\documentclass[11pt,a4paper,openright,twoside]{report}

\usepackage[ngerman,english]{babel}
\usepackage[utf8]{inputenc}

\usepackage{amsmath,amssymb,amsthm}
\usepackage{hyperref}
\usepackage{mathtools}
\usepackage[shortlabels]{enumitem}
\usepackage{bbm}
\usepackage{graphicx,graphics,float}
\usepackage{subfigure}
\usepackage[left=1in, right=1in, top=1in, bottom=1in, includefoot, headheight=15pt]{geometry}
\usepackage{stmaryrd}

\usepackage{color}
\definecolor{darkgreen}{rgb}{0, .5, 0}
\definecolor{darkred}{rgb}{.5, 0, 0}
\definecolor{orange}{rgb}{1, 0.3, 0.1}

\usepackage{setspace}

\theoremstyle{plain}
\newtheorem{theorem}{Theorem}[section]
\newtheorem*{mocktheorem*}{Mock Theorem}
\newtheorem{proposition}[theorem]{Proposition}
\newtheorem{corollary}[theorem]{Corollary} 
\newtheorem{assumption}[theorem]{Assumption} 
\newtheorem{lemma}[theorem]{Lemma} 
\newtheorem{example}[theorem]{Example}

\theoremstyle{definition} 
\newtheorem{definition}[theorem]{Definition}

\ifx\prop\undefined
\newtheorem{prop}{Proposition}[section]
\fi

\newtheorem{remark}[theorem]{Remark}

\usepackage{etoolbox}
\AfterEndEnvironment{theorem}{\noindent\ignorespaces}
\AfterEndEnvironment{mocktheorem*}{\noindent\ignorespaces}
\AfterEndEnvironment{proposition}{\noindent\ignorespaces}
\AfterEndEnvironment{corollary}{\noindent\ignorespaces}
\AfterEndEnvironment{assumption}{\noindent\ignorespaces}
\AfterEndEnvironment{lemma}{\noindent\ignorespaces}
\AfterEndEnvironment{example}{\noindent\ignorespaces}
\AfterEndEnvironment{definition}{\noindent\ignorespaces}
\AfterEndEnvironment{remark}{\noindent\ignorespaces}

\numberwithin{equation}{section}
\usepackage{multicol}

\usepackage{fancyhdr}
\fancypagestyle{plain}{%
  \fancyhf{}%
}

\usepackage{mathrsfs}

\newcommand{\E}{{\mathbb{E}}}
\providecommand{\R}{{\mathbb{R}}}

\newcommand{\dd}{{\rm d}}

\providecommand{\N}{{\mathbb N}}

\newcommand{\n}{{\rm n}}

\newcommand{\1}{\ensuremath{\mathbf{1}}}

\def\P{{\mathbb P}}

\providecommand{\abs}[1]{\ensuremath{\left\lvert#1\right\rvert}}

\DeclareRobustCommand*\circFSuper[1]{\accentset{\circ}{f}^{#1}}
\newcommand{\circFSuperSub}[2]{\accentset{\circ}{f}^{#1}_{#2}}
\newcommand{\circG}{\accentset{\circ}{g}}

\newcommand{\circGSuperSub}[2]{\accentset{\circ}{g}^{#1}_{#2}}
\newcommand{\circH}{\accentset{\circ}{h}}

\newcommand{\circHSub}[1]{\accentset{\circ}{h}_{#1}}
\newcommand{\circHSuperSub}[2]{\accentset{\circ}{h}^{#1}_{#2}}
\newcommand{\circD}{\accentset{\circ}{D}}

\newcommand{\circDSub}[1]{\accentset{\circ}{D}_{#1}}
\newcommand{\circDSuperSub}[2]{\accentset{\circ}{D}^{#1}_{#2}}

\DeclareRobustCommand*\circP{\accentset{\circ}{P}}
\newcommand{\circRho}{\accentset{\circ}{\rho}}
\newcommand{\circPhi}{\accentset{\circ}{\phi}}
\newcommand{\circPsi}{\accentset{\circ}{\psi}}

\usepackage{relsize}
\usepackage{accents}
\usepackage{array}
\usepackage{bm}
\usepackage{dashrule}

\let\origdoublepage\cleardoublepage
\newcommand{\clearemptydoublepage}{%
  \clearpage
  {\pagestyle{empty}\origdoublepage}%
}
\let\cleardoublepage\clearemptydoublepage

\newcommand{\Abgabedatum}{23. Januar 2019}

\begin{document}
\singlespacing
\pagenumbering{roman}
\selectlanguage{ngerman}

\begin{titlepage}
\newcommand{\HRule}[1]{\rule{\linewidth}{#1}} 	
\centering

\vspace*{4.5cm}

\begin{otherlanguage}{english}{\Large

\setlength{\arrayrulewidth}{2pt}
\begin{tabular}{m{0.00025\textwidth}m{0.92\textwidth}m{0.00025\textwidth}}
\hline
~\vspace*{2.0cm}~&\centering\uppercase{\textbf{Mathematical Modeling of Systemic Risk in Financial Networks:\\\large Managing Default Contagion and Fire Sales}}&\\ [-0.0cm]
\hline
\end{tabular}
}
\end{otherlanguage}

\Large
\vspace*{2cm}

Dissertation\\
an der Fakultät für Mathematik, Informatik und Statistik\\
der Ludwig-Maximilians-Universität München

\vfill

eingereicht von

\vspace*{0.2cm}
Daniel Ritter

\vspace*{2cm}
\Abgabedatum

\end{titlepage}

\thispagestyle{empty}
\cleardoublepage

\begin{titlepage}
\setcounter{page}{3}

\Large
~
\vfill

\noindent 1.~Gutachter: Prof.~Dr.~Thilo Meyer-Brandis

\vspace*{0.2cm}
\noindent 2.~Gutachter: Prof.~Dr.~Konstantinos~Panagiotou

\vspace*{0.2cm}
\noindent 3.~Gutachter: Prof.~Dr.~Rama~Cont

\vspace*{0.6cm}
\noindent Tag der m\"undlichen Pr\"ufung: 13.~Mai 2019

\end{titlepage}

\thispagestyle{empty}
\cleardoublepage

\begin{titlepage}
\setcounter{page}{5}

\centering

{\LARGE \bf Eidesstattliche Versicherung}\\
\normalsize(Siehe Promotionsordnung vom 12.07.11, \S~8, Abs.~2 Pkt.~5.)

\Large
\flushleft

\vspace*{2cm}
Hiermit erkläre ich, Daniel Ritter, an Eidesstatt, dass die Dissertation von mir selbstst\"andig, ohne unerlaubte Beihilfe angefertigt ist.


\vspace*{2cm}
M\"unchen, den \Abgabedatum

\end{titlepage}

\thispagestyle{empty}
\cleardoublepage

\setcounter{page}{7}

\chapter*{Zusammenfassung}

Wie die Finanzkrise in 2007/08 eindrucksvoll zeigte, bergen Ansteckungseffekte in Finanznetzwerken eine große Gefahr für die Stabilität des gesamten Systems. Ohne ausreichende Kapitalanforderungen an Banken und andere Finanzinstitutionen können sich anfangs lokal beschränkte Schocks über verschiedene Ansteckungskanäle im gesamten System ausbreiten und sich dabei um ein Vielfaches verstärken. Das Ziel dieser Dissertation ist es deswegen, zwei ausgewählte Ansteckungskanäle dieses sogenannten \emph{systemischen Risikos} genauer zu untersuchen, mathematisch zu modellieren und Konsequenzen für das systemische Risikomanagement von Finanzinstitutionen abzuleiten.

Der erste Ansteckungskanal, welchen wir betrachten, ist \emph{Default Contagion}. Der zugrundeliegende Effekt ist hierbei, dass insolvente Institutionen ihre Schulden oder andere finanzielle Verpflichtungen nicht mehr -- oder nur teilweise -- bedienen können. Gläubiger oder anderweitig direkt beeinflusste Parteien im System sind deshalb gezwungen, Abschreibungen vorzunehmen, und werden durch die erlittenen finanziellen Verluste möglicherweise selbst in die Insolvenz getrieben. Dies wiederum läutet eine neue Runde im \emph{Default Contagion}-Prozess ein. In unserem Modell beschreiben wir jede Institution vereinfacht durch die Gesamtheit der Finanzpositionen, denen sie ausgesetzt ist, sowie ihr ursprüngliches Kapital. Unser Ausgangspunkt ist hierbei die Arbeit von Detering et al.~\cite{Detering2015a} -- ein Modell für Ansteckung in ungewichteten Netzwerken -- welches insbesondere die exakte Netzwerkkonfiguration als zufällig betrachtet und asymptotische Ergebnisse für große Netzwerke herleitet. Wir erweitern dieses Modell, sodass auch gewichtete Netzwerke betrachtet werden können und dadurch eine Anwendung auf Finanznetzwerke möglich wird. Genauer leiten wir für einen gegebenen anfänglichen Schock einen expliziten, asymptotischen Ausdruck für den durch Ansteckung verursachten Gesamtschaden im System her und liefern ein notwendiges und hinreichendes Kriterium dafür, dass ein ungeschocktes Finanznetzwerk stabil gegenüber kleinen Schocks ist. Ferner entwickeln wir eine explizite Formel für notwendiges und hinreichendes Risikokapital auf Ebene der einzelnen Institutionen, sodass die Stabilität des Finanznetzwerks gewährleistet wird. Durch Simulationen zeigen wir, dass unsere asymptotischen Resultate bereits für Finanznetzwerke in der typischen Größenordnung von einigen tausend Institutionen eine sehr gute Beschreibung liefern.

In einem nächsten Schritt entwickeln wir eine mehrdimensionale Erweiterung unseres Modells für \emph{Default Contagion}, um die in Finanznetzwerken beobachteten komplexen Strukturen abbilden zu können -- allen voran ist das die \emph{Core-Periphery} Struktur, aber auch mehrschichtige Strukturen, regionale Konzentrationen und Mischformen davon. Zu diesem Zweck weisen wir jeder Institution im Netzwerk einen zusätzlichen Parameter zu, welcher deren Typ beschreibt. Das Netzwerk wird dadurch in Subsysteme (Blöcke) eingeteilt. Dieses neue Modell ermöglicht es insbesondere, die Auswirkungen eines lokalen Schocks in einem der Subsysteme (z.\,B.~ein bestimmtes Land) auf das Gesamtsystem zu quantifizieren. Unsere Resultate zeigen, dass diese zusätzliche Komplexität die Stabilität des Finanzsystems stark beeinträchtigen kann, und wir entwickeln Maßnahmen, mit denen sich einzelne Subsysteme vor der Ansteckung durch andere Subsysteme schützen können. Außerdem gelingt es uns, realistischere finanzielle Verpflichtungen zu modellieren, deren Höhe von beiden Vertragsparteien abhängt. Bisher war es nur unter der Annahme, dass das Ausmaß einer Ansteckung lediglich von der exponierten Seite abhängt, möglich, aussagekräftige analytische Ergebnisse abzuleiten. Wie wir demonstrieren, kann diese vereinfachende Annahme zu einer gravierenden Unterschätzung des Risikopotentials in einem System führen, und die zusätzliche Komplexität in unserem Modell ist deswegen essentiell, um die Stabilität eines Systems realistisch einschätzen zu können.

Als nächstes entwickeln wir ein Modell für den Ansteckungskanal \emph{Fire Sales}, bei dem Institutionen auf einen anfänglichen Schock mit dem Verkauf von Aktien reagieren -- z.\,B.~aufgrund entsprechender Regularien. Dadurch geraten die Aktienpreise unter Druck und Investoren erleiden weitere Verluste. Dies wiederum führt erneut zu Verkäufen und der Prozess setzt sich weiter fort. Zur Modellierung dieses Ansteckungsprozesses beschreiben wir jede Institution durch die Anzahl und Art ihrer gehaltenen Aktien sowie ihr ursprüngliches Kapital und den durch einen anfänglichen Schock verursachten Verlust. Zusätzlich nehmen wir an, dass Institutionen ihre Entscheidung zum Verkauf von Aktien anhand einer gegebenen Funktion treffen und auch die Auswirkungen der Verkäufe auf die Aktienpreise durch eine gegebene Funktion beschrieben werden. In unserer Modellierung greifen wir Ideen aus der Literatur zu \emph{Default Contagion} auf und es gelingt uns so, eine rigorose Beschreibung des Prozesses zu liefern. Insbesondere bestimmen wir asymptotisch den Gesamtschaden im System, der durch den Anfangsschock und anschließende \emph{Fire Sales} verursacht wird, und wir liefern eine Klassifikation von stabilen Systemen sowie hinreichendes Risikokapital, um die Stabilität eines Finanzsystems sicherzustellen. Erneut belegen wir die Anwendbarkeit unserer asymptotischen Resultate durch geeignete Simulationen.

Schließlich kombinieren wir die Modelle für \emph{Default Contagion} und \emph{Fire Sales}, um ein kompletteres Bild von Ansteckungseffekten in Krisenzeiten zu bekommen.  Unsere Ergebnisse zeigen, dass sich die beiden Ansteckungskanäle gegenseitig enorm verstärken können, und unterstreichen deswegen die Wichtigkeit von kombinierten Modellen für das Verständnis von systemischem Risiko. Auch für den kombinierten Fall gelingt es uns, Kapitalanforderungen herzuleiten, die ausreichen, um die Stabilität des Systems zu gewährleisten, und deshalb von großem Interesse für regulatorische Einrichtungen sind.

\cleardoublepage
\selectlanguage{english}
\chapter*{Abstract}

As impressively shown by the financial crisis in 2007/08, contagion effects in financial networks harbor a great threat for the stability of the entire system. Without sufficient capital requirements for banks and other financial institutions, shocks that are locally confined at first can spread through the entire system and be significantly amplified by various contagion channels. The aim of this thesis is thus to investigate in detail two selected contagion channels of this so-called \emph{systemic risk}, provide mathematical models and derive consequences for the systemic risk management of financial institutions.

The first contagion channel we consider is \emph{default contagion}. The underlying effect is here that insolvent institutions cannot service their debt or other financial obligations anymore -- at least partially. Debtors and other directly impacted parties in the system are thus forced to write off their losses and can possibly be driven into insolvency themselves due to their incurred financial losses. This on the other hand starts a new round in the default contagion process. In our model we simplistically describe each institution by all the financial positions it is exposed to as well as its initial capital. In doing so, our starting point is the work of Detering et al.~\cite{Detering2015a} -- a model for contagion in unweighted networks -- which particularly considers the exact network configuration to be random and derives asymptotic results for large networks. We extend this model such that weighted networks can be considered and an application to financial networks becomes possible. More precisely, for any given initial shock we deduce an explicit asymptotic expression for the total damage caused in the system by contagion and provide a necessary and sufficient criterion for an unshocked financial system to be stable against small shocks. Moreover, we develop an explicit formula for necessary and sufficient risk capital at the level of single institutions that ensures stability of the financial network. We demonstrate by simulations that our asymptotic results give a good description for financial networks of the size of a few thousand institutions already.

In the next step, we develop a multi-dimensional extension of our model for default contagion such that we can describe the complex structures observed in financial networks -- particularly \emph{core-periphery}-structures but also multi-layered structures, regional concentrations and mixtures thereof. To this end, we assign to each institution in the network an additional parameter describing its type. The network is thereby divided in different subsystems (blocks). In particular, this new model enables us to quantify the impact of a local shock in one of the subsystems (e.\,g.~a certain country) to the global system. Our results show that the additional complexity can significantly affect the stability of the financial system and we develop measures for the individual subsystems to secure themselves against contagion from other subsystems. Furthermore, we accomplish a more realistic modeling of financial obligations whose size may depend on both contracting parties. So far, meaningful analytical results could only be derived under the assumption that the amount of contagion only depends on the exposed party. We demonstrate that this simplifying assumption can lead to a grave underestimation of the risk potential of a system and the additional complexity in our model is thus essential for a realistic assessment of a system's stability.

Next, we develop a model for the contagion channel of \emph{fire sales} at which institutions react to an initial shock by selling asset shares -- forced by regulations for instance. As a result the share prices come under pressure and investors suffer further losses. This in turn leads again to asset sales and the process proceeds. For the modeling of this contagion channel, we describe each institution by the number and kind of its held asset shares as well as its initial capital and the losses suffered due to some initial shock. Additionally we assume that institutions make their decision to sell shares according to some given function and also the price impact of sales is described by a given function. In our modeling we resort to ideas from the \emph{default contagion} literature and we thus achieve a rigorous description of the process. In particular, we asymptotically determine the total damage to the system caused by the initial shock and the subsequent \emph{fire sales}, and we provide a classification of stable systems as well as sufficient risk capital to ensure stability of a financial system. Again we verify the applicability of our asymptotic results by suitable simulations.

Finally, we combine the models for \emph{default contagion} and \emph{fire sales} to get a more complete picture of contagion effects in periods of crisis. Our results show that the two contagion channels can tremendously amplify each other and thus stress the importance of combined models for the understanding of \emph{systemic risk}. Also for the combined case we achieve to derive capital requirements sufficient to ensure stability of the system that are hence of great interest to regulatory institutions.


\cleardoublepage
\setcounter{tocdepth}{2}
\tableofcontents

%

\cleardoublepage
\pagenumbering{arabic}
\chapter{Introduction}

We are living in an ever more connected world today providing us with uncountable possibilities to make our lives more informed, more efficient, more profitable, more enjoyable and easier. These benefits come at a cost, however. More precisely, the growing complexity of dependencies in any kind of connected structure gives rise to concern regarding its stability. In many settings an initial local shock to a system that considered for itself may be ever so innocuous can be transmitted along the connections and grow to become a major threat to the entire system -- so-called \emph{systemic risk}. While this type of risk certainly is of general importance, in this thesis we focus on the area of \emph{financial systems} where financial institutions are linked by various types of dependencies that allow financial distress to spread. In particular, we aim to make a contribution towards a better understanding of the underlying mechanisms, their joint impact and how to prevent systemic cascades. As will be further discussed in Section \ref{1:sec:systemic:risk}, two of the main drivers of systemic risk are the contagion channels \emph{default contagion} and \emph{fire sales} which we will introduce in Sections \ref{1:sec:default:contagion} and \ref{1:sec:fire:sales} respectively. In Section \ref{1:sec:contribution}, we will give an outlook on the contributions made in the chapters to follow.

\section{Systemic Risk}\label{1:sec:systemic:risk}
Today there are not one but many definitions of \emph{systemic risk} emphasizing different aspects of the same rough underlying idea \cite{BIS1994,Duffie2003b,Kaufman1995,Kaufman2003,Mishkin1995,Schwarcz2008,Taylor2010}. As Hurd points out in \cite{Hurd2016}, however, the following three ingredients are essential for the concept of systemic risk.
\begin{enumerate}
\item \textbf{A triggering event of some kind:} This can be any stress scenario (external or internal) impairing the institutions in the system. Examples could be a sudden drop in asset values, the failure of one or more institutions due to mismanagement or crime, new legislation, natural disasters, or terrorist attacks to mention just a few.
\item \textbf{Propagation of distress within the financial system:} The initially local shock event (e.\,g.~the burst of the US housing bubble in 2008) spreads to other parts of the system by direct or indirect relations between institutions. This process bears an undeniable resemblance with the contagion of diseases between humans -- initially healthy people can become infected if in contact with disease carriers and further communicate the disease themselves afterwards -- and it is thus termed \emph{financial contagion process} nowadays. Hurd lists the following four main channels of contagion. 
\begin{enumerate}[\alph*.]
\item\label{1:enum:contagion:channels:asset:correlation} \emph{Asset Correlation:} Actually not being a contagion channel in the narrower sense, similar or highly correlated asset portfolios can make different institutions susceptible to the same kinds of initial shock events, thus considerably weakening the system as a whole and fueling other types of contagion.
\item \emph{Default Contagion:} Also called \emph{balance sheet contagion}, it describes the effect that upon default of some institution $i$ in the system all other institutions need to write off according interbank assets (pending financial obligations of $i$) in their balance sheet. These assets can be usual \emph{loans} but also \emph{securities cross-holdings}, \emph{derivatives} or \emph{foreign exchange} (see \cite{Thurner2015} for instance). As a result other institutions may default and more write-offs may be the consequence. For more details see Section \ref{1:sec:default:contagion}.
\item \emph{Liquidity Contagion:} Concerned with liquidity rather than solvency this contagion channel explains how the shortage of funding can spread through the system. Some institution $i$ might find itself in the situation of not having enough liquidity to meet short term obligations, thus recalling or not rolling over its issued loans. By this decision on the other hand, $i$'s debtors can run into the danger of a liquidity shortage and reclaim their issued loans for their part and so on.
\item \emph{Market Illiquidity and Asset Fire Sales:} As in \ref{1:enum:contagion:channels:asset:correlation}, overlapping asset portfolios are the catalyst of this channel of contagion. More than only a correlated initial shock, however, it describes how financial distress can force institutions to sell off their assets, hence depressing market prices and possibly forcing other institutions to react to dropped values of their portfolios by selling off assets as well. For more details see Section \ref{1:sec:fire:sales}.
\end{enumerate}
\item \textbf{Significant macroeconomic impact:} As the financial system serves several purposes for the wider economy and society such as the provision of liquidity and credits or the infrastructure for payment systems, a breakdown of large parts of it generally does great harm also beyond the participating institutions.
\end{enumerate}
Returning to the analogy of financial contagion to the contagion of diseases between humans, in his famous speech \cite{Haldane2009} Haldane compared the global financial crisis of 2007/08 to the SARS epidemic of 2002/03. Both started with an external trigger event. Uncertainty about its causes and consequences lead to panic and overreaction spreading across the globe. In the end, the macroeconomic impact was huge compared to the relatively moderate implications of the triggering event.

In particular, the tremendous repercussions for economy and society of a global spread of distress necessitate the intervention by regulating institutions to mitigate or even prevent such cascades from happening in the future. In classical risk management before the latest crisis it was generally accepted to measure risk on an institution level only, taking into account market risk of falling asset prices and credit risk of defaulting direct counterparties, but neglecting second-order and feedback effects due to contagion which were deemed negligible. It was one of the main insights from 2007/08, however, that this was indeed an oversimplification of the real situation.

Much work has hence been put into the proper understanding of contagion and amplification effects over the last decade and despite it being a relatively young field of research the sheer number of publications devoted to systemic risk once again stresses the importance of the topic. In the following, we briefly summarize three of the main lines of research. For a more detailed overview of research on systemic risk see \cite{book:sr} for example.

\subsection{Systemic Risk Measures}
The basic idea of this line of research is to generalize the classical notion of monetary risk measures (see \cite{Foellmer2002} for instance) to account for systemic effects. One can distinguish four different approaches to this problem. First, a rather direct approach applying the concept of value-at-risk and expected-shortfall to systemic quantities. Prominent work here is SRISK from \cite{Acharya2012,Brownlees2017}, the \emph{conditional value at risk} (CoVaR) from \cite{Adrian2016} and the \emph{marginal expected shortfall} (MES) as well as the \emph{systemic expected shortfall} (SES) from \cite{Acharya2017}. While CoVaR is concerned with the extent of systemic damage given that a certain individual institution experiences large losses, SRISK, MES and SES turn the tables and investigate the individual damage of a certain institution in the event of a systemic crisis. Also see \cite{Bisias2012} for more details on the above mentioned and other risk measures. Second, in a more axiomatic fashion, it is possible to consider a multi-variate random vector of risk factors in a certain financial system and first aggregate those to a single uni-variate systemic risk quantity before applying a classical uni-variate risk measure to it (adding sufficient capital to make the risk factor acceptable) -- see \cite{Chen2013} and \cite{Kromer2016} for instance. In the third approach, the order of adding capital and aggregating is reversed. That is, first a sufficient multi-variate vector of capitals is determined to make the system acceptable from a regulator's perspective. Aggregating those capitals then leads to the systemic risk measure. This approach is taken for example in the works \cite{Armenti2015}, \cite{Biagini2017}, \cite{Biagini2018b} and \cite{Feinstein2017}. Finally, emphasizing the dependence of distress in financial systems, there is a line of research concerned with conditional risk measures. See \cite{Acciaio2013}, \cite{Bion-Nadal2004}, \cite{Detlefsen2005}, \cite{Filipovic2012}, \cite{Hoffmann2016} and \cite{Hoffmann2018} for example.

\subsection{Mean-field games}
A further line of research was started by \cite{Fouque2013} and considers a system of diffusion dynamics as a model for the capitals of institutions in the system. More precisely, the logarithm of an institution $i$'s capital $Y_t^{(i)}$ at time $t$ is supposed to be driven by an individual Brownian motion $W^{(i)}$ and borrowing/lending-activity between institutions $i$ and $j$ is modeled by a drift term $\alpha\big(Y_t^{(j)}-Y_t^{(i)}\big)$ for some positive constant $\alpha$. In the mean-field limit, as the size of the system becomes large, a propagation of chaos result is derived in the sense that the individual capitals decouple and the random variables $Y^{(i)}$ are given as independent Ornstein-Uhlenbeck processes. Furthermore, one of the main results in \cite{Fouque2013} is that the borrowing/lending according to $\alpha$ stabilizes the system against small and medium-sized shocks, making a collapse of the system less probable, but at the same time exacerbating the extent of cascades if those occur. The model from \cite{Fouque2013} has been extended for example in \cite{Carmona2018}, \cite{Carmona2015}, \cite{Fouque2013b} and \cite{Kley2014} to capture effects of core/periphery structures and borrowing from resp.~lending to a central bank. By similar means, \cite{Biagini2018} considers the effect of an asset bubble on the robustness of individual institutions in a large network. Note that this line of research is rather abstract with regard to the modeling of borrowing/lending as an institution $j$ with higher capital will lend money to an institution $i$ with lower capital that it will never reclaim unless the sign of $Y_t^{(j)}-Y_t^{(i)}$ flips at some point in time -- even then there is no memory about the total amount owed between institutions.

\subsection{Network models}\label{1:ssec:network:models}
Maybe the most direct approach of modeling contagion in financial systems -- and the one followed in this thesis -- is to consider explicitly the financial system with its institutions and linkages between them. That is, if we consider a financial system of $n\in\N$ institutions we describe the direct exposures (liabilities) between institutions by a matrix $e=(e_{i,j})_{1\leq i,j\leq n}$ where $e_{i,j}$ denotes the exposure that institution $j$ has to $i$. Also other model parameters can be considered for each institution such as its capital reserves or assets held for instance. Probably the most prominent models following the network approach are the Eisenberg-Noe model from \cite{Eisenberg2001} and the Gai-Kapadia model from \cite{Gai2010} which we will discuss along with their respective extensions in the following section. Moreover, in \cite{Capponi2016} the authors use matrix majorization tools to compare systemic losses in financial systems according to different concentration of liabilities. The work \cite{Chong2016} develops a structural model for default and allows to compute explicitly the joint probabilities of default and survival by a Bayesian network approach. One drawback of network models is that typically the whole network needed to be observed for calibration, while often only aggregated data is available. The authors of \cite{Gandy2017} thus develop a Bayesian methodology that allows to estimate the individual entries of the liability matrix $e$ from this aggregated data and they apply it to stress testing of European banks.

\section{Default Contagion}\label{1:sec:default:contagion}
As roughly outlined above, \emph{default contagion} considers the progressive infection of direct creditors in a financial system. Initially some distressed institutions in the financial system may not be able to repay (in full) their pending loans or other liabilities and declare bankruptcy. As a consequence exposed institutions in the network need to write off their losses. This is completely analogue to classical credit risk management. However, due to these first-order losses further institutions may become insolvent and more write-offs ensue. In particular, some institution in the system that does not maintain a direct relationship with any of the initially distressed institutions but is linked to them by some vulnerable joint business partner can still get into trouble even though it may not even have knowledge of this indirect link. Furthermore, even for directly exposed institutions classical risk management can significantly underestimate losses as it only considers losses from the particular business relation itself but neglects exposures to institutions that are driven into default themselves. During the financial crisis in 2007/08 it was demonstrated that these second order effects neglected in classical risk management actually play an important role when assessing a system's stability with respect to initial shock events and it is thus necessary to better understand those effects.

In this section, we want to elaborate on two intensively pursued approaches within the network line of research from Subsection \ref{1:ssec:network:models}. While both consider contagion on a network represented by some exposure matrix $e$, the methods used are quite different. Whereas the Eisenberg-Noe (EN) model proposed in \cite{Eisenberg2001} and its extensions consider the final state of an explicitly observed system after contagion via an equilibrium representation, the research branch started by \cite{Gai2010} and \cite{Cont2016} models the network configuration $e$ as a random sample calibrated to some observed network and then considers the contagion cascade round by round.

\subsection{Eisenberg-Noe-type Models}
The seminal work \cite{Eisenberg2001} was one of the first papers of financial mathematics that systematically addressed feedback effects of contagion in financial networks. The authors considered a financial network consisting of $n\in\N$ institutions and described it by the observed liability matrix $e\in\R_{+,0}^{n\times n}$ and the vector $c\in\R_{+,0}^n$ of cash values available to each institution which may also include external assets minus external liabilities that are of higher seniority than interbank liabilities. The idea is now to assume that the whole financial system was cleared in the sense that all liabilities were settled. If each institution $i$ in the system is solvent in the sense that its total assets $c+\sum_{1\leq j\leq n}e_{j,i}$ (cash + interbank assets) exceed or equal its total interbank liabilities $\tilde{p}_i:=\sum_{1\leq j\leq n}e_{i,j}$, then it is always possible that all liabilities are repaid. The clearing vector of the system is then given by $p=\tilde{p}$, where $p_i$ shall denote the total amount paid by institution $i$. If, however, $c+\sum_{1\leq j\leq n}e_{j,i} < \tilde{p}_i$, then institution $i$ is insolvent and cannot settle all its liabilities in full. In fact, we assume that $i$ repays as much of its debt as possible and the money available to it is split among its creditors $j$ according to the relative liabilities
\[ \Pi_{i,j} = \begin{cases}e_{i,j}\,\tilde{p}_i \hspace*{-0.04cm}\big.^{-1},&\text{if }\tilde{p}_i>0,\\0,&\text{otherwise.}\end{cases} \]
Note, however, that the value of $i$'s interbank assets further depends on the actual amount that $i$'s debtors are able to clear and there is thus a potential interdependence. Nevertheless, for a feasible clearing vector $p=(p_i)_{1\leq i\leq n}$ it cannot be the case that $p_i<(c_i+\sum_{1\leq j\leq n}\Pi_{j,i}p_j)\wedge\tilde{p}_i$ as $i$ would end up with a positive amount of money although not repaying all of its debt. On the other hand, $i$ will not pay out more than its total outstanding loans $\tilde{p}_i$ and it cannot repay more than it has available. Hence, it cannot be the case that $p_i>(c_i+\sum_{1\leq j\leq n}\Pi_{j,i}p_j)\wedge\tilde{p}_i$. Any adequate clearing vector $p$ thus needs to satisfy
\begin{equation}\label{1:eqn:eisenber:noe:fixed:point}
p = \left( c + \Pi^\top p \right) \wedge \tilde{p}.
\end{equation}
It is then the main result in \cite{Eisenberg2001} that indeed such $p$ exists. Moreover, in general there are a least and a largest solution to \eqref{1:eqn:eisenber:noe:fixed:point} but under mild assumptions on the model parameters those actually coincide and there is a unique way to clear the financial system.

The EN-model therefore describes the final state of the system from an equilibrium perspective and thus loses the sequential notion briefly outlined before. We can, however, translate this idea into an algorithm to determine the largest clearing vector $p^+$ in rounds $k\in\N$. We initially assume that each institution clears all of its liabilities and set $p^{(0)}=\tilde{p}$. Nevertheless, any insolvent institution $i$ can only repay a total amount of
\[ p_i^{(1)}:=c_i+\sum_{1\leq j\leq n}\Pi_{j,i}p_j^{(0)} = c_i + \sum_{1\leq j\leq n}e_{j,i} < \tilde{p_i} = p_i^{(0)}. \]
Thus all institutions in the system have to write off their losses with initially insolvent institutions and the money paid out by them is described by the vector
\[ p^{(2)} := \left( c + \Pi^\top p^{(1)}\right) \wedge \tilde{p}. \]
This in turn may further reduce the money available to already insolvent institutions and also institutions that were solvent before may go bankrupt. Continuing the process, in round $k\geq3$ we thus derive
\[ p^{(k)} = \left(c+\Pi^\top p^{(k-1)}\right) \wedge \tilde{p} \]
and it can be shown that indeed $p^{(k)}\to p^+$ as $k\to\infty$. 

For a given financial system and an according clearing vector we can now determine the losses in the system and identify the institutions that default due to contagion effects. As remarked earlier already, however, it is a strong assumption that the complete network configuration is available. If at all possible (especially for international relations), above calculations have to be performed by some regulating institution after collecting all the necessary data from participating institutions. Moreover, the fixed point equation \eqref{1:eqn:eisenber:noe:fixed:point} is very high-dimensional and as thus intuition about the solution $p$ itself and its robustness to changes of the model parameters is hard to find in general.

Nevertheless, due to its tractability the EN-model has proven to be a popular model. In particular, there have been made many extensions to the original model from \cite{Eisenberg2001} to account for a more realistic setting of the financial networks. In \cite{Elliott2014}, \cite{GLASSERMAN2015383} and \cite{Rogers2013} for instance, bankruptcy costs have been included in the model representing losses in the enforcement process of claims and deductions on asset values when liquidating them. The latter one should not be equated with fire sales, however, as there is no contagion considered in the before listed works. This channel (for more details see Section \ref{1:sec:fire:sales}) has been included in the EN-model by the works \cite{Amini2016}, \cite{Cecchetti2016} and \cite{Cifuentes2005} for example. Moreover, the extensions \cite{Elliott2014} and \cite{Elsinger2009} consider cross-holdings between banks. In the model of \cite{Weber2016}, all the above mentioned effects are consolidated.

\subsection{Random Graph Models}\label{1:ssec:cascade:models}
As mentioned earlier already, it is often unrealistic to assume a completely known network configuration of the financial system. Even if a regulating institution may be able to collect all the data about direct exposures in its area of responsibility, this becomes less reasonable when considering also cross-border relations. Moreover, the precise network configuration may exhibit changes over time maybe even on a daily basis. However, as has been shown in empirical works (see \cite{Cont2013} for instance) global statistics such as the degree distributions in the network are relatively constant over time. A second popular approach to model default contagion in financial networks is thus by means of random graphs, where vertices represent institutions in the system and weighted, directed edges the exposures between them. The underlying probability measure can then be calibrated to observed data and the random samples are typical configurations of the actual present or a future network. It is then one possibility to apply the EN-methodology described in the previous subsection to any random configuration in order to compute the final clearing vector (see \cite{Gandy2017} for example). In general this requires a numerical dealing with the problem to derive statements about the typical magnitude of a cascade and the corresponding stability of a system. Especially for large networks, however, one can employ the powerful methodology of probabilistic limit theorems in random settings and the results derived will hold for all typical realizations of the random network and in particular for the observed configuration. Moreover, as those results are typically given in terms of the global statistics of the system and these are relatively constant over time as remarked before, a strong robustness over time is achieved.

To this end, we change our view of the default contagion process away from the equilibrium perspective and towards a more direct modeling of the default process that resembles the algorithm to obtain the largest clearing vector in the EN-model. That is, starting with a set of initially defaulted institutions distress is transmitted to their direct neighbors in the first round. This may cause new defaults in the system and start a second round of default contagion and so on. Different than in the EN-setting where the actually repaid debt can be described according to the severity of the default (the actual loss incurred), it is necessary here to consider a fixed recovery rate (typically even $0\%$) regardless of the amount of money available to the defaulted institution. This assumption can be justified by the fact that the processing of defaults may take months or even longer and at the time of bankruptcy the actual value of the according interbank assets are highly uncertain. One good example of this reasoning is the insolvency of Lehman Brothers in 2008 which took years to liquidate completely and whose traded recovery-rate amounted to only $8.625\%$ in a bond auction for settling credit default swaps just three weeks after the default \cite{creditfixings}.

The random graphs approach to default contagion was started with \cite{Gai2010} and a well known representative of this line of research is the model in \cite{Cont2016} (also see \cite{Amini2014c}). There the authors chose to describe the financial network by the so-called \emph{configuration model} that takes as input the observed empirical degree distribution (jointly for in- and out-degrees) of some financial network and draws uniformly at random a configuration satisfying this distribution. The results in \cite{Cont2016} then allow for any given initial shock event to compute asymptotically for large networks the size of the cluster of finally defaulted institutions. Similar at a first glance to the EN-setting, the final state of the system is described by a fixed point equation. Despite the large system size, however, this equation is actually uni-variate and thus much more tractable than in the EN-model. Moreover, in \cite{Cont2016} a measure of resilience is derived for the financial system that essentially counts the number of so-called \emph{contagious links}, i.\,e.~exposures that alone are large enough to transmit default from one institution to another. A rather strong consequence of this result is that financial systems can be seen as being resilient to small initial shocks as soon as such contagious links are prohibited by a regulating institution. Put in other words, only local effects are responsible for the spread of contagion. The description of resilience here shows a general advantage of the asymptotic viewpoint in the random graphs line of research: For infinitely large networks it is possible to choose arbitrarily small initial shocks for which the final state of the system can be investigated. In comparison, for finite networks of size $n\in\N$ the least possible positive shock size (expressed as the initial default fraction) is $n^{-1}$. It is thus necessary to choose certain parameters for the initial shock and specify a maximal amplification to define resilience of a financial system. For infinitely large networks the notion of resilience emerges completely natural, letting the initial shock size tend to zero (see Chapters \ref{chap:systemic:risk}, \ref{chap:block:model}, \ref{chap:fire:sales} and \ref{chap:fire:sales:default} for more details).

To describe more realistic network topologies the works \cite{Hurd2016,Hurd2017} consider an assortative version of the configuration model where edges are categorized according to the degrees of their adjacent vertices. This allows for a better description of core/periphery structures where periphery banks are almost exclusively connected to a small set of highly connected core banks.

There is a major point of criticism of the models \cite{Cont2016,Amini2014c,Hurd2016,Hurd2017}, however. To be precise, the very heterogeneous nature of real financial networks cannot be reproduced sufficiently. Many empirical analyses such as \cite{Boss2004} for the Austrian banking network or \cite{Cont2013} for the Brazilian one show that typically observed degree sequences are asymptotically heavy-tailed and have infinite second moment. In this case, however, \cite{Janson2009} shows that with high probability (probability converging to one as the network size diverges) the configuration model produces non-simple graphs, i.\,e.~there occur self-loops or multiple edges of the same direction between the same institutions. That is, if we calibrate the configuration model to a specific observed financial network which is by its nature described by a simple graph (no institution is exposed to itself and if there are several contracts between the same institutions those are aggregated to a single one) still in most cases the resulting network turns out to be non-simple and its characteristics are distorted from the originally observed one. While this problem could in general be solved by conditioning on the random graph being simple, the results from \cite{Cont2016} are not applicable in this case anymore as those are all formulated asymptotically with high probability only.


In \cite{Detering2015a}, the authors thus switch to a different kind of random graph model that by definition always produces simple configurations while at the same time allowing for degree sequences with asymptotically infinite second moment and thus capturing the strong heterogeneity of empirically observed data. The model they chose is a directed version of the well-known Chung-Lu random graph \cite{ChungLu2002} that assigns to each institution a certain vertex-weight describing its tendency to form edges to/from other institutions. Especially in the literature of statistical physics random graph models of this kind are referred to as \emph{fitness models} \cite{Caldarelli2002,Gandy2017,Servedio2004}.

\cite{Detering2015a} is formulated for unweighted edges only, however, and contagion is abstractly defined according to certain individual thresholds for each institution that describe the number of debtors that need to fail in order for the particular institution to fail itself. Calibrating these thresholds is clearly a very challenging task as they depend on the capital of the institution, the amounts of credit issued to other institutions and the order of their default. Nevertheless, it is one of the main results in \cite{Detering2015a} that the mere absence of contagious links does not ensure resilience of a financial system with strongly heterogeneous degrees. Moreover, a resilience criterion extending the one from \cite{Cont2016} is derived.

In this thesis we will follow the random graphs approach to default contagion and in fact \cite{Detering2015a} will be the starting point for two different models in the following. We thus summarize its main findings in Subsection \ref{2:ssec:special:case:threshold:model}.

\section{Fire Sales}\label{1:sec:fire:sales}
While the previously discussed channel of \emph{default contagion} distributes financial distress in a system via direct contractual dependencies, the links between institutions relevant for the channel of \emph{fire sales} are their overlapping asset portfolios. The underlying dynamics are as follows. Consider for simplicity a financial system in which institutions can invest in a single common asset. Some initial shock event may now diminish the capital of at least one institution in the system which reacts to this loss by selling some of its asset shares either due to external regulations such as leverage constraints or internal preferences. If the number of sold shares is large compared to typical trading volumes on the market (especially if multiple institutions start to sell shares at the same time) the surplus on the market will reduce the share price by so-called price impact. This has two consequences for the institutions in the system. First, the selling institutions incur even further losses on the particular trades themselves and second, as asset portfolios are marked-to-market each institution invested in the asset effectively loses money (at least in the short term). These additional losses and a general uncertainty about the situation on the financial markets can now provoke even more institutions to sell even more of their asset shares thus further reducing prices and portfolio values and so on. By this iterated process a propagation of distress through the system is described and the initial losses can be amplified considerably. In particular, even institutions that were spared from the initial triggering event come under pressure due to their asset portfolio's overlap with distressed institutions that they might not even have known of.

Additionally to the above mentioned works \cite{Amini2016,Cecchetti2016,Cifuentes2005,Weber2016} extending the EN-model to account for fire sales there is a variety of different approaches to model fire sales in financial systems. In a sense similar to the cascade models from Subsection \ref{1:ssec:cascade:models}, in \cite{Caccioli2014} the authors describe fire sales as a branching process and they find certain criteria for stability of a network. In \cite{Ibragimov2011} the consequences of portfolio diversification are investigated and it is found that benefits for individual institutions can come at a disadvantage for the system as a whole. Similar considerations are made in \cite{Beale2011}. The work \cite{Wagner2010} shows that institutions are motivated to create heterogeneous portfolios to avoid the risk of joint liquidation. Related to these works is also \cite{Kley2016} for the setting of a reinsurance market with overlapping risky objects. For a continuous time setting the effect of fire sales on the dynamics and correlations of asset prices are investigated in \cite{Cont2013a}. Moreover, extending the work \cite{Khandani2011}, in \cite{Cont2016b} the authors develop a model that can explain volatility spikes and the increase of correlations in times of financial distress. The effects of fire sales have also been the object of interest in a series of empirical papers. \cite{Guo2016} for instance analyzes the topology of the induced network of overlapping asset portfolios and \cite{Braverman2014} uses a network representation to quantify dependencies. The authors of \cite{Cont2017} propose a framework for stress testing systems with regard to fire sales and develop centrality indices for institutions in a network of common asset holdings. Other methods to quantify the dependencies due to asset holdings have been developed for example in \cite{girardi18} by means of the scalar product of portfolio weights or in \cite{Kritzman2011} by the so-termed absorption ratio using a principal component decomposition of asset returns. In \cite{Duarte2013} the authors develop a systemic index of aggregate vulnerability and demonstrate the decreasing stability of the financial system in the years before the global crisis.

\section{Contribution of This Thesis}\label{1:sec:contribution}
We will now give an outlook on the contributions towards the understanding of systemic risk made in the following. The chapters of this thesis are in large parts adopted from the papers \cite{Detering2016}, \cite{Detering2018}, \cite{Detering2018b} and \cite{Detering2018c} and are all devoted to different aspects of default contagion and fire sales. The common thread of the different chapters is that we use random graph methods to model particular empirically observed network characteristics, we analyze asymptotically the effect of shock events on the system and compute the final state of the system at the end of the contagion cascade, we derive criteria for resilience and non-resilience in our models that allow us to understand which network characteristics promote or hinder the spread of distress in the financial system, and we derive sufficient systemic capital requirements that can be prescribed by some regulating institution to all participants in the system \mbox{in order to contain and prevent large cascades.}

\paragraph{Contribution in Chapter \ref{chap:systemic:risk}:} In Chapter \ref{chap:systemic:risk}, for the purpose of analyzing default contagion, we develop a weighted, directed random graph model for financial networks of direct exposures that allows to describe degree sequences with infinite second moment. The model thus extends \cite{Cont2016} and is applicable also to more realistic settings of financial networks.

Using a uni-variate fixed point equation, we then derive asymptotic results for large networks about the final systemic damage caused by some given initial shock event. In particular, we introduce an index of systemic importance for each bank and we consider the final systemic damage to be the total systemic importance of all finally defaulted institutions. This approach is new compared to \cite{Cont2016} and allows for a more relevant assessment of systemic risk than only considering the number of finally defaulted institutions. Further, this approach is in line with and can easily be applied to current regulatory methods where institutions are assigned to certain classes of systemic importance \cite{BaselCommittee2013,Fed2015}.

We then apply our asymptotic results to derive criteria for resilience respectively non-resilience. At this step we make use of the previously mentioned naturally emerging notion of resilience where we let the relative size of the initial shock be arbitrarily small and we particularly call a system non-resilient if nevertheless the final relative systemic damage is lower bounded by some positive constant.

We then go even one step further and apply the resilience criteria in order to obtain explicit formulas for sharp capital requirements for each individual institution. In general (if asymptotically the degree distributions have an infinite second moment), these requirements are more restrictive than the mere absence of contagious links in the system as in \cite{Cont2016}. Instead they can be identified with a sublinear function of the specific in-degree of an institution. In particular, diversification is encouraged by our systemic capital requirements which is for example in line with the results in \cite{Capponi2016}. Moreover, a very appealing feature of our formulas for capital requirements is that those can be computed locally by each institution itself if only a regulating institution with knowledge about the global system statistics publishes two global constants. In particular, it is completely transparent for each institution how the capitals are determined and as those depend on the particular institution's own business relations only, they prevent manipulation by individual institutions either of their own capital requirements or their competitors. As the derived capital requirements are sharp, one can interpret them as an allocation of the total systemic risk (the total necessary capital to secure the system) to each individual institution. This adds to research questions discussed in \cite{Biagini2018b}, \cite{Feinstein2017} and \cite{Hoffmann2017}, but to the best of our knowledge is the first such allocation that does not rely on complete knowledge of the system but can be computed locally.

Finally, we verify by numerical simulations that our asymptotic results and in particular the systemic capital requirements are applicable also to financial systems of reasonable finite size.

\paragraph{Contribution in Chapter \ref{chap:block:model}:} Looking for ways to benefit from portfolio diversification and access to different international markets \cite{Aoki2010,Artis2007,Artis2012,Dell'ariccia2008,Faria2007}, over the last decades institutions have entered more and more cross-border relationships \cite{Chinazzi2013,Degryse2010,Halaj2013,Minoiu2013} thus leading to an increasingly complex global network structure and linking different regional financial systems all over the globe. Moreover, the modern financial landscape is typically described by tiered structures usually referred to as \emph{core/periphery}. Empirical evidence of this fact for several countries is for example given in \cite{Boss2004} for Austria, \cite{Cont2013} for Brazil, \cite{Craig2014} for Germany, \cite{Fricke2015} for Italy, \cite{Veld2014} for the Netherlands, \cite{Langfield2014} for the United Kingdom and \cite{Alves2013} for the European interbank network.

Motivated by this clustered and tiered observed network structures, in Chapter \ref{chap:block:model} we propose a random graph model for default contagion in which each institution is given a certain type (such as country, core/periphery or mixtures thereof) assigning it to a certain subsystem. Our model constitutes a multi-variate extension of \cite{Detering2015a} and assigns to each institution a vector of different vertex-weights describing the different tendencies to develop edges to certain subsystems. By this approach the above explained assortativity can be achieved. Note that for calibration purposes the assignment to a certain subsystem poses an additional challenge but methods for community detection and the identification of core respectively periphery vertices can be used. See for example \cite{Blondel2008,Clauset2004,Copic2009,Fortunato2016,Zhang2014,Zhao2011} respectively \cite{Holme2005,Rombach2017}.

Moreover, in our proposed model we overcome an issue that persisted in the literature so far and regards the distribution of exposures. While certainly exposures between core institutions should be larger than an exposure between a core and a periphery institution intuitively, in the literature (also in Chapter \ref{chap:systemic:risk}) it has so far been assumed that the distribution of such an exposure only depends on the creditor bank. This allowed to decouple the actual exposures from the network configuration. In Chapter \ref{chap:block:model}, we develop new multi-variate tools to tackle this issue directly, further increasing the dimensionality of the problem.

Yet, we achieve to derive analogue results as in \cite{Detering2015a} thus giving an explicit formula for the final systemic damage caused by an initial shock and providing an understanding of network structures promoting or obstructing the propagation of distress by deriving sharp criteria for resilience and non-resilience. Due to the multi-dimensionality of the model parameters we have to deal with a multi-variate fixed point equation but it is worth to note that the problem stays low-dimensional compared to an EN-type problem. Using the criteria for resilience, we then describe a family of sufficient systemic capital requirements securing the system against small initial shocks. Furthermore, we integrate ideas of exposure modeling from Chapter \ref{chap:systemic:risk} into our new model to significantly reduce the dimensionality of the problem while at the same time keeping the advantage of counterparty-dependent exposures.

Finally, we provide applications of our model in particular analyzing the effect of several subsystems on the stability of the global network. Moreover, we demonstrate by a simple example that the possibility to model counterparty-dependent exposures is a necessary feature in models of financial networks in order not to underestimate the impact of contagion.

\paragraph{Contribution in Chapter \ref{chap:fire:sales}:} In Chapter \ref{chap:fire:sales}, we seize on the idea of an asymptotic modeling from the literature on default contagion and develop an asymptotic model for fire sales. The underlying model parameters are for each institution a vector of share holdings of different illiquid assets, its initial capital and an amount of initially incurred losses. Moreover, we choose to consider an abstract function that describes the sales strategy for each institution and depends on the incurred losses relative to an institution's initial capital. Drops of asset values shall be modeled by a certain price impact function. Our setup allows for a large degree of freedom regarding the choice of these functions.

\pagebreak
Using global statistics of the institutions' model parameters we are then able to describe the final state of the system by a multi-variate fixed point equation where each dimension represents one of the considered assets. An arising difficulty at this description, however, is the emergence of discontinuities in the underlying system of functions. In addition to the final systemic damage due to defaults the final state also includes information about the total number of shares of each asset that were sold during the fire sales process and the hence induced total price impact that is an important quantity as every investor in the asset -- even if not participating in the fire sales process -- will effectively  lose money by the devaluation of its portfolio (at least in the short term). This includes in particular private investors, the wider economy and even countries (e.\,g.~the government pension fund of Norway). 

Similar as in the setting of default contagion a notion of resilience is emerging in the limit of large financial systems and we characterize systems as stable or unstable. Moreover, we are again able to compute systemic capital requirements that ensure resilience of the system and by numerical simulations we demonstrate their positive effect for system stability over classical risk management policies according to Basel III \cite{BaselCommittee2011} for instance.

We then employ our new theory to investigate positive and negative effects of portfolio diversification adding to current research in \cite{Beale2011,Frey2018,Ibragimov2011,Wagner2010}.

\paragraph{Contribution in Chapter \ref{chap:fire:sales:default}:} In Section \ref{1:sec:systemic:risk}, we explained that financial distress is propagated through financial systems via several channels of contagion and we already analyzed in detail two of the main drivers -- \emph{default contagion} and \emph{fire sales} -- individually. Clearly the different effects can amplify each other, however, if considered simultaneously. Consider for example a fire sales process and assume that due to the dropped value of its asset portfolio some institution in the system defaults. Then as described before all creditors of this institution need to write off their respective interbank assets and thus incur further losses. This in turn may lead to further asset sales and devaluation of asset portfolios by price impact.

It is thus important to model the interplay of different contagion channels in order to get a better picture of their joint impact on systemic risk. This is the purpose of Chapter \ref{chap:fire:sales:default}. In particular, we combine our models of default contagion and fire sales from Chapters \ref{chap:block:model} and \ref{chap:fire:sales} and consider a simultaneous cascade process for heterogeneous and assortative networks. As before we take an asymptotic approach to the problem and describe the final systemic damage as well as the final number of sold asset shares by a multi-variate fixed point equation where the underlying functions reflect the joint impact of both contagion drivers.

In order to understand which structures of financial systems facilitate or prevent the propagation of financial distress we extend our previous results on resilience and non-resilience and we further derive sufficient systemic capital requirements to secure a system against the combined effects of default contagion and fire sales.

Moreover, we demonstrate in a specific example that in fact the two channels can significantly amplify each other which stresses the importance of integrated models in systemic risk research.

%

\cleardoublepage
\chapter{A Model for Default Contagion in Financial Networks}\label{chap:systemic:risk}

In this chapter, we propose a first model for default contagion in financial networks. Our main motivation is to represent the large degree of heterogeneity that constitutes one of the defining features of real financial networks (see \cite{Boss2004} for Austria or \cite{Cont2013} for Brazil for instance) but could not be described sufficiently by previous works such as \cite{Cont2016}. We build on the network model from \cite{Detering2015a} and represent the financial network by means of a random graph in particular. Compared to \cite{Detering2015a}, however, we augment the underlying probability space such that exposures between institutions as well as capitals of the institutions can be modeled by random variables. We can thus describe exposures of different size in the network which makes our model applicable to reasonable financial networks. Still the underlying random graph allows to construct asymptotic degree sequences with infinite second moment, thus reflecting the networks' strong heterogeneity.

In Section \ref{2:random:graph}, we then describe the default contagion process and we determine the final state of large financial systems hit by some initial shock. At this, in particular, rather than considering the fraction of finally defaulted institutions only (cf.~\cite{Cont2016,Detering2015a}), we introduce an index of systemic importance of each institution and we are thus able to determine the damage to the financial system itself but also to the wider economy or society. In Section \ref{2:sec:resilience} we go one step further and investigate the vulnerability of an a priori unshocked financial system to a small shock. We derive criteria for resilience and apply them to derive sharp capital requirements securing the system. In particular for the relevant case of degree sequences with infinite second moment, these regulations are considerably more restrictive than the ones derived in \cite{Cont2016}. We end this chapter with a series of simulation results in Section \ref{2:simulation:study} demonstrating the applicability of our asymptotic theory also for finite systems of a size typical for the real world. All the proofs are deferred to Section \ref{2:sec:proofs}.

\vspace*{-7pt}
\paragraph{My own contribution:} This chapter is in large parts adopted from \cite{Detering2016} and is thus joint work with Nils Detering, Thilo Meyer-Brandis and Konstantinos Panagiotou. I was significantly involved in the development of all parts of that paper and did most of the editorial work. In particular, I made major contributions to Lemma \ref{2:lem:f:continuous}, Theorem \ref{2:thm:asymp:2}, Proposition \ref{2:prop:convergence:speed}, Theorem \ref{2:thm:cont:res}, Theorem \ref{2:threshold:res}, Theorem \ref{2:thm:functional:nonres}, Theorem \ref{2:thm:Pareto:type}, Proposition \ref{2:prop:robust:capital:requirements} and Theorem \ref{2:cor:threshold:res}. Moreover, all simulations have been designed, implemented and interpreted by myself. Subsection \ref{2:ssec:weights:estimation} is my own work and only contained in an earlier version of \cite{Detering2016}. 

\section{Default Contagion on a Weighted, Directed Random Graph}\label{2:random:graph}
We shall present a stochastic model for a weighted, directed financial network. It will be based on the directed random graph model proposed in \cite{Detering2015a} (see Subsection \ref{2:ssec:special:case:threshold:model}) but complemented by edge weights. 
The main objective will be to assess the damage caused by default contagion asymptotically when the network size grows to infinity.

\subsection{Default Contagion and Systemic Importance}\label{2:ssec:default:contagion:systemic:importance}
We first describe the process of default contagion on a given (deterministic) financial network. If $n\in\N$ is the size of the network, we label the institutions (for simplicity called banks hereafter) by indices $i \in [n]$, where $[n]:=\{1,\dots,n \}$, and interpret them as vertices in a graph. If furthermore $e_{i,j}\in\R_{+,0}$ describes the exposure of bank $j$ to bank $i$, then we draw a directed edge of weight $e_{i,j}$ from $i$ to $j$ in the graph if $e_{i,j}>0$. We do not allow for self-loops or multiple edges between two vertices pointing in the same direction. That is, $e_{i,i}=0$ for all $i\in[n]$ and the network structure is completely determined by the exposure matrix $(e_{i,j})_{i,j\in[n]}$. Moreover, consider for each bank $i\in[n]$ its capital/equity $c_i\in\R_{+,0,\infty}:=\R_{+,0}\cup\{\infty\}$ and a value of \emph{systemic importance} $s_i\in\R_{+,0}$ which measures the potential damage caused by the default of bank $i$ and could, for instance, be computed by a regulating institution according to the indicator-based approach developed by the \textit{Basel Committee on Banking Supervision} in its framework text from 2013 in order to measure  global systemic importance of banks \cite{BaselCommittee2013}. 
The focus is on the impact of a potential default on the global financial system and wider economy. A 
similar approach is pursued by the \textit{Board of Governors of the Federal Reserve System} \cite{Fed2015}. 
Another example of deriving systemic importance values is to use \textit{DebtRank} as introduced in \cite{Battiston2012a}. It focuses on the relative monetary impact of a bank in an interbank network.

We call bank $i$ \emph{solvent} if $c_i>0$ and \emph{insolvent/defaulted} if $c_i=0$ (due to an exogenous shock to the network). The set of \emph{initially defaulted} banks is thus given by \mbox{$\mathcal{D}_0 = \{ i \in [n] \,:\, c_i = 0\}$.} They trigger a default cascade $ \mathcal{D}_0\subseteq  \mathcal{D}_1\subseteq ...$ given by 
\begin{equation}\label{2:eqn:default:cascade}
\mathcal{D}_k = \Bigg\{ i \in [n] \,:\, c_i \le \sum_{j\in \mathcal{D}_{k-1}} e_{j,i} \Bigg\},
\end{equation} 
where in each round $k\geq1$ of the cascade process bank $i$ has to write off its exposures to banks that defaulted in round $k-1$ and goes bankrupt as soon as its total write-offs exceed its initial capital. The chain of default sets clearly stabilizes after at most $n-1$ rounds and we call $\mathcal{D}_n=\mathcal{D}_{n-1}$ the {\em final default cluster} in the network induced by $\mathcal{D}_0$. We could easily introduce a constant recovery rate $R\in[0,1)$ to our model by multiplying exposures $e_{j,i}$ by a factor $1-R$ in \eqref{2:eqn:default:cascade}.

A first approach, that is often pursued in current literature, is to identify the damage caused to the financial network with the fraction $n^{-1}\vert\mathcal{D}_n\vert$. That is, damage is bearable if only few banks default as a result of the external shock event and the thereby started cascade process and it becomes the more threatening the larger the final fraction of defaulted banks $n^{-1}\vert\mathcal{D}_n\vert$ gets. In line with current regulator considerations, however, it is more realistic to consider the more general index of systemic importance of defaulted banks to really measure the damage to the economy. Instead of the size of the final default cluster $\mathcal{D}_n$, in the following we hence consider its total systemic importance $\mathcal{S}_n:=\sum_{i\in\mathcal{D}_n}s_i$ as a measure for the damage caused. Clearly, the special case $\mathcal{S}_n=\vert\mathcal{D}_n\vert$ is covered by setting $s_i=1$ for each $i\in[n]$.

\subsection[A Special Case: the Threshold Model]{A Special Case: the Threshold Model from \cite{Detering2015a}}\label{2:ssec:special:case:threshold:model}
Consider for now the special case that $e_{i,j}\in\{0,1\}$ and $c_i\in\N_{0,\infty}:=\N_0\cup\{\infty\}$. That is, whether or not a bank in the network defaults depends on the number of defaulted debtors (and the bank's individual integer-valued capital $c_i$). This setting (we call it the \emph{threshold model}) was considered in \cite{Detering2015a} and we recall its model assumptions and the main result here.

Instead of a deterministic network structure, we describe the network as a random graph. To this end, (in addition to capital $c_i$ and systemic importance $s_i$) assign to each vertex $i\in[n]$ two deterministic vertex-weights $w_i^-\in\R_{+,0}$ and $w_i^+\in\R_{+,0}$ and define the probability $p_{i,j}$ of a directed edge from vertex $i$ to vertex $j$ being present by
\begin{equation}\label{2:conn:prob}
p_{i,j}=\begin{cases}\min \{1,n^{-1}w^+_i w^-_j \},&i\neq j,\\0,& i=j.\end{cases}
\end{equation}
Further, let $X_{i,j}$ be the indicator function for the event of edge $(i,j)$ sent from vertex $i$ to vertex $j$ being present and assume that these events are independent for all $(i,j)\in[n]^2$. The role of in-weight $w^-_i$ respectively out-weight $w^+_i$ is to determine the tendency of vertex $i\in [n]$ to have incoming respectively outgoing edges. The vertex-weights are deterministic and purely used as a mean to specify the edge probabilities. They should not be confused with the edge-weights $e_{i,j}$. The construction of the random graph via vertex-weights resembles the one in \cite{ChungLu2002} or more general in \cite{Bollobas2007}. Note, however, that our random graph is a directed generalization as in \cite{Detering2015a}.

For each random configuration of the network, we can then consider the cascade process \eqref{2:eqn:default:cascade} to derive the random final default cluster $\mathcal{D}_n$ as well as its random systemic importance $\mathcal{S}_n$. The idea in the following is to let the network grow in a regular fashion (see Assumption \ref{2:ass:regularity}) and to use law-of-large-numbers effects in order to derive a deterministic limit for $n^{-1}\mathcal{S}_n$.

For each network size $n\in\N$ let $\bm{w}^-(n)=(w_1^-(n),\ldots,w_n^-(n))$, $\bm{w}^+(n)=(w_1^+(n),\ldots,w_n^+(n))$, $\bm{s}(n)=(s_1(n),\ldots,s_n(n))$ and $\bm{c}(n)=(c_1(n),\ldots,c_n(n))$ sequences of in-weights, out-weights, systemic importances and capitals of the individual banks. We impose the following regularity conditions:
\begin{assumption}\label{2:ass:regularity}
For each $n\in\N$, denote the joint empirical distribution function of $\bm{w}^-(n)$, $\bm{w}^+(n)$, $\bm{s}(n)$ and $\bm{c}(n)$ by
\[ F_n(x,y,v,l)=n^{-1}\sum_{i\in[n]}\1\{w_i^-(n)\leq x,w_i^+(n)\leq y,s_i(n)\leq v,c_i(n)\leq l\},\hspace{9pt} (x,y,v,l)\in\R_{+,}^3\times\N_{0,\infty}, \]
and let $(W_n^-,W_n^+,S_n,C_n)$ a random vector distributed according to $F_n$. We assume that:
\begin{enumerate}
\item\label{2:ass:regularity:1} \textbf{Convergence in distribution:} There exists a distribution function $F$ on $\R_{+,0}^3\times\N_{0,\infty}$ such that $F(x,y,v,l)=0$ for all $x,y\leq x_0$ and $x_0>0$ small enough, and such that at all points $(x,y,v,l)$ for which $F_l(x,y,v):=F(x,y,v,l)$ is continuous in $(x,y,v)$, it holds $\lim_{n\to\infty}F_n(x,y,v,l)=F(x,y,v,l)$. Denote by $(W^-,W^+,S,C)$ a random vector distributed according to $F$.
\item \textbf{Convergence of average weights and systemic importance:} $W^-$, $W^+$ and $S$ are integrable and $\E[W_n^-]\to\E[W^-]$, $\E[W_n^+]\to\E[W^+]$ as well as $\E[S_n]\to\E[S]$ as $n\to\infty$.
\end{enumerate}
\end{assumption}
This assumption is of a technical nature and concerned with the behavior of the network parameters as the size of the network tends to infinity. For practical purposes one can think of Assumption \ref{2:ass:regularity} ensuring that the limiting network keeps the observed parameter distribution of some real network we want to investigate. In particular, the expected weights are assumed to stay finite. In \cite{Detering2015a} it was derived that Assumption \ref{2:ass:regularity} implies $D_i^-\sim\mathrm{Poi}(w_i^-\E[W^+])$ respectively $D_i^+\sim\mathrm{Poi}(w_i^+\E[W^-])$ in the limit $n\to\infty$, where $D_i^-$ and $D_i^+$ denote the random in- respectively out-degree of vertex $i$ with weights $(w_i^-,w_i^+)$. Conversely, one can show that for an observed network topology, i.\,e.~given in- and out-degrees, maximum likelihood estimators of the in- and out-weights are approximately given by the in- and out-degrees (normalized by some global factor) -- see Subsection \ref{2:ssec:weights:estimation}. That is, morally one can think of the in- respectively out-weight of a vertex as its in- respectively out-degree.

Further, note that we did not assume $W^-$ or $W^+$ to have finite second moment. By the result from \cite{Detering2015a} that the empirical degree distribution for the model above converges weakly to a random vector $(D^-,D^+)$ distributed as $\mathrm{Poi}(W^-\E[W^+],W^+\E[W^-])$, we see that our model is hence capable of modeling networks without a second moment condition on their degree-sequences. In particular, choosing $W^-$ and $W^+$ power law distributed with parameters $\beta^-$ respectively $\beta^+$ results in power law distributions for the degrees $D^-$ and $D^+$ with the very same parameters. This allows to calibrate our model parameters to observed empirical in- and out-degree sequences. As we will see in Subsections \ref{2:ssec:threshold:requirements} and \ref{2:ssec:capital:requirements}, these power law parameters carry the most important information about the network when it comes to determining resilient capital requirements.

Consider now the following heuristics: Let $\zeta\in[0,\E[W^+]]$ denote the total out-weight of finally defaulted banks divided by $n$. Then in the limit $n\to\infty$ for any bank $i\in[n]$ the number of finally defaulted neighbors in the network is given by a random variable $\mathrm{Poi}(w_i^-\zeta)$. Bank $i$ is thus finally defaulted itself if and only if $\mathrm{Poi}(w_i^-\zeta)\geq c_i$. Summing over all banks in the network we therefore derive the following identity:
\[ \E[W^+\psi_C(W^-\zeta)] = \zeta, \]
where
\[ \psi_l(x) := \P\left(\mathrm{Poi}(x)\geq l\right) = \begin{cases}\sum_{j\geq l}e^{-x}x^j/j!,&0\leq l<\infty,\\0,&l=\infty.\end{cases} \]
Moreover, summing up the systemic importance values, the final damage caused by defaulted banks should be given by $\E[S\psi_C(W^-\zeta)]$.

Motivated by these heuristics consider now the function
\[ f(z;(W^-,W^+,C)):=\E\left[W^+\psi_C(W^-z)\right]-z. \]
By the dominated convergence theorem, $f(z;(W^-,W^+,C))$ is continuous and has a smallest root $\hat{z}\in[0,\E[W^+]]$. Furthermore, let
\[ d(z;(W^-,W^+,C)) := \E[W^-W^+\phi_C(W^-z)]-1, \]
the weak derivative of $f$ (see Lemma \ref{2:lem:f:continuous}), where
\[ \phi_l(x) := \P\left(\mathrm{Poi}(x)=l-1\right)\1\{l\geq1\}. \]

A sequence of events $(E_n)_{n\in\N}$ shall hold with high probability (w.\,h.\,p.) if $\P(E_n)\to1$, as $n\to\infty$. The following theorem for the threshold model will be used in the proofs of our main results in this chapter. 

\pagebreak
\begin{theorem}[adapted from {\cite[Theorem 7.2]{Detering2015a}}]\label{2:thm:threshold:model}
Consider a sequence of financial systems satisfying Assumption \ref{2:ass:regularity}. Then the following holds:
\begin{enumerate}
\item For all $\epsilon>0$ with high probability:
\[ n^{-1}\mathcal{S}_n \geq \E\left[S\psi_C(W^-\hat{z})\right] - \epsilon. \]
\item If $d(z)$ is bounded from above by some constant $\kappa<0$ on a neighborhood of $\hat{z}$, then
\[ n^{-1}\mathcal{S}_n \xrightarrow{p} \E\left[S\psi_C(W^-\hat{z})\right],\quad\text{as }n\to\infty. \]
\end{enumerate}
\end{theorem}
The theorem thus allows us to compute the final damage $\mathcal{S}_n$ for $n\to\infty$.

\subsection{The Exposure Model}\label{2:ssec:exposure:model}
On the base of the threshold model from the previous subsection, we will now construct our weighted, directed random graph model for financial systems. At this, we uncouple the occurrence of an edge sent from $i$ to $j$ from the size of its possible edge-weight. That is, we model the occurrence of edges by the random matrix $X=X(n)=(X_{i,j})_{i,j\in[n]}$ from Subsection \ref{2:ssec:special:case:threshold:model} and we assign to each pair $(i,j)\in[n]^2$ with $i\neq j$ a random variable $E_{i,j}>0$ representing $j$'s \emph{possible} exposure to $i$ (set $E_{i,i}=0$ for all $i\in[n]$) such that $E=E(n)=(E_{i,j})_{i,j\in[n]}$ is independent of $X$ (clearly one can construct such a joint probability space). The random exposure of $j$ to $i$ is then given by $e_{i,j}=X_{i,j}E_{i,j}$.

To make the model analytically tractable, we assume that for each bank $j$ the list of possible exposures $E_{1,j},\ldots,E_{j-1,j},E_{j+1,j},\ldots,E_{n,j}$ is an exchangeable sequence of random variables. That is, for each $j\in[n]$ and each permutation $\pi$ of  $[n]\backslash\{j\}$ 
\[ \left(E_{1,j},\ldots,E_{j-1,j},E_{j+1,j},\ldots,E_{n,j}\right) \stackrel{d}{=} \left(E_{\pi(1),j},\ldots,E_{\pi(j-1),j},E_{\pi(j+1),j},\ldots,E_{\pi(n),j}\right). \]
This is equivalent to taking for each bank $j\in[n]$ an arbitrary sequence of random variables $\tilde{E}_{1,j},\ldots,\tilde{E}_{j-1,j},\tilde{E}_{j+1,j},\ldots,\tilde{E}_{n,j}$ and transforming them into a list of exposures $E_{1,j},\ldots,E_{j-1,j}$,\linebreak $E_{j+1,j},\ldots,E_{n,j}$ by $E_{i,j}=\tilde{E}_{\pi(i),j}$ for some random permutation $\pi$ independent of $\{\tilde{E}_{i,j}\}_{i\in[n]\backslash\{j\}}$ and uniformly drawn from the set of all permutations of $[n]\backslash\{j\}$. We remark that the requirement of exchangeable exposures is a typical assumption made in the literature (such as \cite{Cont2016}). Note, however, that in this setting the distribution of the exposure size $E_{i,j}$ only depends on the creditor bank $j$ and not on the debtor bank $i$, which might be a criticizable assumption for example in strongly pronounced core/periphery networks where also the exposures might exhibit stronger heterogeneity. The relaxation of this assumption is technically more demanding and deferred to Chapter \ref{chap:block:model}.

Furthermore, assign to each bank $i\in[n]$ a possibly stochastic capital value $c_i\in\R_{+,0,\infty}$ and a deterministic systemic importance value $s_i\in\R_{+,0}$. Using \eqref{2:eqn:default:cascade} we can then again determine the random final default cluster $\mathcal{D}_n$ and its systemic importance $\mathcal{S}_n$. It is the aim of the following subsection to derive results about convergence and deterministic bounds similar as in Subsection \ref{2:ssec:special:case:threshold:model} for the threshold model.

Table \ref{2:tab:parameters} summarizes all important parameters in the exposure model and compares them to the observed quantities in a financial network.
\begin{table}
\begin{center}
\begin{tabular}{|c|c|}
\hline\textbf{Observed Network} & \textbf{Exposure Model}\\\hline
capital $c_i\in\R_{+,0}$, & capital $c_i\in\mathcal{L}^0(\R_{+,0,\infty})$,\\
systemic importance $s_i\in\R_{+,0}$, & systemic importance $s_i\in\R_{+,0}$,\\
in-degree $d_i^-\in\N_0$, & in-weight $w_i^-\in\R_{+,0}$,\\
out-degree $d_i^+\in\N_0$, & out-weight $w_i^+\in\R_{+,0}$,\\
&  (edge probability $p_{i,j}=\min\{1,n^{-1}w_i^+w_j^-\}\1\{i\neq j\}$)\\
exposure sequence & exchangeable sequence of possible edge weights\\
$(e_{i,j})_{j\in[n]\backslash\{i\}}\subset\R_{+,0}$ & $(E_{i,j})_{j\in[n]\backslash\{i\}}\subset\mathcal{L}^0(\R_+)$\\\hline
\end{tabular}
\caption{Comparison of observed quantities in a financial network and the model parameters in the exposure model}\label{2:tab:parameters}
\end{center}
\end{table}

\subsection{Asymptotic Results for Default Contagion in the Exposure Model}\label{2:ssec:asymptotic:results:exposure:model}
The setting in the exposure model is more complex than in the threshold model since we cannot decide if a bank defaults only based on the number of its neighbors that default: \eqref{2:eqn:default:cascade} asserts that this also depends on the actual exposures between the banks. However, one crucial assumption that we made is that these exposures are exchangeable, so intuitively it should make no difference which neighbors of a given bank default, but just their actual number. 

To formalize this intuitive argument, define for each bank $i\in [n]$ the random threshold value
\begin{equation}\label{2:ex1:perc:thres}
\tau_i(n):= 
\inf \Bigg\{ s \in \{0\}\cup[n-1] \,:\, \sum_{\ell\leq s} E_{\rho_i(\ell),i} \geq c_i \Bigg\},
\quad \text{where}\quad
\rho_i(\ell):=\ell+\1\{\ell\geq i\},
\end{equation}
with the usual convention $\inf\emptyset:=\infty$, that is, $\tau_i$ is allowed to take the value $\infty$ if capital $c_i$ is larger than the sum of all possible exposures. In this case, bank $i$ can never default. The use of the enumeration $\rho_i$ becomes necessary in (\ref{2:ex1:perc:thres}) since we want to spare $i$ in this natural ordering. The value $\tau_i$ then determines the \textit{hypothetical} default threshold of $i$, assuming that $i$'s neighbors default in the order of their natural index given by $\rho_i$ and that all edges $(j,i)$, $1\leq j\leq \rho_i(\tau_i)$, $i\neq j$, are present in the graph. We denote the hypothetical threshold sequence by $
\pmb{\tau}(n)=(\tau_1(n),\ldots,\tau_n(n))$. 
The thresholds are only hypothetical, since not all of the first $\rho_i(\tau_i)$ exposures must be present in the graph and the vertices do usually not default in their natural order. However, we know that the exposures are exchangeable, so all these simplifications should have no effect; it will turn out in the proof of Theorem~\ref{2:thm:asymp:1} that indeed the value $\tau_i$ captures the actual dynamics: the qualitative characteristics of the contagion process in the exposure model are the same as in the threshold model with capital sequence $\pmb{\tau}(n)$.

As an equivalent of Assumption \ref{2:ass:regularity} for the threshold model we need to impose the following regularity conditions:

\begin{assumption}\label{2:vertex:assump}
For each $n\in\N$, denote the random joint empirical distribution function of $\bm{w}^-(n)$, $\bm{w}^+(n)$, $\bm{s}(n)$ and $\pmb{\tau}(n)$ by
\[ G_n(x,y,v,l)=n^{-1}\sum_{i\in[n]}\1\{w_i^-(n)\leq x,w_i^+(n)\leq y,s_i(n)\leq v,\tau_i(n)\leq l\},\hspace{8pt} (x,y,v,l)\in\R_{+,0}^3\times\N_{0,\infty}. \]
Then we assume that:
\begin{enumerate}
\item \textbf{Almost sure convergence in distribution:} There exists a deterministic distribution function $G$ on $\R_{+,0}^3\times\N_{0,\infty}$ such that $G(x,y,v,l)=0$ for all $x,y\leq x_0$ and $x_0>0$ small enough, and such that at all points $(x,y,v,l)$ for which $G_l(x,y,v):=G(x,y,v,l)$ is continuous in $(x,y,v)$, it holds almost surely $\lim_{n\to\infty}G_n(x,y,v,l)=G(x,y,v,l)$. Denote by $(W^-,W^+,S,C)$ a random vector distributed according to $G$.
\pagebreak
\item \textbf{Convergence of average weights and systemic importance:} $W^-$, $W^+$ and $S$ are integrable and $\int_{\R_{+,0}^3\times\N_{0,\infty}} x\,\dd G_n(x,y,v,l) \to \E[W^-]$, $\int_{\R_{+,0}^3\times\N_{0,\infty}} y\,\dd G_n(x,y,v,l) \to \E[W^+]$ as well as $\int_{\R_{+,0}^3\times\N_{0,\infty}} v\,\dd G_n(x,y,v,l) \to \E[S]$ as $n\to\infty$.
\end{enumerate}
\end{assumption}
To ensure that Assumption \ref{2:vertex:assump}~holds, a twofold regularity is needed. Firstly, for a vertex with given in- and out-weight, the distribution of the threshold value must stabilize, even though the number of exposures appearing in the sum in (\ref{2:ex1:perc:thres}) increases. Secondly, a law of large numbers for the empirical distribution of the threshold values has to hold. See Subsection \ref{2:ssec:examples} for general examples of financial systems satisfying Assumption \ref{2:vertex:assump}.

For the remainder of this subsection, we consider a sequence of financial systems denoted as\linebreak $(\bm{w}^-(n),\bm{w}^+(n),\bm{s}(n),E(n),\bm{c}(n))$ and satisfying Assumption \ref{2:vertex:assump}. In particular, we denote by $(W^-,W^+,S,T)$ a random vector distributed according to the limiting distribution $G$ from Assumption \ref{2:vertex:assump}. We assume that the financial systems have experienced an external shock such that a positive fraction of banks have capital zero. In the notation from above this means $\P(T=0)>0$.  
Hence we are in a situation in which a default cascade is about to happen and we are interested in the damage to the financial system and the wider economy, given by $\mathcal{S}_n=\sum_{i\in\mathcal{D}_n}s_i$ the total systemic importance of defaulted banks after the contagion process. This damage $\mathcal{S}_n$ is a random number for each $n\in\N$. As the network size gets large, however, we show that $n^{-1}\mathcal{S}_n$ converges to a deterministic value which we can determine exactly. To this end, we denote
\[ f(z;(W^-,W^+,T)):=\E\left[W^+\psi_T(W^-z)\right]-z, \]
where as in Subsection \ref{2:ssec:special:case:threshold:model}
\[ \psi_l(x) := \P\left(\mathrm{Poi}(x)\geq l\right) = \begin{cases}\sum_{j\geq l}e^{-x}x^j/j!,&0\leq l<\infty,\\0,&l=\infty,\end{cases} \]
and
\[ d(z;(W^-,W^+,T)) := \E[W^-W^+\phi_T(W^-z)]-1, \]
where again as in Subsection \ref{2:ssec:special:case:threshold:model}
\[ \phi_l(x) := \P\left(\mathrm{Poi}(x)=l-1\right)\1\{l\geq1\}. \]
Whenever $(W^-,W^+,T)$ is clear from the context, we abbreviate $f(z;(W^-,W^+,T))$ by $f(z)$ and $d(z;(W^-,W^+,T))$ by $d(z)$. The following lemma summarizes some properties of $f$ and $d$. 
\begin{lemma}\label{2:lem:f:continuous}
The function $f(z)$ is continuous on $[0,\infty)$ and admits the following representation:
\begin{equation}\label{2:eqn:integral:representation}
f(z) = \E[W^+\1\{T=0\}] + \int_0^z d(\xi)\dd\xi
\end{equation}
If $\P(T=0)>0$, then $f(z)$ has a strictly positive root $\hat{z}$.
\end{lemma}
In particular, $d(z)$ is the weak derivative of $f(z)$ and if $d(z)$ is continuous on some interval $I\subset[0,\infty)$, then $f(z)$ is continuously differentiable on $I$ with derivative $d(z)$.

We derive the following result about $\mathcal{S}_n$, the damage caused to the system. 
It resembles Theorem \ref{2:thm:threshold:model} for the threshold model and indeed in the proof 
we make use of this result. However, due to the hypothetical nature of the threshold sequence $\pmb{\tau}(n)$ this application is not straight-forward and requires considerable effort.

\pagebreak
\begin{theorem}\label{2:thm:asymp:1}
Under Assumption \ref{2:vertex:assump}, suppose $\P(T=0)>0$ and let $\hat{z}$ be the smallest positive root of $f(z
)$. If the weak derivative $d(z)$ of $f(z)$ is bounded from above by some constant $\kappa<0$ on a neighborhood of $\hat{z}$, then
\[ n^{-1}\mathcal{S}_n \xrightarrow{p} \E [S \psi_T(W^-\hat{z}) ], \quad \text{ as } n\rightarrow \infty. \]
\end{theorem}
Two remarks are in order: First, if $f(z)$ is continuously differentiable on a neighborhood of $\hat{z}$ with $f'(\hat{z})<0$ (i.\,e.~$\hat{z}$ stable), then Theorem \ref{2:thm:asymp:1} is applicable. This is a standard assumption in current literature. In \cite{Cont2016}, for instance, the authors assume degree sequences of finite second moment. In this case, it is straight forward to show that $f(z)$ is continuously differentiable. Secondly, without the assumption of stableness, $n^{-1}\mathcal{S}_n$ does not converge to a deterministic number in general (see \cite{Janson2012} for a comparable result in a much simpler setting). However, in the following theorem we are still able to state asymptotic bounds rather than an exact limiting value. We believe that the derived bounds are sharp in the sense that they cannot be improved without further assumptions. Proving this, however, is beyond the scope of this work.

\begin{theorem}\label{2:thm:asymp:2}
Under Assumption \ref{2:vertex:assump}, suppose $\P(T=0)>0$ and let $\hat{z}$ be the smallest positive root of $f(z)$. Further, let $z^*$ be the smallest value of $z>0$ at which $f(z)$ crosses zero,
\[ z^* := \inf\left\{z>0\,:\,f(z)<0\right\}. \]
Then the following holds:
\begin{enumerate}
\item\label{2:thm:asymp:2:1} For all $\epsilon>0$, with high probability $n^{-1}\mathcal{S}_n \geq  \E [S \psi_T(W^-\hat{z}) ]-\epsilon$.
\item\label{2:thm:asymp:2:2} If further $d(z)$ is continuous on some neighborhood of $z^*$, then for all $\epsilon>0$ with high probability $n^{-1}\mathcal{S}_n \leq  \E [S \psi_T(W^-z^*) ]+\epsilon$.
In particular, if $\hat{z}=z^*$, then
\[ n^{-1}\mathcal{S}_n \xrightarrow{p} \E [S \psi_T(W^-\hat{z}) ], \quad \text{ as } n\rightarrow \infty. \]
\end{enumerate}
\end{theorem}
Note that under the assumptions of Theorem \ref{2:thm:asymp:2} it holds that $f(0)>0$. Together with continuity of $f$ from Lemma \ref{2:lem:f:continuous} this implies that $z^*\geq\hat{z}$. In general, however, it is possible that $z^*>\hat{z}$, for example if $f$ has a local minimum at $\hat{z}$.

\subsection{Examples for Financial Systems Satisfying Assumption \ref{2:vertex:assump}}\label{2:ssec:examples}
In this section, we want to demonstrate the wide applicability of our model. In Example \ref{2:example:1}, we describe a financial system where banks are sorted into buckets according to their weight and capitals as well as exposures have the same distributions for all banks within each bucket.

\begin{example}\label{2:example:1}
Let $({\bf w^-}(n), {\bf w^+}(n), {\bf s}(n))$ be a triple consisting of in-weight, out-weight and systemic importance sequences such that the empirical distribution
\[ \tilde{F}_n(x,y,v)=n^{-1}\sum_{i\in[n]}\1\left\{w_i^-(n)\leq x,w_i^+(n)\leq y, s_i\leq v\right\},\quad(x,y,v)\in\R_{+,0}^3, \]
converges to some distribution function $\tilde{F}$. To be consistent with Assumption \ref{2:vertex:assump}, we require $\tilde{F}(x,y,v)=0$ for all $x,y\leq x_0$ and some $x_0>0$.  Further, assume $\lim_{n\to\infty}\E[(W_n^-,W_n^+,S_n)]=\E[(W^-,W^+,S)]$, where $(W_n^-,W_n^+,S_n)\sim\tilde{F}_n$ and $(W^-,W^+,S)\sim\tilde{F}$. Choose some partition of $[0,\infty)^3$ into countably many sets $D_k$, $k\in \N$, and denote \mbox{$\mathcal{W}_k:=\{i\in[n]\,:\, (w_i^-,w_i^+,s_i)\in D_k\}$.} Let the distributions of $E_{j,i}$, $j\in[n]$, and $c_i$ be equal across vertices $i \in \mathcal{W}_k$ and assume them all to be independent. From this assumption it follows that for $n\geq l\in\N_0$ the 
distribution of $Y_i^l:=\1\{\tau_i\leq l\}$ only depends on the category $\mathcal{W}_k\ni i$. Therefore, denote $q_k^l:=\P(Y_i^l=1)$ for $i\in\mathcal{W}_k$.

We show that Assumption~\ref{2:vertex:assump} is satisfied. For this let $(x,y,v)\in \mathbb{R}_{+,0}^3$ a point of continuity of $\tilde{F}$ and define
\[ \mathcal{W}^{(x,y,v)}:=\{i \in [n] \,:\, w^-_i \leq x, w^+_i \leq y, s_i\leq v \}. \]
The cardinality of $\mathcal{W}^{(x,y,v)}\cap\mathcal{W}_k$ stabilizes as a fraction of $n$ for all $k$ by the regularity of the weight sequences. Since clearly $Y^l_i\in\mathcal{L}^1$, this yields that 
\begin{align*}
G (x,y,v,l) := \lim_{n\rightarrow \infty }n^{-1}\sum_{i \in \mathcal{W}^{(x,y,v)}} \E[ Y_i^l] &= 
\lim_{n\rightarrow \infty }n^{-1}\sum_{k\in\N}\abs{\mathcal{W}^{(x,y,v)}\cap\mathcal{W}_k}q_k^l
\end{align*}
exists. By the strong law of large numbers applied to each category, it follows that almost surely
\[ \lim_{n\rightarrow \infty} G_n(x,y,v,l)=\lim_{n\rightarrow \infty }n^{-1}\sum_{i \in \mathcal{W}^{(x,y,v)}} Y_i^l=G (x,y,v,l),\quad(x,y,v,l)\in\R_{+,0}^3\times\N_0. \]
Note that, although dropped from the notation, the random variables $\{Y_i^l\}_{i \in \mathcal{W}^{(x,y,v)}}$ depend on $n$. Therefore, we need a strong law of large numbers for independent triangular arrays here. Since $Y_i^l\in\mathcal{L}^4$, the standard proofs of the strong law of large numbers using the generalized Chebyshev's inequality and the Borel-Cantelli Lemma carry over to this case.

Finally, complete $G(x,y,v,l)$ also for points $(x,y,v)\in\R_{+,0}^3$ which are not continuity points of $\tilde{F}$ simply by taking limits from above.
\end{example}
As already observed in a similar setting in \cite{Cont2016}, the independence of the exposure random variables 
can be weakened.
\begin{example}\label{2:example:2}
Similar as above, assume that vertices are partitioned into the $K$ classes $\mathcal{W}_1,\dots,\mathcal{W}_K$  with vertices with the same marginal distributions of the capitals and exposures. The sets may depend on the network size $n$ but we shall assume that $\lim_{n\rightarrow \infty} n^{-1} \abs{\mathcal{W}_{k}} =:\lambda (k)$, $k\in[K]$, i.\,e.~the fraction of vertices of a given class stabilizes. 
For each class $k \in [K]$ we are given generating sequences $\{c_l^k\}_{l\in\N}$ and $\{E_l^k\}_{l\in\N}$ of random variables in $\R_{+,0}$ resp.~$\R_+$. Further, we assume that $\{c_l^k\}_{l\in\N}$ and $\{E_l^k\}_{l\in\N}$ are infinite exchangeable systems and independent of each other for every $k\in[K]$ 
(see for example \cite{Aldous1985} for the definition of exchangeability and infinite exchangeable systems). For each network size $n$, assign now to every vertex $i\in\mathcal{W}_k$ a capital from $\{c_l^k\}_{l\in\N}$ and $n-1$ exposures from $\{E_l^k\}_{l\in\N}$ according to any deterministic rule such that no capital or exposure is used more than once. 
Define then the threshold value as in (\ref{2:ex1:perc:thres}) and for a fixed $m\in\N$ and vertex $i$ in class $k\in [K]$ the indicator random variable $Y^m_{k,i}:=\1\{\tau_i = m\}$, that determines whether vertex $i$ has threshold value $m$. Observe that for $n\geq m+1 $ every vertex $i\in [n]$ has more than $m$ exposures and the distribution of $Y^m_{k,i}$ is thus independent of $n$.
Let $\beta_k (1),\dots , \beta_k (\abs{\mathcal{W}_k} )$ the indices of the vertices in $\mathcal{W}_k$. By construction then  
\[ \mathcal{L} ( Y^m_{k,\beta_k (1)}, \dots ,  Y^m_{k,\beta_k (\abs{\mathcal{W}_k} } ) = \mathcal{L} ( Y^m_{k,\sigma_k (\beta_k (1))}, \dots ,  Y^m_{k,\sigma_k ( \beta_k (\abs{\mathcal{W}_k} ) }) \]
for all $\sigma_k\in \Sigma (\mathcal{W}_k)$, that is, for each $k\in [K]$, the random variables $\{Y^m_{k,i}\}$ build an exchangeable system. 
Since for fixed $n$ the sequence $\{Y^m_{k,i} \}_{i\in \mathcal{W}_k}$ is just the restriction to a finite subset of variables of an infinite exchangeable system for $\abs{\mathcal{W}_k} \rightarrow \infty$ it converges in law to an infinite exchangeable system. Let $\Delta_k$ be its directing measure. 
This implies that the system of random variables $(Y^m_{k,i})_{k\in [K],i\in \mathcal{W}_k}$ forms a multi-exchangeable system (see \cite{Graham2008} for definition).
Define the empirical measure by
\[ \Lambda^m_k := \abs{\mathcal{W}_k}^{-1} \sum_{i=1}^{\abs{\mathcal{W}_k}} \delta_{Y^m_{k,\beta_k (i)}} \]
for each $k \in [K]$. By \cite[Theorem 2]{Graham2008} convergence in distribution of $\{Y^m_{k,i} \}_{i\in \mathcal{W}_k}$ implies convergence in distribution of the empirical measure sequence $(\Lambda_k)_{k\in [K]}$, without any assumptions on the dependency structure across classes. Since the above considerations apply for all $m\in \mathbb{N}$, convergence in distribution of the empirical measure sequence 
\[ \Lambda_k := n^{-1} \sum_{i \in [n]} \delta_{\tau_i} \]
follows for all $k \in [K]$. By the Skorohod Coupling Theorem \cite[Theorem 4.30]{Kallenberg2001}, there exists a probability space with random elements $\{ \tilde{\Lambda}_k \}_{k \in [K]}$ distributed as $\{ \Lambda_k \}_{k \in [K]}$ such that $\{ \tilde{\Lambda}_k \}_{k \in [K]}$ converges almost surely as required.
\end{example}

\section{Resilient Networks and Systemic Capital Requirements}\label{2:sec:resilience}
In the previous section, we quantified the default propagation in financial networks after an external shock to banks' capitals, i.\,e.~in a network with initially defaulted banks. It is of interest, however, to be able to determine how systemically risky a network is prior to a shock event. That is, for a sequence of financial systems $(\bm{w}^-(n),\bm{w}^+(n),\bm{s}(n),E(n),\bm{c}(n))$ satisfying Assumption \ref{2:vertex:assump} with $\P(T>0)=1$ we want to observe today's network topology and exposures and, keeping them unchanged, apply some small random shock to the capitals only ex post. A \textit{resilient}, systemically unrisky network should only experience minor damage by this whereas in \textit{non-resilient}, systemically risky networks even a small shock can cause huge harm to the whole system. An advantage over static models such as the Eisenberg-Noe model \cite{Eisenberg2001} is that we can assess stability already for an unshocked system. Further, this section will show that whether a financial network is judged resilient or non-resilient only depends on the distributions of $W^-$, $W^+$ and $T$. These have been shown to be relatively stable over time even if locally the network might change noticeably.
\subsection{Resilience Criteria for Unshocked Networks}
In order to incorporate such small random shocks into our model, we introduce a sequence $\bm{m}(n)=(m_1(n),\ldots,m_n(n))$ of binary marks $m_i\in\{0,1\}$ to $(\bm{w}^-(n),\bm{w}^+(n),\bm{s}(n),E(n),\bm{c}(n))$, where $m_i=0$ means that bank $i$ defaults ex post due to some shock event and hence loses all its capital to start the cascade process, i.\,e.~the new capital is given by $c_im_i$ for each bank $i\in[n]$. Otherwise, the capital distribution stays the same. 
We extend Assumption \ref{2:vertex:assump} such that there exists a distribution $\overline{G}$ \mbox{and the new empirical distribution}
\[ \overline{G}_n(x,y,v,l,k) = n^{-1}\sum_{i\in[n]}\1\{w_i^-(n)\leq x,w_i^+(n)\leq y, s_i\leq v, \tau_i(n)\leq l, m_i(n)\leq k\} \]
converges almost surely at all continuity points $(x,y,v)$ of $\overline{G}_{l,k}(x,y,v):=\overline{G}(x,y,v,l,k)$ and denote by $(W^-,W^+,S,T,M)$ a random vector distributed according to 
$\overline{G}$. We assume that $\P(T=0)=0<\P(M=0)$ such that indeed $M$ causes ex post defaults in an unshocked system.

We want to consider a financial system as being non-resilient to initial shocks if even very small shocks $M$ can cause significant damage $\mathcal{S}_n^M$, measured by the total systemic importance of the defaulted banks $\mathcal{D}_n^M$ at the end of the contagion process triggered by $M$. Mathematically this is expressed as follows:

\begin{definition}\label{2:def:non:resilience}
A financial system is said to be \emph{non-resilient} if there exists a constant $\Delta>0$ such that for each ex post default $M$ with $\P(M=0)>0$ it holds that
\[ n^{-1}\mathcal{S}_n^M \geq \Delta\quad\text{w.\,h.\,p.} \]
\end{definition}
The following theorem states a sufficient criterion for a system to be non-resilient. 
\begin{theorem}[Non-resilience Criterion]\label{2:prop:nonres}
Under Assumption \ref{2:vertex:assump} suppose that $\P(T=0)=0$ and that there exists $z_0>0$ such that
\begin{equation}\label{2:condition:nonres}
f(z)>0,\quad\text{for all }0<z<z_0.
\end{equation}
Then it holds for all $M$ with $\P(M=0)>0$ that
\[ n^{-1}\mathcal{S}_n^M \geq \E\left[S\psi_T(W^-z_0)\right] \quad\text{w.\,h.\,p.} \]
In particular, if $\E[S\1\{T<\infty\}]>0$, then the system is non-resilient.
\end{theorem}
The proof of Theorem~\ref{2:prop:nonres} follows from Part \ref{2:thm:asymp:2:1}.~of Theorem \ref{2:thm:asymp:2} and arguments analogue to ones used in \cite[Theorem 7.3]{Detering2015a} and is thus omitted here.

We can interpret Theorem \ref{2:prop:nonres} as follows: If a financial network satisfies condition (\ref{2:condition:nonres}), then no matter how small the fraction of banks which are driven to bankruptcy by an external shock event, after the cascade process of defaults always a damage larger than the constant $\E\left[S\psi_T(W^-z_0)\right]$ is caused to the system. In the reasonable case that $\E[S\1\{T<\infty\}]>0$, this lower bound for the damage is strictly positive and the system is hence non-resilient according to Definition \ref{2:def:non:resilience}. In particular, by choosing $s_i=1$ for all $i\in[n]$ and hence $S\equiv 1$, we derive that the final default fraction $n^{-1}\vert\mathcal{D}_n^M\vert$ is lower bounded by the constant $\E\left[\psi_T(W^-z_0)\right]$, which is positive (unless $\P(T=\infty)=1$).

\begin{figure}[t]
    \hfill\subfigure[]{\includegraphics[width=0.4\textwidth]{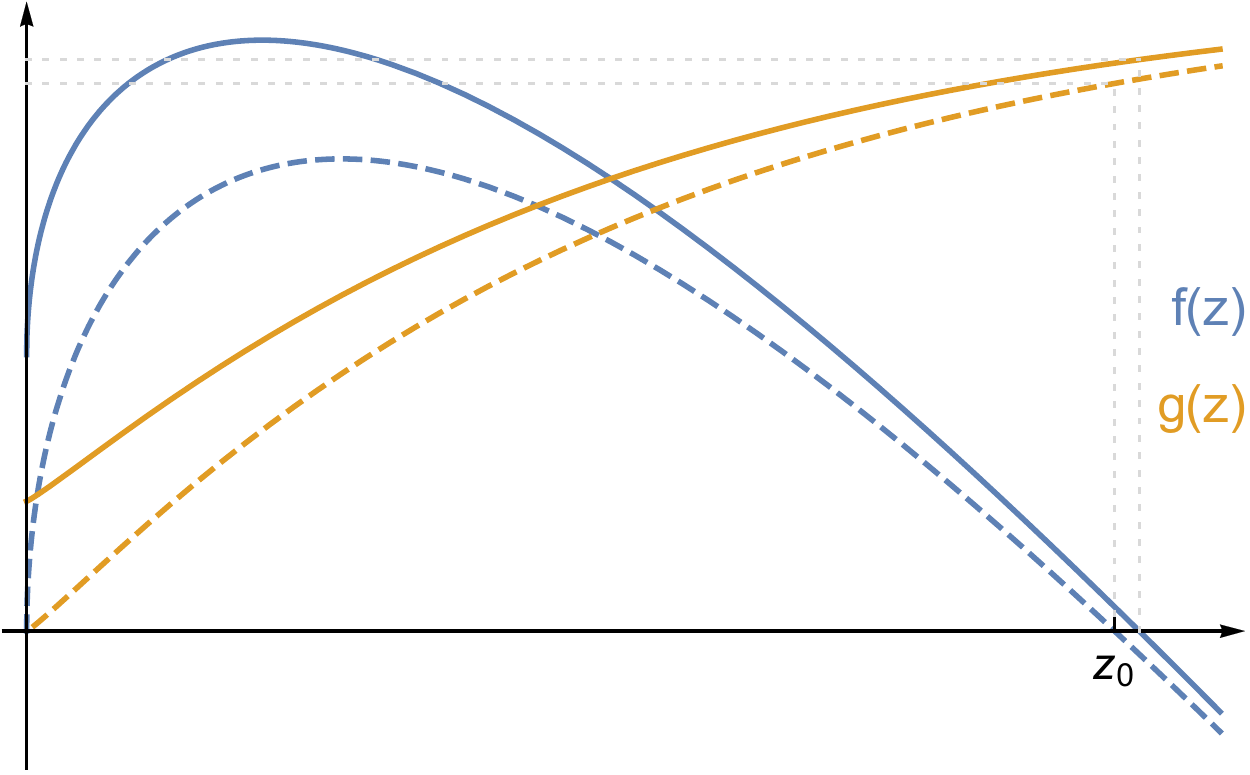}\label{2:fig:condition:nonres}}
    \hfill\subfigure[]{\includegraphics[width=0.4\textwidth]{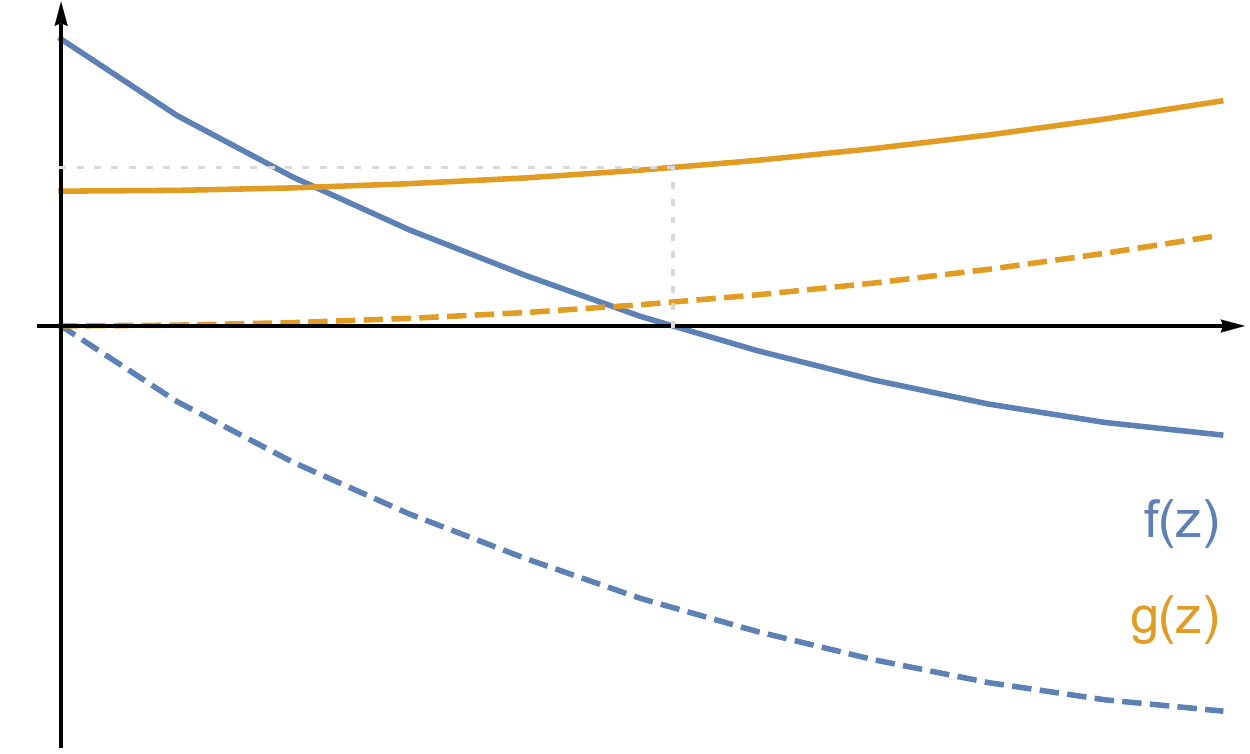}\label{2:fig:condition:res}}\hfill
\caption{Examples of functions $f(z)=\E[W^+\psi_T(W^-z)]-z$ (blue) satisfying conditions (\ref{2:condition:nonres}) (a) respectively (\ref{2:condition:res:2}) (b). Orange: the function $g(z)=\E[S\psi_T(W^-z)]$. Dashed: the unshocked functions. Solid: the shocked functions.}\label{2:fig:conditions:nonres:res}\hfill
\vspace{-8mm}
\end{figure}

Condition (\ref{2:condition:nonres}) is an assumption on $f$ which is illustrated in Figure \ref{2:fig:condition:nonres}. Whereas for $\P(M=0)=0$ the first non-negative root of the function is zero, any howsoever small increase in $\P(M=0)$, and hence upwards shift of $f(z)$, makes the first root jump above $z_0$ and causes default of a set of size larger than $n\E[\psi_T(W^-z_0)]$ and systemic importance larger than $n\E[S\psi_T(W^-z_0)]$.

If on the other hand function $f(z;(W^-,W^+,T))$ is such behaved that the first positive root $\hat{z}^M$ of $f(z;(W^-,W^+,TM))$ tends to zero as $\P(M=0)$ becomes smaller, one can expect that also the final default cluster $\mathcal{D}_n^M$ and its systemic importance $\mathcal{S}_n^M$ vanish and the system can hence be regarded as \textit{resilient} to small shocks. See Figure \ref{2:fig:condition:res} for an exemplary illustration. This intuition is formalized in the following definition, theorems and proposition. 
\begin{definition}\label{2:def:resilience}
A financial system is said to be \emph{resilient} if for each $\epsilon>0$ there exists $\delta>0$ such that
\[ n^{-1}\mathcal{S}_n^M < \epsilon\quad\text{w.\,h.\,p.~for all }M\text{ such that }\P(M=0)<\delta. \]
\end{definition}
In words this means that the final damage to the system $\mathcal{S}_n^M$ can be controlled by the initial default fraction $\P(M=0)$. Theorem \ref{2:prop:res} is then the analogue of \cite[Theorem 7.4]{Detering2015a} transferred to our exposure model.

\begin{theorem}[Resilience Criterion]\label{2:prop:res}
Under Assumption \ref{2:vertex:assump} suppose $\P(T=0)=0$ and that there exists $z_0>0$ such that
\begin{equation}\label{2:condition:res:2}
d(z)<0,\quad\text{for all }0<z<z_0.
\end{equation}
Then for any sequence of ex post defaults $\{M_i\}_{i\in\N}$ with $\lim_{i\to\infty}\P(M_i=0)=0$, it follows that for any $\epsilon>0$, there exists $i_\epsilon$ such that
\[ n^{-1}\mathcal{S}_n^{M_i} \leq \epsilon \quad \text{w.\,h.\,p. for all } i\geq i_\epsilon. \]
In particular, the system is resilient.
\end{theorem}
Theorem \ref{2:prop:res} states that the systemic importance of all finally defaulted banks tends to zero as the initial default fraction tends to zero, which is exactly our definition of resilience. However, it makes no statement about the rate of convergence. If we assume not only that $d(z) < 0$ for $z$ small enough but even $\limsup_{z\to0+}d(z) < 0$, then we derive the following result concerning convergence speed. 

\begin{proposition}\label{2:prop:convergence:speed}
Under Assumption \ref{2:vertex:assump} suppose $\P(T=0)=0$,
\[ \kappa:=\limsup_{z\to0+}d(z)<0\quad\text{and}\quad \kappa_S:=\limsup_{z\to0+}\E\left[W^-S\phi_T(W^-z)\right]<\infty. \]
Then for any sequence $\{M_i\}_{i\in\N}$ with $\lim_{i\to\infty}\P(M_i=0)=0$, it follows that w.\,h.\,p.
\[ n^{-1}\mathcal{S}_n^{M_i} \leq \E[S\1_{\{M_i=0\}}] - \kappa^{-1}\kappa_S\E[W^+\1_{\{M_i=0\}}] + o(\E[W^+\1_{\{M_i=0\}}]). \]
If $f(z)$ and $\E[S\psi_T(W^-z)]$ are continuously differentiable from the right at $z=0$ with derivatives $\kappa<0$ and $\kappa_S<\infty$, then we derive
\[ n^{-1}\mathcal{S}_n^{M_i} \xrightarrow{p} \E[S\1_{\{M_i=0\}}] - \kappa^{-1}\kappa_S\E[W^+\1_{\{M_i=0\}}] + o(\E[W^+\1_{\{M_i=0\}}]). \]
\end{proposition}
In particular, if $\{M_i\}_{i\in\N}$ is independent of $W^+$ and $S$, then w.\,h.\,p.
\[ n^{-1}\mathcal{S}_n^{M_i} \leq \P(M_i=0)\left(\E[S]-\kappa^{-1}\kappa_S\E[W^+]\right) + o(\P(M_i=0)) = \mathcal{O}(\P(M_i=0)) \]
and $1-\kappa^{-1}\kappa_S\E[W^+]/\E[S]$ can be regarded as the maximal amplification factor of the systemic importance of initially defaulted banks $\E[S\1\{M_i=0\}]=\P(M_i=0)\E[S]$. If further $S\equiv1$ and $W^-W^+$ is integrable, above result is the analogon to \cite[Corollary 20]{Amini2014c}:
\[ n^{-1}\vert\mathcal{D}_n^{M_i}\vert \xrightarrow{p} \P(M_i=0)\left(1+\frac{\E[W^+]\E[W^-\1\{T=1\}]}{1-\E[W^-W^+\1\{T=1\}]}\right) + o(\P(M_i=0)) \]
Both Theorem \ref{2:prop:res} and Proposition \ref{2:prop:convergence:speed} are concerned with the behavior of the weak derivative $d(z)$ of $f(z)$ near $z=0$. The following criterion that rather focuses on the behavior of $f(z)$ near $z=0$ will turn out to be useful later. 
\begin{theorem}\label{2:thm:cont:res}
Under Assumption \ref{2:vertex:assump} suppose $\P(T=0)=0$, $d(z)$ to be continuous on $(0,z_0)$ for some $z_0>0$ and
\begin{equation}\label{2:thm:cont:res:ass}
\inf\left\{z>0\,:\,f(z)<0\right\} = 0.
\end{equation} 
Then for any sequence $\{M_i\}_{i\in\N}$ with $\lim_{i\to\infty}\P(M_i=0)=0$, it follows that for any $\epsilon>0$, there exists $i_\epsilon$ such that
\[ n^{-1}\mathcal{S}_n^{M_i} \leq \epsilon \quad \text{w.\,h.\,p. for all } i\geq i_\epsilon. \]
In particular, the system is resilient.
\end{theorem}
Assumption (\ref{2:thm:cont:res:ass}) describes that $f(z)$ becomes negative immediately after $z=0$. It is in some sense the opposite of assumption (\ref{2:condition:nonres}) and ensures that the roots $z^*_i$ (analogue as in Theorem \ref{2:thm:asymp:2} but for the shocked systems) tend to zero as the shock size $\P(M_i=0)$ shrinks to zero.

\subsection{Systemic Threshold Requirements}\label{2:ssec:threshold:requirements}
A natural problem that is also of highest interest to regulators is to identify capital requirements for the individual banks which can be determined from observable quantities of the network and that are sufficient to make the network resilient to external shocks. Observable quantities are the in- and out-degrees $(d_i^-)_{i\in[n]}$ respectively $(d_i^+)_{i\in[n]}$, which function as estimators of the in- and out-weights $(w_i^-)_{i\in[n]}$ respectively $(w_i^+)_{i\in[n]}$, and interbank exposures. In this section we will first focus on identifying threshold requirements in the threshold model (see Subsection \ref{2:ssec:special:case:threshold:model})
that guarantee resilience. In the next section, we then discuss how to translate the threshold requirements into systemic capital requirements in the exposure model.

More precisely, in this section we seek threshold requirements for bank $i$ of the form \mbox{$\tau_i=\tau(w_i^-)$,} where $\tau:\R_{+,0}\to\N$ is a non-decreasing function. Such a functional form has the interpretation that the threshold (capital) requirement of a bank only depends on its risk of defaulting due to default of debtors (exposure risk). In contrast, if bank $i$'s threshold (capital) requirement $\tau(w_i^-,w_i^+)$ was also depending on the out-weight $w_i^+$, this would also take possible defaults caused by bank $i$ into account. This risk management policy would not be in line with traditional risk management techniques and would certainly be harder to communicate to the banks.

Note, in particular, the changed nature of the threshold values $\tau_i$. While, so far, the capitals $c_i$ (which equal $\tau_i$ in the threshold model) were exogenous quantities in our model, we now aim to determine them from the in-weights $w_i^-$ under the constraint of building a resilient network, hence making them endogenous quantities.

To investigate resilience of a financial system implementing threshold requirements given by $\tau$ we want to use the resilience criteria from the previous section. In particular we have to ensure that Assumption \ref{2:vertex:assump} is satisfied for the functional thresholds $\tau_i=\tau(w_i^-)$, i.\,e.~we need
\[ \lim_{n\to\infty}n^{-1}\sum_{i\in[n]}\1\{w_i^-(n)\leq x,w_i^+(n)\leq y,s_i(n)\leq v, \tau(w_i^-(n))\leq l\} = G(x,y,v,l) \]
for some distribution $G:\R_{+,0}^3\times\N_{0,\infty}\to[0,1]$ and all points $(x,y,v,l)\in\R_{+,0}^3\times\N_{0,\infty}$ for which $G_l(x,y,v):=G(x,y,v,l)$ is continuous. Note that depending on the choice of $\tau$ and $G$ for the limiting random vector $(W^-,W^+,S,T)\sim G$ it does not necessarily hold that \mbox{$\P(T=\tau(W^-))=1$.} This is because $W^-$ could have positive mass at some point of discontinuity of $\tau$ and it would then be important whether the in-weight distributions converge from below or from above. Instead one easily derives that $\P(\accentset{\circ}{\tau}(W^-)\leq T\leq \bar{\tau}(W^-))=1$ where $\accentset{\circ}{\tau}(w):=\lim_{\epsilon\to0+}\tau((1-\epsilon)w)$ and $\bar{\tau}(w):=\lim_{\epsilon\to0+}\tau(1+\epsilon)w)$ are the left-continuous resp.~right-continuous modifications of $\tau$. If, however, $\tau$ only admits discontinuities at $\tilde{w}\in\R_{+,0}$ such that $\P(W^-=\tilde{w})=0$, then in fact $\P(T=\tau(W^-))=1$ and we will assume this from now on. All our results on resilience and non-resilience in the following can easily be extended for the functions $\accentset{\circ}{\tau}$ resp.~$\bar{\tau}$.
\begin{assumption}\label{2:ass:tau}
Consider sequences $\mathbf{w}^-(n)$, $\mathbf{w}^+(n)$ and $\mathbf{s}(n)$ of in-weights, out-weights and systemic importance values such that their empirical distribution and mean converge to those of a random vector $(W^-,W^+,S)$. Moreover, let $\tau:\R_{+,0}\to\N_0$ be a non-decreasing function and assume that its points of discontinuity are all null-sets of $W^-$. In particular, letting $\tau_i(n)=\tau(w_i^-(n))$, $i\in[n]$, Assumption \ref{2:vertex:assump} is satisfied and it holds $T=\tau(W^-)$ a.\,s.
\end{assumption}
Empirical studies of financial networks such as \cite{Boss2004} or \cite{Cont2013} show that degrees follow Pareto distributions (at least in the tail). We denote in the following $X\sim\mathrm{Par}(\beta,x_\text{min})$, $\beta>1$, $x_\text{min}>0$, if the random variable $X$ has density
\[ f_X(x)=(\beta-1)x_\text{min}^{\beta-1}x^{-\beta}\1\{x\geq x_\text{min}\}. \]
As mentioned before, to reproduce Pareto distributed degrees in our model we need to choose the weights Pareto distributed as well. Hence let $W^-\sim\mathrm{Par}(\beta^-,w_\text{min}^-)$ and $W^+\sim\mathrm{Par}(\beta^+,w_\text{min}^+)$, where $\beta^->2$, $\beta^+>2$, $w_\text{min}^->0$ and $w_\text{min}^+>0$. In particular, any choice of an increasing function $\tau$ will satisfy Assumption \ref{2:ass:tau}. The main result of this section then identifies a criterion for function $\tau$ that ensures resilience of the financial network. 
\begin{theorem}\label{2:threshold:res}
Suppose Assumption \ref{2:ass:tau} for a non-decreasing function $\tau:\R_{+,0}\to\N\backslash\{0,1\}$ such that for each bank $i\in[n]$ the threshold value $\tau_i$ depends on in-weight $w_i^-$ by $\tau_i=\tau(w_i^-)$. Moreover, assume for the limiting weight distributions that $W^-\sim\mathrm{Par}(\beta^-,w_\text{min}^-)$ respectively $W^+\sim\mathrm{Par}(\beta^+,w_\text{min}^+)$, $\beta^-,\beta^+>2$, $w_\text{min}^-,w_\text{min}^+>0$. Set
\[ \gamma_\text{\normalfont c}:=2+\frac{\beta^--1}{\beta^+-1}-\beta^-\quad\text{and}\quad \alpha_\text{\normalfont c}:=\frac{\beta^+-1}{\beta^+-2}w_\text{\normalfont min}^+\left(w_\text{\normalfont min}^-\right)^{1-\gamma_\text{\normalfont c}}. \]
Then the system is resilient if one of the following holds:
\begin{enumerate}
\item \label{2:threshold:res:1} $\gamma_\text{\normalfont c}<0$
\item \label{2:threshold:res:2} $\gamma_\text{\normalfont c}=0$ and $\liminf_{w\to\infty}\tau(w)>\alpha_\text{\normalfont c}+1$.
\item \label{2:threshold:res:3} $\gamma_\text{\normalfont c}>0$ and $\liminf_{w\to\infty}w^{-\gamma_\text{\normalfont c}}\tau(w)>\alpha_\text{\normalfont c}$.
\end{enumerate}
\end{theorem}

\noindent The theorem identifies different criteria for $\tau$ depending on the quantity $\gamma_\text{c}$ and hence the values of $\beta^-$ and $\beta^+$. Since $\beta^->2$ and $\beta^+>2$, we note that always $\gamma_\text{c}<1$. That is, also in Part \ref{2:threshold:res:3}.~of the theorem it is possible to choose a sub-linear threshold function $\tau$ that ensures resilience. On the other hand, even the constant threshold function $\tau(w)=2$ for all $w\in\R_{+,0}$ ensures resilience by Part \ref{2:threshold:res:1}.~whenever $\gamma_\text{c}<0$. This is in particular the case if $\beta^->3$ and $\beta^+>3$, i.\,e.~if $W^-$ and $W^+$ both admit finite second moments. This is in line with the results from \cite{Cont2016}. In addition, the theorem makes statements about cases when $\beta^-<3$ and $\beta^+>3$ or vice versa. Such parameters were observed on real markets for example in \cite{Cont2013}. In these cases, all $\gamma_\text{c}<0$, $\gamma_\text{c}=0$ or $\gamma_\text{c}>0$ are possible and only the exact values of $\beta^-$ and $\beta^+$ determine the condition for resilience.

\pagebreak
\begin{remark}\label{2:rem:contagious:links}
In Theorem \ref{2:threshold:res} we make the assumption of $\tau(w)\geq2$. In other words, each bank must at least be capable of sustaining the default of its largest debtor. This requirement has already been implemented in an even stricter form in the \textit{Supervisory framework for measuring and controlling large exposures} by the \textit{Basel Committee on Banking Supervision} from 2014 which will become applicable as from 2019 \cite{BaselCommittee2014}. While being economically sensible, the assumption is actually not necessary in order to derive analytical results. For the case of $\gamma_\text{c}<0$ it is enough to postulate $\E\left[W^-W^+\1\{\tau(W^-)=1\}\right]<1$ in order to ensure resilience. Also in the case of $\gamma_\text{c}\geq0$, it suffices to adjust $\alpha_\text{c}$ for a factor $(1-\E[W^-W^+\1\{\tau(W^-)=1\}])^{-1}$, whenever $\E\left[W^-W^+\1\{\tau(W^-)=1\}\right]<1$. Both results follow from a simple modification of our proof.
\end{remark}

\noindent Note that Theorem \ref{2:threshold:res} is formulated with assumptions on the marginal distributions of $W^-$ and $W^+$ only. Hence, the result is robust with respect to the dependency structure of the weights, i.\,e.~the resilience criteria are sufficient for all dependency structures. As Theorem \ref{2:thm:functional:nonres} will show, in the case of comonotone weights, the values of $\gamma_\text{c}$ and $\alpha_\text{c}$ are sharp. Also in the case of upper tail dependent weights (a reasonable assumption for real financial networks) the value of $\gamma_c$ is sharp. By $W^-$ and $W^+$ being upper tail dependent we mean that
\[ \lambda:=\liminf_{p\to0}\P(F_{W^+}(W^+)>1-p \mid F_{W^-}(W^-)>1-p) > 0. \]
If even
\[ \Lambda(x) := \lim_{p\to0}\P(F_{W^+}(W^+)>1-xp \mid F_{W^-}(W^-)>1-p) \]
exists for all $x\geq0$, we are able to determine explicitly sharp thresholds $\alpha_\text{c}(\Lambda)$ given by
\[ \alpha_\text{c}(\Lambda) := w_\text{min}^+(w_\text{min}^-)^{1-\gamma_\text{c}}\int_0^\infty\Lambda\left(x^{1-\beta^+}\right)\dd x. \]
For comonotone dependence (i.\,e.~$\Lambda(x)=1\wedge x$), $\alpha_\text{c}(\Lambda)$ coincides with $\alpha_\text{c}$ from Theorem \ref{2:threshold:res}.

\begin{theorem}\label{2:thm:functional:nonres}
Consider the same situation as in Theorem \ref{2:threshold:res}. If $\gamma_c>0$, the following holds:
\begin{enumerate}
\item \label{2:thm:functional:nonres:1} If $\limsup_{w\to\infty}w^{-\gamma_\text{\normalfont c}}\tau(w) < \lambda \frac{\beta^+-2}{\beta^+-1} \alpha_\text{c}$, then the system is non-resilient.
\item \label{2:thm:functional:nonres:2} If $\Lambda(x)$ exists for each $x\geq0$ and $\limsup_{w\to\infty}w^{-\gamma_\text{\normalfont c}}\tau(w) < \alpha_\text{\normalfont c}(\Lambda)$, then the system is non-resilient. If $\liminf_{w\to\infty}w^{-\gamma_\text{\normalfont c}}\tau(w) > \alpha_\text{\normalfont c}(\Lambda)$, then the system is resilient.
\end{enumerate}
\end{theorem}
In Part \ref{2:thm:functional:nonres:2}.~of the theorem, we characterize threshold functions $\tau$ that are asymptotically smaller respectively larger than $\alpha_\text{c}(\Lambda)w^{\gamma_\text{c}}$. In the proof we calculate the derivative of $f(z)$ at $z=0$ in order to show non-resilience ($f'(0)>0$) respectively resilience ($f'(0)<0$). If $\tau(w)$ asymptotically behaves like $\alpha_\text{c}(\Lambda)w^{\gamma_\text{c}}$, we obtain $f'(0)=0$ and hence both (\ref{2:condition:nonres}) and (\ref{2:thm:cont:res:ass}) are possible (not simultaneously). In this case, the exact form of $\tau$ and not only its asymptotics are important to decide whether the system is resilient or non-resilient.

\begin{remark}
If the weights $W^-$ and $W^+$ are not upper tail dependent, the conditions from Theorem \ref{2:threshold:res} are generally too strict. If their dependency is such that $\E[W^+(W^-)^{1-\gamma}]<\infty$ for some $\gamma\in(0,\gamma_\text{c}]$ for example, then $\liminf_{w\to\infty}w^{-\gamma}\tau(w)>0$ is already a sufficient criterion for resilience of the system. This can easily be derived from line (\ref{2:eqn:expectation:upper:part}) in the proof of Theorem \ref{2:threshold:res}.
\end{remark}

\noindent Theorems \ref{2:threshold:res} and \ref{2:thm:functional:nonres} both describe financial systems whose weights are given by Pareto distributed random variables. While such random variables model the tails of empirical degree distributions very well, typically for small weights there is a non-negligible discrepancy. However, the proofs of Theorems \ref{2:threshold:res} and \ref{2:thm:functional:nonres} show that it is in fact only the tail that determines (non-)resilience of a financial system. Therefore, assume in the following that there exist constants $K^-,K^+\in(0,\infty)$ and $\beta^-,\beta^+>2$ such that
\begin{equation}\label{2:eqn:Pareto:type}
1-F_{W^\pm}(w) \leq \left(\frac{w}{K^\pm}\right)^{1-\beta^\pm}
\end{equation}
for $w$ large enough. That is, the tails of the survival functions of $W^-$ and $W^+$ are bounded by the powers $1-\beta^-$ resp.~$1-\beta^+$. Then the following version of Theorem \ref{2:threshold:res} holds.
\begin{theorem}\label{2:thm:Pareto:type}
Suppose Assumption \ref{2:ass:tau} for a non-decreasing function $\tau:\R_{+,0}\to\N\backslash\{0,1\}$ such that for each bank $i\in[n]$ the threshold value $\tau_i$ depends on in-weight $w_i^-$ by $\tau_i=\tau(w_i^-)$. Moreover, let the distribution functions of $W^-$ and $W^+$ satisfy \eqref{2:eqn:Pareto:type}. For $\gamma_\text{\normalfont c}$ defined as before, the system is resilient if one of following holds:
\begin{enumerate}
\item \label{2:thm:Pareto:type:1} $\gamma_\text{\normalfont c}<0$
\item \label{2:thm:Pareto:type:2} $\gamma_c=0$ and $\liminf_{w\to\infty}\tau(w)>\frac{\beta^+-1}{\beta^+-2}K^+K^-+1$
\item \label{2:thm:Pareto:type:3} $\gamma_c>0$ and $\liminf_{w\to\infty}w^{-\gamma_c}\tau(w)>\frac{\beta^+-1}{\beta^+-2}K^+(K^-)^{1-\gamma_c}$
\end{enumerate}
\end{theorem}
Note that by the same means also Theorem \ref{2:thm:functional:nonres} can be generalized. For non-resilience the inequality in \eqref{2:eqn:Pareto:type} needs to be inverted such that it describes a lower bound on the tail of the distributions.

\subsection{Systemic Capital Requirements}\label{2:ssec:capital:requirements}
In this section, we translate the threshold requirements from Theorem \ref{2:threshold:res} to capital requirements in the exposure model. That is, we state explicit amounts of capital each bank has to be able to procure in stress scenarios in order for the system to be resilient. As for the threshold requirements, it is important to note that each bank can compute its capital requirements on its own by just knowing its local neighborhood in the network. Further, a bank's capital requirement only depends on the default risk the bank exposes itself to and not on the default risk the bank poses to other banks. Proposition \ref{2:prop:robust:capital:requirements} states a straightforward robust way to translate threshold requirements into sufficient capital requirements. In general, it might lead to capital requirements that are too high and hence unnecessarily reduce interbank lending and liquidity, however. Thus, we further provide Theorem \ref{2:cor:threshold:res} below, which accurately determines capital requirements under a certain regularity assumption on the exposure lists.

\begin{proposition}\label{2:prop:robust:capital:requirements}
Suppose Assumption \ref{2:ass:tau} for non-decreasing $\tau:\R_{+,0}\to\N\backslash\{0,1\}$ and limiting weights $W^-\sim\mathrm{Par}(\beta^-,w_\text{min}^-)$ respectively $W^+\sim\mathrm{Par}(\beta^+,w_\text{min}^+)$ with $\beta^-,\beta^+>2$, $w_\text{min}^-,w_\text{min}^+>0$. Further, assume that $\liminf_{w\to\infty}\tau(w)>\alpha_\text{\normalfont c}+1$ if $\gamma_\text{\normalfont c}=0$ respectively $\liminf_{w\to\infty}w^{-\gamma_\text{\normalfont c}}\tau(w)>\alpha_\text{\normalfont c}$ if $\gamma_\text{\normalfont c}>0$, where the quantities $\gamma_\text{\normalfont c}$ and $\alpha_\text{\normalfont c}$ are as in Theorem \ref{2:threshold:res}. Then the system is resilient if 
\[ c_i > \max\Bigg\{\sum_{j\in J}E_{j,i}~\Bigg\vert~J\subset[n], \vert J\vert=\tau(w_i^-)-1\Bigg\} \quad\text{almost surely for all }i\in[n], \]
i.\,e.~capital $c_i$ of bank $i\in[n]$ is larger than the sum of the $\tau(w_i^-)-1$ largest exposures of $i$.
\end{proposition}
Analogously, a robust translation of Theorems \ref{2:thm:functional:nonres} and \ref{2:thm:Pareto:type} to the exposure model is possible.

Proposition \ref{2:prop:robust:capital:requirements} requires each bank $i$ to be able to cope with default of its $\tau(w_i^-)$ largest exposures. But as we have seen in the proof of Theorem \ref{2:threshold:res}, only the thresholds and hence the capitals of large banks in the network matter for resilience. For large banks with many exposures on the other hand one can expect an averaging effect of the exposure sizes to occur if they are not too irregular. Hence, one can presume that in this case multiplying threshold values from the threshold model by average exposure sizes for each bank leads to the same resilience characteristics. We formalize this in Theorem \ref{2:cor:threshold:res} under Assumption \ref{2:ass:exposures} on the exposure sequences. This assumption is motivated by the following reasoning:

For each bank $i$, let $\{E_{j,i}\}_{j\in\N\backslash\{i\}}$ be a sequence of i.\,i.\,d.~positive random variables. Let $\lambda_i:=\E[E_{\rho_i(1),i}]<\infty$ be their mutual expectation and denote $S_k^i:=\sum_{j=1}^k E_{\rho_i(j),i}$. If there is some $t>1$ such that $\E\left[\vert E_{\rho_i(1),i}\vert^t\right]<\infty$, then by the Baum-Katz-Theorem from \cite{Baum1965} for all $\epsilon>0$,
\begin{equation}\label{2:eqn:Baum:Katz:1}
k^{t-1}\P\left(S_{k}^i\geq (1+\epsilon)k\lambda_i\right) \to 0,\quad \text{as }k\to\infty,
\end{equation}
and for all $x>1$,
\begin{equation}\label{2:eqn:Baum:Katz:2}
k^{tx-1}\P\left(S_{k}^i\geq\epsilon \lambda_i k^x\right) \to 0,\quad \text{as }k\to\infty.
\end{equation}

\begin{assumption}\label{2:ass:exposures}
Motivated by the above, we assume that for each bank $i\in[n]$ with exposure list $\{E_{j,i}\}_{j\in\N\backslash\{i\}}$ of mutual mean $\lambda_i$, we can find $t>1$ such that the convergences in (\ref{2:eqn:Baum:Katz:1}) and (\ref{2:eqn:Baum:Katz:2}) hold. Moreover, we assume them to be uniform for $i\in[n]$ (but not necessarily for $\epsilon$ or $x$).
\end{assumption}

\noindent Assumption \ref{2:ass:exposures} ensures a certain regularity of the exposures without bounding their mean.

\begin{theorem}\label{2:cor:threshold:res}
Suppose Assumption \ref{2:ass:tau} for non-decreasing $\tau:\R_{+,0}\to\N\backslash\{0,1\}$ and such that $W^-\sim\mathrm{Par}(\beta^-,w_\text{min}^-)$ and $W^+\sim\mathrm{Par}(\beta^+,w_\text{min}^+)$ with $\beta^-,\beta^+>2$, $w_\text{min}^-,w_\text{min}^+>0$. The quantities $\gamma_\text{\normalfont c}$ and $\alpha_\text{\normalfont c}$ shall be defined as in Theorem \ref{2:threshold:res}. Further, assume $c_i>\max_{j\in[n]\backslash\{i\}}E_{j,i}$ almost surely for all $i\in[n]$. Then the following holds:
\begin{enumerate}
\item \label{2:cor:threshold:res:1} If $\gamma_\text{\normalfont c}<0$, then the system is always resilient.
\end{enumerate}
Now further assume that the exposure lists $\{E_{j,i}\}_{j\in\N\backslash\{i\}}$, $i\in\N$, satisfy Assumption \ref{2:ass:exposures} for some $t>1$. Then the system is resilient if one of the following holds:
\begin{enumerate}
\setcounter{enumi}{1}
\item \label{2:cor:threshold:res:2} $\gamma_\text{\normalfont c}=0$ and there exist some $\gamma>0$ 
such that $\liminf_{w\to\infty}w^{-\gamma}\tau(w)>0$ and for all $i\in[n]$, $c_i\geq \tau(w_i^-)\lambda_i$ almost surely.
\item \label{2:cor:threshold:res:3} $\gamma_\text{\normalfont c}>0$, $\liminf_{w\to\infty}w^{-\gamma_\text{\normalfont c}}\tau(w)>\alpha_\text{\normalfont c}$ and for all $i\in[n]$, $c_i\geq \tau(w_i^-)\lambda_i$ almost surely.
\end{enumerate}
\end{theorem}
Theorem \ref{2:cor:threshold:res} provides the banks with a formula that is easy to use and only requires the regulator to announce $\alpha_\text{c}$ and $\gamma_\text{c}$. Resilient capital requirements are then determined according to average exposure size $\lambda_i$ and number of exposures $d_i^-\sim w_i^-$. Since the average exposure size $\lambda_i$ is proportional to $(d_i^-)^{-1}$ while the factor $\alpha_\text{c}(d_i^-)^{\gamma_\text{c}}$ is sublinear in $d_i^-$, in particular a deconcentration of loans is favorable for the banks to reduce systemic risk charges.

\begin{remark}
Theorem \ref{2:cor:threshold:res} extends Theorem \ref{2:threshold:res} to the exposure model under Assumption \ref{2:ass:exposures} for the exposure sequences. By the same means, also Theorems \ref{2:thm:functional:nonres} and \ref{2:thm:Pareto:type} can be extended.
\end{remark}

\section{Simulation Study}\label{2:simulation:study}
All previous chapters have been formulated in the limit as the number of banks $n$ tends to $\infty$ and the fraction of initially defaulted banks $p$ tends to 0. It is hence reasonable to investigate whether the results are good approximations also for real networks which are finite with only a few thousand institutions and experience a shock of a positive fraction of banks. Since our model is based on the non-observable weight-parameters, one would have to estimate them from the degree-sequences which are observable for real network configurations at least by regulating institutions. Hence, we will start this section by a short note on weight-estimation. Since specific transactions between banks are not disclosed to the public there is no data basis for us to investigate real networks, however. Instead, we will subsequently discuss our findings by simulating networks. For simplicity we consider the final default fraction $n^{-1}\vert\mathcal{D}_n\vert$ as the systemic risk measure, i.\,e.~$s_i=1$ for all banks $i\in[n]$.

\subsection{Estimation of Weights}\label{2:ssec:weights:estimation}
Since the weight sequences are not directly observable from real networks, we give a few lines here on how to estimate them from data that is observable. First note that for a network $G$ of size $n$ with edge set $E(G)$ 
the likelihood of weight sequences $\bm{w}^-=(w_1^-,\ldots,w_n^-)$ and $\bm{w}^+=(w_1^+,\ldots,w_n^+)$ is given by
\[ L(w_1^-,w_1^+,\ldots,w_n^-,w_n^+\mid E(G)) = \prod_{(i,j)\in E(G)}\left(\frac{w_i^+w_j^-}{n}\wedge1\right)\prod_{\substack{(i,j)\not\in E(G)\\i\neq j}}\left(1-\frac{w_i^+w_j^-}{n}\wedge1\right). \]
One can always derive the maximum-likelihood estimators $\hat{w}_1^-,\ldots,\hat{w}_n^-,\hat{w}_1^+,\ldots,\hat{w}_n^+$ by numerically maximizing $L$. 
In order to obtain some intuition about them, we further want to derive an approximation of the estimators. For this, we assume that $w_i^+w_j^-\ll n$ for all $i,j\in[n]$ which is a reasonable assumption at least when $W^+$, $W^-$ are square-integrable. We can hence approximate
\[ L(w_1^-,w_1^+,\ldots,w_n^-,w_n^+\mid E(G)) \approx \frac{1}{n^s} \prod_{i\in[n]}\left(w_i^-\right)^{d_i^-}\left(w_i^+\right)^{d_i^+} \exp\left(-w_i^+\frac{\sum_{j\in[n]}w_j^-}{n}\right), \]
where $s:=\sum_{i\in[n]}d_i^-=\sum_{i\in[n]}d_i^+$. By the product form $w_i^+w_j^-$ in (\ref{2:conn:prob}), we are free to multiply all out-weights $w_i^+$ by some constant $\eta$ if, at the same time, we multiply all in-weights by its inverse $\eta^{-1}$. Motivated by the fact that $\sum_{i\in[n]}d_i^-=\sum_{i\in[n]}d_i^+$, we use this degree of freedom to set $\sum_{i\in[n]}w_i^-=\sum_{i\in[n]}w_i^+$ and want to maximize the approximated likelihood function under this constraint. (Other constraints, such as $\sum_{i\in[n]}w_i^-=\mathrm{const}$, are also possible and lead to the same result in the end.) By Lagrange's multiplier method this leads to a maximization of
\[ \prod_{i\in[n]}\left(w_i^-\right)^{d_i^-}\left(w_i^+\right)^{d_i^+} \exp\left(-w_i^+\frac{\sum_{j\in[n]}w_j^-}{n}\right) + \lambda\left(\sum_{k\in[n]}w_k^--w_k^+\right). \]
Differentiating with respect to $w_l^-$ resp.~$w_l^+$ for all $l\in[n]$, we are left with solving the equations
\[ 0 = \prod_{i\in[n]}\left(w_i^-\right)^{d_i^-}\left(w_i^+\right)^{d_i^+} \exp\left(-w_i^+\frac{\sum_{j\in[n]}w_j^-}{n}\right) \left(\frac{d_l^-}{w_l^-}-\frac{\sum_{k\in[n]}w_k^+}{n}\right) + \lambda \]
respectively
\[ 0 = \prod_{i\in[n]}\left(w_i^-\right)^{d_i^-}\left(w_i^+\right)^{d_i^+} \exp\left(-w_i^+\frac{\sum_{j\in[n]}w_j^-}{n}\right) \left(\frac{d_l^+}{w_l^+}-\frac{\sum_{k\in[n]}w_k^-}{n}\right) - \lambda. \]
In particular, $d_l^-/w_l^-$ resp.~$d_l^+/w_l^+$ must be independent of $l$ and we can thus find constants $\lambda^-$ and $\lambda^+$ such that $w_l^-=\lambda^- d_l^-$ and $w_l^+=\lambda^+ d_l^+$. Using the constraints $\sum_{i\in[n]}w_i^-=\sum_{i\in[n]}w_i^+$ and $\sum_{i\in[n]}d_i^-=\sum_{i\in[n]}d_i^+$, we obtain that $\lambda=0$ and $\lambda^-=\lambda^+=\sqrt{n/\sum_{i\in[n]}d_i^-}$ such that the approximated likelihood function is maximized by
\[ w_i^-=d_i^-\sqrt{\frac{n}{\sum_{j\in[n]}d_j^-}},\qquad w_i^+=d_i^+\sqrt{\frac{n}{\sum_{j\in[n]}d_j^-}}. \]
That is, the approximated weight estimators are proportional to the observed degrees and only normalized in a certain sense. The normalization is necessary due to our choice of $p_{i,j}$ in (\ref{2:conn:prob}). If we had chosen $p_{i,j}=1\wedge w_i^+w_j^-/\sum_{k\in[n]}w_k^+$ instead for example, then $w_i^-$ and $w_i^+$ could be interpreted directly as expected degrees and estimated by $d_i^-$ respectively $d_i^+$. All previous calculations would need to be adjusted by a factor $n/\sum_{k\in[n]}w_k^+\approx1/\E[W^+]$ but analogous results would still hold.

The smaller the observed fraction $\max_{i,j\in[n]}d_i^+d_j^-/\sum_{k\in[n]}d_k^-$, the better is above approximation of $w_i^+w_j^-=d_i^+d_j^-n/\sum_{k\in[n]}d_k^-\ll n$. On networks where $\max_{i,j\in[n]}d_i^+d_j^-/\sum_{k\in[n]}d_k^-$ is large, $\bm{w}^-$ and $\bm{w}^+$ have to be estimated numerically.

\subsection{Simulations for the Threshold Model}
For our simulations, we make use of the findings in \cite{Cont2013} that the empirical in- and out-degrees as well as the exposure sizes in the Brazilian banking network are power law distributed. For November 2008, the authors of \cite{Cont2013} estimated the power law exponents $\beta^-=2.132$ and \linebreak$\beta^+=2.8861$ for the degree sequences and $\xi=2.5277$ for the exposures. In our weight-based model, these degree distributions are obtained by choosing in- and out-weights power law distributed with exponents $\beta^-$ and $\beta^+$ as well. In addition to this, we assume them to be comonotone and Pareto distributed with minimal weights $w_\text{min}^-=w_\text{min}^+=1$. 

In a first simulation, we consider a threshold model with above weight parameters and assume absence of contagious links but nothing more. That is, we set $\tau_i=2$ for all $i\in[n]$. In order to start the cascade process, we assume initial default of $p=1\%$ uniformly chosen banks in the network. We then simulate the default process for $n\in\{100k\,:\, k\in[100]\}$ and $100$ different configurations of the random network for each $n$. The results for the final fraction of defaulted banks are plotted in Figure \ref{2:fig:Convergence}. As can be seen from Figure \ref{2:fig:TheoreticalFraction}, the theoretical value of the final default fraction as $n$ tends to infinity can be determined to be approximately $84.54\%$. This value is drawn as a red line in Figure \ref{2:fig:Convergence}. Already for small $n$, most of the simulations yield results that are close to this theoretical value and the networks can hence be understood as being non-resilient. As $n$ grows to $10^4$ the final fractions become even more precise. In particular, there is not a single resilient sample anymore for $n\geq500$. Here and in the following, by resilience for finite networks and positive shocks sizes we informally mean that the final default fraction is small compared to some reference value (here $84.54\%$).

\begin{figure}[t]
    \hfill\subfigure[]{\includegraphics[width=0.4\textwidth]{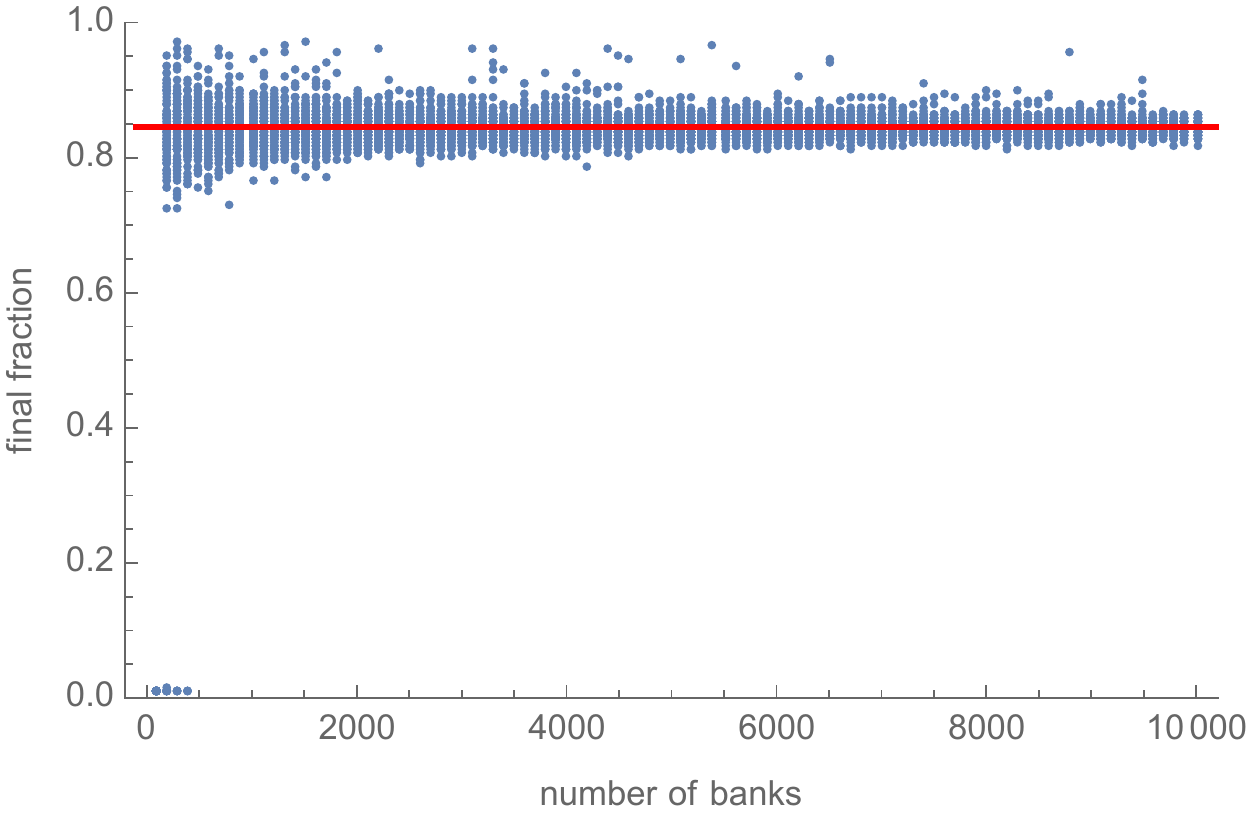}\label{2:fig:Convergence}}
    \hfill\subfigure[]{\includegraphics[width=0.4\textwidth]{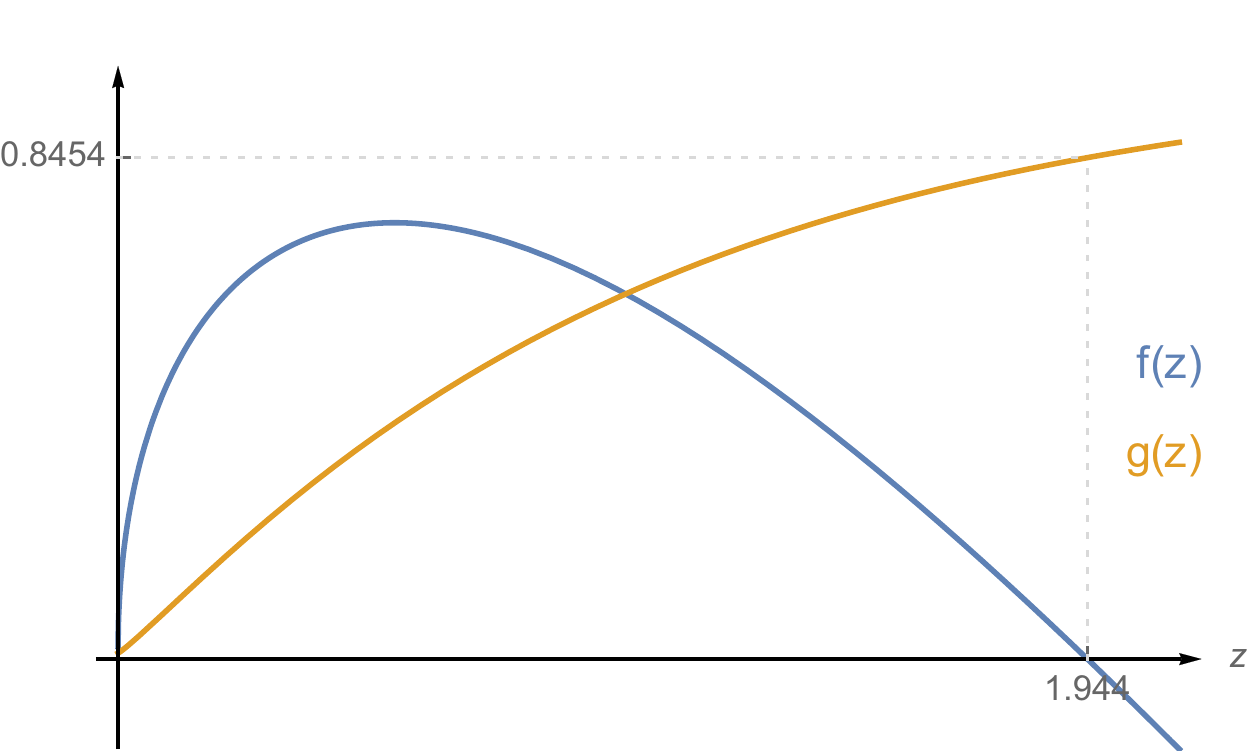}\label{2:fig:TheoreticalFraction}}\hfill
\caption{(a) Convergence of the final fraction of defaulted banks in the threshold model for networks of finite size. (b) Determination of the theoretical final default fraction in the threshold model for networks whose sizes grow to infinity and with $p=1\%$ initial defaults and constant threshold $2$. Blue: $f(z)=(1-p)\E[W^+\psi_2(W^-z)]+p\E[W^+]-z$ with root $\hat{z}\approx 1.94433$. Orange: $g(z)=(1-p)\E[\psi_2(W^-z)]+p$ with $g(\hat{z})\approx 0.845434$.}\label{2:fig:Convergence:TheoreticalFraction}
\end{figure}
\begin{figure}[t]
	\hfill\subfigure[]{\includegraphics[width=0.4\textwidth]{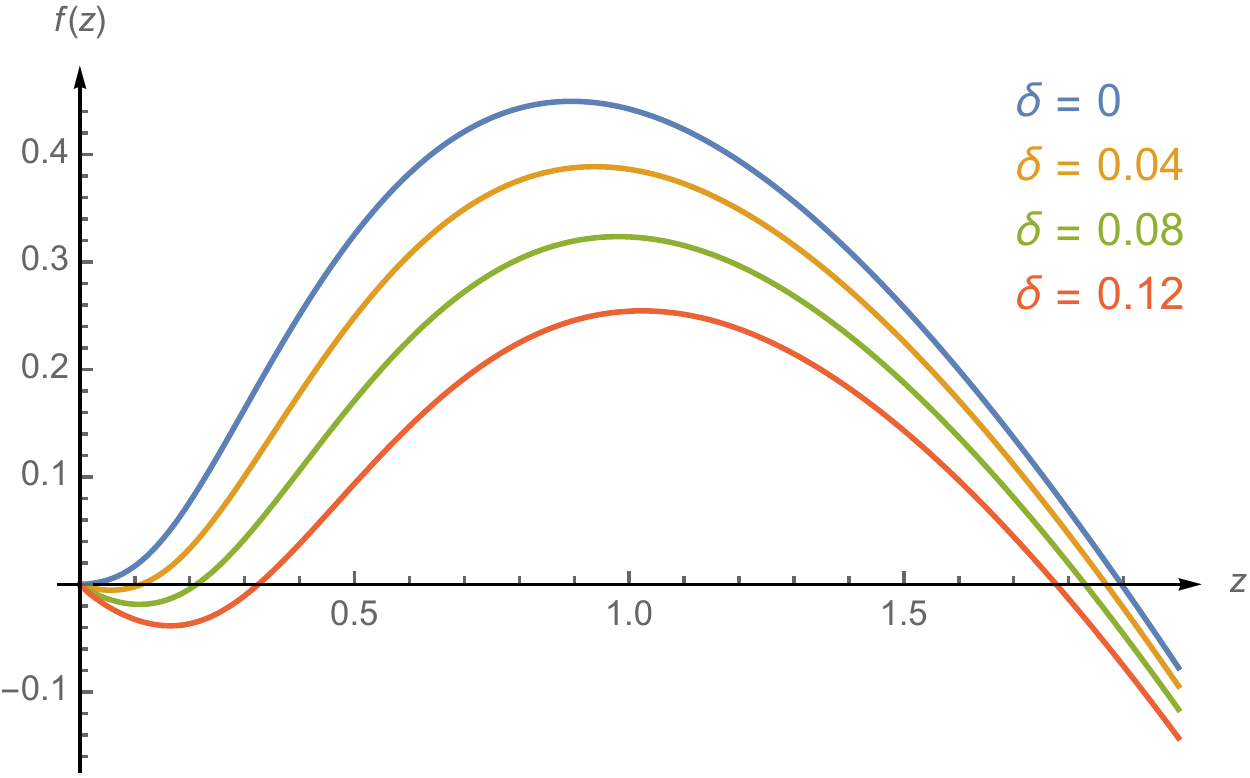}\label{2:fig:influenceOfDelta}}
    \hfill\subfigure[]{\includegraphics[width=0.4\textwidth]{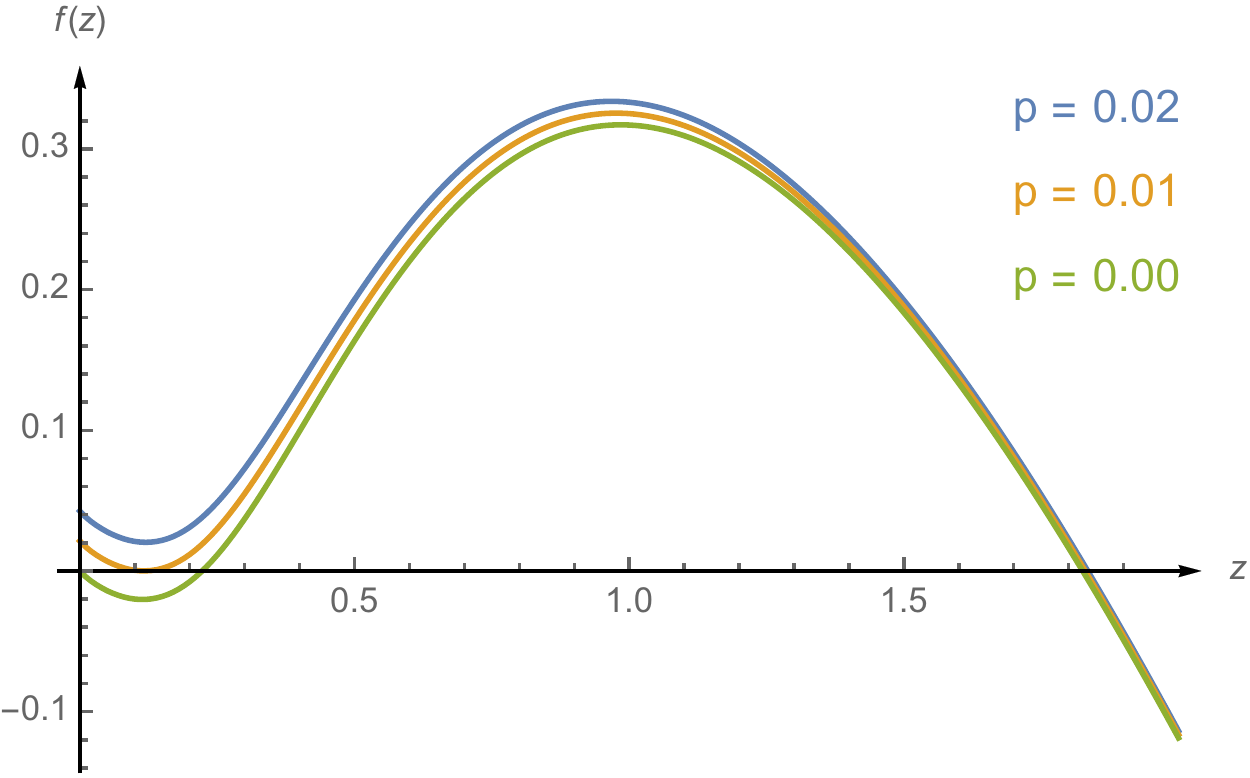}\label{2:fig:influenceOfP}}\hfill
\caption{(a) Influence of $\delta$ on the shape of $f(z)=\E[W^+\psi_T(W^-z)]-z$ with capital requirements $\tau_i=\max\{2,\lfloor (\alpha_\text{c}(1+\delta) (w_i^-)^{\gamma_\text{c}(1+\delta)}\rfloor\}$. (b) Influence of $p$ on the shape of function $f(z)=(1-p)\E[W^+\psi_T(W^-z)]+p\E[W^+]-z$ for the example of $\delta=0.0839$.}\label{2:fig:influenceOfDelta:influenceOfP}
\end{figure}

Instead of the absence of contagious links, Theorem \ref{2:threshold:res} predicts certain threshold requirements to make our network model resilient to small initial shocks. Keeping above network parameters unchanged, we compute $\alpha_\text{c}\approx 2.13$ and $\gamma_\text{c}\approx0.468$. A natural choice for the threshold of bank $i\in[n]$ is then $\tau_i=\max\{2,\lfloor \alpha (w_i^-)^\gamma\rfloor\}$, where $\alpha=\alpha_\text{c}(1+\delta)$, $\gamma=\gamma_\text{c}(1+\delta)$ and $\delta\in[-1,\infty)$ denotes a (possibly negative) buffer. By Theorems \ref{2:threshold:res} and \ref{2:thm:functional:nonres}, networks are resilient to initial shocks for $\delta>0$ and non-resilient for $\delta<0$. The influence of $\delta$ on $f(z)$ can be seen in Figure \ref{2:fig:influenceOfDelta}. In particular, one notes that resilience for positive $\delta$ stems from the negative hump of $f(z)$ subsequent to zero. Further note, however, that resilience is only guaranteed to shocks whose size tends to zero. Even networks, where the number of banks tends to infinity but which are shocked by a strictly positive initial default fraction $p$, will only be resilient for $\delta>\delta_p$ for a certain $\delta_p>0$. This is because $f(z)$ depends on $p$ by $f(z)=(1-p)\E[W^+\psi_T(W^-z)]+p\E[W^+]-z$ if a uniformly chosen fraction $p$ of all banks in the network defaults at the beginning. The influence of $p$ on $f(z)$ can be seen in Figure \ref{2:fig:influenceOfP}. In order for a network to be resilient to an initial shock of $p$ the hump subsequent to $0$ needs to become negative in Figure \ref{2:fig:influenceOfP}. By this, it is always possible to determine the least necessary buffer $\delta$ to make a system resilient to a shock of initial default fraction $p$ numerically (see Table \ref{2:tab:deltaFromP} for the corresponding values $\delta_p$ for $p=0.001k$, $k\in[10]$). Note that a buffer of $\delta=0.0839$ yields $\alpha=2.31$ and $\gamma=0.507$ and hence the thresholds required to make the system resilient to shocks of $1\%$ are still strongly sublinear.
\begin{table}[h]
	\caption{List of values for buffer $\delta$ corresponding to initial default of $\lfloor pn\rfloor$ banks}
		\hfill\begin{tabular}{r|c|c|c|c|c|c|c|c|c|c}
			$p$ $[\%]$ & 0.1 & 0.2 & 0.3 & 0.4 & 0.5 & 0.6 & 0.7 & 0.8 & 0.9 & 1.0\\\hline
			$\delta_p$ $[\%]$ & 2.35 & 3.44 & 4.30 & 5.04 & 5.71 & 6.36 & 6.89 & 7.42 & 7.91 & 			8.39
		\end{tabular}
	\label{2:tab:deltaFromP}
	\hfill
	\vspace{-0mm}
\end{table}

\begin{figure}[t]
    \hfill\subfigure[]{\includegraphics[width=0.4\textwidth]{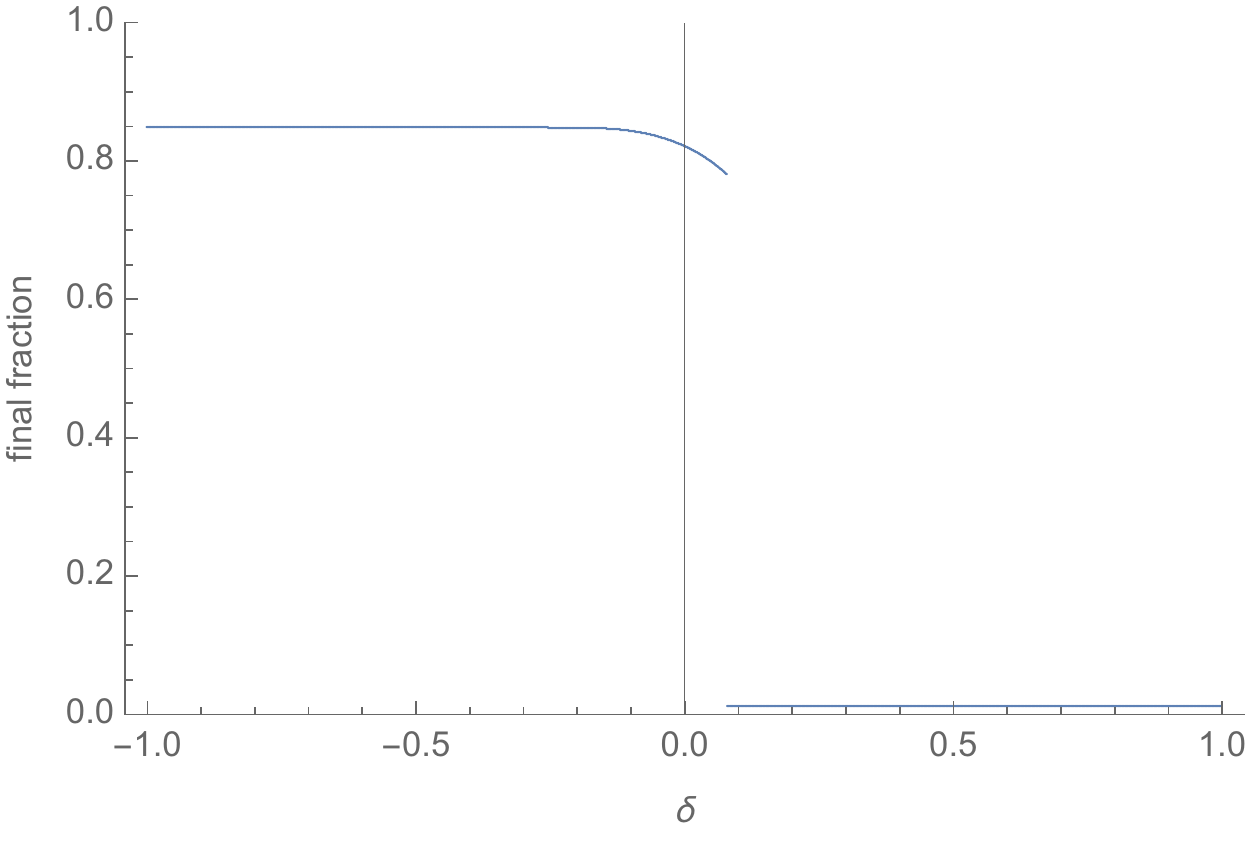}\label{2:fig:SimulationDelta1000000}}
    \hfill\subfigure[]{\includegraphics[width=0.4\textwidth]{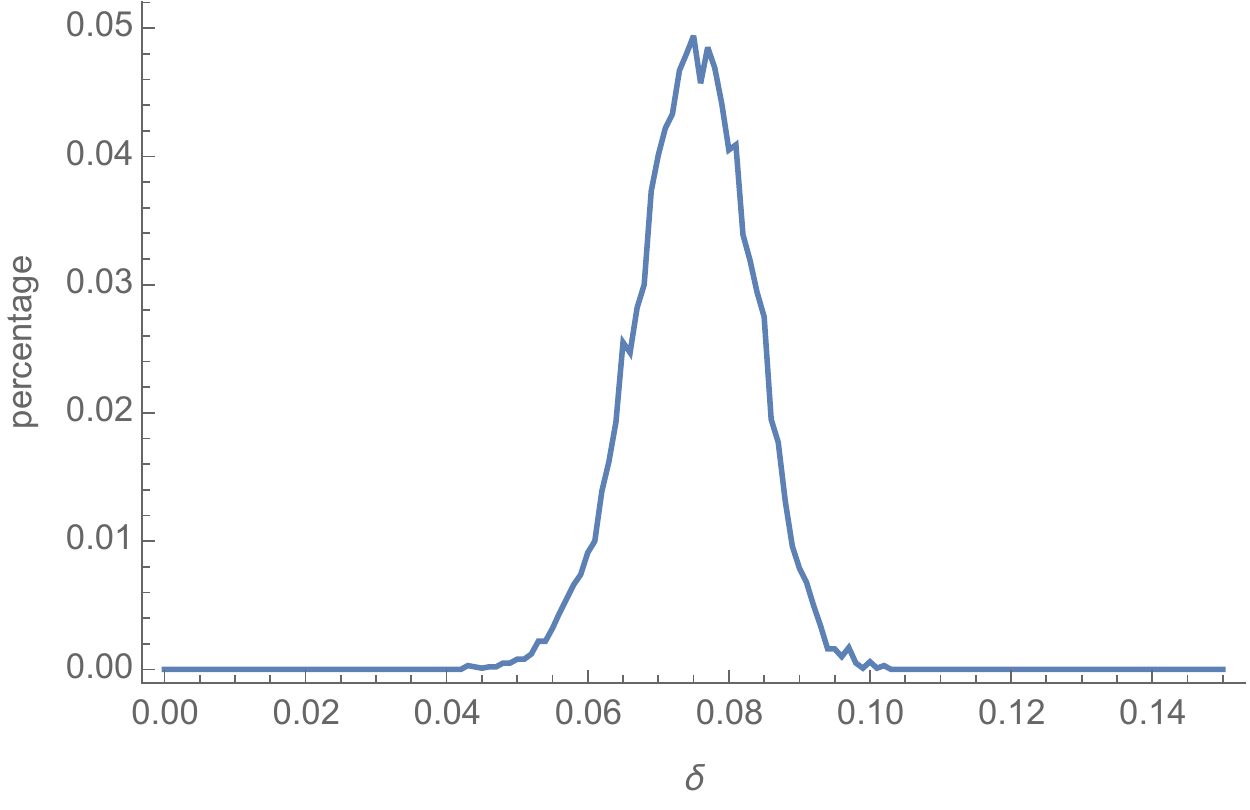}\label{2:fig:DeltaDistribution_0.01_1000000}}\hfill
\caption{(a) A typical result for the final fraction in a network of $10^6$ banks with initial default fraction of $p=1\%$ as $\delta$ varies between $-1$ and $1$ in steps of $10^{-3}$. (b) The distribution of jump points for $10^4$ networks of size $n=10^6$ with initial default fraction $p=1\%$.}\label{2:fig:SimulationDelta1000000:DeltaDistribution_0.01_1000000}
\vspace{0.5cm}
	\hfill\includegraphics[width=0.4\textwidth]{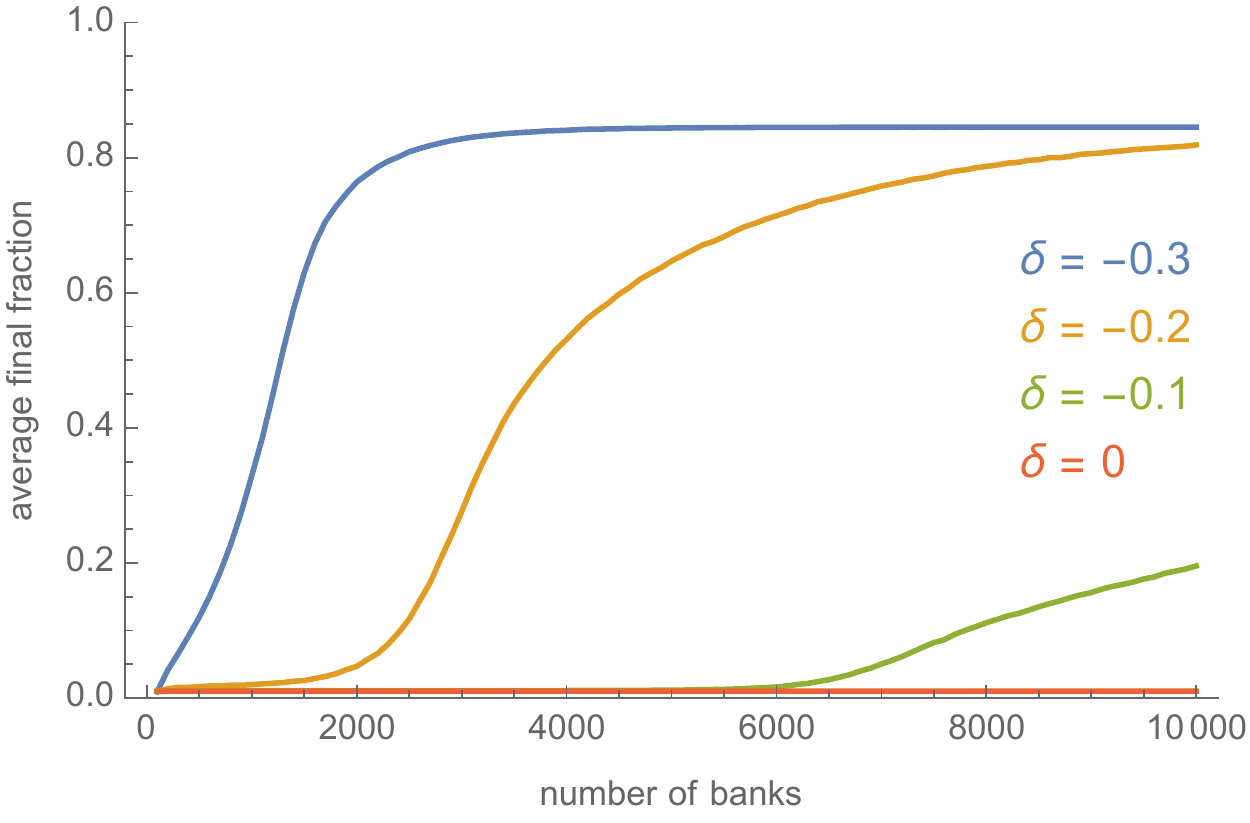}\hfill
		\caption{Average final fraction of defaulted banks in finite networks}
		\label{2:fig:gemittelt}
\end{figure}

\pagebreak
We want to verify above results by simulations. For this, we simulate a very large network consisting of $n=10^6$ banks and keeping the network topology constant we let $\delta$ vary between $-1$ and $1$ in steps of $10^{-3}$. For each simulated network, we then find that it becomes resilient for $\delta$ large enough. This becomes visible by a jump of the final fraction of defaulted banks at this particular $\delta$ as illustrated in Figure \ref{2:fig:SimulationDelta1000000} for a sample network. The jump shows that in the end it is only one bank whose default lets the whole system crash.

Keeping track of the values of $\delta$ at which the final fraction drops near $p=1\%$ for $10^4$  simulated networks yields the distribution shown in Figure \ref{2:fig:DeltaDistribution_0.01_1000000}. It shows a peak at about \mbox{$\delta=0.076$} and hence supports our theoretical findings from above. Deviations from the theoretical value $\delta_{0.01}\approx 0.0839$ are small and can be explained by the finite (albeit very large) network size.

Having looked at the theoretical capital requirements for very large networks, it is now sensible to turn our attention to networks of a few thousand banks as they arise in the real world. Figure \ref{2:fig:gemittelt} shows the final fraction in networks of size $n\leq10^4$ with initial default fraction $p=0.01$ for $\delta$ between $-0.3$ and $0$. For each $n$, we averaged over $10^5$ simulations. The figure shows that networks of size $n\leq10^4$ are already resilient for $\delta=0$. Even for $\delta=-0.2$ the network is rather resilient if $n\leq2,000$ resp.~for $\delta=-0.1$ if $n\leq6,000$. That is, our result is robust in the sense that already lower threshold requirements are sufficient to make the systems resilient to small shocks. The deviations stem from the relatively small network sizes of only a few thousand. Here, rare extreme values of vertex weights fail to appear despite the missing second moment condition or those large banks are not infected by the uniform initial infection. 

For managing systemic risk in real networks it might, however, be of interest not only how some uniform initial default influences the system but also how the default of the largest banks does. In a further simulation, we hence choose the $\lfloor pn\rfloor$ largest (by weights) banks in the network to default at the beginning. The function $f(z)$ then qualitatively keeps its shape as in Figure \ref{2:fig:influenceOfDelta:influenceOfP} but is shifted upwards. Again, we can compute corresponding values of $\delta$ and $p$ numerically. We list our results in Table \ref{2:tab:deltaFromPWithLargest}. As one expects, the values of $\delta_p$ are larger in this case than the ones we obtained for uniform infection in Table \ref{2:tab:deltaFromP}, but only by a factor of about $2$ and as before the resulting capital requirements are strongly sublinear.

\begin{table}[t]
	\caption{List of values for buffer $\delta$ corresponding to initial default of the $\lfloor pn\rfloor$ largest banks}
		\hfill\begin{tabular}{r|c|c|c|c|c|c|c|c|c|c}
		$p$ $[\%]$& 0.1 & 0.2 & 0.3 & 0.4 & 0.5 & 0.6 & 0.7 & 0.8 & 0.9 & 1.0\\\hline
		$\delta_p$ $[\%]$& 4.09 & 6.05 & 7.61 & 8.90 & 10.0 & 11.0 & 11.9 & 12.7 & 13.4 & 14.1
		\end{tabular}\hfill
	\label{2:tab:deltaFromPWithLargest}
\end{table}

\subsection{Simulations for the Exposure Model}

We can now turn to the simulation of a weighted network as in the exposure model. In addition to the network parameters of the threshold model in the previous section, we assume that for $i\neq j$, exposures $E_{j,i}$ are given by $E_{j,i}\stackrel{d}{=}E_i$ for Pareto distributed random variables $E_i$ with exponent $\xi=2.5277$ as in \cite{Cont2013} and minimal value $E_{\text{min},i}$. The exposures are assumed independent of each other and the network topology. The minimal exposures $E_{\text{min},i}$, can be chosen arbitrarily since they act as a constant factor for all exposures $E_{j,i}$ and capital $c_i$.

\begin{figure}[t]
    \hfill\subfigure[]{\includegraphics[width=0.4\textwidth]{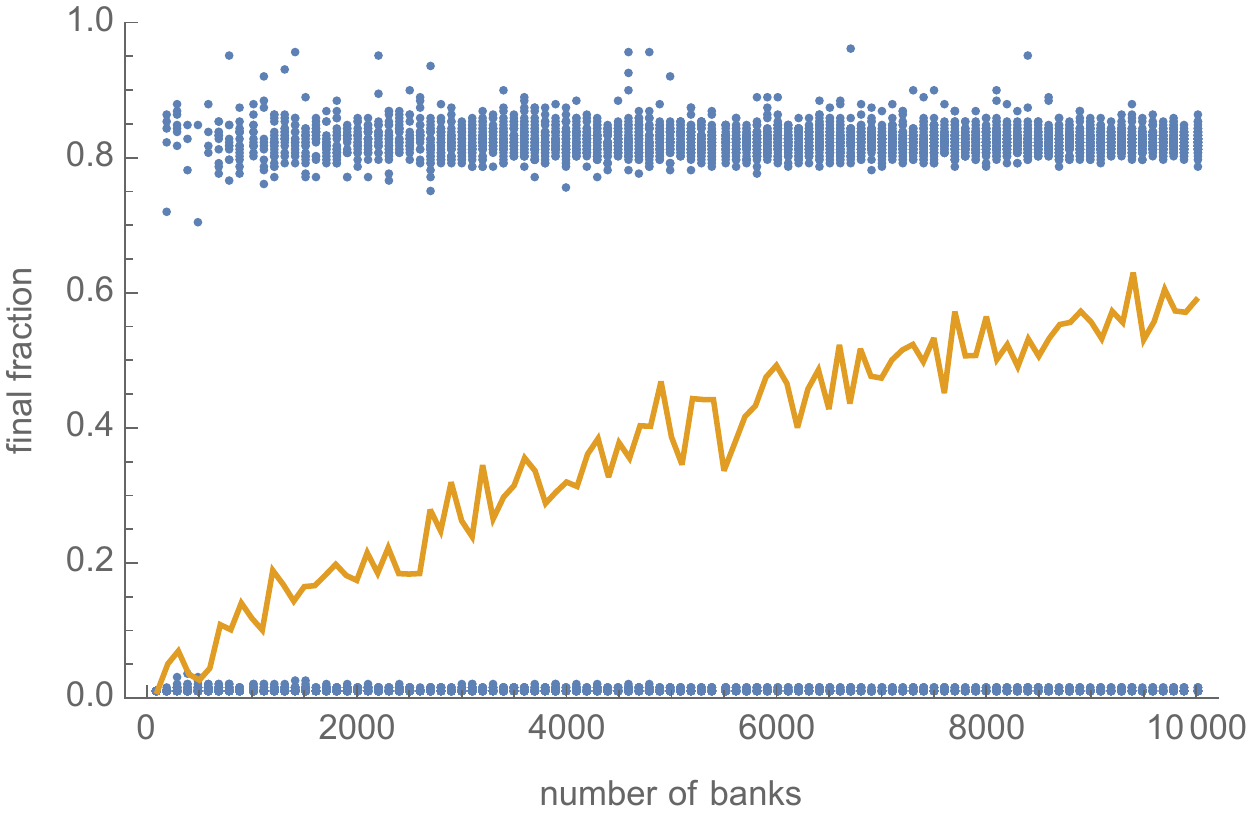}\label{2:fig:exposure:model}}
    \hfill\subfigure[]{\includegraphics[width=0.4\textwidth]{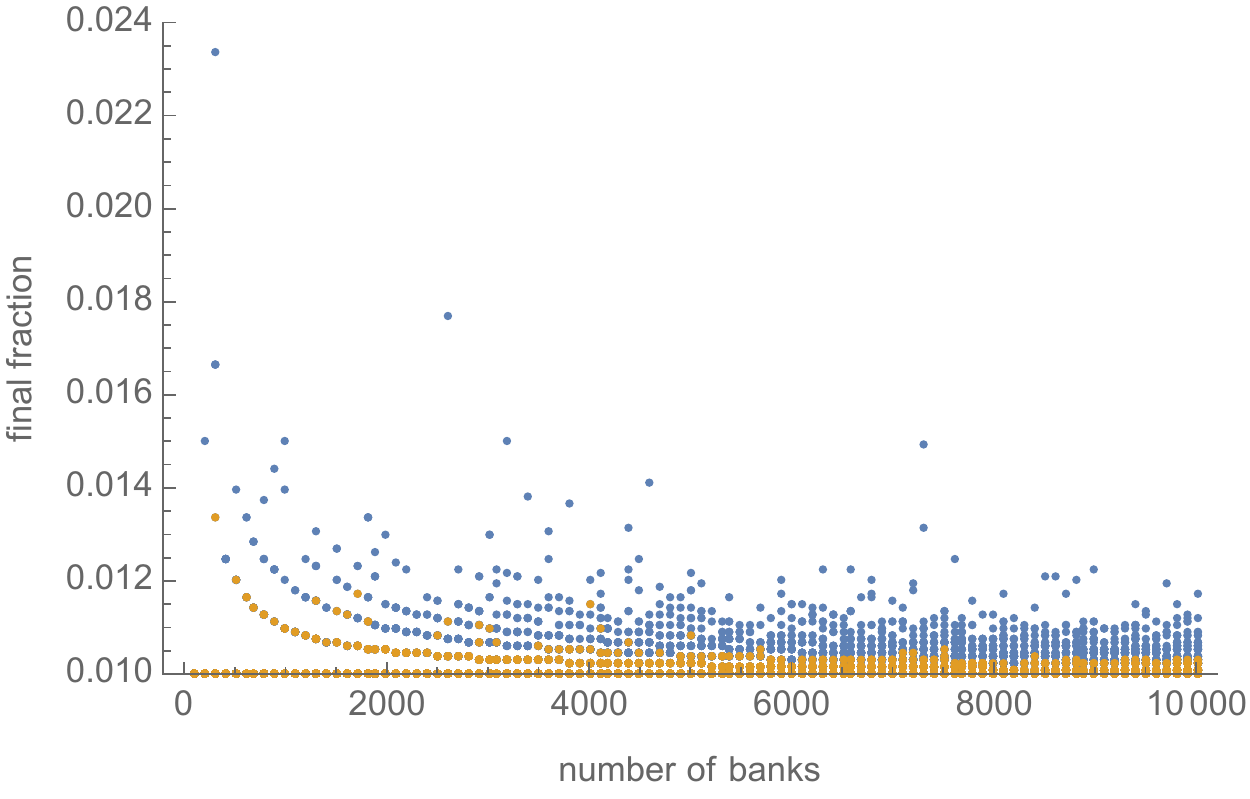}\label{2:fig:capital:requirements}}\hfill
\caption{(a) Scatter plot of the final fraction of defaulted banks for weighted networks of finite size without contagious links. Orange: average over all $100$ configurations for \mbox{each size.} \mbox{(b) Scatter} plot of the final fraction of defaulted banks for weighted networks of finite size. Blue: Capitals determined by Theorem \ref{2:cor:threshold:res}. Orange: Capitals determined by Proposition \ref{2:prop:robust:capital:requirements}}
		\label{2:fig:exposure:model:capital:requirements}
\end{figure}

In a first simulation, again we assume absence of contagious links but nothing more. That is, we first simulate the network skeleton and the edge-weights independently and then determine the banks' capitals as their largest exposure value plus some small buffer $\epsilon>0$. For our simulation, we choose $\epsilon=10^{-3}\E[E_i]=10^{-3}E_{\text{min},i}(\xi-1)/(\xi-2)$. As before, we assume initial default of $p=1\%$ uniformly chosen banks in the network and simulate the default process for $n\in\{100k\,:\, k\in[100]\}$ and $100$ different configurations of the random network for each $n$. The results for the final fraction of defaulted banks are plotted in Figure \ref{2:fig:exposure:model}. We notice that already for small network sizes there are some non-resilient network samples with final default fraction of about $80\%$. As the number of banks $n$ grows, also the probability that the networks are non-resilient significantly increases. This can be seen from the orange curve in Figure \ref{2:fig:exposure:model} which shows the average final default fraction taken over all $100$ configurations. The simulation supports our analytical result that for networks without a second moment condition on their degree sequences, simply the absence of contagious links does not ensure resilience.

In a second simulation, we keep the network topology and the exposure sizes from the first simulation unchanged and choose capitals according to the formula in Proposition \ref{2:prop:robust:capital:requirements} with $\tau(w)=\max\{2,\lfloor\alpha w^\gamma\rfloor\}$ for $\alpha=\alpha_\text{c}(1+\delta)$, $\gamma=\gamma_\text{c}(1+\delta)$ and $\delta=8.39\%$ as in Table \ref{2:tab:deltaFromP}. As can be seen from Figure \ref{2:fig:capital:requirements}, already for typical network sizes of less than $10^4$, these capital allocations make the system resilient (note the axis scale). The maximal final fraction we observed was given by $1.33\%$. As mentioned before, the capital requirements in Proposition \ref{2:prop:robust:capital:requirements} are too robust in general, however. In another simulation, we hence choose capitals as determined in Theorem \ref{2:cor:threshold:res} again for $\tau(w)=\max\{2,\lfloor\alpha w^\gamma\rfloor\}$. Figure \ref{2:fig:capital:requirements} shows that under these requirements the fundamental defaults still do not spread through the network. All observed final fractions were less or equal $2.33\%$. However, keeping track of the total capitalization of the system further reveals that the capital requirements from Theorem \ref{2:cor:threshold:res} only amount to about $61\%$ of the ones from Proposition \ref{2:prop:robust:capital:requirements} for our chosen network parameters.

\section{Proofs}\label{2:sec:proofs}
\subsection{Proofs for Section \ref{2:random:graph}}\label{2:ssec:proofs:2}
\begin{proof}[Proof of Lemma \ref{2:lem:f:continuous}]
Continuity of $f$ follows directly from Lebesgue's dominated convergence, noting that $W^+$ is integrable by Assumption \ref{2:vertex:assump}. Further, \mbox{$f(0)=\E[W^+\1\{T=0\}]>0$} and $\lim_{z\to\infty}f(z;(W^-,W^+,T))=-\infty$. Hence by the intermediate value theorem function $f$ must have a positive root $\hat{z}$. Representation (\ref{2:eqn:integral:representation}) follows by an application of Fubini's theorem:
\begin{align}\label{2:eqn:integral:representation}
f(z) &= \E\left[W^+\1\{T=0\} + \int_0^z W^-W^+\P\left(\mathrm{Poi}(W^-\xi)=T-1\right)\1\{T\geq1\}\,\dd\xi\right]-z\nonumber\\
&= \E\left[W^+\1\{T=0\}\right] + \int_0^z\left(\E\left[W^-W^+\phi_T(W^-\xi)\right]-1\right)\dd\xi\qedhere
\end{align}
\end{proof}

\begin{proof}[Proof of Theorem~\ref{2:thm:asymp:1}]
We want to make use of Theorem \ref{2:thm:threshold:model} for the threshold model. 
Thus we describe an alternative description of default contagion compared to Subsection \ref{2:ssec:default:contagion:systemic:importance}:

At the beginning we declare all initially defaulted vertices to be \emph{defaulted} but yet \emph{unexposed}. At each step, a single defaulted, unexposed vertex $i \in [n]$ is picked and exposed to its neighbors, i.\,e.~weighted edges to its neighbors are drawn. If bank $j$ goes bankrupt due to the new edge that is sent from $i$, it is added to the set of defaulted, unexposed vertices. Otherwise, the capital of $j$ is reduced by the amount $e_{i,j}$. Afterwards, we remove $i$ from the set of unexposed vertices.

We keep track of the following sets and quantities at different steps $0\leq t\leq n-1$:
\begin{enumerate}[leftmargin=*,label=\alph*.]
\item $U(t)\subset[n]$: the unexposed vertices at step $t$. We set $U(0):=\{ i \in [n] \,:\, c_i=0 \}$.
\item $N (t)\subset[n]$: the solvent vertices at step $t$. At $t=0$, we set $N(0):=[n]\backslash U(0)$.
\item The updated capitals $\{\tilde{c}_i (t) \}_{i\in [n]}$ with $\tilde{c}_i(0)=c_i$ for all $i\in[n]$.
\end{enumerate}

\pagebreak
At step $t\in [n-1]$ the sets and quantities are updated according to the following scheme:
\begin{enumerate}[leftmargin=*]
\item \label{2:chose:rule} Choose a vertex $v\in U(t-1)$ according to any rule.
\item \label{2:update:capitals:rule} Expose $v$ to all of its neighbors in $N (t-1)$. That is, for all vertices $w\in N(t-1)$ set $\tilde{c}_w (t):=\max \{ 0 ,\tilde{c}_w (t-1) - e_{v,w} \}$. Note that $\tilde{c}_w (t)= \tilde{c}_w (t-1)$ if $e_{v,w}=0$.
\item \label{2:update:sets:rule} Set $N (t):= \{ i \in N(t-1) \,:\,  \tilde{c}_i (t)> 0 \}$ and $U(t):= ( U(t-1)\setminus \{v\} ) \cup \{ i \in N(t-1) \,:\,  \tilde{c}_i (t)= 0 \}$.
\end{enumerate}
Edges that are sent to already insolvent vertices are not exposed (but they could). Above steps are repeated until step $\hat{t}$, the first time that $U(t)=\emptyset$. Note that $\hat{t}$ is the final number of infected vertices independent of the rule chosen in Step~\ref{2:chose:rule}. Further, we can complete the exposition of the entire graph by exposing also links to defaulted vertices and links sent from vertices in $N (\hat{t})$.

Now, observe that the rule chosen in Step \ref{2:chose:rule}.~defines a permutation of the $\hat{t}$ elements of $[n]$ that go bankrupt. Further, for each $j \in [n]$ it defines an ordering of the set of insolvent vertices that send an edge to $j$, describing the order in which the edges are exposed. This ordering can be completed to a bijective map $\pi_j:[n-1]\to[n]\setminus \{ j \}$ by adding vertices that either send no edge to $j$ or are still solvent in the end. To be precise, let $\pi_j$ denote the ordering for vertex $j$ and let this vertex (after the exposition) have $l$ links sent from insolvent vertices. Then the entries $\pi_j(1),\dots , \pi_j(l)$ list defaulted neighbors in $[n]\backslash\{j\}$ in the order their edges are sent to vertex $j$. The entries $\pi_j(l+1),\dots,\pi_j(n-1)$ are, in their natural order, the remaining vertices in $[n]\backslash\{j\}$.

In order to reduce the model to the threshold model from Subsection \ref{2:ssec:special:case:threshold:model}, we now want to give a meaning to the so far only hypothetical threshold values $\tau_i$, $i\in[n]$. The idea is to construct a new random graph that has the same distribution (also of the threshold) as the graph constructed in Subsection~\ref{2:ssec:exposure:model} but with thresholds that have a direct meaningful interpretation:

We work on the same probability space as before but instead of assigning weight $E_{i,j}$ to a potential edge sent from $i\in[n]\backslash\{j\}$ to $j$, now the $i-$th ($i\in[n-1]$) edge that is sent to vertex $j$ during the sequential exposition described above shall receive weight $E_{\rho_j(i),j}$, where as before $\rho_j$ is the natural enumeration of $[n]\backslash\{j\}$. That is, edge-weights are not linked to the natural indices of their vertices anymore, but instead to the order of the exposition of the edges. One notes, however, that the random graph constructed that way has the same distribution as the random graph constructed before. To see this, observe that by the sequential procedure described by the orderings $\{\pi_j\}_{j\in[n]}$ and the assignment of exposures as described above, a potential edge sent from vertex $i$ to vertex $j$ is now assigned the edge-weight $E_{\rho_j(\pi_j^{-1}(i)),j}$. By exchangeability of the lists $\{E_{i,j}\}_{i\in[n]\backslash\{j\}}$ for $j\in[n]$, the new random variables $\{E_{\rho_j(\pi_j^{-1}(i)),j}\}_{i\in[n]\backslash\{j\}}$ have the same multivariate distribution as $\{E_{i,j}\}_{i\in[n]\backslash\{j\}}$. Obviously, also the new exposures are independent of the edge-indicator functions $\{X_{i,j}\}_{i,j\in[n]}$.

Hence, both constructions result in the same distribution for the random graph. Further, note that the assignment of edge-weights has been conducted in such a way that the threshold values in both versions of the network coincide. As before, they are given as
\[ \tau_i(n):= \inf \{ s \in \{0\}\cup[n-1] \,:\, \sum_{l\leq s} E_{\rho_i(l),i} \geq c_i \},\quad i\in[n].\]
In the new random graph, however, the thresholds $\tau_i$ have the interpretation of actual thresholds meaning that bank $i$ goes bankrupt at the $\tau_i$-th of one of its neighbors' default. The sequential description of the cascade process has then the advantage that we can reduce it to the threshold model as described in Subsection \ref{2:ssec:special:case:threshold:model}. We can replace the capitals $\tilde{c}_i (t)$, which represent monetary thresholds, by integer values $\tilde{\tau}_i(t)$ (we set $\tilde{\tau}_i(0):=\tau_i$), which count numbers of neighbors, and alter Steps \ref{2:update:capitals:rule}.~and \ref{2:update:sets:rule}.~in the description of the sequential cascade process according to the rule that if there is an edge sent from $v$ to $w$ ($e_{v,w}>0$), then set $\tilde{\tau}_w (t):=\tilde{\tau}_w (t-1) - 1$. If there is no edge from $v$ to $w$ ($e_{v,w}=0$), set $\tilde{\tau}_w (t):= \tilde{\tau}_w (t-1)  $.
Then the sets $N (t)$ and $U(t)$ are defined by $N (t):= \{ i \in N(t-1) \,:\,  \tilde{\tau}_i (t)> 0 \}$ respectively $U(t):= \left(U(t-1)\backslash\{v\}\right)\cup \{ i \in N(t-1) \,:\,  \tilde{\tau}_i (t)= 0 \}$. Everything else stays unchanged. Note that the resulting threshold values $\tilde{\tau}_i(t)$ are only valid for an exposition in the order as specified above. In the threshold model, however, we are free to choose exactly the same rule as we chose in Step \ref{2:chose:rule}.~of our model since this does not affect the final set of defaulted vertices. Hence, we can replace our exposure model by the threshold model from Subsection \ref{2:ssec:special:case:threshold:model}, resulting in the same final set of defaulted vertices.
In Theorem \ref{2:thm:threshold:model}, a regularity condition on the capital (here threshold) distribution $T$ is required. This is ensured after conditioning on the values of $\{\tau_i\}_{i\in[n]}$ by Assumption~\ref{2:vertex:assump} that $G_n(x,y,v,l)$ converges to $G(x,y,v,l)$ almost surely for all $(x,y,v,l)$. Applying Theorem \ref{2:thm:threshold:model} hence yields the desired statement.
\end{proof}

\begin{proof}[Proof of Theorem~\ref{2:thm:asymp:2}]
Part \ref{2:thm:asymp:2:1}.~follows from Theorem \ref{2:thm:threshold:model} by the same arguments as before.

In order to prove the second part, we will apply an additional small shock to the system such that each bank $i$, regardless of its attributes $w_i^-$, $w_i^+$ and $\tau_i$, has its capital $c_i$ and hence its threshold $\tau_i$ set to $0$ with probability $p$, where $p$ is some fixed small number. The new limiting distribution of the system is then given by $(W^-,W^+,TM_p)$, where $M_p$ is a \mbox{$\{0,1\}$-valued} random variable independent of $(W^-,W^+,T)$ and with $\P(M_p=0)=p$. Instead of \mbox{$f(z)=f(z;(W^-,W^+,T))$} we then have to consider the function
\[ f_p(z):=f(z;(W^-,W^+,TM_p)) = p(\E[W^+]-z)+(1-p)f(z). \]
Assuming $P(T=0)<1$ (the case $\P(T=0)=1$ is trivial), it holds $f_p(z)>f(z)$ and hence we conclude that the first positive root $\hat{z}_p$ of $f_p(z)$ is larger than $z^*$. By definition of $z^*$ we further derive that $\hat{z}_p\to z^*$ as $p\to0$. The idea is therefore to choose $p$ in such a way that $\hat{z}_p$ satisfies the assumptions of Theorem \ref{2:thm:asymp:1} and then conclude by coupling the original system with the additionally shocked one to derive $n^{-1}\mathcal{S}_n\leq n^{-1}\mathcal{S}_n^{(p)}$, where $\mathcal{S}_n^{(p)}:=\sum_{i\in\mathcal{D}_n^{(p)}}s_i$ and $\mathcal{D}_n^{(p)}$ denotes the set of finally defaulted vertices in the additionally shocked system. Since $z^*$ is a root of the continuously differentiable function $f(z)$ it must hold $d(z^*)\leq0$. We distinguish two cases:

In the first case, we assume that $\kappa:=d(z^*)<0$. Then by continuity of $d(z)$ on a neighborhood of $z^*$, also
\[ d_p(z):=\E\left[W^-W^+\phi_{TM_p}(W^-z)\right]-1 
 < \frac{\kappa}{2} < 0 \]
on a neighborhood of $\hat{z}_p$ for $p$ small enough. As indicated above, an application of Theorem \ref{2:thm:asymp:1} together with a coupling argument then yields
\begin{align*}
n^{-1}\mathcal{S}_n &\leq n^{-1}\mathcal{S}_n^{(p)}= \E\left[S\psi_{TM_p}(W^-\hat{z}_p)\right] + o_p(1)\leq \E\left[S\psi_T(W^-\hat{z}_p)\right] + \E[S\1_{\{M_p=0\}}] + o_p(1)\\
&\leq \E\left[S\psi_T(W^-z^*)\right] + \epsilon + o_p(1)
\end{align*}
by continuity of $\E\left[S\psi_T(W^-z)\right]$ ($S$ is assumed integrable) and choosing $p$ small enough.

In the second case, we have $0 = d(z^*) = \lim_{\epsilon\to0}\epsilon^{-1} f(z^*+\epsilon)$. For $\tilde{\epsilon}>0$, let
\[ \delta(\tilde{\epsilon}) :=  -\inf_{0<\epsilon\leq\tilde{\epsilon}}\epsilon^{-1}f(z^*+\epsilon),\]
which is positive for all $\tilde{\epsilon}$ by definition of $z^*$. We can therefore find $\tilde{\epsilon}>0$ arbitrarily small such that $\delta(\epsilon)<\delta(\tilde{\epsilon})$ for all $\epsilon<\tilde{\epsilon}$. We then derive that
\[ 0 \leq f(z^*+\epsilon) + \delta(\epsilon)\epsilon \leq f(z^*+\epsilon) + \delta(\tilde{\epsilon})\epsilon \]
for all $\epsilon\leq\tilde{\epsilon}$ with equality only for $\epsilon=\tilde{\epsilon}$. Hence at $\epsilon=\tilde{\epsilon}$ the derivative of the last term must be non-positive, i.\,e.~$d(z^*+\tilde{\epsilon})\leq -\delta(\tilde{\epsilon}) < 0$. By continuity, also $d(z) \leq -\delta(\tilde{\epsilon})/2<0$ on a neighborhood of $z^*+\tilde{\epsilon}$. Hence $z^*+\tilde{\epsilon}$ is a good candidate for the first positive root of the additionally shocked system. All that is left to show is that there exists a certain value for the shock size $p$ such that $z^*+\tilde{\epsilon}$ becomes the first positive root. To this end, let
\[ p(\tilde{\epsilon}) := \frac{\tilde{\epsilon}\delta(\tilde{\epsilon})}{\E[W^+]-z^*-\tilde{\epsilon}(1-\delta(\tilde{\epsilon}))}. \]
Note that for $\P(T=0)<1$ the root $z^*$ is always less than $\E[W^+]$ and hence for $\tilde{\epsilon}$ small enough $p(\tilde{\epsilon})$ becomes positive. As $\tilde{\epsilon}\to0$, also $p(\tilde{\epsilon})$ tends to zero. Now note that for all $0<\epsilon\leq\tilde{\epsilon}$,
\[ f_{p(\tilde{\epsilon})}(z^*+\epsilon) \geq (1-p(\tilde{\epsilon}))(-\epsilon\delta(\tilde{\epsilon})) + p(\tilde{\epsilon})\left(\E[W^+]-(z^*+\epsilon)\right) \geq (\tilde{\epsilon}-\epsilon)\delta(\tilde{\epsilon}) \geq 0 \]
with equality only for $\epsilon=\tilde{\epsilon}$. The additional shock strictly increases $f(z)$ and hence there cannot be any root less or equal $z^*$. In particular, $z^*+\tilde{\epsilon}$ is the first positive root of the additionally shocked system. By letting $\tilde{\epsilon}\to0$, we conclude for arbitrarily small $\epsilon>0$ that
\[ n^{-1}\mathcal{S}_n \leq \E\left[S\psi_T(W^-z^*)\right] + \epsilon + o_p(1). \]
For the case of $\hat{z}=z^*$, we simply need to combine Parts \ref{2:thm:asymp:2:1}.~and \ref{2:thm:asymp:2:2}.~of the theorem.
\end{proof}

\subsection{Proofs for Section \ref{2:sec:resilience}}\label{2:ssec:proofs:3}
\begin{proof}[Proof of Proposition \ref{2:prop:convergence:speed}]
By (\ref{2:eqn:integral:representation}), we derive
\[ f_i(z):=f(z;(W^-,W^+,TM_i))\leq \E[W^+\1\{M_i=0\}] + \kappa z + o(z) \]
and similarly
\[ \E[S\psi_{TM_i}(W^-z)] \leq \E[S\1\{M_i=0\}] + \kappa_S z + o(z). \]
Hence we derive
\[ \hat{z}_i \leq -\kappa^{-1}\E[W^+\1\{M_i=0\}] + o(\E[W^+\1\{M_i=0\}]), \]
where $\hat{z}_i$ denotes the first positive root of $f_i(z)$. Together with Theorem \ref{2:thm:asymp:1} this shows that
\[ \text{w.\,h.\,p.} \quad n^{-1}\mathcal{S}_n^{M_i} \leq \E[S\1\{M_i=0\}] - \kappa^{-1}\kappa_S\E[W^+\1\{M_i=0\}] + o(\E[W^+\1\{M_i=0\}]). \]
If $f(z)$ and $\E[S\psi_T(W^-z)]$ are continuously differentiable from the right at $z=0$ with derivatives $\kappa<0$ and $\kappa_S<\infty$, then above inequalities are equalities and hence
\[ n^{-1}\mathcal{S}_n^{M_i} \xrightarrow{p} \E[S\1\{M_i=0\}] - \kappa^{-1}\kappa_S\E[W^+\1\{M_i=0\}] + o(\E[W^+\1\{M_i=0\}]).\qedhere \]
\end{proof}

\begin{proof}[Proof of Theorem \ref{2:thm:cont:res}]
By \eqref{2:thm:cont:res:ass}, we derive that $z^*_i\to0$ as $i\to\infty$, where
\[ z^*_i := \inf\left\{ z>0\,:\,f(z;(W^-,W^+,TM_i))<0\right\}. \]
For $i$ large enough such that $z^*_i<z_0$, we can then apply Part \ref{2:thm:asymp:2:2}.~of Theorem \ref{2:thm:asymp:2} to derive that
\[ \text{w.\,h.\,p.} \quad n^{-1}\mathcal{S}_n^{M_i} \leq \E\left[S\psi_{TM_i}(W^-z^*_i)\right] + \frac{\epsilon}{2} \leq \E\left[S\psi_T(W^-z^*_i)\right] + \E\left[S\1\{M_i=0\}\right] + \frac{\epsilon}{2}.\]
Note that from continuity of $d(z)$ it follows that also $\E\left[W^-W^+\phi_{TM_i}(W^-z)\right]$ is continuous by dominated convergence. Since $S$ is integrable, the first summand tends to zero as $z^*_i\to0$ and also the second summand vanishes as $\P(M_i=0)\to0$. In particular, we can choose $i$ large enough such that
\[ \E\left[S\psi_T(W^-z^*_i)\right] + \E\left[S\1\{M_i=0\}\right]\leq\frac{\epsilon}{2}.\qedhere \]
\end{proof}

We now turn to the proofs of Theorems \ref{2:threshold:res} -- \ref{2:thm:Pareto:type}. To show resilience of the financial system in Theorem \ref{2:threshold:res}, we want to use Theorem \ref{2:thm:cont:res}. In order for this to work, we need to ensure that $d(z)$ is continuous for $z>0$. This is done in the following lemma.

\begin{lemma}\label{2:lem:cont:differentiability}
Assume that $T=\tau(W^-)$ for some integer-valued function $\tau(w)=o(w)$. Then $d(z)$ is continuous on $(0,\infty)$.
\end{lemma}
\begin{proof}
We fix some $\tilde{z}<\infty$ and aim to show that for some arbitrarily fixed $\delta<\tilde{z}$ the family $\left\{W^-W^+\phi_T(W^-z)\right\}_{z\in[\tilde{z}-\delta,\tilde{z}+\delta]}$ is bounded by some integrable random variable almost surely. This will show continuity of $d(z)$ by Lebesgue's dominated convergence theorem.

By definition of a Poisson random variable, we derive
\[ W^-W^+\phi_T(W^-z) = W^-W^+\P\left(\mathrm{Poi}(W^-z)=T-1\right) = W^-W^+e^{-W^-z}\frac{(W^-z)^{T-1}}{\Gamma(T)}. \]
Then, by applying Stirling's approximation to the $\Gamma$-function,
\[ W^-W^+\phi_T(W^-z) \leq W^-W^+\exp\left\{-W^-z\left(1-\frac{T-1}{W^-z}+\frac{T-1}{W^-z}\log\left(\frac{T-1}{W^-z}\right)\right)\right\}. \]
In the exponent, we identify the expression $g\left((T-1)/(W^-z)\right)$, where $g(x):=1-x+x\log(x)$. The continuous function $g$ admits the unique minimum $g(x^*)=0$ at $x^*=1$. Since further $\lim_{x\to0+}g(x)=1$, it holds that $g(x)\geq G$ for $x<1/2$ and some suitable $G>0$. Now choose $\tilde{w}$ large enough such that $(\tau(w)-1)/(w(\tilde{z}-\delta))<1/2$ for all $w>\tilde{w}$. We derive,
\[ W^-W^+\phi_T(W^-z) \leq W^+\left(W^-\exp\left\{-W^-(\tilde{z}-\delta)G\right\}+\tilde{w}\right) \leq W^+M(\tilde{z}) \in L^1, \]
almost surely, where $M(\tilde{z})$ is a positive constant depending on $\tilde{z}$ only.
\end{proof}

\begin{remark}\label{2:rem:cont:continuous:differentiability}
In Theorem \ref{2:threshold:res}, we have $\tau(w)=\mathcal{O}(w^\gamma)$ for $0\leq\gamma<1$ and hence $\tau(w)=o(w)$.
\end{remark}

\begin{proof}[Proof of Theorem \ref{2:threshold:res}]
In order to ease notation, we will assume throughout all the proofs that $w_\text{min}^-=w_\text{min}^+=1$. The arguments for general $w_\text{min}^-$ and $w_\text{min}^+$ are completely analogue.

Since $\E[W^-W^+]$ is maximized for comonotone weights, i.\,e.~$W^+=F_{W^+}^{-1}(F_{W^-}(W^-))=(W^-)^\frac{\beta^--1}{\beta^+-1}$, we get $\E\left[W^-W^+\right] \leq \E\left[(W^-)^{1+\frac{\beta^--1}{\beta^+-1}}\right] = (\beta^--1)\int_1^\infty w^{\gamma_\text{c}-1}\dd w < \infty$ for $\gamma_\text{c}<0$. By dominated convergence we conclude that $f$ is continuously differentiable on $[0,\infty)$ with $f'(z)=d(z)$. By $T\geq2$, in particular, $f'(0)=-1$ and hence by Theorem \ref{2:prop:res} the system is resilient to small shocks.

Now let $\gamma_\text{c}=0$ and $\alpha:=\liminf_{w\to\infty}\tau(w)>\alpha_\text{c}+1$. We assume that $\alpha<\infty$, otherwise we could truncate $\tau(w)$ at some $\N\ni\alpha>\alpha_\text{c}+1$. Since $\tau(w)\geq2$ is an integer-valued function, we observe that $\alpha\in\N\backslash\{0,1\}$ and $\tau(w)\geq\alpha$ for all $w>\tilde{w}$ and some constant $\tilde{w}>0$. Since $\psi_l(x)$ is monotonically decreasing in $l$, we derive that, as $z\to0$,
\begin{align*}
\E\left[W^+\psi_T(W^-z)\right] &\leq \E\left[W^+\psi_\alpha(W^-z) + W^+\psi_2(W^-z)\1\{W^-\leq\tilde{w}\}\right]\leq 
\E\left[W^+\psi_\alpha(W^-z)\right] + o(z).
\end{align*}
Note that since $\psi_\alpha(x)$ is a strictly increasing function in $x$, this expression becomes maximized for comonotone dependence of $W^-$ and $W^+$. We derive
\begin{align*}
\limsup_{z\to0+}z^{-1}\E\left[W^+\psi_T(W^-z)\right] &\leq \limsup_{z\to0+}z^{-1} \sum_{k\geq\alpha} \E\left[(W^-)^\frac{\beta^--1}{\beta^+-1}e^{-W^-z}\frac{(W^-z)^k}{k!}\right]\\
&= \limsup_{z\to0+}\alpha_\text{c}\sum_{k\geq\alpha} \int_z^\infty \frac{x^{\frac{\beta^--1}{\beta^+-1}-\beta^-+k}e^{-x}}{k!}\dd x\\
&= \alpha_\text{c}\sum_{k\geq\alpha} \frac{\Gamma(k-1)}{k!} = \frac{\alpha_\text{c}}{\alpha-1}<1.
\end{align*}
In particular, $\limsup_{z\to0+}z^{-1}f(z)<0$. By Lemma \ref{2:lem:cont:differentiability} and Remark \ref{2:rem:cont:continuous:differentiability}, we also know that $d(z)$ is continuous for $z>0$ (we can simply cut off $\tau(w)$ at $\alpha$) and hence by Theorem \ref{2:thm:cont:res} we can conclude that the system must be resilient.

Finally, assume that $\gamma_\text{c}>0$ and $\alpha:=\liminf_{w\to\infty}w^{-\gamma_\text{c}}\tau(w) > \alpha_\text{c}$. We can then choose some $\alpha_\text{c}<\tilde{\alpha}<\alpha$ and $\tilde{w}<\infty$ such that $\tau(w)\geq\left\lceil\tilde{\alpha}w^{\gamma_\text{c}}\right\rceil$ for all $w>\tilde{w}$. Hence we derive that
\[ \E\left[W^+\psi_T(W^-z)\right] \leq \E\left[W^+\psi_{\left\lceil\tilde{\alpha}(W^-)^{\gamma_\text{c}}\right\rceil}(W^-z)\1\{W^->\tilde{w}\}\right] + o(z). \]
By a Chernoff bound we get that $\psi_l(x)\leq(xe/l)^le^{-x}$ for $x<l$. Thus for $w\leq\left(\tilde{\alpha}^{-1}(1+\epsilon)z\right)^\frac{1}{\gamma_\text{c}-1}$ and $\epsilon>0$,
\[ \psi_{\left\lceil\tilde{\alpha}w^{\gamma_\text{c}}\right\rceil}(wz) \leq \left(\frac{wze}{\tilde{\alpha}w^{\gamma_\text{c}}}\right)^{\tilde{\alpha}w^{\gamma_\text{c}}}e^{-wz} = \exp\left\{-z^\frac{\gamma_\text{c}}{\gamma_\text{c}-1}g\left(wz^\frac{1}{1-\gamma_\text{c}}\right)\right\}, \]
where $g(x) := x-\tilde{\alpha}x^{\gamma_\text{c}}\log(ex^{1-\gamma_\text{c}}/\tilde{\alpha})$. For arbitrary $\lambda>0$, we can hence choose $\tilde{w}$ large enough such that for all $\tilde{w}<w\leq\left(\tilde{\alpha}^{-1}(1+\epsilon)ez\right)^\frac{1}{\gamma_\text{c}-1}$ it holds
\[ \psi_{\left\lceil\tilde{\alpha}w^{\gamma_\text{c}}\right\rceil}(wz) \leq \left(\frac{wze}{\tilde{\alpha}w^{\gamma_\text{c}}}\right)^{\tilde{\alpha}w^{\gamma_\text{c}}} \leq \left(\frac{wze}{\tilde{\alpha}w^{\gamma_\text{c}}}\right)^{1+\lambda}(1+\epsilon)^{-\tilde{\alpha}w^{\gamma_\text{c}}+1+\lambda} \leq \left(\tilde{\alpha}^{-1}ze\right)^{1+\lambda} \]
and
\[ \E\left[W^+\psi_{\left\lceil\tilde{\alpha}(W^-)^{\gamma_\text{c}}\right\rceil}(W^-z)\1\left\{\tilde{w}<W^-\leq \left(\tilde{\alpha}^{-1}(1+\epsilon)ez\right)^\frac{1}{\gamma_\text{c}-1}\right\}\right] \leq \left(\tilde{\alpha}^{-1}ze\right)^{1+\lambda}\E[W^+] = o(z). \]
For $(\tilde{\alpha}^{-1}(1+\epsilon)ez)^\frac{1}{\gamma_\text{c}-1} < w \leq (\tilde{\alpha}^{-1}(1+\epsilon)z)^\frac{1}{\gamma_\text{c}-1}$, we have $g(wz^\frac{1}{1-\gamma_\text{c}}) \geq \delta$ for some $\delta>0$. Thus
\begin{align*}
&\E\left[W^+\psi_{\left\lceil\tilde{\alpha}(W^-)^{\gamma_\text{c}}\right\rceil}(W^-z)\1\left\{\left(\tilde{\alpha}^{-1}(1+\epsilon)ez\right)^\frac{1}{\gamma_\text{c}-1}<W^-\leq\left(\tilde{\alpha}^{-1}(1+\epsilon)z\right)^\frac{1}{\gamma_\text{c}-1}\right\}\right]\\
&\hspace{11cm}\leq \E[W^+]e^{-\delta z^\frac{\gamma_\text{c}}{\gamma_\text{c}-1}} = o(z).
\end{align*}
Hence, as $z\to0$, only $W^->\left(\tilde{\alpha}^{-1}(1+\epsilon)z\right)^\frac{1}{\gamma_\text{c}-1}$ contributes to $z^{-1}\E[W^+\psi_T(W^-z)]$. On this set, by bounding with the comonotone dependence, we compute
\begin{align}\label{2:eqn:expectation:upper:part}
&\E\Bigg[W^+\psi_T(W^-z)\1\left\{W^->\left(\tilde{\alpha}^{-1}(1+\epsilon)z\right)^\frac{1}{\gamma_\text{c}-1}\right\}\Bigg] \leq \E\Bigg[W^+\1\left\{W^->\left(\tilde{\alpha}^{-1}(1+\epsilon)z\right)^\frac{1}{\gamma_\text{c}-1}\right\}\Bigg]\\
&\hspace{3cm}\leq (\beta^--1)\int_{\left(\tilde{\alpha}^{-1}(1+\epsilon)z\right)^\frac{1}{\gamma_\text{c}-1}}^\infty w^{\frac{\beta^--1}{\beta^+-1}-\beta^-}\dd w = \alpha_\text{c}\frac{(1+\epsilon)z}{\tilde{\alpha}}\nonumber
\end{align}
Hence, $\limsup_{z\to0+}z^{-1}f(z) = \limsup_{z\to0+}z^{-1}\E[W^+\psi_T(W^-z)]-1<0$ by choosing $\epsilon>0$ small enough such that $\tilde{\alpha}^{-1}\alpha_\text{c}(1+\epsilon)<1$. This shows resilience by the same arguments as in Part \ref{2:threshold:res:2}., noting that we can cut $\tau(w)$ at $w^\eta$ for some $\gamma_\text{c}<\eta<1$.
\end{proof}

\begin{proof}[Proof of Theorem \ref{2:thm:functional:nonres}]
Again, we simplify notation be setting $w_\text{min}^-=w_\text{min}^+=1$.

We start by proving the second statement. Let $\alpha:=\limsup_{w\to\infty}w^{-\gamma_\text{c}}\tau(w)$ and choose $\alpha < \tilde{\alpha} < \int_0^\infty\Lambda\left(x^{1-\beta^+}\right)\dd x$ and $\tilde{w}<\infty$ such that $\tau(w)\leq\left\lfloor\tilde{\alpha}w^{\gamma_\text{c}}\right\rfloor$ for all $w>\tilde{w}$. Moreover, choose $\epsilon>0$ and $\delta>0$ such that $\tilde{\alpha}<(1-\epsilon)(1-\delta)\int_0^\infty\Lambda\big(x^{1-\beta^+}\big)\dd x$ and let $z>0$ small enough such that $\tilde{w}\leq\left(\tilde{\alpha}^{-1}(1-\epsilon)z\right)^\frac{1}{\gamma_\text{c}-1}$ as well as $z<(\epsilon^2\delta)^{\frac{1}{\gamma_\text{c}}-1}(\tilde{\alpha}/(1-\epsilon))^\frac{1}{\gamma_\text{c}}$. For such $z$ and $w>(\tilde{\alpha}^{-1}(1-\epsilon)z)^\frac{1}{\gamma_\text{c}-1}$ it holds that
\[ \P\left(\mathrm{Poi}(wz)\geq\left\lfloor\tilde{\alpha}w^{\gamma_\text{c}}\right\rfloor\right) \geq 1-\P\left(\left\vert\mathrm{Poi}(wz)-wz\right\vert \geq \epsilon wz\right) \geq 1-\frac{1}{\epsilon^2wz} \geq 1-\frac{1-\epsilon}{\epsilon^2\tilde{\alpha}w^{\gamma_\text{c}}} \geq 1-\delta, \]
by Chebyshev's inequality. Therefore,
\begin{align}\label{2:eqn:int:cond:prob}
&\E\left[W^+\psi_T(W^-z)\right] \geq (1-\delta)\E\left[W^+\1\left\{W^->\left(\tilde{\alpha}^{-1}(1-\epsilon)z\right)^\frac{1}{\gamma_\text{c}-1}\right\}\right]\nonumber\\
&\hspace{0.382cm}= (1-\delta)\int_0^\infty \P\left(W^+>x,W^->\left(\tilde{\alpha}^{-1}(1-\epsilon)z\right)^\frac{1}{\gamma_\text{c}-1}\right)\dd x\nonumber\\
&\hspace{0.382cm}= (1-\delta)\left(\tilde{\alpha}^{-1}(1-\epsilon)z\right)^{\frac{1}{\gamma_\text{c}-1}\frac{\beta^--1}{\beta^+-1}}\nonumber\\
&\hspace{3.675cm} \times\int_0^\infty \P\left(F_{W^+}(W^+)>1-x^{1-\beta^+}p(z),F_{W^-}(W^-)>1-p(z)\right)\dd x\nonumber\\
&\hspace{0.382cm}= \tilde{\alpha}^{-1}(1-\epsilon)(1-\delta)z \int_0^\infty \P\left(F_{W^+}(W^+)>1-x^{1-\beta^+}p(z) \,\middle\vert\, F_{W^-}(W^-)>1-p(z)\right)\dd x,
\end{align}
where we substituted $p(z) := (\tilde{\alpha}^{-1}(1-\epsilon)z)^\frac{1-\beta^-}{\gamma_\text{c}-1}$. Note that the conditional probability is bounded by $1\wedge x^{1-\beta^+}$. Hence, by Lebesgue's dominated convergence theorem
\[ z^{-1}\E\left[W^+\psi_T(w^-z)\right] \geq \tilde{\alpha}^{-1}(1-\epsilon)(1-\delta)\int_0^\infty\Lambda\left(x^{1-\beta^+}\right)\dd x + o(1) > 1 + o(1) \]
and thus $\E\left[W^+\psi_T(W^-z)\right]>z$ for $z$ small enough which implies non-resilience by Theorem \ref{2:prop:nonres}.

For resilience in Part \ref{2:thm:functional:nonres:2}., we follow the same proof as for Theorem \ref{2:threshold:res} until we arrive at expression (\ref{2:eqn:expectation:upper:part}) which we can now evaluate with the same means as above. That is,
\[ z^{-1}\E\bigg[W^+\1\bigg\{W^->\left(\tilde{\alpha}^{-1}(1+\epsilon)z\right)^\frac{1}{\gamma_\text{c}-1}\bigg\}\bigg] \to \tilde{\alpha}^{-1}(1+\epsilon)\int_0^\infty\Lambda\left(x^{1-\beta^+}\right)\dd x <1,\quad\text{as }z\to0. \]
For the first statement of the theorem, note that we can lower bound the integral in (\ref{2:eqn:int:cond:prob}) by $\P(F_{W^+}(W^+)>1-p(z) \mid F_{W^-}(W^-)>1-p(z))$ and hence we derive non-resilience as above by
\[ \liminf_{z\to0}z^{-1}\E[W^+\psi_T(W^-z)] \geq \tilde{\alpha}^{-1}(1-\epsilon)(1-\delta) \lambda > 1.\qedhere \]
\end{proof}

\begin{proof}[Proof of Theorem \ref{2:thm:Pareto:type}]
Let $U\sim\mathcal{U}[0,1]$ and
\[ X^\pm := (1-U)^\frac{1}{1-\beta^\pm}K^\pm. \]
Then $X^\pm\sim\mathrm{Par}(\beta^\pm,K^\pm)$. Further, let $\tilde{W}^\pm := F_{W^\pm}^{-1}(U)$ such that $\tilde{W}^\pm\stackrel{d}{=} W^\pm$ but $\tilde{W}^-$ and $\tilde{W}^+$ are comonotone. By \eqref{2:eqn:Pareto:type} we then now that $\tilde{W}^\pm=F_{W^\pm}^{-1}(U)\leq F_{X^\pm}^{-1}(U)=X^\pm$ for $U\geq\tilde{u}$ large enough and hence
\[ \E[W^-W^+]\leq\E[\tilde{W}^-\tilde{W}^+]\leq \tilde{w}(\E[W^-]+\E[W^+]) + \E[X^-X^+], \]
where $\tilde{w}=\max\{F_{W^-}^{-1}(\tilde{u}),F_{W^+}^{-1}(\tilde{u})\}<\infty$. Then for the case of $\gamma_c<0$, by the same calculations as in the proof of Theorem \ref{2:threshold:res} we derive $\E[X^-X^+]<\infty$ and thus resilience of the system.

For $\gamma_c=0$, as in the proof of Theorem \ref{2:threshold:res} we have
\[ \E[W^+\psi_T(W^-z)] \leq \E[W^+\psi_\alpha(W^-z)]+o(z), \]
where $\alpha=\liminf_{w\to\infty}\tau(w)\in\N\backslash\{0,1\}$. Hence
\begin{align*}
\limsup_{z\to0+} \frac{\E[W^+\psi_T(W^-z)]}{z} &\leq \limsup_{z\to0+} \frac{\E[\tilde{W}^+\psi_\alpha(\tilde{W}^-z)]}{z}\\
&\leq \limsup_{z\to0+} \frac{\E[X^+\psi_\alpha(X^-z)] + \tilde{w}\psi_2(\tilde{w}z)}{z}\\
&=\frac{\beta^+-1}{\beta^+-2}K^+K^-\frac{1}{\alpha-1}
\end{align*}
and the system is resilient for $\liminf_{w\to\infty}\tau(w)=\alpha>\frac{\beta^+-1}{\beta^+-2}K^+K^-+1$.

For $\gamma_c>0$, we follow the proof of Theorem \ref{2:threshold:res} and replace the right-hand side of \eqref{2:eqn:expectation:upper:part} by
\begin{align*}
&\E\bigg[X^+\1\left\{X^->\left(\tilde{\alpha}^{-1}(1+\epsilon)z\right)^\frac{1}{\gamma_\text{c}-1}\right\}\bigg] + \tilde{w}\1\left\{\tilde{w}>\left(\tilde{\alpha}^{-1}(1+\epsilon)z\right)^\frac{1}{\gamma_\text{c}-1}\right\}\\
&\hspace{8cm}= \frac{\beta^+-1}{\beta^+-2}K^+(K^-)^{1-\gamma_c}\frac{(1+\epsilon)z}{\tilde{\alpha}} + o(z).\qedhere
\end{align*}
\end{proof}

\begin{proof}[Proof of Proposition \ref{2:prop:robust:capital:requirements}]
The capitals $c_i$ are chosen such that the threshold values $\tau_i$ are at least $\tau(w_i^-)$. Hence coupling the weighted network to the corresponding threshold network yields the result.
\end{proof}

\begin{proof}[Proof of Theorem \ref{2:cor:threshold:res}]
By the assumption of $c_i>\max_{j\in[n]\backslash\{i\}}E_{j,i}$ almost surely for all \mbox{$i\in[n]$,} we get $\tau_i\geq2$ almost surely. The proof of Part \ref{2:cor:threshold:res:1}.~is thus completely analogue to the one of Theorem \ref{2:threshold:res}.

We continue by proving Part \ref{2:cor:threshold:res:3}. By the means of the proof of Theorem \ref{2:threshold:res} we derive that $\limsup_{z\to0+}z^{-1}\E\left[W^+\psi_T(W^-z)\1\left\{T>(1+\epsilon)\alpha_\text{c}(W^-)^{\gamma_\text{c}}\right\}\right] < 1$ for each $\epsilon>0$. It will hence suffice to prove $\E\left[W^+\psi_T(W^-z)\1\left\{T\leq(1+\epsilon)\alpha_\text{c}(W^-)^{\gamma_\text{c}}\right\}\right] = o(z)$ in order to show resilience. To this end, choose $0<\delta<\gamma_\text{c}(t-1)$. Then
\begin{align*}
\E\Big[W^+\psi_T(W^-z)\1\Big\{T\leq(W^-)^\delta\Big\}\Big] &\leq \E\left[W^+\psi_2(W^-z)\1\left\{T\leq(W^-)^\delta\right\}\right]\\
&= \lim_{M\to\infty}\E\left[\left(W^+\wedge M\right)\psi_2(W^-z)\1\left\{T\leq(W^-)^\delta\right\}\right]\\
&\leq \liminf_{M\to\infty} \liminf_{n\to\infty} n^{-1}\sum_{i\in[n]}\left(w_i^+\wedge M\right)\psi_2(w_i^-z)\1\left\{\tau_i\leq(w_i^-)^\delta\right\},
\end{align*}
where we approximated the non-continuous integrand by continuous ones and used almost sure weak convergence by Assumption \ref{2:vertex:assump}. Now taking expectations with respect to the exposure lists, by Fatou's lemma we derive
\begin{align*}
\E\left[W^+\psi_T(W^-z)\1\left\{T\leq(W^-)^\delta\right\}\right] &\leq \liminf_{M\to\infty} \liminf_{n\to\infty} n^{-1}\sum_{i\in[n]}\left(w_i^+\wedge M\right)\psi_2(w_i^-z)\P\left(\tau_i\leq(w_i^-)^\delta\right)\\
&\leq K_1 \liminf_{M\to\infty} \liminf_{n\to\infty} n^{-1}\sum_{i\in[n]}\left(w_i^+\wedge M\right)\psi_2(w_i^-z)(w_i^-)^{\delta-t\gamma_\text{c}}\\
&= K_1 \liminf_{M\to\infty} \E\left[\left(W^+\wedge M\right)\psi_2(W^-z)(W^-)^{\delta-t\gamma_\text{c}}\right]\\
&= K_1 \E\left[W^+\psi_2(W^-z)(W^-)^{\delta-t\gamma_\text{c}}\right],
\end{align*}
where we used Assumption \ref{2:ass:exposures} to bound
\begin{align*}
\P\left(\tau_i\leq(w_i^-)^\delta\right) &= \P\Bigg(\sum_{1\leq j\leq(w_i^-)^\delta}E_{\rho_i(j),i}\geq c_i\Bigg) \leq \P\Bigg(\sum_{1\leq j\leq(w_i^-)^\delta}E_{\rho_i(j),i}\geq \tau(w_i^-)\lambda_i\Bigg)\\
& \leq \P\Bigg(\sum_{1\leq j\leq(w_i^-)^\delta}E_{\rho_i(j),i}\geq \alpha_\text{c}(w_i^-)^{\gamma_\text{c}}\lambda_i\Bigg)\leq K_1(w_i^-)^{\delta-t\gamma_\text{c}},
\end{align*}
for $w_i^-$ large enough and some uniform constant $K_1>\infty$. Note that for $W^-\leq\tilde{w}$, we have $\E\left[W^+\psi_T(W^-z)\1\{W^-\leq\tilde{w}\}\right] = o(z)$ as in the proof of Theorem \ref{2:threshold:res}. Hence it holds
\[ z^{-1}\E\left[W^+\psi_T(W^-z)\1\left\{T\leq(W^-)^\delta\right\}\right] \leq K_1 \E\left[W^+\frac{\psi_2(W^-z)}{W^-z}(W^-)^{1+\delta-t\gamma_\text{c}} \right] + o(1). \]
Since $\psi_2(x)=o(x)$, by dominated convergence it is enough to prove $\E\left[W^+(W^-)^{1+\delta-t\gamma_\text{c}} \right]<\infty$ in order for $\E\left[W^+\psi_T(W^-z)\1\left\{T\leq(W^-)^\delta\right\}\right]=o(z)$. We can easily choose $t>1$ and $\delta>0$ in such a way that $1+\delta-t\gamma_\text{c}>0$ and can therefore estimate $\E[W^+(W^-)^{1+\delta-t\gamma_\text{c}} ]$ by the comonotone expectation $\E\big[(W^-)^{\frac{\beta^--1}{\beta^+-1}+1+\delta-t\gamma_\text{c}} \big]$ which is finite since by our choice of $\delta$,
\[ \frac{\beta^--1}{\beta^+-1}+1+\delta-t\gamma_\text{c}-\beta^- = \gamma_\text{c}(1-t)+\delta-1 < -1. \]
Now let $2\leq N<\gamma_\text{c}/\delta$ and consider $\E\left[W^+\psi_T(W^-z)\1\left\{(W^-)^{(N-1)\delta}<T\leq(W^-)^{N\delta}\right\}\right]$. As in the proof of Theorem \ref{2:threshold:res} it is enough to consider $\E\big[W^+\1\big\{W^->z^\frac{1}{(N-1)\delta-1},T\leq(W^-)^{N\delta}\big\}\big]$. Similar as above, we derive $\P\left(\tau_i\leq(w_i^-)^{N\delta}\right) \leq K_N(w_i^-)^{N\delta-t\gamma_\text{c}}$ for some uniform $K_N<\infty$ and
\begin{align*}
\E\Big[W^+\1\Big\{W^->z^\frac{1}{(N-1)\delta-1},T\leq(W^-)^{N\delta}\Big\}\Big] &\leq K_N \E\left[W^+\1\left\{W^->z^\frac{1}{(N-1)\delta-1}\right\}(W^-)^{N\delta-t\gamma_\text{c}}\right]\\
&\leq K_N\E\left[(W^-)^{\frac{\beta^--1}{\beta^+-1}}\1\left\{W^->z^\frac{1}{(N-1)\delta-1}\right\}\right]z^\frac{N\delta-t\gamma_\text{c}}{(N-1)\delta-1}\\
&= K_N\frac{\beta^--1}{1-\gamma_\text{c}} z^\frac{\gamma_\text{c}-1+N\delta-t\gamma_\text{c}}{(N-1)\delta-1}= o(z)
\end{align*}
since by the choice of $\delta$ and $N$, it holds $\gamma_\text{c}-1+N\delta-t\gamma_\text{c} < (N-1)\delta-1 < 0$.

Finally, we have to consider the part $\E\left[W^+\psi_T(W^-z)\1\{(W^-)^{\gamma_\text{c}-\delta}<T\leq(1+\epsilon)\alpha_\text{c}(W^-)^{\gamma_\text{c}}\}\right]$.
If we choose $\epsilon>0$ small enough such that $(1+2\epsilon)\alpha_\text{c}<\liminf_{w\to\infty}\tau(w)/w^{\gamma_\text{c}}$ and denote $\tilde{\tau}(w):=(1+\epsilon)\alpha_\text{c}w^{\gamma_\text{c}}$, then we observe that by Assumption \ref{2:ass:exposures} for $w_i^-$ large enough
\begin{align*}
\P\left(\tau_i\leq\tilde{\tau}(w_i^-)\right) &\leq 
\P\left(\sum_{j=1}^{\tilde{\tau}(w_i^-)} E_{\rho_i(j),i}\geq\frac{1+2\epsilon}{1+\epsilon}\tilde{\tau}(w_i^-)\lambda_i\right)
\leq K\left((1+\epsilon)\alpha_\text{c}(w_i^-)^{\gamma_\text{c}}\right)^{-(t-1)}
\end{align*}
for some $K<\infty$. Since $\gamma_\text{c}-1-(t-1)\gamma_\text{c} < \gamma_\text{c}+\delta-1 < 0$, we then get
\begin{align*}
&\E\left[W^+\1\{W^->z^\frac{1}{\gamma_\text{c}+\delta-1},T\leq(1+\epsilon)\alpha_\text{c}(W^-)^{\gamma_\text{c}}\}\right]\\
&\hspace*{6cm}\leq K \E\left[W^+\1\{W^->z^\frac{1}{\gamma_\text{c}+\delta-1}\}\left((1+\epsilon)\alpha_\text{c}(W^-)^{\gamma_\text{c}}\right)^{-(t-1)}\right]\\
&\hspace*{6cm}\leq K \left((1+\epsilon)\alpha_\text{c}\right)^{-(t-1)} \frac{\beta^--1}{1-\gamma_\text{c}}z^\frac{\gamma_\text{c}-1-(t-1)\gamma_\text{c}}{\gamma_\text{c}+\delta-1}= o(z).
\end{align*}
Altogether, we derive (note that we decomposed the expectation in finitely many summands)
\[ \E\left[W^+\psi_T(W^-z)\1\left\{T\leq(1+\epsilon)\alpha_\text{c}(W^-)^{\gamma_\text{c}}\right\}\right] = o(z) \]
and hence $\limsup_{z\to0+}z^{-1}f(z) = \limsup_{z\to0+}z^{-1}\E\left[W^+\psi_T(W^-z)\right] - 1 < 0$, which shows resilience by Theorem \ref{2:thm:cont:res}. Note that we can cut off $T$ at $(W^-)^\eta$ for some $\gamma_\text{c}<\eta<1$ to ensure continuity of $d(z)$ by Lemma \ref{2:lem:cont:differentiability} and Remark \ref{2:rem:cont:continuous:differentiability}.

Part \ref{2:cor:threshold:res:2}.~follows by the same calculations replacing $\gamma_\text{c}$ by $\gamma$ and using $\gamma>\gamma_\text{c}$.
\end{proof}

\cleardoublepage
\chapter{A Model for Default Contagion in Multi-type Financial Networks}\label{chap:block:model}

Real financial networks exhibit various forms of complex structures. Typical examples are core-periphery structures, geographically induced community structures, or mixtures of the two. See Figure \ref{3:fig:core:periphery:network} for an exemplary illustration. While our model from Chapter \ref{chap:systemic:risk} can describe heavy-tailed degree distributions, one can argue that compared to real observed networks the probability of exposures between periphery institutions is still too large. Moreover, communities cannot be modeled by the means from Chapter \ref{chap:systemic:risk}. It is thus the aim of this chapter to provide a model capable of representing these features. Again our starting point is the model from \cite{Detering2015a} (called \emph{threshold model} in Chapter \ref{chap:systemic:risk}) and we extend it in two directions in Section \ref{3:sec:default:fin}. First, we divide the global financial system into subsystems of different institution types such as country, core/periphery, a combination of the two, or any other reasonable segmentation. We then let the connection probability of two institutions in the system depend on their respective types -- larger within communities or within cores, smaller between communities and for periphery institutions for instance. Second, we allow for different exposure distributions between different institution types -- larger exposures between core-institutions for example. Our model can then be considered a \emph{multi-dimensional threshold model} in the sense that due to the two extensions the analysis of the contagion process turns into a multi-dimensional problem. Thus, compared to \cite{Detering2015a} novel techniques are required to derive asymptotic results for the final systemic damage induced by some initial shock event in Section \ref{3:sec:asymptotic:results} and for the characterization of resilient and non-resilient system structures in Section \ref{3:sec:resilience}. Moreover in Section \ref{3:sec:capital:block:model}, we provide a family of capital requirements sufficient to secure a multi-type network against small initial shocks. In Section \ref{3:sec:exposure:model}, we integrate the idea of exchangeable exposures from Chapter \ref{chap:systemic:risk} into our multi-type model to gain more flexibility at the modeling of exposures. This extension can be seen as a \emph{multi-dimensional exposure model} in analogy to Chapter \ref{chap:systemic:risk}. We then demonstrate the effects that the complex structures of financial systems modeled in this chapter have on their stability in Section \ref{3:sec:applications}. Finally, we provide all our proofs for this chapter in Section \ref{3:sec:proofs}.

\begin{figure}[t]
\centering
    \includegraphics[width=0.52\textwidth]{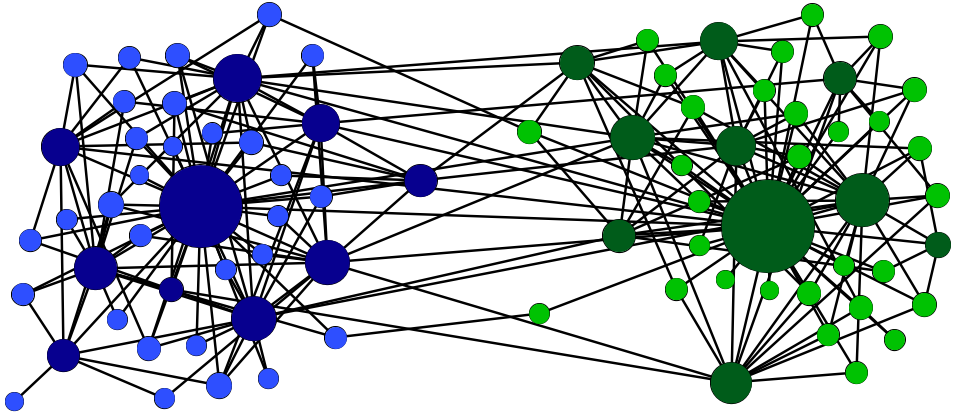}
\caption{A sample network consisting of two cores (darkblue resp.~darkgreen) and the associated peripheries (lightblue resp.~lightgreen). Vertex sizes correspond to the respective degrees. For simplicity the network is depicted undirected and the strength of links is omitted. }
\label{3:fig:core:periphery:network}
\end{figure}

\vspace*{-7pt}
\paragraph{My own contribution:} Except for Sections \ref{3:sec:capital:block:model} and \ref{3:sec:exposure:model} which are entirely my own work and have not been published elsewhere, so far, this chapter is mostly reproducing \cite{Detering2018} which is joint work with Nils Detering, Thilo Meyer-Brandis and Konstantinos Panagitou. I was significantly involved in the development of all parts of that paper and did most of the editorial work. In particular, I made major contributions to the model design, Corollary \ref{3:cor:poisson:degrees}, Lemmas \ref{3:lem:properties:f}, \ref{3:lem:existence:hatz} and \ref{3:lem:sufficient:criteria:z:star}, Theorem \ref{3:thm:general:weights}, Remark \ref{3:rem:fraction:subsystems}, Theorem \ref{3:thm:resilience}, Definition \ref{3:def:non:resilience}, Corollary \ref{3:cor:non:resilience:independent}, Lemma \ref{3:lem:z0}, Theorem \ref{3:thm:non-resilience}, Lemma \ref{3:lem:z0:equals:z*}, Examples \ref{3:ex:non-res:subsystem}, \ref{3:ex:res:global:system} and \ref{3:ex:neighbor:dependent}, the numerical simulations, Proposition \ref{3:prop:finitary:weights}, Theorem \ref{3:thm:finitary:weights}, as well as Lemmas \ref{3:lem:z^*:epsilon}, \ref{3:lem:stochastic:domination} and \ref{3:lem:convergence:g}. 
Compared to \cite{Detering2018}, in this chapter we additionally consider values of systemic importance for each institution as in Chapter \ref{chap:systemic:risk} to measure a more flexible final systemic damage than only the final default fraction.

\section{Default Contagion on a Random Weighted Multi-type Network}\label{3:sec:default:fin}
We describe a model for a financial network consisting of $n\in\N$ vertices (institutions) and random directed edges between them. We usually think of an institution \mbox{$i\in[n]:=\{1,\ldots,n\}\subset\N$} as a bank in an interbank network and of a directed edge going from institution $i\in[n]$ to $j\in[n]$ as a financial exposure of $j$ which emanates from $i$, for example by an outstanding interbank-loan from $i$ to $j$. Our model accounts for two more features. First, we assign weights to the edges which for now shall represent the amount of the loan. Later in Section \ref{3:sec:exposure:model}, we will generalize our model and identify an edge-weight as an abstract multiplier for the exposure. The model from this section can thus be interpreted as a multi-dimensional extension of the \emph{threshold model} from Subsection \ref{2:ssec:special:case:threshold:model} whereas Section \ref{3:sec:exposure:model} describes a multi-dimensional extension of the \emph{exposure model} from Subsection \ref{2:ssec:exposure:model}. As will be clear from the construction in the following, the assignment of edge-weights depends on both the creditor- and the debtor-institution. This feature is new as compared to previous literature such as~\cite{Cont2016} and Chapter \ref{chap:systemic:risk}, where the amount of each loan did only depend on the creditor bank. Second, we assign different types to the institutions in the network. This allows to describe more involved network structures such as core-periphery networks -- a two-type network in our terminology -- and (dis-)assortative structures.

\subsection{Vertex Types}
We begin by assigning to each institution $i\in[n]$ a type $\alpha_i\in[T]$, where $T\in\N$ is the fixed number of types. In the prominent case of a core-periphery network we choose $T=2$ and a bank $i\in[n]$ shall be a core bank if $\alpha_i=1$ resp.~a periphery bank if $\alpha_i=2$. Hence the financial network is partitioned into sets of institutions of different types, which we also call blocks.

\subsection{Vertex Weights and Random Weighted Edges}\label{3:ssec:random:edges}
Next, we fix $R\in\N$ and we construct a random network with edge-weights in $[R]$. To this end, assign to each institution $i\in[n]$ a set $\{w_i^{-,r,\alpha},w_i^{+,r,\alpha}\}_{1\leq r\leq R, 1\leq \alpha\leq T}$ of non-negative vertex-weights and denote $\bm{w}_i^-=(w_i^{-,r,\alpha})_{r\in[R],\alpha\in[T]}\in\R_{+,0}^{[R]\times [T]}$ resp.~$\bm{w}_i^+=(w_i^{+,r,\alpha})_{r\in[R],\alpha\in[T]}\in\R_{+,0}^{[R]\times [T]}$. The weight $w_i^{-,r,\alpha}$ describes the tendency of bank $i$ to develop incoming edges of weight $r$ from institutions of type $\alpha$. Similarly, $w_i^{+,r,\alpha}$ describes the tendency of $i$ to form outgoing edges of weight $r$ to institutions of type $\alpha$. To formalize this, 
let $X_{i,j}^r$ be the indicator random variable which is $1$ if there is an edge of weight $r$ going from $i$ to $j$ and $0$ otherwise and let $X_{i,j}^r\sim\mathrm{Be}(p_{i,j}^r)$ be a Bernoulli random variable with  expectation
\begin{equation}\label{3:eqn:edge:prob}
p_{i,j}^r := \begin{cases} \min\{R^{-1}, n^{-1}w_i^{+,r,\alpha_j}w_j^{-,r,\alpha_i}\},&i\neq j,\\0,&i=j.\end{cases}
\end{equation}
To avoid multiple edges of different weights between the institutions, we assume $\{X_{i,j}^r\}_{1\leq r\leq R}$ to be mutually exclusive in the sense that $\sum_{1\le r\le R} X_{i,j}^r\leq 1$. Also, we assume that edges between different pairs of institutions are independent, i.\,e.~$X_{i_1,j_1}^{r_1}\perp X_{i_2,j_2}^{r_2}$ for all $r_1,r_2\in[R]$ if $(i_1,j_1)\neq (i_2,j_2)$. In particular, $X_{i,j}^{r_1}\perp X_{j,i}^{r_2}$ for all $i\neq j$, $r_1,r_2\in[R]$. This can for example be achieved by introducing a sequence of independent random variables $U_{i,j}$, each distributed uniformly on the interval $[0,1]$, and letting $X_{i,j}^r=\1\big\{U_{i,j}\in\big[\sum_{s\leq r-1}p_{i,j}^s,\sum_{s\leq r}p_{i,j}^s\big)\big\}$. The upper bound $R^{-1}$ in \eqref{3:eqn:edge:prob} then ensures that $\sum_{s\leq R}p_{i,j}^s\leq1$. 


\subsection{Capital and Default Contagion}
We assign to each institution $i\in[n]$ an initial amount of capital (equity) $c_i\in\N_{0,\infty}$. We call an institution solvent if $c_i>0$ and insolvent if $c_i=0$ and we denote by $\mathcal{D}_0:=\{i\in[n]\,:\,c_i=0\}$ the set of initially defaulted institutions. The initial default shall be due to some exogenous event such as a stock market crash. Because of the interconnections in the network the default of the institutions in $\mathcal{D}_0$ will spread through the network. This happens since the defaulted banks cannot (fully) repay their loans to their creditors. As first suggested in \cite{Gai2010} it is a reasonable assumption that defaulted debtors cannot repay any of their debts since processing their default may take months or even years while financial contagion is a short term process. In fact, one can generalize our model to the case of a fixed constant recovery rate simply by adjusting the capitals. The default contagion process can then be described as follows. In round $k\geq 1$ of the default cascade the set of defaulted institutions is 
\begin{equation}
\label{3:eqn:default:contagion}
\mathcal{D}_k := \Bigg\{ i\in[n]\,:\,c_i\leq\sum_{r \in [R]} r\sum_{j\in\mathcal{D}_{k-1}}X_{j,i}^r\Bigg\}.
\end{equation}
In particular, $\mathcal{D}_0\subseteq\mathcal{D}_1\subseteq\cdots$ and the chain of default sets stabilizes at round $n-1$ the latest. We hence denote the final default set by $\mathcal{D}_n:=\mathcal{D}_{n-1}$. Note that the only randomness in this process stems from the random links in the network. Once a network configuration has been fixed the whole default contagion sequence is fully determined.
%

\begin{remark}\label{3:rem:exchangeable:exposures}
To model realistic financial networks with general exposure values, it would be a priori necessary to choose $R$ very large and our model would become very high-dimensional. Instead of considering an edge-weight $r$ as the exposure between two institutions, however, one can also interpret it as a more general factor of impact. It is then possible to model unbounded exposure distributions also with a considerably small choice of $R$. We will discuss this idea more precisely in Chapter \ref{chap:fire:sales:default}.

\end{remark}

\subsection{Systemic Importance}
The set $\mathcal{D}_n$ introduced in the previous subsection describes the set of banks that are finally driven into default by the initial default set $\mathcal{D}_0$. Only considering the size $\vert\mathcal{D}_n\vert$ of this set to measure the damage to the financial system neglects the fact that there are larger/more important banks and smaller/less important banks in the system, however. Instead, analogously to Chapter \ref{chap:systemic:risk} we finally assign to each bank $i\in[n]$ a systemic importance value $s_i\in\R_{+,0}$ which reflects the importance of the bank for the whole financial system or even the wider economy. See Chapter~\ref{chap:systemic:risk} for suitable examples of $s_i$. Instead of $\vert\mathcal{D}_n\vert$ we then consider
\[ \mathcal{S}_n:=\sum_{i\in\mathcal{D}_n}s_i \]
to measure the final damage caused by defaulted banks. Note that $\mathcal{S}_n=\vert\mathcal{D}_n\vert$ if $s_i=1$ for all $i\in[n]$, which is the special case considered in \cite{Detering2018}.

\subsection{Regular Vertex Sequences}
In the previous subsections we introduced several parameters and in particular, any random ensemble is described by the weight sequences $\bm{w}^-:=(\bm{w}_1^-,\ldots,\bm{w}_n^-)$ and $\bm{w}^+:=(\bm{w}_1^+,\ldots,\bm{w}_n^+)$, the systemic importance sequence $\bm{s}:=(s_1,\ldots,s_n)$, the capital sequence $\bm{c}:=(c_1,\ldots,c_n)$, and the vertex type sequence $\bm{\alpha}:=(\alpha_1,\ldots,\alpha_n)$. That is, much of the information about the system -- in particular the complexity of the network configuration -- is contained in the empirical distribution function
\[ F_n(\bm{x}, \bm{y}, v, l, m) := n^{-1}\sum_{i\in[n]}~\prod_{\substack{r\in [R],\alpha\in [T]}}\1\Big\{w_i^{-,r,\alpha}\leq x^{r,\alpha},w_i^{+,r,\alpha}\leq y^{r,\alpha}\Big\} \1\left\{s_i\leq v, c_i\leq l, \alpha_i\leq m\right\}, \]
for $\bm{x},\bm{y}\in\R_{+,0}^{[R]\times[T]}$, $v\in\R_{+,0}$, $l\in\N_{0,\infty}$ and $m\in[T]$. 

The setting described so far puts us in the position to model a system with a given number~$n$ of institutions. However, as already described in the introduction, our main focus is on studying how the complex structures in the underlying network affect the contagion process and, more generally, the (in-)stability of the system as a whole. Towards this aim we proceed as follows. Instead of restricting our attention  to a single system configuration, we consider an ensemble of systems that are similar, where the structural similarity is measured precisely in terms of the joint empirical distribution of all parameters that we consider. In particular, we assume that we have a collection of systems with a varying number $n$ of
institutions with the property that the sequence $(F_n)_{n \in \mathbb{N}}$ of empirical distributions converges.

\begin{definition}\label{3:def:regular:vertex:sequence}
A sequence $(\bm{w}^-(n),\bm{w}^+(n), \bm{s}(n), \bm{c}(n), \bm{\alpha}(n))_{n\in\N}$ of model parameters for different network sizes $n\in\N$ is called a \emph{regular vertex sequence} if the following conditions hold. 
\begin{enumerate}[(a)]
\item \textbf{Convergence in distribution:} 
Let $(\bm{W}_n^-,\bm{W}_n^+,S_n,C_n, A_n)$ be a random vector distributed according to the empirical distribution function $F_n$ with $\bm{W}_n^\pm=(W_n^{\pm,r,\alpha})_{r\in[R],\alpha\in[T]}$. Then there exists a distribution function $F$ such that $F_n(\bm{x},\bm{y},v,l,m)\to F(\bm{x},\bm{y},v,l,m)$ for all points $(\bm{x},\bm{y},v,l,m)$ at which $F_{l,m}(\bm{x},\bm{y},v):=F(\bm{x},\bm{y},v,l,m)$ is continuous. Denote by $(\bm{W}^-,\bm{W}^+,S,C,A)$ a random vector distributed according to 
$F$ where similar to above $\bm{W}^\pm=(W^{\pm,r,\alpha})_{r\in[R],\alpha\in[T]}$.
\item \textbf{Convergence of average weights and systemic importance:} $S$, $W^{-,r,\alpha}$ and $W^{+,r,\alpha}$ are integrable for all $r\in[R]$ and $\alpha\in[T]$, and as $n\to\infty$ it holds
\[ \E[S_n]\to\E[S],\quad \E\left[W_n^{-,r,\alpha}\right]\to\E\left[W^{-,r,\alpha}\right],\quad\text{and}\quad\E\left[W_n^{+,r,\alpha}\right]\to\E\left[W^{+,r,\alpha}\right]. \]
\end{enumerate}
\end{definition}
Note that if we extract the subnetwork of edges with weight $r$ going from banks of type $\alpha$ to banks of type $\beta$ we are exactly in the setting of \cite{Detering2015a} with limiting in-weight $W^{-,r,\beta}$ and out-weight $W^{+,r,\alpha}$ (each of them scaled to account for the changed number of banks). We hence derive the following corollary of \cite[Theorem 3.3]{Detering2015a} which gives some intuition about the geometry of the described random graph in our model.
\begin{corollary}\label{3:cor:poisson:degrees}
Consider a financial system described by a regular vertex sequence and let $D_n^{\pm,r,\alpha}$ be the $r$-out/in-degree with respect to banks of type $\alpha$ of some bank in the network chosen uniformly at random. 
Then for each $(r,\alpha)\in[R]\times[T]$, as $n\to\infty$, in distribution
\begin{equation}\label{3:eqn:degree:convergence}
D_n^{\pm,r,\alpha} \to \mathrm{Poi}\Bigg(W^{\pm,r,\alpha}\sum_{\beta\in[T]}\zeta^{r,\beta,\alpha}_\mp\1\{A=\beta\}\Bigg),
\end{equation}
where $\zeta^{r,\beta,\alpha}_\mp:=\E\left[W^{\mp,r,\beta}\1\{A=\alpha\}\right]$.
In particular, the $r$-out/in-degree with respect to banks of type $\alpha$ of some uniformly chosen bank of type $\beta$ converges in distribution to $\mathrm{Poi}\big(W^{\pm,r,\alpha}\zeta_\mp^{r,\beta,\alpha}\big)$. 
\end{corollary}

\section{Asymptotic Results}\label{3:sec:asymptotic:results}
This section presents results for the final default fraction $n^{-1}\vert\mathcal{D}_n\vert$ triggered by some set of initial defaults, which is one of the main contributions of this chapter -- see Subsection~\ref{3:subsecFinDef}. In Subsection~\ref{3:subsecPrel}, we first introduce some functions that will play an important role. For the sake of readability, most of the -- mainly technical -- proofs of this section are moved to the appendix.

\subsection{Preliminaries}\label{3:subsecPrel}
Denote in the following $V:=[R]\times[T]^2$. Let $\psi_l(x_1,\ldots,x_R) := \P\left(\sum_{s\in[R]}sX_s\geq l\right)$ for independent Poisson random variables $X_s\sim\mathrm{Poi}(x_s)$, $s\in[R]$.  In the following, the functions $f^{r,\alpha,\beta}:\R_{+,0}^V\to\R$, $(r,\alpha,\beta)\in V$,  and $g:\R_{+,0}^V\to\R_{+,0}$ will play a central role. They are given by
\begin{equation*}
\begin{gathered}
f^{r,\alpha,\beta}(\bm{z}) = \E\Bigg[W^{+,r,\alpha}\psi_C\Bigg(\sum_{\gamma\in[T]}W^{-,1,\gamma}z^{1,\beta,\gamma},\ldots,\sum_{\gamma\in[T]}W^{-,R,\gamma}z^{R,\beta,\gamma}\Bigg)\1\{A=\beta\}\Bigg] - z^{r,\alpha,\beta},\\
g(\bm{z}) = \sum_{\beta\in[T]}\E\Bigg[S \psi_C\Bigg(\sum_{\gamma\in[T]}W^{-,1,\gamma}z^{1,\beta,\gamma},\ldots,\sum_{\gamma\in[T]}W^{-,R,\gamma}z^{R,\beta,\gamma}\Bigg)\1\{A=\beta\}\Bigg]
\end{gathered}
\end{equation*}
We begin by investigating some basic but important properties of these functions.
\begin{lemma}\label{3:lem:properties:f}
The functions $f^{r,\alpha,\beta}(\bm{z})$, $(r,\alpha,\beta)\in V$, and $g(\bm{z})$ are continuous at all $\bm{z}\in\R_{+,0}^V$. Further, each function $f^{r,\alpha,\beta}(\bm{z})$ is monotonically increasing in all of its coordinates except $z^{r,\alpha,\beta}$.
\end{lemma}
\begin{proof}
Continuity 
follows from Lebesgue's dominated convergence theorem noting that the integrands are continuous in $\bm{z}$ and bounded by the integrable random variables $W^{+,r,\alpha}$ resp.~$S$. Monotonicity of $f^{r,\alpha,\beta}$ follows directly from the monotonicity of the Poisson-probabilities.
\end{proof}
Let now $P:=\bigcap_{(r,\alpha,\beta)\in V}\{\bm{z}\in\R_{+,0}^V\,:\,f^{r,\alpha,\beta}(\bm{z})\geq0\}$. Since $f^{r,\alpha,\beta}(\bm{z})<0$ for any $\bm{z}\in\R_{+,0}^V$ with $z^{r,\alpha,\beta}>\E[W^{+,r,\alpha}\1\{A=\beta\}]$ and $P$ is an intersection of closed sets, $P$ is in fact compact. Note that clearly $\bm{0}\in P$. In general $P$ might consist of several disjoint, compact, connected components. Let in the following $P_0$ denote the component (i.\,e.~the largest connected subset) of $P$ containing $\bm{0}$. Since $P$ is a compact subset of $\R_{+,0}^V$, so is $P_0$. Define now $\bm{z}^*\in\R_{+,0}^V$ by $(z^*)^{r,\alpha,\beta}:=\sup_{\bm{z}\in P_0}z^{r,\alpha,\beta}$. The following lemma shows that in fact $\bm{z}^*\in P_0$ and it can hence be thought of as the maximal point of $P_0$. Further it identifies $\bm{z}^*$ as a joint root of all the functions $f^{r,\alpha,\beta}$, $(r,\alpha,\beta)\in V$, and shows the existence of a smallest joint root $\hat{\bm{z}}$. It will turn out later (see in particular Theorem \ref{3:thm:general:weights}) that the final systemic damage $n^{-1}\mathcal{S}_n$ is intimately related to these two (typically coinciding) joint roots of the functions $f^{r,\alpha,\beta}$, $(r,\alpha,\beta)\in V$. 

\begin{lemma}\label{3:lem:existence:hatz}
There exists a smallest joint root $\hat{\bm{z}}\in\R_{+,0}^V$ of all functions $f^{r,\alpha,\beta}$, \mbox{$(r,\alpha,\beta)\in V$,} in the sense that $\hat{\bm{z}}\leq\bar{\bm{z}}$ componentwise for all joint roots $\bar{\bm{z}}\in\R_{+,0}^V$. Further, $\bm{z}^*$ as defined above is a joint root of the functions $f^{r,\alpha,\beta}$, $(r,\alpha,\beta)\in V$, and both $\hat{\bm{z}}\in P_0$ and $\bm{z}^*\in P_0$.
\end{lemma}
The lemma identifies $\bm{z}^*$ as 
the maximal joint root of $f^{r,\alpha,\beta}$, $(r,\alpha,\beta)\in V$, in $P_0$. However, if $P_0\subsetneq P$, then there will exist joint roots $\tilde{\bm{z}}\not\in P_0$ such that $\bm{z}^*\lneq\tilde{\bm{z}}$.

Often $\hat{\bm{z}}$ and $\bm{z}^*$ will coincide and then Theorem \ref{3:thm:general:weights} below will show that the final systemic damage $n^{-1}\mathcal{S}_n$ converges to $g(\hat{\bm{z}})$ in probability. But in some pathological situations this is not the case and Theorem \ref{3:thm:general:weights} will yield a lower bound on $n^{-1}\mathcal{S}_n$ in terms of $\hat{\bm{z}}$ and an upper bound in terms of $\bm{z}^*$. Figures \ref{3:fig:one:joint:root} and \ref{3:fig:two:joint:roots} show two-dimensional examples of $f$. In both examples, we chose $R=2$ and $T=1$. In the first example, we further chose all weights to be $1$ and the capital of each bank to be $3$ with probability $80\%$ respectively $0$ with probability $20\%$. The functions $f^1(z^1,z^2):=f^{1,1,1}(z^1,z^2)$ and $f^2(z^1,z^2):=f^{2,1,1}(z^1,z^2)$, where $z^1:=z^{1,1,1}$ and $z^2:=z^{2,1,1}$, then have a unique joint root, i.\,e.~$\hat{\bm{z}}=\bm{z}^*$. In the second example, we chose all weights to be $2$ and the capital of each bank to be $3$ with probability $\approx 94.14\%$ respectively $0$ with probability $\approx 5.86\%$. In this case, there exist two distinct joint roots $\hat{\bm{z}}\neq\bm{z}^*$ in $P_0$. At $\hat{\bm{z}}$, the root sets of $f^1$ and $f^2$ do not cross each other but only touch.

\begin{figure}[t]
    \hfill\subfigure[]{\includegraphics[width=0.4\textwidth]{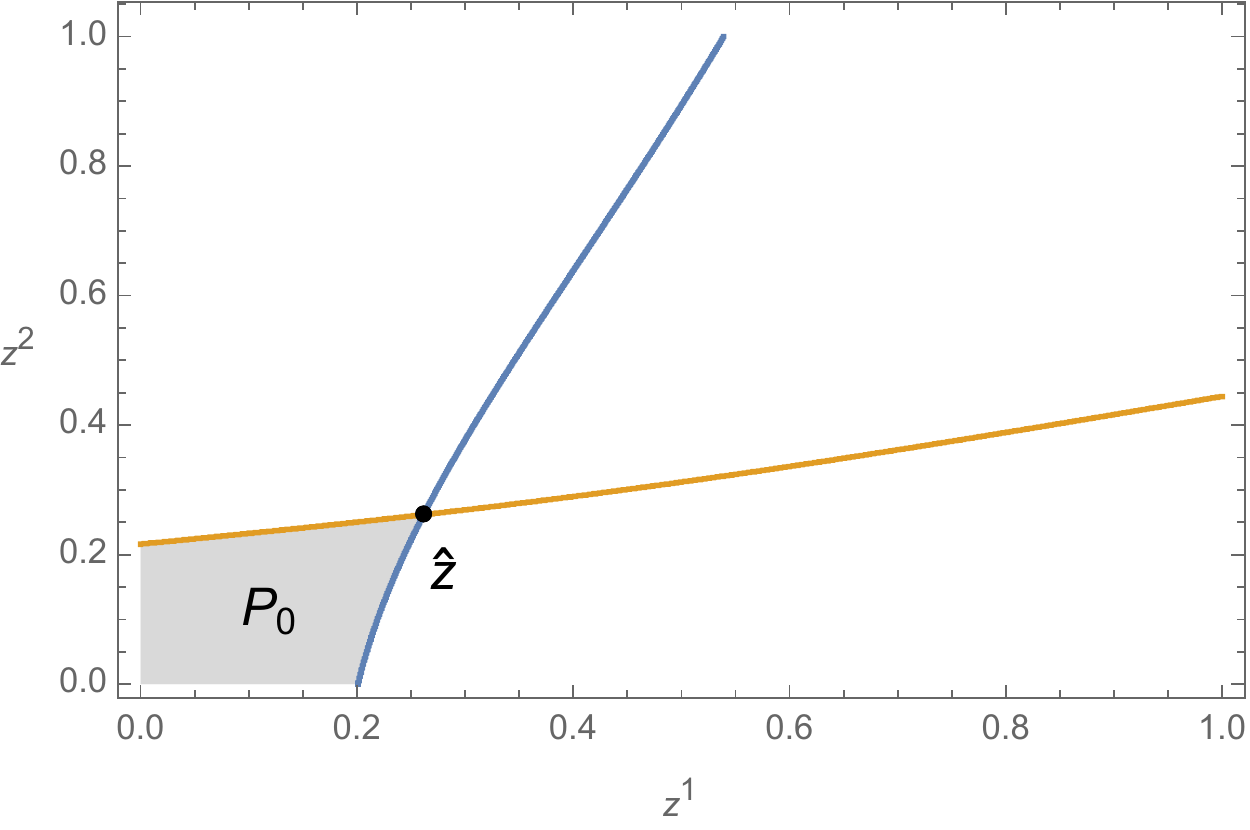}\label{3:fig:one:joint:root}}
    \hfill\subfigure[]{\includegraphics[width=0.4\textwidth]{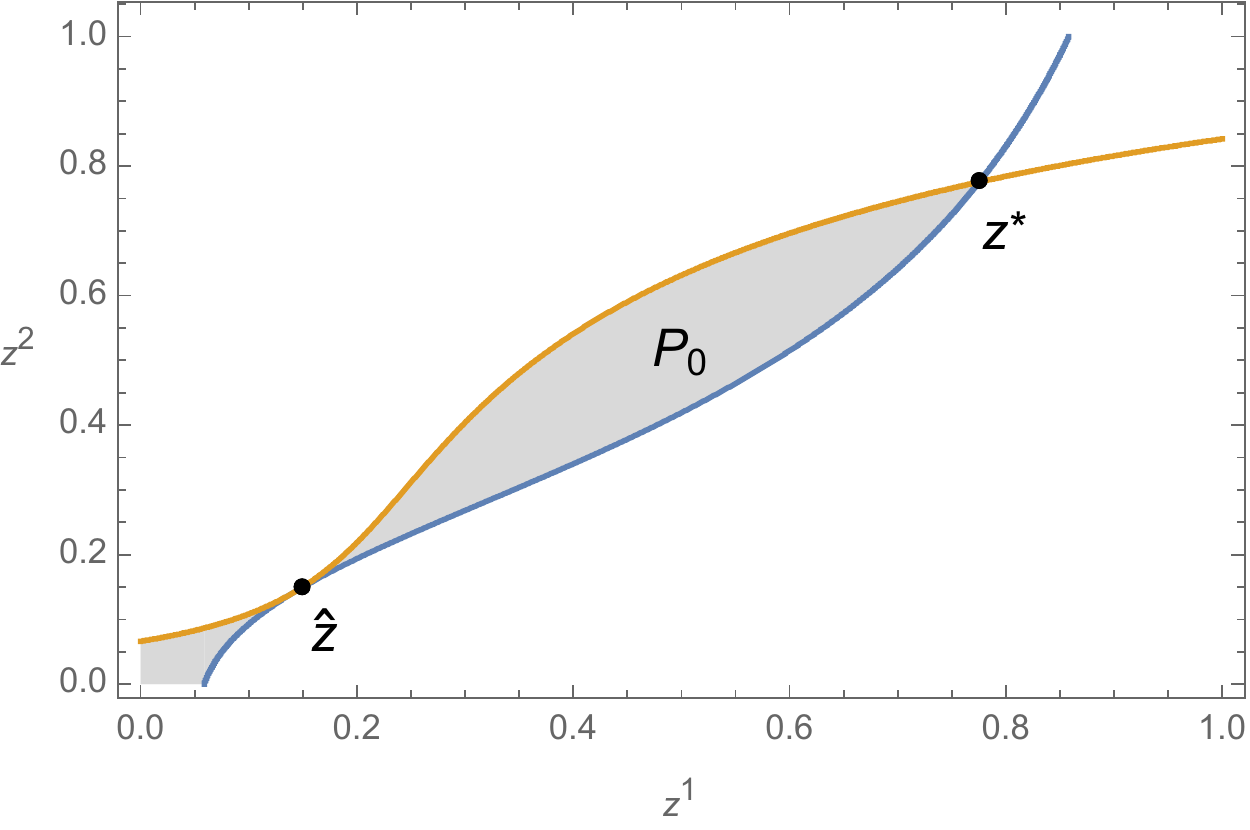}\label{3:fig:two:joint:roots}}\hfill
\caption{Plot of the root sets of the functions $f^1(z^1,z^2)$ (blue) and $f^2(z^1,z^2)$ (orange) for two different example networks.}
\end{figure}

The next lemma provides two sufficient criteria to check if a joint root, such as $\hat{\bm{z}}$, equals $\bm{z}^*$. These depend on the (weak) directional derivatives of $f^{r,\alpha,\beta}$, $(r,\alpha,\beta)\in V$, and are hence natural extensions of the stable fixed point assumption in previous literature such as \cite{Cont2016,Detering2015a} and Chapter \ref{chap:systemic:risk}. 

\begin{lemma}\label{3:lem:sufficient:criteria:z:star}
If $\bar{\bm{z}}\in P_0$ is a joint root of the functions $f^{r,\alpha,\beta}$, $(r,\alpha,\beta)\in V$, then $\bar{\bm{z}}=\bm{z}^*$ if one of the following holds:
\begin{enumerate}[(a)]
\item There exists $\bm{v}\in\R_+^V$ such that for all $(r,\alpha,\beta)\in V$ the directional derivatives $D_{\bm{v}}f^{r,\alpha,\beta}(\bar{\bm{z}})$ exist and $D_{\bm{v}}f^{r,\alpha,\beta}(\bar{\bm{z}})<0$.
\item There exist $\bm{v}\in\R_+^V$, $\kappa<1$ and $\Delta>0$ such that for every $\delta\in(0,\Delta)$,
\begin{align*}
\kappa v^{r,\alpha,\beta} &\geq \sum_{r'\in[R]}\E\Bigg[W^{+,r,\alpha}\Bigg(\sum_{\beta'\in[T]}v^{r',\beta,\beta'}W^{-,r',\beta'}\Bigg)\1\{A=\beta\}\\
&\hspace{1.5cm} \times\P\Bigg(\sum_{s\in[R]}s\mathrm{Poi}\Bigg(\sum_{\gamma\in[T]}W^{-,s,\gamma}\left(\bar{z}^{s,\beta,\gamma}+\delta v^{s,\beta,\gamma}\right)\Bigg)\in\{C-r',\ldots,C-1\}\Bigg)\Bigg].
\end{align*}
\end{enumerate}
\end{lemma}

\subsection{The Main Result for the Final Systemic Damage}\label{3:subsecFinDef}
We now provide an asymptotic formula for the final systemic damage $n^{-1}\mathcal{S}_n$ in terms of function $g$ and the joint roots $\hat{\bm{z}}$ and $\bm{z}^*$. 

\begin{theorem}\label{3:thm:general:weights}
Consider a financial system described by a regular vertex sequence and let $\hat{\bm{z}}$ and $\bm{z}^*$ be the smallest respectively largest joint root in $P_0$ of the functions $f^{r,\alpha,\beta}$, $(r,\alpha,\beta)\in V$. Then
\[ g(\hat{\bm{z}}) + o_p(1) \leq n^{-1}\mathcal{S}_n \leq g(\bm{z}^*) + o_p(1). \]
In particular, if $\hat{\bm{z}}=\bm{z}^*$, then $n^{-1}\mathcal{S}_n \stackrel{p}{\to} g(\hat{\bm{z}})$ as $n\to\infty$.
\end{theorem}
Note that for the case $R=T=1$, Theorem \ref{3:thm:general:weights} extends \cite[Theorem 7.2]{Detering2015a} by an upper bound even if the requirements for the second part there are not satisfied. 
Also we do not 
require continuous differentiability of $f^{1,1,1}(z^{1,1,1})$ in a neighborhood of $(z^*)^{1,1,1}$ as 
in Theorem \ref{2:thm:asymp:2}.

\begin{remark}\label{3:rem:fraction:subsystems}
Theorem \ref{3:thm:general:weights} determines $ng(\hat{\bm{z}})$ as a lower bound on the damage caused by finally defaulted banks in the network. In fact, $g(\bm{z})$ is given by a sum over all the different types $\beta\in[T]$ in the network and it is thus no surprise that by small changes in the proofs of Theorems \ref{3:thm:finitary:weights} and \ref{3:thm:general:weights}, one derives that the damage caused by finally defaulted banks of type $\beta$ is lower~bounded~by
\[ n\E\Bigg[S\P\Bigg(\sum_{s\in[R]}s\mathrm{Poi}\Bigg(\sum_{\gamma\in[T]}W^{-,s,\gamma}\hat{z}^{s,\beta,\gamma}\Bigg)\geq C\Bigg)\1\{A=\beta\}\Bigg] + o_p(n). \]
The same reasoning allows to derive an upper bound in terms of $\bm{z}^*$.
\end{remark}


\section{Resilient and Non-resilient Networks}\label{3:sec:resilience}

In the previous section, we derived results that allow us to determine the typical final default fraction  in large financial systems caused by an exogenous shock. Another important question from a regulator's point of view that we study in this section is whether a given system in an \emph{initially} unshocked state is likely to be resilient to small shocks or susceptible to default~cascades.

Note that for some fixed financial network $(\bm{W}^-,\bm{W}^+,S,C,A)$ all information about the initial shock stems from $C$ and by ``initially unshocked'' we mean that $c_i>0$ for all $i\in[n]$. We model small shocks to the system by an \emph{ex post infection} in the following sense: we introduce indicators $m_i\in\{0,1\}$, $i\in[n]$, 
with the meaning that (the initially solvent) bank $i$ becomes insolvent if $m_i=0$. This amounts to setting its capital to $c_im_i$. In analogy to Definition \ref{3:def:regular:vertex:sequence} we assume regularity of $\{m_i\}_{i\in[n]}$ (jointly with the rest of the parameters) and we denote by $M$ the limiting random variable of ex post infection. 
In particular, the financial system 
shall be described by the random vector $(\bm{W}^-,\bm{W}^+,S,C,A,M)$ with $\P(C=0)=0$ and $\P(M=0)>0$. Denote by $\mathcal{D}_n^M$ the random final default set that $M$ triggers, by $\mathcal{S}_n^M=\sum_{i\in\mathcal{D}_n}s_i$ the final systemic damage, and by $(f^M)^{r,\alpha,\beta}$, $g^M$ and $(\bm{z}^*)^M$ the analogues of $f^{r,\alpha,\beta}$, $g$ respectively $\bm{z}^*$ with $C$ replaced by $CM$. 

From a regulator's point of view a desirable property of a financial system is the ability to absorb small local shocks $M$ without larger systemically important parts of the system being harmed. In our asymptotic setting, we can even choose $M$ arbitrarily small and we call a system \emph{resilient} 
if the relative final damage $n^{-1}\mathcal{S}_n^M$ tends to $0$ as $\P(M=0)\to0$. If on the other hand $n^{-1}\mathcal{S}_n^M$ is lower bounded by some positive constant, we call the system \emph{non-resilient} (see Definition \ref{3:def:non:resilience} below). 
\begin{definition}[Resilience]\label{3:def:resilience}
A financial system is said to be \emph{resilient} if for each $\epsilon>0$ there exists $\delta>0$ such that for all $M$ with $\P(M=0)<\delta$ it holds $n^{-1}\mathcal{S}_n^M \leq \epsilon$ w.\,h.\,p.
\end{definition}
It will turn out that the resilience of the system strongly depends on the form of the set $P_0$ which was introduced in Subsection~\ref{3:subsecPrel}.
Our first result is a criterion guaranteeing resilience. 
\begin{theorem}[Resilience Criterion]\label{3:thm:resilience}
Consider a financial system described by a regular vertex sequence and assume that 
$P_0=\{\bm{0}\}$. Then 
the system is resilient.
\end{theorem}
In particular for $R=T=1$, resilience is ensured if $0=\inf\{z>0\,:\,f^{1,1,1}(z)<0\}$. Theorem~\ref{3:thm:resilience} therefore extends \cite[Theorem 2.7]{Detering2015a}. Moreover, by Lemma~\ref{3:lem:sufficient:criteria:z:star}, if for some $\bm{v}\in\R_+^V$, $D_{\bm{v}}f^{r,\alpha,\beta}(\bm{0})$ exists and is negative for each $(r,\alpha,\beta)\in V$, then $P_0=\{\bm{0}\}$ and Theorem \ref{3:thm:resilience} is applicable.

Figure \ref{3:fig:resilient} shows a two-dimensional example satisfying the condition in Theorem \ref{3:thm:resilience}. We chose $R=2$, $T=1$, $W^{\pm,1}=W^{\pm,2}=1$ and $C=3$. It can be seen from the figure that small shocks (here $5\%$ of all banks) do only cause small jumps of the smallest joint root of $f^1(z^1,z^2)=f^{1,1,1}(z^1,z^2)$ and $f^2(z^1,z^2)=f^{2,1,1}(z^1,z^2)$, where $z^1:=z^{1,1,1}$ and $z^2:=z^{2,1,1}$.

\begin{figure}[t]
	\hfill\includegraphics[width=0.578\textwidth]{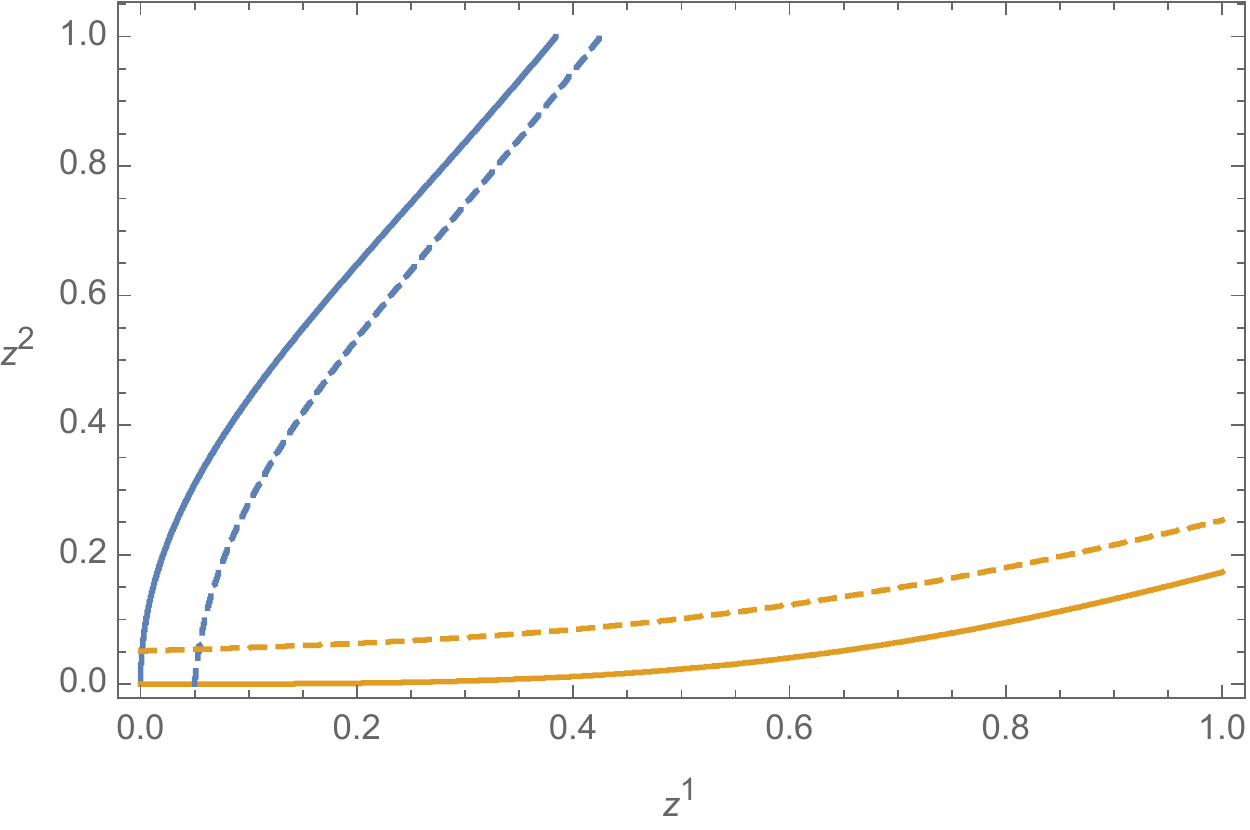}\hfill
		\caption{Plot of the root sets of the functions $f^1(z^1,z^2)$ (blue) and $f^2(z^1,z^2)$ (orange) for a financial system satisfying the condition in Theorem \ref{3:thm:resilience}. Solid: the unshocked functions. Dashed: the shocked functions.}\label{3:fig:resilient}
\end{figure}

On the other hand, concerning the characterization of non-resilient networks, a difficulty that arises is that the ex post shock $M$ possibly targets only certain subnetworks. More precisely, if $\tilde{V}:=\{(r,\alpha,\beta)\in V\,:\,\E[W^{+,r,\alpha}\1\{A=\beta\}]>0\}$, then it is still possible that for $(r,\alpha,\beta)\in \tilde{V}$ it holds $\E[W^{+,r,\alpha}\1\{A=\beta\}\1\{M=0\}]=0$. Consider for example a financial network consisting of banks of two types which are isolated of each other. Further, one of the two subnetworks shall be resilient, whereas the other one is non-resilient (in the sense of Definition \ref{3:def:non:resilience} (b)
). In order for the whole system to experience large damage, it is then necessary that $M$ does not only infect banks in the resilient subsystem but also in the non-resilient one. This explains why we have to differentiate between different choices for $M$ in the following to fully understand non-resilience in our model.
\begin{definition}[Non-resilience]\label{3:def:non:resilience}
\begin{enumerate}[(a)]
\item Let $I\subseteq \tilde{V}$. A financial system is called \emph{non-resilient with respect to shocks on $I$} if there exists a constant $\Delta_I>0$ such that $n^{-1}\mathcal{S}_n^M \geq \Delta_I$ w.\,h.\,p.~for each shock $M$ with $\E[W^{+,r,\alpha}\1\{A=\beta\}\1\{M=0\}]>0$ for all $(r,\alpha,\beta)\in I$.
\item We call a financial system \emph{non-resilient} if it is non-resilient w.\,r.\,t.~shocks on 
some $I\subseteq \tilde{V}$.
\end{enumerate} 
\end{definition}
Clearly a system is non-resilient if and only if it is non-resilient w.\,r.\,t.~shocks on $\tilde{V}$.


Let us start by considering the special case that $M$ infects every part of the system (i.\,e.~a shock on $\tilde{V}$). This is the case for example if $M$ is independent of type $A$, vertex-weights $W^{\pm,r,\alpha}$ and capital $C$. We can then formulate a corollary of Theorem \ref{3:thm:non-resilience} following later:
\begin{corollary}\label{3:cor:non:resilience:independent}
Consider a financial system described by a regular vertex sequence and  any ex post infection $M$ with $\E[W^{+,r,\alpha}\1\{A=\beta\}\1\{M=0\}]>0$ for all $(r,\alpha,\beta)\in \tilde{V}$. Then for any $\epsilon>0$ w.\,h.\,p.{} $n^{-1}\mathcal{S}_n^M \geq g(\bm{z}^*)-\epsilon$. If 
$g(\bm{z}^*)>0$, 
then the system is non-resilient.
\end{corollary}
Recall from Theorem \ref{3:thm:resilience} that a financial system is resilient if $P_0=\{\bm{0}\}$. For most practical purposes (e.\,g.~$\P(S>0)=1$) it will hold that $g(\bm{z}^*)>0$ if $P_0\supsetneq\{\bm{0}\}$ and hence $\bm{z}^*\neq\bm{0}$. Theorem \ref{3:thm:resilience} and Corollary \ref{3:cor:non:resilience:independent} then completely characterize resilience of a financial system in terms of $P_0$. 

See Figure \ref{3:fig:non:resilient} for an example where $P_0\neq\{\bm{0}\}$. In this example, we chose $R=2$, $T=1$, weights $W^{\pm,1}=W^{\pm,2}=3/2$ and capital $C=2$. 
The figure shows the jump of the smallest joint root from $\bm{0}$ to above $\bm{z}^*$ for any small shock (here $10\%$ of all banks).
\begin{figure}[t]
	\hfill\includegraphics[width=0.578\textwidth]{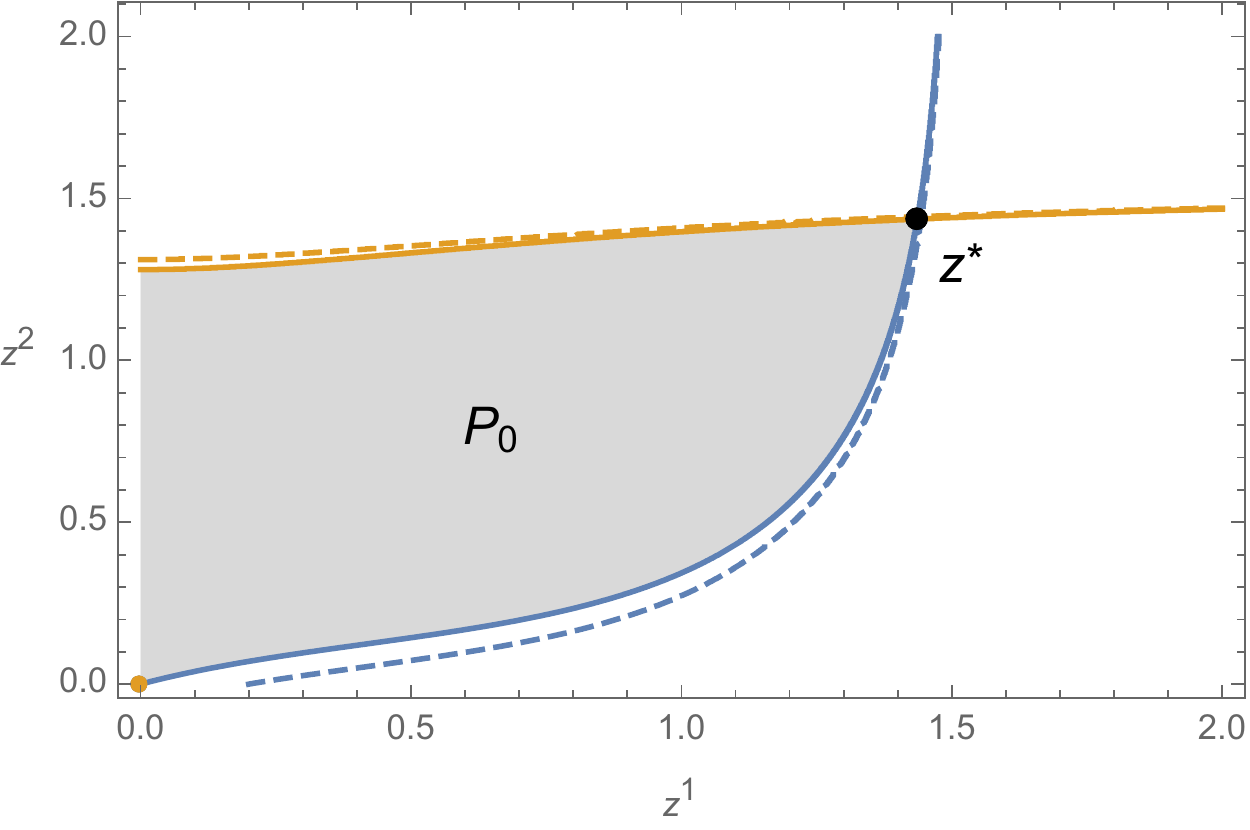}\hfill
		\caption{Plot of the root sets of the functions $f^1(z^1,z^2)$ (blue) and $f^2(z^1,z^2)$ (orange) for a financial system with 
		$P_0\neq\{\bm{0}\}$. Solid: the unshocked functions. Dashed: the shocked functions.
}\label{3:fig:non:resilient}
\end{figure}

We now aim to describe non-resilience with respect to shocks on $I\subsetneq \tilde{V}$. 
That is, we consider shocks $M$ such that $\E\left[W^{+,r,\alpha}\1\{A=\beta\}\1\{M=0\}\right]>0$ for $(r,\alpha,\beta)\in I$ but possibly $\E\left[W^{+,r,\alpha}\1\{A=\beta\}\1\{M=0\}\right]=0$ for 
$(r,\alpha,\beta)\in \tilde{V}\backslash I$. To this end, denote
\[ T(I) := \overline{\bigcap_{(r_1,\alpha_1,\beta_1)\in I}\left\{\bm{z}\in\R_{+,0}^V\,:\,f^{r_1,\alpha_1,\beta_1}(\bm{z})<0\right\} \cap \bigcap_{(r_2,\alpha_2,\beta_2)\in \tilde{V}\backslash I} \left\{\bm{z}\in\R_{+,0}^V\,:\,f^{r_2,\alpha_2,\beta_2}(\bm{z}) \leq 0\right\}} \]
and define $\bm{z}_0(I)$ 
by $z_0^{r,\alpha,\beta}(I) := \inf_{\bm{z}\in T(I)}z^{r,\alpha,\beta}$. Lemma \ref{3:lem:z0} shows that $\bm{z}_0(I)$ is the smallest joint root of 
$f^{r,\alpha,\beta}$, $(r,\alpha,\beta)\in V$, that is stable with respect to shocks in the $I$-coordinates. 
\begin{lemma}\label{3:lem:z0}
It holds 
$\bm{z}_0(I) \in P_0 \cap T(I)$ and it is thus a joint root of 
$f^{r,\alpha,\beta}$, $(r,\alpha,\beta)\in V$.
\end{lemma}
We can state a general theorem for non-resilience in terms of $\bm{z}_0(I)$, where $I$ denotes the set of coordinates impacted by 
$M$. 
\begin{theorem}[Non-resilience Criterion]\label{3:thm:non-resilience}
Consider a financial system described by a regular vertex sequence and any ex post infection $M$ with $\E[W^{+,r,\alpha}\1\{A=\beta\}\1\{M=0\}]>0$ for all $(r,\alpha,\beta)\in I$, where $\emptyset\neq I\subseteq \tilde{V}$. Then for any $\epsilon>0$ it holds w.\,h.\,p.~$n^{-1}\mathcal{S}_n^M \geq g(\bm{z}_0(I))-\epsilon$. If $g(\bm{z}_0(I))>0$, then the system is non-resilient with respect to shocks on $I$.
\end{theorem}
As before, usually (e.\,g.~if $\P(S>0)=1$) it will hold that $g(\bm{z}_0(I))>0$ whenever $\bm{z}_0(I)\neq\bm{0}$.

By Theorem \ref{3:thm:non-resilience} we derive for shocks on $\tilde{V}$ that for any $\epsilon>0$ w.\,h.\,p.~$n^{-1}\mathcal{S}_n^M \geq g(\bm{z}_0(\tilde{V}))-\epsilon$ while in Corollary \ref{3:cor:non:resilience:independent} we claimed $n^{-1}\mathcal{S}_n^M \geq g(\bm{z}^*)-\epsilon$ w.\,h.\,p. By the following lemma the two are in fact equivalent.
\begin{lemma}\label{3:lem:z0:equals:z*}
It holds $\bm{z}_0(\tilde{V})=\bm{z}^*$.
\end{lemma}
The identity $\bm{z}_0(I)=\bm{z}^*$ on the other hand does not necessarily imply $I=\tilde{V}$. 

\section{Systemic Capital Requirements}\label{3:sec:capital:block:model}
In this section, we apply the asymptotic theory developed in the previous sections in order to determine sufficient capital requirements ensuring resilience of a multi-type financial network to small initial shocks. This goes one step further than the resilience criterion from Theorem \ref{3:thm:resilience}, enabling a regulator to actively manage the system's stability. In fact, in Subsection \ref{3:ssec:suff:cap:requ}, we derive a whole family of different capital requirements parametrized by a vector $\bm{v}\in\R_+^V$. Thus, in Subsection \ref{3:ssec:interpretation:v}, we further provide an interpretation of $\bm{v}$ and we demonstrate how for different stress scenarios, amplification due to contagion effects can be controlled by means of our derived capitals.

\subsection{A Family of Capital Requirements}\label{3:ssec:suff:cap:requ}
By Theorem \ref{3:thm:resilience} and Lemma \ref{3:lem:sufficient:criteria:z:star}, for resilience of an a priori unshocked system it is sufficient that the directional derivative $D_{\bm{v}}f^{r,\alpha,\beta}(\bm{0})<0$ for all $(r,\alpha,\beta)\in V$ and some direction $\bm{v}\in\R_+^V$. It is thus the aim of this subsection to construct capitals $c_i(\bm{v})$, $i\in[n]$, for any given $\bm{v}\in\R_+^V$, such that $D_{\bm{v}}f^{r,\alpha,\beta}(\bm{0})$ exists and satisfies above condition.

In analogy to Chapter \ref{chap:systemic:risk}, for an institution $i\in[n]$ of type $\alpha_i=\beta$, we choose the following form for the capital:
\begin{equation}\label{3:eqn:def:capitals}
c_i(\bm{v}) = \max\left\{R+1,\left\lceil\mu^\beta \left( \frac{e_i(\bm{v})}{\Vert\bm{v}\Vert} \right)^{\nu^\beta}\right\rceil\right\},
\end{equation}
where $\{\mu^\beta\}_{\beta\in[T]}\subset\R_+$, $\{\nu^\beta\}_{\beta\in[T]}\subset\R_{+,0}$ and
\[ e_i(\bm{v}) := \sum_{s\in[R]} s \sum_{\gamma\in[T]}w_i^{-,s,\gamma} v^{s,\beta,\gamma} \]
can be thought of as the total exposure of $i$ weighted according to $\bm{v}$. 
Taking the maximum with $R+1$ in \eqref{3:eqn:def:capitals} ensures that $i$ does not have any contagious links to other institutions. If $W^{-,R,\gamma}\big\vert_{A=\beta}=0$ for all $\gamma\in[T]$ and institutions of type $\beta$ thus do not have any $R$-edges (in the limit of large systems), then one could decrease this lower bound on $c_i(\bm{v})$ accordingly. Moreover, as $e_i(\bm{v})$ linearly depends on $\bm{v}$, it makes sense to normalize it in \eqref{3:eqn:def:capitals}. By this choice only the direction of $\bm{v}$ is important and $\mu^\beta$ can be chosen independently of $\bm{v}$ in the following. Finally, assume that the uncapitalized system is a regular vertex sequence according to Definition \ref{3:def:regular:vertex:sequence} with limiting weights $W^{\pm,r,\alpha}$ and limiting vertex type $A$. Denote the limiting random variable of $\{e_i(\bm{v})\}_{i\in[n]}$ by
\[ E(\bm{v}) = \sum_{\beta\in[T]} \1\{A=\beta\} \sum_{s\in[R]}s\sum_{\gamma\in[T]}W^{-,s,\gamma}v^{s,\beta,\gamma} \]
and assume that also the limiting random variable of $\{c_i(\bm{v})\}_{i\in[n]}$ exists and satisfies
\[ C(\bm{v}) = \sum_{\beta\in[T]} \1\{A=\beta\} \max\left\{ R+1 , \left\lceil\mu^\beta \left( \frac{E(\bm{v})}{\Vert\bm{v}\Vert} \right)^{\nu^\beta}\right\rceil \right\}\quad\text{a.\,s.,} \]
which is true at least for absolutely continuously distributed in-weights (cf.~Subsection \ref{2:ssec:threshold:requirements}). Note, however, that also for more general distributions we can proceed similarly and derive analogue results as in the following.

We can then describe the behavior of $f^{r,\alpha,\beta}(\bm{z})$ near $\bm{0}$.
\begin{proposition}\label{3:prop:limsup:liminf}
For $\nu^\beta\geq1$, it holds
\[ D_{\bm{v}}f^{r,\alpha,\beta}(\bm{0}) = \lim_{h\to0+}h^{-1}f^{r,\alpha,\beta}(h\bm{v}) = -v^{r,\alpha,\beta}. \]
For $\nu^\beta<1$, it holds
\[ \limsup_{h\to0+}h^{-1}f^{r,\alpha,\beta}(h\bm{v}) \leq \Vert\bm{v}\Vert \limsup_{h\to0+}h^{-1}\E\left[ W^{+,r,\alpha}\1\{A=\beta\}\1\left\{ \frac{E(\bm{v})}{\Vert\bm{v}\Vert} > \left( \frac{h}{\mu^\beta} \right)^\frac{1}{\nu^\beta-1} \right\} \right] - v^{r,\alpha,\beta}, \]
and
\[ \liminf_{h\to0+}h^{-1}f^{r,\alpha,\beta}(h\bm{v}) \geq \Vert\bm{v}\Vert \liminf_{h\to0+}h^{-1}\E\left[ W^{+,r,\alpha}\1\{A=\beta\}\1\left\{ \frac{E(\bm{v})}{\Vert\bm{v}\Vert} > \left( \frac{h}{\mu^\beta} \right)^\frac{1}{\nu^\beta-1} \right\} \right] - v^{r,\alpha,\beta}. \]
\end{proposition}
We can then immediately conclude that the financial system becomes resilient if $\nu^\beta\geq1$ for all $\beta\in[T]$ (also see Corollary \ref{3:cor:resilient}). For $\nu^\beta<1$, however, more work is needed and in particular the choice of $\mu^\beta$ becomes important. Note that analogue to Lemma \ref{3:lem:sufficient:criteria:z:star} it is sufficient to assume $\limsup_{h\to0+}h^{-1}f^{r,\alpha,\beta}(h\bm{v})<0$ to ensure that $\bm{z}^*=\bm{0}$ and the system to be resilient. We could therefore work with upper bounds throughout all of this section. For a more concise presentation and a notion of sharpness regarding the capital requirements, however, we make the following exact assumption on the distribution tails of $W^{+,r,\alpha}\big\vert_{A=\beta}$ and $E(\bm{v})\big\vert_{A=\beta}$:
\[ u^{k^{+,r,\alpha,\beta}-1}\left(1-F_{W^{+,r,\alpha}\vert_{A=\beta}}(u)\right)\to \left(w_\text{min}^{+,r,\alpha,\beta}\right)^{k^{+,r,\alpha,\beta}-1},\quad\text{as }u\to\infty \]
for some constants $k^{+,r,\alpha,\beta}>2$ and $w_\text{min}^{+,r,\alpha,\beta}\in\R_+$, respectively
\[ u^{k^{-,\beta}-1}\left(1-F_{E(\bm{v})\vert_{A=\beta}}(u)\right)\to \left(e_\text{min}^\beta(\bm{v})\right)^{k^{-,\beta}-1},\quad\text{as }u\to\infty \]
for some constants $k^{-,\beta}>2$ and $e_\text{min}^\beta(\bm{v})\in\R_+$. That is $W^{+,r,\alpha}\big\vert_{A=\beta}$ and $E(\bm{v})\big\vert_{A=\beta}$ resemble Pareto distributions in the tail. Moreover, for the dependence structure (i.\,e.~the copula) of the two random variables we make the assumption that
\[ \Lambda^{r,\alpha,\beta}(x) = \P(A=\beta)^{-1} \lim_{p\to0}\P\left(\left. F_{W^{+,r,\alpha}\vert_{A=\beta}}(W^{+,r,\alpha})>1-xp\,\right\vert\,F_{E(\bm{v})\vert_{A=\beta}}(E(\bm{v}))>1-p, A=\beta \right) \]
exists for all $x\geq0$. Denote now
\[ \nu_c^{r,\alpha,\beta} := 2+\frac{k^{-,\beta}-1}{k^{+,r,\alpha,\beta}-1}-k^{-,\beta} \]
and
\[ \mu_c^{r,\alpha,\beta} := \Vert\bm{v}\Vert^{\nu_c^\beta} \P(A=\beta) \left( e_\text{min}^\beta(\bm{v}) \right)^{1-\nu_c^\beta} \frac{w_\text{min}^{+,r,\alpha,\beta} \int_0^\infty \Lambda^{r,\alpha,\beta}(x^{1-k^{+,r,\alpha,\beta}})\,\mathrm{d}x}{v^{r,\alpha,\beta}}. \]
Then for each function $f^{r,\alpha,\beta}$, the values $\nu_c^{r,\alpha,\beta}$ and $\mu_c^{r,\alpha,\beta}$ correspond to $\gamma_c$ resp.~$\alpha_c$ from Subsection \ref{2:ssec:threshold:requirements}. In particular, the following holds:
\begin{proposition}\label{3:prop:directional:derivative}
Under above assumptions, the directional derivative of $f^{r,\alpha,\beta}$ in direction $\bm{v}$ is given by
\[ D_{\bm{v}}f^{r,\alpha,\beta}(\bm{0}) = \begin{cases}\infty,&\text{if }\nu^\beta<\nu_c^\beta,\\\left( \frac{\mu_c^{r,\alpha,\beta}}{\mu^\beta}-1 \right)v^{r,\alpha,\beta},&\text{if }\nu^\beta=\nu_c^{r,\alpha,\beta},\\-v^{r,\alpha,\beta},&\text{if }\nu^\beta>\nu_c^{r,\alpha,\beta}.
\end{cases} \]
\end{proposition}
Then as hinted at above using Theorem \ref{3:thm:resilience} and Lemma \ref{3:lem:sufficient:criteria:z:star} we derive the following criterion for resilience.
\begin{corollary}\label{3:cor:resilient}
The system is resilient for a certain choice of $\bm{v}\in\R_+^V$ and capitals as above if for all $(r,\alpha,\beta)\in V$ one of the following holds
\begin{enumerate}
\item $\nu^\beta>\nu_c^{r,\alpha,\beta}$,
\item $\nu^\beta=\nu_c^{r,\alpha,\beta}$ and $\mu^\beta>\mu_c^{r,\alpha,\beta}$.
\end{enumerate}
\end{corollary}
Analogue to Lemma \ref{3:lem:sufficient:criteria:z:star}, we also derive that $\bm{z}^*\neq\bm{0}$ if there exists $\bm{v}\in\R_+^V$ such that for all $(r,\alpha,\beta)\in V$ the directional derivative $D_{\bm{v}}f^{r,\alpha,\beta}(\bm{0})$ exists and $D_{\bm{v}}f^{r,\alpha,\beta}(\bm{0})>0$. Together with Corollary \ref{3:cor:non:resilience:independent} we then derive the following criterion for non-resilience.
\begin{corollary}\label{3:cor:non:resilient}
Let $\P(S>0)=1$. Then the system is non-resilient for a certain choice of $\bm{v}\in\R_+^V$ and capitals as above if for all $(r,\alpha,\beta)\in V$ one of the following holds
\begin{enumerate}
\item $\nu^\beta<\nu_c^{r,\alpha,\beta}$,
\item $\nu^\beta=\nu_c^{r,\alpha,\beta}$ and $\mu^\beta<\mu_c^{r,\alpha,\beta}$.
\end{enumerate}
\end{corollary}
To construct a resilient system by capitals of the form \eqref{3:eqn:def:capitals} we thus choose
\[ \nu^\beta\geq\nu_c^\beta := \max_{r\in[R],\alpha\in[T]}\nu_c^{r,\alpha,\beta} \]
and $\mu^\beta>0$ arbitrary if $\nu^\beta>\nu_c^\beta$ respectively
\[ \mu^\beta > \mu_c^\beta := \max\left\{ \mu_c^{r,\alpha,\beta}\,:\,(r,\alpha)\in[R]\times[T]\text{ such that }\nu_c^{r,\alpha,\beta}=\nu^\beta \right\} \]
if $\nu_c^{r,\alpha,\beta}=\nu^\beta$ for some $(r,\alpha)\in[R]\times[T]$.

\begin{remark}
Assume that $\nu^\beta=\nu_c^{r,\alpha,\beta}>0$ for all $(r,\alpha,\beta)\in V$ (in particular $\nu_c^{r,\alpha,\beta}\geq0$). Moreover, suppose that for a certain choice of $\bm{v}\in\R_+^V$ we find $(s,\gamma,\delta)\in V$ such that $\mu_c^{s,\gamma,\delta}<\mu_c^\delta$. We can then decrease the coordinate $v^{s,\gamma,\delta}$ to
\[ \tilde{v}^{s,\gamma,\delta} = \frac{w_\text{min}^{+,s,\gamma,\delta}\int_0^\infty\Lambda\big( x^{1-k^{+,s,\gamma,\delta}} \big)\mathrm{d}x}{\max_{r\in[R],\alpha\in[T]}\left( v^{r,\alpha,\delta} \right)^{-1} w_\text{min}^{+,r,\alpha,\delta}\int_0^\infty\Lambda\left( x^{1-k^{+,r,\alpha,\delta}} \right)\mathrm{d}x} < v^{s,\gamma,\delta} \]
and leaving all other coordinates unchanged this leads to new critical values $\tilde{\mu}_c^\beta<\mu_c^\beta$, \mbox{$\beta\in[T]$,} (while clearly $\nu_c^\beta$ stays the same). That is, the capital requirements for each individual institution in the system can be lowered as compared to the initial choice $\bm{v}$.

It is thus favorable to choose $\bm{v}$ such that $\mu_c^{r,\alpha,\beta}=\mu_c^\beta$ for all $(r,\alpha,\beta)\in V$ or equivalently
\[ v^{r,\alpha,\beta} = \kappa^\beta w_\text{min}^{+,r,\alpha,\beta}\int_0^\infty\Lambda^{r,\alpha,\beta}\left(x^{1-k^{+,r,\alpha,\beta}}\right)\mathrm{d}x \]
for $\{\kappa^\beta\}_{\beta\in[T]}\subset\R_+$. The family $\{\kappa^\beta\}_{\beta\in[T]}$ can then be understood as a choice by the regulator which loans (i.\,e.~the debtors of which subsystems $\beta$) to regulate more than others.
\end{remark}

\subsection{Interpreting Direction \bf{\emph{v}}}\label{3:ssec:interpretation:v}

In the previous subsection, we have constructed capitals parametrized by some arbitrarily chosen vector $\bm{v}\in\R_+^V$ that ensure resilience of a given financial system as $D_{\bm{v}}f^{r,\alpha,\beta}(\bm{0})<0$ for all $(r,\alpha,\beta)\in V$. The proofs, however, do not convey any intuition about the direction $\bm{v}$.

In this subsection, we therefore change our point of view on $\bm{v}$ and interpret it as a specific stress scenario for the financial system. This perspective allows us to secure the system against any anticipated shock and at this in particular to ensure that amplification does not exceed a certain arbitrarily chosen factor.

Recall from Section \ref{3:sec:resilience} the notion of an ex post shock and consider a family of ex post shocks $\{M_\epsilon\}_{\epsilon>0}$ such that
\begin{equation}\label{3:eqn:def:v}
\E\left[ W^{+,r,\alpha}\1\{A=\beta\}\1\{M_\epsilon=0\} \right] = \epsilon v^{r,\alpha,\beta} + o(\epsilon)
\end{equation}
for some $\bm{v}\in\R_+^V$. If for example each institution defaults ex post with probability $\epsilon$ and independent of its parameters, then
\[ \E\left[ W^{+,r,\alpha}\1\{A=\beta\}\1\{M_\epsilon=0\} \right] = \epsilon \E\left[ W^{+,r,\alpha}\1\{A=\beta\}\right] \]
and $\bm{v}=\bm{\zeta}$, where $\zeta^{r,\alpha,\beta}=\E[W^{+,r,\alpha}\1\{A=\beta\}]$ (if $\zeta^{r,\alpha,\beta}=0$ for some $(r,\alpha,\beta)\in V$, we leave out the corresponding coordinate in the following). If, however, an institution of type $\beta\in[T]$ defaults ex post with probability $\epsilon\eta^\beta$ for a family $\{\eta^\beta\}_{\beta\in[T]}\subset\R_+$ and independent of all other parameters, then $v^{r,\alpha,\beta}=\eta^\beta\zeta^{r,\alpha,\beta}$. While shocks of this type already provide a rich basis of different stress scenarios, other shocks $M_\epsilon$, targeting institutions according to their weights, are also feasible as long as they satisfy \eqref{3:eqn:def:v} for some $\bm{v}\in\R_+^V$.

The vector $\bm{v}$ then describes how different parts of the financial system are affected by the initial shock. More precisely, by Corollary \ref{3:cor:poisson:degrees} and \eqref{3:eqn:def:v} one finds
\[ n \E[W^{+,r,\alpha}\1\{A=\beta\}\1\{M_\epsilon=0\}] \E[W^{-,r,\beta}\1\{A=\alpha\}] + o(n) = n \E[W^{-,r,\beta}\1\{A=\alpha\}] \epsilon v^{r,\alpha,\beta} + o(n) \]
for the number of $r$-edges from initially defaulted $\beta$-institutions to $\alpha$-institutions. Thus
\begin{equation}\label{3:eqn:initial:damage}
\tau^\beta := \sum_{r\in[R]} r \sum_{\alpha\in[T]} \E[W^{-,r,\beta}\1\{A=\alpha\}] \epsilon v^{r,\alpha,\beta}
\end{equation}
can asymptotically be identified as the total monetary damage (not to be confused with the systemic damage according to their systemic importance values) that initially defaulted institutions of type $\beta$ cause to all the institutions in the system (as always in this thesis under the assumption of zero recovery rate). 

We now want to construct capitals $c_i(\bm{v})$ such that this initial damage is amplified by some prescribed amplification factor $B\geq1$ at most. Note that analogue to \eqref{3:eqn:initial:damage} and using Theorem \ref{3:thm:general:weights} the final damage is upper bounded (and typically precisely given) by
\[ \sum_{r\in[R]} r \sum_{\alpha\in[T]} \E[W^{-,r,\beta}\1\{A=\alpha\}] (z_\epsilon^*)^{r,\alpha,\beta}, \]
where $\bm{z}_\epsilon^*$ denotes the equivalent of $\bm{z}^*$ for the shocked system and the proof of Theorem \ref{3:thm:general:weights} shows that $(z_\epsilon^*)^{r,\alpha,\beta}$ asymptotically equals the total $(r,\alpha)$-out-weights of all finally defaulted $\beta$-institutions divided by $n$.

Our goal can then be achieved in the spirit of the previous subsection choosing $\bm{v}$ from \eqref{3:eqn:def:v} for $c_i(\bm{v})$ in \eqref{3:eqn:def:capitals} and letting $\nu^\beta\geq\nu_c^\beta$ and $\mu^\beta\geq\mu_c^\beta/(1-B^{-1})$ for some $B>1$ that can be considered the maximal amplification factor according to the following result.
\begin{theorem}\label{3:thm:amplification}
With the above described capitals, the final damage in the system caused by $\beta$-institutions is upper bounded by $B\tau^\beta$ asymptotically as $\epsilon\to0$.
\end{theorem}
Note that any choice for $\bm{v}\in\R_+^V$ would make the financial system resilient to any initial shock by Corollary \ref{3:cor:resilient}. Whereas resilience according to Definition \ref{3:def:resilience} is to be understood in an asymptotic sense with amplification factors not being relevant, however, for practical purposes the notion of an amplification factor is very desirable. Albeit still in an asymptotic sense, using Theorem \ref{3:thm:amplification} this can be achieved when calibrating $\bm{v}$ to the anticipated shock according to \eqref{3:eqn:def:v}. Finally note that different amplification factors $B^\beta$ for the different subsystems of type $\beta\in[T]$ can be achieved by the same means as in Theorem \ref{3:thm:amplification} but choosing capitals $c_i(\tilde{\bm{v}})$ with
\[ \tilde{v}^{r,\alpha,\beta} = B^\beta v^{r,\alpha,\beta} \]
and $\mu^\beta\geq\mu_c^\beta/(1-(B^{\beta})^{-1})$.

\section{A Multi-type Exposure Model}\label{3:sec:exposure:model}
So far, in this chapter we have proposed and analyzed a model for default contagion that can be seen as a multi-variate version of the threshold model from Subsection \ref{2:ssec:special:case:threshold:model} accounting for multiple institution-types and allowing for assortative edge-weights regarding these types. While it is possible to interpret these edge-weights as exposures between institutions directly, a reasonable model for a financial network would certainly require a very large parameter $R$ and the model would thus become very high-dimensional. In this section, we therefore propose an alternative model combining the ideas from this chapter so far with the exposure model from Chapter \ref{chap:systemic:risk}.

In addition to all the previous model parameters, we thus assign to each institution $i\in[n]$ an exchangeable sequence of almost surely positive random variables $\{E_i^{j,r}\}_{j\in[n]\backslash\{i\},r\in[R]}\subset\mathcal{L}^0(\R_+)$ that we interpret as possible exposures of $i$. Note that compared to Chapter \ref{chap:systemic:risk} we consider not one but $R$ possible exposures for every pair $(i,j)$. This allows us to interpret the edge-weight from Subsection \ref{3:ssec:random:edges} as a multiplier in the following sense:

Let as before $X_{i,j}^r$ be the indicator random variable of an $r$-weighted edge going from institution $i$ to $j$. As $\sum_{r\in[R]}X_{i,j}^r\in\{0,1\}$, we can thus denote the edge-weight between $i$ and $j$ by $r(i,j)$ with $r(i,j)=0$ if $\sum_{r\in[R]}X_{i,j}^r=0$, or such that $X_{i,j}^{r(i,j)}=1$ otherwise. Let then the exposure $e_{i,j}$ be defined by
\[ e_{i,j} := \sum_{1\leq s\leq r(i,j)} E_j^{i,s}, \]
i.\,e.~as the sum over the first $r(i,j)$ possible exposures from the list $\{E_j^{i,r}\}_{r\in[R]}$. In particular, $e_{i,j}=0$ if there is no edge going from $i$ to $j$.

Moreover, we can allow for more freedom in the choice of the capitals $c_i$, $i\in[n]$. Instead of deterministic integer values we consider $c_i\in\mathcal{L}^0(\R_{+,0,\infty})$ in the following. We then define for each institution $i\in[n]$ the random threshold value $q_i$ similar as in Chapter \ref{chap:systemic:risk},
\[ q_i := \inf\left\{ q\in\{0\}\cup[(n-1)R]\,:\,\sum_{1\leq s\leq q}E_i^s \geq c_i \right\}, \]
where $\{E_i^s\}_{s\in[(n-1)R]}$ shall be an arbitrary (but fixed) enumeration of the set of random variables $\{E_i^{j,r}\}_{j\in[n]\backslash\{i\},r\in[R]}$ and $\inf\emptyset=\infty$ by convention. By exchangeability the distribution of $q_i$ does not depend on the choice of the enumeration. We now adjust Definition \ref{3:def:regular:vertex:sequence} to this new setting:
\begin{assumption}\label{3:ass:regularity:thresholds}
Consider a sequence $(\bm{w}^-(n),\bm{w}^+(n),\bm{s}(n),\bm{q}(n),\bm{\alpha}(n))$ of financial systems, where for $n\in\N$ as before $\bm{w}^-(n)=(w_i^-(n))_{i\in[n]}$, $\bm{w}^+(n)=(w_i^+(n))_{i\in[n]}$, $\bm{s}(n)=(s_i)_{i\in[n]}$, $\bm{\alpha}(n)=(\alpha_i(n))_{i\in[n]}$ and additionally $\bm{q}(n)=(q_i)_{i\in[n]}$. For each $n$, we denote the random empirical distribution function by
\[ G_n(\bm{x},\bm{y},v,l,m) = n^{-1} \sum_{i\in[n]}\prod_{r\in[R],\alpha\in[T]}\1\left\{ w_i^{-,r,\alpha}\leq x^{r,\alpha},w_i^{+,r,\alpha}\leq y^{r,\alpha} \right\} \1\{ s_i\leq v, q_i\leq l, \alpha_i\leq m \}, \]
for $(\bm{x},\bm{y},v,l,m)\in U := \R_{+,0}^{[R]\times[T]}\times\R_{+,0}^{[R]\times[T]}\times\R_{+,0}\times\N_{0,\infty}\times[T]$. Then we assume the following:
\begin{enumerate}
\item \textbf{Almost sure convergence in distribution:} There exists a deterministic distribution function $G$ on $U$ such that all points $(\bm{x},\bm{y},v,l,m)$ for which $G_{l,m}(\bm{x},\bm{y},v):=G(\bm{x},\bm{y},v,l,m)$ is continuous in $(\bm{x},\bm{y},v)$, it holds almost surely $\lim_{n\to\infty}G_n(\bm{x},\bm{y},v,l,m)=G(\bm{x},\bm{y},v,l,m)$. Denote by $(\bm{W}^-,\bm{W}^+,S,Q,A)$ a random vector distributed according to $G$.
\item \textbf{Convergence of average weights and systemic importance:} The random variables $W^{\pm,r,\alpha}$, $(r,\alpha)\in[R]\times[T]$, and $S$ are integrable and $\int_U x^{r,\alpha}\dd G(\bm{x},\bm{y},v,l,m) \to \E[W^{-,r,\alpha}]$, $\int_U y^{r,\alpha}\dd G(\bm{x},\bm{y},v,l,m) \to \E[W^{+,r,\alpha}]$ as well as $\int_U v\,\dd G(\bm{x},\bm{y},v,l,m) \to \E[S]$ as $n\to\infty$.
\end{enumerate}
\end{assumption}
We can then define the analogues of $f^{r,\alpha,\beta}$ and $g$ from Section \ref{3:sec:asymptotic:results} for the new setting by
\begin{equation*}
\begin{gathered}
f_Q^{r,\alpha,\beta}(\bm{z}) = \E\Bigg[W^{+,r,\alpha}\psi_Q\Bigg(\sum_{\gamma\in[T]}W^{-,1,\gamma}z^{1,\beta,\gamma},\ldots,\sum_{\gamma\in[T]}W^{-,R,\gamma}z^{R,\beta,\gamma}\Bigg)\1\{A=\beta\}\Bigg] - z^{r,\alpha,\beta},\\
g_Q(\bm{z}) = \sum_{\beta\in[T]}\E\Bigg[S \psi_Q\Bigg(\sum_{\gamma\in[T]}W^{-,1,\gamma}z^{1,\beta,\gamma},\ldots,\sum_{\gamma\in[T]}W^{-,R,\gamma}z^{R,\beta,\gamma}\Bigg)\1\{A=\beta\}\Bigg],
\end{gathered}
\end{equation*}
where as before $\psi_l(x_1,\ldots,x_R) := \P\left(\sum_{s\in[R]}sX_s\geq l\right)$ for independent Poisson random variables $X_s\sim\mathrm{Poi}(x_s)$, $s\in[R]$. Clearly all results from Section \ref{3:sec:asymptotic:results} directly transfer to the new functions $f_Q^{r,\alpha,\beta}$ and $g_Q$. In particular, denote the analogues of $\hat{\bm{z}}$ and $\bm{z}^*$ by $\hat{\bm{z}}_Q$ resp.~$\bm{z}_Q^*$. Then we derive the following generalization of Theorem \ref{3:thm:general:weights}.
\begin{theorem}\label{3:exposures:final:fraction}
Consider a financial system satisfying Assumption \ref{3:ass:regularity:thresholds} and let $\hat{\bm{z}}_Q$, $\bm{z}_Q^*$ as described above. Then for the final systemic damage due to default contagion $\mathcal{S}_n$ it holds
\[ g_Q(\hat{\bm{z}}_Q) + o_p(1) \leq n^{-1}\mathcal{S}_n \leq g_Q(\bm{z}_Q^*) + o_p(1). \]
In particular, if $\hat{\bm{z}}_Q=\bm{z}_Q^*$, then $n^{-1}\mathcal{S}_n\stackrel{p}{\to}g_Q(\hat{\bm{z}}_Q)$ as $n\to\infty$.
\end{theorem}
By the same means as in Section \ref{3:sec:resilience} but applying Theorem \ref{3:exposures:final:fraction} instead of Theorem \ref{3:thm:general:weights}, we then derive the following characterizations of resilient (according to Definition \ref{3:def:resilience}) and non-resilient (according to Definition \ref{3:def:non:resilience}) financial systems. To this end, we consider an ex post infection that is described by the limiting random variable $M\in\{0,1\}$ (extend Assumption \ref{3:ass:regularity:thresholds}). In particular, note that it makes no difference if we apply the ex post shock on the capital $c_i$ or the threshold $q_i$.
\begin{corollary}\label{3:cor:resilience:exposures}
Consider a financial system satisfying Assumption \ref{3:ass:regularity:thresholds} that is a priori unshocked in the sense that $\P(Q=0)=0$. If $\bm{z}_Q^*=\bm{0}$, then the system is resilient.
\end{corollary}

\begin{corollary}\label{3:cor:non:resilience:exposures}
Consider a financial system satisfying Assumption \ref{3:ass:regularity:thresholds} that is a priori unshocked in the sense that $\P(Q=0)=0$. Moreover, let the ex post shock $M$ be such that $\E[W^{+,r,\alpha}\1\{A=\beta\}\1\{M=0\}]>0$ for all $(r,\alpha,\beta)\in V$ such that $\E[W^{+,r,\alpha}\1\{A=\beta\}]>0$. Then for any $\epsilon>0$ w.\,h.\,p.{} $n^{-1}\mathcal{S}_n^M \geq g(\bm{z}^*)-\epsilon$. If 
$g(\bm{z}^*)>0$, 
then the system is non-resilient.
\end{corollary}
In Corollary \ref{3:cor:non:resilience:exposures} we consider shocks that infect every possible part of the system. In the same way also shocks on subsystems only can be considered as a corollary of Theorem \ref{3:thm:non-resilience}.

Finally, we can also derive capital requirements for the multi-type exposure model similar as in Subsection \ref{2:ssec:capital:requirements} but using Corollary \ref{3:cor:resilient}. To this end, denote for $i\in[n]$ by $\{E_i^s\}_{s\in[(n-1)R]}$ an enumeration of $i$'s exposure list, by $\lambda_i<\infty$ their mutual mean and $S_k^i=\sum_{j=1}^k E_i^j$. Moreover, assume that there exists $t>1$ such that for all $\epsilon>0$ and uniformly for all $i\in[n]$ it holds
\[ k^{t-1}\P\left( S_k^i\geq(1+\epsilon)k\lambda_i \right) \to0,\quad\text{as }k\to\infty, \]
and for all $x>1$,
\[ k^{tx-1}\P\left( S_k^i\geq\epsilon\lambda_i k^x \right) \to0,\quad\text{as }k\to\infty. \]

\begin{corollary}\label{3:cor:exposures:capital:requirements}
Fix some direction $\bm{v}\in\R_{+,0}^V$ and let $\nu_c^{r,\alpha,\beta}$ and $\mu_c^{r,\alpha,\beta}$ be as in Section \ref{3:sec:capital:block:model}. Consider then the setting described above and assume that
\[ c_i > \max\left\{ \sum_{s\in S} E_i^s\,:\, \vert S\vert=R \right\} \]
as well as
\[ c_i \geq \lambda_i \mu^\beta \left( \frac{\sum_{s\in[R]}s\sum_{\gamma\in[T]}w_i^{-,s,\gamma}v^{s,\beta,\gamma}}{\Vert\bm{v}\Vert} \right)^{\nu^\beta},\quad\text{if }\alpha_i=\beta, \]
for every $i\in[n]$, where $\{\nu^\beta\}_{\beta\in[T]}\subset\R_{+,0}$ and $\{\mu^\beta\}_{\beta\in[T]}\subset\R_+$. Then the system is resilient if for all $(r,\alpha,\beta)\in V$ one of the following holds
\begin{enumerate}
\item $\nu^\beta>\nu_c^{r,\alpha,\beta}$,
\item $\nu^\beta=\nu_c^{r,\alpha,\beta}>0$ and $\mu^\beta>\mu_c^{r,\alpha,\beta}$.
\end{enumerate}
\end{corollary}

\section{Applications}\label{3:sec:applications}
The theory developed in the previous sections allows to investigate many interesting novel settings as compared to previous literature. In this section, we discuss some of them and highlight their implications. Further, we demonstrate the applicability of our asymptotic results also for finite networks of reasonable size by numerical simulations. To make the notion of non-resilience easier accessible we consider the case of $s_i=1$ for all $i\in[n]$ and hence $\mathcal{S}_n=\vert\mathcal{D}_n\vert$ throughout this subsection. Generalizations under mild assumptions on $S$ are straight-forward.

In the first example, we investigate the influence of a non-resilient subsystem in a global system. Unsurprisingly the global system turns out to be non-resilient as well and we can further show that even resilient network parts become non-resilient by their connections to the non-resilient subsystem, i.\,e.~every howsoever small infection that occurs only within the resilient part of the system finally spreads to a lower bounded fraction of the resilient subsystem. 

\begin{example}\label{3:ex:non-res:subsystem}
For simplicity assume $R=1$ and denote $z^{\alpha,\beta}:=z^{1,\alpha,\beta}$, $f^{\alpha,\beta}(\bm{z}):=f^{1,\alpha,\beta}(\bm{z})$ and $W^{\pm,\alpha}:=W^{\pm,1,\alpha}$ in the following. Consider then a $1$-type banking system, described by the random vector $(\tilde{W}^-,\tilde{W}^+,\tilde{C})$, where $\P(\tilde{W}^+>0)=1$ and $\P(\tilde{C}=0)=0$, and assume that it is non-resilient
. In the $1$-dimensional case this breaks down to the existence of $\tilde{z}_0>0$ such that

\[ \tilde{f}(z) := \E\left[\tilde{W}^+\P\left(\mathrm{Poi}\left(\tilde{W}^-z\right)\geq\tilde{C}\right)\right] - z \geq 0, \quad\text{for all }z\in[0,\tilde{z}_0]. \]

\noindent Now introduce a second (possibly resilient) subsystem to the network. That is, the system is now described by the random vector $(W^{\pm,1},W^{\pm,2},C,A)$, where $\P(C=0)=0$, $A\in\{1,2\}$ and $\alpha_i=1$ means that bank $i\in[n]$ is in the non-resilient subsystem, whereas $\alpha_i=2$ means that $i$ is part of the second subsystem. In order for the characteristics of the non-resilient subsystem to be preserved, we require that $W^{-,1}\vert_{A=1} \stackrel{d}{=} \tilde{W}^-$, $W^{+,1}\vert_{A=1} \stackrel{d}{=} \tilde{W}^+/\P(A=1)$ (to account for the changed number of banks; due to the multiplicative form in \eqref{3:eqn:edge:prob} it is sufficient to adjust either in- or out-weights by $\P(A=1)$) and $C\vert_{A=1}\stackrel{d}{=}\tilde{C}$. We derive that

\[ f^{1,1}(\bm{z}) = \E\left[W^{+,1}\P\left(\mathrm{Poi}\left(W^{-,1}z^{1,1}+W^{-,2}z^{1,2}\right)\geq C\right)\1\{A=1\}\right] - z^{1,1} \geq \tilde{f}\left(z^{1,1}\right) \geq 0 \]

\noindent for all $\bm{z}=(z^{1,1},z^{1,2},z^{2,1},z^{2,2})$ with $z^{1,1}\in[0,\tilde{z}_0]$ and in particular $z_0^{1,1}(I)\geq\tilde{z}_0>0$, where $I:=\{(1,1)\}$. 
An application of Theorem \ref{3:thm:non-resilience} then yields that the fraction of finally defaulted banks in the network is lower bounded by

\begin{align*}
g(\bm{z}_0(I)) &= \E\left[\P\left(\mathrm{Poi}\left(W^{-,1}z_0^{1,1}(I)+W^{-,2}z_0^{1,2}(I)\right)\geq C\right)\1\{A=1\}\right]\\
&\hspace{4.5cm}+ \E\left[\P\left(\mathrm{Poi}\left(W^{-,1}z_0^{2,1}(I)+W^{-,2}z_0^{2,2}(I)\right)\geq C\right)\1\{A=2\}\right]
\end{align*}

\noindent w.\,h.\,p.~for any ex post infection $M$ satisfying $\P(M=0, A=1)>0$ (i.\,e.~infecting some banks in the non-resilient subsystem). That is, if a small fraction of banks in the non-resilient subsystem defaults due to an external shock event, then this infection spreads to the whole system and the fraction of finally defaulted banks in the second subsystem is w.\,h.\,p.~lower bounded~by

\begin{equation}\label{3:eqn:lower:bound:second:subsystem}
\E\left[\left.\P\left(\mathrm{Poi}\left(W^{-,1}z_0^{2,1}(I)+W^{-,2}z_0^{2,2}(I)\right)\geq C\right)\,\right\vert\,A=2\right].
\end{equation}

\noindent In fact, if we assume that $W^{+,2}\vert_{A=1}>0$ almost surely and $\P(W^{-,1}>0,C<\infty,A=2)>0$ (that is, there are some banks in the second subsystem lending to banks in the non-resilient subsystem), then it must hold that
\[ z_0^{2,1}(I) \geq \E\left[W^{+,2}\P\left(\mathrm{Poi}\left(W^{-,1}z_0^{1,1}(I)\right)\geq C\right)\1\{A=1\}\right] >0 \]

\noindent and hence the lower bound \eqref{3:eqn:lower:bound:second:subsystem} is strictly positive. That is, every howsoever small infected fraction in the non-resilient subsystem spreads to a lower bounded fraction of finally defaulted banks in the second subsystem as well.

Now finally assume that $W^{+,1}\vert_{A=2}>0$ almost surely and $\P(W^{-,2}>0,C<\infty,A=1)>0$ (that is, there are some banks in the non-resilient subsystem lending to banks in the second subsystem). By considering the function

\[ f^{1,2}(\bm{z}) = \E\left[W^{+,1}\P\left(\mathrm{Poi}\left(W^{-,1}z^{2,1}+W^{-,2}z^{2,2}\right)\geq CM\right)\1\{A=2\}\right] - z^{1,2}, \]

\noindent we derive that $(\hat{z}^M)^{1,2}>0$ for any ex post infection $M$ such that $\P(M=0,A=2)>0$ (that is, infecting some banks in the second subsystem) and hence also $(\hat{z}^M)^{1,1}>0$ by the form of $f^{1,1}(\bm{z})$ (see above). By the same means as before, we hence conclude that in fact $\hat{\bm{z}}^M\geq\bm{z}_0(I)$ and so the lower bounds derived above still hold. In particular, this means that every howsoever small initial shock to the second (possibly resilient) subsystem causes the default of a lower bounded fraction of banks in the second subsystem. That is, by connecting to the non-resilient subsystem the a priori possibly resilient second subsystem becomes non-resilient as well.

\end{example}
To better understand the phenomenon in Example \ref{3:ex:non-res:subsystem}, consider the following example: Let

\begin{align*}
W^{-,1}\vert_{A=1} &= w_1, & W^{+,1}\vert_{A=1} &= 2w_1, & C\vert_{A=1} &= 1,\\
W^{-,2}\vert_{A=2} &= w_2, & W^{+,2}\vert_{A=2} &= 2w_2, & C\vert_{A=2} &= 2,
\end{align*}

\noindent for $w_1>1$ and $w_2>0$. It is then easy to confirm that the type-$1$ subsystem is in fact non-resilient and the type-$2$ subsystem resilient. (Both subnetworks are Erd\"{o}s-R\'{e}nyi random graphs. In the first subnetwork every edge is contagious, in the second none is.) Additionally, we assume $\P(A=1)=\P(A=2)=1/2$ and $W^{\pm,1}\vert_{A=2} = W^{\pm,2}\vert_{A=1} = w_3 >0$. In particular, 

\begin{align*}
w_3\left(f^{1,1}(\bm{z})+z^{1,1}\right) &= 2w_1\left(f^{2,1}(\bm{z})+z^{2,1}\right), & w_3\left(f^{2,2}(\bm{z})+z^{2,2}\right) &= 2w_2\left(f^{1,2}(\bm{z})+z^{1,2}\right)
\end{align*}

\noindent and hence it must hold that $z_0^{2,1}(I) = (2w_1)^{-1}w_3\,z_0^{1,1}(I)$ resp.~$z_0^{1,2}(I) =(2w_2)^{-1}w_3\,z_0^{2,2}(I)$. The problem then reduces to $f^1(z^1,z^2) = 0$ and $f^2(z^1,z^2)=0$, where $z^1:=z^{1,1}$, $z^2:=z^{2,2}$ and

\begin{align*}
f^1(z^1,z^2) &:= w_1\left(1-e^{-w_1z^1-(2w_2)^{-1}w_3^2z^2}\right) - z^1,\\
f^2(z^1,z^2) &:= w_2\left(1-e^{-(2w_1)^{-1}w_3^2z^1-w_2z^2}\left(1+(2w_1)^{-1}w_3^2z^1+w_2z^2\right)\right) - z^2.
\end{align*}

\noindent Depending on the choice of the weights $w_i$, $i=1,2,3$, the system shows slightly different behavior, as illustrated in Figures \ref{3:fig:ex1:1jointRoot}, \ref{3:fig:ex1:1jointRoot:hump} and \ref{3:fig:ex1:3jointRoots}. In all cases, 
$\bm{z}_0(I)=\bm{z}^*\neq\bm{0}$ which determines a strictly positive lower bound on the final default fraction as shown in Example \ref{3:ex:non-res:subsystem}. However, $\bm{z}^*$ shows a jump discontinuity at certain choices for the weights $w_i$, $i=1,2,3$, where in particular $(z^*)^2$ changes drastically. Revisiting the definition of $z^2=z^{2,2}$, this can be interpreted as the resilient subsystem suddenly experiencing a lot more defaults. See \cite{Janson2012} for similar results on the Erd\"{o}s-R\'{e}nyi random graph.

\begin{figure}[t]
    \hfill\subfigure[]{\includegraphics[width=0.32\textwidth]{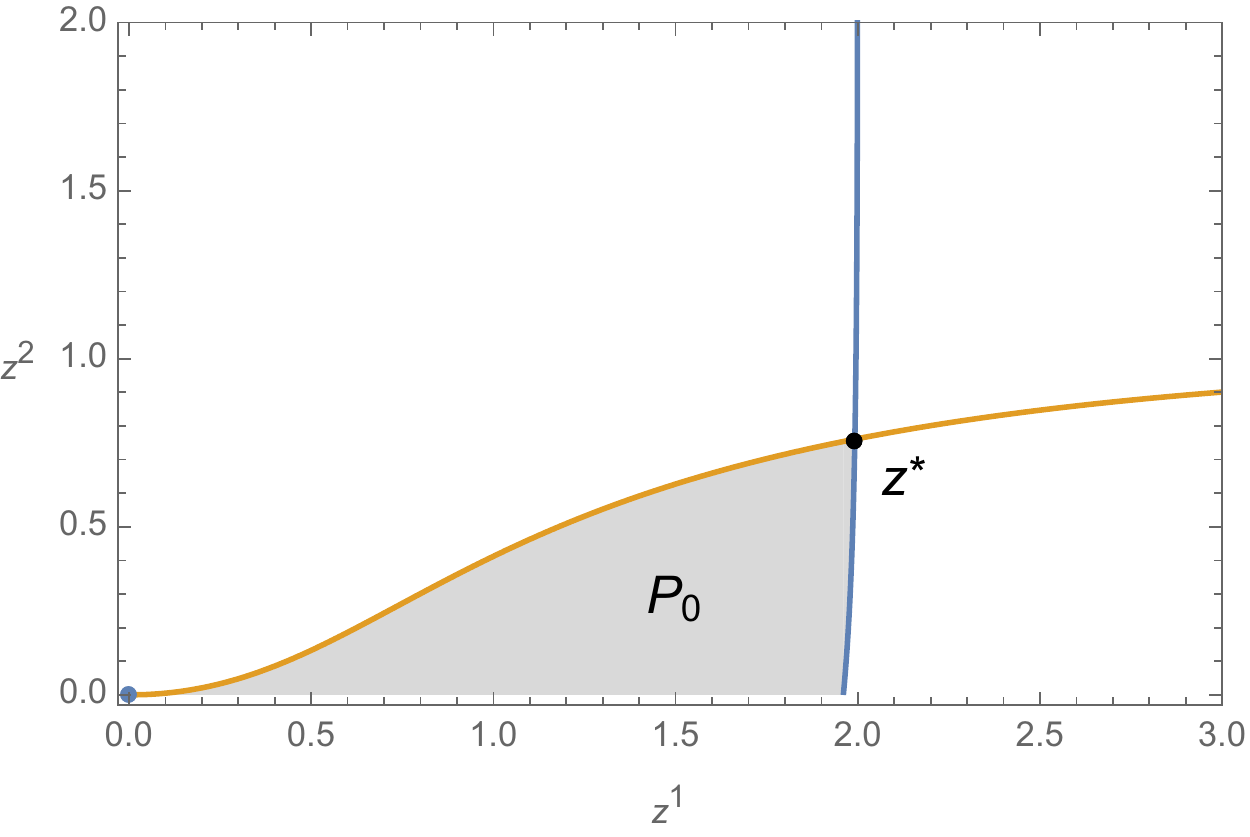}\label{3:fig:ex1:1jointRoot}}
    \hfill\subfigure[]{\includegraphics[width=0.32\textwidth]{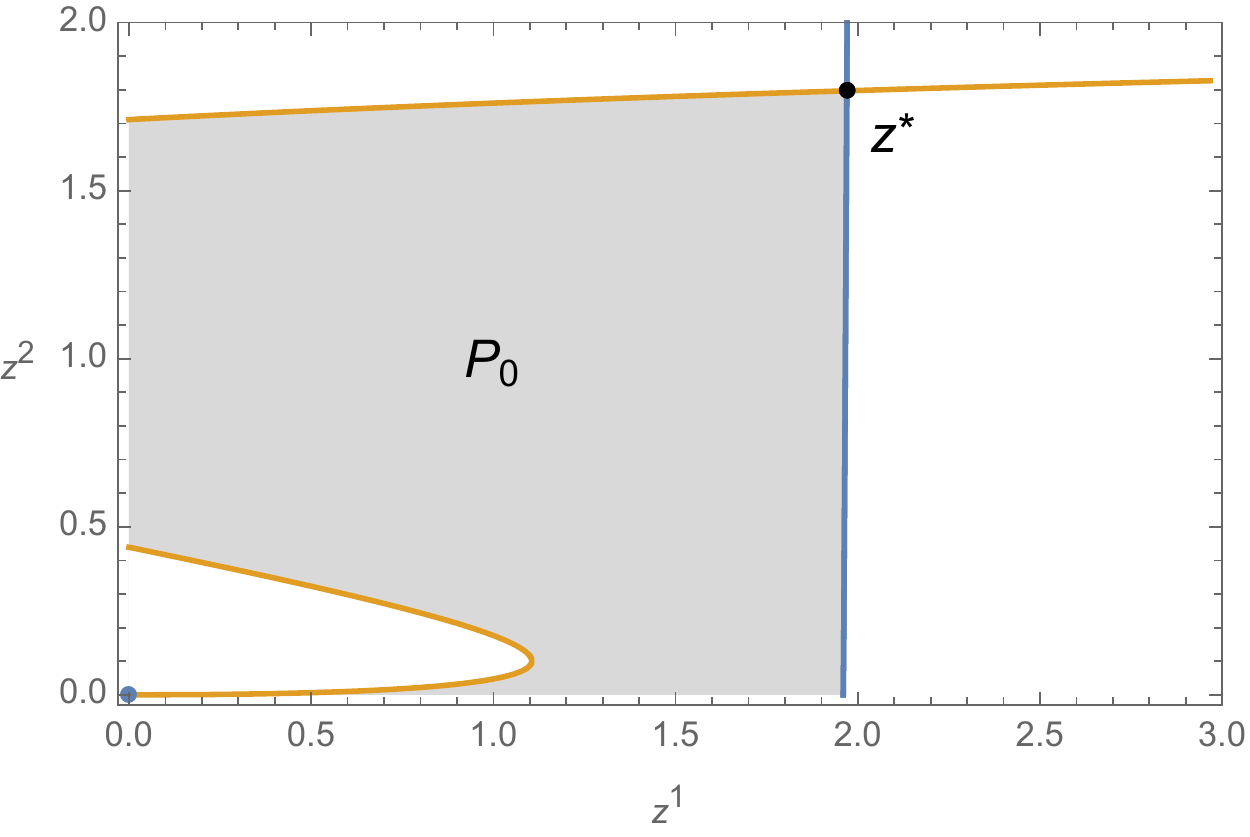}\label{3:fig:ex1:1jointRoot:hump}}
    \hfill\subfigure[]{\includegraphics[width=0.32\textwidth]{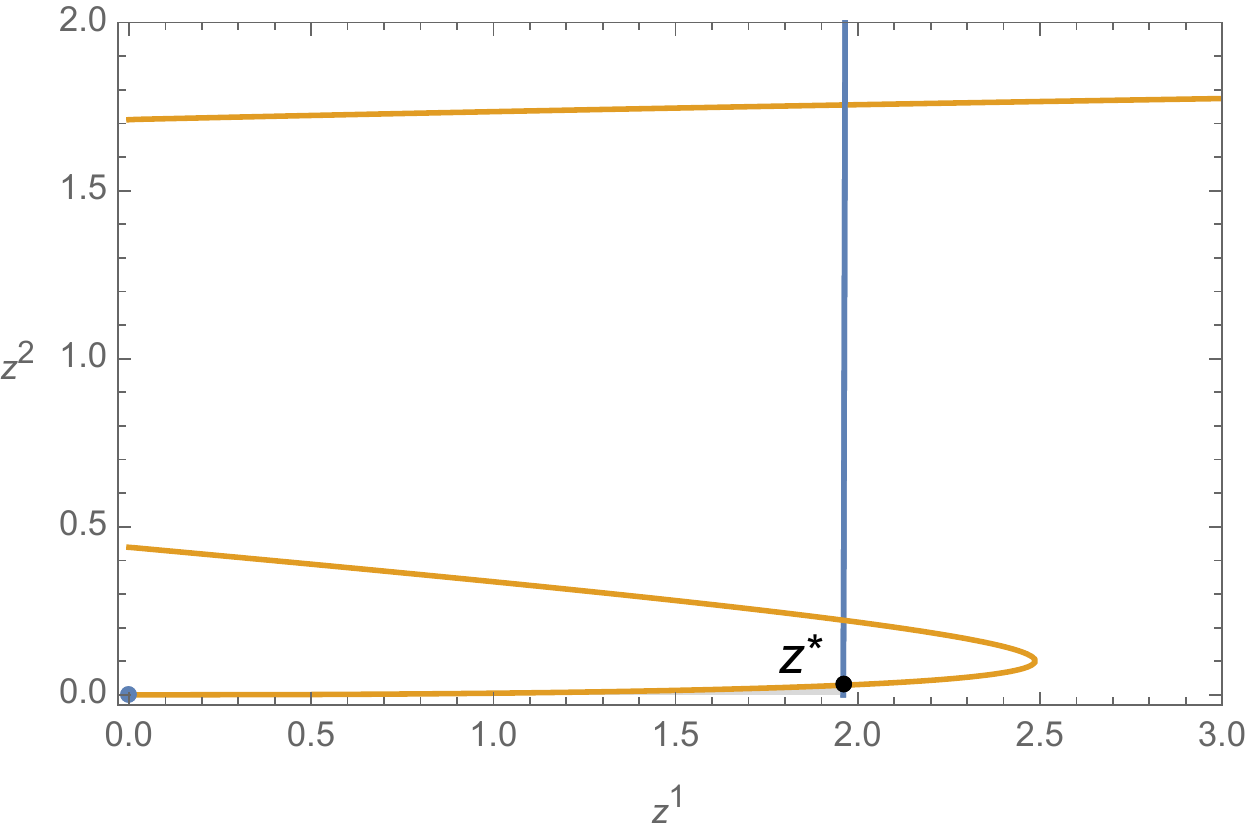}\label{3:fig:ex1:3jointRoots}}\hfill
\caption{Plot of the root sets of the functions $f^1(z^1,z^2)$ (blue) and $f^2(z^1,z^2)$ (orange) for the system with (a) $w_1=2$, $w_2=1$ and $w_3=2$, (b) $w_1=2$, $w_2=2$ and $w_3=3/4$ respectively (c) $w_1=2$, $w_2=2$ and $w_3=1/2$.
}\label{3:fig:Ex1}
\end{figure}

\pagebreak
Since all the main results of this chapter and the derivations in Example \ref{3:ex:non-res:subsystem} are asymptotical for $n\to\infty$, we demonstrate the applicability for finite networks numerically: For each of the scenarios (a)-(c) in Figure \ref{3:fig:Ex1} we performed $10^4$ simulations on networks of varying size \mbox{$n\in\{100k\,:\,k\in[100]\}$} with $1\%$ initially defaulted banks. The outcomes are plotted in Figure \ref{3:fig:Convergence:Ex1} together with the theoretical asymptotic final default fraction (taking into account the initial default fraction of $1\%$). For case (a), except for 6 simulations at $n=100$ all results lie considerably close to the theoretical final fraction of about $87.98\%$ and their deviation becomes smaller the larger $n$ grows. For case (b) and $n<10^3$, some of the simulations ended in final default fractions around $55\%$. (Graphically these come from deviations of the hump (root set of $f^2$, orange) in Figure \ref{3:fig:ex1:1jointRoot:hump} such that it intersects with the root of $f^1$ (blue)). For all other simulations and especially for $n\geq10^3$, the simulation results clearly converge to the theoretical value of about $94.25\%$. For case (c) finally, some of the simulation outcomes for $n\leq500$ were close to $0$ and few around $92.63\%$ (the value if one plugs in the largest of the three joint roots into $g$). The majority of the simulations (in particular for $n\geq4000$), however, resulted in final default fractions close to the theoretical value of $50.02\%$ and again deviations decrease as $n$ increases. Altogether we conclude that already for finite \mbox{networks of a few thousand vertices our asymptotic results are applicable with good accuracy.}

\begin{figure}[t]
\centering
    \hfill\subfigure[]{\includegraphics[width=0.47\textwidth]{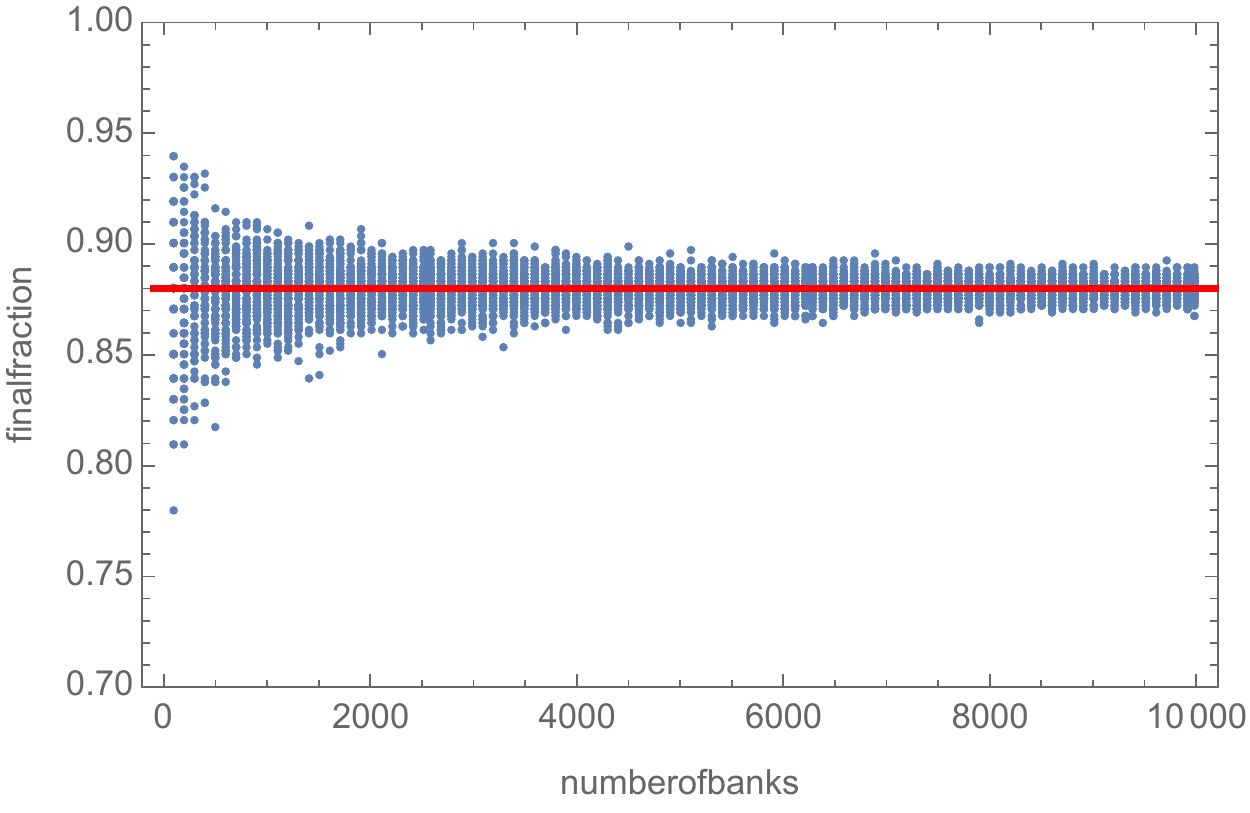}\label{3:fig:ex1:convergence:1jointRoot}}
    \hfill\subfigure[]{\includegraphics[width=0.47\textwidth]{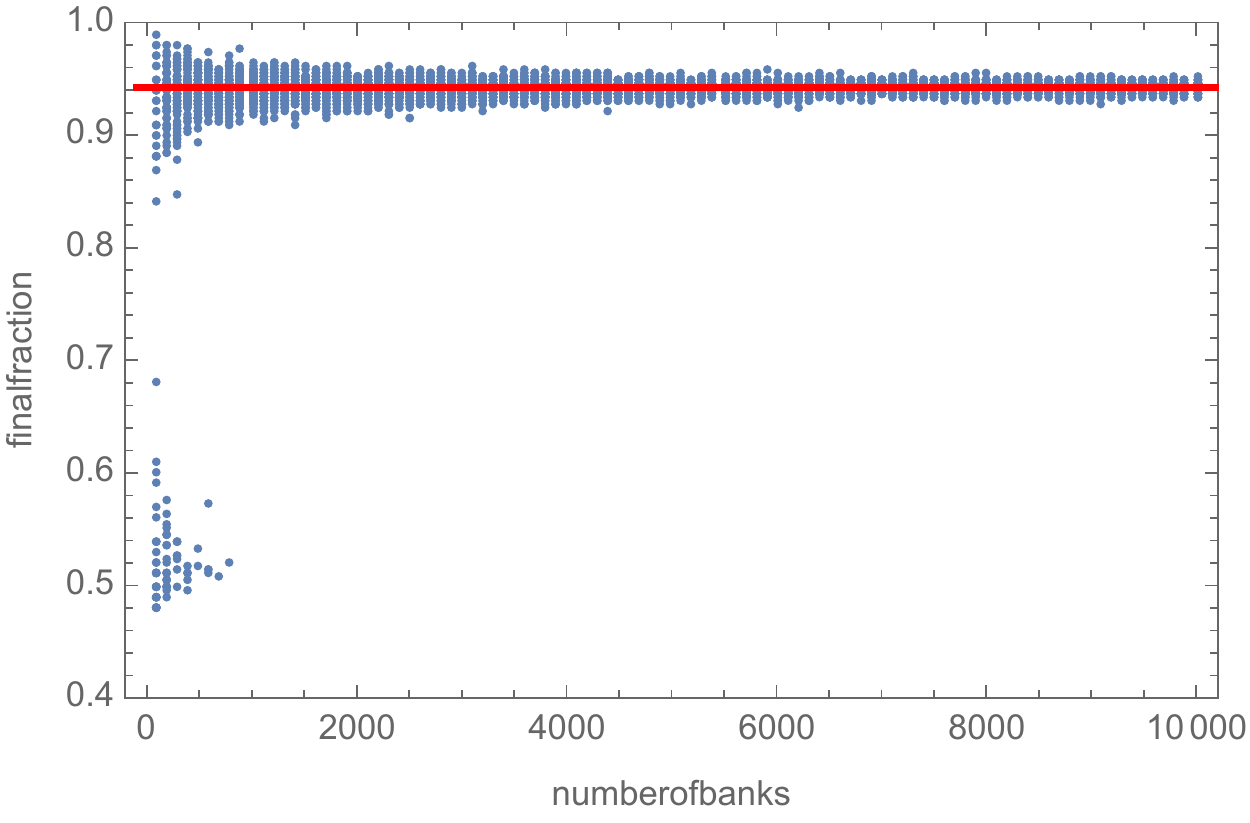}\label{3:fig:ex1:convergence:1jointRoot:hump}}
    \hfill\\
    \hfill\subfigure[]{\includegraphics[width=0.47\textwidth]{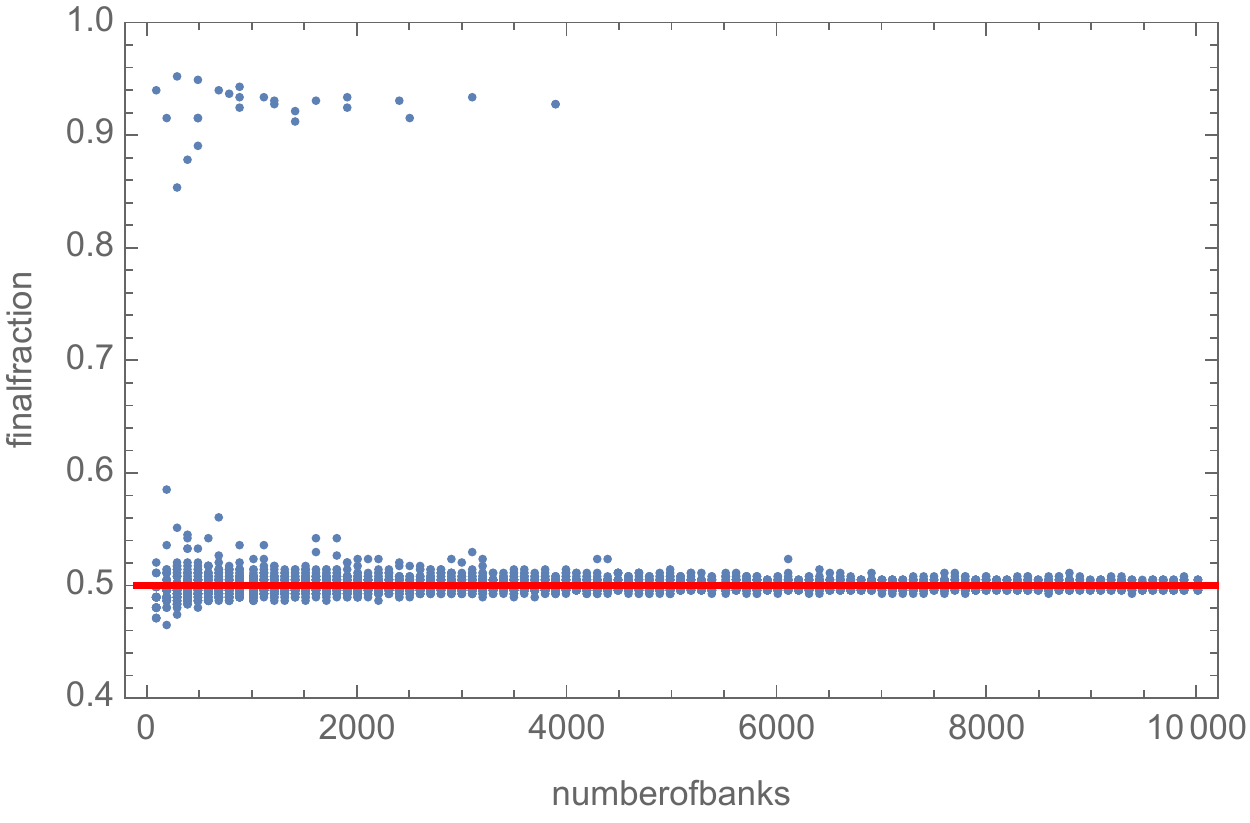}\label{3:fig:ex1:convergence:3jointRoots}}\hfill
\caption{Plot of the simulation results on networks of varying size (blue) and the theoretical asymptotic final default fraction (red) for the system with (a) $w_1=2$, $w_2=1$ and $w_3=2$, \mbox{(b) $w_1=2$,} $w_2=2$ and $w_3=3/4$ respectively (c) $w_1=2$, $w_2=2$ and $w_3=1/2$.
}\label{3:fig:Convergence:Ex1}
\end{figure}

In particular, what we learn from Example \ref{3:ex:non-res:subsystem} is that in order to ensure resilience of a particular subsystem, one needs to completely prohibit links to other non-resilient subsystems. It is, however, also possible that two subsystems which are resilient on their own form a non-resilient global system once connected to each other. It is therefore an interesting regulatory question how to ensure also resilience of a global system composed of various resilient subsystems. In general for our model the answer to this question is provided by Theorem \ref{3:thm:resilience}. However, in the following example we state a more intuitive criterion.

\begin{example}\label{3:ex:res:global:system}
Again, for simplicity assume that $R=1$. Consider a financial network consisting of $T$ subnetworks (types) which shall satisfy the following $1$-dimensional resilience conditions: For each $\epsilon>0$ there exists $z_\epsilon>0$ such that for all $z\in(0,z_\epsilon)$ and $\alpha\in[T]$ it holds
\begin{equation}\label{3:eqn:capital:requirements:subsystems}
\epsilon>\E\left[W^{+,\alpha}W^{-,\alpha}\P\left(\mathrm{Poi}\left(W^{-,\alpha}z\right)=C-1\right)\1\{A=\alpha\}\right].
\end{equation}
Note that this condition implies $\E[W^{+,\alpha}\P(\mathrm{Poi}(W^{-,\alpha}z)\geq C)\vert A=\alpha]<0$ for all $z$ small enough and hence indeed it implies resilience of the subsystem by Theorem \ref{3:thm:resilience}. In Section \ref{2:sec:resilience} explicit capital requirements (i.\,e.~a formula for $C\vert_{A=\alpha}$ in dependence of $W^{-,\alpha}\vert_{A=\alpha}$) were derived for the case of Pareto distributed weights (which are typically observed in real networks) which ensure 
\eqref{3:eqn:capital:requirements:subsystems}. 

\pagebreak
Now further assume that there exists a constant $K<\infty$ such that
\begin{equation}\label{3:eqn:bound:on:international:links}
W^{\pm,\beta}\vert_{A=\alpha} \leq KW^{\pm,\alpha}\vert_{A=\alpha}\quad\text{almost surely},
\end{equation}
for all $\alpha\neq\beta\in[T]$, i.\,e.~the tendency of institutions to develop links with institutions outside their subnetwork is bounded by a constant multiple of their tendency to develop links with institutions within their subnetwork. In particular, this is the case if the external weights are bounded from above and the internal weights are bounded from below.

Replacing $W^{\pm,\beta}\vert_{A=\alpha}$ by $KW^{\pm,\alpha}\vert_{A=\alpha}$ only makes the system less resilient (if the weights increase, the number of links increases and hence the total exposure of each institution). Hence set $\tilde{W}^{\pm,\beta}\vert_{A=\alpha} = KW^{\pm,\alpha}\vert_{A=\alpha}$ for $\alpha\neq\beta$ and $\tilde{W}^{\pm,\alpha}\vert_{A=\alpha}=W^{\pm,\alpha}\vert_{A=\alpha}$. Now define $\bm{v}\in\R_+^{[T]\times[T]}$ by $v^{\alpha,\beta}=K^{\1\{\alpha\neq\beta\}}$, $\alpha,\beta\in[T]$. Then we derive that
\begin{align*}
&\E\Bigg[\tilde{W}^{+,\alpha}\Bigg(\sum_{\beta'\in[T]}v^{\alpha,\beta'}\tilde{W}^{-,\beta'}\Bigg)\P\Bigg(\mathrm{Poi}\Bigg(\sum_{\gamma\in[T]}\tilde{W}^{-,\gamma}z^{\alpha,\gamma}\Bigg)=C-1\Bigg)\1\{A=\alpha\}\Bigg]\\
&\hspace{0.5cm} = \E\Bigg[W^{+,\alpha}W^{-,\alpha}(1+K^2(T-1))\P\Bigg(\mathrm{Poi}\Bigg(W^{-,\alpha}\Bigg(z^{\alpha,\alpha}+K\sum_{\gamma\neq\alpha}z^{\alpha,\gamma}\Bigg)\Bigg)=C-1\Bigg)\1\{A=\alpha\}\Bigg]\\
&\hspace{0.5cm} < (1+K^2(T-1))\epsilon = v^{\alpha,\alpha}(1+K^2(T-1))\epsilon,
\end{align*}
for $z^{\alpha,\alpha}+K\sum_{\gamma\neq\alpha}z^{\alpha,\gamma}<z_\epsilon$, and
\begin{align*}
&\E\Bigg[\tilde{W}^{+,\alpha}\Bigg(\sum_{\beta'\in[T]}v^{\beta,\beta'}\tilde{W}^{-,\beta'}\Bigg)\P\Bigg(\mathrm{Poi}\Bigg(\sum_{\gamma\in[T]}\tilde{W}^{-,\gamma}z^{\beta,\gamma}\Bigg)=C-1\Bigg)\1\{A=\beta\}\Bigg]\\
&\hspace{0.24cm} = \E\Bigg[KW^{+,\beta}W^{-,\beta}(1+K^2(T-1))\P\Bigg(\mathrm{Poi}\Bigg(W^{-,\beta}\Bigg(z^{\beta,\beta}+K\sum_{\gamma\neq\beta}z^{\beta,\gamma}\Bigg)\Bigg)=C-1\Bigg)\1\{A=\beta\}\Bigg]\\
&\hspace{0.24cm} < K(1+K^2(T-1))\epsilon = v^{\alpha,\beta}(1+K^2(T-1))\epsilon,
\end{align*}
for $\alpha\neq\beta$ and $z^{\beta,\beta}+K\sum_{\gamma\neq\beta}z^{\beta,\gamma}<z_\epsilon$. If we now choose $\epsilon<(1+K^2(T-1))^{-1}$, then 
\[ \tilde{f}^{\alpha,\beta}(\delta\bm{v}):=\E\Bigg[\tilde{W}^{+,\alpha}\P\Bigg(\mathrm{Poi}\Bigg(\delta\sum_{\gamma\in[T]}\tilde{W}^{-,\gamma}v^{\beta,\gamma}\Bigg)\geq C\Bigg)\1\{A=\beta\}\Bigg] - \delta v^{\alpha,\beta} < 0 \]
for all $\delta>0$ small enough. It thus holds $\bm{z}^*\leq\lim_{\delta\to0+}\delta\bm{v}=\bm{0}$ and therefore $P_0=\{\bm{0}\}$. We can then apply Theorem \ref{3:thm:resilience} and obtain that the combined system is still resilient.

From a regulatory perspective it is hence enough to impose capital requirements 
described by \eqref{3:eqn:capital:requirements:subsystems} and to restrict links between different subsystems in the sense of \eqref{3:eqn:bound:on:international:links}.
\end{example}
In our first two examples we concentrated on the (non-)resilience of multi-type networks. For simplicity, we assumed that all edges carry the same weight ($R=1$). 
Another interesting feature of our model, however, is that it allows for edge weights that depend on the types of both the creditor and the debtor bank. The following example shows that this can indeed make a huge difference as compared to previous models in which exposures (edge weights) could only depend on the size/degree/type of the creditor bank. It considers two very similar financial systems whose only difference is that in one exposures depend on both the creditor and debtor type and in the other they depend on the type of the creditor bank only. As a consequence the first system will turn out to be non-resilient whereas the second one is resilient.

\begin{example}\label{3:ex:neighbor:dependent}
Consider a network of size $n\in\N$ in which (asymptotically) $p=1/3$ of all banks have type $1$ and the remaining $1-p=2/3$ banks have type $2$. That is, $T=2$. Further assume that for each pair of vertices $(i,j)\in[n]^2$ an edge from $i$ to $j$ shall be present with probability $4/n$. Edges between two banks of type $1$ shall carry weight $2$ and all other edges weight $1$. That is, if $\alpha_i=1$, then $w_i^{\pm,2,1}=w_i^{\pm,1,2}=2$ and $w_i^{\pm,1,1}=w_i^{\pm,2,2}=0$. If $\alpha_i=2$, then $w_i^{\pm,1,1}=w_i^{\pm,1,2}=2$ and $w_i^{\pm,2,1}=w_i^{\pm,2,2}=0$. Finally, all banks shall have capital $2$.

Then similarly as for Example \ref{3:ex:non-res:subsystem} the originally eight-dimensional system reduces to 
\begin{align*}
f^1(z^1,z^2) &= 2p\P\left(\mathrm{Poi}(2z^2)+2\mathrm{Poi}(2z^1)\geq 2\right) - z^1,\\
f^2(z^1,z^2) &= 2(1-p)\P\left(\mathrm{Poi}(2(z^1+z^2))\geq 2\right) - z^2.
\end{align*}
See Figure \ref{3:fig:neighbor:dependent} for an illustration of the root sets of $f^1$ and $f^2$. This figure already shows 
that $\bm{z}^*\neq\bm{0}$ and hence non-resilience by Theorem \ref{3:thm:non-resilience}. 
\begin{figure}[t]
    \hfill\subfigure[]{\includegraphics[width=0.45\textwidth]{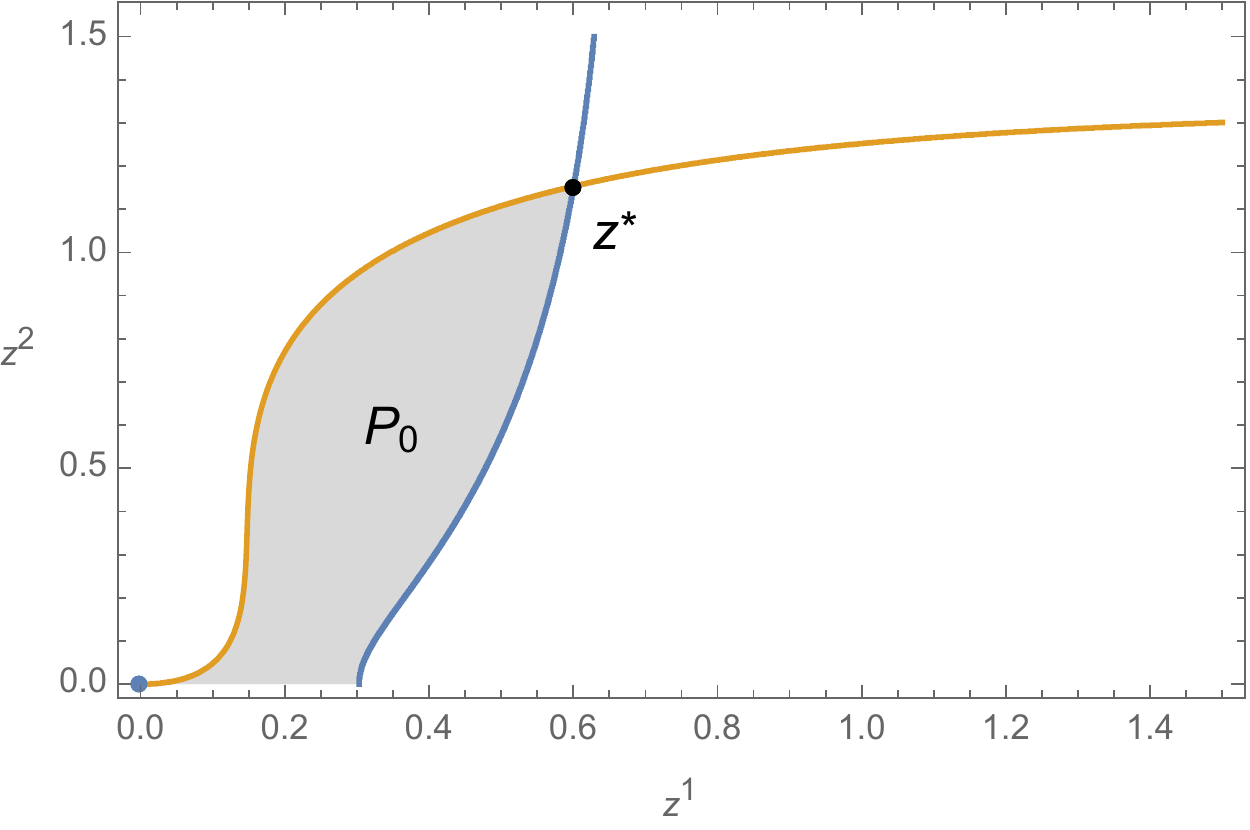}\label{3:fig:neighbor:dependent}}
    \hfill\subfigure[]{\includegraphics[width=0.45\textwidth]{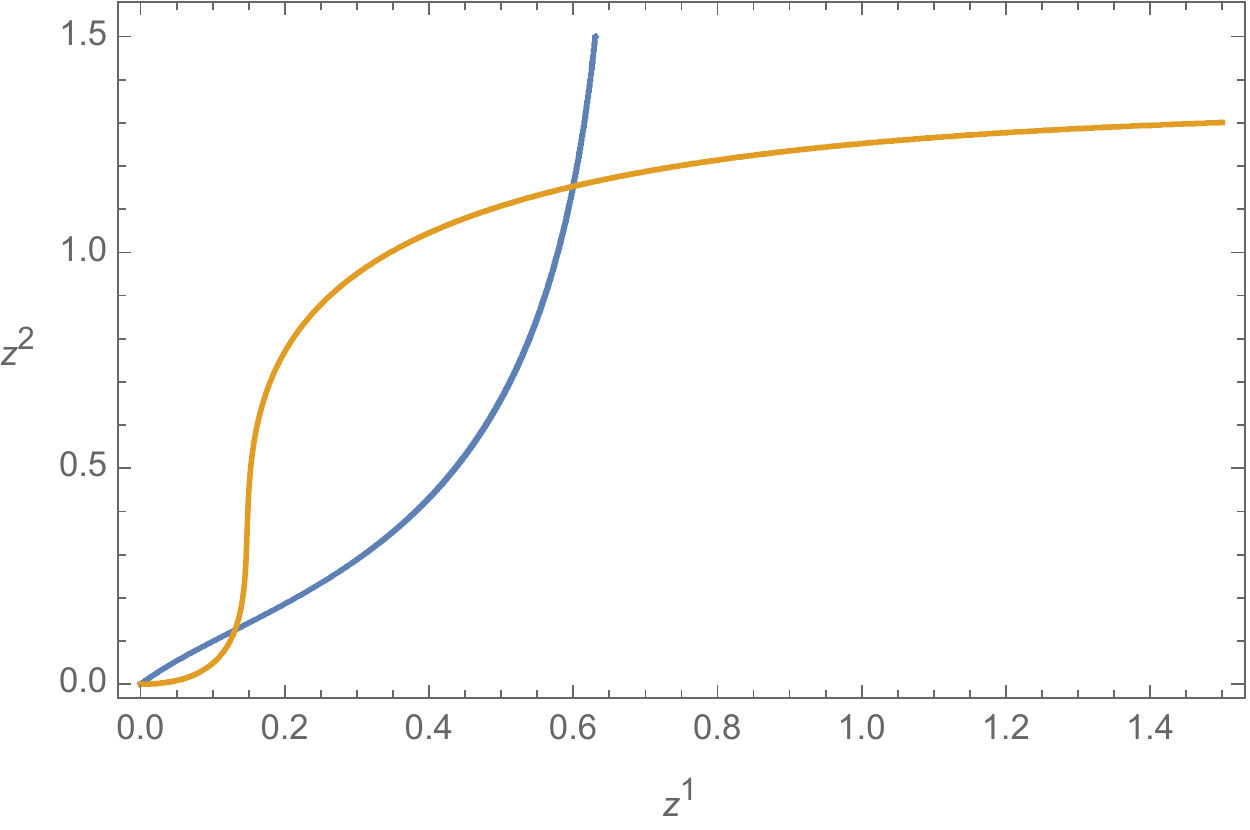}\label{3:fig:neighbor:independent}}\hfill
\caption{Plot of the root sets of the functions $f^1(z^1,z^2)$ (blue) and $f^2(z^1,z^2)$ (orange) for the system with (a) neighbor-dependent exposures respectively (b) neighbor-independent exposures.
}
\end{figure}
Also for $z^1,z^2\to0$, we can compute
\[ \frac{\partial f^1}{\partial z^1}(z^1,z^2) = 4p \P\left(\mathrm{Poi}(2z^2)+2\mathrm{Poi}(2z^1)\in\{0,1\}\right) - 1 \to 4p - 1 = \frac{1}{3} > 0, \]
which rigorously proves that the type-$1$ subnetwork and then also the whole system is non-resilient (cf.~Example \ref{3:ex:non-res:subsystem}). Numerically one derives that $\bm{z}^*\approx(0.601,1.153)$ and $g(\bm{z}^*)\approx 0.877$. In order to test this prediction, we performed $10^4$ simulations of financial networks of size $n=10^4$ with initial default probability $10^{-3}$. In only $5.32\%$ of the simulations, we observed a resilient nature in the sense that the simulated final default fraction was lower than $3\%$. All of the remaining simulations ended with a final default fraction within $[85.65\%,89.73\%]$ and are hence of a non-resilient nature. Averaging over the latter ones yields a mean final default fraction of $87.71\%$.


Now consider the following modified network: Instead of assigning weight $2$ to all links between two type-$1$ banks and weight $1$ to all other links, this time assign weight $2$ with probability $p$ to any edge going to a type-$1$ bank (all other edges are assigned weight $1$). That is, we keep the skeleton of the network but we redistribute the edge-weights in such a way that they do only depend on the creditor bank and not on the debtor bank. The total number of weight-$2$ edges stays the same (note that in the first network the type-$1$ banks accounted to a fraction of $p$ of all the debtor banks of type-$1$ banks). This can be achieved by assigning the following new vertex-weights: $w_i^{+,1,1}=w_i^{+,2,1}=2$ for all $i\in[n]$. Further, if $\alpha_i=1$, then $w_i^{-,1,1}=w_i^{-,1,2}=2(1-p)$ and $w_i^{-,2,1}=w_i^{-,2,2}=2p$. All other vertex-weights shall stay the same. The new system then reduces to the following two functions, whose root sets are shown in Figure \ref{3:fig:neighbor:independent}:
\begin{align*}
f^1(z^1,z^2) &= 2p\P\left(\mathrm{Poi}(2(1-p)(z^1+z^2))+2\mathrm{Poi}(2p(z^1+z^2))\geq 2\right) - z^1,\\
f^2(z^1,z^2) &= 2(1-p)\P\left(\mathrm{Poi}(2(z^1+z^2))\geq 2\right) - z^2
\end{align*}
Figure \ref{3:fig:neighbor:independent} shows that the root set of $f^2$ is being shifted to the left, now starting off above the root set of $f^1$. One can hence already expect that $P_0=\{\bm{0}\}$ and the new system to be resilient. Also more rigorously, as $z^1,z^2\to0$, we derive that
\begin{align*}
\frac{\partial f^1}{\partial z^1}(z^1,z^2) &= 4p(1-p)\P\left(\mathrm{Poi}(2(1-p)(z^1+z^2))+2\mathrm{Poi}(2p(z^1+z^2))=1\right)\\
&\hspace{3.15cm} + 4p^2\P\left(\mathrm{Poi}(2(1-p)(z^1+z^2))+2\mathrm{Poi}(2p(z^1+z^2))\in\{0,1\}\right) - 1\\
&\to 4p^2 - 1 = -\frac{5}{9},\\
\frac{\partial f^1}{\partial z^2}(z^1,z^2) &= 4p(1-p)\P\left(\mathrm{Poi}(2(1-p)(z^1+z^2))+2\mathrm{Poi}(2p(z^1+z^2))=1\right)\\
&\hspace{3.83cm} + 4p^2\P\left(\mathrm{Poi}(2(1-p)(z^1+z^2))+2\mathrm{Poi}(2p(z^1+z^2))\in\{0,1\}\right)\\
&\to 4p^2 = \frac{4}{9},
\end{align*}
\begin{align*}
\frac{\partial f^2}{\partial z^1}(z^1,z^2) &= 4(1-p)\P\left(\mathrm{Poi}(2(z^1+z^2))=1\right) \to 0,\\
\frac{\partial f^2}{\partial z^2}(z^1,z^2) &= 4(1-p)\P\left(\mathrm{Poi}(2(z^1+z^2))=1\right) - 1 \to -1.
\end{align*}
The directional derivatives $D_{\bm{v}}f^1(\bm{0})$ and $D_{\bm{v}}f^2(\bm{0})$ thus exist for every $\bm{v}\in\R_+^V$. Choose then for example $\bm{v}=(v^1,v^2)=(1,1)$ such that $D_{\bm{v}}f^1(\bm{0})=-1/9$ and $D_{\bm{v}}f^2(\bm{0})=-1$. From Lemma \ref{3:lem:sufficient:criteria:z:star} we thus derive that $\bm{z}^*=\bm{0}$ and hence $P_0=\{\bm{0}\}$. This allows us to apply Theorem \ref{3:thm:resilience} and hence the modified system is indeed resilient. Again this can be validated numerically. On the same skeleton as for the previous simulation but with random edge-weights as described above the simulated final default fractions are now all within the interval $[0.11\%,0.63\%]$ with an average of $0.20\%$. The system is hence indeed of a resilient nature.
\end{example}
Although Example \ref{3:ex:neighbor:dependent} is too simple to model a real financial network, it still shows that counterparty-dependent exposures may have a significant impact on the stability of the system. In general, it is also possible that they increase stability of the system, however.

\section{Proofs}\label{3:sec:proofs}
In this section, we provide the proofs of 
our results in Sections \ref{3:sec:asymptotic:results}, \ref{3:sec:resilience} and \ref{3:sec:capital:block:model}. Theorem \ref{3:thm:general:weights} will be proved in two steps. At this the underlying ideas are similar to \cite{Detering2015a} but at a considerable number of steps novel methods have to be used and we will particularly discuss them in detail.  We use the notation
\[ [\bm{a},\bm{b}]:=\bigcap_{(r,\alpha,\beta)\in V}\{\bm{z}\in\R^V\,:\,a^{r,\alpha,\beta}\leq z^{r,\alpha,\beta}\leq b^{r,\alpha,\beta}\} \]
for the cuboid spanned by the vectors $\bm{a}$ and $\bm{b}$ in $\R^V$ in the following. Further, let $\bm{\zeta}\in\R_{+,0}^V$ be defined by $\zeta^{r,\alpha,\beta}:=\E[W^{+,r,\alpha}\1\{A=\beta\}]$.

\subsection{Proofs of Lemmas \ref{3:lem:existence:hatz} and \ref{3:lem:sufficient:criteria:z:star}}\label{3:ssec:proofs:lemmas}
\begin{proof}[Proof of Lemma \ref{3:lem:existence:hatz}]
Existence of a smallest joint root $\hat{\bm{z}}\in[\bm{0},\bm{\zeta}]$ is ensured by the Knaster-Tarski fixed point theorem. We now construct a joint root $P_0\ni\bar{\bm{z}}\leq\hat{\bm{z}}$ which shows that $\hat{\bm{z}}=\bar{\bm{z}}\in P_0$, in particular:
It holds $f^{r,\alpha,\beta}(\hat{\bm{z}})=0$ for all $(r,\alpha,\beta)\in V$ and then $f^{r,\alpha,\beta}(\bm{z})\leq0$ for all $\hat{\bm{z}}\geq\bm{z}\in\R_{+,0}^V$ such that $z^{r,\alpha,\beta}=\hat{z}^{r,\alpha,\beta}$ 
by monotonicity of $f^{r,\alpha,\beta}$ from Lemma \ref{3:lem:properties:f}. 
Consider then the following sequence $(\bm{z}_n)_{n\in\N}\subset\R_{+,0}^V$:
\begin{itemize}
\item $\bm{z}_0=\bm{0}$
\item $\bm{z}_1=(z_1^{1,1,1},0,\ldots,0)$, where $z_1^{1,1,1}\geq0$ is the smallest possible value such that \mbox{$f^{1,1,1}(\bm{z}_1)=0$.} It is possible to find such $z_1^{1,1,1}$ by the intermediate value theorem since $f^{1,1,1}$ is continuous, $f^{1,1,1}(\bm{0})\geq 0$ and $f^{1,1,1}(\hat{z}^{1,1,1},0,\ldots,0)\leq0$. By Lemma \ref{3:lem:properties:f}, it then holds $f^{r,\alpha,\beta}(\bm{z}_1)\geq f^{r,\alpha,\beta}(\bm{0})\geq0$ for all $(1,1,1)\neq(r,\alpha,\beta)\in V$. In particular, $\bm{z}_1\in P_0$.
\item $\bm{z}_2=\bm{z}_1+(0,z_2^{1,1,2},0,\ldots,0)$, where $z_2^{1,1,2}\geq0$ is the smallest value such that $f^{1,1,2}(\bm{z}_2)=0$. Again it is possible to find such $z_2^{1,1,2}$ by the intermediate value theorem since $f^{1,1,2}$ is continuous, $f^{1,1,2}(\bm{z}_1)\geq0$ and $f^{1,1,2}(\bm{z}_1+(0,\hat{z}^{1,1,2},0,\ldots,0))\leq0$. Since $\bm{z}_1\in P_0$, by Lemma \ref{3:lem:properties:f} it then holds $f^{r,\alpha,\beta}(\bm{z}_2)\geq f^{r,\alpha,\beta}(\bm{z}_1)\geq0$ for all $(1,1,2)\neq(r,\alpha,\beta)\in V$. In particular, $\bm{z}_2\in P_0$.
\item $\bm{z}_i$, $i\in\{3,\ldots,RT^2\}$, are found analogously, changing only the corresponding coordinate.
\item $\bm{z}_{RT^2+1}=\bm{z}_{RT^2}+(z_{RT^2+1}^{1,1,1}-z_{RT^2}^{1,1,1},0,\ldots,0)$, where $z_{RT^2+1}^{1,1,1}\geq z_{RT^2}^{1,1,1}$ is the smallest value such that $f^{1,1,1}(\bm{z}_{RT^2+1})=0$, which is again possible by the intermediate value theorem. In particular, it still holds $z^{1,1,1}_{RT^2+1}\leq\hat{z}^{1,1,1}$. As before also $\bm{z}_{RT^2+1}\in P_0$.
\item Continue for $\bm{z}_i$, $i\geq RT^2+2$.
\end{itemize}
The sequence $(\bm{z}_n)_{n\in\N}$ constructed this way has the following properties: It is non-decreasing in each coordinate and $(\bm{z}_n)_{n\in\N}\subset P_0$. Further, it is bounded inside 
$[\bm{0},\hat{\bm{z}}]$. Hence by monotone convergence, each coordinate of $\bm{z}_n$ converges and so $\bar{\bm{z}}=\lim_{n\to\infty}\bm{z}_n$ exists. Now suppose there is $(r,\alpha,\beta)\in V$ such that $f^{r,\alpha,\beta}(\bar{\bm{z}})>0$. By continuity of $f^{r,\alpha,\beta}$ then also $f^{r,\alpha,\beta}(\bm{z}_n)>\epsilon$ for some $\epsilon>0$ and $n$ large enough. This, however, is in contradiction with the construction of the sequence $(\bm{z}_n)_{n\in\N}$ since $f^{r,\alpha,\beta}(\bm{z}_n)=0$ in every $RT^2$-th step. Hence $f^{r,\alpha,\beta}(\bar{\bm{z}})\leq 0$ for all $(r,\alpha,\beta)\in V$. Also $\bar{\bm{z}}\in P_0$, however, since this is a closed set
. Hence $f^{r,\alpha,\beta}(\bar{\bm{z}})\geq0$ for all $(r,\alpha,\beta)\in V$ and altogether this shows that $\bar{\bm{z}}$ is a joint root of all functions $f^{r,\alpha,\beta}$, $(r,\alpha,\beta)\in V$.


Now turn to the proof that $\bm{z}^*\in P_0$ and it is a joint root of all functions $f^{r,\alpha,\beta}$, $(r,\alpha,\beta)\in V$: First define the following sets for each $\epsilon>0$:
\[ P(\epsilon) := \bigcap_{(r,\alpha,\beta)\in V}\{\bm{z}\in\R_{+,0}^V\,:\,f^{r,\alpha,\beta}(\bm{z})\geq -\epsilon\}
\]
Further denote by $P_0(\epsilon)$ 
the connected component 
of $\bm{0}$ in $P(\epsilon)$
. By the same procedure as for $\hat{\bm{z}}$ above, we now derive that there exists a smallest (componentwise) point $\hat{\bm{z}}(\epsilon)\in P_0(\epsilon)$ such that $f^{r,\alpha,\beta}(\hat{\bm{z}}(\epsilon))=-\epsilon$ for all $(r,\alpha,\beta)\in V$. Clearly, $\hat{\bm{z}}(\epsilon)$ is non-decreasing in $\epsilon$ (componentwise) and hence 
$\tilde{\bm{z}}:=\lim_{\epsilon\to0+}\hat{\bm{z}}(\epsilon)$ exists (we will show 
that $\tilde{\bm{z}}=\bm{z}^*$ in fact).

Now by monotonicity of $P_0(\epsilon)$, we derive that $\hat{\bm{z}}(\delta)\in P_0(\delta)\subseteq P_0(\epsilon)$ for all $\delta\leq \epsilon$. Since $P_0(\epsilon)$ is a closed set, it must thus hold that also $\tilde{\bm{z}}=\lim_{\delta\to0+}\hat{\bm{z}}(\delta)\in P_0(\epsilon)$ for all $\epsilon>0$ and in particular, $\tilde{\bm{z}}\in\bigcap_{\epsilon>0}P_0(\epsilon)$. Further, by continuity of $f^{r,\alpha,\beta}$, $(r,\alpha,\beta)\in V$, we derive that $\bigcap_{\epsilon>0}P_0(\epsilon)\subseteq\bigcap_{\epsilon>0}P(\epsilon)\subseteq P$. Moreover, $\bigcap_{\epsilon>0}P_0(\epsilon)$ is the intersection of a chain of connected, compact sets in the Hausdorff space $\R^V$ and it is hence a connected, compact set itself. Since it further contains $\bm{0}$, we can then conclude that $\bigcap_{\epsilon>0}P_0(\epsilon)\subseteq P_0$ and thus $\tilde{\bm{z}}\in P_0$.

We now want to show that $\bm{z}\leq\tilde{\bm{z}}$ componentwise for arbitrary $\bm{z}\in P_0$. 
This clearly proves $\tilde{\bm{z}}=\bm{z}^*$. 
It thus suffices to show $P_0\subset[\bm{0},\hat{\bm{z}}(\epsilon)]$. Then $\bm{z}\leq\hat{\bm{z}}(\epsilon)$ and $\bm{z}\leq\lim_{\epsilon\to0+}\hat{\bm{z}}(\epsilon)=\tilde{\bm{z}}$. Hence assume that $P_0\not\subset[\bm{0},\hat{\bm{z}}(\epsilon)]$. By connectedness of $P_0$ we find $\bar{\bm{z}}\in P_0$ with $\bar{z}^{r,\alpha,\beta}\leq\hat{z}^{r,\alpha,\beta}(\epsilon)$ for all $(r,\alpha,\beta)\in V$ and equality for at least one coordinate. W.\,l.\,o.\,g.~let this coordinate be $(1,1,1)$. By monotonicity of $f^{1,1,1}$ with respect to $z^{r,\alpha,\beta}$ for every $(r,\alpha,\beta)\neq(1,1,1)$, we thus derive that
\[ f^{1,1,1}(\bar{\bm{z}})\leq f^{1,1,1}(\hat{\bm{z}}(\epsilon)) = -\epsilon. \]
However, we also assumed that $\bar{\bm{z}}\in P_0$ and hence $f^{1,1,1}(\bar{\bm{z}})\geq0$, a contradiction.

Finally, we obtain that $f^{r,\alpha,\beta}(\bm{z}^*)=f^{r,\alpha,\beta}(\tilde{\bm{z}}) = \lim_{\epsilon\to0+}f^{r,\alpha,\beta}(\hat{\bm{z}}(\epsilon)) = \lim_{\epsilon\to0+}(-\epsilon)=0$, by continuity of $f^{r,\alpha,\beta}$
. Hence $\bm{z}^*$ is in fact a joint root of all the functions $f^{r,\alpha,\beta}$, $(r,\alpha,\beta)\in V$.
\end{proof}
\begin{proof}[Proof of Lemma \ref{3:lem:sufficient:criteria:z:star}]
Note that it is sufficient to construct a sequence $(\bm{z}_n)_{n\in\N}\subset\R_+^V$ such that $\lim_{n\to\infty}\bm{z}_n=\bar{\bm{z}}$ and $f^{r,\alpha,\beta}(\bm{z}_n)<0$ for all $(r,\alpha,\beta)\in V$, $n\in\N$. By monotonicity of $f^{r,\alpha\beta}$ from Lemma \ref{3:lem:properties:f} it then follows that $\bm{z}^*\leq\bm{z}_n$ 
and hence $\bm{z}^*\leq\bar{\bm{z}}$. If condition 1.~is satisfied, we get $\lim_{n\to\infty}nf^{r,\alpha,\beta}\left(\bar{\bm{z}}+n^{-1}\bm{v}\right) = D_{\bm{v}}f^{r,\alpha,\beta}(\bar{\bm{z}}) < 0$ and we can hence choose $\bm{z}_n:=\bar{\bm{z}}+n^{-1}\bm{v}$.

If condition 2.~is satisfied, note that by Fubini's theorem for $n>\Delta^{-1}$ we derive
\begin{align*}
f^{r,\alpha,\beta}\bigg(\bar{\bm{z}}+n^{-1}\bm{v}\bigg) &= \int_0^{n^{-1}} -v^{r,\alpha,\beta} + \sum_{r'\in[R]}\E\Bigg[W^{+,r,\alpha}\Bigg(\sum_{\beta'\in[T]}v^{r',\beta,\beta'}W^{-,r',\beta'}\Bigg) \1\{A=\beta\}\\
&\hspace*{0.345cm}\times\P\Bigg(\sum_{s\in[R]}s\mathrm{Poi}\Bigg(\sum_{\gamma\in[T]}W^{-,s,\gamma}\Big(\bar{z}^{s,\beta,\gamma}+\delta v^{s,\beta,\gamma}\Big)\Bigg)\hspace*{-0.75ex}\in\hspace*{-0.5ex}\{C-r',\ldots,C-1\}\Bigg)\Bigg]\dd\delta\\
&\leq -n^{-1}(1-\kappa)v^{r,\alpha,\beta}
\end{align*}
and 
$\limsup_{n\to\infty}nf^{r,\alpha,\beta}\left(\bar{\bm{z}}+n^{-1}\bm{v}\right) \leq (1-\kappa)v^{r,\alpha,\beta} < 0$. 
Thus choose $\bm{z}_n:=\bar{\bm{z}}+n^{-1}\bm{v}$ again.
\end{proof}

\subsection{Proof of the Main Result for Finitary Weights}\label{3:ssec:proof:main:finitary}
In this section, we consider the special case that the vertex weights $w_i^{r,\alpha,\beta}$ and capitals $c_i$ can only take values in a finite set. To formalize this, consider the following definition:

\begin{definition}[Finitary Regular Vertex Sequence]
A regular vertex sequence denoted by $(\bm{w}^-,\bm{w}^+,\bm{s},\bm{c},\bm{\alpha})$ is called finitary if there exist $J\in\N$, $\tilde{\bm{w}}_j^-\in\R_{+,0}^{[R]\times[T]}$, $\tilde{\bm{w}}_j^+\in\R_{+,0}^{[R]\times[T]}$ and $\tilde{s}_j\in\R_{+,0}$, $j\in[J]$, as well as $c_\text{max}\in\N_0$ such that for all $n\in\N$ and $i\in[n]$, there exists $j=j(n,i)\in[J]$ such that $\bm{w}_i^\pm(n)=\tilde{\bm{w}}_j^\pm$, $s_i=\tilde{s}_j$ and $c_i(n)\in[c_\text{max}]\cup\{0,\infty\}$.
\end{definition}
That is, in a finitary system there is a partition of the set of all institutions into $TJ(c_\text{max}+2)$ 
sets. In particular, in this case all weights $w_i^{\pm,r,\alpha}$ are bounded from above by some constant $\overline{w}\in\R_+$ and hence by dominated convergence we can compute the partial derivatives of $f^{r,\alpha,\beta}$:
\begin{align*}
\frac{\partial f^{r,\alpha,\beta}}{\partial z^{r',\alpha',\beta'}}(\hat{\bm{z}}) &= -\delta_{r,r'}\delta_{\alpha,\alpha'}\delta_{\beta,\beta'} + \delta_{\beta,\alpha'}\E\Bigg[ W^{+,r,\alpha}W^{-,r',\beta'} \1\{A=\beta\}\\
&\hspace{3cm}\times\P\Bigg(\sum_{s\in[R]}s\mathrm{Poi}\Bigg(\sum_{\gamma\in[T]}W^{-,s,\gamma}z^{s,\beta,\gamma}\Bigg)\in\{C-r',\ldots,C-1\}\Bigg)\Bigg],
\end{align*}
where $\delta_{a,b}:=\1\{a=b\}$. Hence for any vector $\bm{v}\in\R^V$, the directional derivative of $f^{r,\alpha,\beta}$ in direction $\bm{v}$ is given by the following continuous expression:
\begin{align*}
D_{\bm{v}} f^{r,\alpha,\beta}(\hat{\bm{z}})&= 
- v^{r,\alpha,\beta} + \sum_{r'\in[R]} \E\Bigg[W^{+,r,\alpha}\Bigg(\sum_{\beta'\in[T]}v^{r',\beta,\beta'}W^{-,r',\beta'}\Bigg)\1\{A=\beta\}\\
&\hspace{3.6cm}\times\P\Bigg(\sum_{s\in[R]}s\mathrm{Poi}\Bigg(\sum_{\gamma\in[T]}W^{-,s,\gamma}z^{s,\beta,\gamma}\Bigg)\in\{C-r',\ldots,C-1\}\Bigg)\Bigg]
\end{align*}
We can then prove the following asymptotic results for the final systemic damage in the network. 
\begin{proposition}\label{3:prop:finitary:weights}
Consider a financial system described by a finitary regular vertex sequence and let $\hat{\bm{z}}$ be the smallest joint root of the functions $f^{r,\alpha,\beta}$, $(r,\alpha,\beta)\in V$. Then it holds that $n^{-1}\mathcal{S}_n \geq g(\hat{\bm{z}}) + o_p(1)$. If additionally there exists 
$\bm{v}\in\R_+^V$ such that $D_{\bm{v}}f^{r,\alpha,\beta}(\hat{\bm{z}})<0$ for all $(r,\alpha,\beta)\in V$, then $n^{-1}\mathcal{S}_n = g(\hat{\bm{z}}) + o_p(1)$.
\end{proposition}
See Figure \ref{3:fig:one:joint:root} for an example where such $\bm{v}\in\R_+^V$ exists respectively Figure \ref{3:fig:two:joint:roots} for an example where it does not. Theorem \ref{3:thm:finitary:weights} below will analyze systems of the latter type as well.

\begin{proof}
We begin by proving the lower bound: As in \cite{Detering2015a} and Chapter \ref{chap:systemic:risk} we switch to a sequential default contagion process. The idea is to collect defaulted institutions and instead of exposing them all at once (as in \eqref{3:eqn:default:contagion}), only select one defaulted institution uniformly at random in each round $t\geq1$ and expose it to its neighbors (draw edges). Using the finitary assumption, it is then sufficient to keep track of the following sets and quantities during the default process:

\begin{align*}
U^\alpha(t) &:= \left\{i\in[n]\,:\,\alpha_i=\alpha\text{ and }i\text{ is defaulted but unexposed at time }t\right\},\\
S_{j,m,l}^\alpha(t) &:= \left\{ i\in[n]\,:\,\alpha_i=\alpha, j(n,i)=j, c_i=m\text{ and }i\text{ has total exposure of }l\text{ at time }t\right\},\\
D_j(t) &:= \{i\in[n]\,:\,j(n,i)=j\text{ and }i\text{ is defaulted at time }t\},
\end{align*}
\begin{align*}
u^\alpha(t) &:= \left\vert U^\alpha(t)\right\vert, & c_{j,m,l}^\alpha(t) &:= \left\vert S_{j,m,l}^\alpha(t)\right\vert, & s(t) &:= \sum_{j\in[J]}\tilde{s}_j \vert D_j(t)\vert, & 
w^{r,\alpha,\beta}(t) &:= \sum_{i\in U^\beta(t)}w_i^{+,r,\alpha}.
\end{align*}

\noindent Let $h(t):=(u^\alpha(t),c_{j,m,l}^\alpha(t),s(t),w^{r,\alpha,\beta}(t))$ the vector of all tracked quantities at time $t$ (for the sake of a better readability we omitted the index sets) and $H(t)=(h(s))_{s\leq t}$. Then for $n$ large enough such that all $p_{i,j}^r<R^{-1}$ (possible by finitary weights), the expected evolution of the system at 
time $t$ is 

\begin{align*}
&\E\left[\left.c_{j,m,l}^\alpha(t+1)-c_{j,m,l}^\alpha(t)\,\right\vert\, H(t)\right]\\
&\hspace{1.5cm}= \frac{1}{\sum_{\beta\in[T]}u^\beta(t)}\sum_{\beta\in[T]}\sum_{v\in U^\beta(t)} \sum_{r\in[R]}\left(\sum_{i\in S_{j,m,l-r}^\alpha(t)}\frac{w_v^{+,r,\alpha}w_i^{-,r,\beta}}{n} - \sum_{i\in S_{j,m,l}^\alpha(t)}\frac{w_v^{+,r,\alpha}w_i^{-,r,\beta}}{n}\right)\\
&\hspace{1.5cm}= \sum_{r\in[R]}\frac{\sum_{\beta\in[T]}w^{r,\alpha,\beta}(t) \tilde{w}_j^{-,r,\beta}}{\sum_{\beta\in[T]}u^\beta(t)}\frac{c_{j,m,l-r}^\alpha(t)-c_{j,m,l}^\alpha(t)}{n}.
\end{align*}
\begin{align*}
\E\left[\left.s(t+1)-s(t)\,\right\vert\, H(t)\right] &= \sum_{j\in[J]}\tilde{s}_j \sum_{\alpha\in[T]} \sum_{r\in[R]} \sum_{m=0}^{c_\text{max}} \sum_{l=m-r}^{m-1} \frac{\sum_{\beta\in[T]}w^{r,\alpha,\beta}(t) \tilde{w}_j^{-,r,\beta}}{\sum_{\beta\in[T]}u^\beta(t)} \frac{c_{j,m,l}^\alpha(t)}{n},
\end{align*}
\begin{align*}
\E\left[\left.u^\alpha(t+1)-u^\alpha(t)\,\right\vert\, H(t)\right] &= -\frac{u^\alpha(t)}{\sum_{\beta\in[T]}u^\beta(t)} + \sum_{j,m} \sum_{r\in[R]}\frac{\sum_{\beta\in[T]}w^{r,\alpha,\beta}(t) \tilde{w}_j^{-,r,\beta}}{\sum_{\beta\in[T]}u^\beta(t)}\sum_{l=m-r}^{m-1}\frac{c_{j,m,l}^\alpha(t)}{n},
\end{align*}
\begin{align*}
&\E\left[\left.w^{r,\alpha,\beta}(t+1)-w^{r,\alpha,\beta}(t)\,\right\vert\, H(t)\right]\\
&\hspace{3.5cm}= -\frac{w^{r,\alpha,\beta}(t)}{\sum_{\gamma\in[T]}u^\gamma(t)} + \sum_{j,m} \tilde{w}_j^{+,r,\alpha} \sum_{s\in[R]}\frac{\sum_{\gamma\in[T]}w^{s,\beta,\gamma}(t)\tilde{w}_j^{-,s,\gamma}}{\sum_{\gamma\in[T]}u^\gamma(t)}\sum_{l=m-s}^{m-1}\frac{c_{j,m,l}^\beta(t)}{n}.
\end{align*}

\noindent The expressions on the right-hand side are all Lipschitz functions of $u^\alpha(t),w^{r,\alpha,\beta}(t),c_{j,m,l}^\alpha(t)$ as long as $\sum_{\beta\in[T]}u^\beta(t)$ is bounded away from zero. All the remaining conditions in Wormald’s theorem \cite{wormald1995} can be checked by similar means as in \cite{Detering2015a}. We can thus uniformly approximate 
\begin{align}
n^{-1}c_{j,m,l}^\alpha(t) &= \gamma_{j,m,l}^\alpha(n^{-1}t) + o_p(1),\label{3:eqn:approx:gamma}\\
n^{-1}s(t) &= \sigma(n^{-1}t) + o_p(1),\label{3:eqn:approx:delta}\\
n^{-1}u^\alpha(t) &= \nu^\alpha(n^{-1}t) + o_p(1),\label{3:eqn:approx:nu}\\
n^{-1}w^{r,\alpha,\beta}(t) &= \mu^{r,\alpha,\beta}(n^{-1}t) + o_p(1),\label{3:eqn:approx:mu}
\end{align}
where the functions $\gamma_{j,m,l}^\alpha(\tau)$, $\sigma(\tau)$, $\nu^\alpha(\tau)$ and $\mu^{r,\alpha,\beta}(\tau)$ are defined as the unique solution of
\begin{align}
\frac{\dd}{\dd \tau}\gamma_{j,m,l}^\alpha(\tau) &= \sum_{r\in[R]}\frac{\sum_{\beta\in[T]}\mu^{r,\alpha,\beta}(\tau)\tilde{w}_j^{-,r,\beta}}{\sum_{\beta\in[T]}\nu^\beta(\tau)}\left(\gamma_{j,m,l-r}^\alpha(\tau)-\gamma_{j,m,l}^\alpha(\tau)\right),\label{3:eqn:diff:gamma}\\
\frac{\dd}{\dd \tau}\sigma(\tau) &= \sum_{j\in[J]}\tilde{s}_j \sum_{\alpha\in[T]} \sum_{r\in[R]} \sum_{m=0}^{c_\text{max}} \sum_{l=m-r}^{m-1}\frac{\sum_{\beta\in[T]}\mu^{r,\alpha,\beta}(\tau)\tilde{w}_j^{-,r,\beta}}{\sum_{\beta\in[T]}\nu^\beta(\tau)}\gamma_{j,m,l}^\alpha(\tau),\label{3:eqn:diff:delta}\\
\frac{\dd}{\dd \tau} \nu^\alpha(\tau) &= -\frac{\nu^\alpha(\tau)}{\sum_{\beta\in[T]}\nu^\beta(\tau)}+\sum_{j,m}\sum_{r\in[R]}\frac{\sum_{\beta\in[T]}\mu^{r,\alpha,\beta}(\tau)\tilde{w}_j^{-,r,\beta}}{\sum_{\beta\in[T]}\nu^\beta(\tau)}\sum_{l=m-r}^{m-1}\gamma_{j,m,l}^\alpha(\tau),\label{3:eqn:diff:nu}\\
\frac{\dd}{\dd \tau}\mu^{r,\alpha,\beta}(\tau) &= -\frac{\mu^{r,\alpha,\beta}(\tau)}{\sum_{\gamma\in[T]}\nu^\gamma(\tau)}+\sum_{j,m}\tilde{w}_j^{+,r,\alpha}\sum_{s\in[R]}\frac{\sum_{\gamma\in[T]}\mu^{s,\alpha,\gamma}(\tau)\tilde{w}_j^{-,s,\gamma}}{\sum_{\gamma\in[T]}\nu^\gamma(\tau)}\sum_{l=m-r}^{m-1}\gamma_{j,m,l}^\alpha(\tau).\label{3:eqn:diff:mu}
\end{align}
Approximations \eqref{3:eqn:approx:gamma}-\eqref{3:eqn:approx:mu} hold uniformly for $t/n<\hat{\tau} := \inf\{\tau\in\R_{+,0}\,:\,\sum_{\beta\in[T]}\nu^\beta(\tau)=0\}$. For $z^{r,\alpha,\beta}(\tau) := \int_0^\tau\mu^{r,\alpha,\beta}(s)/\sum_{\gamma\in[T]}\nu^\gamma(s)\dd s$, an implicit solution of \eqref{3:eqn:diff:gamma}-\eqref{3:eqn:diff:mu} is given by
\begin{align*}
\gamma_{j,m,l}^\alpha(\tau) &= \P(\bm{W}^\pm=\tilde{\bm{w}}_j^\pm, C=m, A=\alpha) \P\Bigg(\sum_{s\in[R]}s\mathrm{Poi}\Bigg(\sum_{\beta\in[T]}\tilde{w}_j^{-,s,\beta}z^{s,\alpha,\beta}(\tau)\Bigg)=l\Bigg),\\
\sigma(\tau) &= \E\Bigg[S\P\Bigg(\sum_{s\in[R]}s\mathrm{Poi}\Bigg(\sum_{\beta\in[T]}W^{-,s,\beta}z^{s,\alpha,\beta}(\tau)\Bigg)\geq C\Bigg)\1\{A=\alpha\}\Bigg],\\
\nu^\alpha(\tau) &= \E\Bigg[\P\Bigg(\sum_{s\in[R]}s\mathrm{Poi}\Bigg(\sum_{\beta\in[T]}W^{-,s,\beta}z^{s,\alpha,\beta}(\tau)\Bigg)\geq C\Bigg)\1\{A=\alpha\}\Bigg] - \int_0^\tau\frac{\nu^\alpha(s)}{\sum_{\beta\in[T]}\nu^\beta(s)},
\end{align*}
\begin{align*}
\mu^{r,\alpha,\beta}(\tau) &= \E\Bigg[W^{+,r,\alpha}\P\Bigg(\sum_{s\in[R]}s\mathrm{Poi}\Bigg(\sum_{\gamma\in[T]}W^{-,s,\gamma}z^{s,\beta,\gamma}(\tau)\Bigg)\geq C\Bigg)\1\{A=\beta\}\Bigg] - z^{r,\alpha,\beta}(\tau).
\end{align*}
In particular, note that $\sigma(\tau)=g(\bm{z}(\tau))$ and $\mu^{r,\alpha,\beta}(\tau) = f^{r,\alpha,\beta}(\bm{z}(\tau))$. Thus for $\tau<\hat{\tau}$, it holds $f^{r,\alpha,\beta}(\bm{z}(\tau)) = n^{-1}w^{r,\alpha,\beta}(\lfloor\tau n\rfloor) + o_p(1) \geq 0 + o_p(1)$ 
and by letting $n\to\infty$, it follows $f^{r,\alpha,\beta}(\bm{z}(\tau))\geq0$. 
By continuity of $\bm{z}(\tau)$ and $\bm{z}(0)=0$, hence $\bm{z}(\tau)\in P_0$. 
Further,
\[ f^{r,\alpha,\beta}(\bm{z}(\tau)) = \mu^{r,\alpha,\beta}(\tau) = n^{-1}w^{r,\alpha,\beta}(\lfloor\tau n\rfloor) + o_p(1) \leq n^{-1} \overline{w} u^\beta(\lfloor\tau n\rfloor) + o_p(1) = \overline{w}\nu^\beta(\tau) + o_p(1) \]
and as $n\to\infty$, $f^{r,\alpha,\beta}(\bm{z}(\tau)) \leq \overline{w} \nu^\beta(\tau)$. 

As $\tau\to\hat{\tau}$, $\sum_{\beta\in[T]}\nu^\beta(\tau)\to0$ and hence by continuity of $\bm{z}(\tau)$, $f^{r,\alpha,\beta}(\bm{z}(\hat{\tau}))=0$ for all $(r,\alpha,\beta)\in V$. Again by continuity of $\bm{z}(\tau)$ and closedness of $P_0$, we then conclude that $\bm{z}(\hat{\tau})\geq\hat{\bm{z}}$. 
In particular, using continuity of $\sigma$
\begin{equation}\label{3:eqn:tau:hat}
g(\hat{\bm{z}}) \leq g(\bm{z}(\hat{\tau})) = \lim_{\tau\to\hat{\tau}}g(\bm{z}(\tau)) =  \lim_{\tau\to\hat{\tau}}\sigma(\tau) = \sigma(\hat{\tau}).
\end{equation}
Let now $\hat{t}$ denote the first time that $\sum_{\alpha\in[T]}u^\alpha(t)=0$ (i.\,e.~the number of steps until the contagion process stops). We show that $\hat{t}/n\geq\hat{\tau}+o_p(1)$. Define $X_n:=(\lfloor\hat{\tau}n\rfloor\wedge\hat{t})/n-\hat{\tau}$. Then $\hat{t}/n\geq\hat{\tau}+X_n$. Further, for $\epsilon>0$ and $n$ large enough such that $\hat{\tau}-\lfloor\hat{\tau} n\rfloor/n\leq\epsilon$, we obtain
\[ \P(\vert X_n\vert>\epsilon) = \P(\hat{\tau}-\hat{t}/n>\epsilon) \leq \P\Bigg(\sum_{\alpha\in[T]}\nu^\alpha(\hat{t}/n)>\frac{1}{2}\min_{\tau\in[0,\hat{\tau}-\epsilon]}\sum_{\alpha\in[T]}\nu^\alpha(\tau) , \hat{t}/n<\hat{\tau}\Bigg), \]
using continuity of $\sum_{\alpha\in[T]}\nu^\alpha(\tau)$. Let now $(Y_n)_{n\in\N}$ such that $\left\vert \sum_{\alpha\in[T]}u^\alpha(t)/n-\nu^\alpha(t/n)\right\vert\leq Y_n$ and $Y_n=o_p(1)$ (existence of $(Y_n)_{n\in\N}$ ensured by \eqref{3:eqn:approx:nu}). Since $\sum_{\alpha\in[T]}u^\alpha(\hat{t})=0$, we conclude that
\[ \P(\vert X_n\vert>\epsilon)\leq\P\Bigg(Y_n>\frac{1}{2}\min_{\tau\in[0,\hat{\tau}-\epsilon]}\sum_{\alpha\in[T]}\nu^\alpha(\tau)\Bigg)\to0,\quad\text{as }n\to\infty. \]
But then by \eqref{3:eqn:diff:delta}, for arbitrary $\epsilon>0$,
\[ n^{-1}\mathcal{S}_n = n^{-1}s(\hat{t}) \geq n^{-1}s\left( \hat{t} \wedge n(1-\epsilon)\hat{\tau} \right) = \sigma\left(\frac{\hat{t}}{n} \wedge (1-\epsilon)\hat{\tau}\right) + o_p(1) = \sigma\left((1-\epsilon)\hat{\tau}\right) + o_p(1), \]
where for the last equality we used $\hat{t}/n\geq\hat{\tau}+o_p(1)$. Now letting $\epsilon\to0$, using continuity of $\sigma$ and combining with \eqref{3:eqn:tau:hat} shows the lower bound.


In order to prove the second part, we first want to show that the existence of $\bm{v}$ implies that in fact $\bm{z}(\hat{\tau})=\hat{\bm{z}}$. To this end, assume that $\bm{z}(\hat{\tau})\neq\hat{\bm{z}}$. Then there exists $(r,\alpha,\beta)\in V$ and $\delta>0$ such that $z^{r,\alpha,\beta}(\hat{\tau})>\hat{z}^{r,\alpha,\beta}+\delta v^{r,\alpha,\beta}$. Without loss of generality assume that $z^{r,\alpha,\beta}(\tau)$ is the first coordinate that reaches $\hat{z}^{r,\alpha,\beta}+\delta v^{r,\alpha,\beta}$, that is there exists $\tau_\delta\in[0,\hat{\tau}]$ such that $\bm{z}(\tau_\delta)\leq \hat{\bm{z}}+\delta\bm{v}$ componentwise and $z^{r,\alpha,\beta}(\tau_\delta)=\hat{z}^{r,\alpha,\beta}+\delta v^{r,\alpha,\beta}$. But by $D_{\bm{v}}f^{r,\alpha,\beta}(\hat{\bm{z}})<0$ and continuity of $D_{\bm{v}}f^{r,\alpha,\beta}(\bm{z})$, we then derive for $\delta>0$ small enough that $0 > f^{r,\alpha,\beta}(\hat{\bm{z}}+\delta\bm{v}) \geq f^{r,\alpha,\beta}(\bm{z}(\tau_\delta))$, where we used monotonicity of $f^{r,\alpha,\beta}$ from Lemma \ref{3:lem:properties:f}. This contradicts that $f^{r,\alpha,\beta}(\bm{z}(\tau))\geq0$ for all $\tau\in[0,\hat{\tau}]$ and hence it must hold that $\bm{z}(\hat{\tau})=\hat{\bm{z}}$. In particular, 
$g(\hat{\bm{z}})=\sigma(\hat{\tau})$ (cf.~\eqref{3:eqn:tau:hat}).

The difficulty in the following is that the system is only described by the functions $\gamma_{j,m,l}^\alpha(\tau)$, $\sigma(\tau)$, $\nu^\alpha(\tau)$ and $\mu^{r,\alpha,\beta}(\tau)$ as long as $\tau<\hat{\tau}$, the first time at which $\sum_{\alpha\in[T]}\nu^\alpha(\tau)=0$. Wormald's theorem makes no statement about the system at or after $\hat{\tau}$, however. The idea is hence the following: We let $\tau_\epsilon$ be the first time at which $\mu^{r,\alpha,\beta}(\tau)\leq v^{r,\alpha,\beta}\epsilon$ for all $(r,\alpha,\beta)\in V$ and choose a sequence $(\epsilon_n)_{n\in\N}\subset\R_+$ such that $\epsilon_n\to0$ as $n\to\infty$. We then consider the cascade process as before for the first $\lfloor\tau_{\epsilon_n}n\rfloor$ steps and we show that the number of remaining defaults $R_n$ divided by $n$ converges to $0$ in probability as $n\to\infty$. In particular, this will show that
\[ n^{-1}\vert\mathcal{S}_n\vert = n^{-1}s\left(\hat{t}\right) = n^{-1}s\left(\lfloor\tau_{\epsilon_n}n\rfloor+R_n\right) \leq \sigma\left(\tau_{\epsilon_n}\right) + o_p(1) \leq \sigma\left(\hat{\tau}\right) + o_p(1) = g(\hat{\bm{z}}) + o_p(1). \]
In order to show $n^{-1}R_n=o_p(1)$, we will expose the defaulted banks round by round as in \eqref{3:eqn:default:contagion}, i.\,e.~we expose the banks in $\bigcup_{\alpha\in[T]}U^\alpha(\lfloor\tau_{\epsilon_n}n\rfloor)$ at once and so on. However, banks with 
$w_i^{+,r,\alpha}=0$ for all $r\in[R]$ and $\alpha\in[T]$ will never infect any new banks. Thus, we only need to consider banks with $\sum_{r\in[R]}\sum_{\alpha\in[T]}w_i^{+,r,\alpha}>0$ in the following. Since we are in a finitary setting, this means that there exists $w_0>0$ such that $\sum_{r\in[R]}\sum_{\alpha\in[T]}w_i^{+,r,\alpha}\geq w_0$ for all banks. Taking into account also banks with total out-weight of \mbox{zero only causes an extra bounded factor for $R_n$}.

For each solvent bank at step $\lfloor\tau_{\epsilon_n}n\rfloor$ there are two possible ways to default: Either there is one exposure to a defaulted bank that is larger than the remaining capital at step $\lfloor\tau_{\epsilon_n}n\rfloor$ (the bank defaults \emph{directly}) or there are at least two exposures to defaulted banks that add up to an amount larger than the remaining capital (the bank defaults \emph{indirectly}). Therefore, for $\alpha\in[T]$ and $l\geq1$ we define the following sets:

\begin{align*}
\mathcal{D}_l^\alpha &= \mathcal{D}_l^\alpha(\tau_{\epsilon_n}) = \{i\in[n]\,:\,\alpha_i=\alpha\text{ and }i\text{ defaults directly in the }l\text{-th round after step }\lfloor\tau_{\epsilon_n}n\rfloor\}\\
\mathcal{I}_l^\alpha &= \mathcal{I}_l^\alpha(\tau_{\epsilon_n}) = \{i\in[n]\,:\,\alpha_i=\alpha\text{ and }i\text{ defaults indirectly in the }l\text{-th round after step }\lfloor\tau_{\epsilon_n}n\rfloor\}
\end{align*}

\noindent Further, let $\mathcal{T}^\alpha_l = \mathcal{D}^\alpha\cup\mathcal{I}^\alpha$. In particular, $R_n=\sum_{\alpha\in[T]}\sum_{l\geq1}\left\vert\mathcal{T}_l^\alpha(\tau_{\epsilon_n})\right\vert$. Further, the following quantities will play an important role:

\[ D_l^{r,\alpha,\beta} = \sum_{i\in\mathcal{D}_l^\beta}w_i^{+,r,\alpha}, I_l^{r,\alpha,\beta} = \sum_{i\in\mathcal{I}_l^\beta}w_i^{+,r,\alpha}\quad\text{and}\quad T_l^{r,\alpha,\beta}=\sum_{i\in\mathcal{T}_l^\beta}w_i^{+,r,\alpha},\quad l\geq 1, (r,\alpha,\beta)\in V \]

\noindent We now exploit again the assumption that $D_{\bm{v}}f^{r,\alpha,\beta}(\hat{\bm{z}})<0$ for all $(r,\alpha,\beta)\in V$. Also recall that the expression $D_{\bm{v}}f^{r,\alpha,\beta}(\bm{z})$ is continuous in $\bm{z}$ since the weights are assumed finitary. Further $\bm{z}(\tau)$ is continuous in $\tau$. Hence for $\epsilon>0$ small enough (i.\,e.~$\bm{z}(\tau_\epsilon)$ close to $\hat{\bm{z}}$), it holds

\begin{align*}
0 &> D_{\bm{v}}f^{r,\alpha,\beta}(\bm{z}(\tau_\epsilon))\\
&= \sum_{r'\in[R]} \E\Bigg[W^{+,r,\alpha}\Bigg(\sum_{\beta'\in[T]}v^{r',\beta,\beta'}W^{-,r',\beta'}\Bigg)\\
&\hspace{1.76cm}\times\P\Bigg(\sum_{s\in[R]}s\mathrm{Poi}\Bigg(\sum_{\gamma\in[T]}W^{-,s,\gamma}z^{s,\beta,\gamma}(\tau_\epsilon)\Bigg)\in\{C-r',\ldots,C-1\}\Bigg)\1\{A=\beta\}\Bigg] - v^{r,\alpha,\beta}\\
&= \sum_{j,m}\tilde{w}_j^{+,r,\alpha}\sum_{r'\in[R]}\Bigg(\sum_{\beta'\in[T]}v^{r',\beta,\beta'}\tilde{w}_j^{-,r',\beta'}\Bigg)\sum_{s=1}^{r'}\gamma_{j,m,m-s}^\beta(\tau_\epsilon) - v^{r,\alpha,\beta}.
\end{align*}

\noindent We can hence find $c_1<1$ such that for all $(r,\alpha,\beta)\in V$ it holds
\[ \sum_{j,m}\tilde{w}_j^{+,r,\alpha}\sum_{r'\in[R]}\Bigg(\sum_{\beta'\in[T]}v^{r',\beta,\beta'}\tilde{w}_j^{-,r',\beta'}\Bigg)\sum_{s=1}^{r'}\gamma_{j,m,m-s}^\beta(\tau_\epsilon) \leq c_1v^{r,\alpha,\beta}. \]

\noindent By \eqref{3:eqn:approx:gamma} and possibly slightly increasing $c_1$, we then derive that

\[ \sum_{j,m}\tilde{w}_j^{+,r,\alpha}\sum_{r'\in[R]}\Bigg(\sum_{\beta'\in[T]}v^{r',\beta,\beta'}\tilde{w}_j^{-,r',\beta'}\Bigg)\sum_{s=1}^{r'}\frac{c_{j,m,m-s}^\beta(\lfloor \tau_\epsilon n\rfloor)}{n} \leq c_1v^{r,\alpha,\beta} \]

\noindent on a $\sigma(h(\lfloor\tau_\epsilon n\rfloor))$-measurable set $\Omega_n^\epsilon$ such that $\lim_{n\to\infty}\P(\Omega_n^\epsilon)=1$ for every $\epsilon>0$. Further, by the definition of $\tau_\epsilon$ and \eqref{3:eqn:approx:mu}, we can choose $\Omega_n^\epsilon$ in such a way that $n^{-1}w^{r,\alpha,\beta}(\lfloor\tau_\epsilon n\rfloor)\leq2\epsilon v^{r,\alpha,\beta}$ holds on $\Omega_n^\epsilon$ for all $(r,\alpha,\beta)\in V$. We can then compute on $\Omega_n^\epsilon$
\begin{align*}
n^{-1}\E\Big[D_1^{r,\alpha,\beta}\,\Big\vert\,h(\lfloor\tau_\epsilon n\rfloor)\Big] 
&\leq \sum_{j,m} \tilde{w}_j^{+,r,\alpha}  \sum_{l=0}^{m-1}\frac{c_{j,m,l}^\beta(\lfloor\tau_\epsilon n\rfloor)}{n} \sum_{\beta'\in[T]}\sum_{v\in U^{\beta'}(\lfloor\tau_\epsilon n\rfloor)}\sum_{r'=m-l}^R\frac{w_v^{r',+,\beta}\tilde{w}_j^{-,r',\beta'}}{n}\\
&\leq 2\epsilon\sum_{j,m} \tilde{w}_j^{+,r,\alpha} \sum_{l=0}^{m-1}\frac{c_{j,m,l}^\beta(\lfloor\tau_\epsilon n\rfloor)}{n} \sum_{\beta'\in[T]}\sum_{r'=m-l}^R v^{r',\beta,\beta'}\tilde{w}_j^{-,r',\beta'}
\leq 2\epsilon c_1v^{r,\alpha,\beta},
\end{align*}
\begin{align*}
n^{-1}\E\Big[I_1^{r,\alpha,\beta}\,\Big\vert\,h(\lfloor\tau_\epsilon n\rfloor)\Big] 
&\leq n^{-1}\sum_{j,m}\sum_{l=0}^{m-1}\sum_{i\in S_{j,m,l}^\beta(\lfloor\tau_\epsilon n\rfloor)} \overline{w} \Bigg( \sum_{\beta'\in[T]}\sum_{v\in U^{\beta'}(\lfloor\tau_\epsilon n\rfloor)}\sum_{r'=m-l}^R\frac{w_v^{r',+,\beta}w_i^{-,r',\beta'}}{n} \Bigg)^2\\
&\leq \overline{w}\Bigg( \sum_{\beta'\in[T]}\sum_{r'=0}^R\frac{w^{r',\beta,\beta'}(\lfloor\tau_\epsilon n\rfloor)}{n}\overline{w} \Bigg)^2
\leq 
C \epsilon^2,
\end{align*}
where $C:=4T^2(R+1)^2\overline{w}^3\Vert\bm{v}\Vert_\infty^2$. In particular, for $\epsilon>0$ small enough we find $c_2\in(0,1-c_1)$ such that $C\epsilon^2\leq 2\epsilon c_2 v^{r,\alpha,\beta}$ 
for all $(r,\alpha,\beta)\in V$ and hence on $\Omega_n^\epsilon$ it holds that
\[ n^{-1}\E\Big[T_1^{r,\alpha,\beta}\,\Big\vert\,h(\lfloor\tau_\epsilon n\rfloor)\Big] = n^{-1}\E\Big[D_1^{r,\alpha,\beta}\,\Big\vert\,h(\lfloor\tau_\epsilon n\rfloor)\Big] + n^{-1}\E\Big[I_1^{r,\alpha,\beta}\,\Big\vert\,h(\lfloor\tau_\epsilon n\rfloor)\Big] \leq 2\epsilon (c_1+c_2)v^{r,\alpha,\beta}. \]
Let then $c:=c_1+c_2\in(0,1)$. We continue inductively: Assume that on $\Omega_n^\epsilon$ it holds for $l\geq1$ that $n^{-1}\E\left[\left.T_l^{r,\alpha,\beta}\,\right\vert\,h(\lfloor\tau_\epsilon n\rfloor)\right] \leq 2\epsilon c^l v^{r,\alpha,\beta}$. We then derive on $\Omega_n^\epsilon$ that
\begin{align*}
n^{-1}\E\left[\left.D_{l+1}^{r,\alpha,\beta}\,\right\vert\,h(\lfloor\tau_\epsilon n\rfloor)\right] 
&= n^{-1}\sum_{j,m}\sum_{l=0}^{m-1}\sum_{i\in S_{j,m,l}^\beta(\lfloor\tau_\epsilon n\rfloor)} w_i^{+,r,\alpha}\P\left(\left.i\in\mathcal{D}_{l+1}^\beta\,\right\vert\,h(\lfloor\tau_\epsilon n\rfloor)\right)\\
&\hspace*{-1.65cm}\leq \sum_{j,m}\tilde{w}_j^{+,r,\alpha}\sum_{l=0}^{m-1}\frac{c_{j,m,l}^\beta(\lfloor\tau_\epsilon n\rfloor)}{n}\sum_{\beta'\in[T]}\sum_{r'=m-l}^R\tilde{w}_j^{-,r',\beta'}n^{-1}\E\left[\left.T_l^{r',\beta,\beta'}\,\right\vert\,h(\lfloor\tau_\epsilon n\rfloor)\right]\\
&\hspace*{-1.65cm}\leq 
2\epsilon c^l \sum_{j,m}\tilde{w}_j^{+,r,\alpha}\sum_{r'\in[R]}\Bigg(\sum_{\beta'\in[T]}v^{r',\beta,\beta'}\tilde{w}_j^{-,r',\beta'}\Bigg)\sum_{s=1}^{r'}\frac{c_{j,m,m-s}^\beta(\lfloor\tau_\epsilon n\rfloor)}{n}
\leq 2\epsilon c^l c_1v^{r,\alpha,\beta},
\end{align*}
\begin{align*}
n^{-1}\E\left[\left.I_{l+1}^{r,\alpha,\beta}\,\right\vert\,h(\lfloor\tau_\epsilon n\rfloor)\right] 
&= n^{-1}\sum_{j,m}\sum_{l=0}^{m-1}\sum_{i\in S_{j,m,l}^\beta(\lfloor\tau_\epsilon n\rfloor)} w_i^{+,r,\alpha}\P\left(\left.i\in\mathcal{I}_{l+1}^\beta\,\right\vert\,h(\lfloor\tau_\epsilon n\rfloor)\right)\\
&\hspace{-1.3cm}\leq 
\sum_{j,m} \tilde{w}_j^{+,r,\alpha}  \sum_{l=0}^{m-1}\frac{c_{j,m,l}^\beta(\lfloor\tau_\epsilon n\rfloor)}{n} \left(\sum_{\beta'\in[T]}\sum_{r'=m-l}^R \tilde{w}_j^{-,r',\beta'}n^{-1}\E\left[\left.T_l^{r',\beta,\beta'}\,\right\vert\,h(\lfloor\tau_\epsilon n\rfloor)\right]\right)\\
&\hspace{2.0cm}\times \left(\sum_{\beta'\in[T]}\sum_{r'=m-l}^R\tilde{w}_j^{-,r',\beta'}n^{-1}\sum_{k\leq l}\E\left[\left.T_k^{r',\beta,\beta'}\,\right\vert\,h(\lfloor\tau_\epsilon n\rfloor)\right]\right)\\
&\hspace{-1.3cm}\leq \overline{w}\left(2\epsilon T(R+1)\overline{w} c^l v^{r',\beta,\beta'}\right)\Bigg(2\epsilon T(R+1)\overline{w}\sum_{k\leq l} c^k v^{r',\beta,\beta'}\Bigg) \leq C c^l \frac{1}{1-c} \epsilon^2
\end{align*}
Now choose $\epsilon>0$ small enough such that even $\frac{C}{1-c}\epsilon^2 \leq 2\epsilon c_2v^{r,\alpha,\beta}$ and conclude that on $\Omega_n^\epsilon$
\begin{align*}
&\hspace*{-5.2cm} n^{-1}\E\left[\left.T_{l+1}^{r,\alpha,\beta}\,\right\vert\,h(\lfloor\tau_\epsilon n\rfloor)\right] = n^{-1}\E\left[\left.D_{l+1}^{r,\alpha,\beta}\,\right\vert\,h(\lfloor\tau_\epsilon n\rfloor)\right] + n^{-1}\E\left[\left.I_{l+1}^{r,\alpha,\beta}\,\right\vert\,h(\lfloor\tau_\epsilon n\rfloor)\right] \leq 2\epsilon c^{l+1} v^{r,\alpha,\beta},\\
n^{-1}\sum_{\alpha\in[T]}\sum_{l\geq1}\E\left[\left.\vert\mathcal{T}_l^\alpha\vert\,\right\vert\,h(\lfloor\tau_\epsilon n\rfloor)\right] &\leq n^{-1}\sum_{\alpha\in[T]}\sum_{l\geq1}\E\Bigg[\sum_{i\in\mathcal{T}_l^\alpha} \frac{\sum_{r\in[R]}\sum_{\gamma\in[T]}w_i^{+,r,\gamma}}{w_0}\,\Bigg\vert\,h(\lfloor\tau_\epsilon n\rfloor)\Bigg]\\
&=w_0^{-1}n^{-1}\sum_{r\in[R]}\sum_{\alpha,\gamma\in[T]}\sum_{l\geq1}\E\left[\left.T_l^{r,\gamma,\alpha}\,\right\vert\,h(\lfloor\tau_\epsilon n\rfloor)\right]\\
&\leq 2\epsilon w_0^{-1}(R+1)T^2\frac{1}{1-c}\Vert\bm{v}\Vert_\infty.
\end{align*}
Consider now again the sequence $(\epsilon_n)_{n\in\N}$ from before and let $\epsilon,\delta>0$ arbitrary. For $n$ large enough such that $\epsilon_n\leq\epsilon$, we derive (using Markov's inequality in the penultimate step)
\begin{align*}
\P\left(n^{-1}R_n\geq\delta\right) &
\leq \P\Bigg(n^{-1} \sum_{\alpha\in[T]}\sum_{l\geq1}\vert\mathcal{T}_l^\alpha(\tau_\epsilon)\vert\geq\delta\Bigg)
\leq \E\Bigg[\delta^{-1}n^{-1}\sum_{\alpha\in[T]}\sum_{l\geq1}\E\left[\vert\mathcal{T}_l^\alpha(\tau_\epsilon)\vert\,\big\vert\,h(\lfloor\tau_\epsilon n\rfloor)\right]\Bigg]\\
&\leq 2\epsilon \delta^{-1}w_0^{-1}(R+1)T^2\frac{1}{1-c}\Vert\bm{v}\Vert_\infty + (1-\P(\Omega_n^\epsilon)).
\end{align*}
Choosing $\epsilon$ small enough and $n$ large enough, this quantity becomes arbitrarily small. Hence, $n^{-1}R_n=o_p(1)$ and this finishes the proof as explained above. 
\end{proof}

\begin{theorem}\label{3:thm:finitary:weights}
Consider a financial system described by a finitary regular vertex sequence and let $\hat{\bm{z}}$ and $\bm{z}^*$ be the smallest resp.~largest joint root in $P_0$ of the functions $f^{r,\alpha,\beta}$, $(r,\alpha,\beta)\in V$. Then $g(\hat{\bm{z}}) + o_p(1) \leq n^{-1}\mathcal{S}_n \leq g(\bm{z}^*) + o_p(1)$. In particular, if $\hat{\bm{z}}=\bm{z}^*$, then $n^{-1}\mathcal{S}_n = g(\hat{\bm{z}}) + o_p(1)$.
\end{theorem}
The idea for the proof of Theorem \ref{3:thm:finitary:weights} is to apply a small further shock to the financial system such that the second statement in Proposition \ref{3:prop:finitary:weights} becomes applicable. That is, we let each solvent bank in the system default independently with probability $\epsilon>0$ and denote the analogues of $f^{r,\alpha,\beta}$, $g$, $\hat{\bm{z}}$ and $\bm{z}^*$ by $f_\epsilon^{r,\alpha,\beta}$, $g_\epsilon$, $\hat{\bm{z}}(\epsilon)$ respectively $\bm{z}^*(\epsilon)$. That is,
\begin{align*}
f_\epsilon^{r,\alpha,\beta}(\bm{z}) &= \epsilon\left(\E\left[W^{+,r,\alpha}\1\{A=\beta\}\right] - z^{r,\alpha,\beta}\right) + (1-\epsilon)f^{r,\alpha,\beta}(\bm{z}),\hspace{0.3cm} & g_\epsilon(\bm{z}) &= \epsilon\E[S] + (1-\epsilon)g(\bm{z}).
\end{align*}
We can assume in the following that $\E\left[W^{+,r,\alpha}\1\{A=\beta\}\right]>0$ and hence $f_\epsilon^{r,\alpha,\beta}(\bm{z})>f^{r,\alpha,\beta}(\bm{z})$ for all $\bm{z}$, since otherwise $f^{r,\alpha,\beta}(\bm{z})=-z^{r,\alpha,\beta}$ and we can simply leave out the $(r,\alpha,\beta)$-component in the proof. The following lemma describes $\bm{z}^*(\epsilon)$ for small $\epsilon$:
\begin{lemma}\label{3:lem:z^*:epsilon}
The function $\bm{z}^*:\R_{+,0}\to\R_{+,0}^V$ is right-continuous and monotonically increasing in each component. In particular, the derivative $(\bm{z}^*)'(\epsilon)$ exists for Lebesgue-almost every $\epsilon>0$ and $\bm{z}^*(\epsilon)-\bm{z}^*\geq\int_0^\epsilon(\bm{z}^*)'(\xi)\dd\xi$ componentwise.
\end{lemma}
\begin{proof}
For every $\bm{z}\in P_0$ and $\epsilon>0$ it holds $f_\epsilon^{r,\alpha,\beta}(\bm{z})\geq f^{r,\alpha,\beta}(\bm{z})\geq0$ and hence $P_0\subseteq P_0(\epsilon)$, where $P_0(\epsilon)$ denotes the analogue of $P_0$ for the additionally shocked case. In particular, $\bm{z}^*\in P_0(\epsilon)$ and hence $\bm{z}^*\leq\bm{z}^*(\epsilon)$ componentwise. The same argument shows that $\bm{z}^*(\epsilon_1)\leq\bm{z}^*(\epsilon_2)$ for any $\epsilon_1\leq\epsilon_2$ and hence $\bm{z}^*(\epsilon)$ is monotonically increasing in each component.

In particular, $\lim_{\epsilon\to0+}\bm{z}^*(\epsilon)$ exists and $\lim_{\epsilon\to0+}\bm{z}^*(\epsilon)\in\bigcap_{\epsilon>0}P_0(\epsilon)$. Now let $\delta>0$. By continuity of $f_\epsilon^{r,\alpha,\beta}(\bm{z})$ with respect to $\epsilon$ and $\bm{z}$, it holds $f^{r,\alpha,\beta}(\bm{z})\geq f_\epsilon^{r,\alpha,\beta}(\bm{z})-\delta$ for $\epsilon$ small enough and all $\bm{z}$ in the compact set $[\bm{0},\bm{\zeta}]$
. Hence for $\bm{z}\in\bigcap_{\epsilon>0}P_0(\epsilon)$, we obtain $f^{r,\alpha,\beta}(\bm{z})\geq-\delta$ for every $\delta>0$ and so $\bigcap_{\epsilon>0}P_0(\epsilon)\subseteq P$. However, since $\bigcap_{\epsilon>0}P_0(\epsilon)$ is the intersection of a chain of connected, compact sets in the Hausdorff space $\R^V$, it is itself a connected, compact set. Since further $\bm{0}\in\bigcap_{\epsilon>0}P_0(\epsilon)$, we thus derive that $\bigcap_{\epsilon>0}P_0(\epsilon)=P_0$. That is, $\lim_{\epsilon\to0+}\bm{z}^*(\epsilon)\in P_0$ and hence $\lim_{\epsilon\to0+}\bm{z}^*(\epsilon)=\bm{z}^*$. 
The same arguments show that $\lim_{h\to0+}\bm{z}^*(\epsilon+h)=\bm{z}^*(\epsilon)$ for every $\epsilon>0$, hence proving right-continuity of $\bm{z}^*(\epsilon)$.

A classical result for derivatives of monotone functions (see \cite[Theorem 7.21]{Wheeden1977} for instance) then yields the existence of $(\bm{z}^*)'$ almost everywhere and
\[ \bm{z}^*(\epsilon)-\bm{z}^*\geq \lim_{h\to0-}\bm{z}^*(\epsilon+h) - \lim_{h\to0+}\bm{z}^*(h) \geq \int_0^\epsilon(\bm{z}^*)'(\xi)\dd\xi. \qedhere\]
\end{proof}

\begin{proof}[Proof of Theorem \ref{3:thm:finitary:weights}]
As outlined above, in order to reduce this general setting to the special case from Proposition \ref{3:prop:finitary:weights}, we apply an additional small shock to the system. 
That is, if we can find a vector $\bm{v}(\epsilon)\in\R_+$ such that $D_{\bm{v}(\epsilon)}f_\epsilon^{r,\alpha,\beta}(\hat{\bm{z}}(\epsilon))<0$ for all $(r,\alpha,\beta)\in V$, then applying Proposition \ref{3:prop:finitary:weights}, we derive for the final damage $n^{-1}\mathcal{S}_n^\epsilon$ in the additionally shocked system
\[ n^{-1}\mathcal{S}_n^\epsilon \leq g_\epsilon(\hat{\bm{z}}(\epsilon)) + o_p(1) \leq g_\epsilon(\bm{z}^*(\epsilon)) + o_p(1) \leq \epsilon +  g(\bm{z}^*(\epsilon)) + o_p(1). \]
We then conclude $n^{-1}\mathcal{S}_n \leq g(\bm{z}^*) + o_p(1)$, as for arbitrary $\delta>0$ and $\epsilon
>0$ small enough,
\[ \P\left(n^{-1}\mathcal{S}_n - g(\bm{z}^*) > \delta\right)  \leq \P\left(n^{-1}\mathcal{S}_n^\epsilon - (\epsilon+g(\bm{z}^*(\epsilon)) > \delta/2\right) \to 0,\quad\text{as }n\to\infty. \]
So let us show the existence of the vectors $\bm{v}(\epsilon)$: By Lemma \ref{3:lem:z^*:epsilon} we know that $(\bm{z}^*)'(\epsilon)$ exists almost everywhere and that $\int_0^\epsilon(\bm{z}^*)'(\xi)\dd\xi \leq \bm{z}^*(\epsilon)-\bm{z}^* < \infty$. By the integrability, we can hence find a sequence $(\epsilon_n)_{n\in\N}\subset(0,1)$ such that $\epsilon_n\to0$ as $n\to\infty$ and for each $\epsilon_n$ it holds
\[ (\bm{z}^*)'(\epsilon_n)< \epsilon_n^{-1} \left(\pmb\zeta-\bm{z}^*-\delta\1\right) \]
componentwise, where 
$0<\delta<\zeta^{r,\alpha,\beta}-(z^*)^{r,\alpha,\beta}$ for all $(r,\alpha,\beta)\in V$ and \mbox{$\1=(1,\ldots,1)\in\R_+^V$.} This bound can be achieved with the same $\epsilon_n$ for each component, noting that the sum $\sum_{(r,\alpha,\beta)\in V}((z^*)^{r,\alpha,\beta})'(\xi)$ which is still integrable. With
\[ 0 = f_\epsilon^{r,\alpha,\beta}(\bm{z}^*(\epsilon)) = (1-\epsilon)f^{r,\alpha,\beta}(\bm{z}^*(\epsilon)) + \epsilon\left(\E\left[W^{+,r,\alpha}\1\{A=\beta\}\right] - (z^*)^{r,\alpha,\beta}(\epsilon)\right), \]
we then derive that
\begin{align*}
\frac{\dd}{\dd \epsilon}f^{r,\alpha,\beta}(\bm{z}^*(\epsilon))\big\vert_{\epsilon=\epsilon_n} 
&= -\frac{1}{(1-\epsilon_n)^2}\left(\zeta^{r,\alpha,\beta}-(z^*)^{r,\alpha,\beta}(\epsilon_n)\right) + \frac{\epsilon_n}{1-\epsilon_n}\frac{\dd}{\dd\epsilon}(z^*)^{r,\alpha,\beta}(\epsilon)\big\vert_{\epsilon=\epsilon_n}\\
&< -\frac{1}{1-\epsilon_n}\left(\zeta^{r,\alpha,\beta}-(z^*)^{r,\alpha,\beta}(\epsilon_n) - \epsilon_n\frac{\dd}{\dd\epsilon}(z^*)^{r,\alpha,\beta}(\epsilon)\big\vert_{\epsilon=\epsilon_n}\right)\\
&< -\frac{1}{1-\epsilon_n} \left((z^*)^{r,\alpha,\beta}+\delta-(z^*)^{r,\alpha,\beta}(\epsilon_n)\right)<0
\end{align*}
for $n$ large enough such that $(z^*)^{r,\alpha,\beta}(\epsilon_n)<(z^*)^{r,\alpha,\beta}+\delta$. On the other hand,
\[ \frac{\dd}{\dd \epsilon}f^{r,\alpha,\beta}(\bm{z}^*(\epsilon))\big\vert_{\epsilon=\epsilon_n} = D_{\bm{v}(\epsilon_n)}f^{r,\alpha,\beta}(\bm{z}^*(\epsilon_n)) \geq D_{\bm{v}(\epsilon_n)}f_{\epsilon_n}^{r,\alpha,\beta}(\bm{z}^*(\epsilon_n)), \]
where $\bm{v}(\epsilon_n):=(\bm{z}^*)'(\epsilon_n)$. Hence altogether,
\[ D_{\bm{v}(\epsilon_n)}f_{\epsilon_n}^{r,\alpha,\beta}(\bm{z}^*(\epsilon_n)) < 0. \]
In fact, it also holds that
\begin{align*}
v^{r,\alpha,\beta}(\epsilon_n)
&\geq \lim_{h\to0}\E\Bigg[W^{+,r,\alpha}\1\{A=\beta\}\P\Bigg(\sum_{s\in[R]}s\mathrm{Poi}\Bigg(\sum_{\gamma\in[T]}W^{-,s,\gamma}(z^*)^{s,\beta,\gamma}(\epsilon_n+h)\Bigg)< C\Bigg)\Bigg]\\
&\geq \E\Bigg[W^{+,r,\alpha}\1\{A=\beta\}\P\Bigg(\sum_{s\in[R]}s\mathrm{Poi}\Bigg(\sum_{\gamma\in[T]}W^{-,s,\gamma}\zeta^{s,\beta,\gamma}\Bigg)< C\Bigg)\Bigg] > 0,
\end{align*}
using that $f_{\epsilon_n}^{r,\alpha,\beta}(\bm{z}^*(\epsilon_n))=f_{\epsilon_n+h}^{r,\alpha,\beta}(\bm{z}^*(\epsilon_n+h))=0$.
The proof is hence finished, if $\hat{\bm{z}}(\epsilon_n)=\bm{z}^*(\epsilon_n)$. 
Otherwise, note the following: For each $\delta>0$ it holds $f_{\epsilon_n}^{r,\alpha,\beta}(\hat{\bm{z}}(\epsilon_n+\delta))<f_{\epsilon_n+\delta}^{r,\alpha,\beta}(\hat{\bm{z}}(\epsilon_n+\delta))=0$. By monotonicity of $f_{\epsilon_n}^{r,\alpha,\beta}$ from Lemma \ref{3:lem:properties:f}, we derive that \mbox{$\bm{z}^*(\epsilon_n)\leq\hat{\bm{z}}(\epsilon_n+\delta)\leq\bm{z}^*(\epsilon_n+\delta)$.} Hence as $\delta\to0$, using $\bm{z}^*(\epsilon_n+\delta)\to\bm{z}^*(\epsilon_n)$ by Lemma \ref{3:lem:z^*:epsilon}, we conclude that $\hat{\bm{z}}(\epsilon_n+\delta)\to\bm{z}^*(\epsilon_n)$ as $\delta\to0$. Thus we derive that for $\delta_n>0$ small enough, it holds $D_{\bm{v}(\epsilon_n)}f_{\epsilon_n+\delta_n}^{r,\alpha,\beta}(\hat{\bm{z}}(\epsilon_n+\delta_n))<0$ by continuity of $D_{\bm{v}}f_\epsilon^{r,\alpha,\beta}(\bm{z})$ w.\,r.\,t.~$\epsilon$ and $\bm{z}$. Hence apply Proposition \ref{3:prop:finitary:weights} to the financial systems 
additionally shocked by $\epsilon_n+\delta_n$ and choose vectors $\bm{v}(\epsilon_n)$ for the directional derivative.
\end{proof}

\subsection{Proof of Theorem \ref{3:thm:general:weights}}\label{3:ssec:proof:main:general}
In the previous section, we derived an explicit asymptotic expression for the final default fraction if in our model we choose vertex-weights only from a finite set. While this gives a first important insight into the behavior of large financial networks, it is not possible to model 
heavy tailed degree distributions as observed for real financial networks by bounded vertex-weights. Theorem \ref{3:thm:general:weights} hence extends Theorem \ref{3:thm:finitary:weights} to the case of general (non-finitary) regular vertex sequences.

The outline for the rest of this section is the following: We want to approximate the general regular vertex sequence by two sequences of finitary vertex sequences such that one of them describes a system that experiences less defaults and damage, and the other one experiences more defaults and damage. To this end, we first construct the corresponding limiting distribution functions $\{F_k^A\}_{k\in\N}$ respectively $\{F_k^B\}_{k\in\N}$ and then investigate the finitary systems with help of Theorem \ref{3:thm:finitary:weights}.

Let $D_\infty:=\left(\R_{+,0}^{[R]\times[T]}\right)^2\times\R_{+,0}\times\N_0\times[T]$ and for $(r,\alpha,\beta)\in V$, $(\bm{z},\bm{x},\bm{y},v,l,m)\in\R_{+,0}^V\times D_\infty$,
\begin{align*}
h_f^{r,\alpha,\beta}(\bm{z},\bm{x},\bm{y},v,l,m) &:= y^{r,\alpha}\psi_l\Bigg(\sum_{\gamma\in[T]}x^{1,\gamma}z^{1,\beta,\gamma},\ldots,\sum_{\gamma\in[T]}x^{R,\gamma}z^{R,\beta,\gamma}\Bigg) \1\{m=\beta\},\\
h_g(\bm{z},\bm{x},\bm{y},v,l,m) &:= v\sum_{\beta\in[T]}\psi_l\Bigg(\sum_{\gamma\in[T]}x^{1,\gamma}z^{1,\beta,\gamma},\ldots,\sum_{\gamma\in[T]}x^{R,\gamma}z^{R,\beta,\gamma}\Bigg)\1\{m=\beta\},
\end{align*}
where as before $\psi_l(x_1,\ldots,x_R) := \P(\sum_{r\in[R]}r\mathrm{Poi}(x_r)\geq l)$.
Note that although $D_\infty$ does not contain $\left(\R_{+,0}^{[R]\times[T]}\right)^2\times\R_{+,0}\times\{\infty\}\times[T]$, it holds
\[ f^{r,\alpha,\beta}(\bm{z})=\int_{D_\infty} h_f^{r,\alpha,\beta}(\bm{z},\bm{x},\bm{y},v,l,m) \dd F(\bm{x},\bm{y},v,l,m) - z^{r,\alpha,\beta} \]
and
\[ g(\bm{z})=\int_{D_\infty} h_g(\bm{z},\bm{x},\bm{y},v,l,m) \dd F(\bm{x},\bm{y},v,l,m), \]
for $F$ the limiting distribution of the weights, capital and type as given in Definition \ref{3:def:regular:vertex:sequence}.
This is because $\psi_\infty(x_1,\ldots,x_R)
=0$. Let then $Z:=[\bm{0},\pmb\zeta]$ 
and $H := \{h_g\}\cup\bigcup_{(r,\alpha,\beta)\in V}\{h_f^{r,\alpha,\beta}\}$.

As a first approximation of $F$, we choose the discretizations
\[ F_j^A(\bm{x},\bm{y},v,l,m) := F\bigg(\frac{\lceil j\bm{x}\rceil}{j}, \frac{\lceil j\bm{y}\rceil}{j},\frac{\lceil jv\rceil}{j},l,m\bigg),~ F_j^B(\bm{x},\bm{y},v,l,m) := F\bigg(\frac{\lfloor j\bm{x}\rfloor}{j}, \frac{\lfloor j\bm{y}\rfloor}{j},\frac{\lfloor jv\rfloor}{j},l,m\bigg) \]
for $j\in\N$, where $\lceil\cdot\rceil$ and $\lfloor\cdot\rfloor$ shall be applied componentwise. That is, the sequences $\{F_j^A\}_{j\in\N}$ and $\{F_j^B\}_{j\in\N}$ approximate $F$ from above respectively below and the approximations become finer as $j$ increases. Since every $h\in H$ is continuous in $\bm{z}$, $\bm{x}$, $\bm{y}$ and $v$, it is easy to obtain (cf.~\cite{Detering2015a}) that for each $k\in\N$ there exists $j_k$ large enough such that for all $j\geq j_k$ it holds
\[ \left\vert\int_{D_k} h(\bm{z},\bm{x},\bm{y},v,l,m)\dd F_j^{A,B}(\bm{x},\bm{y},v,l,m) - \int_{D_k} h(\bm{z},\bm{x},\bm{y},v,l,m)\dd F(\bm{x},\bm{y},v,l,m)\right\vert \leq k^{-1} \]
for all $\bm{z}\in Z$, where \mbox{$D_k:=\{(\bm{x},\bm{y},v,l,m)\,:\,x^{r,\alpha}\leq k,y^{r,\alpha}\leq k,v\leq k, l\leq k\}\subset D_\infty$.} We denote $\overline{F}_k^A:=F_{j_k}^A$ and $\overline{F}_k^B:=F_{j_k}^B$ in the following.

By construction, the distribution functions $\overline{F}_k^{A,B}$ clearly correspond to discrete weight and systemic importance sequences that can be obtained from the original regular vertex sequence by adjusting weights and systemic importance values upward respectively downward. However, $\overline{F}_k^{A,B}$ (potentially) still assigns mass to infinitely many weights and capitals. For the case of $\overline{F}_k^A$, we can overcome this issue by setting
\[ F_k^A(\bm{x},\bm{y},v,l,m) := \begin{cases}\overline{F}_k^A(\bm{x}\wedge k,\bm{y}\wedge k,v\wedge k,l\wedge k,m),&\text{if }l<\infty,\\1,&\text{else},\end{cases} \]
where $\cdot\wedge k$ denotes componentwise truncation at $k$. That is, if the capital, the systemic importance value or one of the weights of some bank in the system exceeds $k$ (we call this bank \emph{large} in the following), then in the approximating finitary system described by $F_k^A$, this bank's weights and systemic importance value are all set to $0$ and its capital is increased to $\infty$ (cf.~\cite{Detering2015a} for a rigorous definition of the approximating vertex sequences). Note that the type of each bank stays the same. Clearly, this further reduces defaults in the system and the corresponding systemic damage in the sense that if we couple the original system with the finitary approximating system, then the final systemic damage $n^{-1}(\mathcal{S}_k^A)_n$ is stochastically dominated by $n^{-1}\mathcal{S}_n$ for all $k\in\N$.

If we wanted to apply exactly the same idea also to $\overline{F}_k^B$, we would need to set all weights and systemic importance values of large banks to $\infty$, which is not possible by the definition of a finitary regular vertex sequence. Still it will be possible to adjust weights, systemic importance and capitals of large banks to finitely many values such that the final damage in the finitary approximating system stochastically dominates $n^{-1}\mathcal{S}_n$. To this end, let
\[ \gamma_k^\beta := \int_{D_k^c}\1\{m=\beta\}\dd F(\bm{x},\bm{y},v,l,m), \]
where $D_k^c:=D_\infty\backslash D_k$, and
\begin{align*}
(\overline{w}_k^\beta)^{r,\alpha} &:= \begin{cases}2\left(\gamma_k^\beta\right)^{-1}\displaystyle\int_{D_k^c}y^{r,\alpha}\1\{m=\beta\}\dd F(\bm{x},\bm{y},v,l,m)\geq 2k,&\text{if }\gamma_k^\beta>0,\\2k,&\text{if }\gamma_k^\beta=0,\end{cases}
\end{align*}
as well as
\begin{align*}
\overline{s}_k^\beta &:= \begin{cases}2\left(\gamma_k^\beta\right)^{-1}\displaystyle\int_{D_k^c}v\1\{m=\beta\}\dd F(\bm{x},\bm{y},v,l,m)\geq 2k,&\hspace*{0.4cm}\text{if }\gamma_k^\beta>0,\\2k,&\hspace*{0.4cm}\text{if }\gamma_k^\beta=0.\end{cases}
\end{align*}
Let then $F_k^B$ be given by $\overline{F}_k^B$ on $D_k$. By this definition we know that it holds $F_k^B(k,\ldots,k,\beta)=F(k,\ldots,k,\beta)$ for each $\beta\in[T]$. Moreover, let $F_k^B$ assign the remaining mass $\gamma_k^\beta$ to the points $(\bm{0},\overline{\bm{w}}_k^\beta,\overline{s}_k^\beta,0,\beta)$. That is, if a large bank of type $\beta$ originally has finite capital, then its approximated capital is set to $0$ (it initially defaults), its in-weights are set to $0$, its out-weights are set to $(\overline{w}_k^\beta)^{r,\alpha}$ and its systemic importance is set to $\overline{s}_k^\beta$ (again cf.~\cite{Detering2015a} for a rigorous definition of the approximating vertex sequences). As before, their type does not change. Finally, we assign the remaining mass $\P(A=\beta,C=\infty)$ to the points $(\bm{0},\bm{0},0,\infty,\beta)$ for each $\beta\in[T]$.

By construction, all large banks are initially defaulted in the approximating finitary system. Also all the weights of small banks are increased as compared to the original system. To show that there occurs more damage in the approximating system than in the original one (i.\,e.~$n^{-1}(\mathcal{S}_k^B)_n$ stochastically dominates $n^{-1}\mathcal{S}_n$), all that is left to show is that for each \mbox{$r\in[R]$} the total $r$-out-weight of large $\beta$-type banks with respect to each type $\alpha\in[T]$ in the approximating system is larger than in the original one. But the total $r$-out-weight of large $\beta$-banks with respect to type $\alpha$ is given by
\[ n(\overline{w}_k^\beta)^{r,\alpha}\left(\gamma_k^\beta+o(1)\right) = 2n \int_{D_k^c}y^{r,\alpha}\1\{m=\beta\}\dd F(\bm{x},\bm{y},v,l,m) (1+o(1)) \]
in the approximating system, whereas for the original system it is
\[ n\int_{D_k^c}y^{r,\alpha}\1\{m=\beta\}\dd F(\bm{x},\bm{y},v,l,m) (1+o(1)). \]
Hence for each small bank $i\in[n]$ the number of incoming $r$-edges from large banks in the original system is stochastically dominated by the corresponding number in the approximating system (for more details see \cite{Detering2015a}). In particular, the total exposure of $i$ to the set of large banks (the weighted sum of incoming edges) is stochastically dominated. This shows the following:
\begin{lemma}\label{3:lem:stochastic:domination}
Consider a regular vertex sequence and let sequences $\{F_k^A\}$ and $\{F_k^B\}$ be constructed as above. Further, let $(\mathcal{S}_k^A)_n$ and $(\mathcal{S}_k^B)_n$ be the total damage caused by finally defaulted banks in the finitary approximating systems. Then with $\preceq$ denoting stochastic domination it holds that
\[ n^{-1}\left(\mathcal{S}_k^A\right)_n \preceq n^{-1}\mathcal{S}_n \preceq n^{-1}\left(\mathcal{S}_k^B\right)_n. \]
\end{lemma}
We have hence bounded the final damage $n^{-1}\mathcal{S}_n$ from below and from above using finitary approximations. We now want to compute the precise final damages for these approximating systems using Theorem \ref{3:thm:finitary:weights}. Let
\begin{gather*}
\left(f_k^{A,B}\right)^{r,\alpha,\beta}(\bm{z}) = \int_{D_\infty} h_f^{r,\alpha,\beta}(\bm{z},\bm{x},\bm{y},v,l,m)\dd F_k^{A,B}(\bm{x},\bm{y},v,l,m) - z^{r,\alpha,\beta},\\
g_k^{A,B}(\bm{z}) = \int_{D_\infty} h_g(\bm{z},\bm{x},\bm{y},v,l,m)\dd F_k^{A,B}(\bm{x},\bm{y},v,l,m)
\end{gather*}
the corresponding analogues of $f^{r,\alpha,\beta}$ and $g$. Further, denote by $\hat{\bm{z}}_k^A$ and $(\bm{z}^*)_k^A$ resp.~$\hat{\bm{z}}_k^B$ and $(\bm{z}^*)_k^B$ the smallest and largest joint roots of all functions $\left(f_k^A\right)^{r,\alpha,\beta}$ resp.~$\left(f_k^B\right)^{r,\alpha,\beta}$, $(r,\alpha,\beta)\in V$. Then we derive the following result comparing these quantities to the original system:
\begin{lemma}\label{3:lem:convergence:g}
It holds $\liminf_{k\to\infty} g_k^A\left(\hat{\bm{z}}_k^A\right) \geq g(\hat{\bm{z}})$ and $\limsup_{k\to\infty} g_k^B\left(\left(\bm{z}^*\right)_k^B\right) \leq g(\bm{z}^*)$.
\end{lemma}
\begin{proof}
First note that uniformly for all $\bm{z}\in Z$ and $h\in H$ it holds
\begin{align*}
&\left\vert \int_{D_k} h(\bm{z},\bm{x},\bm{y},v,l,m) \dd F_k^{A,B}(\bm{x},\bm{y},v,l,m) - \int_{D_k} h(\bm{z},\bm{x},\bm{y},v,l,m) \dd F(\bm{x},\bm{y},v,l,m)\right\vert\\
&\hspace{1.7cm}= \left\vert \int_{D_k} h(\bm{z},\bm{x},\bm{y},v,l,m) \dd \overline{F}_k^{A,B}(\bm{x},\bm{y},v,l,m) - \int_{D_k} h(\bm{z},\bm{x},\bm{y},v,l,m) \dd F(\bm{x},\bm{y},v,l,m)\right\vert\\
&\hspace{1.7cm}\leq k^{-1} \to 0,\quad\text{as }k\to\infty.
\end{align*}
Since further $\int_{D_k^c}v\dd F\to0$ and $\int_{D_k^c}y^{r,\alpha}\1\{m=\beta\}\dd F\to0$, as $k\to\infty$, and each $h\in H$ is bounded by the integrands $v$ or $y^{r,\alpha}\1\{m=\beta\}$, it holds uniformly for all $\bm{z}\in Z$ and $h\in H$ that $\int_{D_k^c}h\dd F\to0$.  Together with $\int_{D_k^c}h\dd F_k^A=0$, this implies
\begin{equation}\label{3:eqn:strong:G:convergence:A}
\int_{D_\infty}h(\bm{z},\bm{x},\bm{y},v,l,m)\dd F_k^A(\bm{x},\bm{y},v,l,m) - \int_{D_\infty} h(\bm{z},\bm{x},\bm{y},v,l,m)\dd F(\bm{x},\bm{y},v,l,m) = o(1)
\end{equation}
uniformly for all $\bm{z}\in Z$ and $h\in H$.

For $\{F_k^B\}_{k\in\N}$, we further need to consider the term $\int_{D_k^c}h(\bm{z},\bm{x},\bm{y},v,l,m)\dd F_k^B(\bm{x},\bm{y},v,l,m)$. Thus
\begin{align*}
\int_{D_k^c} v\, \dd F_k^B(\bm{x},\bm{y},v,l,m) &= \sum_{\beta\in[T]}\overline{s}_k^\beta \gamma_k^\beta, & \int_{D_k^c}y^{r,\alpha}\1\{m=\beta\}\dd F_k^B(\bm{x},\bm{y},v,l,m) &= (\overline{w}_k^\beta)^{r,\alpha} \gamma_k^\beta.
\end{align*}
All these quantities tend to $0$ as $k\to\infty$ (note that $(\overline{w}_k^\beta)^{r,\alpha}\gamma_k^\beta=2\int_{D_k^c}y^{r,\alpha}\1\{m=\beta\}\dd F$ if $\gamma_k^\beta>0$). Since each function $h\in H$ is bounded by one of the (finitely many) integrands from above, this implies that $\int_{D_k^c}h(\bm{z},\bm{x},\bm{y},v,l,m)\dd F_k^B(\bm{x},\bm{y},v,l,m) \to 0$, as $k\to\infty$, uniformly for all $\bm{z}\in Z$ and $h\in H$. Therefore we can conclude that also uniformly for all $\bm{z}\in Z$ and $h\in H$,
\begin{equation}\label{3:eqn:strong:G:convergence:B}
\int_{D_\infty}h(\bm{z},\bm{x},\bm{y},v,l,m)\dd F_k^B(\bm{x},\bm{y},v,l,m) - \int_{D_\infty} h(\bm{z},\bm{x},\bm{y},v,l,m)\dd F(\bm{x},\bm{y},v,l,m) = o(1).
\end{equation}
We now turn to the proof of the first statement: Let $\epsilon>0$ and define
\[ D_\epsilon:=\bigcap_{(r,\alpha,\beta)\in V}\{\bm{z}\in\R_{+,0}^V\,:\,f^{r,\alpha,\beta}(\bm{z})\in[0,\epsilon]\}. \]
Further let $\bm{z}_\epsilon\in\R_{+,0}^V$ be defined by $z_\epsilon^{r,\alpha,\beta}:=\inf_{\bm{z}\in D_\epsilon}z^{r,\alpha,\beta}$, $(r,\alpha,\beta)\in V$. Then in particular, $\bm{z}_\epsilon\leq\hat{\bm{z}}$ componentwise since $\hat{\bm{z}}\in D_\epsilon$. Further, $\bm{z}_\epsilon$ is clearly increasing componentwise as \mbox{$\epsilon\to0$.} Hence the limit $\tilde{\bm{z}}:=\lim_{\epsilon\to0}\bm{z}_\epsilon\leq\hat{\bm{z}}$ exists. 
Now note that for fixed $(r,\alpha,\beta)\in V$, by definition of $\bm{z}_\epsilon$, we find a sequence $(\bm{z}_n)_{n\in\N}\subset D_\epsilon$ such that $\lim_{n\to\infty}z_n^{r,\alpha,\beta}=z_\epsilon^{r,\alpha,\beta}$ and $\bm{z}_n\geq\bm{z}_\epsilon$ componentwise. By monotonicity and uniform continuity of $f^{r,\alpha,\beta}$ on $D_\epsilon$, we then get
\[ f^{r,\alpha,\beta}(\bm{z}_\epsilon) \leq f^{r,\alpha,\beta}(z_n^{1,1,1},\ldots,z_\epsilon^{r,\alpha,\beta},\ldots,z_n^{R,T,T}) = f^{r,\alpha,\beta}(\bm{z}_n)+o(n)\leq\epsilon + o(n) \]
and hence $f^{r,\alpha,\beta}(\bm{z}_\epsilon)\leq\epsilon$. Again by continuity of $f^{r,\alpha,\beta}$, we obtain $f^{r,\alpha,\beta}(\tilde{\bm{z}}) = \lim_{\epsilon\to0}f^{r,\alpha,\beta}(\bm{z}_\epsilon) \leq \lim_{\epsilon\to0}\epsilon=0$. Replacing $\hat{\bm{z}}$ by $\tilde{\bm{z}}$ in the proof of Lemma \ref{3:lem:existence:hatz}, we now get the existence of a joint root $\bar{\bm{z}}\leq\tilde{\bm{z}}$ of all the functions $f^{r,\alpha,\beta}$, $(r,\alpha,\beta)\in V$. Since $\hat{\bm{z}}$ is the smallest joint root by definition, it thus follows that $\tilde{\bm{z}}=\hat{\bm{z}}$. Now note that by \eqref{3:eqn:strong:G:convergence:A} for $k$ large enough we derive that $(f_k^A)^{r,\alpha,\beta}(\bm{z})\geq f^{r,\alpha,\beta}(\bm{z})-\epsilon$ for all $\bm{z}\in Z$. Further, by construction of $F_k^A$, it holds that $(f_k^A)^{r,\alpha,\beta}(\bm{z})\leq f^{r,\alpha,\beta}(\bm{z})$. In particular, we can conclude that $\hat{\bm{z}}_k^A\in D_\epsilon$ for $k$ large enough and hence $\hat{\bm{z}}_k^A\geq \bm{z}_\epsilon$. Thus, for each $\epsilon>0$ by \eqref{3:eqn:strong:G:convergence:A} we derive 
$\liminf_{k\to\infty}g_k^A(\hat{\bm{z}}_k^A) \geq \lim_{k\to\infty}g_k^A(\bm{z}_\epsilon) = g(\bm{z}_\epsilon)$. Finally, using continuity of $g$ and $\lim_{\epsilon\to0}\bm{z}_\epsilon=\hat{\bm{z}}$, we get the first statement:
\[ \liminf_{k\to\infty}g_k^A(\hat{\bm{z}}_k^A) \geq g(\hat{\bm{z}}) \]
If now as in the proof of Theorem \ref{3:thm:finitary:weights} $\bm{z}^*(\epsilon)$ is the largest joint root of the additionally shocked system, we derive by \eqref{3:eqn:strong:G:convergence:B} that for $k$ large enough it holds $\left(f_k^B\right)^{r,\alpha,\beta}(\bm{z}^*(\epsilon))\leq f^{r,\alpha,\beta}(\bm{z}^*(\epsilon))/2<0$ for all $(r,\alpha,\beta)\in V$ and hence $(\bm{z}^*)_k^B\leq\bm{z}^*(\epsilon)$ componentwise. (Assume $\E[W^{+,r,\alpha}\1\{A=\beta\}]>0$ for all $(r,\alpha,\beta)\in V$ such that $f^{r,\alpha,\beta}(\bm{z}^*(\epsilon))<0$. Otherwise, we can simply leave out the coordinates $z^{r,\alpha,\beta}$ in all the proof since the $(r,\alpha,\beta)$-coordinate of all joint roots will be $0$.) Again by \eqref{3:eqn:strong:G:convergence:B}, we then derive 
\[ \limsup_{k\to\infty}g_k^B\left((\bm{z}^*)_k^B\right) \leq \lim_{k\to\infty} g_k^B(\bm{z}^*(\epsilon)) = g(\bm{z}^*(\epsilon)) \]
and by letting $\epsilon\to0$,
\[ \limsup_{k\to\infty} g_k^B\left((\bm{z}^*)_k^B\right) \leq g(\bm{z}^*). \qedhere\]
\end{proof}

\begin{proof}[Proof of Theorem \ref{3:thm:general:weights}]
Let $\epsilon>0$. By Lemma \ref{3:lem:stochastic:domination}, we obtain
\[ \P\left(n^{-1}\mathcal{S}_n-g(\hat{\bm{z}})<-\epsilon\right)\leq \P\left(n^{-1}\left(\mathcal{S}_k^A\right)_n-g(\hat{\bm{z}})<-\epsilon\right). \]
Further, by Lemma \ref{3:lem:convergence:g}, for $k$ large enough, we have $g_k^A(\hat{z}_k^A)>g(\hat{\bm{z}})-\epsilon/2$ and hence
\[ \P\left(n^{-1}\mathcal{S}_n-g(\hat{\bm{z}})<-\epsilon\right) \leq \P\left(n^{-1}\left(\mathcal{S}_k^A\right)_n-g_k^A(\hat{\bm{z}}_k^A)<-\epsilon/2\right). \]
Applying Theorem \ref{3:thm:finitary:weights} to the finitary system, as $n\to\infty$ we derive \mbox{$\P\left(n^{-1}\mathcal{S}_n-g(\hat{\bm{z}}\right)<-\epsilon)\to0$,} which shows the first part of the theorem.

Similarly, for the second part, by Lemma \ref{3:lem:stochastic:domination}
\[ \P\left(n^{-1}\mathcal{S}_n-g(\bm{z}^*)>\epsilon\right) \leq \P\left(n^{-1}\left(\mathcal{S}_k^B\right)_n-g(\bm{z}^*)>\epsilon\right) \]
and by Lemma \ref{3:lem:convergence:g}, for $k$ large enough it holds that $g_k^B((\bm{z}^*)_k^B)<g(\bm{z}^*)+\epsilon/2$. Hence an application of Theorem \ref{3:thm:finitary:weights} yields that
\begin{align*}
\P\left(n^{-1}\mathcal{S}_n-g(\bm{z}^*)>\epsilon\right) &\leq \P\left(n^{-1}\left(\mathcal{S}_k^B\right)_n-g_k^B((\bm{z}^*)_k^B)>\frac{\epsilon}{2}\right) \to 0,\quad\text{as }n\to\infty. \qedhere
\end{align*}
\end{proof}

\subsection{Proofs for Section \ref{3:sec:resilience}}\label{3:ssec:proofs:resilience}
\begin{proof}[Proof of Theorem \ref{3:thm:resilience}]
Let $\gamma\in(0,1)$ and define
\[ (f^\gamma)^{r,\alpha,\beta}(\bm{z}):=(1-\gamma)f^{r,\alpha,\beta}(\bm{z})+\gamma\left(\zeta^{r,\alpha,\beta}-z\right). \]
Further, let
\[ P^\gamma:=\bigcap_{(r,\alpha,\beta)\in V}\{\bm{z}\in\R_{+,0}^V\,:\,(f^\gamma)^{r,\alpha,\beta}(\bm{z})\geq0\} \]
and denote by $P_0^\gamma$ the largest connected component of $P^\gamma$ containing $\bm{0}$. Finally, define $\bm{z}^*(\gamma)$ by
\[ (z^*)^{r,\alpha,\beta}(\gamma):=\sup_{\bm{z}\in P_0^\gamma}z^{r,\alpha,\beta}. \]
By the proof of Lemma \ref{3:lem:z^*:epsilon} we know that $\bm{z}^*(\gamma)\to\bm{0}$, as $\gamma\to0$, and hence also $g(\bm{z}^*(\gamma))\to0$, using continuity of $g$. Choose now $\gamma>0$ small enough such that $g(\bm{z}^*(\gamma))\leq\epsilon/3$ and $\delta>0$ small enough such that $(f^M)^{r,\alpha,\beta}(\bm{z})<(f^\gamma)^{r,\alpha,\beta}(\bm{z})$ uniformly for all $\bm{0}\leq\bm{z}\leq\bm{\zeta}$ componentwise and $\P(M=0)<\delta$. 
Then in particular $(\bm{z}^*)^M\leq\bm{z}^*(\gamma)$ and $g((\bm{z}^*)^M)\leq g(\bm{z}^*(\gamma)) \leq\epsilon/3$.

If we now possibly decrease $\delta$ such that $\delta\leq\epsilon/3$, then by Theorem \ref{3:thm:general:weights}, we derive for the final damage caused by defaulted banks in the shocked system $n^{-1}\mathcal{S}_n^M$ that w.\,h.\,p.
\[ n^{-1}\mathcal{S}_n^M \leq g^M((\bm{z}^*)^M) + \epsilon/3 \leq g((\bm{z}^*)^M) + 2\epsilon/3 \leq \epsilon. \qedhere \]
\end{proof}

\begin{proof}[Proof of Lemma \ref{3:lem:z0}]
Let $P_0(\epsilon,I)$ denote the largest connected subset of
\[ P(\epsilon,I) := \bigcap_{(r,\alpha,\beta)\in V}\left\{\bm{z}\in\R_{+,0}^V\,:\,f^{r,\alpha,\beta}(\bm{z})\geq -\epsilon\1\{(r,\alpha,\beta)\in I\}\right\} \]
containing $\bm{0}$. Then by replacing $P_0$ in the proof of Lemma~\ref{3:lem:existence:hatz} with $P_0(\epsilon,I)$, we obtain existence of a smallest (componentwise) point $\hat{\bm{z}}(\epsilon,I)\in\R_{+,0}^V$ such that $f^{r,\alpha,\beta}(\hat{\bm{z}}(\epsilon,I))=-\epsilon$ for $(r,\alpha,\beta)\in I$ and $f^{r\alpha,\beta}(\hat{\bm{z}}(\epsilon,I))=0$ for $(r,\alpha,\beta)\in V\backslash I$. 
In particular, $\hat{\bm{z}}(\epsilon,I)\in P_0(\epsilon,I)$. Let now
\[ T(\epsilon,I) := \bigcap_{(r,\alpha,\beta)\in V}\left\{\bm{z}\in\R_{+,0}^V\,:\,f^{r,\alpha,\beta}(\bm{z})\leq -\epsilon\1\{(r,\alpha,\beta)\in I\}\right\}. \]
Then clearly $\hat{\bm{z}}(\epsilon,I)\in T(\epsilon,I)$. Further, in the proof of the existence of $\hat{\bm{z}}(\epsilon,I)$ we can 
use any upper bound $\bm{z}\in T(\epsilon,I)$, which shows that $\hat{\bm{z}}(\epsilon,I)\leq\bm{z}$ componentwise. In particular, 
$\hat{\bm{z}}(\epsilon,I)$ is monotone in $\epsilon$ and therefore $\tilde{\bm{z}}(I):=\lim_{\epsilon\to0+}\hat{\bm{z}}(\epsilon,I)$ exists.

Let now $\bar{\bm{z}}\in T(I)$ arbitrary. Then there exists a sequence $(\bm{z}_k)_{k\in\N}\subset\R_{+,0}^V$ with $f^{r,\alpha,\beta}(\bm{z}_k)<0$ for $(r,\alpha,\beta)\in I$ respectively $f^{r,\alpha,\beta}(\bm{z}_k)\leq0$ for $(r,\alpha,\beta)\in V\backslash I$ such that $\lim_{k\to\infty}\bm{z}_k=\bar{\bm{z}}$. By finiteness of $I$, we can then find $\epsilon_k>0$ such that $f^{r,\alpha,\beta}(\bm{z}_k)\leq-\epsilon_k\1\{(r,\alpha,\beta)\in I\}$ for any $(r,\alpha,\beta)\in V$ and $k\in\N$. In particular, $\bm{z}_k\in T(\epsilon_k,I)$ and hence $\bm{z}_k\geq \hat{\bm{z}}(\epsilon_k,I)\geq\tilde{\bm{z}}$. As $k\to\infty$, we can thus conclude that $\tilde{\bm{z}}\leq\bar{\bm{z}}$ for any $\bar{\bm{z}}\in T(I)$ and hence $\tilde{\bm{z}}\leq\bm{z}_0(I)$. On the other hand, $\tilde{\bm{z}}(I)\in T(I)$ by definition and therefore $\tilde{\bm{z}}(I)=\bm{z}_0(I)$.

Finally, note that
\[ \bm{z}_0(I)=\lim_{\epsilon\to0+}\hat{\bm{z}}(\epsilon,I) \in \bigcap_{\epsilon>0}P_0(\epsilon,I)=P_0, \]
where the last equality follows from $\bigcap_{\epsilon>0}P_0(\epsilon,I)\subset P$ and that $\bigcap_{\epsilon>0}P_0(\epsilon,I)$ must be a connected set containing $\bm{0}$ since $P_0(\epsilon,I)$ is a chain of connected, compact sets containing $\bm{0}$.
\end{proof}

\begin{proof}[Proof of Theorem \ref{3:thm:non-resilience}]
Let $\hat{\bm{z}}^M$ denote the analogue of $\hat{\bm{z}}$ for the ex post shocked system. Then 
\begin{align*}
&f^{r,\alpha,\beta}(\hat{\bm{z}}^M) + \E\Bigg[W^{+,r,\alpha}\P\Bigg(\sum_{s\in[R]}s\mathrm{Poi}\Bigg(\sum_{\gamma\in[T]}W^{-,s,\gamma}(\hat{z}^M)^{s,\beta,\gamma}\Bigg)\leq C-1\Bigg)\1\{A=\beta\}\1\{M=0\}\Bigg]\\
&= (f^M)^{r,\alpha,\beta}(\hat{\bm{z}}^M) = 0
\end{align*}
and hence $f^{r,\alpha,\beta}(\hat{\bm{z}}^M)\leq0$ with equality if and only if $\E[W^{+,r,\alpha}\1\{A=\beta\}\1\{M=0\}]=0$. Define now
\[ \epsilon:=-\max_{(r,\alpha,\beta)\in I}f^{r,\alpha,\beta}(\hat{\bm{z}}^M)>0, \]
so that $\hat{\bm{z}}^M\in T(\epsilon,I)$, where $T(\epsilon,I)$ as in the proof of Lemma \ref{3:lem:z0}. In the construction of $\hat{\bm{z}}(\epsilon,I)$ (see Lemma \ref{3:lem:z0}), we can then use the upper bound $\hat{\bm{z}}^M$ and obtain that 
$\hat{\bm{z}}(\epsilon,I)\leq\hat{\bm{z}}^M$ and hence $\hat{\bm{z}}^M\geq\bm{z}_0(I)$. 
We can then apply Theorem \ref{3:thm:general:weights} to conclude that
\[ \lim_{n\to\infty}\P\left(n^{-1}\mathcal{S}_n^M < g(\bm{z}_0(I))-\epsilon\right) \leq \lim_{n\to\infty} \P\left(n^{-1}\mathcal{S}_n^M < g\left(\hat{\bm{z}}^M\right) - \epsilon\right) = 0 \]
and hence $n^{-1}\mathcal{S}_n^M \geq g(\bm{z}_0(I))-\epsilon$ w.\,h.\,p.
\end{proof}

\begin{proof}[Proof of Lemma \ref{3:lem:z0:equals:z*}]
By Lemma \ref{3:lem:z0} clearly $\bm{z}_0(\tilde{V})\leq\bm{z}^*$. Assume now that $\bm{z}_0(\tilde{V})\lneq\bm{z}^*$. Then for some $(r,\alpha,\beta)\in\tilde{V}$ it must hold $z_0^{r,\alpha,\beta}(\tilde{V})<(z^*)^{r,\alpha,\beta}$ and by the construction of $\bm{z}_0(\tilde{V})$ in the proof of Lemma \ref{3:lem:z0} we can find $\epsilon>0$ such that $z_0^{r,\alpha,\beta}(\tilde{V})\leq\hat{z}^{r,\alpha,\beta}(\epsilon,\tilde{V})<(z^*)^{r,\alpha,\beta}$. Now by the definition of $\bm{z}^*$ and connectedness of $P_0$, we find a point $P_0\ni\tilde{\bm{z}}\leq\hat{\bm{z}}(\epsilon,\tilde{V})$ such that $\tilde{z}^{r,\alpha,\beta}=\hat{z}^{r,\alpha,\beta}(\epsilon,\tilde{V})$. But then $f^{r,\alpha,\beta}(\tilde{\bm{z}})\leq f^{r,\alpha,\beta}(\hat{\bm{z}}(\epsilon,\tilde{V}))=-\epsilon<0$ which contradicts $\tilde{\bm{z}}\in P_0$.
\end{proof}

\subsection{Proofs for Sections \ref{3:sec:capital:block:model} and \ref{3:sec:exposure:model}}
\begin{proof}[Proof of Proposition \ref{3:prop:limsup:liminf}]
Consider the case that $\nu^\beta<1$. We start with the upper bound and first derive a certain Chernoff bound. To this end, let $\{\lambda^s\}_{s\in[R]}\subset\R_{+,0}$ and $c\geq\sum_{s\in[R]}s\lambda^s$. Then using Markov's inequality, we compute for a sum of independent Poisson random variables and arbitrary $\theta\geq0$ that

\vspace*{-0.4cm}
\begin{align*}
\P\Bigg( \sum_{s\in[R]}s\mathrm{Poi}(\lambda^s)\geq c \Bigg) &= \P\Bigg( \exp\Bigg\{\theta \sum_{s\in[R]}s\mathrm{Poi}(\lambda^s)\Bigg\} \geq e^{\theta c} \Bigg) \leq e^{-\theta c} \prod_{s\in[R]} \E\Bigg[ \exp\Bigg\{ \theta s \mathrm{Poi}(\lambda^s) \Bigg\} \Bigg]\\
&= e^{-\theta c} \prod_{s\in[R]} \exp\left\{\lambda^s \left(e^{\theta s}-1\right)\right\} \leq e^{-\theta c} \prod_{s\in[R]}\exp\left\{ \frac{s\lambda^s}{R}\left( e^{\theta R}-1 \right) \right\}.
\end{align*}

\vspace*{-0.15cm}
\noindent This expression is minimized for $\theta^*=R^{-1}(\log c - \log(\sum_{s\in[R]}s\lambda^s))\geq0$ and thus
\begin{equation}\label{3:eqn:chernoff:bound}
\P\Bigg( \sum_{s\in[R]}s\mathrm{Poi}(\lambda^s)\geq c \Bigg) \leq \exp\Bigg\{ R^{-1}\Bigg( c \log\Bigg( \frac{e\sum_{t\in[R]}t\lambda^t}{c} \Bigg) - \sum_{s\in[R]}s\lambda^s \Bigg) \Bigg\}.
\end{equation}
Let now $\mu,h\in\R_+$, $0\leq\nu<1$ and $\{d^s\}_{s\in[R]}\subset\R_{+,0}$. Moreover, set $\lambda^s=h d^s$, $d=\sum_{s\in[R]}sd^s$ and $c=\lceil\mu d^\nu\rceil$. Then for $d \leq (h\mu^{-1})^\frac{1}{\nu-1}$ and $\omega(u):=u-\mu u^\nu \log(e \mu^{-1} u^{1-\nu})$,
\[ \P\Bigg( \sum_{s\in[R]}s\mathrm{Poi}(h d^s) \geq c \Bigg) \leq \exp\left\{ -h^\frac{\nu}{\nu-1} \omega\left( d z^\frac{1}{1-\nu} \right) R^{-1} \right\}. \]
Let now $\delta>0$ arbitrary. Then for $\tilde{d}$ large enough and $\tilde{d}<d\leq((1+\delta)eh\mu^{-1})^\frac{1}{\nu-1}$, we derive
\begin{equation}\label{3:eqn:inequ:1}
\P\Bigg( \sum_{s\in[R]}s\mathrm{Poi}(h d^s) \geq c \Bigg) \leq \left( \frac{edh}{\mu d^\nu} \right)^\frac{\mu d^\nu}{R} \leq \left( \frac{edh}{\mu d^\nu} \right)^2 (1+\delta)^{2-R^{-1}\mu d^\nu} \leq h^2 = o(h).
\end{equation}
Also for $((1+\delta)eh\mu^{-1})^\frac{1}{\nu-1}<d\leq((1+\delta)h\mu^{-1})^\frac{1}{\nu-1}$, we obtain
\begin{equation}\label{3:eqn:inequ:2}
\P\Bigg( \sum_{s\in[R]}s\mathrm{Poi}(h d^s) \geq c \Bigg) \leq \exp\left\{ -h^\frac{\nu}{\nu-1}k \right\} = o(h),
\end{equation}
where
\[ k:=\min\left\{\omega(u)\,:\,\left((1+\delta)e\mu^{-1}\right)^\frac{1}{\nu-1} \leq u \leq \left((1+\delta)\mu^{-1}\right)^\frac{1}{\nu-1}\right\}>0. \]
Finally, for $d\leq \tilde{d}$ and $c\geq R+1$, we compute
\begin{align}
\P\left( \sum_{s\in[R]}s\mathrm{Poi}(h d^s) \geq c \right) &\leq \P\left( \sum_{s\in[R]}s\mathrm{Poi}(h d^s) \geq R+1 \right) \leq \P\left( \mathrm{Poi}\left(h \sum_{s\in[R]} d^s\right) \geq \frac{R+1}{R} \right)\nonumber\\
&\leq \P\left( \mathrm{Poi}(hd)\geq2 \right) \leq h^2d^2 \leq h^2 \tilde{d}^2 = o(h).\label{3:eqn:inequ:3}
\end{align}
Combining \eqref{3:eqn:inequ:1}, \eqref{3:eqn:inequ:2} and \eqref{3:eqn:inequ:3}, we can thus conclude that
\begin{align*}
&\E\left[ W^{+,r,\alpha} \P\left( \sum_{s\in[R]}s\mathrm{Poi}\left( \sum_{\gamma\in[T]} W^{-,s,\gamma} v^{s,\beta,\gamma} h \right) \geq C(\bm{v}) \right) \1\{A=\beta\}\right.\\
&\hspace{7.7cm}\left. \times \1\left\{ E(\bm{v}) \leq \left( \frac{(1+\delta)h\Vert\bm{v}\Vert^{\nu^\beta}}{\mu^\beta} \right)^\frac{1}{\nu^\beta-1} \right\} \right] = o(h)
\end{align*}
and hence
\[ f^{r,\alpha,\beta}(h\bm{v}) \leq o(h) + \E\left[W^{+,r,\alpha} \1\{A=\beta\} \1\left\{ E(\bm{v}) \leq \left( \frac{(1+\delta)h\Vert\bm{v}\Vert^{\nu^\beta}}{\mu^\beta} \right)^\frac{1}{\nu^\beta-1} \right\} \right] - hv^{r,\alpha,\beta}. \]
In particular, this yields
\begin{align*}
&\limsup_{h\to0+}h^{-1}f^{r,\alpha,\beta}(h\bm{v})\\
&\hspace{1cm}\leq \limsup_{h\to0+}h^{-1}\E\left[ W^{+,r,\alpha}\1\{A=\beta\}\1\left\{ E(\bm{v}) > \left( \frac{(1+\delta)h \Vert\bm{v}\Vert^{\nu^\beta}}{\mu^\beta} \right)^\frac{1}{\nu^\beta-1} \right\} \right] - v^{r,\alpha,\beta}\\
&\hspace{1cm}= (1+\delta)\Vert\bm{v}\Vert \limsup_{h\to0+}h^{-1}\E\left[ W^{+,r,\alpha}\1\{A=\beta\}\1\left\{ \frac{E(\bm{v})}{\Vert\bm{v}\Vert} > \left( \frac{h}{\mu^\beta} \right)^\frac{1}{\nu^\beta-1} \right\} \right] - v^{r,\alpha,\beta}\\
\end{align*}
and letting $\delta\to0$ ends the proof of the upper bound.

For the lower bound, a similar calculation as above yields that for $c-1\leq\sum_{s\in[R]}s\lambda^s$ and $\theta\leq0$, it holds
\[ \P\left( \sum_{s\in[R]}s\mathrm{Poi}(\lambda^s) \leq c-1 \right) \leq e^{-\theta (c-1)} \prod_{s\in[R]}\exp\left\{ s\lambda^s \left( e^\theta-1 \right) \right\} \]
and minimizing over $\theta$ we derive
\[  \P\left( \sum_{s\in[R]}s\mathrm{Poi}(\lambda^s) \leq c-1 \right) \leq \exp\left\{ (c-1) \log\left( \frac{e\sum_{t\in[R]}t\lambda^t}{c-1} \right) - \sum_{s\in[R]}s\lambda^s \right\}. \]
Let now $\delta>0$ arbitrary. Then for $c-1\leq (1-\delta)\sum_{s\in[R]}s\lambda^s$ and $\sum_{s\in[R]}s\lambda^s$ large enough it holds
\[ \P\left( \sum_{s\in[R]}s\mathrm{Poi}(\lambda^s) \leq c-1 \right) \leq \exp\left\{ \left( (1-\delta)\log\left( \frac{e}{1-\delta} \right) - 1 \right) \sum_{s\in[R]}s\lambda^s \right\} \leq \delta. \]
In particular,
\begin{align*}
&\liminf_{h\to0+} h^{-1}f^{r,\alpha,\beta}(h\bm{v})\\
&\hspace{1cm} \geq (1-\delta) \liminf_{h\to0+}h^{-1}\E\left[ W^{+,r,\alpha}\1\{A=\beta\}\1\left\{ E(\bm{v}) > \left( \frac{(1-\delta)h \Vert\bm{v}\Vert^{\nu^\beta}}{\mu^\beta} \right)^\frac{1}{\nu^\beta-1} \right\} \right] - v^{r,\alpha,\beta}\\
&\hspace{1cm} = (1-\delta)^2 \Vert\bm{v}\Vert \liminf_{h\to0+}h^{-1}\E\left[ W^{+,r,\alpha}\1\{A=\beta\}\1\left\{ \frac{E(\bm{v})}{\Vert\bm{v}\Vert} > \left( \frac{h}{\mu^\beta} \right)^\frac{1}{\nu^\beta-1} \right\} \right] - v^{r,\alpha,\beta}
\end{align*}
and letting $\delta\to0$ ends the proof for $\nu^\beta<1$.

Consider now the case that $\nu^\beta\geq1$. Clearly, $\liminf_{h\to0}h^{-1}f^{r,\alpha,\beta}(h\bm{v})\geq -v^{r,\alpha,\beta}$ by the definition of $f^{r,\alpha,\beta}$. For the upper bound, let $\mu\in\R_+$, $\nu\geq1$, $\{d^s\}_{s\in[R]}\subset\R_{+,0}$, $\tilde{d}>(R/\mu)^{1/\nu}$ and $h\leq \mu\tilde{d}^{\nu-1}/e$. Moreover, set $\lambda^s=hd^s$, $d=\sum_{s\in[R]}sd^s$ and $c=\lceil\mu d^\nu\rceil$. Then for $d\geq\tilde{d}$, by \eqref{3:eqn:chernoff:bound} we derive
\[ \P\left( \sum_{s\in[R]}s\mathrm{Poi}(\lambda^s)\geq c \right) \leq \exp\left\{ R^{-1}\left(\mu d^\nu \log\left(\frac{ehd}{\mu d^\nu}\right) - hd \right) \right\}
\leq \left(\frac{eh}{\mu \tilde{d}^{\nu-1}}\right)^\frac{\mu \tilde{d}^\nu}{R}
= o(h). \]
Together with \eqref{3:eqn:inequ:3} for $d\leq\tilde{d}$, we conclude that for all $d\geq0$,
\[ \P\left( \sum_{s\in[R]}s\mathrm{Poi}(\lambda^s)\geq c \right) \leq \left(\frac{eh}{\mu \tilde{d}^{\nu-1}}\right)^{\mu \tilde{d}^\nu/R} + h^2\tilde{d}\,^2 \]
and hence $f^{r,\alpha,\beta}(h\bm{v})\leq o(h)-v^{r,\alpha,\beta}$. In particular, this yields
\[ \limsup_{h\to0}h^{-1}f^{r,\alpha,\beta}(h\bm{v})\leq -v^{r,\alpha,\beta} \]
and ends the proof.
\end{proof}

\begin{proof}[Proof of Proposition \ref{3:prop:directional:derivative}]
For $\nu^\beta\geq1$, this is just the result from Proposition \ref{3:prop:limsup:liminf}. Therefore, assume that $\nu^\beta<1$ in the following.

For any fixed $\delta>0$, we find $\tilde{u}\geq0$ large enough such that for $u\geq\tilde{u}$ it holds
\[ 1 - (1+\delta) \left(\frac{u}{w_\text{min}^{+,r,\alpha,\beta}}\right)^{1-k^{+,r,\alpha,\beta}} \leq F_{W^{+,r,\alpha}\vert_{A=\beta}}(u) \leq 1 - (1-\delta) \left(\frac{u}{w_\text{min}^{+,r,\alpha,\beta}}\right)^{1-k^{+,r,\alpha,\beta}} \]
and
\[ 1- (1+\delta) \left(\frac{u}{e_\text{min}^\beta(\bm{v})}\right)^{1-k^{-,\beta}} \leq F_{E(\bm{v})\vert_{A=\beta}}(u) \leq 1- (1-\delta) \left(\frac{u}{e_\text{min}^\beta(\bm{v})}\right)^{1-k^{-,\beta}}. \]
Thus for $h$ large enough
\begin{align*}
&\E\left[ W^{+,r,\alpha}\1\{A=\beta\}\1\left\{ \frac{E(\bm{v})}{\Vert\bm{v}\Vert}>\left( \frac{h}{\mu^\beta} \right)^\frac{1}{\nu^\beta-1} \right\} \right]\\
&\hspace{1cm}= \int_0^\infty \P\left( W^{+,r,\alpha}>y,A=\beta,E(\bm{v})>\Vert\bm{v}\Vert\left( \frac{h}{\mu^\beta} \right)^\frac{1}{\nu^\beta-1} \right)\,\mathrm{d}y\\
&\hspace{1cm}\leq \int_0^{\tilde{u}} \P\left( A=\beta,E(\bm{v})>\Vert\bm{v}\Vert\left( \frac{h}{\mu^\beta} \right)^\frac{1}{\nu^\beta-1} \right)\,\mathrm{d}y\\
&\hspace{2cm} + \int_{\tilde{u}}^\infty \P\left( F_{W^{+,r,\alpha}\vert_{A=\beta}}\left( W^{+,r,\alpha}\big\vert_{A=\beta} \right)>1-(1+\delta)\left(\frac{y}{w_\text{min}^{+,r,\alpha,\beta}}\right)^{1-k^{+,r,\alpha,\beta}},\right.\\
&\hspace{5.9cm}\left.  F_{E(\bm{v})\vert_{A=\beta}}\left( E(\bm{v})\big\vert_{A=\beta} \right)>1-(1+\delta)p(h)  \right)\mathrm{d}y\; \P(A=\beta)\\
&\hspace{1cm}\leq \int_0^{\tilde{u}} \P(A=\beta)\P\left( E(\bm{v})\big\vert_{A=\beta}>\Vert\bm{v}\Vert\left( \frac{h}{\mu^\beta} \right)^\frac{1}{\nu^\beta-1} \right)\,\mathrm{d}y\\
&\hspace{2cm} + \int_{\tilde{u}\theta^{-1}(h)}^\infty \P\left( F_{W^{+,r,\alpha}\vert_{A=\beta}}\left( W^{+,r,\alpha}\big\vert_{A=\beta} \right)>1-(1+\delta)p(h)x^{1-k^{+,r,\alpha,\beta}},\right.\\
&\hspace{5.9cm}\left.  F_{E(\bm{v})\vert_{A=\beta}}\left( E(\bm{v})\big\vert_{A=\beta} \right)>1-(1+\delta)p(h)  \right)\mathrm{d}x\; \theta(h) \P(A=\beta)  \\
&\hspace{1cm}\leq \tilde{u} (1+\delta)e_\text{min}^\beta(\bm{v})\Vert\bm{v}\Vert^{1-k^{-,\beta}}\left( \frac{h}{\mu^\beta} \right)^\frac{1-k^{-,\beta}}{\nu^\beta-1} \\
&\hspace{2cm} + \int_0^\infty \P\left( F_{W^{+,r,\alpha}\vert_{A=\beta}}\left( W^{+,r,\alpha}\big\vert_{A=\beta} \right)>1-(1+\delta)p(h)x^{1-k^{+,r,\alpha,\beta}}\,\right\vert\\
&\hspace{4.05cm}\left. F_{E(\bm{v})\vert_{A=\beta}}\left( E(\bm{v})\big\vert_{A=\beta} \right)>1-(1+\delta)p(h)  \right)\mathrm{d}x\; \theta(h)\P(A=\beta) (1+\delta)p(h)\\
&\hspace{1cm}= o(h) + \left( \int_0^\infty \Lambda^{r,\alpha,\beta}\left( x^{1-k^{+,r,\alpha,\beta}} \right)\,\mathrm{d}x + o(1) \right)\\
&\hspace{6.25cm} \times\P(A=\beta) w_\text{min}^{+,r,\alpha,\beta} \left( \frac{\Vert\bm{v}\Vert}{e_\text{min}^\beta(\bm{v})}\left( \frac{h}{\mu^\beta} \right)^\frac{1}{\nu^\beta-1} \right)^{\nu_c^{r,\alpha,\beta}-1} (1+\delta)
\end{align*}
where
\[ p(h) = \left( \frac{\Vert\bm{v}\Vert}{e_\text{min}^\beta(\bm{v})}\left( \frac{h}{\mu^\beta} \right)^\frac{1}{\nu^\beta-1} \right)^{1-k^{-,\beta}}, \]
we substituted $y=\theta(h)x$ with
\[ \theta(h) = w_\text{min}^{+,r,\alpha,\beta} \left( p(h) \right)^\frac{1}{1-k^{+,r,\alpha,\beta}}, \]
and in the last line we used dominated convergence noting that (bounding by comonotone dependence) $\Lambda^{r,\alpha,\beta}(x^{1-k^{+,r,\alpha,\beta}})\leq 1\wedge x^{1-k^{+,r,\alpha,\beta}}$ which is integrable as $k^{+,r,\alpha,\beta}>2$.

In particular, with Proposition \ref{3:prop:limsup:liminf}, we derive that $\limsup_{h\to0+}h^{-1}f^{r,\alpha,\beta}(h\bm{v}) \leq -v^{r,\alpha,\beta}$ for $1>\nu^\beta>\nu_c^{r,\alpha,\beta}$, and for $\nu^\beta=\nu_c^{r,\alpha,\beta}$ as $\delta\to0$,
\[ \limsup_{h\to0+}h^{-1}f^{r,\alpha,\beta}(h\bm{v}) \leq \left( \frac{\mu_c^{r,\alpha,\beta}}{\mu^\beta} - 1 \right) v^{r,\alpha,\beta} \]
For the lower bounds, similarly as above
\begin{align*}
&\E\left[ W^{+,r,\alpha}\1\{A=\beta\}\1\left\{ \frac{E(\bm{v})}{\Vert\bm{v}\Vert}>\left( \frac{h}{\mu^\beta} \right)^\frac{1}{\nu^\beta-1} \right\} \right]\\
&\hspace{1cm}\geq \int_{\tilde{u}\theta^{-1}(h)}^\infty \P\left( F_{W^{+,r,\alpha}\vert_{A=\beta}}\left( W^{+,r,\alpha}\big\vert_{A=\beta} \right)>1-(1-\delta)p(h)x^{1-k^{+,r,\alpha,\beta}}\,\right\vert\\
&\hspace{4.0cm}\left. F_{E(\bm{v})\vert_{A=\beta}}\left( E(\bm{v})\big\vert_{A=\beta} \right)>1-(1-\delta)p(h)  \right)\mathrm{d}x\; \theta(h) \P(A=\beta) (1-\delta)p(h) \\
&\hspace{1cm}\geq \left( \int_0^\infty \Lambda^{r,\alpha,\beta}\left( x^{1-k^{+,r,\alpha,\beta}} \right)\,\mathrm{d}x + o(1) \right)\\
&\hspace{2cm}\times \P(A=\beta) w_\text{min}^{+,r,\alpha,\beta} \left( \frac{\Vert\bm{v}\Vert}{e_\text{min}^\beta(\bm{v})}\left( \frac{h}{\mu^\beta} \right)^\frac{1}{\nu^\beta-1} \right)^{\nu_c^{r,\alpha,\beta}-1} (1-\delta),
\end{align*}
additionally noting that the lower integral bound $\tilde{u}\theta^{-1}(h)=o(1)$. Thus by Proposition \ref{3:prop:limsup:liminf}, we conclude that for $\nu^\beta<\nu_c^{r,\alpha,\beta}$,
\[ \liminf_{h\to0+}h^{-1}f^{r,\alpha,\beta}(h\bm{v}) = \infty,  \]
for $\nu^\beta=\nu_c^{r,\alpha,\beta}$ as $\delta\to0$,
\[ \liminf_{h\to0+}h^{-1}f^{r,\alpha,\beta}(h\bm{v}) \geq \left( \frac{\mu_c^{r,\alpha,\beta}}{\mu^\beta} - 1 \right) v^{r,\alpha,\beta} \]
and obviously for $\nu^\beta>\nu_c^{r,\alpha,\beta}$,
\[ \liminf_{h\to0+}h^{-1}f^{r,\alpha,\beta}(h\bm{v}) \geq -v^{r,\alpha,\beta}. \qedhere \]
\end{proof}

\begin{proof}[Proof of Theorem \ref{3:thm:amplification}]
Denote by $f_\epsilon^{r,\alpha,\beta}$, $(r,\alpha,\beta)\in V$, the functions for the system shocked by $M_\epsilon$. Then
\[ f_\epsilon^{r,\alpha,\beta}(h\bm{v}) \leq f^{r,\alpha,\beta}(h\bm{v}) + \E\left[ W^{+,r,\alpha}\1\{A=\beta\}\1\{M_\epsilon=0\} \right] \leq -\frac{h+o(h)}{B}v^{r,\alpha,\beta} + \epsilon v^{r,\alpha,\beta}. \]
In particular, choosing $h=(1+\delta)B\epsilon$ for some arbitrary $\delta>0$, we derive
\[ f_\epsilon^{r,\alpha,\beta}((1+\delta)B\epsilon\bm{v}) \leq (-\delta\epsilon+o(\epsilon))v^{r,\alpha,\beta}. \]
For $\epsilon>0$ small enough, this expression becomes negative and thus it holds $\bm{z}_\epsilon^*\leq(1+\delta)B\epsilon\bm{v}$. Finally, we can let $\delta\to0$ to finish the proof.
\end{proof}

\begin{proof}[Proof of Theorem \ref{3:exposures:final:fraction}]
The idea is similar as in Theorem \ref{2:thm:asymp:1}. We consider the contagion process in sequential form, i.\,e.~for steps $0\leq t\leq n-1$ we let
\begin{enumerate}[\alph*.]
\item $U(t)\subset[n]$ be the unexposed institutions at step $t$ with $U(0):=\{i\in[n]\,:\,c_i=0\}$,
\item $N(t)\subset[n]$ the solvent institutions at step $t$ with $N(0):=[n]\backslash U(0)$,
\item \label{3:sequential:c}the updated capitals $\{\tilde{c}_i(t)\}_{i\in[n]}$ with $\tilde{c}_i(0):= c_i$ for all $i\in[n]$
\end{enumerate}
and at step $t\in[n-1]$ we update those sets and quantities as follows:
\begin{enumerate}
\item \label{3:sequential:1}We choose an institution $v\in U(t-1)$ according to any rule.
\item \label{3:sequential:2}We expose $v$ to all of its solvent creditors in $N(t-1)$ by setting
\[ \tilde{c}_w(t):=\max\{0,\tilde{c}_w(t-1)-e_{v,w}\}. \]
\item \label{3:sequential:3}We let $N(t):=\{i\in N(t-1)\,:\,\tilde{c}_i(t)>0\}$ the new set of solvent institutions and $U(t):=(U(t-1)\backslash\{v\})\cup\{i\in N(t-1)\,:\,\tilde{c}_i(t)=0\}$ the new set of unexposed institutions.
\end{enumerate}
Then the rule chosen in Step \ref{3:sequential:1}.~determines an ordering for the defaults (exposition) of all institutions and thus also an ordering of the exposure list $\{E_i^{j,r}\}_{j\in[n]\backslash\{i\},r\in[R]}$ where unused exposures $E_i^{j,r}$ are moved to the end of this ordered list if $r>r(j,i)$. The ordering of the unused exposures can be chosen arbitrarily. Let then $\{E_i^s\}_{s\in[(n-1)R]}$ be the ordered enumeration of $i$'s exposure list. In particular, the random variable
\begin{equation}\label{3:eqn:proof:def:q}
p_i = \inf\left\{p\in\{0\}\cup[(n-1)R]\,:\,\sum_{1\leq s\leq p}E_i^s \geq c_i\right\}
\end{equation}
can be interpreted as a threshold in the sense of Section \ref{3:sec:default:fin} (there called capital) since equivalently to \ref{3:sequential:c}~we can keep track of updated thresholds $\tilde{p}_i(t)$ with $\tilde{p}_i(0)=p_i$ and in Steps \ref{3:sequential:2}.~and \ref{3:sequential:3}.~we set $\tilde{p}_w(t)=\max\{0,\tilde{p}_w(t-1)-r(v,w)\}$, $N(t)=\{i\in N(t-1)\,:\,\tilde{p}_i(t)>0\}$ and $U(t)=(U(t-1)\backslash\{v\})\cup\{i\in[n]\,:\,\tilde{p}_i(t)=0\}$. Conditioning on $\{p_i\}_{i\in[n]}$ we are thus in the setting of Section \ref{3:sec:asymptotic:results} -- note in particular that $p_i$ defined in \eqref{3:eqn:proof:def:q} has the same distribution as $q_i$ in Section \ref{3:sec:exposure:model} and thus by Assumption \ref{3:ass:regularity:thresholds} the conditioned threshold system almost surely defines a regular vertex sequence according to Definition \ref{3:def:regular:vertex:sequence} with a fixed deterministic limiting distribution function. An application of Theorem \ref{3:thm:general:weights} thus shows the desired result.
\end{proof}

\begin{proof}[Proof of Corollary \ref{3:cor:exposures:capital:requirements}]
By the proof of Proposition \ref{3:prop:directional:derivative} all we need to show is 
\begin{align*}
&\E\Bigg[W^{+,r,\alpha}\P\left( \sum_{s\in[R]}s\mathrm{Poi}\left( \sum_{\gamma\in[T]}W^{-,s,\gamma}hv^{s,\beta,\gamma} \right) \geq R+1 \right)\\
&\hspace{6cm} \times\1\left\{ Q\leq (1+\epsilon)\mu_c^\beta\left( \frac{E(\bm{v})}{\Vert\bm{v}\Vert} \right)^{\nu^\beta} \right\} \1\{A=\beta\}\Bigg] = o(h).
\end{align*}
But this follows by exactly the same means as in the proof of Theorem \ref{2:cor:threshold:res}, replacing $W^-$ by $\Vert\bm{v}\Vert^{-1}E(\bm{v})$ and noting that in the present setting
\[ \psi_{R+1}\left( \sum_{\gamma\in[T]}w^{-,1,\gamma}hv^{1,\beta,\gamma} , \ldots, \sum_{\gamma\in[T]}w^{-,R,\gamma}hv^{R,\beta,\gamma} \right) = o(h). \qedhere \]
\end{proof}

\cleardoublepage
\chapter{A Model for Fire Sales in Financial Networks}\label{chap:fire:sales}
In the previous chapters, we gave a detailed analysis of the contagion channel \emph{default contagion}. As the latest financial crisis revealed, however, system instability is driven by multiple different channels and \emph{fire sales} are of particular importance. In this chapter, we thus propose a model for the contagion effects due to distressed asset sales. We choose similar techniques as in the previous chapters and adopt in particular the asymptotic perspective which will allow us later in Chapter \ref{chap:fire:sales:default} to integrate our models into one. Similar as in Chapter \ref{chap:block:model} financial systems are described by a set of multi-dimensional functions. One of the biggest challenges, however, is connected to the fact that these functions can become discontinuous. We first derive results about the final default fraction and the final price impact in a system  hit by some initial shock in Section \ref{4:sec:fire:sales}, and then derive criteria for whether some initially unshocked system is vulnerable to small shocks as well as formulas for capital requirements sufficient to secure the system in Section \ref{4:sec:resilience}. In Section \ref{4:sec:applications}, we apply our theory to investigate positive and negative effects of asset diversification in financial systems and demonstrate the benefits of our systemic capital requirements compared to the classical risk management approach. Proofs are postponed to Section \ref{4:sec:proofs}.

\vspace*{-7pt}
\paragraph{My own contribution:} This chapter strongly resembles joint work with Nils Detering, Thilo Meyer-Brandis and Konstantinos Panagiotou \cite{Detering2018b}. I was significantly involved in the development of all parts of that paper and did much of the editorial work. In particular, I made major contributions to the design of the model and the contagion process, Lemma \ref{4:lem:final:state:cont:rho}, Proposition \ref{4:prop:final:state:bounds}, Lemma \ref{4:lem:existence:chi:hat}, Theorem \ref{4:thm:fire:sale:final:fraction}, Examples \ref{4:ex:cont} and \ref{4:ex:discont}, Theorem \ref{4:thm:resilience}, Corollary \ref{4:cor:resilience}, Theorem \ref{4:thm:non-resilience}, Corollaries  \ref{4:cor:non:resilience}, \ref{4:cor:one:asset:sales:default}, \ref{4:cor:intermediate:sales} and \ref{4:cor:multiple:assets}, Examples  \ref{4:ex:diversification} and \ref{4:ex:diversification:similarity}, Subsection \ref{4:ssec:initial:asset:shocks}, Remark \ref{4:rem:sequence:chi:*}, and Lemma \ref{4:lem:convergence:chi:hat}. 
Whereas in \cite{Detering2018b} the final default fraction was chosen as the systemic risk measure, here we include general systemic importance values as in the previous chapters.

\section{A Model for Fire Sales}\label{4:sec:fire:sales}
In this section, we define two models of fire sales. The first considers an explicitly given finite system of financial institutions and the second takes a stochastic, asymptotic perspective. We first describe the parameters and assumptions and then determine the final state of the system after the fire sales cascade has ended both in the deterministic as well as in the stochastic setting.

\vspace{-7pt}
\paragraph{Model parameters} We consider a financial system consisting of $n\in\N$ institutions which can invest in $M\in\N$ different (not perfectly liquid) assets or asset classes. That is, to each institution $i\in[n]:=\{1,\ldots,n\}$ we assign a number $x_i^m\in\R_{+,0}$ of held shares of asset $m\in[M]$ (or any other index set of size $M$). See Figure \ref{4:fig:bipartite} for an illustration. Further, we denote by $c_i\in\R_{+,\infty}:=\R_+\cup\{\infty\}$  the initial capital of institution $i$ (for example the equity for leveraged institutions or the portfolio value for institutions that exclusively invest in assets) and we assume that it incurs exogenous losses $\ell_i\in\R_{+,0}$ due to some shock event.
In the case of a market crash for instance, it could be that $\ell_i=\sum_{1\le m\le M} x_i^m\delta^mp^m$, where $p^m$ denotes the initial price of one share of asset $m\in[M]$ and $\delta^m\in(0,1]$ is the relative price shock on the asset. Finally, assign to each institution a value of systemic importance $s_i\in\R_{+,0}$ that measures the potential damage to economy and society in the case that institution $i$ defaults (see Chapter \ref{chap:systemic:risk} for more details). Let then the empirical distribution function $F_n:\R_{+,0}^{M}\times\R_{+,0}\times\R_{+,\infty}\times\R_{+,0}\to[0,1]$ of the institutions' parameters be denoted by
\begin{equation}
\label{4:eq:empDistFixedSize}
F_n(\bm{x},s,c,\ell) = n^{-1}\sum_{i\in[n]}\1\{x_i^1\leq x^1,\ldots,x_i^M\leq x^M, s_i\leq s, c_i\leq c, \ell_i\leq \ell\}
\end{equation}
and let in the following $(\bm{X}_n,S_n,C_n,L_n)$ be a random vector with distribution $F_n$.

\begin{figure}
\centering
\includegraphics[width=0.75\textwidth]{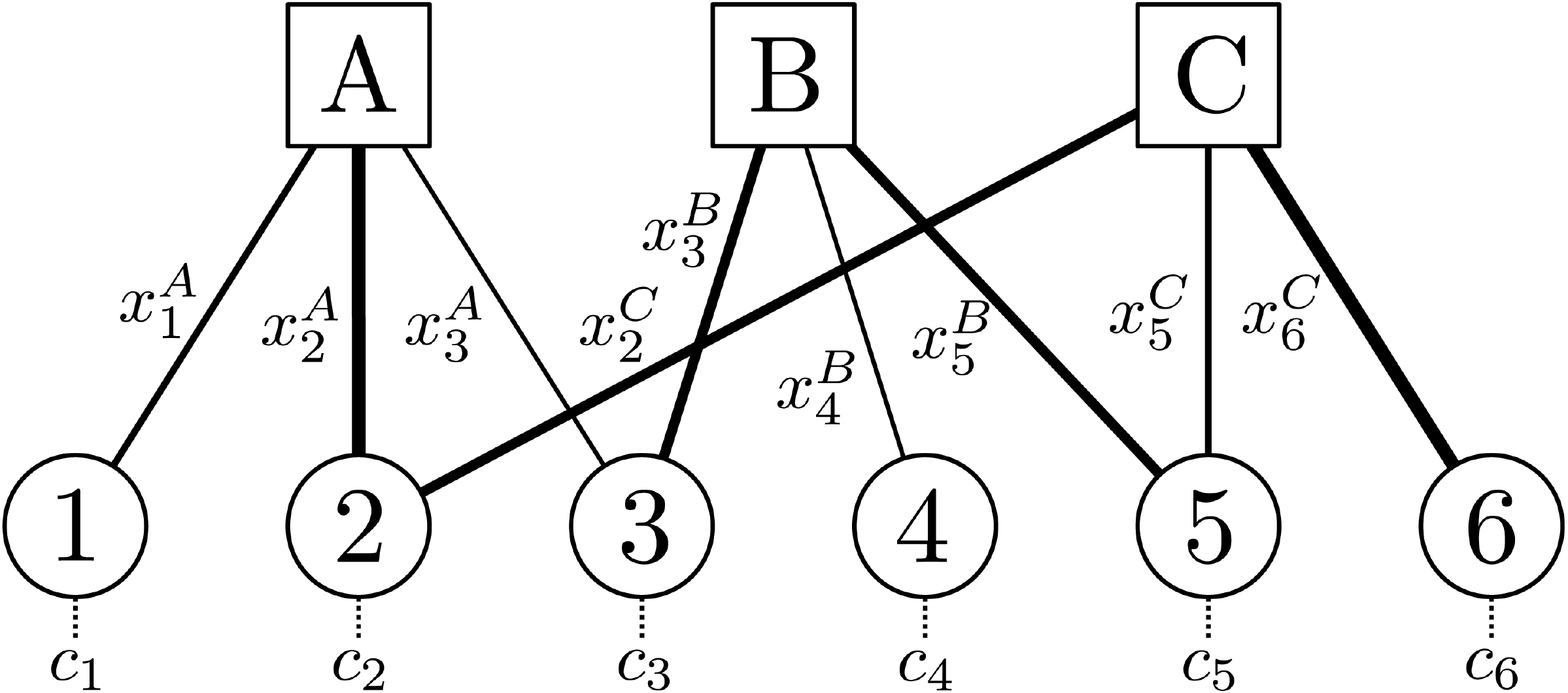} 
\caption{An illustration of a system with $n=6$ financial institutions (circles) and $M=3$ assets (squares). Edges represent investments of the institutions in the assets and their thickness indicates the investment volume $x_i^m$. Furthermore, capitals $c_i$ are attached to the institutions.}\label{4:fig:bipartite}
\end{figure}

\vspace{-7pt}
\paragraph{Asset sales} We assume that due to the exogenous losses some of the institutions are forced to liquidate parts of their asset holdings in order to comply with regulatory or market-imposed constraints (e.\,g.~leverage constraints), self-imposed risk preferences and policies to adjust the portfolio size, or to react to investor redemption. These sales are described by a non-decreasing
function $\rho:\R_{+,0}\to[0,1]$ such that each institution $i\in[n]$ incurring a loss of $\Lambda$ sells $x_i^m\rho(\Lambda/c_i)$ of its shares of asset $m$. 
The fraction $\Lambda/c_i$ describes the relative loss of institution $i$ measured against its initial equity. It is hence sensible to assume that
$$\rho(0)=0, \quad \rho(u)\leq1  \text{ and } \rho(u)=\rho(1) \text{ for all }  u\ge 1.$$
If at default of an institution the whole portfolio is to be liquidated, then $\rho(1)=1$. In general, however, the remaining assets at default may be frozen by the insolvency administrator and only be sold to the market on a longer time scale.
In this case,  $\rho(1)\in[0,1)$. Our assumptions on $\rho$ are rather mild and allow for a flexible description of various scenarios. Some concrete examples for sales functions $\rho$ are as follows:
\begin{itemize}
\itemsep-0.1em
\item The perhaps simplest non-trivial example is $\rho(u)=\1\{u\geq1\}$. It describes complete liquidation of the portfolio at default (if the institution is leveraged) resp.~dissolution.
\item A more involved example can be derived from a leverage constraint that prohibits an institution from investing more money into risky assets than a certain multiple $\lambda_\text{max}\geq 1$ of its capital/equity. In the one-asset case this means that $xp/c=:\lambda\leq\lambda_\text{max}$, where $x$ denotes the number of shares held, $p$ is the price per share and $c$ denotes the institution's capital. Assume now that while the asset price $p$ stays constant, the institution suffers an exogenous shock $\ell$ and $c$ is reduced to $\tilde{c}=c-\ell$. 
\begin{itemize}
\itemsep-0.2em
\item If $\ell\leq (1-\lambda/\lambda_\text{max})c$, then the leverage constraint $xp/\tilde{c}\leq\lambda_\text{max}$ is satisfied and no reaction is required by the institution.
\item However, if $\ell> (1-\lambda/\lambda_\text{max})c$, then the institution must divest some of its shares; suppose that it sells $\delta x$ of them, for some $0<\delta\le 1$. In order for the leverage constraint $(1-\delta)xp/\tilde{c}\leq\lambda_\text{max}$ to hold, it is easy to verify that $\delta\geq 1-(1-\frac\ell{c})\frac{\lambda_\text{max}}\lambda$.
\end{itemize}
The relative asset sales are hence given by $\rho(\ell/c)$, where $\rho(u):=(1-(1-u)\lambda_\text{max}/\lambda)^+$ for $u\in[0,1]$; this amounts to linear sales (with respect to the losses) once the threshold $1-\lambda\lambda_\text{max}^{-1}$ is reached.
\item Taking an alternative route in the previous example, suppose that the loss of the institution stems only from  a price change $p\to \tilde{p}<p$, which reduces the capital to $\tilde{c} = c - x(p-\tilde{p})$. If $\tilde{p}\geq p(1-\frac1\lambda)/(1-\frac1{\lambda_\text{max}})$, then no action is required  to comply with the leverage constraint. In the remaining cases, the institution must sell a fraction of $1-\lambda_\text{max}+\lambda_\text{max}\frac{p}{\tilde{p}}(1-\frac{1}{\lambda})$ of its assets, and we obtain  $\rho(u)=(1-\lambda_\text{max}(1-u)/(\lambda-u))^+$ for $u\in[0,1]$. 
\item Finally, it can be shown that for price changes combined with exogenous losses the sale function is bounded from above and below by the two previous cases. Leverage constraints hence imply a sale function which is $0$ below a certain threshold and then grows linearly. 
\end{itemize}
The actual reasoning behind asset sales is in general more complex than the presented examples, and this is why we consider a general sale function $\rho$ in this thesis. A natural assumption is that $\rho$ is right-continuous. By replacing $\rho(u)$ with its right-continuous modification \mbox{$\overline{\rho}(u):=\lim_{\epsilon\to0+}\rho((1+\epsilon)u)$} throughout this chapter, our results become applicable also for arbitrary (not right-continuous) sale functions $\rho$. Finally, denote by $\accentset{\circ}{\rho}(u):=\lim_{\epsilon\to0+}\rho((1-\epsilon)u)$ the left-continuous modification of $\rho$.

Let us remark that more generally we may choose different sale functions $\rho^m$ for all assets $m\in[M]$. It is then possible to replace the scalar function $\rho(u)$ by the diagonal matrix $\mathrm{diag}(\rho^1(u),\ldots,\rho^M(u))$ in all the following considerations. Further, we may partition the set of institutions into different types (banks, insurance companies, hedge funds, \dots) and choose different $\rho$ or $\rho^m$ for each type. Finally, our proofs in this chapter also work for other arguments than $ \Lambda/c_i$ for $\rho$ (where $\Lambda$ are the losses), but for simplicity we stick to this particular form.

\vspace{-7pt}
\paragraph{Price impact} Since the assets are not perfectly liquid (the limit order book has finite depth), the sales of shares triggered by the exogenous shock cause prices to go down. This on the other hand causes losses for all the institutions invested in the assets due to mark-to-market accounting.
We model the price loss of asset $m\in[M]$ by a continuous function $h^m:\R_{+,0}^M\to[0,1]$ which is non-decreasing in each coordinate. That is, if $\bm{y}=(y^1,\ldots,y^M) \in \R_{+,0}^M$ and $ny^m$ shares of asset $m$ have been sold in total, then we assume that the price of the asset $m$ drops by $h^m(\bm{y})$; hence each institution $i\in[n]$ suffers losses of $\bm{x}_i\cdot h(\bm{y})$, where $\bm{x}_i:=(x_i^1,\ldots,x_i^M)$ and $h(\bm{y})=(h^1(\bm{y}),\ldots,h^M(\bm{y}))$.

\pagebreak
Two remarks are appropriate. First, note the relative parametrization with the number of institutions $n$, where we assumed that $ny^m$ (instead of $y^m$) shares of asset $m$ are sold.
For fixed $n$ this is arbitrary; however, when we later consider the stochastic model (see Assumption \ref{4:ass:regularity:fire:sales}), this parametrization will turn out to be rather convenient to state our results.

Further note that $\bm{x}_i\cdot h(\bm{y})$ usually only describes an upper bound on institution $i$'s losses at the time that $\bm{y}$ shares were sold, since in general $i$ might already have sold parts of its shares at an earlier time and thus higher prices. Our model is in this sense conservative and also incorporates implementation losses (price changes for the particular trade itself) by selling institutions. Further, this will allow for explicit analytic results in the following. It is an interesting question for future research to extend the model so that it also accounts for intermediate sales.

\vspace{-7pt}
\paragraph{Fire sales} The fire sales process that we consider is described by the combination of the previous two ingredients. Triggered by some exogenous event the institutions start selling a portion of their assets hence driving down prices. Due to mark-to-market effects, however, this means that institutions experience further losses and are forced into further sales. This iterative process continues until the system stabilizes and no further sales, losses and price changes  occur.

\subsection{Fire Sales -- The Deterministic Model}\label{4:ssec:fire:sales:finite}
In this section, we provide a complete description of the final state of the system after the fire sales process is completed. We are interested in the vector $\bm{\chi}_n$ of the number of finally sold shares divided by $n$ after the fire sales process and hence the final price impact $h^m(\bm{\chi}_n)$ on any asset $m\in[M]$. Further, for leveraged institutions such as banks or hedge funds, it makes sense to consider also the the size of the set of finally defaulted institutions $\mathcal{D}_n$ and more generally the damage $\mathcal{S}_n=\sum_{i\in\mathcal{D}_n}s_i$ caused by their default. Given $\bm{\chi}_n$, we readily obtain that $\mathcal{D}_n:=\{i\in[n]\,:\,\ell_i+\bm{x}_i\cdot h(\bm{\chi}_n) \geq c_i\}$ and hence

\begin{equation}
\label{4:eq:finalDefFixedSize}
n^{-1}\mathcal{S}_n = n^{-1}\sum_{i\in[n]}s_i\1\{\ell_i+\bm{x}_i\cdot h(\bm{\chi}_n)\geq c_i\} = \E\left[S_n\1\left\{L_n+\bm{X}_n\cdot h(\bm{\chi}_n)\geq C_n\right\}\right].
\end{equation}
Note that for the special case of $s_i=1$ for all $i\in[n]$ considered in \cite{Detering2018b}, $\mathcal{S}_n=\vert\mathcal{D}_n\vert$ and by our results below we can thus also make statements about the final default fraction $n^{-1}\vert\mathcal{D}_n\vert$. 

In order to derive $\bm{\chi}_n$ we first consider the special case that the sale function $\rho$ is continuous. We consider the fire sales process in rounds, where in each round institutions react to the price changes from the previous round. Denote by $\bm{\sigma}_{(k)}=(\sigma_{(k)}^1,\ldots,\sigma_{(k)}^M)$ the vector of cumulatively sold shares in round $k$. For $k=1$, we readily obtain

\[ \bm{\sigma}_{(1)} = \sum_{i\in[n]}\bm{x}_i \rho\left(\frac{\ell_i}{c_i}\right) = n\E\left[\bm{X}_n\rho\left(\frac{L_n}{C_n}\right)\right]. \]

\noindent Similarly, in round $k\geq2$,

\begin{equation}
\label{4:eq;selling}
\bm{\sigma}_{(k)} = \sum_{i\in[n]} \bm{x}_i\rho\left(\frac{\ell_i + \bm{x}_i\cdot h(n^{-1}\bm{\sigma}_{(k-1)})}{c_i}\right) = n\E\left[\bm{X}_n\rho\left(\frac{L_n + \bm{X}_n\cdot h(n^{-1}\bm{\sigma}_{(k-1)})}{C_n}\right)\right].
\end{equation}

\noindent Thus, by induction $(\bm{\sigma}_{(k)})_{k\in\N}$ is non-decreasing componentwise and bounded by $n\E[\bm{X}_n]$. The limit $n\bm{\chi}_n := \lim_{k\to\infty}\bm{\sigma}_{(k)}$ -- the vector of finally sold shares -- must hence exist.

\pagebreak
\begin{lemma}\label{4:lem:final:state:cont:rho}
Consider the fire sales process with continuous $\rho$. Then  $\bm{\chi}_n = n^{-1}\lim_{k\to\infty}\bm{\sigma}_{(k)}$, the number of sold shares divided by $n$ at the end of the fire sales process, is the smallest (componentwise) solution of
\begin{equation}\label{4:eqn:fixed:point:sigma}
{
\E\left[\bm{X}_n\rho\left(\frac{L_n+\bm{X}_n\cdot h(\bm{\chi})}{C_n}\right)\right] - \bm{\chi} = \bm{0}.}
\end{equation}
\end{lemma}
\begin{proof}
By continuity of $\rho$ and the dominated convergence theorem 
\begin{align*}
\bm{\chi}_n &= n^{-1}\lim_{k\to\infty}\bm{\sigma}_{(k)} = \lim_{k\to\infty}\E\left[\bm{X}_n\rho\left(\frac{L_n + \bm{X}_n\cdot h(n^{-1}\bm{\sigma}_{(k-1)})}{C_n}\right)\right]\\
&= \E\left[\bm{X}_n\rho\left(\frac{L_n + \bm{X}_n\cdot h(n^{-1}\lim_{k\to\infty}\bm{\sigma}_{(k-1)})}{C_n}\right)\right] = \E\left[\bm{X}_n\rho\left(\frac{L_n + \bm{X}_n\cdot h(n^{-1}\bm{\chi}_n)}{C_n}\right)\right]
\end{align*}
and $\bm{\chi}_n$ is thus a solution of \eqref{4:eqn:fixed:point:sigma}. By the Knaster-Tarski theorem there exists a least fixed point $\hat{\bm{\chi}}_n$. Clearly, \mbox{$\bm{\sigma}_{(0)}:=\bm{0}\leq n\hat{\bm{\chi}}_n$.} Hence assume inductively that $\bm{\sigma}_{(k)}\leq n\hat{\bm{\chi}}_n$ for $k\geq1$. Then
\begin{equation}\label{4:eqn:induction:sigma}
\bm{\sigma}_{(k+1)} = \sum_{i\in[n]}\bm{x}_i \rho\left(\frac{\ell_i + \bm{x}_i\cdot h(n^{-1}\bm{\sigma}_{(k)})}{c_i}\right) \leq \sum_{i\in[n]}\bm{x}_i \rho\left(\frac{\ell_i + \bm{x}_i\cdot h(\hat{\bm{\chi}}_n)}{c_i}\right) = n\hat{\bm{\chi}}_n
\end{equation}
by monotonicity of $\rho$, and hence $\bm{\chi}_n=n^{-1}\lim_{k\to\infty}\bm{\sigma}_{(k)}\leq\hat{\bm{\chi}}_n$. By definition of $\hat{\bm{\chi}}_n$ it thus holds that $\bm{\chi}_n=\hat{\bm{\chi}}_n$.
\end{proof}
It remains to study the case where the sale function $\rho$ is only right-continuous. The following simple example of a non-continuous $\rho$ shows that also in this case it may be possible to determine the final state of the system by the smallest solution of \eqref{4:eqn:fixed:point:sigma}. Consider $\rho(u)=\1\{u\geq1\}$, that is, institutions sell their portfolio as they go bankrupt. Then $\bm{\sigma}_{(k)}\neq\bm{\sigma}_{(k-1)}$ only if in round $k$ at least one institution defaults that was solvent in round $k-1$. Since there are only $n$ institutions, the fire sales process stops after at most $n-1$ rounds and the vector $\bm{\chi}_n$ of finally sold shares divided by $n$ solves \eqref{4:eqn:fixed:point:sigma}. Again by \eqref{4:eqn:induction:sigma} we then obtain $\bm{\chi}_n=\hat{\bm{\chi}}_n$ is the smallest solution of \eqref{4:eqn:fixed:point:sigma}.

Finally, consider an arbitrary right-continuous sale function $\rho$.  Again by the Knaster-Tarski theorem \eqref{4:eqn:fixed:point:sigma} has a smallest solution ${\overline{\bm{\chi}}_n}$ and by \eqref{4:eqn:induction:sigma} it holds $n^{-1}\lim_{k\to\infty}\bm{\sigma}_{(k)}\leq{\overline{\bm{\chi}}_n}$. For left-continuous $\rho$ in fact we would derive equality but for right-continuous $\rho$ it is in general possible that $n^{-1}\lim_{k\to\infty}\bm{\sigma}_{(k)}\lneq{\overline{\bm{\chi}}_n}$. This is the case if $\lim_{k\to\infty}\bm{\sigma}_{(k)}$ sold shares would be enough to start a new round of fire sales but this quantity is actually never reached in finitely many rounds. Then the following holds; the proof is straight-forward by bounding $\rho$ from below with its left-continuous modification $\accentset{\circ}{\rho}$.
\begin{proposition}\label{4:prop:final:state:bounds}
\label{4:prop:upperlowerfixedsize}
Consider the fire sales process with a right-continuous sale function $\rho$ and the corresponding left-continuous modification $\accentset{\circ}{\rho}$. Let {$\overline{\bm{\chi}}_n\in\R_{+,0}^M$} denote the smallest solution of \eqref{4:eqn:fixed:point:sigma}. Moreover, let $\hat{\bm{\chi}}_n\in\R_{+,0}^M$ be the smallest solution of 
\[
{\E\left[\bm{X}_n\accentset{\circ}{\rho}\left(\frac{L_n+\bm{X}_n\cdot h(\bm{\chi})}{C_n}\right)\right] - \bm{\chi} = 0.}
\]
Then the number of sold shares divided by $n$ at the end of the fire sales process satisfies 
\begin{equation}\label{4:eqn:bounds:chi:n}
\hat{\bm{\chi}}_n \leq n^{-1}\lim_{k\to\infty}\bm{\sigma}_{(k)}\leq{\overline{\bm{\chi}}_n}.
\end{equation}
\end{proposition}

\pagebreak\noindent
The equilibrium vector $n\overline{\bm{\chi}}_n$ (in the sense of \eqref{4:eqn:fixed:point:sigma}) is thus a conservative bound on the final number of sold shares $n\bm{\chi}_n=\lim_{k\to\infty}\bm{\sigma}_{(k)}$. However, as discussed above the convergence of the fire sales process to a non-equilibrium heavily relies on the assumption of \emph{arbitrarily} small sale sizes towards the end of the process. For real  systems this is obviously not realistic since the least possible number of shares sold by an institution is lower bounded by $1$. For all practical purposes it will therefore hold that $\bm{\chi}_n=\overline{\bm{\chi}}_n$ and fire sales stop at an equilibrium state.

\subsection{Fire Sales -- The Stochastic Model}
The previous section describes fire sales in any specific (finite) system. Our aim is, however, to understand qualitatively how and which characteristics of a system promote or hinder the spread of fire sales. In the following, we thus consider an ensemble of systems that are similar in the sense that they all share some (observed) statistical characteristics. This similarity is measured in terms of the most natural parameters, namely the joint empirical distribution function~\eqref{4:eq:empDistFixedSize} of the asset holdings, the capital/equity, and the initial losses. In particular, we assume that we have a collection of systems with a varying number $n$ of institutions with the property that the sequence $(F_n)_{n\in N}$ stabilizes, i.\,e. has a limit. Additionally, we assume convergence of the average asset holdings to a finite value; this is a standard assumption avoiding condensation of the distribution of the asset holdings. Our assumptions are collected in the following definition.
\begin{assumption}\label{4:ass:regularity:fire:sales}
Let $M \in \N$. For each $n\in\N$ consider a system with $n$ institutions and~$M$ assets specified by the sequences $\bm{x}(n)=(\bm{x}_i(n))_{1\le i\le n}$ of asset holdings, $\bm{s}(n)=(s_i(n))_{1\le i\le n}$ of systemic importance values, $\bm{c}(n)=(c_i(n))_{1\le i\le n}$ of capitals and $\bm{\ell}(n)=(\ell_i(n))_{1\le i\le n}$ of exogenous losses. Let $F_n$ be the empirical distribution function of these parameters for $n\in\N$ (as in~\eqref{4:eq:empDistFixedSize}) and let $$(\bm{X}_n,S_n,C_n,L_n) = \big((X_n^1, \dots, X_n^M),S_n,C_n,L_n)\big) \sim F_n.$$ Then assume the following.
\begin{enumerate}[(a)]
\itemsep-1pt
\item \textbf{Convergence in distribution:} There is a distribution function $F$ such that as $n \to \infty$, $F_n(\bm{x},s,c,\ell)\to F(\bm{x},s,c,\ell)$ at all continuity points of $F$.
\item \textbf{Convergence of means:} Let $(\bm{X},S,C,L) = ((X^1, \dots, X^M),S,C,L)\sim F$. Then as $n\to\infty$,
\[ \E[S_n]\to\E[S]<\infty\quad\text{and}\quad\E[X_n^m]\to\E[X^m]<\infty, \quad m\in[M]. \]
\end{enumerate}
\end{assumption}
An ensemble of systems satisfying Assumption~\ref{4:ass:regularity:fire:sales} will be called an \emph{$(\bm{X},S,C)$-system with initial shock $L$} in the sequel. A particular and probably the most relevant scenario is as follows. Suppose that the distribution $F$ is specified, for example by considering a real system. Then, for each $n\in \N$ we construct a system by assigning to each institution $i\in [n]$ independently asset holdings, capital and losses distributed like $F$. Then, by the strong law of large numbers, with probability 1, the sequence of systems we obtain satisfies Assumption~\ref{4:ass:regularity:fire:sales}.

As in the \emph{deterministic model} our aim is to describe in this broader setting the final state of the system. Before we do so, let us give some definitions that are handy in the forthcoming description. First, recall that for $n\in\N$ the eventual number of sold shares is characterized by the smallest solution to~\eqref{4:eqn:fixed:point:sigma}, and the systemic importance of defaulted institutions is given by~\eqref{4:eq:finalDefFixedSize}.
Let therefore $f^m,g:\R_{+,0}^M\to\R$, $m\in [M]$ be analogously defined for the ``limiting object'' by
\begin{align}
f^m(\bm{\chi}) &:= \E\left[X^m\rho\left(\frac{L+\bm{X}\cdot h(\bm{\chi})}{C}\right) \right] - \chi^m,\quad m\in[M],\label{4:eqn:fire_sales:fm}\\
g(\bm{\chi}) &:= \E\left[S\1\left\{L+\bm{X}\cdot h(\bm{\chi})\geq C\right\}\right],\nonumber
\end{align}
which are clearly upper semi-continuous, and let (cf.~Proposition~\ref{4:prop:upperlowerfixedsize})
\begin{equation}
\label{4:eqn:fire_sales:circ:fm}
\circFSuper{m}(\bm{\chi}) := \E\left[X^m\accentset{\circ}{\rho}\left(\frac{L+\bm{X}\cdot h(\bm{\chi})}{C}\right) \right] - \chi^m, m\in[M], ~~\text{and}~~
\circG(\bm{\chi}) := \E\left[S\1\left\{L+\bm{X}\cdot h(\bm{\chi})>C\right\}\right]
\end{equation}
be their lower semi-continuous modifications. Further, define the sets
\[
\accentset{\circ}{P}:=\bigcap_{m\in[M]}\left\{\bm{\chi}\in\R_{+,0}^M\,:\,\circFSuper{m}(\bm{\chi})\geq0\right\} \quad\text{and}
\quad
P:=\bigcap_{m\in[M]}\left\{\bm{\chi}\in\R_{+,0}^M\,:\,f^m(\bm{\chi})\geq0\right\}
\]
and denote by $\accentset{\circ}{P}_0$ resp.~$P_0$ the largest connected subsets of $\accentset{\circ}{P}$ and $P$ containing $\bm{0}$ (clearly $f^m(\bm{0})\geq\circFSuper{m}(\bm{0})\geq0$ for all $m\in[M]$). Note that $P$ and $P_0$ are closed sets by upper semi-continuity of $f^m$, $m\in[M]$. 

Let us immediately give an illustrative explanation of all the quantities above. Leave aside for the moment the $\accentset{\circ}{\cdot}$ modifications and assume that during the fire sales process assets are sold continuously in time $\tau$, that is,  let $n\bm{\chi}(\tau)$ at any given time $\tau$ be the current vector of sold shares. Then by definition of the process, $n\E[X^m\rho((L+\bm{X}\cdot h(\bm{\chi}(\tau)))/C) ]$ is the number of shares of asset $m$ the system as a whole needs to sell. 
Moreover, as $\bm{\chi}^m(\tau)$ is the current number of sold shares, at any point in time $\tau$ it must hold $f^m(\bm{\chi}(\tau))\geq0$. Since the process starts with zero sold shares ($\bm{\chi}(0)=\bm{0}$) it is therefore intuitive that $\bm{\chi}(\tau)$ will never leave the set $P_0$. At the end of the fire sales process, however, the final number of sold shares must equal the required number of sold shares for each asset $m\in[M]$ (cf.~\eqref{4:eqn:fixed:point:sigma} for the \emph{deterministic model}). So the final state will be described by a joint root of the functions $f^m$, $m\in[M]$, that lies in $P_0$. The $\accentset{\circ}{\cdot}$ modifications come into play because in certain pathological cases it can be important if the limiting system (as $n\to\infty$) is approached from below or from above.

To formalize this intuition about joint roots in $P_0$, consider the following lemma. Let
\[
\bm{\chi}^*\in\R_{+,0}^M \quad\text{with}\quad (\chi^*)^m:=\sup_{\bm{\chi}\in P_0}\chi^m.
\]
\begin{lemma}\label{4:lem:existence:chi:hat}
There exists a smallest joint root $\hat{\bm{\chi}}$ of all functions $\circFSuper{m}(\bm{\chi})$, $m\in[M]$ with $\hat{\bm{\chi}}\in\accentset{\circ}{P}_0$. Further, $\bm{\chi}^*$ as defined above is a joint root of the functions 
$f^m$, $m\in[M]$, and $\bm{\chi}^*\in P_0$.
\end{lemma}
At the end of this section we will give a couple of examples illustrating the situation in concrete settings. Using the quantities $\hat{\bm{\chi}}$ and $\bm{\chi}^*$ as well as the functions $\circG$ and $g$, we can then describe the final state of the system after the fire sales process asymptotically as $n\to\infty$.
\begin{theorem}\label{4:thm:fire:sale:final:fraction}
Consider a  system satisfying Assumption \ref{4:ass:regularity:fire:sales}. Then for the final systemic damage $n^{-1}\mathcal{S}_n$ and $\chi_n^m$, the number of finally sold shares of asset $m\in[M]$ divided by $n$ it holds
\[\begin{gathered}
\circG(\hat{\bm{\chi}}) + o(1) \leq n^{-1}\mathcal{S}_n \leq g(\bm{\chi}^*) + o(1), \qquad
\hat{\chi}^m+o(1) \leq \chi_n^m \leq (\chi^*)^m+o(1).
\end{gathered}\]
In particular, for the final price impact $h^m(\bm{\chi}_n)$ on asset $m\in[M]$, we obtain
\[ h^m(\hat{\bm{\chi}})+o(1) \leq h^m(\bm{\chi}_n) \leq h^m(\bm{\chi}^*)+o(1). \]
\end{theorem}
Only in rather pathological cases, when $\hat{\bm{\chi}}$ is a point of discontinuity for some $f^m$ (cf.~Figure \ref{4:fig:example:non:continuous}) or it is instable (compare Figure \ref{4:fig:two:joint:roots} to Figures \ref{4:fig:one:joint:root:small} and \ref{4:fig:one:joint:root:large}), it happens that $\hat{\bm{\chi}}\neq\bm{\chi}^*$. Similarly it usually holds that $\circG(\hat{\bm{\chi}})=g(\bm{\chi}^*)$. Then Theorem \ref{4:thm:fire:sale:final:fraction} determines the limits of $\bm{\chi}_n$ and $n^{-1}\mathcal{S}_n$ as $n\to\infty$. We conclude this section with two illustrative examples.

\begin{example}\label{4:ex:cont}
It is often (e.\,g.~if $\rho$ is continuous or $\bm{X}$ is absolutely continuous) the case that $f^m(\bm{\chi})=\circFSuper{m}(\bm{\chi})$ and hence $P=\circP$ as well as $P_0=\circP_0$. See Figure \ref{4:fig:examples:chi:hat} for an illustration of three different two-dimensional examples. We chose $h^m(\bm{\chi})=\chi^m$, $m=1,2$, $\rho(y)=\1\{y\geq1\}$, $X^1=X^2\sim\mathrm{Exp}(1)$, $C=c$, $\P(L=c)=0.15$ and $\P(L=0)=0.85$, where $c=1.9$ in (a), $c\approx1.82$ in (b) and $c=1.5$ in (c). Whereas in (a) and (c) it holds $\hat{\bm{\chi}}=\bm{\chi}^*$, in (b) the points are distinct.
\end{example}

\begin{figure}[t]
    \hfill\subfigure[]{\includegraphics[width=0.32\textwidth]{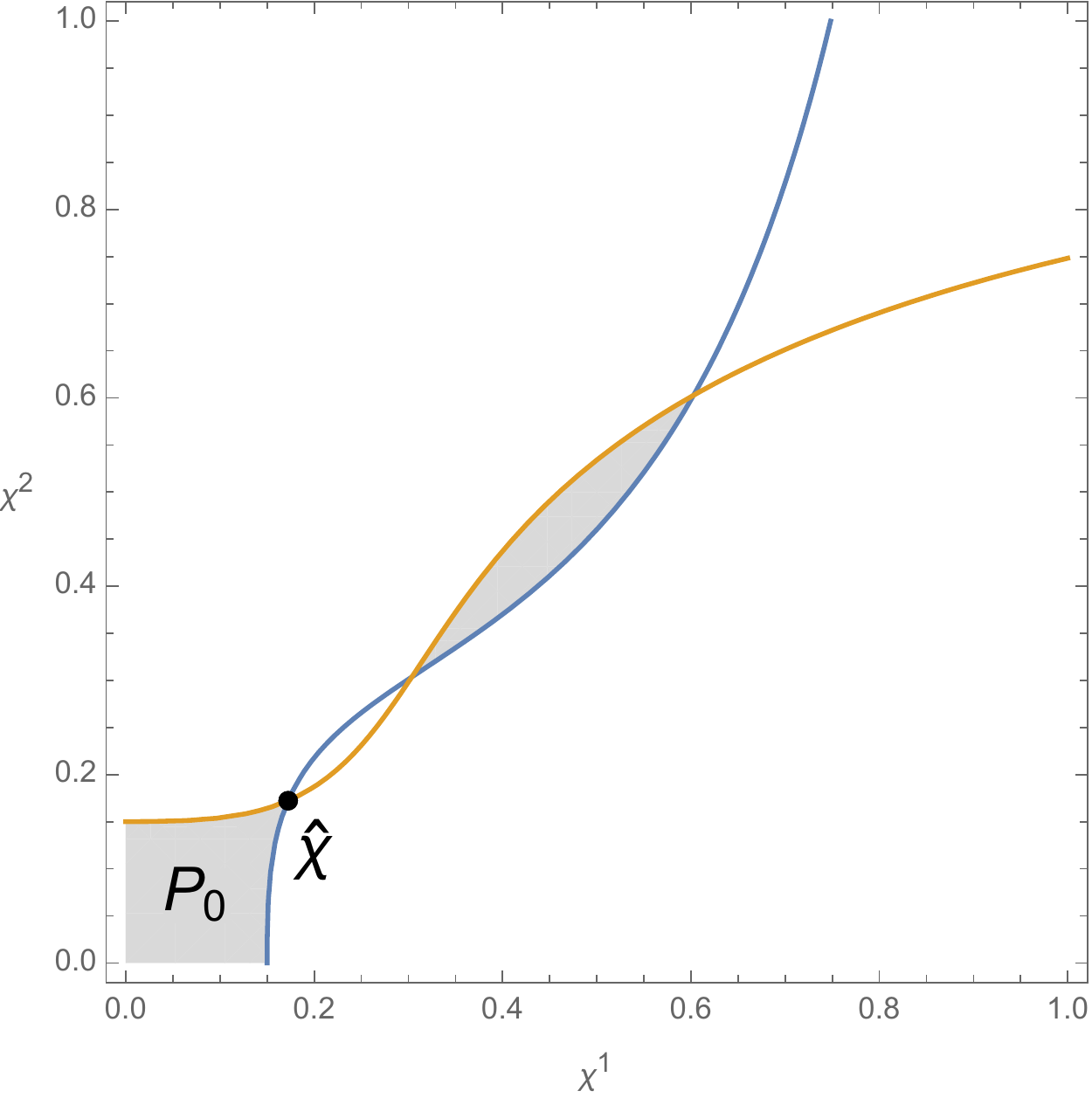}\label{4:fig:one:joint:root:small}}
    \hfill\subfigure[]{\includegraphics[width=0.32\textwidth]{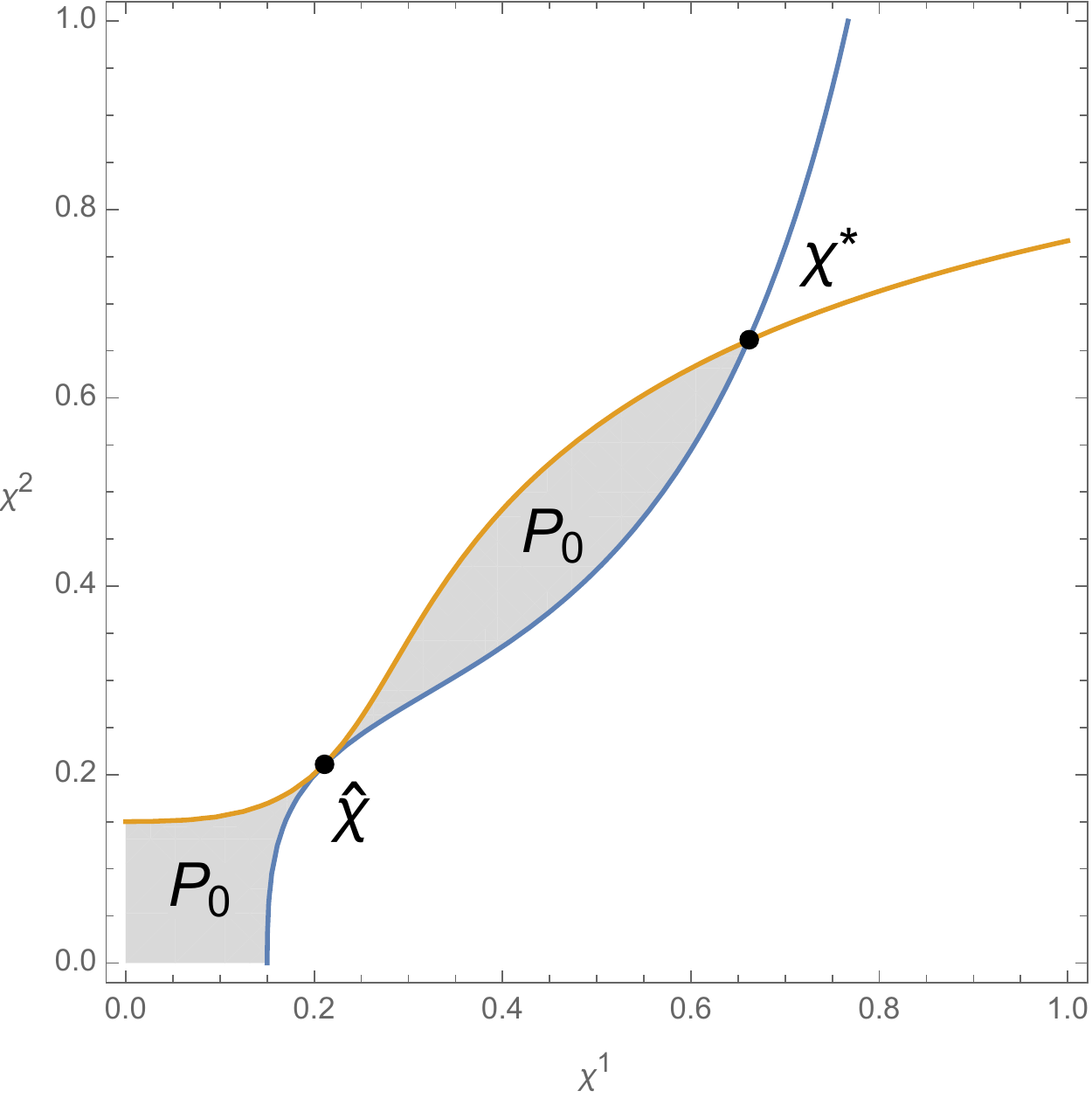}\label{4:fig:two:joint:roots}}\hfill\subfigure[]{\includegraphics[width=0.32\textwidth]{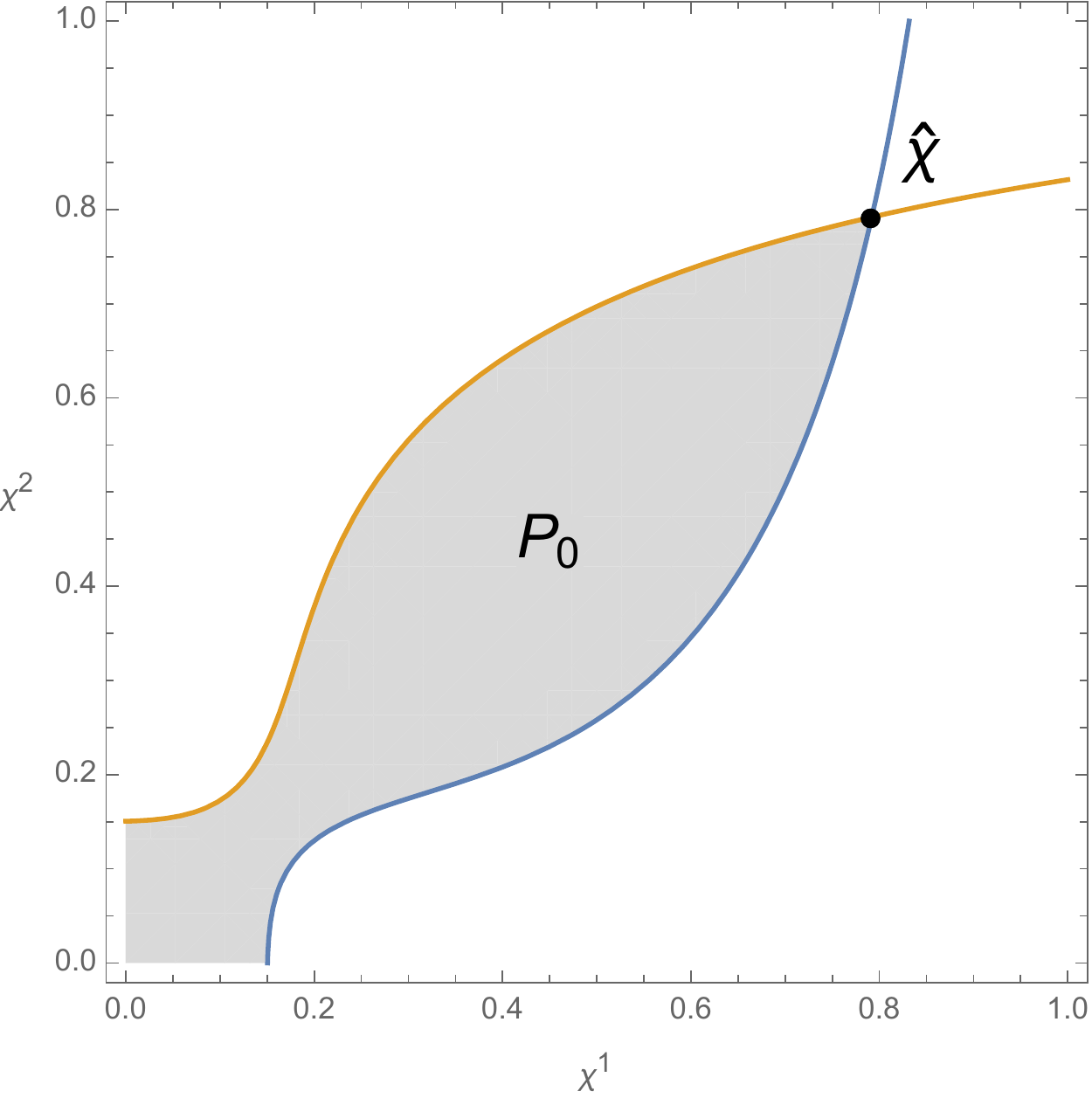}\label{4:fig:one:joint:root:large}}
    \hfill
\caption{Plot of the root sets of the functions $f^1(\chi^1,\chi^2)$ (blue) and $f^2(\chi^1,\chi^2)$ (orange) for three different example systems. In gray the set $P$ is depicted.}\label{4:fig:examples:chi:hat}
\end{figure}

\begin{example}\label{4:ex:discont}
Since it is of advantage for the proofs in the rest of this chapter to gain intuition about the general case as well, in Figure \ref{4:fig:example:non:continuous} we further provide an example of a system where $\circFSuper{m}\neq f^m$ and $\circP_0\subsetneq P_0$. Again we chose $h^m(\bm{\chi})=\chi^m$, $m=1,2$, $\rho(y)=\1\{y\geq1\}$ and $X^1=X^2$. Further, $C=0.1$, $\P(X^1=1,L=0.1)=0.1$, $\P(X^1=0.5,L=0)=0.1$ and with the remaining probability of $0.8$ it holds that $L=0$ and $X^1$ is uniformly distributed on the interval $[0,0.5]$.
\end{example}

\begin{figure}[t]
    \hfill\subfigure[]{\includegraphics[width=0.38\textwidth]{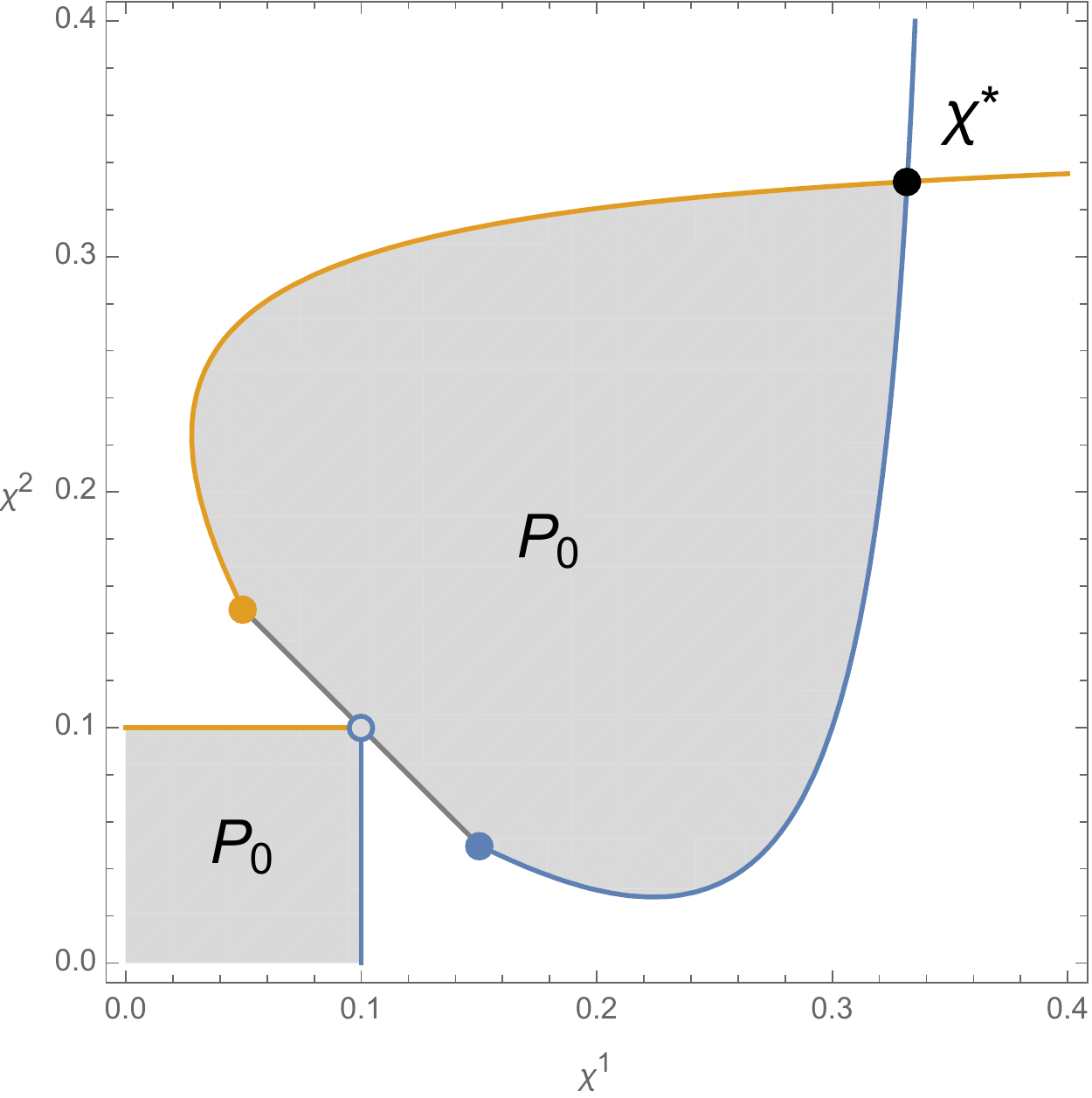}\label{4:fig:ex:non:continuous}}
    \hfill\subfigure[]{\includegraphics[width=0.38\textwidth]{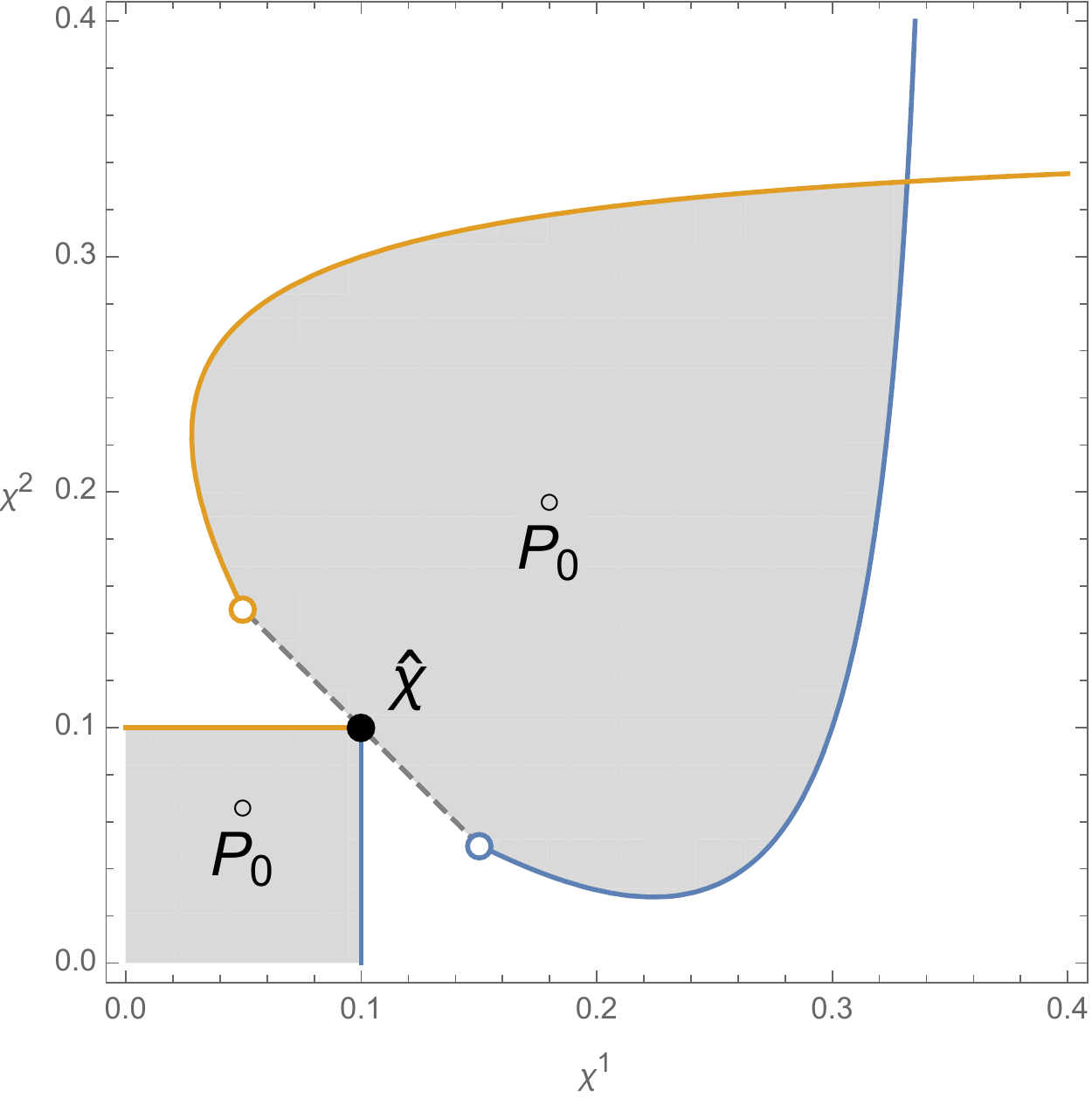}\label{4:fig:ex:non:continuous:circ}}\hfill
\caption{Plot of the root sets of the functions $f^1(\chi^1,\chi^2)$ (blue), $f^2(\chi^1,\chi^2)$ (orange) in (a) and $\circFSuper{1}(\chi^1,\chi^2)$ (blue), $\circFSuper{2}(\chi^1,\chi^2)$ (orange) in (b) respectively. In gray the sets $P_0$ and $\circP_0$ respectively, where the solid line in (a) belongs to $P_0$ whereas the dashed line in (b) does not.}\label{4:fig:example:non:continuous}
\end{figure}

\section{Resilient and Non-resilient Systems}\label{4:sec:resilience}

In the previous section, we derived results that allow us to determine the final default fraction in $({\bf X},S,C)$-systems caused by fire sales and sparked by some exogenous shock $L$. In this section, we go one step further and investigate whether a given  system in an \emph{initially} unshocked state is likely to be resilient to small shocks or susceptible to fire sales. 

Note that all information about an initial shock comes from the random variable $L$, whereas the system itself is specified by~$(\bm{X},S,C)$. So we can easily consider shocks of different magnitude~$L$ on the same a priori unshocked system. In the following, if we use the notation~$g$, $f^m$, $\circG$, $\circFSuper{m}$, $\hat{\bm{\chi}}$ and $\bm{\chi}^*$, we mean the quantities from the previous section (see Equations~\eqref{4:eqn:fire_sales:fm} and \eqref{4:eqn:fire_sales:circ:fm}) for the $(\bm{X},S,C)$-system, that is, with initial shock $L\equiv0$.

\subsection{Resilience}\label{4:ssec:resilience}

From a regulator's perspective a desirable property of an $(\bm{X},S,C)$-system is the ability to absorb small local shocks $L$ without larger parts of the system being harmed. In our model, we can even choose $L$ arbitrarily small. The following way of defining resilience is hence natural: we let the shock $L$ become small in the sense that $\E[L/C]\to0$, and we call the system \emph{resilient} if the asymptotic final systemic damage $n^{-1}\mathcal{S}_{n,L}$ caused by $L$ also tends to $0$.
\begin{definition}[Resilience]\label{4:def:resilience}
An $(\bm{X},S,C)$-system  is said to be \emph{resilient} if for each $\epsilon>0$ there exists $\delta>0$ such that for all $L$ with $\E[L/C]<\delta$ it holds $\limsup_{n\to\infty}n^{-1}\mathcal{S}_{n,L} \leq \epsilon$.
\end{definition}
We chose our definition of resilience (and non-resilience in Definition \ref{4:def:non:resilience} below) in terms of the final systemic damage and thus in line with Chapters \ref{chap:systemic:risk} and \ref{chap:block:model}. Alternatively, depending on the quantity of interest, it can also be sensible to define resilience via the final number of sold shares $n\bm{\chi}_{n,L}$ (and hence the final price impacts $h^m(\bm{\chi}_{n,L})$ which also affect the wider economy). Theorem \ref{4:thm:resilience} determines upper bounds for both $n^{-1}\mathcal{S}_{n,L}$ and $\bm{\chi}_{n,L}$ in the limit $\E[L/C]\to0$.
\begin{theorem}\label{4:thm:resilience}
For each $\epsilon>0$ there exists $\delta>0$ such that for all $L$ with $\E[L/C]<\delta$ it holds for the final damage caused by defaulted institutions $n^{-1}\mathcal{S}_{n,L}$ and the number $n\chi_{n,L}^m$ of finally sold shares of each asset $m\in[M]$ in the shocked system that
\[ \limsup_{n\to\infty}n^{-1}\mathcal{S}_{n,L}\leq g(\bm{\chi}^*)+\epsilon\quad \text{and}\quad \limsup_{n\to\infty} \chi_{n,L}^m \leq(\chi^*)^m+\epsilon, ~~ m\in[M]. \]
\end{theorem}
We immediately obtain the following handy resilience criterion.
\begin{corollary}[Resilience Criterion]\label{4:cor:resilience}
If $g(\bm{\chi}^*)=0$, then the $(\bm{X},S,C)$-system is resilient.
\end{corollary}
Note that $g(\bm{0})=0$ and hence the system is resilient if $\bm{\chi}^*=\bm{0}$ (i.\,e.~$P_0=\{\bm{0}\}$). It is possible, however, that $g(\bm{\chi}^*)=0$ while $\bm{\chi}^*\neq\bm{0}$. In this case, by Theorem \ref{4:thm:non-resilience} below it is possible that a large fraction of shares of assets is sold, but Corollary \ref{4:cor:resilience} ensures that the fraction of finally defaulted institutions stays small.

\subsection{Non-resilience}\label{4:ssec:non:resilience}
To a large degree we can also characterize non-resilient systems. Note, however, that in our model description we made the conservative assumption that each institution $i\in[n]$ in the system is exposed to the final price impact $h(\bm{\chi}_n)$ with its total initial asset holdings $\bm{x}_i$. One can argue that institutions sell off their assets gradually and are hence not exposed to the total price change. The following result considers non-resilience under this conservative assumption. For other scenarios it can serve as a first indication of non-resilience.

In this subsection, we restrict ourselves to initial shocks of the form $\ell_i\in\{0,2c_i\}$ for all $i\in[n]$, where $\P(L=2C)>0$ and $L/C$ is independent of $(\bm{X},C)$. That is, each institution $i$ defaults initially with positive probability.
Rather than $\ell_i=2c_i$, the first natural choice to model the default of institution $i$ would be $\ell_i=c_i$. Note, however, that in the setting of Section~\ref{4:sec:fire:sales}, even if $\P(L=C)>0$, it is possible that no initial defaults occur since $(L,C)$ is defined as the weak limit of a sequence $(L_n,C_n)$ and it is possible that $L_n<C_n$ for all $n\in\N$ and still $L=C$ almost surely. In order to derive meaningful results one therefore has to choose $\ell_i=2c_i$ (or any other multiple larger than $1$). Note that this does not change the outcome of the fire sales process since $\rho(u)=\rho(1)$ for all $u\geq1$.

We now call a financial system \emph{non-resilient} if the damage caused by finally defaulted institutions is lower bounded by some positive constant.
\begin{definition}[Non-resilience]\label{4:def:non:resilience}
An $(\bm{X},S,C)$-system is said to be \emph{non-resilient} if there exists $\Delta\in\R_+$ such that $\liminf_{n\to\infty}n^{-1}\mathcal{S}_{n,L}>\Delta$ for any initial shock $L$ with the above listed properties.
\end{definition}
We derive the following lower bound on the final default fraction and finally sold shares. 

\begin{theorem}\label{4:thm:non-resilience}
If the initial shock $L$ satisfies above properties and $h^m(\bm{\chi})$ is strictly increasing in $\chi^m$ for all $m\in[M]$, then for any $\epsilon>0$ it holds
\[ \liminf_{n\to\infty}n^{-1}\mathcal{S}_{n,L} > \circG(\bm{\chi}^*)-\epsilon \quad\text{and}\quad \liminf_{n\to\infty}\chi_{n,L}^m > (\chi^*)^m-\epsilon. \]
\end{theorem}
The assumption that $h^m(\bm{\chi})$ is strictly increasing in $\chi^m$ excludes a rather pathological case. It is satisfied by all the standard price impact functions in the literature such as linear price impact $h_\text{lin}^m(\bm{\chi})=p^m\alpha\chi^m$ or log-linear price impact $h_\text{loglin}^m(\bm{\chi})=p^m(1-\exp(-\alpha\chi^m))$ for initial share price $p^m$ and some parameter $\alpha>0$.

\begin{corollary}[Non-resilience Criterion]\label{4:cor:non:resilience}
If $h^m(\bm{\chi})$ is strictly increasing in $\chi^m$ for all \mbox{$m\in[M]$,} and $\circG(\bm{\chi}^*)>0$, then the $(\bm{X},S,C)$-system is non-resilient.
\end{corollary}
As remarked earlier already, for most practical purposes it will hold that $\circG(\bm{\chi}^*)=g(\bm{\chi}^*)$. For the reasonable case that $h^m(\bm{\chi})$ is strictly increasing in $\chi^m$ for all $m\in[M]$, we can thus fully describe stability of an $(\bm{X},S,C)$-system by the combination of Corollaries \ref{4:cor:resilience} and \ref{4:cor:non:resilience}. Only for rather pathological cases we cannot decide if an $(\bm{X},S,C)$-system is resilient or non-resilient.

\subsection{Systemic Capital Requirements}
\label{4:ssec:capital:requirements}

Theorems \ref{4:thm:resilience} and \ref{4:thm:non-resilience} can be used to derive sufficient and necessary capital requirements to make a given system resilient with respect to initial shocks. That is, given the asset holdings of each institution, we want to determine sharp bounds for the capital that each institution must hold so that the system is (non-)resilient in the sense of Definitions~\ref{4:def:resilience} and \ref{4:def:non:resilience}. Recall that non-resilience according to Definition \ref{4:def:non:resilience} is always under our conservative model assumption that price impact is applied to all initially held shares of an asset. The derived requirements will generally depend on the actual sale function $\rho$ and price impact $h$. In the following, we demonstrate the procedure of deriving capital requirements along a series of examples of ever increasing complexity culminating in a quite general setting with multiple assets. We always consider systemic importance values $s_i=1$ for all $i\in[n]$.

\vspace{-7pt}
\paragraph{One Asset with Sales at Default} We start with considering a system with institutions investing in one asset only. The distribution $F_X$ of asset holdings is assumed to have a power law tail in the sense that there exist constants $B_1,B_2\in(0,\infty)$ such that for $x$ large enough
\begin{equation}
\label{4:eq:Xpowerlaw}
B_1 x^{1-\beta}\leq 1-F_X(x) \leq B_2 x^{1-\beta},
\end{equation}
for some $\beta > 2$. Whereas the reduction to one asset is a strong simplification, there is empirical evidence for power laws in investment volumes, see e.\,g.~\cite{Garlaschelli2005}.
Moreover, we assume \mbox{$\P(X\geq 1)=1$} -- institutions involved in the fire sales process hold at least one share. Further, assume \mbox{$\rho(u)=\1\{u\geq1\}$,} i.\,e.~institutions do not sell their assets until they default. Moreover, for the price impact we also assume a power-law, that is, there are $\nu,\mu_1,\mu_2 \in \R_{+}$ such that for small $\chi$
\begin{equation}
\label{4:eq:hpowerlaw}
\mu_1 \chi^\nu\leq h(\chi) \leq \mu_2 \chi^\nu.
\end{equation}
A typical assumption in the fire sales literature, see e.\,g.~\cite{Braverman2014,Caccioli2014,Cifuentes2005,Cont2017,Cont2018}, is (log-)linear price impact, i.\,e.~$\nu=1$. This choice can for example be justified by the fact that limit order books' shape functions are approximately constant close to the bid price, see also \cite{Gu2008}.

We now derive necessary and sufficient requirements for the institutions' capital buffers to make the financial system resilient. That is, given the asset holdings $(x_1(n), \dots, x_n(n))_{n \in \N}$ with empirical distribution converging to $F_X$ we want to determine a 
sequence of minimal capitals $(c_1(n), \dots, c_n(n))_{n \in \N}$ sufficient for ensuring resilience of the system in the sense of Definition \ref{4:def:resilience}. It turns out that a natural description emerges when we choose the capitals in dependence on the asset holdings by the following power form: $c_i = \alpha x_i^\gamma$ for $\alpha\in\R_+$ and $\gamma\in\R_{+,0}$. 

\begin{corollary}\label{4:cor:one:asset:sales:default}
Consider a system as specified above. Then,
\begin{enumerate}
\item if $\gamma>1-\nu(\beta-2)$, then the system is resilient.
\item if $\gamma=1-\nu(\beta-2)$ and $\alpha>\mu_2 \left(B_2\frac{\beta-1}{\beta-2}\right)^\nu$, then the system is resilient.
\item if $\gamma=1-\nu(\beta-2)$ and $\alpha<\mu_1\left(B_1\frac{\beta-1}{\beta-2}\right)^\nu$, then the system is non-resilient.
\item if $\gamma<1-\nu(\beta-2)$, then the system is non-resilient.
\end{enumerate}
\end{corollary}
Typically one can choose $B_1$ and $B_2$ resp.~$\mu_1$ and $\mu_2$ arbitrarily close as $x\to\infty$ resp.~$\chi\to0$. Corollary \ref{4:cor:one:asset:sales:default} hence states necessary and sufficient conditions on the capital to make the financial system resilient. 

\vspace{-7pt}
\paragraph{One Asset with Intermediate Sales} In the previous example, we considered the conservative case of sales at default only. Intermediate sales will make the system less resilient, however, and we consider an example of this kind here as well. We assume again~\eqref{4:eq:Xpowerlaw} and~\eqref{4:eq:hpowerlaw}, that is, the asset holdings are power-law distributed with parameter $\beta$ and the price impact has a power-law approximation with exponent $\nu$ when $\chi\to0$. Moreover, capitals are again specified by $c_i = \alpha x_i^\gamma$. In contrast to the assumptions in the previous paragraph we now choose $\rho(u)=1\wedge u^q$ for some $q\in\R_+$. The parameter $q$ can be understood as a measure for the institutions' confidence in the asset, since it describes the speed at which they sell it. The outbreak of fire sales is then governed by $\nu$ (price impact) and $q$ (speed of selling). In fact, the product $\nu q$ is a crucial quantity in the decision whether a system is resilient. It is easy to show that if $\nu q<1$ the system is always non-resilient, as the institutions sell overproportionally fast compared to the price fall. Note that this quantity cannot be influenced by regulations but is intrinsic to the market in our setting.

\begin{corollary}\label{4:cor:intermediate:sales}
Consider a system as described above and assume that $\nu q>1$. If\linebreak $\gamma>1-\nu(\beta-2)$, then the system is resilient. If $\gamma<1-\nu(\beta-2)$, then the system is non-resilient.
\end{corollary}
Similar as in the proof of Corollary \ref{4:cor:one:asset:sales:default} it is possible to derive sufficient capital requirements also at the critical values $\nu q=1$ resp.~$\gamma=1-\nu(\beta-2)$; the details are omitted. 

\vspace{-7pt}
\paragraph{Multiple Assets} The two previous examples have already given first important insights into the calculation of sufficient capital requirements for stability in a given system. Whereas these concentrated on systems with one asset only, however, in reality institutions are invested in a large number of assets $M$. We will derive sufficient capital requirements also in this setting. In practice, linear capital requirements seem reasonable and the previous examples have shown already that these are sufficient in the one-asset case. Furthermore, linear capitals allow for tractable calculations in the following.

We keep the same assumptions as in the previous example. In particular, $\rho^m(y)=1\wedge y^{q^m}$ for some $q^m\in\R_+$, $m\in[M]$ (recall the note in the beginning of Section \ref{4:sec:fire:sales} about different sale functions for different assets) and $h^m(\bm{\chi}) \le \mu^m(\chi^m)^\nu$ for $\chi^m\to0$ for some $\nu,\mu^m\in\R_+$. Consider then linear capital $c_i=\sum_{m\in[M]}\theta^mx_i^m$ for each institution $i\in[n]$, where $\theta^m,\in\R_+$, $m\in[M]$.
\begin{corollary}\label{4:cor:multiple:assets}
Consider a  system with multiple assets as described above. Then the system is resilient if for each $m\in[M]$ one of the following holds:
\begin{enumerate}
\item \label{4:cor:multiple:assets:1} $q^m>\nu^{-1}$,
\item \label{4:cor:multiple:assets:2} $q^m=\nu^{-1}$ and $\theta^m>\mu^m\E[X^m]$.
\end{enumerate}
\end{corollary}
The first condition reflects the interplay of price impact and the speed of asset sales as for the one-asset case. The second condition gives an explicit linear fraction of the institutions' holdings of each asset that ensures resilience. It depends on the price impact function by $\mu^m$ and the number of shares of the asset held on average by each other institution in the system $\E[X^m]$. While the condition $q^m>\nu^{-1}$ is sufficient in the limit $n\to\infty$ by our theory, for real networks of finite size the quantity $\mu^m\E[X^m]$ is of big interest as it gives a proper scaling factor also for values of $q^m$ other than $\nu^{-1}$. 
The following sample calculations show that $\theta^m$ is of a reasonable magnitude 
also under our conservative model assumptions: Assume that the price impact is log-linear with $h^m(\chi)=1-e^{-\alpha^m\chi}$ for some $\alpha^m\in\R_+$ and that the sale of all assets in the considered system reduces the asset price by $50\%$. This implies $\alpha^m=\log(2)/\E[X^m]$. Hence $\mu^m=\alpha^m$ ($\nu=1$) and $\theta^m>\log(2)\approx0.69$ ensures resilience.

Corollary \ref{4:cor:multiple:assets} thus derives linear capital requirements for institutions investing in more than one asset which are already used in Basel III for instance. We can explicitly determine the coefficients for these linear capital requirements in our model.

\section{Applications and Simulations}\label{4:sec:applications}
In this section, we apply the theory developed in Sections \ref{4:sec:fire:sales} and \ref{4:sec:resilience} to investigate which structures or properties of systems promote the emergence and spread of fire sales. To achieve this our route is as follows. In Subsection \ref{4:ssec:diversification:similarity} we consider systems parametrized by two orthogonal characteristic quantities: \emph{portfolio diversification} and \emph{portfolio similarity}. We first analyze their effect in the setting from Sections \ref{4:sec:fire:sales} and \ref{4:sec:resilience}, and we verify our findings also with simulations for finite systems of reasonable size. At this we assume initial shocks on institutions' capitals directly rather than shocks on certain asset prices. In Subsection \ref{4:ssec:initial:asset:shocks}, we concentrate on three fundamentally different system configurations and we test our derived capital requirements for shocks on asset prices. As we will show, it is beneficial to combine our capital requirements with classical risk capital in form of \emph{value-at-risk}. While the \emph{value-at-risk} part of the capital ensures for any institution that an initial shock can be absorbed with probability $1-\epsilon$ (for some small $\epsilon>0$), the additional \emph{systemic surcharge} in form of our capital requirements makes sure that also in the unlikely event of initial distress the spread of fire sales is locally confined. For simplicity we consider the case that $s_i=1$ for all $i\in[n]$ only.

\subsection{The Effect of Portfolio Diversification and Similarity}\label{4:ssec:diversification:similarity}
For simplicity, throughout this section we assume that the limiting total asset holdings given as $X^\text{tot}=X^1+\ldots+X^M$ are Pareto distributed with density $f_{X^\text{tot}}(x)=(\beta-1)x^{-\beta}\1\{x\geq1\}$ for some exponent $\beta>2$. One can generalize our results also to more general distributions. Further, we make the assumptions that $\rho(u)=\1\{u\geq1\}$ and $h^m(\bm{\chi})=1-e^{-\chi^m}$ to simplify calculations, but also for other sensible choices our observations below are applicable.

In a first example, we consider a system of institutions whose investment in each asset $m\in[M]$ makes up a fraction $\lambda^m\in\R_+$ of their total asset holdings, where $\sum_{m \in [M]}\lambda^m=1$. We show that perfect diversification ($\lambda^1=\ldots=\lambda^M=M^{-1}$) maximizes stability of the system.
\begin{example}\label{4:ex:diversification}
For a system as described above the functions $f^m(\bm{\chi})$ are given by
\[ f^m(\bm{\chi}) = \lambda^m\E\left[X^\text{tot}\1\left\{X^\text{tot}\sum_{k=1}^M\lambda^k\left(1-e^{-\chi^k}\right)\geq C\right\}\right] - \chi^m,\quad m\in[M].
\]
Let us write $t = \sum_{1\le k \le M}\lambda^k(1-e^{-\chi^k})$ for short. Now assume similar to Corollary \ref{4:cor:one:asset:sales:default} that $C=\alpha(X^\text{tot})^\gamma$ for some constants $\alpha,\gamma\in\R_{+,0}$. Then
\begin{align*}
f^m(\bm{\chi})
& = \lambda^m\E\left[X^\text{tot}\1\left\{X^\text{tot}\geq\left({\alpha}/t\right)^\frac{1}{1-\gamma}\right\}\right] - \chi^m\\
& = \lambda^m \int_{\max\left\{1 , \left({\alpha}/t\right)^\frac{1}{1-\gamma}\right\}}^\infty (\beta-1)x^{1-\beta}\,\dd x - \chi^m
=  \lambda^m \frac{\beta-1}{\beta-2} \min\left\{1, (t\alpha^{-1})^\frac{\beta-2}{1-\gamma}\right\} - \chi^m.
\end{align*}
Motivated by the symmetry of the functions, we consider $f^m(\bm{\chi})$ along direction $\bm{v}\in\R_+^M$, with $v^m=(\lambda^m)^{-1}$. Then 
\[
\frac{f^m(\chi\bm{v})}{\lambda^m} = \frac{\beta-1}{\beta-2}\left(\alpha^{-1}\sum_{k=1}^M \lambda^k \left(1-e^{-\chi/\lambda^k}\right)\right)^\frac{\beta-2}{1-\gamma} - \frac{\chi}{(\lambda^m)^2}.
\]
Let
\[ \gamma_c:=3-\beta\quad\text{and}\quad\alpha_c:=\sum_{1\le m \le M}(\lambda^m)^2\frac{\beta-1}{\beta-2}. \]
We infer that if $\chi\in\R_{+,0}$ is small enough and $\gamma>\gamma_c$ or $\gamma=\gamma_c$ and $\alpha>\alpha_c$, then $\frac{\dd}{\dd \chi}f^m(\chi\bm{v})<0$ for all $m\in[M]$. That is, $\bm{\chi}^*=\bm{0}$ and the system is resilient by Corollary~\ref{4:cor:resilience}. On the other hand, if either $\gamma<\gamma_c$ or $\gamma=\gamma_c$ and $\alpha<\alpha_c$, then $\frac{\dd}{\dd \chi}f^m(\chi\bm{v})>0$ for all $m\in[M]$ and the system is non-resilient by Corollary \ref{4:cor:non:resilience}, as $\bm{\chi}^*\neq\bm{0}$ and $\circG(\bm{\chi}^*)=g(\bm{\chi}^*)>0$. Since $\gamma_c$ does not depend on the choice of $\{\lambda^m\}_{m\in[M]}$, it makes sense to consider $\alpha_c$ as a measure for stability of the system (the smaller $\alpha_c$, the more stable the system). Clearly, $\alpha_c$ becomes minimized for $\lambda^m=M^{-1}$ for all $m\in[M]$ and hence a perfectly diversified system is the most stable.
\end{example} 
Next, we consider a financial system that comprises of $U\in\N$ subsystems of equal size $n/U$. For each subsystem $u\in[U]$ there shall exist a set of $D_u=D\in\N$ specialized assets that can only be invested in by institutions from subsystem $s$. In addition to these $U\cdot D$ specialized assets, there shall exist a set of $J\in\N$ joint assets that can be invested in by any institution of the whole system and that hence connect the different subsystems. Thus, each institution can choose from $\Delta:=D+J$ different assets to invest in. We call $\Delta$ the \emph{diversification} of the system. Further, for each institution a fraction $\Sigma:=J/\Delta$ of its available assets is available also to every other institution in the system. We call $\Sigma$ the \emph{(portfolio) similarity} in the system. Then, as in Example \ref{4:ex:diversification} we could compute the optimal investments for each institution (which is shifted towards investing in the specialized assets to avoid overlap with other subsystems). Instead we assume in the following example that each institution still perfectly diversifies its investment over the $D+J$ assets available to it. This is reasonable if the single institutions do not have a perfect overview of the whole financial system. The effect of diversification $\Delta$ and similarity $\Sigma$ is similar for the two different allocations.
\begin{example}\label{4:ex:diversification:similarity}
Consider a system as described above consisting of $U$ subsystems and allowing each institution to invest in $D$ specialized assets and in $J$ joint assets in equal shares. Then the system is described by the following functions:
\begin{align*}
f^j(\bm{\chi}) &:= U^{-1}\sum_{u=1}^U\E\left[\frac{X^\text{tot}}{D+J}\1\left\{\frac{X^\text{tot}}{D+J}\left(\sum_{k=1}^J\left(1-e^{-\chi^k}\right)+\sum_{d=1}^D\left(1-e^{-\chi^{u,d}}\right)\right)\geq C\right\}\right] - \chi^j,\\
f^{u,d}(\bm{\chi}) &:= U^{-1}\E\left[\frac{X^\text{tot}}{D+J}\1\left\{\frac{X^\text{tot}}{D+J}\left(\sum_{j=1}^J\left(1-e^{-\chi^j}\right)+\sum_{e=1}^D\left(1-e^{-\chi^{u,e}}\right)\right)\geq C\right\}\right] - \chi^{u,d},
\end{align*}
where $j\in[J]$, $u\in[U]$, $d\in[D]$ and $\bm{\chi}=(\chi^1,\ldots,\chi^J,\chi^{1,1},\ldots,\chi^{U,D})\in\R_{+,0}^{J+UD}$ with small misuse of notation. Similar as in Example \ref{4:ex:diversification} we derive that
\[ \gamma_c=3-\beta\quad \text{and}\quad \alpha_c=\frac{J+\frac{D}{U}}{(D+J)^2}\frac{\beta-1}{\beta-2}=\frac{1+(U-1)\Sigma}{\Delta U}\frac{\beta-1}{\beta-2}. \]
From the formula it is obvious that $\alpha_c$ decreases (i.\,e.~stability of the system increases) as $\Delta$ increases or $\Sigma$ decreases. 
\end{example}
Example \ref{4:ex:diversification:similarity} hence shows that diversification makes the system more stable (as already seen in Example \ref{4:ex:diversification}) whereas stronger similarity between the institutions' portfolios makes the system more fragile.

\begin{figure}[t]
    \hfill\subfigure[]{\includegraphics[width=0.43\textwidth]{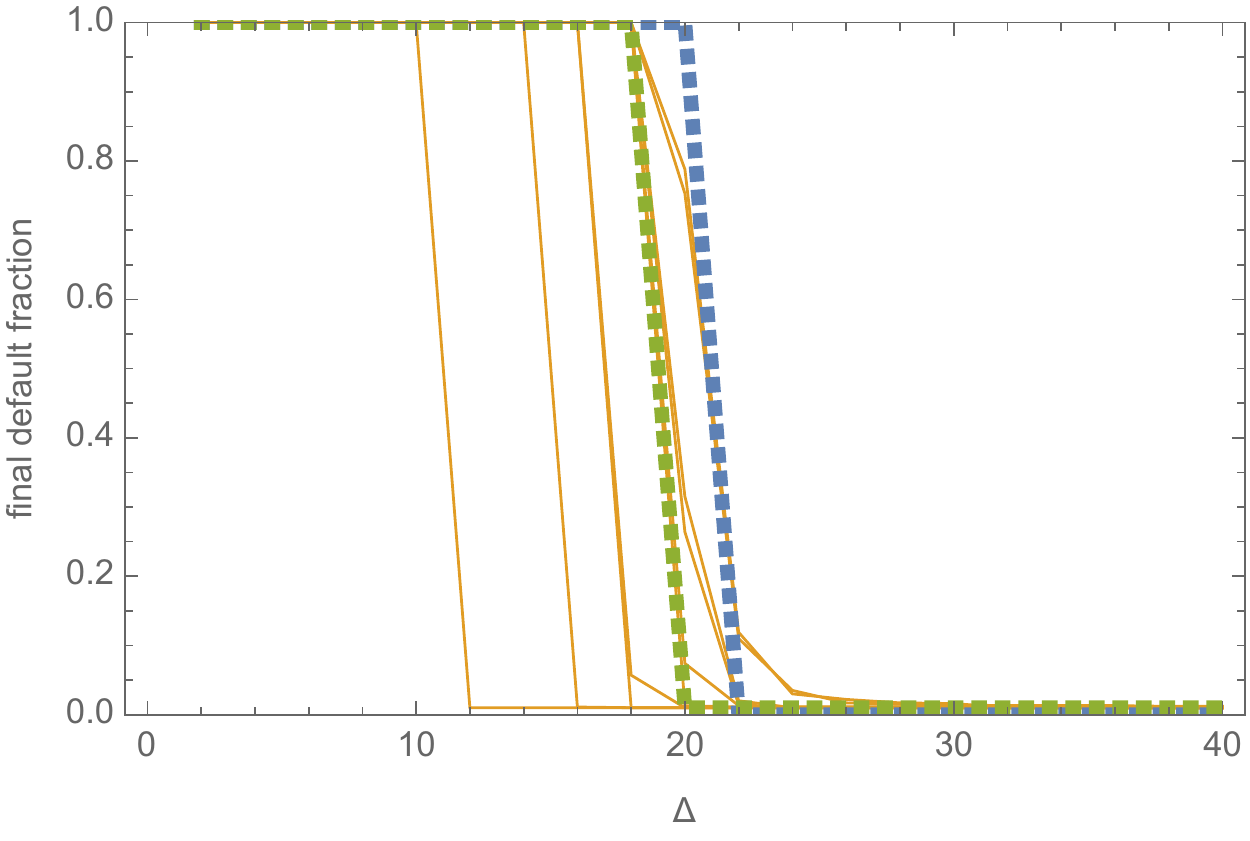}\label{4:fig:varying:diversification}}
    \hfill\subfigure[]{\includegraphics[width=0.43\textwidth]{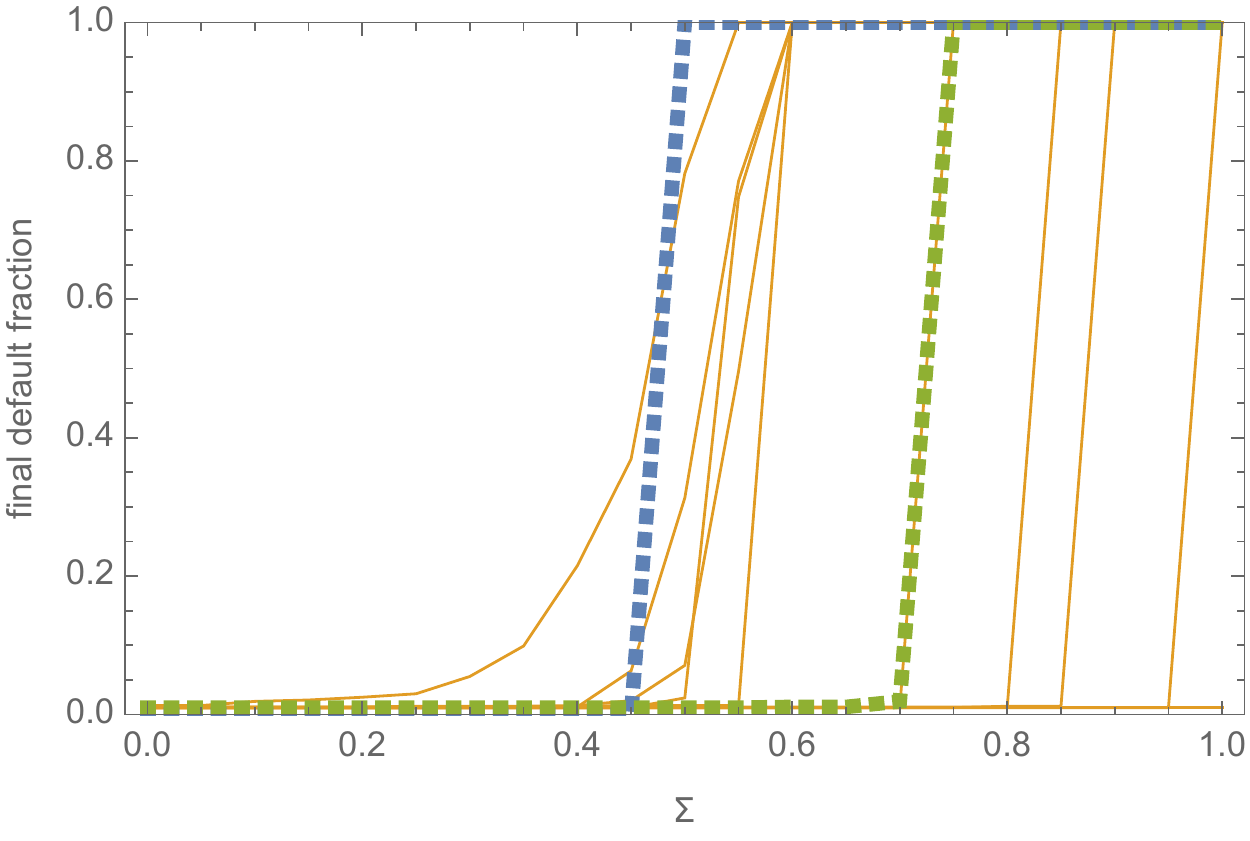}\label{4:fig:varying:similarity}}\hfill
\caption{(a)~The effect of varying portfolio diversification $\Delta$ as $\Sigma=0.5$ is fixed. (b)~The effect of varying portfolio similarity as $\Delta=20$ is fixed. In blue: the theoretical final default fraction. In orange: $10$ exemplary simulations. In green: the median over $10^3$ simulations.}
\vspace*{-0.25cm}
\end{figure}

Note that all previous conclusions build on the (asymptotic) theory from Sections \ref{4:sec:fire:sales} and \ref{4:sec:resilience}. To verify and back up the result for finite systems, however, we also give a simulation based verification for a series of moderate size  ($n=10^4$) systems. We chose $\beta=3$ and $U=2$. For $D=J=10$, we then derive $\Delta=20$, $\Sigma=0.5$, $\gamma_c=0$ and $\alpha_c=0.075$. We therefore assigned to each institution the capital $c_i=\alpha_c$. Further, we drew the total asset holdings $x_i^\text{tot}$ for each institution $i\in[n]$ as random numbers according to the above described Pareto distribution. Finally, we randomly chose a set of initially defaulted institutions of size $0.01n$ and equally distributed across the $U$ subsystems. To see the effect of diversification, we first fix $\Sigma=0.5$ and let $D=J$ vary from $1$ to $20$ (i.\,e.~$\Delta\in[40]$). The results are plotted in Figure \ref{4:fig:varying:diversification}. Since we calibrated the capitals $c_i=\alpha_c$ to the values $\Delta=20$ and $\Sigma=0.5$, the theoretical (asymptotic) final default fraction is $1$ for $\Delta\leq20$ and $0$ otherwise. This curve is shown in blue. In orange we exemplarily illustrate $10$ of the $10^3$ simulations. One can see that in each simulation the final default fraction rapidly decreases at a certain value for $\Delta$ close to the theoretical value of $20$. In green finally, we plot the median over all $10^3$ simulations which is very close to the theoretical curve despite the finite system size and hence verifies that systems become more resilient as $\Delta$ increases. Deviations from the theoretical curve 
become smaller as $n$ increases. 

Furthermore, in the same setting we conducted simulations for systems of fixed diversification $\Delta=20$ and varying similarity $\Sigma$ between $0$ and $1$ ($J\in[0,20]$ and $D=20-J$). The results are shown in Figure \ref{4:fig:varying:similarity}. Again, in blue we plot the theoretically predicted curve which is $0$ for $\Sigma<0.5$ and $1$ otherwise. In orange $10$ exemplary simulations are shown. For these, one can see that either there exists an individual threshold for $\Sigma$ close to $0.5$ at which the final default fraction rapidly increases or the final fraction stays constant at $1\%$. The median over the $10^3$ simulations for each $\Sigma$ can be seen in green 
and it verifies that the system becomes less resilient as the similarity $\Sigma$ increases. Again deviations from the theoretical curve 
are due to finite size effects and become smaller as $n$ increases.

\subsection{Testing the Capital Requirements by Simulations}\label{4:ssec:initial:asset:shocks}

In contrast to the previous subsection, we will consider initial stress in a system in form of shocks on asset prices. There are then two dimensions to consider regarding stability of a system. First, for each institution the probability of initial default should be small. Second, in the rare event that some institutions become initially distressed the remaining capital of the institutions still needs to be high enough to stop the spread of fire sales. The latter of the two is precisely the systemic risk capital derived in Subsection \ref{4:ssec:capital:requirements}. To further ensure rare initial distress we increase capital $c_i$ by $i$'s \emph{value-at-risk} with respect to some level $\epsilon>0$ which is a classical risk capital for example used in the \emph{Basel III} framework. In that sense, our systemic risk capital becomes a \emph{systemic risk surcharge} to the classical risk capital:
\begin{equation}\label{4:eqn:value:at:risk:plus:surcharge}
c_i = \text{value-at-risk}(i) + \text{systemic risk surcharge}(i),\qquad i\in[n]
\end{equation}
The aim of this subsection is  to verify by simulations that capitals of the form as in \eqref{4:eqn:value:at:risk:plus:surcharge} indeed ensure resilience of a system to initial asset shocks.

In this analysis, we further want to demonstrate the effect of different system characteristics.
For simplicity, we choose to consider two subsystems $S_1=[n/2]\subset[n]$ and $S_2=[n]\backslash S_1$; the considerations extend readily to a larger number of subsystems. Further, there are two assets $A$ and $B$ (i.\,e.~$M=2$) with uncorrelated price changes. Denote by $x_i^\text{tot}=x_i^A+x_i^B$ the total number of shares held by institution $i\in[n]$ and denote the limiting random variable by $X^\text{tot}$. We then consider the following three scenarios, see also Figure \ref{4:fig:three:scenarios}:
\begin{enumerate}[(a)]
\itemsep0pt
\item The two subsystems invest in different assets. 
\item Each institution's portfolio is perfectly diversified. 
\item All institutions in the system invest in the same asset. 
\end{enumerate}
\begin{figure}[t]
\subfigure[]{\includegraphics[width=0.28\textwidth]{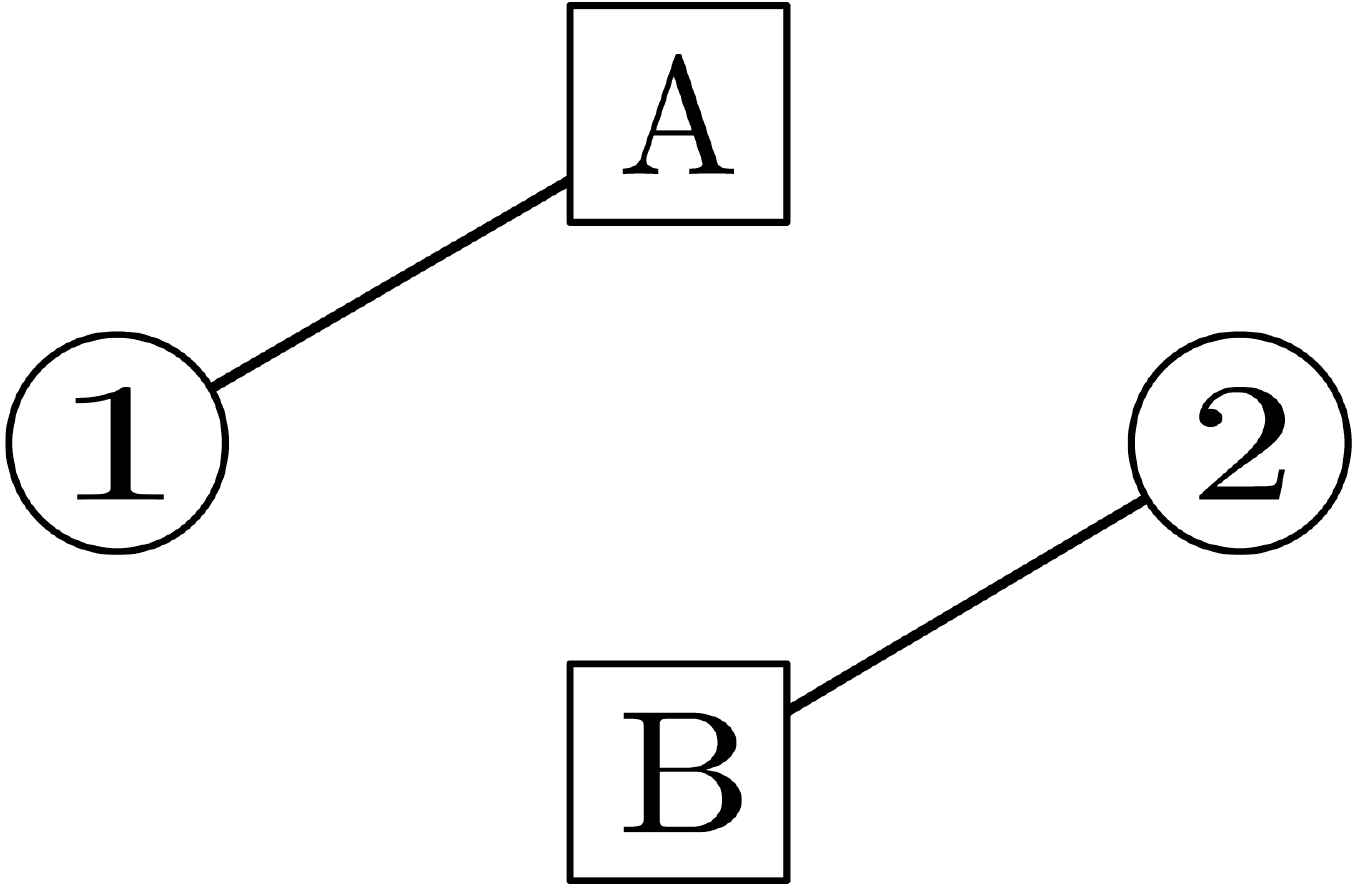}\label{4:fig:disjoint}}
    \hfill\subfigure[]{\includegraphics[width=0.28\textwidth]{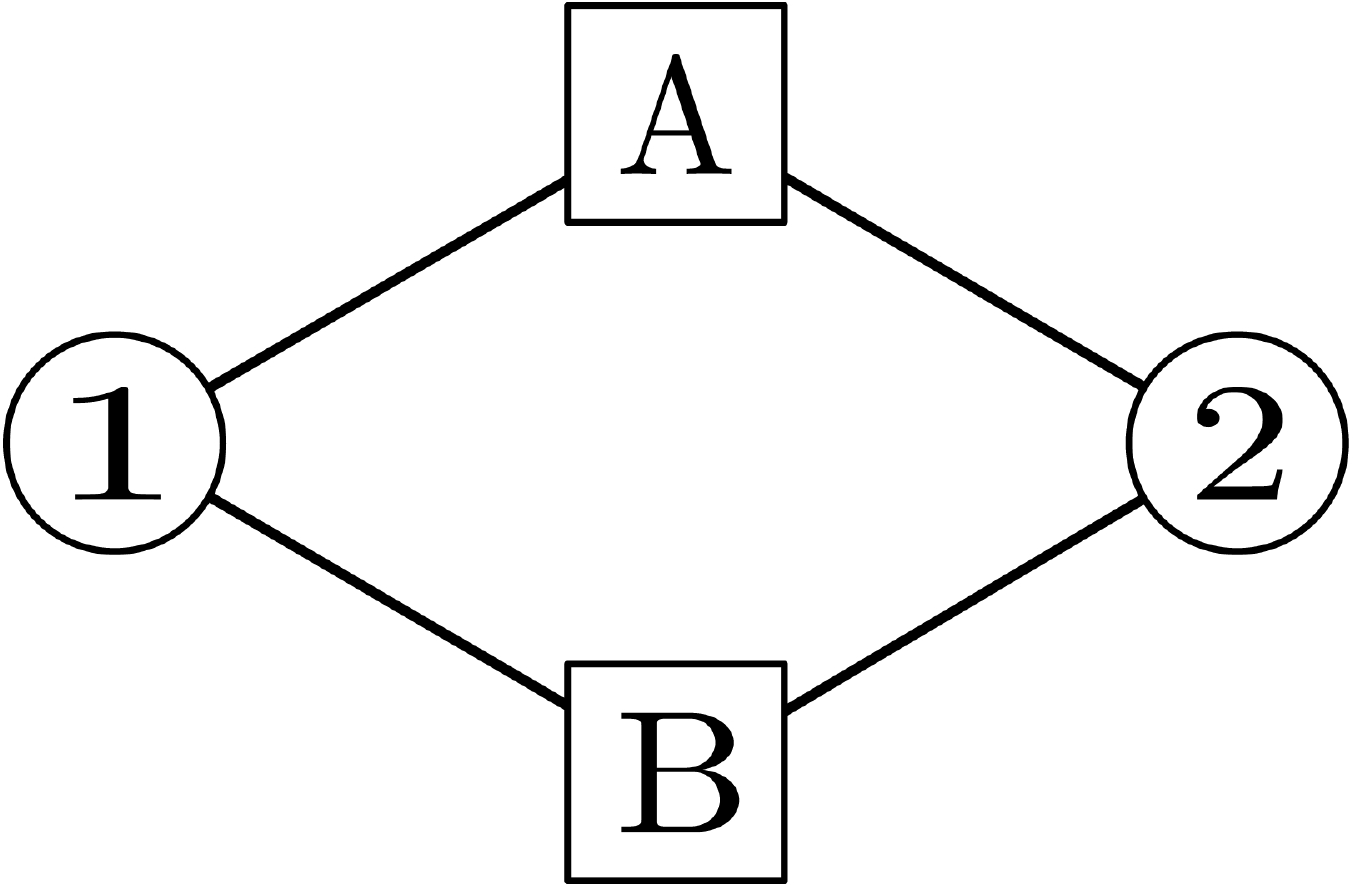}\label{4:fig:diversified}}
    \hfill\subfigure[]{\includegraphics[width=0.28\textwidth]{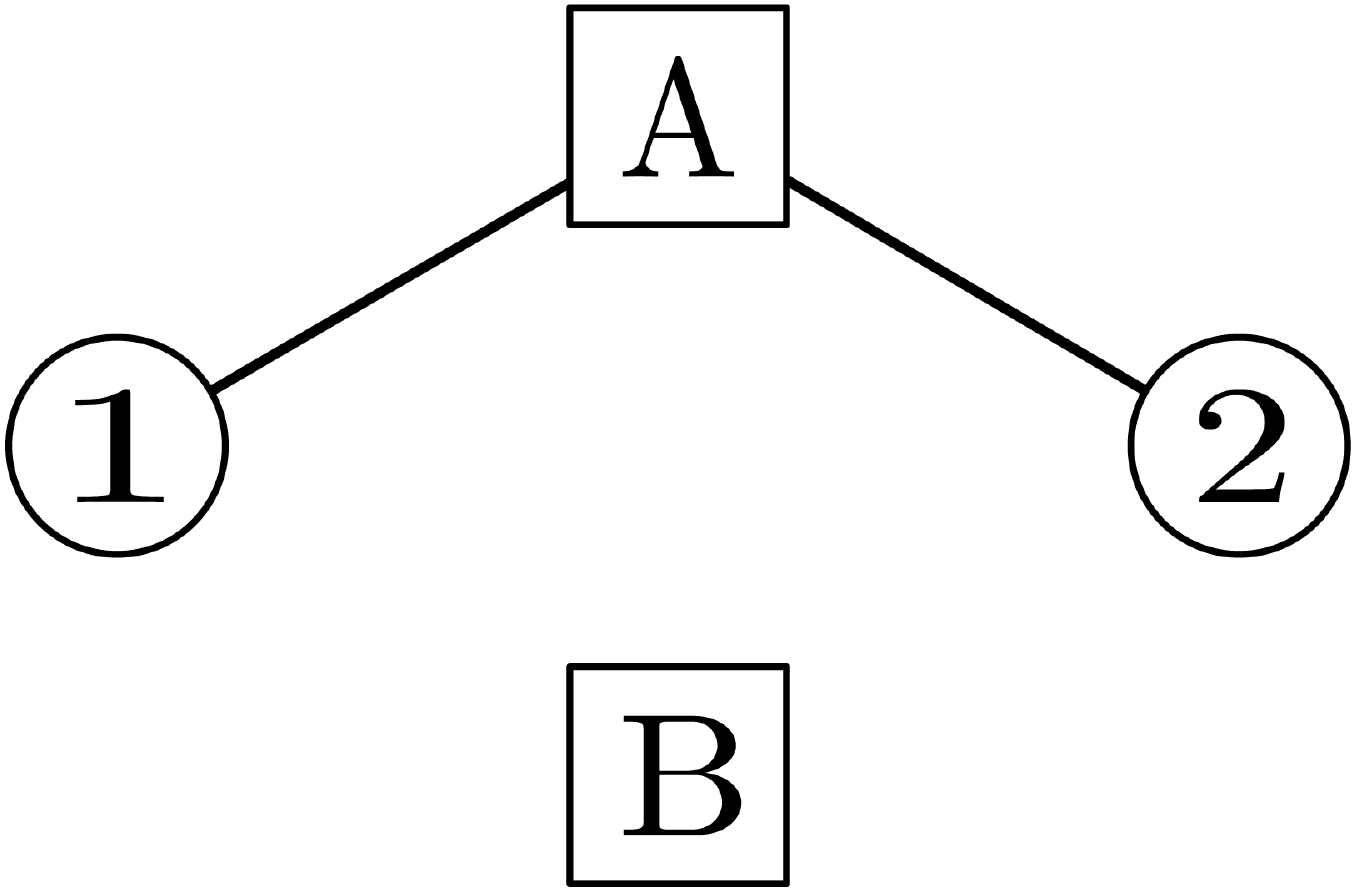}\label{4:fig:bubble}}
\caption{Illustrations of different system configurations for subsystems $1$ and $2$, and assets $A$ and $B$. (a) Separated undiversified subsystems. (b) Connected diversified subsystems. (c)~Connected undiversified subsystems.}\label{4:fig:three:scenarios}
\vspace*{-0.25cm}
\end{figure}
More formally, if we denote by $\pi^{jA}$ resp.~$\pi^{jB}$ the proportion invested in assets $A$ and $B$ by an institution in subsystem $S_j$, $j=1,2$, we can express the three scenarios as follows:
\begin{enumerate}[(a)]
\itemsep0pt
\item $\pi^{1A}=\pi^{2B}=1$ and $\pi^{1B}=\pi^{2A}=0$.
\item $\pi^{1A}=\pi^{1B}=\pi^{2A}=\pi^{2B}=1/2$.
\item $\pi^{1A}=\pi^{2A}=1$ and $\pi^{1B}=\pi^{2B}=0$.
\end{enumerate}
\begin{remark}
In terms of the dimensions $\Delta$ and $\Sigma$ from the previous subsection we can characterize the three cases by
\[ \hspace{1.7cm} \text{(a) }\Delta=1,~\Sigma=0, \hspace{1.7cm} \text{(b) }\Delta=2,~\Sigma=1, \hspace{1.7cm} \text{(c) }\Delta=1,~\Sigma=1.\hspace{1.7cm} \]
From the previous results we would therefore expect that configurations (a) and (b) are more stable than configuration (c), while we cannot make any statements about the relation between (a) and (b). Note, however, that in the previous subsection we considered the initial shock to consist of initial defaults of some institutions independent of the asset holdings. In this subsection, we will consider initial shocks on the financial system by reducing asset prices (see below). It will turn out that diversification can then have a positive effect on stability of the system as in \cite{Frey2018} because the variance of each institution's initial loss becomes smaller, or it can have a negative effect as in \cite{Beale2011,Ibragimov2011,Wagner2010} because the system is exposed to more assets or becomes more connected via assets.
\end{remark} 
Again, we assume that $\rho(u)=\1\{u\geq1\}$ and $h^m(\bm{\chi})=1-e^{-\chi^m}$, $m=A,B$. Further, let
\begin{equation}\label{4:eqn:density:total:assets}
f_{X^\text{tot}}(x) = (\beta-1)x^{-\beta}\left(\kappa\1\{1\leq x\leq b\} + \mu^{\beta-1}\1\{x>b\}\right),
\end{equation}
where $\beta=2.5$, $\mu=0.25$, $b=10^{1/(\beta-1)}\mu$ and $\kappa=(1-\mu^{\beta-1}b^{1-\beta})/(1-b^{1-\beta})$. That is the tail of the distribution resembles a Pareto distribution with exponent $\beta$ and $b$ is chosen such that the tail describes $10\%$ of the probability mass. For $x\leq b$ the exponent $\beta$ stays the same but the coefficient $\kappa$ is chosen such that the remaining mass of $90\%$ is distributed on the interval $[1,b]$ (instead of $[\mu,b]$). Note that for the computations of $\gamma_c$ and $\alpha_c$ in the previous subsection only the tail of the distribution was relevant and so we know that $\alpha_c=\frac{1}{2}\frac{\beta-1}{\beta-2}\mu^{\beta-1}$ for cases (a) and (b) resp.~$\alpha_c=\frac{\beta-1}{\beta-2}\mu^{\beta-1}$ for case (c) and $\gamma_c=3-\beta$ in all cases. We consider these as the systemic risk surcharges as discussed above. Further the \emph{value-at-risk} capital for some institution $i\in[n]$ is given by $\theta x_i^\text{tot}$ for some global parameter $\theta\in[0,1]$ that needs to be calibrated to the confidence level $\epsilon$ and the initial shock distribution (see below). The piecewise form of $f_{X^\text{tot}}$ in \eqref{4:eqn:density:total:assets} rather than for example a perfect Pareto distribution ensures that capital $c_i$ is actually smaller then the maximum potential loss $x_i^\text{tot}$ also for institutions with small investments (note that the systemic surcharge $\alpha_c$ is determined from the tail of the distribution).

\begin{table}[t]
\begin{center}
\begin{tabular}{|c|c|c|c|}\hline
\textbf{Case} & \textbf{$\varnothing$ initially infected fraction} & \textbf{$\varnothing$ finally infected fraction} & \textbf{amplification}\\\hline
(a) & $1.35\%$ & $3.82\%$ & $1.83$\\
(b) & $0.80\%$ & $3.26\%$ & $3.06$\\
(c) & $0.56\%$ & $2.13\%$ & $2.84$\\
(c') & $1.35\%$ & $5.00\%$ & $2.70$\\\hline
\end{tabular}
\caption{Simulation results for capital allocation determined by the value-at-risk plus the systemic risk surcharge}\label{4:tab:VaR+riskCapital}
\vspace{1ex}
\begin{tabular}{|c|c|c|c|}\hline
\textbf{Case} & \textbf{$\varnothing$ initially infected fraction} & \textbf{$\varnothing$ finally infected fraction} & \textbf{amplification}\\\hline
(a) & $2.38\%$ & $8.95\%$ & $2.76$\\
(b) & $1.76\%$ & $15.1\%$ & $7.53$\\
(c) & $0.84\%$ & $5.58\%$ & $5.66$\\
(c') & $2.38\%$ & $10.0\%$ & $3.21$\\\hline
\end{tabular}
\caption{Simulation results for capital allocation determined by the systemic risk surcharge}\label{4:tab:riskCapital}
\end{center}
\vspace*{-0.5cm}
\end{table}

Finally, we model the initial shock on asset $m=A,B$ as $e^{-R^m}$, where $R^A\stackrel{d}{=}R^B$ are independent random variables such that $\P(R^A=0)=90\%$ and with the remaining probability of $10\%$ it holds $R^A=\vert T\vert/10$ for $T$ having Student's $t$-distribution with $1.5$ degrees of freedom (Student's $t$-distribution is a popular choice in market models since heavy tails can be modeled by less than $2$ degrees of freedom). 

For each realization of $(R^A,R^B)$ we can then numerically determine the asymptotic final default fraction as in Section \ref{4:sec:fire:sales}. We choose $\epsilon=5\%$ which yields $\theta\approx 8.36\%$ in cases (a) and (c) resp.~$\theta\approx 8.47\%$ in case (b) as the parameter for the \emph{value-at-risk} capital. The numerical average final default fraction for the three configurations is listed in Table \ref{4:tab:VaR+riskCapital}. The most stable configuration is (c). Recall, however, that in this case $\alpha_c$ is the double amount than in cases (a) and (b). Therefore, we further included case (c') where we adjusted the value of $\alpha_c$ accordingly. It can then be seen that this configuration is in fact the least stable one and (b) becomes the most stable one. Diversification is thus beneficial for the capital allocation value-at-risk plus systemic risk surcharge. In comparison to this, we performed the same simulations with capitals determined solely by the systemic risk charge. The results are listed in Table \ref{4:tab:riskCapital}. Clearly, the system becomes less stable. In particular, the least stable configuration is now (b), the diversified one. For completeness also consider Table \ref{4:tab:VaR} listing the average fractions for systems equipped exactly with the value-at-risk as capital. In this case, all final infections are already initial infections and there is no amplification (final infection divided by initial infection minus 1). So rather counterintuitively our systemic risk surcharge increases the amplification. This is because due to the additional capital some institutions are initially saved from infection but become infected in the course of the fire sales process. Overall, the combination of value-at-risk with the systemic risk surcharge significantly increases stability of the financial system.

\begin{table}[t]
\begin{center}
\begin{tabular}{|c|c|c|c|}\hline
\textbf{Case} & \textbf{$\varnothing$ initially infected fraction} & \textbf{$\varnothing$ finally infected fraction} & \textbf{amplification}\\\hline
(a) & $5.00\%$ & $5.00\%$ & $0.00$\\
(b) & $5.00\%$ & $5.00\%$ & $0.00$\\
(c) & $5.00\%$ & $5.00\%$ & $0.00$\\\hline
\end{tabular}
\caption{Simulation results for capital allocation determined by the value-at-risk}\label{4:tab:VaR}
\end{center}
\end{table}

\section{Proofs}\label{4:sec:proofs}
\subsection{Proofs for Section \ref{4:sec:fire:sales}}\label{4:ssec:proofs:2}
\begin{proof}[Proof of Lemma \ref{4:lem:existence:chi:hat}]
Existence of $\hat{\bm{\chi}}$ follows from the Knaster-Tarski theorem. We now construct a joint root $\accentset{\circ}{P}_0\ni\bar{\bm{\chi}}\leq\hat{\bm{\chi}}$ such that we can conclude $\hat{\bm{\chi}}=\bar{\bm{\chi}}\in \accentset{\circ}{P}_0$.

It holds $\circFSuper{m}(\hat{\bm{\chi}})=0$ for all $m\in[M]$ and thus (for any fixed $m\in[M]$) $\circFSuper{m}(\bm{\chi})\leq0$ for all $\hat{\bm{\chi}}\geq\bm{\chi}\in\R_{+,0}^M$ such that $\chi^m=\hat{\chi}^m$ by monotonicity of $\circFSuper{m}$. 
Consider then the following sequence $(\bm{\chi}_{(k)})_{k\in\N}\subset\R_{+,0}^M$:
\begin{itemize}
\item $\bm{\chi}_{(0)}=\bm{0}\in \accentset{\circ}{P}_0$
\item $\bm{\chi}_{(1)}=(\chi_{(1)}^1,0,\ldots,0)$, where $0\leq\chi_{(1)}^1\leq\hat{\chi}^1$ is the smallest possible value such that $\circFSuper{1}(\bm{\chi}_{(1)})=0$. It is possible to find such $\chi_{(1)}^1$ since $\circFSuper{1}(\bm{\chi})+\chi^1$ is monotonically increasing in $\chi^1$, $\circFSuper{1}(\bm{0})\geq 0$ and $\circFSuper{1}(\hat{\chi}^1,0,\ldots,0)\leq0$. By monotonicity of $\circFSuper{m}$ with respect to $\chi^1$ for all $m\in[M]\backslash\{1\}$, it then holds $\circFSuper{m}(\bm{\chi}_{(1)})\geq \circFSuper{m}(\bm{0})\geq0$ for all $1\neq m\in [M]$ and in particular $\bm{\chi}_{(1)}\in \accentset{\circ}{P}_0$.
\item $\bm{\chi}_{(2)}=\bm{\chi}_{(1)}+(0,\chi_{(2)}^2,0,\ldots,0)$, where $0\leq\chi_{(2)}^2\leq\hat{\chi}^2$ is the smallest value such that $\circFSuper{2}(\bm{\chi}_{(2)})=0$. Again it is possible to find such $\chi_{(2)}^2$ since $\circFSuper{2}(\bm{\chi})+\chi^2$ is monotonically increasing in $\chi^2$, $\circFSuper{2}(\bm{\chi}_{(1)})\geq0$ and $\circFSuper{2}(\bm{\chi}_{(1)}+(0,\hat{\chi}^2,0,\ldots,0))\leq0$. By monotonicity of $\circFSuper{m}$ with respect to $\chi^2$ for all $m\in[M]\backslash\{2\}$, it then holds $\circFSuper{m}(\bm{\chi}_{(2)})\geq \circFSuper{m}(\bm{\chi}_{(1)})\geq0$ for all $2\neq m\in [M]$ and in particular $\bm{\chi}_{(2)}\in \accentset{\circ}{P}_0$.
\item $\bm{\chi}_{(i)}$, $i\in\{3,\ldots,M\}$, are found analogously, changing only the corresponding coordinate.
\item $\bm{\chi}_{(M+1)}=\bm{\chi}_{(M)}+(\chi_{(M+1)}^1-\chi_{(M)}^1,0,\ldots,0)$, where $\chi_{(M)}^1\leq\chi_{(M+1)}^1\leq\hat{\chi}^1$ is the smallest value such that $\circFSuper{1}(\bm{\chi}_{(M+1)})=0$, which is again possible by monotonicity of $\circFSuper{1}(\bm{\chi})+\chi^1$
, $\circFSuper{1}(\bm{\chi}_{(M)})\geq0$ and $\circFSuper{1}(\bm{\chi}_{(M)}+(\hat{\chi}^1-\chi_{(M)}^1,0,\ldots,0))\leq0$. Further, it still holds $\bm{\chi}_{(M+1)}\in \accentset{\circ}{P}_0$.
\item Continue for $\bm{\chi}_{i}$, $i\geq M+2$.
\end{itemize}
The sequence $(\bm{\chi}_{(k)})_{k\in\N}$ constructed this way has the following properties: It is non-decreasing in each coordinate and 
bounded inside $[\bm{0},\hat{\bm{\chi}}]$. Hence by monotone convergence, each coordinate of $\bm{\chi}_{(k)}$ converges and so $\bar{\bm{\chi}}=\lim_{k\to\infty}\bm{\chi}_{(k)}$ exists. Since the convergence is from below, it holds
\begin{align*}
\circFSuper{m}(\bar{\bm{\chi}}) &= \E\left[X^m\circRho\left(\frac{L+\bm{X}\cdot h(\lim_{k\to\infty}\bm{\chi}_{(k)})}{C}\right)\right] - \lim_{k\to\infty}\chi_{(k)}^m\\
&= \lim_{k\to\infty}\E\left[X^m\circRho\left(\frac{L+\bm{X}\cdot h(\bm{\chi}_{(k)})}{C}\right)\right]-\chi_{(k)}^m = \lim_{k\to\infty}\circFSuper{m}(\bm{\chi}_{(k)}) \geq 0
\end{align*}
and thus $\bar{\bm{\chi}}\in\accentset{\circ}{P}_0$. 
Now suppose there is $m\in[M]$ such that $\circFSuper{m}(\bar{\bm{\chi}})>0$. By lower semi-continuity of $\circFSuper{m}$ then also $\circFSuper{m}(\bm{\chi}_{(k)})>\epsilon$ for some $\epsilon>0$ and $k$ large enough. This, however, is a contradiction to the construction of the sequence $(\bm{\chi}_{(k)})_{k\in\N}$ since $\circFSuper{m}(\bm{\chi}_{(k)})=0$ in every $M$-th step. Hence $\circFSuper{m}(\bar{\bm{\chi}})= 0$ for all $m\in[M]$ and $\bar{\bm{\chi}}$ is a joint root of all functions $\circFSuper{m}$, $m\in[M]$.

Now turn to the proof that $\bm{\chi}^*\in P_0$
: We first consider the case that $\rho$ is continuous. We approximate $\bm{\chi}^*\in P_0$ by the sequence $(\hat{\bm{\chi}}(\epsilon))_{\epsilon>0}$ of smallest fixpoints for the functions $f^m(\bm{\chi})+ \epsilon$. This allows us to apply the Knaster-Tarski Theorem and the monotonicity properties of $f^m+ \epsilon$ similar as above. Simple topological arguments will then allow us to conclude that $\bm{\chi}^*\in P_0$. Let for $\epsilon>0$
\[ P(\epsilon) := \bigcap_{m\in[M]}\left\{\bm{\chi}\in\R_{+,0}^M\,:\,f^m(\bm{\chi})\geq -\epsilon\right\} \]
and denote by $P_0(\epsilon)$ the connected component of $\bm{0}$ in $P(\epsilon)$. By the same procedure as for $\hat{\bm{\chi}}$ above, we now derive that there exists a smallest (componentwise) point $\hat{\bm{\chi}}(\epsilon)\in P_0(\epsilon)$ such that $f^m(\hat{\bm{\chi}}(\epsilon))=-\epsilon$ for all $m\in[M]$. Clearly, $\hat{\bm{\chi}}(\epsilon)$ is non-decreasing (componentwise) in $\epsilon$ and hence $\tilde{\bm{\chi}}:=\lim_{\epsilon\to0+}\hat{\bm{\chi}}(\epsilon)$ exists (we will show that $\tilde{\bm{\chi}}=\bm{\chi}^*$ in fact).

Now by monotonicity of $P_0(\epsilon)$, we derive that $\hat{\bm{\chi}}(\delta)\in P_0(\delta)\subseteq P_0(\epsilon)$ for all $\delta\leq \epsilon$. Since $P_0(\epsilon)$ is a closed set, it must thus hold that also $\tilde{\bm{\chi}}=\lim_{\delta\to0+}\hat{\bm{\chi}}(\delta)\in P_0(\epsilon)$ for all $\epsilon>0$ and in particular, $\tilde{\bm{\chi}}\in\bigcap_{\epsilon>0}P_0(\epsilon)$. Further, we derive that $\bigcap_{\epsilon>0}P_0(\epsilon)\subseteq\bigcap_{\epsilon>0}P(\epsilon)\subseteq P$. Moreover, $\bigcap_{\epsilon>0}P_0(\epsilon)$ is the intersection of a chain of connected, compact sets in the Hausdorff space $\R^M$ and it is hence a connected, compact set itself. Since it further contains $\bm{0}$, we can then conclude that $\bigcap_{\epsilon>0}P_0(\epsilon)\subseteq P_0$ and thus $\tilde{\bm{\chi}}\in P_0$.

Consider now an arbitrary $\bm{\chi}\in P_0$. We want to show that $\bm{\chi}\leq\tilde{\bm{\chi}}$ and thus $\tilde{\bm{\chi}}=\bm{\chi}^*$. It suffices to show that $P_0\subset[\bm{0},\hat{\bm{\chi}}(\epsilon)]$  for all $\epsilon$. Then $\bm{\chi}\leq\hat{\bm{\chi}}(\epsilon)$ and \mbox{$\bm{\chi}\leq\lim_{\epsilon\to0+}\hat{\bm{\chi}}(\epsilon)=\tilde{\bm{\chi}}$}. Hence assume that $P_0\not\subset[\bm{0},\hat{\bm{\chi}}(\epsilon)]$. By connectedness of $P_0$ we find $\bar{\bm{\chi}}\in P_0$ with $\bar{\chi}^m\leq\hat{\chi}^m(\epsilon)$ for all $m\in[M]$ and equality for at least one coordinate (otherwise $P_0\cap\partial[\bm{0},\hat{\bm{\chi}}(\epsilon)]=\emptyset$ and $P_0=\left(P_0\cap\left(\R_{+,0}^M\backslash[\bm{0},\hat{\bm{\chi}}(\epsilon)]\right)\right)\cup\left(P_0\cap[\bm{0},\hat{\bm{\chi}}(\epsilon))\right)$ is the union of two open non-empty sets and hence not connected). W.\,l.\,o.\,g.~let this coordinate be $\bar{\chi}^1$. By monotonicity of $f^1$ with respect to $\chi^m$ for every $1\neq m\in[M]$, we thus derive that $f^1(\bar{\bm{\chi}})\leq f^1(\hat{\bm{\chi}}(\epsilon))=-\epsilon<0$ which is a contradiction to $\bar{\bm{\chi}}\in P_0$.

Now consider the general case that $\rho$ is right-continuous and let $(\rho_r(u))_{r\in\N}$ be a sequence of continuous sale functions approximating $\rho$ from above. Denoting by $P^r$ the analogue of $P$ for the sale function $\rho_r$, we derive that $P=\bigcap_{r\in\N}P^r$ since clearly $P^r\supseteq P$ for all $r\in\N$ and further by dominated convergence for every $\bm{\chi}\in\bigcap_{r\in\N}P^r$,
\[ \chi^m \leq \E\left[X^m\rho_r\left(\frac{L+\bm{X}\cdot h(\bm{\chi})}{C}\right)\right] \to \E\left[X^m\rho\left(\frac{L+\bm{X}\cdot h(\bm{\chi})}{C}\right)\right],\quad\text{as }r\to\infty, \]
so that $\bigcap_{r\in\N}P^r\subseteq P$. If we further let $P_0^r$ denote the largest connected subset of $P^r$ containing $\bm{0}$, then $P_0^r$ is compact and connected for every $r\in\N$ and hence so is $\bigcap_{r\in\N}P_0^r$. Since further $\bm{0}\in\bigcap_{r\in\N}P_0^r$, we derive that $\bigcap_{r\in\N}P_0^r=P_0$. Let now $\bm{\chi}_r^*$ denote the analogue of $\bm{\chi}^*$ for the sale function $\rho_r$. Then $\lim_{r\to\infty}\bm{\chi}_r^*\in P_0^R$ for all $R\in\N$ and hence $\lim_{r\to\infty}\bm{\chi}_r^*\in\bigcap_{R\in\N}P_0^R=P_0$. Now suppose there existed a vector $\bm{\chi}\in P_0$ and $m\in[M]$ such that $\chi^m>\lim_{r\to\infty}(\chi_r^*)^m$. Then also for $R$ large enough, $\chi^m>(\chi_R^*)^m$ and hence $\bm{\chi}\not\in P_0^R$. This, however, contradicts the assumption that $\bm{\chi}\in P_0=\bigcap_{R\in\N}P_0^R$. Hence there exists no such $\bm{\chi}\in P_0$ and $\bm{\chi}^*=\lim_{r\to\infty}\bm{\chi}_r^*\in P_0$.

Finally, we show that $\bm{\chi}^*$ is a joint root of $f^m$, $m\in[M]$: Since $\bm{\chi}^*\in P_0$, it holds that $f^m(\bm{\chi}^*)\geq0$ for all $m\in[M]$. Assume now that $f^m(\bm{\chi}^*)>0$ for some $m\in[M]$. We can then gradually increase the $m$-coordinate of $\bm{\chi}^*$ (until $f^m(\bm{\chi}^*)=0$). By monotonicity of $f^k(\bm{\chi})$ with respect to $\chi^m$ for every $m\neq k\in[M]$, however, we can be sure that we do not leave the set $P_0$ by this procedure which is a contradiction to the definition of $\bm{\chi}^*$. Hence $\bm{\chi}^*$ is a joint root of $f^m$, $m\in[M]$.
\end{proof}

\begin{remark}\label{4:rem:sequence:chi:*}
In the proof of Lemma \ref{4:lem:existence:chi:hat}, for the case that $\rho$ is continuous, we constructed $\bm{\chi}^*$ as the limit of a sequence $(\hat{\bm{\chi}}(\epsilon))_{\epsilon>0}$ such that $f^m(\hat{\bm{\chi}}(\epsilon))=-\epsilon$ for all $m\in[M]$. For non-continuous $\rho$ by the Knaster-Tarski theorem we still know that there exists a smallest vector $\hat{\bm{\chi}}(\epsilon)$ such that $f^m(\hat{\bm{\chi}}(\epsilon))=-\epsilon$, but the construction of $\hat{\bm{\chi}}(\epsilon)$ as for $\hat{\bm{\chi}}$ in the proof of Lemma \ref{4:lem:existence:chi:hat} fails and we can hence not be sure a priori that $\hat{\bm{\chi}}(\epsilon)\in P_0(\epsilon)$. Hence let further $\tilde{\bm{\chi}}(\epsilon)$ be defined as the smallest vector in $P_0(\epsilon)$ such that $f^m(\tilde{\bm{\chi}}(\epsilon))=-\epsilon$. This vector exists again by the Knaster-Tarski theorem noting that analogue to Lemma \ref{4:lem:existence:chi:hat} $P_0(\epsilon)$ contains its componentwise supremum $\bm{\chi}^*(\epsilon)$. Then by the same means as above, we derive that $\bm{\chi}^*=\lim_{\epsilon\to0}\tilde{\bm{\chi}}(\epsilon)$.
\end{remark}
In Theorem \ref{4:thm:fire:sale:final:fraction} we are considering a sequence of financial systems. The following lemma shows the convergence of the smallest joint roots under certain assumptions:
\begin{lemma}\label{4:lem:convergence:chi:hat}
Let a sequence (for $r\in\N$) of financial systems be described by functions $\circFSuperSub{m}{r}$, $m\in[M]$, with smallest joint root $\hat{\bm{\chi}}_r$. If $\liminf_{r\to\infty}\circFSuperSub{m}{r}(\bm{\chi})\geq\circFSuper{m}(\bm{\chi})$ pointwise for every \mbox{$m\in[M]$,} then $\liminf_{r\to\infty}\hat{\bm{\chi}}_r\geq\hat{\bm{\chi}}$, where $\hat{\bm{\chi}}$ denotes the smallest joint root of the functions $\circFSuper{m}$, $m\in[M]$.
\end{lemma}
\begin{proof}
The main difficulty in showing the result is that we have $\liminf_{r\to\infty}\circFSuperSub{m}{r}(\bm{\chi})\geq\circFSuper{m}(\bm{\chi})$ only pointwise but not uniformly in $\bm{\chi}$. A further difficulty is the multidimensionality. The main idea is to construct a path in analogy to the construction in Lemma \ref{4:lem:existence:chi:hat} that leads to a point $\tilde{\bm{\chi}}(\epsilon)$ smaller but close to $\hat{\bm{\chi}}$. On this path the functions $\circFSuperSub{m}{r}$, $m\in[M]$ are all positive for $r$ large. It can then be compared componentwise with a path leading to $\hat{\bm{\chi}}_r$.

For this consider the construction of $\hat{\bm{\chi}}$ in Lemma \ref{4:lem:existence:chi:hat} and change it in such a way that in each step $k=LM+m$ (where $L\in\N_0$ and $m\in[M]$) a point $\bm{\chi}_{(k)}(\epsilon)$ is chosen such that $\circFSuper{m}(\bm{\chi}_{(k)}(\epsilon))\leq\epsilon$ for some fixed $\epsilon>0$ (choose $\chi_{(k)}^m(\epsilon)\geq\chi_{(k-1)}^m(\epsilon)$ as the smallest possible value such that this inequality holds; it will then either be $\circFSuper{m}(\bm{\chi}_{(k)}(\epsilon))=\epsilon$ or $\bm{\chi}_{(k)}(\epsilon)=\bm{\chi}_{(k-1)}(\epsilon)$). Note that $\circFSuper{m}(\bm{\chi}_{(k)}(\epsilon))<\epsilon$ can only happen if $\circFSuper{m}(\bm{0}) <\epsilon$ in which case there exists $k_0\in \N_{\infty}$ such that  $\chi^m_{(k)}=0$ and $\circFSuper{m}(\bm{\chi}_{(k)}) <\epsilon$ for all $k\leq k_0$ but $\circFSuper{m}(\bm{\chi}_{(k)}) \geq \epsilon$ and $\chi^m_{(k)}>0$ for $k>k_0$. Then $(\bm{\chi}_{(k)}(\epsilon))_{k\in\N}$ is a non-decreasing (componentwise) sequence bounded by $\hat{\bm{\chi}}$ and hence \mbox{$\tilde{\bm{\chi}}(\epsilon)=\lim_{k\to\infty}\bm{\chi}_{(k)}(\epsilon)$} exists. Further, it holds that \mbox{$\circFSuper{m}(\tilde{\bm{\chi}}(\epsilon)) \leq \liminf_{k\to\infty}\circFSuper{m}(\bm{\chi}_{(k)}(\epsilon)) \leq \epsilon$.} Finally, $\tilde{\bm{\chi}}(\epsilon)$ is non-increasing componentwise in $\epsilon$ and bounded inside $[\bm{0},\hat{\bm{\chi}}]$ and thus the limit \mbox{$\tilde{\bm{\chi}}=\lim_{\epsilon\to0+}\tilde{\bm{\chi}}(\epsilon)$} exists. Moreover, $\circFSuper{m}(\tilde{\bm{\chi}}) \leq \liminf_{\epsilon\to0+}\circFSuper{m}(\tilde{\bm{\chi}}(\epsilon)) \leq \liminf_{\epsilon\to0+}\epsilon = 0$ and 
in particular $\tilde{\bm{\chi}}=\hat{\bm{\chi}}$.

Fix now $\delta>0$ and choose $\epsilon>0$ small enough such that $\tilde{\chi}^m(\epsilon)>\hat{\chi}^m(1-\delta)^{1/2}$ for all $m\in[M]$. Further, choose $K=K(\epsilon)\in\N$ large enough such that $\chi_{(K)}^m(\epsilon)>\tilde{\chi}^m(\epsilon)(1-\delta)^{1/2}$ for all $m\in[M]$. In particular, $\chi_{(K)}^m(\epsilon)>\hat{\chi}^m(1-\delta)$. Now note that $\hat{\bm{\chi}}_r$ can be constructed by a sequence $(\bm{\chi}_{(k,r)})_{k\in\N}$ analogue to $\hat{\bm{\chi}}$ in the proof of Lemma \ref{4:lem:existence:chi:hat} as well. We can then in each step $k\in\N$ cap the element of the constructing sequence $\bm{\chi}_{(k,r)}$ at $\bm{\chi}_{(k)}(\epsilon)$, which clearly does not increase the limit of the sequence. We want to make sure that in fact the cap is used in every step $k\leq K$ if we only choose $r$ large enough. Then we can conclude that $\hat{\bm{\chi}}_r \geq \bm{\chi}_{(K)}(\epsilon) \geq \hat{\bm{\chi}}(1-\delta)$ and hence letting $\delta\to0$, $\liminf_{r\to\infty}\hat{\bm{\chi}}_r \geq \hat{\bm{\chi}}$.

We now show that the cap is applied in every step $k\leq K$ for $r$ large enough by an induction argument. For $k=0$, clearly $\bm{\chi}_{(k,r)}=\bm{\chi}_{(k)}(\epsilon)=\bm{0}$ and the cap is applied. Now lets assume it holds for $k\leq k_0<K$. If $\bm{\chi}_{(k_0+1)}(\epsilon)=\bm{\chi}_{(k_0)}(\epsilon)$, then of course the cap is also applied in step $k_0+1$ as the sequence $\bm{\chi}_{(k,r)}$ is increasing. Otherwise, note that by definition of $\bm{\chi}_{(k_0+1)}(\epsilon)$, it holds $\circFSuper{m}(\bm{\chi})\geq\epsilon$ for all $\bm{\chi}\in\R_{+,0}^M$ such that $\chi^m\in[\chi_{(k_0)}^m(\epsilon),\chi_{(k_0+1)}^m(\epsilon)]$ and $\chi^\ell=\chi_{(k_0)}^\ell(\epsilon)=\chi_{(k_0+1)}^\ell(\epsilon)$ for all $\ell\in[M]\backslash\{m\}$. Now choose a discretization $\{\chi_j\}_{0\leq j\leq J}$ of $[\chi_{(k_0)}^m(\epsilon),\chi_{(k_0+1)}^m(\epsilon)]$ for $J<\infty$ such that $\chi_0 = \chi_{(k_0)}^m(\epsilon)$, $\chi_J=\chi_{(k_0+1)}^m(\epsilon)$ and $\chi_{j-1}<\chi_j<\chi_{j-1}+\epsilon/3$ for all $j\in[J]$. We now use the assumption that $\liminf_{r\to\infty}\circFSuperSub{m}{r}(\bm{\chi}_j)\geq\circFSuper{m}(\bm{\chi}_j)$ for every $0\leq j\leq J$, where $\chi^m_j=\chi_j$ and $\chi^\ell_j = \chi_{(k_0)}^\ell(\epsilon)$ for $\ell\in[M]\backslash\{m\}$. Then for $r$ large enough, $\circFSuperSub{m}{r}(\bm{\chi}_j)\geq\circFSuper{m}(\bm{\chi}_j)-\epsilon/3\geq 2\epsilon/3$. Finally, for any linear interpolation $\bm{\chi}=\alpha\bm{\chi}_{j-1}+(1-\alpha)\bm{\chi}_j$ between $\bm{\chi}_{j-1}$ and $\bm{\chi}_j$ ($\alpha\in[0,1]$), it holds
\[ \circFSuperSub{m}{r}(\bm{\chi}) \geq \circFSuperSub{m}{r}(\bm{\chi}_{j-1}) + \chi_{j-1}^m-\chi_j^m \geq 2\epsilon/3 - \epsilon/3 = \epsilon/3. \]
Hence the cap is applied in step $k_0+1$. As there are only finitely many steps $k\leq K$, this finishes the proof.
\end{proof}

\begin{proof}[Proof of Theorem \ref{4:thm:fire:sale:final:fraction}]
We start by proving the lower bound. Recall from Proposition \ref{4:prop:upperlowerfixedsize} that $\bm{\chi}_n\geq\hat{\bm{\chi}}_n$. Using weak convergence of 
$(\bm{X}_n,S_n,C_n,L_n)$ and approximating $\circRho$ from below by a sequence of continuous sale functions $(\rho_r)_{r\in\N}$, we derive for $U\in\R_+$ that pointwise
\begin{align*}
\liminf_{n\to\infty}\E\left[X_n^m\circRho\left(\frac{L_n+\bm{X}_n\cdot h(\bm{\chi}))}{C_n}\right)\right] &\geq \lim_{n\to\infty}\E\left[\left(X_n^m\wedge U\right)\rho_r\left(\frac{L_n+\bm{X}_n\cdot h(\bm{\chi}))}{C_n}\right)\right]\\
&= \E\left[\left(X^m\wedge U\right)\rho_r\left(\frac{L+\bm{X}\cdot h(\bm{\chi}))}{C}\right)\right].
\end{align*}
Hence as $U\to\infty$ and $r\to\infty$ by monotone convergence,
\begin{equation}\label{4:eqn:convergence:finite:f}
\liminf_{n\to\infty}\E\left[X_n^m\circRho\left(\frac{L_n+\bm{X}_n\cdot h(\bm{\chi})}{C_n}\right)\right] - \chi^m \geq \circFSuper{m}(\bm{\chi})
\end{equation}
and we can use Lemma \ref{4:lem:convergence:chi:hat} to derive that $\liminf_{n\to\infty}\bm{\chi}_n \geq \liminf_{n\to\infty}\hat{\bm{\chi}}_n\geq\hat{\bm{\chi}}$.

We now want to show the lower bound on the final damage. Fix some $\delta>0$ and choose $n$ large enough such that $\bm{\chi}_n\geq\hat{\bm{\chi}}_n\geq(1-\delta)\hat{\bm{\chi}}$ componentwise. Then
\[ n^{-1}\mathcal{S}_n = \E\left[S_n\1\left\{L_n+\bm{X}_n\cdot h(\bm{\chi}_n)\geq C_n\right\}\right] \geq \E\left[S_n\1\left\{L_n+\bm{X}_n\cdot h((1-\delta)\hat{\bm{\chi}})> C_n\right\}\right]. \]
However, using weak convergence of $(\bm{X}_n,S_n,C_n,L_n)$ and approximating the indicator function $\1\{y>1\}$ from below by continuous functions $(\phi_t)_{t\in\N}$, we derive for $U\in\R_+$
\begin{align*}
\liminf_{n\to\infty}n^{-1}\mathcal{S}_n &\geq \lim_{n\to\infty}\E\left[\left( S_n\wedge U \right)\phi_t\left(\frac{L_n+\bm{X}_n\cdot h((1-\delta)\hat{\bm{\chi}})}{C_n}\right)\right]\\
&= \E\left[(S\wedge U)\phi_t\left(\frac{L+\bm{X}\cdot h((1-\delta)\hat{\bm{\chi}})}{C}\right)\right]
\end{align*}
and as $U\to\infty$ and $t\to\infty$,
\[ \liminf_{n\to\infty}n^{-1}\mathcal{S}_n \geq \E\left[S\1\left\{L+\bm{X}\cdot h((1-\delta)\hat{\bm{\chi}})>C\right\}\right] = \accentset{\circ}{g}((1-\delta)\hat{\bm{\chi}}). \]
This quantity now tends to $\accentset{\circ}{g}(\hat{\bm{\chi}})$ as $\delta\to0$ by lower semi-continuity of $\accentset{\circ}{g}$.

Now we approach the second part of the theorem. Recall from Proposition \ref{4:prop:upperlowerfixedsize} that \mbox{$\bm{\chi}_n\leq\overline{\bm{\chi}}_n$.} By the construction of $\bm{\chi}^*$ in the proof of Lemma \ref{4:lem:existence:chi:hat}, we have a non-increasing (as $\epsilon\to0$) sequence $(\hat{\bm{\chi}}(\epsilon))_{\epsilon>0}$  such that $\lim_{\epsilon\to0+}\hat{\bm{\chi}}(\epsilon)=\bm{\chi}^*$. (See Remark \ref{4:rem:sequence:chi:*} for non-continuous $\rho$.) In particular, $\bm{\chi}^*\leq\hat{\bm{\chi}}(\epsilon)$ for every $\epsilon>0$ and $f^m(\hat{\bm{\chi}}(\epsilon))=-\epsilon$. Using weak convergence of $(\bm{X}_n,S_n,C_n,L_n)$ we derive for $U\in\R_+$ and $(\rho_s)_{s\in\N}$ an approximation of $\rho$ from above by continuous sale functions that
\begin{align*}
&\limsup_{n\to\infty}\E\bigg[X_n^m\rho\bigg(\frac{L_n+\bm{X}_n\cdot h(\hat{\bm{\chi}}(\epsilon))}{C_n}\bigg)\bigg] = \E[X^m] - \liminf_{n\to\infty}\E\bigg[X_n^m\bigg(1-\rho\bigg(\frac{L_n+\bm{X}_n\cdot h(\hat{\bm{\chi}}(\epsilon))}{C_n}\bigg)\bigg)\bigg]\\
&\hspace{1cm}\leq \E[X^m] - \liminf_{n\to\infty}\E\left[(X_n^m\wedge U)\left(1-\rho_s\left(\frac{L_n+\bm{X}_n\cdot h(\hat{\bm{\chi}}(\epsilon))}{C_n}\right)\right)\right]\\
&\hspace{1cm}= \E[X^m] - \E\left[(X^m\wedge U)\left(1-\rho_s\left(\frac{L+\bm{X}\cdot h(\hat{\bm{\chi}}(\epsilon))}{C}\right)\right)\right]
\end{align*}
and as $U\to\infty$, $s\to\infty$, by monotone convergence
\[ \limsup_{n\to\infty}\E\left[X_n^m\rho\left(\frac{L_n+\bm{X}_n\cdot h(\hat{\bm{\chi}}(\epsilon))}{C_n}\right)\right] \leq f^m(\hat{\bm{\chi}}(\epsilon)) + \hat{\chi}^m(\epsilon) = \hat{\chi}^m(\epsilon) - \epsilon. \]
Hence for $n$ large enough it holds
\[ \E\left[X_n^m\rho\left(\frac{L_n+\bm{X}_n\cdot h(\hat{\bm{\chi}}(\epsilon))}{C_n}\right)\right] - \hat{\chi}^m(\epsilon) \leq -\epsilon/2 <0 \]
for all $m\in[M]$. In particular, we know that $\overline{\bm{\chi}}_n\leq\hat{\bm{\chi}}(\epsilon)$.  Letting $\epsilon\to0$, this shows that $\limsup_{n\to\infty}\chi_n^m\leq\limsup_{n\to\infty}\overline{\chi}_n^m\leq(\chi^*)^m$ for all $m\in[M]$ and hence completes the proof of the upper bound on finally sold assets.

For the upper bound on the final damage $n^{-1}\mathcal{S}_n=\E[S_n\1\{L_n+\bm{X}_n\cdot h(\bm{\chi}_n)\geq C_n\}]$, approximate the indicator function $\1\{y\geq1\}$ from above by continuous functions $(\psi_t)_{t\in\N}$ and use weak convergence of $(\bm{X}_n,S_n,C_n,L_n)$ to derive for $U\in\R_+$
\begin{align*}
\limsup_{n\to\infty}n^{-1}\mathcal{S}_n &= \E[S] - \liminf_{n\to\infty}\E\left[ S_n\left( 1-\psi_t\left( \frac{L_n+\bm{X}_n\cdot h\left( \hat{\bm{\chi}}(\epsilon) \right)}{C_n} \right) \right) \right]\\
&\leq \E[S] - \liminf_{n\to\infty}\E\left[ (S_n\wedge U) \left( 1-\psi_t\left( \frac{L_n+\bm{X}_n\cdot h\left( \hat{\bm{\chi}}(\epsilon) \right)}{C_n} \right) \right) \right]\\
&= \E[S] - \E\left[ (S\wedge U) \left( 1-\psi_t\left( \frac{L+\bm{X}\cdot h\left( \hat{\bm{\chi}}(\epsilon) \right)}{C} \right) \right) \right]
\end{align*}
and as $U\to\infty$ and $t\to\infty$, $\limsup_{n\to\infty}n^{-1}\mathcal{S}_n\leq g(\hat{\bm{\chi}}(\epsilon))$. Letting $\epsilon\to0$, thus shows the second part of the theorem by upper semi-continuity of $g$.
%
\end{proof}

\subsection{Proofs for Section \ref{4:sec:resilience}}\label{4:ssec:proofs:resilience}
As in Section \ref{4:sec:resilience} we use the notation $g$, $f^m$, $\circG$, $\circFSuper{m}$, $\hat{\bm{\chi}}$ and $\bm{\chi}^*$ for an unshocked $(\bm{X},S,C)$-system. If instead we index these quantities by $\cdot_L$, we mean the system shocked by $L$.
\begin{proof}[Proof of Theorem \ref{4:thm:resilience}]
By 
Remark \ref{4:rem:sequence:chi:*}, there exists a sequence of vectors $\tilde{\bm{\chi}}(\gamma)\in\R_{+,0}^M$ such that $f^m(\tilde{\bm{\chi}}(\gamma))=-\gamma$ for all $m\in[M]$ and arbitrary $\gamma\in\R_+$. Now for arbitrary $\alpha\in\R_+$ it holds that
\begin{align*}
f_L^m(\bm{\chi}) &= \E\left[X^m\rho\left(\frac{L+\bm{X}\cdot h(\bm{\chi})}{C}\right)\right] - \chi^m\\
&\leq \E\left[X^m\1\left\{\frac{L}{C}\geq\alpha\right\}\right] + \E\left[X^m\rho\left(\frac{\alpha C+\bm{X}\cdot h(\bm{\chi})}{C}\right)\right] - \chi^m.
\end{align*}
Since $\E[L/C]<\delta$, by Markov's inequality it holds that $\P(L/C\geq\alpha)\leq\delta/\alpha$ and hence for $\delta>0$ small enough, we have $\E[X^m\1\{L/C\geq\alpha\}]\leq\gamma/3$ (recall that $\E[X^m]<\infty$). By dominated convergence and right-continuity of $\rho$, it thus holds that $f_L^m(\bm{\chi})\leq f^m(\bm{\chi})+2\gamma/3$ for $\alpha>0$ small enough such that
\[ \E\left[X^m\rho\left(\frac{\alpha C+\bm{X}\cdot h(\bm{\chi})}{C}\right)\right] \leq \E\left[X^m\rho\left(\frac{\bm{X}\cdot h(\bm{\chi})}{C}\right)\right] + \gamma/3. \]
In particular, $f_L^m(\tilde{\bm{\chi}}(\gamma)) \leq -\gamma/3<0$ and hence $\bm{\chi}_L^*<\tilde{\bm{\chi}}(\gamma)$ for $\delta$ small enough. By similar means, we further derive that for $\delta$ small enough it holds $g_L(\tilde{\bm{\chi}}(\gamma))\leq g(\tilde{\bm{\chi}}(\gamma))+\epsilon/3$. Together with Theorem \ref{4:thm:fire:sale:final:fraction}, we thus derive that
\[ \limsup_{n\to\infty}n^{-1}\mathcal{S}_{n,L} \leq g_L(\bm{\chi}_L^*) + \epsilon/3 \leq g_L(\tilde{\bm{\chi}}(\gamma)) + \epsilon/3 \leq g(\tilde{\bm{\chi}}(\gamma)) + 2\epsilon/3. \]
Now since $\tilde{\bm{\chi}}(\gamma)\to\bm{\chi}^*$ and by upper semi-continuity of $g$, we can choose $\gamma>0$ small enough such that $g(\tilde{\bm{\chi}}(\gamma))\leq g(\bm{\chi}^*)+\epsilon/3$ and conclude that $\limsup_{n\to\infty}n^{-1}\mathcal{S}_{n,L} \leq g(\bm{\chi}^*)+\epsilon$.

For the bound on $\chi_{n,L}^m$ choose $\gamma$ and $\delta$ small enough such that $(\chi_L^*)^m \leq \tilde{\chi}^m(\gamma)+\epsilon/3 \leq (\chi^*)^m+2\epsilon/3$ and conclude by Theorem \ref{4:thm:fire:sale:final:fraction} that
\[ \limsup_{n\to\infty}\chi_{n,L}^m \leq (\chi_L^*)^m + \epsilon/3 \leq (\chi^*)^m+\epsilon. \qedhere\]
\end{proof}

\begin{proof}[Proof of Theorem \ref{4:thm:non-resilience}]
For any $\epsilon>0$ and any subset $I\subset[M]$ let
\[ T(\epsilon,I) := \bigcap_{m\in I}\left\{\bm{\chi}\in\R_{+,0}^M\,:\,\circFSuper{m}(\bm{\chi})\leq-\epsilon\right\} \cap \bigcap_{k\in I^c}\left\{\bm{\chi}\in\R_{+,0}^M\,:\,\chi^k\geq\E[X^k]\right\}, \]
where $I^c:=[M]\backslash I$. Analogously to the construction of $\hat{\bm{\chi}}$ in the proof of Lemma \ref{4:lem:existence:chi:hat}, we find the smallest (componentwise) point $\hat{\bm{\chi}}(\epsilon,I)\in\R_{+,0}^M$ such that $\circFSuper{m}(\bm{\chi})=-\epsilon$ for $m\in I$ and $\chi^k=\E[X^k]$ for $k\in I^c$. Clearly, $\hat{\bm{\chi}}(\epsilon,I)\in T(\epsilon,I)$ and $\hat{\bm{\chi}}(\epsilon,I)\leq\bm{\chi}$ for any other $\bm{\chi}\in T(\epsilon,I)$ (choose $\bm{\chi}$ as an upper bound in the construction).

In particular, $\hat{\bm{\chi}}(\epsilon,I)$ is non-decreasing. As it is bounded by $\E[\bm{X}]$, we therefore know that it is continuous for almost every $\epsilon>0$. As moreover, $\E[X^m\rho(\bm{X}\cdot h(\hat{\bm{\chi}}(\epsilon,I))/C)]$ is bounded and increasing in $\epsilon$, we derive that for almost every $\epsilon>0$ and arbitrary $\delta>0$, we can find $\gamma>0$ small enough such that
\[ \E\left[ X^m\rho\left( \frac{\bm{X}\cdot h(\hat{\bm{\chi}}(\epsilon,I))}{C} \right) \right] \leq \E\left[ X^m\rho\left( \frac{\bm{X}\cdot h(\hat{\bm{\chi}}(\epsilon-\gamma,I))}{C} \right) \right] + \delta. \]
As $\hat{\chi}^m(\epsilon,I)$ is strictly increasing for $m\in I$ and by the assumption of $h^m(\bm{\chi})$ being strictly increasing in $\chi^m$, we derive on $\{X^m>0\}$ that
\[ \rho\left(\frac{\bm{X}\cdot h\left(\hat{\bm{\chi}}(\epsilon-\gamma,I)\right)}{C}\right) \leq \circRho\left(\frac{\bm{X}\cdot h\left(\hat{\bm{\chi}}(\epsilon,I)\right)}{C}\right) \]
and hence
\[ \E\left[ X^m\circRho\left( \frac{\bm{X}\cdot h(\hat{\bm{\chi}}(\epsilon,I))}{C} \right) \right] \leq \E\left[ X^m\rho\left( \frac{\bm{X}\cdot h(\hat{\bm{\chi}}(\epsilon,I))}{C} \right) \right] \leq \E\left[ X^m\circRho\left( \frac{\bm{X}\cdot h(\hat{\bm{\chi}}(\epsilon,I))}{C} \right) \right] + \delta. \]
Choosing $\delta$ arbitrarily small, we thus conclude that $f^m(\hat{\bm{\chi}}(\epsilon,I))=\circFSuper{m}(\hat{\bm{\chi}}(\epsilon,I)) = -\epsilon$ for $m\in I$.

Suppose now there was some $\bm{\chi}\in P_0\backslash[\bm{0},\hat{\bm{\chi}}(\epsilon,I)]$. As $P_0\subset[\bm{0},\E[\bm{X}]]$ and by monotonicity of $f^m$, we could then find some $m\in I$ and $\tilde{\bm{\chi}}\in P_0$ such that $\tilde{\bm{\chi}}\leq\hat{\bm{\chi}}(\epsilon,I)$ and $\tilde{\chi}^m=\hat{\chi}^m(\epsilon,I)$. This on the other hand would imply $f^m(\tilde{\bm{\chi}}) \leq f^m(\hat{\bm{\chi}}(\epsilon,I)) = -\epsilon$, which contradicts $\tilde{\bm{\chi}}\in P_0$. We can thus conclude that $\bm{\chi}^*\in P_0\subset[\bm{0},\hat{\bm{\chi}}(\epsilon,I)]$.

Consider now a certain $L$ and let
\[ I := \{m\in[M]\,:\,\hat{\chi}_L^m<\E[X^m]\}. \]
Then for $m\in I$, we have
\[ \circFSuper{m}(\hat{\bm{\chi}}_L) = \frac{\circFSuperSub{m}{L}(\hat{\bm{\chi}}_L) - \P(L=2C) (\E[X^m]-\hat{\chi}_L^m)}{\P(L=0)} <0 \]
since $\circFSuperSub{m}{L}(\hat{\bm{\chi}}_L)=0$ by definition. Let now
\[ \epsilon := -\max_{m\in I} \circFSuper{m}(\hat{\bm{\chi}}_L) >0. \]
By construction, $\hat{\bm{\chi}}_L\in T(\epsilon,I)$ and thus $\hat{\bm{\chi}}_L\geq\hat{\bm{\chi}}(\epsilon,I)\geq\bm{\chi}^*$. By Theorem \ref{4:thm:fire:sale:final:fraction} we can thus conclude that
\[ n^{-1}\mathcal{S}_{n,L} \geq \circG_L(\hat{\bm{\chi}}_L) + o(1) \geq \circG(\bm{\chi}^*) + o(1) \]
and
\[ \chi_{n,L}^m \geq \hat{\chi}_L^m + o(1) \geq (\chi^*)^m + o(1).\qedhere \]
\end{proof}

\begin{proof}[Proof of Corollary \ref{4:cor:one:asset:sales:default}]
First, let $\gamma\geq1$. Then
\[ f(\chi) = \E[X\1\{Xh(\chi)\geq \alpha X^\gamma\}] - \chi \leq \E[X\1\{h(\chi)\geq \alpha \}] - \chi = -\chi \]
for $\chi$ small enough such that $h(\chi)<\alpha$. Hence $\chi^*=0$ and $g(\chi^*)=\P(Xh(\chi^*)\geq\alpha X^\gamma)
=0$. The system is hence resilient by Corollary \ref{4:cor:resilience}.

Now assume that $\gamma<1$. Then for $\chi$ small enough
\begin{align*}
f(\chi) &\leq \E\left[X\1\left\{X\geq\left(\frac{\alpha}{\mu_2}\chi^{-\nu}\right)^\frac{1}{1-\gamma}\right\}\right] - \chi\\
&= \left(1-F_X\left(\left(\frac{\alpha}{\mu_2}\chi^{-\nu}\right)^\frac{1}{1-\gamma}\right)\right)\left(\frac{\alpha}{\mu_2}\chi^{-\nu}\right)^\frac{1}{1-\gamma} + \int_{\left(\frac{\alpha}{\mu_2}\chi^{-\nu}\right)^\frac{1}{1-\gamma}}^\infty (1-F_X(t))\dd t - \chi\\
&\leq 
B_2 \frac{\beta-1}{\beta-2} \left(\frac{\mu_2}{\alpha}\chi^\nu\right)^\frac{\beta-2}{1-\gamma} - \chi
\end{align*}
and $\liminf_{\chi\to0+} f(\chi)\chi^{-1} < 0$ for $\gamma > 1-\nu(\beta-2)$ or $\gamma = 1-\nu(\beta-2)$ and $\alpha > \mu_2 \left(B_2\frac{\beta-1}{\beta-2}\right)^\nu$. This implies $\chi^*=0$ and hence resilience as above.

On the other hand, for $\chi$ small enough also
\[ f(\chi) \geq B_1\frac{\beta-1}{\beta-2}\left(\frac{\mu_1}{\alpha}\chi^\nu\right)^\frac{\beta-2}{1-\gamma} - \chi \]
and hence $\chi^*>0$ for $\gamma<1-\nu(\beta-2)$ or $\gamma=1-\nu(\beta-2)$ and $\alpha<\mu_1\left(B_1\frac{\beta-1}{\beta-2}\right)^\nu$. Then
\[ \circG(\chi^*) = \P\left(X>\left(\frac{\alpha}{h(\chi^*)}\right)^\frac{1}{1-\gamma}\right) \geq B_1 \left(\frac{h(\chi^*)}{\alpha}\right)^\frac{\beta-1}{1-\gamma} > 0 \]
and the system is non-resilient by Corollary \ref{4:cor:non:resilience}.
\end{proof}

\begin{proof}[Proof of Corollary \ref{4:cor:intermediate:sales}]
Non-resilience for $\gamma<1-\nu(\beta-2)$ is trivial from Corollary \ref{4:cor:one:asset:sales:default} noting that the intermediate sales only make the system even less resilient.

So assume in the following that $\gamma>1-\nu(\beta-2)$: First, let $\gamma\geq1$. Then for $\chi$ small enough (cf.~the proof of Corollary~\ref{4:cor:one:asset:sales:default}) it holds
\[ f(\chi) \leq \E\left[ X \left( \frac{Xh(\chi)}{\alpha X^\gamma} \right)^q \1\left\{Xh(\chi) < \alpha X^\gamma\right\} \right] - \chi \leq \frac{\E[X] \mu_2^q}{\alpha} \chi^{\nu q} - \chi \]
and by $\nu q>1$, we derive $\chi^*=0$ and resilience of the system by Corollary \ref{4:cor:resilience}.

Now let $\gamma<1$. Using $B_1x^{1-\beta} \leq 1-F_X(x) \leq B_2 x^{1-\beta}$ for $x\geq x_0$, we derive
\begin{align*}
&\E\left[ X^{1+q(1-\gamma)} \1\left\{ X<\left(\frac{\alpha}{h(\chi)}\right)^\frac{1}{1-\gamma} \right\} \right]\\
&\hspace{0.4cm}= \int_0^{\left(\frac{\alpha}{h(\chi)}\right)^\frac{1}{1-\gamma}} (1+q(1-\gamma)) t^{q(1-\gamma)} (1-F_X(t)) \dd t - \left(\frac{\alpha}{h(\chi)}\right)^{\frac{1}{1-\gamma}+q} \left( 1-F_X\left( \left(\frac{\alpha}{h(\chi)}\right)^\frac{1}{1-\gamma} \right) \right)\\
&\hspace{0.4cm}\leq \left(B_2 \frac{1+q(1-\gamma)}{2-\beta+q(1-\gamma)} - \tilde{B}_1\right) \left(\frac{\alpha}{h(\chi)}\right)^{\frac{2-\beta}{1-\gamma}+q} + \kappa,
\end{align*}
where $\kappa>0$ accounts for the lower part of the integral from $0$ to $x_0$ and $\tilde{B}_1$ is chosen such that $\tilde{B}_1x^{1-\beta}\leq 1-F_X(x)$ for all $x\geq0$. For $\chi$ small enough (cf.~the proof of Corollary~\ref{4:cor:one:asset:sales:default}) it then holds
\begin{align*}
f(\chi) &\leq B_2 \frac{\beta-1}{\beta-2} \left(\frac{\mu_2}{\alpha}\chi^\nu\right)^\frac{\beta-2}{1-\gamma} + \E\left[X^{1+q(1-\gamma)}\1\left\{X<\left(\frac{\alpha}{h(\chi)}\right)^\frac{1}{1-\gamma}\right\}\right] \left(\frac{h(\chi)}{\alpha}\right)^q - \chi\\
&\leq B_2 \frac{\beta-1}{\beta-2} \left(\frac{\mu_2}{\alpha}\chi^\nu\right)^\frac{\beta-2}{1-\gamma} + \left(B_2 \frac{1+q(1-\gamma)}{2-\beta+q(1-\gamma)} - \tilde{B}_1\right) \left(\frac{\mu_2}{\alpha}\chi^\nu\right)^\frac{\beta-2}{1-\gamma} + \kappa \left(\frac{\mu_2}{\alpha}\chi^\nu\right)^q - \chi
\end{align*}
and by $\nu(\beta-2)/(1-\gamma)>1$ as well as $\nu q>1$, we derive that $\chi^*=0$ and the system is resilient by Corollary \ref{4:cor:resilience}.
\end{proof}

\begin{proof}[Proof of Corollary \ref{4:cor:multiple:assets}]
Let $\bm{v}\in\R_+^M$ be defined by $v^m:=\theta^m/\mu^m$. The functions $f^m$, $m\in[M]$, are given by
\[ f^m(\bm{\chi}) = \E\left[X^m\rho^m\left(\frac{\sum_{m\in[M]}X^m h^m(\bm{\chi})}{\sum_{m\in[M]}\theta^mX^m}\right)\right] - \chi^m \]
and thus for $\chi\in\R_{+,0}$ we have
\[ f^m(\chi\bm{v}) = \E[X^m \rho^m(\chi^\nu)] - \chi\frac{\theta^m}{\mu^m} = \E[X^m]\left(\1\{\chi\geq1\}+\chi^{\nu q^m}\1\{\chi<1\}\right) - \chi\frac{\theta^m}{\mu^m}. \]
As $\chi\to0$,
\[ \liminf_{\chi\to0+} f^m(\chi\bm{v})\chi^{-1} = \begin{cases} -\theta^m/\mu^m,&\text{if }q^m>\nu^{-1},\\\E[X^m]-\theta^m/\mu^m,&\text{if }q^m=\nu^{-1},\\ \infty,&\text{if }q^m<\nu^{-1}. \end{cases} \]
In particular, both \ref{4:cor:multiple:assets:1}.~and \ref{4:cor:multiple:assets:2}.~imply $\liminf_{\chi\to0+} f^m(\chi\bm{v})\chi^{-1}<0$ and since this holds for all $m\in[M]$, we can conclude that $\bm{\chi}^*=\bm{0}$ and the system is resilient by Corollary \ref{4:cor:resilience}.
\end{proof}

\cleardoublepage
\chapter{An Integrated Model for Default Contagion and Fire Sales in Multi-type Financial Networks}
\label{chap:fire:sales:default}
In the previous chapters, we have studied different aspects of systemic risk focusing on the two contagion channels \emph{default contagion} respectively \emph{fire sales} and particularly derived ways to prevent large default cascades. Each single chapter, however, was devoted to the understanding of specific phenomena and to ensure a concise presentation left out aspects covered in other chapters. The aim of this chapter is now to propose and analyze in detail a model combining the model features from Chapters \ref{chap:block:model} and \ref{chap:fire:sales} and by this obtain a more comprehensive picture of systemic risk. In particular, we will integrate the fire sales channel of systemic risk and the default contagion channel which allows us to better understand their interlocking in a cascade of financial distress. As in the previous chapters we will derive results about the final state of an initially distressed financial system, a characterization of resilience, and capital requirements. Moreover, we demonstrate that default contagion and fire sales can significantly amplify each other and that these amplification effects can even cause a system to become non-resilient. In Section \ref{5:sec:system:model}, we describe the model and particularly emphasize the joint contagion process of default contagion and fire sales. Next, we derive results about the final systemic damage in Section \ref{5:asymp:res:sys} and identify resilient and non-resilient system characteristics in Section \ref{5:sec:resilience}. In Section \ref{5:sec:applications}, we show that the combination of default contagion and fire sales can have tremendous impact on the stability of the system and support this example by numerical simulations. Moreover, we derive sufficient capital requirements to protect a financial system against the joint effects of default contagion and fire sales. Finally, we give proofs for all our results in Section \ref{5:sec:proofs}.

\vspace*{-7pt}
\paragraph{My own contribution:} This chapter presents an extended version of the model from \cite{Detering2018c}. Still many passages are adopted or slightly modified from there. \cite{Detering2018c} is joint work with Nils Detering, Thilo Meyer-Brandis and Konstantinos Panagiotou. I was significantly involved in the development of all parts of that paper and did most of the editorial work. In particular, I made major contributions to the conceptualization of the model and the joint contagion process, as well as Lemma \ref{5:lem:cont:rho}, Proposition \ref{5:prop:bounds:finite:systems}, Theorems \ref{5:thm:final:fraction:combined:general}, \ref{5:thm:resilience} and \ref{5:thm:non-resilience}, Examples \ref{5:ex:simulation} and \ref{5:ex:combined}, Theorem \ref{5:thm:final:fraction:combined:finitary}, and Lemmas \ref{5:lem:stochastic:domination} and \ref{5:lem:convergence:g}, that were included in \cite{Detering2018c} for the special case of the \emph{threshold model} from Chapter \ref{chap:systemic:risk} (i.\,e. \mbox{$R=T=1$}) and for the final default fraction rather than a general measure of systemic importance (i.\,e.~\mbox{$S\equiv 1$}). 

\section{An Integrated Model for the Financial System}\label{5:sec:system:model}

In this section, we state our model for a financial system. It includes all the parameters we need to investigate the interplay of the contagion channels \emph{fire sales} and \emph{default contagion}. We assume that there are $n\in\N$ financial institutions. We use the term financial institutions in a wide sense. It may include banks, insurance companies, mutual funds, asset managers but also non-financial institutions as for example corporations if they hold a large number of the assets and would sell them in case of a decline in value. We denote the set of institutions by $[n]:=\{1,\ldots,n\}$. Further we consider $M\in\N$ assets. These are the assets institutions invest in and that are considered relevant for potential fire sales. We denote by $[M]:=\{1,\ldots,M\}$ the set of these assets.

\subsection{Model Parameters}\label{5:ssec:model:parameters}
Each institution $i\in [n]$ has a set of parameters assigned:
\begin{enumerate}
\item {\bf The value of systemic importance $s_i\in\R_{+,0}$}: It describes the potential damage that a default of institution $i$ will cause for the global financial system or the wider economy. See Chapter \ref{chap:systemic:risk} for more details and possible choices for $s_i$.
\item {\bf The initial capital parameter $c_i\in\R_{+,\infty}:=\R_+\cup\{\infty\}$}: This parameter determines the monetary buffer of institution $i$ against losses. For banks this is usually their equity (assets minus liabilities) which is positive if the bank is solvent. For an asset manager it could be the total value of assets managed. In the following, we refer to $c_i$ as capital for simplicity.
\item {\bf The exogenous loss parameter $\ell_i\in\R_{+,0}$}: It models the impact of some external shock on institution $i$. The specification of $\ell_i$ allows for a variety of stress tests for the financial system, i.\,e.~asset price shock, defaults of some institutions, etc. The actual magnitude of $\ell_i$ will thus crucially depend on the stress testing and the business model of institution $i$. The actual new capital of institution $i$ after the shock is $c_i-\ell_i$. 
\item {\bf The number $x_i^m\in\R_{+,0}$ of shares institution $i$ holds of asset $m\in[M]$}: As we are only interested in the effect of sales, we consider only positive holdings. If an institution $i$ is shortening asset $m$, we set $x_i^m=0$. So for each institution we can assign a vector $\bm{x}_i:=(x_i^1,\ldots,x_i^M)\in \R_{+,0}^M$ of asset holdings.
\item {\bf Direct exposures $e_{i,j}\in \R_{+,0}$}: 
In this section, we consider $e_{i,j}$ to be the observed (deterministic) exposure of $j$ to $i$. If $e_{i,j}>0$, this means that institution $i$ owes a monetary amount to institution $j$ via for example a loan. In the next section, we will propose a random model for $\{e_{i,j}\}_{i,j\in[n]}$ using the ideas from Chapter \ref{chap:block:model}.
\end{enumerate}
In Figure~\ref{5:fig:parameters} we summarize the parameters assigned to each institution.

\begin{figure}
\centering
\includegraphics[width=0.6\textwidth]{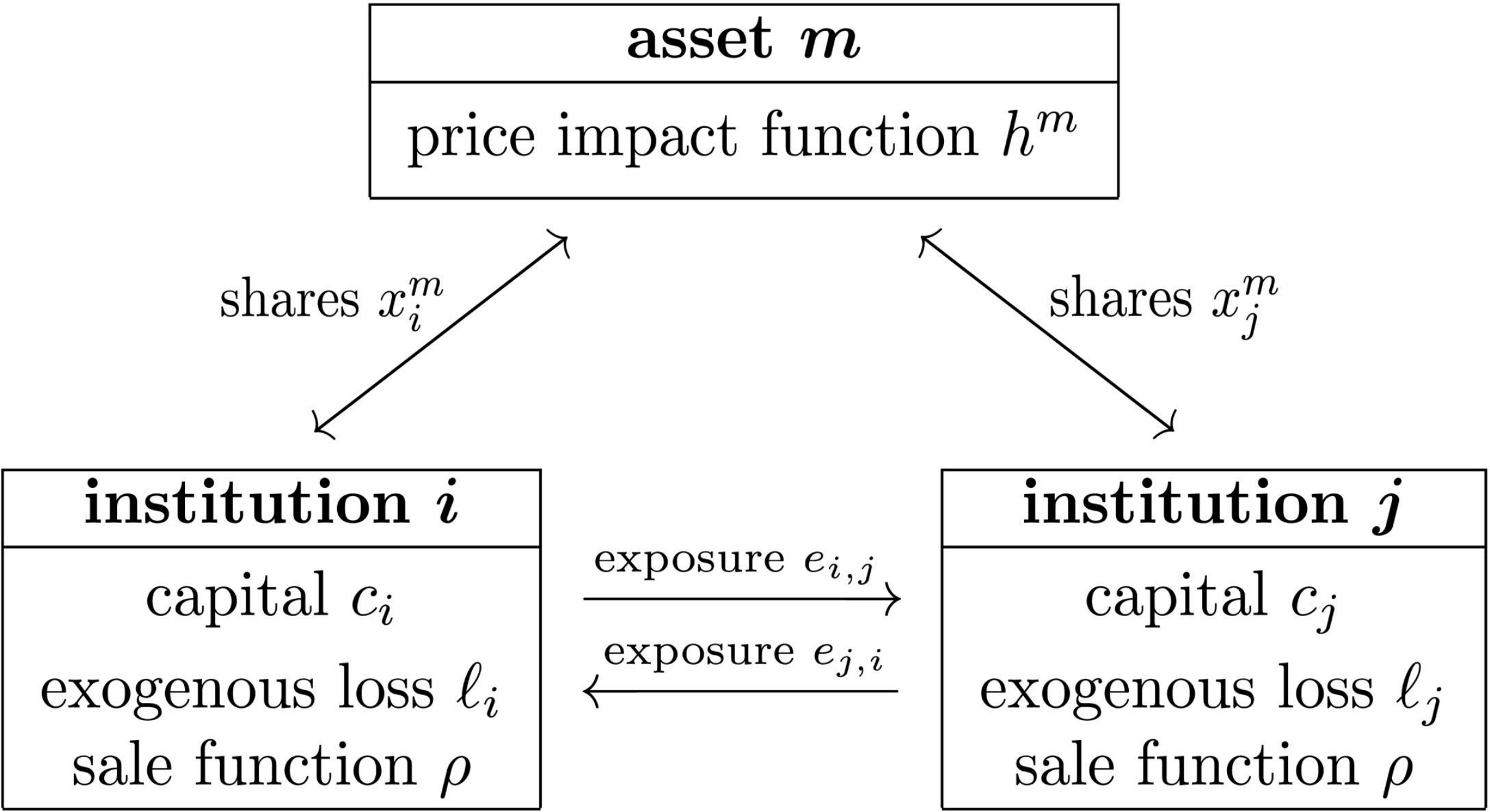} 
\caption{Model parameters for institutions $i$ and $j$ in the financial system as well as asset $m$}\label{5:fig:parameters}
\end{figure}


\subsection{Fire Sales}
We repeat the most important notions about fire sales from Chapter \ref{chap:fire:sales}. Fire sales are the combination of asset sales and price impact. The exogenous losses $\ell_i$, $i\in [n]$, possibly drive some institutions into selling parts of their assets. This can be due to their own risk preference or regulation that forces them to stay within certain risk bounds or leverage constraints. We model these asset sales by a function $\rho:\R_{+,0}\to[0,1]$, which describes the fraction of assets sold after an institution lost a certain fraction of its capital, i.\,e.~$i\in[n]$ sells $x_i^m\rho(\Lambda/c_i)$ of its shares of asset $m$ after it incurred losses of $\Lambda$. We refer to Chapter \ref{chap:fire:sales} for more details on the specification of $\rho$. The setup can easily be extended to account for different sales behaviors for different institutions. This is especially important if one goes beyond the banking network. Then one has to  reflect the fact that banks and insurance companies are regulated differently. Also wealth managers will have a different sale function depending on their risk profile. For the sake of notational simplicity we restrict ourselves to one sale function $\rho$ here. We make the following natural assumptions: the sale function $\rho$ is non-decreasing, $\rho(0)=0$ and $\rho(u)=\rho(1)$ for all $u\geq1$. Moreover to simplify notation in the following, we choose $\rho$ to be right-continuous and denote by $\circRho(u):=\lim_{\epsilon\to0+} \rho((1-\epsilon)u)$ its left-continuous modification.

The sales will cause prices to go down as the assets are not perfectly liquid (the limit order book has finite depth). To model the decline in the asset prices we use functions \mbox{$h^m:\R_{+,0}^M\to[0,1]$} which are non-decreasing and continuous in each coordinate. After \mbox{$ny^ m\hspace*{-0.05cm} \in \hspace*{-0.05cm} \R_{+,0}$} shares of asset $m$ have been sold, we assume that the share price of each asset $m\in[M]$ decreases by $h^m(\bm{y})$, where $\bm{y}=(y^1,\ldots,y^M)$. Each institution $i\in[n]$ is further assumed to suffer losses of $\bm{x}_i\cdot h(\bm{y})$ due to mark-to-market valuation of its portfolio, where $h(\bm{y})=(h^1(\bm{y}),\ldots,h^M(\bm{y}))$.

There are two remarks in order: First, we pick $\bm{y}$ as the argument of $h$ instead of the actual vector of sold shares $n\bm{y}$; this choice is purely conventional for any fixed $n$ but will be convenient for our results (also cf.~Assumption \ref{5:ass:regular:vertex:sequence}). Second, as institutions start selling assets during the contagion process they actually reduce their exposure to future price drops and $\bm{x}_i\cdot h(\bm{y})$ merely functions as an upper bound on $i$'s losses. In this sense, our model is conservative. In particular, implementation costs for each trade are covered. Moreover, this assumption allows for better analytic results in the following.

For a financial system without direct exposures, i.\,e.~$e_{i,j}=0$ for all $i,j\in[n]$, the contagion process is then solely driven by rounds of alternating asset sales and price impact, i.\,e.~fire sales. Denoting by $\bm{\sigma}_{(k)} = (\sigma_{(k)}^1,\ldots,\sigma_{(k)}^M)$ the vector of sold shares in round $k\in\N_0$ with $\bm{\sigma}_{(0)}=\bm{0}$ we then derive
\[ \bm{\sigma}_{(k)} = \sum_{i\in[n]} \bm{x}_i \rho\left( \frac{\ell_i + \bm{x}_i\cdot h(n^{-1}\bm{\sigma}_{(k-1)})}{c_i} \right),\qquad k\geq1. \]
See Chapter \ref{chap:fire:sales} for results on the pure fire sales process.

\subsection{Default Contagion}
If on the other hand, we consider a financial system with $\bm{x}_i=\bm{0}$ for all $i\in[n]$, then contagion completely proceeds via the direct exposure network. That is, if $\ell_i\geq c_i$ and institution $i\in[n]$ is therefore initially defaulted, then each institution $j\in[n]$ suffers losses of $e_{i,j}$. This possibly causes further defaults in the system and so on. Note that this supposes a recovery rate of zero which is a conservative yet reasonable assumption as the time horizon of the default contagion process is short compared to the time the resolution of an insolvent institution takes and there is a huge amount of uncertainty about the actual value of an insolvent institution immediately after its default. One could easily implement other fixed recovery rates in our model.

Again we can consider the pure default contagion process in rounds analogue to Chapter \ref{chap:block:model}. Denoting $\mathcal{D}_{(k)}\subseteq[n]$ the set of defaulted institutions in round $k\in\N_0$ with $\mathcal{D}_{(0)} = \emptyset$ we obtain
\[ \mathcal{D}_{(k)} = \{i\in[n]\,:\,\sum_{j\in\mathcal{D}_{(j)}}e_{j,i} \geq c_i-\ell_i\},\qquad k\geq1, \]
and the contagion process ends after at most $n-1$ rounds. In particular, $\mathcal{D}_n:=\mathcal{D}_{(n-1)}$ consists of all finally defaulted institutions and $\mathcal{S}_n=\sum_{i\in\mathcal{D}_n}s_i$ amounts to the total final systemic damage caused by defaults. See \cite{Cont2016,Detering2015a} and Chapter \ref{chap:block:model} for more results on the pure default contagion process.

\subsection{The Contagion Process}\label{5:sec:cont:prozess}
The focus of this chapter is on the understanding of the joint effects of fire sales and default contagion. We therefore combine the two processes from above and again consider contagion in rounds: Let $\mathcal{D}_{(0)}=\emptyset$ and $\bm{\sigma}_{(0)}=\bm{0}$. Moreover, for $k\geq1$,
\[ \mathcal{D}_{(k)} = \bigg\{i\in[n]\,:\,\sum_{j\in\mathcal{D}_{(k-1)}} e_{j,i} \geq c_i-\ell_i-\bm{x}_i\cdot h\left(n^{-1}\bm{\sigma}_{(k-1)}\right)\bigg\} \]
and
\[ \bm{\sigma}_{(k)} = \sum_{i\in[n]} \bm{x}_i \rho\left( \frac{\sum_{j\in\mathcal{D}_{(k-1)}}e_{j,i} + \ell_i+\bm{x}_i\cdot h\left(n^{-1}\bm{\sigma}_{(k-1)}\right)}{c_i} \right). \]
Then $\mathcal{D}_{(0)}\subseteq\mathcal{D}_{(1)}\subseteq\cdots\subseteq[n]$ and $\bm{\sigma}_{(0)}\leq\bm{\sigma}_{(1)}\leq\cdots\leq\sum_{i\in[n]}\bm{x}_i$. We can thus conclude that the process converges as $k\to\infty$. Let then $\mathcal{D}_n := \bigcup_{k\in\N}\mathcal{D}_{(k)}$ the set of finally defaulted institutions, $\mathcal{S}_n=\sum_{i\in\mathcal{D}_n}s_i$ their systemic importance and $\bm{\chi}_n:=\n^{-1}\lim_{k\to\infty}\bm{\sigma}_{(k)}$ the vector of finally sold shares divided by $n$.

For continuous $\rho$, we derive the following result.
\begin{lemma}\label{5:lem:cont:rho}
Consider above contagion process for continuous $\rho$. Then $\mathcal{D}_n\subseteq[n]$ and $\bm{\chi}_n$ are the smallest solution to
\begin{equation}\label{5:eqn:lem:cont:rho:D}
\mathcal{D} = \bigg\{i\in[n]\,:\,\sum_{j\in\mathcal{D}} e_{j,i} \geq c_i-\ell_i-\bm{x}_i\cdot h(\bm{\chi})\bigg\},
\end{equation}
\begin{equation}\label{5:eqn:lem:cont:rho:chi}
\bm{\chi} = \sum_{i\in[n]} \bm{x}_i \rho\left( \frac{\sum_{j\in\mathcal{D}}e_{j,i} + \ell_i+\bm{x}_i\cdot h(\bm{\chi})}{c_i} \right).
\end{equation}
That is, if $(\tilde{\mathcal{D}},\tilde{\bm{\chi}})$ also solves \eqref{5:eqn:lem:cont:rho:D} and \eqref{5:eqn:lem:cont:rho:chi}, then $\mathcal{D}_n\subseteq\tilde{\mathcal{D}}$ and $\bm{\chi}_n\leq\tilde{\bm{\chi}}$.
\end{lemma}
One particular consequence of Lemma \ref{5:lem:cont:rho} is then that for continuous $\rho$ we can alter the contagion process in the following way. Let $\mathcal{D}_0=\emptyset$ and $\bm{\chi}_0=\bm{0}$. Repeat the following for $k\geq1$ until $\mathcal{D}_k=\mathcal{D}_{k-1}$ and $\bm{\chi}_k=\bm{\chi}_{k-1}$.
\begin{enumerate}[(i)]
\item Let $\mathcal{D}_k\subseteq[n]$ the smallest set such that
\[ \mathcal{D}_k = \bigg\{i\in[n]\,:\,\sum_{j\in\mathcal{D}_k}e_{j,i}\geq c_i-\ell_i-\bm{x}_i\cdot h(\bm{\chi}_{k-1})\bigg\}. \]
\item Let $\bm{\chi}_k$ be the smallest vector such that
\[ \bm{\chi}_k = n^{-1}\sum_{i\in[n]}\bm{x}_i\rho\left(\frac{\sum_{j\in\mathcal{D}_k}e_{j,i} + \ell_i+\bm{x}_i\cdot h(\bm{\chi}_k)}{c_i}\right) \]
which exists by the Knaster-Tarski theorem.
\end{enumerate}
That is, instead of considering fire sales and default contagion simultaneously, we first consider a complete default contagion cascade, then a complete fire sales cascade and so on (compare \cite{Weber2016} for instance). This procedure ends after at most $n$ steps in $\mathcal{D}_n$ and $\bm{\chi}_n$.

For general (not necessarily continuous) $\rho$, clearly the process described by (i) and (ii) still converges to the smallest solution $\overline{\mathcal{D}}_n$ and $\overline{\bm{\chi}}_n$ of \eqref{5:eqn:lem:cont:rho:D} and \eqref{5:eqn:lem:cont:rho:chi}. By similar means as in the proof of Lemma \ref{5:lem:cont:rho} it then holds that $\mathcal{D}_n\subseteq\overline{\mathcal{D}}_n$ as well as $\bm{\chi}_n\leq\overline{\bm{\chi}}_n$.

Furthermore, consider the following process with strict inequality in (i') and $\rho$ replaced by its left-continuous modification $\circRho$ in (ii'): Let $\hat{\mathcal{D}}_0=\emptyset$ and $\hat{\bm{\chi}}_0=\bm{0}$. Repeat the following for $k\geq1$ until $\hat{\mathcal{D}}_k=\hat{\mathcal{D}}_{k-1}$ and $\hat{\bm{\chi}}_k=\hat{\bm{\chi}}_{k-1}$.
\begin{enumerate}[(i')]
\item Let $\hat{\mathcal{D}}_k\subseteq[n]$ the smallest set such that
\[ \hat{\mathcal{D}}_k = \bigg\{i\in[n]\,:\,\sum_{j\in\hat{\mathcal{D}}_k}e_{j,i} > c_i-\ell_i-\bm{x}_i\cdot h(\hat{\bm{\chi}}_{k-1})\bigg\}. \]
\item Let $\hat{\bm{\chi}}_k$ be the smallest vector such that
\[ \hat{\bm{\chi}}_k = n^{-1}\sum_{i\in[n]}\bm{x}_i\circRho\left(\frac{\sum_{j\in\hat{\mathcal{D}}_k}e_{j,i} + \ell_i+\bm{x}_i\cdot h(\hat{\bm{\chi}}_k)}{c_i}\right) \]
which exists by the Knaster-Tarski theorem.
\end{enumerate}
Again this process ends after at most $n$ rounds in $\hat{\mathcal{D}}_n$ and $\hat{\bm{\chi}}_n$. Moreover, by the same means as in the proof of Lemma \ref{5:lem:cont:rho} and since $\circRho$ is a lower bound on $\rho$, we then derive that $\hat{\mathcal{D}}_n\subseteq\mathcal{D}_n$ and $\hat{\bm{\chi}}_n\leq\bm{\chi}_n$. 

Finally, denote $\hat{\mathcal{S}}_n=\sum_{i\in\hat{\mathcal{D}}_n}s_i$ and $\overline{\mathcal{S}}_n=\sum_{i\in\overline{\mathcal{D}}_n}s_i$. Then altogether, we derive the following.
\begin{proposition}\label{5:prop:bounds:finite:systems}
For the set of finally defaulted institutions $\mathcal{D}_n$, their total systemic importance $\mathcal{S}_n$ and the vector $\bm{\chi}_n$ of finally sold shares divided by $n$ it holds
\[ \hat{\mathcal{D}}_n \subseteq \mathcal{D}_n \subseteq \overline{\mathcal{D}}_n,\qquad\hat{\mathcal{S}}_n \leq \mathcal{S}_n \leq \overline{\mathcal{S}}_n,\qquad \hat{\bm{\chi}}_n \leq \bm{\chi}_n \leq \overline{\bm{\chi}}_n. \]
\end{proposition}
The case that $\mathcal{D}_n \subsetneq \overline{\mathcal{D}}_n$ or $\bm{\chi}_n < \overline{\bm{\chi}}_n$ can happen if in the contagion process the sold shares converge to a vector that would be large enough to cause new defaults or trigger further asset sales but is actually never reached in finitely many steps. Then the process converges to a \emph{non-equilibrium} state. As for real financial systems the least possible number of sold shares in each round is lower bounded by $1$, this can actually never happen and for all practical purposes the final set of defaulted institutions is given by $\overline{\mathcal{D}}_n$, their caused systemic damage by $\overline{\mathcal{S}}_n$ and the vector of finally sold shares is given by $\overline{\bm{\chi}}_n$.

\section{The Stochastic Model}\label{5:asymp:res:sys}
In the previous section, we considered the combined contagion process of fire sales and default contagion on any explicitly given financial system. In this section, we go one step further and analyze a whole ensemble of systems simultaneously that share certain statistical characteristics. This will ultimately allow us to understand which system structures promote global contagion or contain it locally.

In a first step, we therefore replace the actual $e_{i,j}$ from the previous section by a random sample instead of using the actual observed edges/exposures. As in Chapter \ref{chap:block:model} we let \mbox{$R\in\N$} the maximal exposure between two institutions and we define a probability measure $\P$ on $\{0,\ldots,R\}^{\abs{E}}$, where $E$ is the set of possible directed edges $E:=\{ (i,j)\in[n]^2 \,:\, i\neq j  \}$. This approach has several advantages:
\begin{enumerate}
\item The network of exposures can change significantly on a microscopic level but as empirical studies show, the global statistics are reasonably stable (see e.\,g.~\cite{Cont2013}).
\item Often only the aggregated exposures  $\sum_{j\in [n]} e_{i,j}$ 
are available to the regulator. Since the individual exposures are unknown it is thus advisable to use the information available and consider probabilistic samples. Ideally one obtains results that hold for all possible realizations. 
\item A random network is analytically more tractable and provides more understanding of the impact of the network characteristics on the combined fire sales and contagion process.
\end{enumerate}
Our choice of $\P$ has to be such that the generated networks share the characteristic of the observed financial network. The actual network is thus replaced by a random network that looks very similar to the original network but has better analytic properties. 

As in Chapter \ref{chap:block:model}, we assume the global financial system to be composed of $T\in\N$ subsystems of different institution types. 
Then for each institution $i\in[n]$ we replace the direct exposures $\{e_{i,j}\}_{j\in[n]}\subset\R_{+,0}$ from Subsection \ref{5:ssec:model:parameters} by
\begin{enumerate}[leftmargin=*,label=\quad 5'.\,(\alph*)]
\item {\bf An institution-type $\alpha_i\in[T]$}: This parameter allocates institution $i$ to a certain subsystem such as country or core/periphery. 
\item {\bf A vector of in-weights $\bm{w}_i^-
\in\R_{+,0}^{[R]\times[T]}$}: The in-weight $w_i^{-,r,\alpha}$ describes the tendency of institution $i$ to be exposed to an institution of type $\alpha$ with an exposure of size $r$.
\item {\bf A vector of out-weights $\bm{w}_i^+
\in\R_{+,0}^{[R]\times[T]}$}: The out-weight $w_i^{+,r,\alpha}$ describes the tendency of institutions of type $\alpha$ to be exposed to $i$ with an exposure of size $r$.
\end{enumerate}
The occurrence of an edge of multiplicity $r\in[R]$ going from $i$ to $j$ is then modeled by a Bernoulli random variable $X_{i,j}^r$ with success probability
\begin{equation}\label{5:edge:prob}
p_{i,j}^r := \begin{cases}\min\left\{R^{-1},n^{-1}w_i^{+,r,\alpha_j}w_j^{-,r,\alpha_i}\right\},&\text{if }i\neq j,\\0,&\text{if }i=j,\end{cases} 
\end{equation}
such that $\sum_{r\in[R]}X_{i,j}^r\leq 1$ and $X_{i_1,j_1}^{r_1}\perp X_{i_2,j_2}^{r_2}$ for all $r_1,r_2$ and $(i_1,j_1)\neq (i_2,j_2)$. 

Now consider a collection of financial systems with varying size $n$. We want to ensure that their statistical characteristics measured by means of the empirical distribution functions stabilize as $n\to\infty$. Moreover, we want to prohibit that exposures or asset holdings condense in one institution. 

\begin{assumption}[Regular Vertex Sequence]\label{5:ass:regular:vertex:sequence}
Let $M\in\N$. For each $n\in\N$ consider a system with $n$ institutions and $M$ assets specified by the sequences $\bm{w}^-(n)=(\bm{w}_i^-(n))_{i\in[n]}$ of in-weights, $\bm{w}^+(n)=(\bm{w}_i^+(n))_{i\in[n]}$ of out-weights
, $\bm{x}(n)=(\bm{x}_i(n))_{i\in[n]}$ of asset holdings, \mbox{$\bm{s}(n)=(s_i(n))_{i\in[n]}$} of systemic importance values, $\bm{c}(n)=(c_i(n))_{i\in[n]}$ of capitals, \mbox{$\bm{\ell}(n)=(\ell_i(n))_{i\in[n]}$} of exogenous losses and $\bm{\alpha}(n)=(\alpha_i(n))_{i\in[n]}$ of institution types. Then the following shall hold:
\begin{enumerate}[(a)]
\item \textbf{Convergence in distribution:} For each $n\in\N$ let the random empirical distribution function of the system parameters be denoted by
\begin{align*}
&F_n(\bm{w}^-,\bm{w}^+,\bm{x},s,c,\ell,\alpha)\\
&\hspace{2.5cm}:= n^{-1} \sum_{i\in[n]} \prod_{r\in[R],\beta\in[T]}\1\left\{w_i^{-r,\beta}\leq w^{-,r,\beta}, w_i^{+,r,\beta}\leq w^{+,r,\beta}\right\}\\
&\hspace{5cm} \times \prod_{m\in[M]} \1\left\{x_i^m\leq x^m\right\} \1\{s_i\leq s, c_i\leq c, \ell_i\leq \ell, \alpha_i\leq \alpha\},
\end{align*}
for $\bm{w}^-,\bm{w}^+\in\R_{+,0}^{[R]\times[T]}$, $\bm{x}\in\R_{+,0}^M$, $s\in\R_{+,0}$, $c\in\R_{+,0,\infty}$, $\ell\in\R_{+,0}$ and $\alpha\in[T]$. Let in the following $(\bm{W}^-_n,\bm{W}^+_n,\bm{X}_n, S_n, C_n, L_n, A_n)$ denote a random vector distributed according to $F_n$. Then there exists a distribution function $F$ such that
\[ F_n(\bm{w}^-,\bm{w}^+,\bm{x},s,c,\ell,\alpha;\bm{\chi})\to F(\bm{w}^-,\bm{w}^+,\bm{x},s,c,\ell,\alpha;\bm{\chi}),\quad\text{as }n\to\infty, \]
at all continuity points of $F_{\alpha}(\bm{w}^-,\bm{w}^+,\bm{x},s,c,\ell):=F(\bm{w}^-,\bm{w}^+,\bm{x},s,c,\ell,\alpha)$.
\item \textbf{Convergence of means:} Denote by $(\bm{W}^-,\bm{W}^+,\bm{X},S,C,L,A)$ a random vector distributed according to the limiting distribution $F$. Then as $n\to\infty$,
\[ \E[W_n^{-,r,\alpha}]\to\E[W^{-,r,\alpha}]<\infty,\quad\E[W_n^{+,r,\alpha}]\to\E[W^{+,r,\alpha}]<\infty,\quad\text{for all }r\in[R],~\alpha\in[T],\]
\[ \E[S_n]\to\E[S]<\infty\quad\text{and}\quad\E[X_n^m]\to\E[X^m]<\infty,\quad \text{for all }m\in[M]. \]
\end{enumerate}
\end{assumption}
Let $V=[R]\times[T]^2$. Define now for $\bm{z}\in\R_{+,0}^V$ and $\bm{\chi}\in\R_{+,0}^M$,
\begin{align*}
g(\bm{z},\bm{\chi}) &:= \sum_{\beta\in[T]}\E\Bigg[S\psi\Bigg( Y_{1,\beta},\ldots,Y_{R,\beta}; C-L-\bm{X}\cdot h(\bm{\chi}) \Bigg)\1\{A=\beta\}\Bigg],\\
f^{r,\alpha,\beta}(\bm{z},\bm{\chi}) &:= \E\Bigg[W^{+,r,\alpha}\psi\Bigg( Y_{1,\beta},\ldots,Y_{R,\beta}; C-L-\bm{X}\cdot h(\bm{\chi}) \Bigg)\1\{A=\beta\}\Bigg] - z^{r,\alpha,\beta},~(r,\alpha,\beta)\in V,\\
f^m(\bm{z},\bm{\chi}) &:= \sum_{\beta\in[T]}\E\Bigg[X^m\phi\Bigg( Y_{1,\beta},\ldots,Y_{R,\beta}; L+\bm{X}\cdot h(\bm{\chi}),C \Bigg)\1\{A=\beta\}\Bigg] - \chi^m,~ m\in[M],
\end{align*}
where we abbreviate
\[ Y_{r,\beta}:= \sum_{\gamma\in[T]}W^{-,r,\gamma}z^{r,\beta,\gamma},\quad r\in[R],~\beta\in[T], \]
and for $\{q_s\}_{s\in[R]}\subset\R_{+,0}$ and independent $Q_s\sim\mathrm{Poi}(q_s)$, $s\in[R]$,
\[ \psi(q_1,\ldots,q_R;t) := \P\left(\sum_{s\in[R]}sX_s \geq t\right), \]
respectively
\[ \phi(q_1,\ldots,q_R;\ell,c) := \E\left[ \rho\left( \frac{\sum_{s\in[R]}sX_s + \ell}{c} \right) \right]. \]
Let us give an intuitive explanation for the functions $g$, $f^{r,\alpha,\beta}$ and $f^m$ first. For this we consider the special case $R=T=1$ and start looking at the fire sales and the default contagion process separately. 

We start with the default contagion process. Heuristically, for an externally given vector of asset sales $\bm{\chi}$, the function $f^{1,1,1}(\cdot,\bm{\chi})$ describes (in the limit $n\rightarrow \infty$) the intensity of the default contagion process over time. Here time refers to steps in a sequential analysis of the process which leads to the same set of defaulted institutions. Let now $\bar{z}\in[0,\E[W^+]]$ denote the total out-weight of finally defaulted banks divided by $n$. Then by the specification of $p_{i,j}$ for any fixed bank $i\in[n]$ the number of incoming edges (exposures) from finally defaulted banks is given by a random variable $\mathrm{Poi}(w_i^{-,1,1}\bar{z})$. Institution $i$ is hence finally defaulted itself if and only if $\mathrm{Poi}(w_i^{-,1,1}\bar{z})\geq c_i- \ell_i-\bm{x_i}\cdot h(\bm{\chi})$. Summing over all banks in the system we thus derive the following identity:
\[  \bar{z} = n^{-1} \sum_{i \in [n]}   w_i^{+,1,1}  \1\big\{\mathrm{Poi}(w_i^{-,1,1}\bar{z}) \geq c_i - \ell_i - \bm{x}_i\cdot h(\bm{\chi}) \big\}\approx   \E\big[W^{+,1,1}\psi(W^{-,1,1}\bar{z};C-L-\bm{X}\cdot h(\bm{\chi}))\big], \]
and therefore $f^{1,1,1} (\bar{z}, \bm{\chi})=0$. Further, the damage by finally defaulted banks is then given by 
\[ n^{-1} \sum_{i \in [n]} s_i \1\big\{\ell_i+\mathrm{Poi}(w_i^{-,1,1}\bar{z}) + \bm{x}_i\cdot h(\bm{\chi})\geq c_i\big\}\approx\E\big[S\psi(W^{-,1,1}\bar{z},C-L-\bm{X}\cdot h(\bm{\chi}))\big] =g(\bar{z},\bm{\chi} ). \]
Hence if fire sales are ignored, meaning the initial capital is simply reduced by a fixed amount accounting for some externally given sales vector $\bm{\chi}$, then in order to get the final state of the system we only need to determine the (first) root $\bar{z}$ of $f^{1,1,1}(\cdot, \bm{\chi})$ and plug it into $g(\cdot,\bm{\chi} )$. 

Let us now look at the fire sales system with an externally given contagion result. For the case of one asset with label $1$, fixing the sum of the out-weights of defaulted institutions, divided by $n$ to be $z$, then the loss institution $i$ receives due to the liabilities to defaulted banks is described by the random quantity $\mathrm{Poi}(w_i^{-,1,1}z)$ and, for continuous $\rho$, similarly as in the derivation of Lemma \ref{5:lem:cont:rho}, the number of finally sold shares $\bar{\chi}$ solves
\begin{equation}
\bar{\chi} = n^{-1}\sum_{i\in[n]}x_i^1\rho\left(\frac{\ell_i+\mathrm{Poi}(w_i^{-,1,1}z)+ x_i^1 h^1(\bar{\chi})}{c_i}\right)\approx\E\left[X^1 \phi(W^{-,1,1}z;L+X^1 h(\bar{\chi}),C)) \right]
\end{equation}
such that $\bar{\chi}$ is a root of $f^1 (z, \cdot )$. Moreover, the final systemic importance of defaulted institutions divided by $n$ is given by
\begin{equation}
n^{-1} \sum_{i\in[n] } s_i \1\big\{ \mathrm{Poi}(w_i^{-,1,1}z) + \ell_i+x_i^1 h^1(\bar{\chi}) \geq c_i\big\} \approx\E\big[S\psi(W^{-,1,1}z,C-L-X^1h(\bar{\chi}))\big] =g(z,\bar{\chi} ).
\end{equation}
So the root $\bar{\chi}$ of $f^1 (z, \cdot )$ determines the end of the process and again $g$ yields the damage by defaulted institutions.

These heuristics show that the joint fire sales and default contagion process should come to an end at a joint root of the functions $f^{r,\alpha,\beta}$, $(r,\alpha,\beta)\in V$, and $f^m, m\in [M]$. Under some circumstances, however, if the distribution of $(\bm{W}^-,\bm{W}^+,\bm{X},S,C,L,A)$ has atoms and the function $\rho$ is discontinuous, also the functions $f^{r,\alpha,\beta}$ and $f^m$ might be discontinuous. Similar as in the previous section, it is then in general not possible to determine the precise end state of the system. Still we will be able to derive lower bounds on the final default fraction and the vector of finally sold shares. To this end, 
define lower semi-continuous modifications of $g$, $f^{r,\alpha,\beta}$, $(r,\alpha,\beta)\in V$, and $f^m$, $m\in[M]$, by
\begin{align*}
\circG(\bm{z},\bm{\chi}) &:= \sum_{\beta\in[T]}\E\Bigg[S\circPsi\Bigg( Y_{1,\beta},\ldots,Y_{R,\beta}; C-L-\bm{X}\cdot h(\bm{\chi}) \Bigg)\1\{A=\beta\}\Bigg],\\
\circFSuper{r,\alpha,\beta}(\bm{z},\bm{\chi}) &:= \E\Bigg[W^{+,r,\alpha}\circPsi\Bigg( Y_{1,\beta},\ldots,Y_{R,\beta}; C-L-\bm{X}\cdot h(\bm{\chi}) \Bigg)\1\{A=\beta\}\Bigg] - z^{r,\alpha,\beta},~(r,\alpha,\beta)\in V,\\
\circFSuper{m}(\bm{z},\bm{\chi}) &:= \sum_{\beta\in[T]}\E\Bigg[X^m\circPhi\Bigg( Y_{1,\beta},\ldots,Y_{R,\beta}; L+\bm{X}\cdot h(\bm{\chi}),C \Bigg)\1\{A=\beta\}\Bigg] - \chi^m,~ m\in[M],
\end{align*}
where as before
\[ Y_{r,\beta}:= \sum_{\gamma\in[T]}W^{-,r,\gamma}z^{r,\beta,\gamma},\quad r\in[R],~\beta\in[T], \]
and for $\{q_s\}_{s\in[R]}\subset\R_{+,0}$ and independent $Q_s\sim\mathrm{Poi}(q_s)$, $s\in[R]$,
\[ \circPsi(q_1,\ldots,q_R;t) := \P\left(\sum_{s\in[R]}sX_s > t\right), \]
respectively
\[ \circPhi(q_1,\ldots,q_R;\ell,c) := \E\left[ \circRho\left( \frac{\sum_{s\in[R]}sX_s + \ell}{c} \right) \right] \]
for $\circRho(u):=\lim_{\epsilon\to0+}\rho((1-\epsilon)u)$. Further, let $\accentset{\circ}{P}_0$ and $P_0$ the largest connected subsets of
\begin{align*}
\accentset{\circ}{P} &:=\bigcap_{(r,\alpha,\beta)\in V} \big\{(\bm{z},\bm{\chi})\in\R_{+,0}^V\times\R_{+,0}^M\,:\,\circFSuper{r,\alpha,\beta}(\bm{z},\bm{\chi})\geq0\big\}\\
&\hspace{6cm}\cap \bigcap_{m\in[M]} \big\{(\bm{z},\bm{\chi})\in\R_{+,0}^V\times\R_{+,0}^M\,:\,\circFSuper{m}(\bm{z},\bm{\chi})\geq0\big\}
\end{align*}
respectively
\begin{align*}
P &:= \bigcap_{(r,\alpha,\beta)\in V} \big\{(\bm{z},\bm{\chi})\in\R_{+,0}^V\times\R_{+,0}^M\,:\,f^{r,\alpha,\beta}(\bm{z},\bm{\chi})\geq0\big\}\\
&\hspace{6cm}\cap \bigcap_{m\in[M]} \big\{(\bm{z},\bm{\chi})\in\R_{+,0}^V\times\R_{+,0}^M\,:\,f^m(\bm{z},\bm{\chi})\geq0\big\}
\end{align*}
that contain $(\bm{0},\bm{0})$ (note that $f^{r,\alpha,\beta}(\bm{0},\bm{0})\geq\circFSuper{r,\alpha,\beta}(\bm{0},\bm{0})\geq0$ for all $(r,\alpha,\beta)\in V$ as well as $f^m(\bm{0})\geq\circFSuper{m}(\bm{0})\geq0$ for all $m\in[M]_0$ and thus $\bm{0}\in \circP$ and $\bm{0}\in P$). We will later make use of the fact that $P$ and $P_0$ are clearly closed sets. Finally, define $\bm{z}^*\in\R_{+,0}^V$ and $\bm{\chi}^*\in\R_{+,0}^M$ by $(z^*)^{r,\alpha,\beta}:=\sup_{(\bm{z},\bm{\chi})\in P_0}z^{r,\alpha,\beta}$ and $(\chi^*)^m:=\sup_{(\bm{z},\bm{\chi})\in P_0}\chi^m$. Then the following holds:
\begin{lemma}\label{5:lem:existence:z:chi:hat}
There exists a smallest joint root $(\hat{\bm{z}},\hat{\bm{\chi}})$ of 
all the functions $\circFSuper{r,\alpha,\beta}$, $(r,\alpha,\beta)\in V$, and $\circFSuper{m}$, $m\in[M]$
. It holds $(\hat{\bm{z}},\hat{\bm{\chi}})\in\accentset{\circ}{P}_0$. Further, $(\bm{z}^*,\bm{\chi}^*)$ as defined above is a joint root of the functions 
$f^{r,\alpha,\beta}$, $f^m$ and $(\bm{z}^*,\bm{\chi}^*)\in P_0$.
\end{lemma}
The proof is analogue to the one of Lemma \ref{4:lem:existence:chi:hat}.

We can then describe the final default fraction and the final number of sold shares asymptotically as $n\to\infty$ in terms of $(\hat{\bm{z}},\hat{\bm{\chi}})$ and $(\bm{z}^*,\bm{\chi}^*)$.
\begin{theorem}\label{5:thm:final:fraction:combined:general}
Consider a financial system that fulfills Assumption~\ref{5:ass:regular:vertex:sequence}. Then for the final systemic damage $n^{-1}\mathcal{S}_n$ and $\chi_n^m$, the number of finally sold shares of asset $m\in[M]$ divided by $n$, it holds
\[\begin{gathered}
\circG(\hat{\bm{z}},\hat{\bm{\chi}}) + o_p(1) \leq n^{-1}\mathcal{S}_n \leq g(\bm{z}^*,\bm{\chi}^*) + o_p(1),\\
\hat{\chi}^m+o_p(1) \leq \chi_n^m \leq (\chi^*)^m+o_p(1).
\end{gathered}\]
In particular, for the final price impact $h^m(\bm{\chi}_n)$ on asset $m\in[M]$ it holds
\[ h^m(\hat{\bm{\chi}})+o_p(1) \leq h^m(\bm{\chi}_n) \leq h^m(\bm{\chi}^*)+o_p(1). \]
\end{theorem}
In most cases, $(\hat{\bm{z}},\hat{\bm{\chi}})$ and $(\bm{z}^*,\bm{\chi}^*)$ will coincide and $\circG(\hat{\bm{z}},\hat{\bm{\chi}})=g(\bm{z}^*,\bm{\chi}^*)$. Theorem \ref{5:thm:final:fraction:combined:general} then describes the limits in probability of $\bm{\chi}_n$ and $n^{-1}\mathcal{S}_n$ for $n\to\infty$.


\section{Resilient and Non-resilient Systems}\label{5:sec:resilience}
Our results from the previous section allow us to compute the final state of a system that was initially hit by some exogenous shock starting a cascade of default contagion and fire sales. We shall now go one step further and describe the vulnerability of an \emph{initially} unshocked system to small shocks. We achieve this goal by considering shocks $L$ of different magnitude on the same initially unshocked system described by $(\bm{W}^-,\bm{W}^+,\bm{X},S,C,A)$. In the following, if we use the notations $g$, $\circG$, $\bm{z}^*$ and $\bm{\chi}^*$ from Section \ref{5:asymp:res:sys} we mean the unshocked system with $L\equiv0$.

Our notion of resilience extends the one for pure fire sales systems in Chapter \ref{chap:fire:sales} and this section heavily borrows from Section \ref{4:sec:resilience}. 

\subsection{Resilience}\label{5:ssec:resilience}
When it comes to regulation of a financial system, one desirable property is the capability to absorb local shocks rather than amplify them through large parts of the system. In our asymptotic model we can consider arbitrarily small shocks $L$ and the following natural notion of resilience emerges: when considering initial shocks $L$ such that $\E[L/C]\to0$, then a system is called \emph{resilient} if also the induced asymptotic final damage $n^{-1}\mathcal{S}_{n,L}$ tends to $0$.

\begin{definition}[Resilience]\label{5:def:resilience}
A financial system $(\bm{W}^-,\bm{W}^+,\bm{X},S,C,A)$ is said to be \emph{resilient} if for each $\epsilon>0$ there exists $\delta>0$ such that for all $L$ with $\E[L/C]<\delta$ it holds $n^{-1}\mathcal{S}_{n,L} \leq \epsilon$ with high probability.
\end{definition}
While this definition (and Corollary \ref{5:cor:resilience} below) is concerned with the final systemic damage only, the following theorem also investigates
the number of sold shares of the assets (and hence the price impacts which also affect the wider economy) in the limit $\E[L/C]\to0$.


\begin{theorem}\label{5:thm:resilience}
For each $\epsilon>0$ there exists $\delta>0$ such that for all $L$ with $\E[L/C]<\delta$ it holds for the final damage by defaulted institutions $n^{-1}\mathcal{S}_{n,L}$ and the number $n\chi_{n,L}^m$ of finally sold shares of each asset $m\in[M]$ in the shocked system that w.\,h.\,p.
\[ n^{-1}\mathcal{S}_{n,L}\leq g(\bm{z}^*,\bm{\chi}^*)+\epsilon \quad\text{and}\quad \chi_{n,L}^m\leq(\chi^*)^m+\epsilon,\quad m\in[M]. \]
\end{theorem}
In particular, we derive the following resilience criterion.
\begin{corollary}[Resilience Criterion]\label{5:cor:resilience}
If $g(\bm{z}^*,\bm{\chi}^*)=0$, then the system is resilient.
\end{corollary}
It is thus sufficient for resilience if $(\bm{z}^*,\bm{\chi}^*)=(\bm{0},\bm{0})$ or equivalently $P_0=\{(\bm{0},\bm{0})\}$. If, however, $g(\bm{z}^*,\bm{\chi}^*)=0$ while $\bm{\chi}^*\neq\bm{0}$, then Corollary \ref{5:cor:resilience} still ensures that the final systemic damage stays small and the system is resilient by Definition \ref{5:def:resilience}, while a large fraction of shares of assets is sold due to fire sales as a reaction to small local shocks -- see Theorem \ref{5:thm:non-resilience} below.

\subsection{Non-resilience}\label{5:ssec:non:resilience}

We now aim at characterizing non-resilient systems. For this note that our fire sales model is in itself a conservative model as for each institution $i\in[n]$ the entire asset holdings $\bm{x}_i$ are exposed to the price impact $h(\bm{\chi}_n)$. It therefore ignores intermediate sales at a more favorable asset price level. We refer to Chapter \ref{chap:fire:sales} for more discussion on intermediate sales. The following results still give a first indication of non-resilience for general financial systems.

We consider shocks of the form $\ell_i\in\{0,2c_i\}$ such that $\P(L=2C)>0$ and $L/C$ is independent of $(\bm{W}^-,\bm{W}^+,\bm{X},S,C,A)$. It may seem odd at first to choose $\ell_i=2c_i$ (or any other multiple strictly larger than $1$) instead of $\ell_i=c_i$ to express the default of institution $i$. The reason is that in the proof of Theorem \ref{5:thm:non-resilience} below we want to use Theorem \ref{5:thm:final:fraction:combined:general} which only considers the limiting random vector $(\bm{W}^-,\bm{W}^+,\bm{X},S,C,L,A)$. It would then be possible that $L=C$ in the limit $n\to\infty$ while $L_n<C_n$ almost surely for all $n\in\N$. This situation would not be distinguishable from $L_n=C_n$ for all $n\in\N$ and in order to derive meaningful results in Theorem \ref{5:thm:non-resilience} we have to choose $\ell_i>c_i$. Since $\circRho(u)=\rho(u)=\rho(1)$ for all $u>1$, this does not affect the contagion process.

In contrast to Definition \ref{5:def:resilience} of resilience, we call a financial system \emph{non-resilient} if any small shock causes a lower bounded linear damage by bankrupt institutions.
\begin{definition}[Non-resilience]\label{5:def:non:resilience}
A financial system is said to be \emph{non-resilient} if there exists $\Delta>0$ such that $n^{-1}\mathcal{S}_{n,L}>\Delta$ w.\,h.\,p.~for any $L$ with the above listed properties.
\end{definition}
The following theorem identifies lower bounds for the final default fraction and finally sold shares.

\begin{theorem}\label{5:thm:non-resilience}
If the initial shock $L$ satisfies above properties and $h^m(\bm{\chi})$ is strictly increasing in $\chi^m$ for all $m\in[M]$, then for any $\epsilon>0$ it holds w.\,h.\,p.~that
\[ n^{-1}\mathcal{S}_{n,L} > \circG(\bm{z}^*,\bm{\chi}^*)-\epsilon \quad\text{and}\quad \chi_{n,L}^m > (\chi^*)^m-\epsilon. \]
\end{theorem}
The assumption on $h(\bm{\chi})$ is a rather mild one and is satisfied for all standard choices for price impact functions such as linear price or log-linear price impact.
\begin{corollary}[Non-resilience Criterion]\label{5:cor:non:resilience}
If $h^m(\bm{\chi})$ is strictly increasing in $\chi^m$ for all \mbox{$m\in[M]$} and $\circG(\bm{z}^*,\bm{\chi}^*)>0$, then the system is non-resilient.
\end{corollary}
For most practical purposes Corollaries \ref{5:cor:resilience} and \ref{5:cor:non:resilience} hence fully determine whether a financial system is resilient or non-resilient.

\section{Applications \& Simulations}\label{5:sec:applications}
In this section, we provide two applications of our theory. Example \ref{5:ex:simulation} has a twofold purpose: It demonstrates the joint impact of default contagion and fire sales. The model parameters are chosen in such a way that the financial system would be resilient with respect to either one of them but non-resilient with respect to their combination. Further, we provide simulations for finite networks in this setting to confirm the applicability of our asymptotic results also for reasonably sized financial systems. In Example \ref{5:ex:combined}, we derive sufficient capital requirements for very general combined financial systems of default contagion and fire sales. This extends results from Chapter \ref{chap:systemic:risk} for pure default contagion and from Chapter \ref{chap:fire:sales} for pure fire sales. Besides certain global parameters that need to be determined by a regulating institution, these capital requirements for each institution $i\in[n]$ only depend on its asset holdings $\bm{x}_i$ and its in-weights $w_i^{-,r,\alpha}$ which can be thought of as in-degrees (cf.~Chapter \ref{chap:systemic:risk}). They are thus very transparent and can be computed locally by the institutions themselves as they only depend on the institutions own business decisions. Moreover, this prevents institutions from manipulating their own or others' capital requirements and serves as a fair allocation of risk in the financial system. For simplicity we assume $S\equiv 1$ throughout this section and hence consider the final default fraction as the measure of systemic damage.

\begin{example}
\label{5:ex:simulation}
Consider a financial system with $R=T=M=1$. For simplicity we omit superscripts throughout this example where appropriate. Let $w_i^-=w_i^+=x_i$ for each $i\in[n]$ and $W^-=W^+=X$ be Pareto distributed with density $f_X(x)=2x^{-3}\1\{x\geq1\}$. Further, let $c_i=3.5$ for each $i\in[n]$ and in particular $C=3.5$. Finally assume $h(\chi)=1-e^{-\chi}$ and $\rho(u)=\1\{u\geq1\}$, that is banks sell assets at default only.

Since $c_i>3$, the system without fire sales would then be resilient (see Theorem \ref{2:threshold:res}
). Also the pure fire sales system without loans would be resilient by Corollary \ref{4:cor:resilience} since for $\chi\leq h^{-1}(3.5)$
\begin{align*}
f(\chi) &= \E[X\1\{X\geq 3.5/h(\chi)\}] - \chi = \int_{3.5/h(\chi)}^\infty 2x^{-2}\dd x - \chi = \frac{4}{7}h(\chi) - \chi = \frac{4}{7}\left(1-e^{-\chi}\right) - \chi
\end{align*}
and hence $f'(0)=-3/7<0$.

However, for the combined contagion system, we derive that
\begin{align*}
f^{1,1,1}(z,\chi) &= 2 +2z\left(\mathrm{Ei}\left(-\frac{7z}{2h(\chi)}\right) - \mathrm{Ei}\left(-\frac{5z}{2h(\chi)}\right)\right)
+ \frac{4h(\chi)}{7}e^{-\frac{7z}{2h(\chi)}} + ze^{-\frac{3z}{2h(\chi)}} - (z+2)e^{-z}\\
&\hphantom{=2}-\frac{z}{3}\left( -e^{-\frac{z}{2h(\chi)}}\left(\frac{z}{2h(\chi)}+1\right) + e^{-z}(z+1) \right)\1\{\chi\leq\log 2\} - z,\\
f^1(z,\chi) &= f_0(z,\chi) + z - \chi,
\end{align*}
where $\mathrm{Ei}(x):=\int_{-\infty}^x t^{-1}e^t\dd t$ denotes the exponential integral. In particular, $f^{1,1,1}(z,z)=f^1(z,z)$ and
\[ \frac{\dd}{\dd z}f^{1,1,1}(z,z)\big\vert_{z=0} = -\frac{1}{3} + \frac{1}{2}e^{-\frac{1}{2}} + e^{-\frac{3}{2}} + \frac{4}{7}e^{-\frac{7}{2}} + 2\left(\mathrm{Ei}\left(-\frac{7}{2}\right) - \mathrm{Ei}\left(-\frac{5}{2}\right)\right) \approx 0.2462 > 0.\]
Hence the directional derivatives of $f^{1,1,1}$ and $f^1$ in direction $(1,1)$ are both positive and thus $z^*>0$ and $\chi^*>0$. See Figure \ref{5:fig:Ex1} for an illustration.

\begin{figure}[t]
\hfill\includegraphics[width=0.5\textwidth]{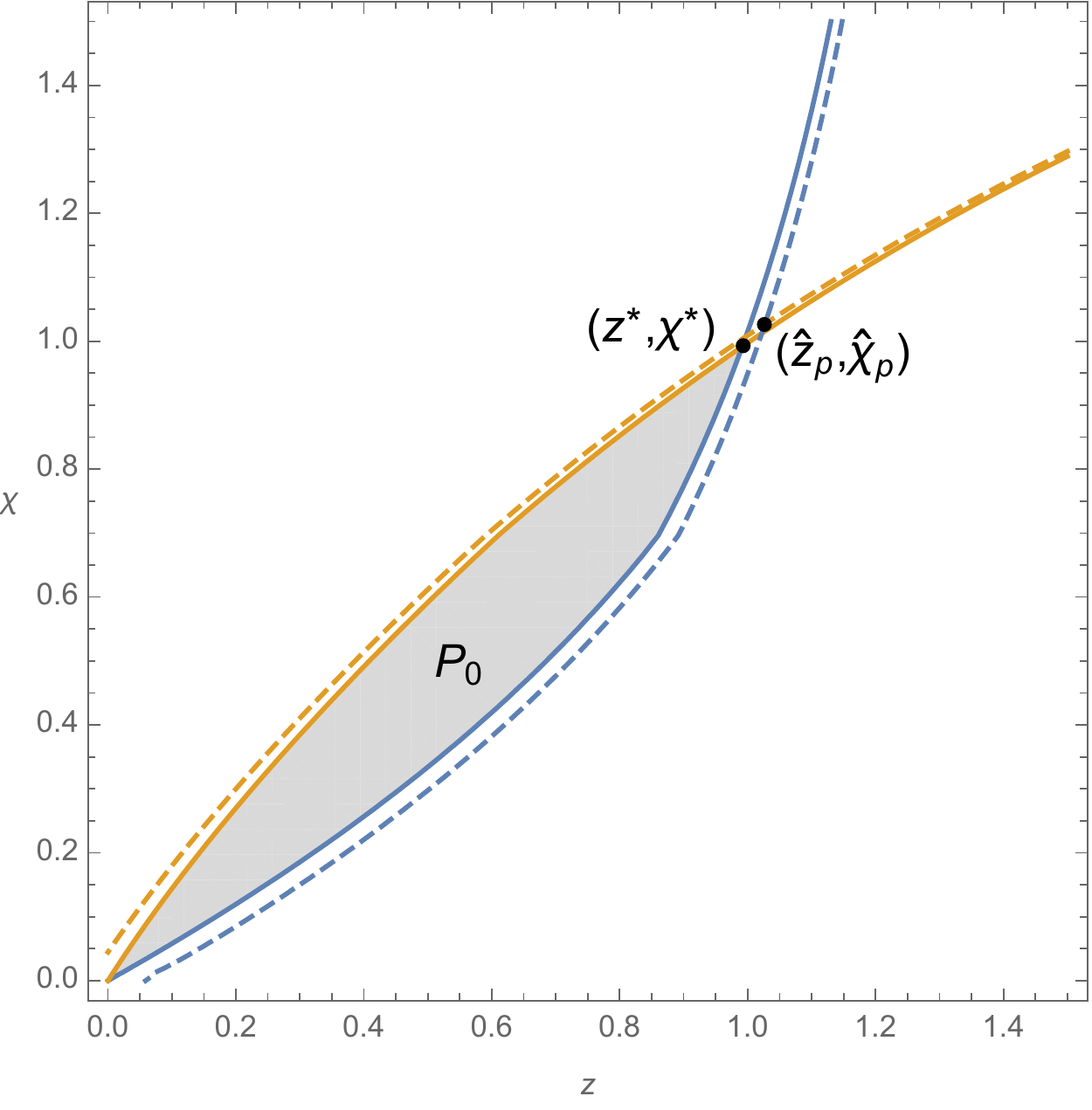} \hfill
\caption{Plot of the root sets of the functions $f^{1,1,1}(z,\chi)$ (blue) and $f^1(z,\chi)$ (orange). Solid: the unshocked functions. Dashed: the shocked functions. In grey the set $P=P_0$}\label{5:fig:Ex1}
\end{figure}
More precisely, we numerically determine $(z^*,\chi^*)\approx(0.992,0.992)$ and since \mbox{$g(z,\chi)=\circG(z,\chi)$} is given by
\begin{align*}
g(z,\chi) &= 1+ z^2\left(2\mathrm{Ei}\left(-\frac{5z}{2h(\chi)}\right)-\mathrm{Ei}\left(-\frac{3z}{2h(\chi)}\right)-\mathrm{Ei}\left(-\frac{7z}{2h(\chi)}\right)\right) - (z+1)e^{-z}\\
&\hspace{1.05cm} + \frac{2h(\chi)}{7}\left(\frac{2h(\chi)}{7} - z\right)e^{-\frac{7z}{2h(\chi)}} + \frac{4}{5}zh(\chi)e^{-\frac{5z}{2h(\chi)}}  + \frac{z^2}{3}\left(e^{-\frac{z}{2h(\chi)}}-e^{-z}\right)\1\{\chi\leq\log 2\}
\end{align*}
a lower bound on the final default fraction is asymptotically given by $g(z^*,\chi^*)\approx 29.24\%$. The combined system is thus non-resilient.

If we let each bank in the system initially default with probability $p=1\%$, then we can determine $(\hat{z}_p,\hat{\chi}_p)=(z_p^*,\chi_p^*)\approx(1.028,1.028)$ as the unique joint root 
of the functions $\circFSuper{0}(z,\chi)=f_p^{1,1,1}(z,\chi)=(1-p)f^{1,1,1}(z,\chi)+p(2-z)$ and $\circFSuper{1}(z,\chi)=f_p^1(z,\chi)=(1-p)f^1(z,\chi)+p(2-\chi)$. Plugging it into $\circG_p(z,\chi)=g_p(z,\chi)=(1-p)g(z,\chi)+p$ yields an asymptotic final fraction of $31.32\%$.

To verify this result for finite systems, we performed $10^5$ simulations on systems of sizes between $10^2$ and $10^4$ ($1000$ simulations for every multiple of $100$) as well as $10^5$ simulations on systems of sizes between $10^3$ and $10^5$ ($1000$ simulations for every multiple of $1000$), where we drew $x_i$ randomly according to the limiting distribution of $X$. Figure \ref{5:fig:simulations} shows the mean over all $1000$ simulations as an orange curve. Additionally, $100$ simulations for every system size are depicted by blue dots. The theoretical final fraction of $31.32\%$ is drawn as a red line. While for small $n$ only few simulations ended in a final default fraction significantly larger than $p=1\%$ and those which did were considerably higher than the theoretical value of $31.32\%$, as $n$ becomes larger, the average final fraction converges to $31.32\%$ and the deviation around this value becomes smaller and smaller. Already for $n\approx 4,000$ the simulated and the theoretical results are considerably close.
\begin{figure}[t]
    \hfill\subfigure[]{\includegraphics[width=0.49\textwidth]{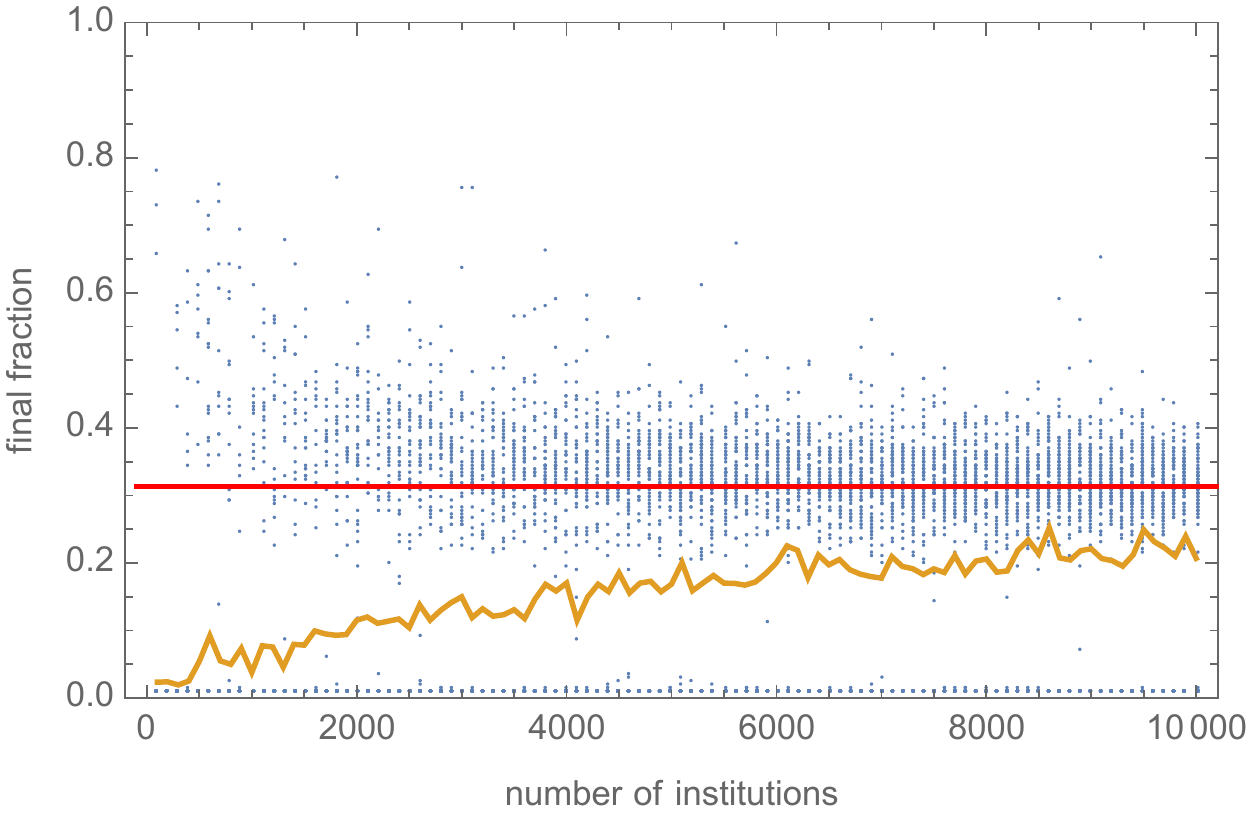}\label{5:fig:convergence:10000}}
    \hfill\subfigure[]{\includegraphics[width=0.49\textwidth]{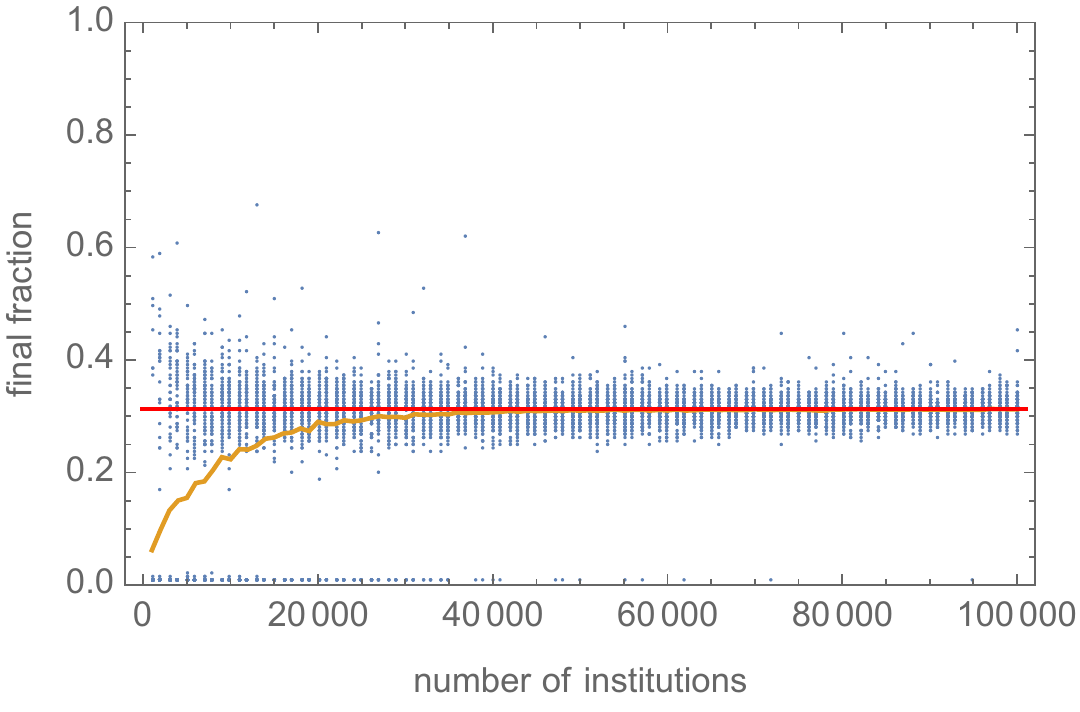}\label{5:fig:convergence100000}}\hfill
\caption{The simulation outcomes for systems as described in Example \ref{5:ex:simulation}. In blue single outcomes, in orange the mean over all outcomes and in red the theoretical asymptotic final fraction. 
}\label{5:fig:simulations}
\end{figure}
\end{example}

\begin{example}[Capital Requirements]\label{5:ex:combined}
In the previous example, we considered the case that \mbox{$\rho(u)=\1\{u\geq1\}$} with sales at default only. Intermediate sales will make the system less resilient, and we shall consider such an example now. We choose $\rho(u)=1\wedge u^q$ for some $q>0$. We consider one asset only and the parameter $q$ could be understood as a measure for the banks' confidence in the asset. Further, assume that the price impact is $h(\chi)=\Theta(\chi^\nu)$ for small $\chi$ and $\nu\geq q^{-1}$, i.\,e.~there exist constants $\mu_1,\mu_2\in(0,\infty)$ such that $\mu_1 \chi^\nu\leq h(\chi) \leq \mu_2 \chi^\nu$ for $\chi\leq\chi_0$ small enough.
The generalization to multiple assets is straight forward in analogy to Corollary \ref{4:cor:multiple:assets}.

The distribution of asset holdings is assumed to have a power law tail in the sense that $1-F_X(x)=\Theta(x^{1-\beta})$ for some $\beta\in(2,\infty)$, i.\,e.~there exist constants $B_1,B_2\in(0,\infty)$ such that $B_1 x^{1-\beta}\leq 1-F_X(x) \leq B_2 x^{1-\beta}$ for $x\geq x_0$ large enough. 

First assume that $R=T=1$. Recall then from Theorem \ref{2:thm:Pareto:type} the sufficient (and necessary) capital requirements for a pure default contagion model without fire sales: Assume \mbox{$1-F_{W^\pm}(w)\leq (w/K^\pm)^{1-\beta^\pm}$} for constants $K^\pm\in(0,\infty)$ and $\beta^\pm>2$, and for $w\geq w_0\in\R_+$. That is, the tails of the distributions of $W^-$ and $W^+$ are at most of power $\beta^-$ resp.~$\beta^+$. If we let $\gamma_c:=2+\frac{\beta^--1}{\beta^+-1}-\beta^-$ and $c_i=c(w_i^-)$ for each bank $i\in[n]$ with $c:\R_{+,0}\to(1,\infty)$, then the (pure default contagion) system is resilient if either $\gamma_c<0$, $\gamma_c>0$ and $\liminf_{w\to\infty}w^{-\gamma_c}c(w)>\frac{\beta^+-1}{\beta^+-2}K^+(K^-)^{1-\gamma_c}=:\alpha_c$ or $\gamma_c=0$ and $\liminf_{w\to\infty}c(w)>\alpha_c+1$. It thus makes sense to define capital requirements $c^\text{dir}(w)=\max\{2,\alpha w^\gamma\}$ for some constants $\alpha>\alpha_c$ and $\gamma\geq\gamma_c$.

Adding the capital requirements $c^\text{dir}$ against direct contagion to the capital requirements $c^\text{ind}(x)=\theta x$ (where $\theta>\mu\E[X]$) 
against fire sales (indirect contagion) found in Corollary \ref{4:cor:multiple:assets}, we thus get the combined capital requirement $c_i\geq c(w_i^-,x_i)$ for each $i\in[n]$, where
\[ c(w,x)=\max\{2,\alpha w^\gamma\} + \theta x. \]
In fact, we can show that these capital requirements make the combined system resilient: By Corollary \ref{4:cor:multiple:assets} it holds
\[ f^1(0,\chi) = \E\left[X\min\left\{1,\left(\frac{Xh(\chi)}{C}\right)^q\right\}\right] - \chi < 0 \]
for $\chi>0$ small enough since $C\geq\theta X$. Since $f^1(z,\chi)$ is continuous in $z$ for fixed $\chi$, we can then choose $z>0$ small enough such that still $f^1(z,\chi)<0$. Furthermore, it holds for $\chi<h^{-1}(\theta)$ and $z>0$ small enough that
\[ f^{1,1,1}(z,\chi) \leq \E\left[W^+\P\left(\mathrm{Poi}(W^-z)\geq\max\left\{2,\alpha (W^-)^\gamma\right\}\right)\right] - z < 0 \] 
by resilience of the pure default contagion system (see the proof of Theorem \ref{2:thm:Pareto:type}). By definition of $(z^*,\chi^*)$ we can then conclude $z^*<z$ and $\chi^*<\chi$. However, $z$ and $\chi$ can be chosen arbitrarily small and thus $z^*=\chi^*=0$. The combined system is then resilient by Corollary \ref{5:cor:resilience}.

\pagebreak
For the case of general $R,T\in\N$, we obtain sufficient capital requirements against default contagion from Corollary \ref{3:cor:resilient}. Thus for an institution $i\in[n]$ of type $\beta\in[T]$ choose
\[ c_i^\text{dir}(\bm{v})=\max\left\{R+1,\left\lceil\mu \left( \frac{\sum_{s\in[R]}s\sum_{\gamma\in[T]}w_i^{-,s,\gamma}v^{s,\beta,\gamma}}{\Vert\bm{v}\Vert} \right)^\nu\right\rceil\right\}, \]
where $\nu\geq\nu_c^\beta$ and $\mu>\mu_c^\beta$ as defined in Section \ref{3:ssec:suff:cap:requ}. By the same means as for the one-dimensional case above, we then derive that
\[ c_i\geq c_i^\text{dir} + c^\text{ind}(x_i) \]
is sufficient for resilience of the financial system. Again the generalization to multiple assets ($M\geq2$) is straightforward by Corollary \ref{4:cor:multiple:assets}.
\end{example}

\section{Proofs}\label{5:sec:proofs}
\subsection{Proofs for Section \ref{5:sec:system:model}}
\begin{proof}[Proof of Lemma \ref{5:lem:cont:rho}]
Clearly,
\begin{align*}
\mathcal{D}_{(k)} &= \bigg\{i\in[n]\,:\,\sum_{j\in\mathcal{D}_{(k-1)}} e_{j,i} \geq c_i-\ell_i-\bm{x}_i\cdot h(n^{-1}\bm{\sigma}_{(k-1)})\bigg\}\\
&\subseteq \bigg\{i\in[n]\,:\,\sum_{j\in\mathcal{D}_n} e_{j,i} \geq c_i-\ell_i-\bm{x}_i\cdot h(\bm{\chi}_n)\bigg\}
\end{align*}
and thus
\[ \mathcal{D}_n=\bigcup_{k\in\N}\mathcal{D}_{(k)} \subseteq \bigg\{i\in[n]\,:\,\sum_{j\in\mathcal{D}_n} e_{j,i} \geq c_i-\ell_i-\bm{x}_i\cdot h(\bm{\chi}_n)\bigg\}. \]
On the other hand, if $i\in\mathcal{D}_n$, then there exists $k_i\in\N$ such that
\[ 0 \geq c_i-\ell_i-\bm{x}_i\cdot h(n^{-1}\bm{\sigma}_{(k_i)}) - \sum_{j\in\mathcal{D}_{(k_i)}} e_{j,i} \geq c_i-\ell_i-\bm{x}_i\cdot h(\bm{\chi}_n) - \sum_{j\in\mathcal{D}_n} e_{j,i} \]
and thus 
\[ \mathcal{D}_n \subseteq \bigg\{i\in[n]\,:\,\sum_{j\in\mathcal{D}_n} e_{j,i} \geq c_i-\ell_i-\bm{x}_i\cdot h(\bm{\chi}_n)\bigg\}. \]
That is, $\mathcal{D}_n$ and $\bm{\chi}_n$ solve \eqref{5:eqn:lem:cont:rho:D}. Moreover, they solve \eqref{5:eqn:lem:cont:rho:chi} as
\begin{align*}
\bm{\chi}_n &= n^{-1}\lim_{k\to\infty}\bm{\sigma}_{(k)} = \lim_{k\to\infty} n^{-1}\sum_{i\in[n]} \bm{x}_i \rho\left( \frac{\sum_{j\in\mathcal{D}_{(k-1)}}e_{j,i} + \ell_i+\bm{x}_i\cdot h(n^{-1}\bm{\sigma}_{(k-1)})}{c_i} \right)\\
&= n^{-1}\sum_{i\in[n]} \bm{x}_i \rho\left( \frac{\sum_{j\in\mathcal{D}_n}e_{j,i} + \ell_i+\bm{x}_i\cdot h(\bm{\chi}_n)}{c_i} \right)
\end{align*}
where we used continuity of $\rho$ and $h$ and the fact that $\mathcal{D}_{(k)}=\mathcal{D}_n$ for $k$ large enough.

Now assume that $\tilde{\mathcal{D}}$ and $\tilde{\bm{\chi}}$ also solve \eqref{5:eqn:lem:cont:rho:D} and \eqref{5:eqn:lem:cont:rho:chi}. Clearly $\mathcal{D}_{(0)}\subseteq\tilde{\mathcal{D}}$ and $\bm{\sigma}_{(0)}\leq n\tilde{\bm{\chi}}$. Hence assume inductively that $\mathcal{D}_{(k)}\subseteq\tilde{\mathcal{D}}$ and $\bm{\sigma}_{(k)}\leq n\tilde{\bm{\chi}}$. Then
\begin{align*}
\mathcal{D}_{(k+1)} &= \bigg\{i\in[n]\,:\,\sum_{j\in\mathcal{D}_{(k)}} e_{j,i} \geq c_i-\ell_i-\bm{x}_i\cdot h(n^{-1}\bm{\sigma}_{(k)})\bigg\}\\
&\subseteq \bigg\{i\in[n]\,:\,\sum_{j\in\tilde{\mathcal{D}}} e_{j,i} \geq c_i-\ell_i-\bm{x}_i\cdot h(\tilde{\bm{\chi}})\bigg\} = \tilde{\mathcal{D}}
\end{align*}
and
\begin{align*}
\bm{\sigma}_{(k+1)} &= \sum_{i\in[n]} \bm{x}_i \rho\left( \frac{\sum_{j\in\mathcal{D}_{(k)}}e_{j,i} + \ell_i+\bm{x}_i\cdot h(n^{-1}\bm{\sigma}_{(k)})}{c_i} \right)\\
&\leq \sum_{i\in[n]} \bm{x}_i \rho\left( \frac{\sum_{j\in\tilde{\mathcal{D}}}e_{j,i} + \ell_i+\bm{x}_i\cdot h(\tilde{\bm{\chi}})}{c_i} \right) = n\tilde{\bm{\chi}}.
\end{align*}
In particular, $\mathcal{D}_n=\bigcup_{k\in\N}\mathcal{D}_{(k)} \subseteq\tilde{\mathcal{D}}$ and $\bm{\chi}_n=n^{-1}\lim_{k\to\infty}\bm{\sigma}_{(k)}\leq\tilde{\bm{\chi}}$.
\end{proof}

\subsection{Proofs for Section \ref{5:asymp:res:sys}}\label{5:ssec:proofs:2}
We first consider the special case, summarized in the following definition, where the weights, asset holdings, capitals and exogenous losses take only finitely many different values.
\begin{definition}[Finitary Regular Vertex Sequence]\label{5:def:finitary:vertex:sequence}
A regular vertex sequence (see Assumption \ref{5:ass:regular:vertex:sequence}) denoted by $(\bm{w}^-(n),\bm{w}^+(n),\bm{x}(n),\bm{s}(n),\bm{c}(n),\bm{\ell}(n),\bm{\alpha}(n))_{n\in\N}$ is called \emph{finitary} if there exist $J\in\N$ and a finite set $\{(\tilde{\bm{w}}_j^-,\tilde{\bm{w}}_j^+,\tilde{\bm{x}}_j,\tilde{s}_j,\tilde{c}_j,\tilde{\ell}_j)\}_{j\in[J]}\subset\R_{+,0}^{[R]\times[T]}\times\R_{+,0}^{[R]\times[T]}\times\R_{+,0}^M\times\R_{+,0}\times\R_{+,\infty}\times\R_{+,0}$ such that for all $n\in\N$ and $i\in[n]$, there exists $j=j(n,i)\in[J]$ such that $(\bm{w}_i^-,\bm{w}_i^+,\bm{x}_i,s_i,c_i,\ell_i)=(\tilde{\bm{w}}_j^-,\tilde{\bm{w}}_j^+,\tilde{\bm{x}}_j,\tilde{s}_j,\tilde{c}_j,\tilde{\ell}_j)$. Denote in the following
\[ p_j^\beta(n):=\P(\bm{W}^-_n=\tilde{\bm{w}}_j^-,\bm{W}^+_n=\tilde{\bm{w}}_j^+,\bm{X}_n=\tilde{\bm{x}}_j,S_n=\tilde{s}_j,C_n=\tilde{c}_j,L_n=\tilde{\ell}_j,A_n=\beta) \]
and
\[ p_j^\beta = \lim_{n\to\infty}p_j^\beta(n) = \P(\bm{W}^-=\tilde{\bm{w}}_j^-,\bm{W}^+=\tilde{\bm{w}}_j^+,\bm{X}=\tilde{\bm{x}}_j,S=\tilde{s}_j,C=\tilde{c}_j,L=\tilde{\ell}_j,A=\beta). \]
\end{definition}
We can then prove a version of Theorem \ref{5:thm:final:fraction:combined:general} for the finitary case:
\begin{theorem}\label{5:thm:final:fraction:combined:finitary}
Consider a financial system described by a finitary regular vertex sequence. Then for the final systemic damage $n^{-1}\mathcal{S}_n$ and $\chi_n^m$, the number of finally sold shares of asset $m\in[M]$ divided by $n$, it holds
\[\begin{gathered}
\circG(\hat{\bm{z}},\hat{\bm{\chi}}) + o_p(1) \leq n^{-1}\mathcal{S}_n \leq g(\bm{z}^*,\bm{\chi}^*) + o_p(1),\\
\hat{\chi}^m+o_p(1) \leq \chi_n^m \leq (\chi^*)^m+o_p(1).
\end{gathered}\]
In particular, for the final price impact $h^m(\bm{\chi}_n)$ on asset $m\in[M]$ it holds
\[ h^m(\hat{\bm{\chi}})+o_p(1) \leq h^m(\bm{\chi}_n) \leq h^m(\bm{\chi}^*)+o_p(1). \]
\end{theorem}
The difficulty for this problem lies in the fact that the functions $f^{r,\alpha,\beta}$ 
are discontinuous in $\bm{\chi}$. That is, there exist values for $\bm{\chi}$ (sold assets) at which a linear fraction of banks defaults. 
However, $f^{r,\alpha,\beta}$ is discontinuous at $(\bm{z},\bm{\chi})$  only if $(\tilde{c}_j-\tilde{\ell}_j-\tilde{\bm{x}}_j\cdot h(\bm{\chi}))\in\N$ for some $j\in J$ and there are hence only finitely many (possibly degenerated) hyperplanes of discontinuities. 
\begin{proof}
We start with the proof of the lower bounds. That is, for arbitrary $\epsilon>0$ we will show that $n^{-1}\mathcal{S}_n \geq n^{-1}\hat{\mathcal{S}}_n\geq (1-\epsilon)\circG(\hat{\bm{z}},\hat{\bm{\chi}})$ and $\chi_n^m\geq\hat{\chi}_n^m\geq(1-\epsilon)\hat{\chi}^m$ w.\,h.\,p. We therefore consider the contagion process given by rounds (i') and (ii'). That is, we first consider a cascade of default contagion. Once this cascade has ended (after at most $n-1$ steps) we start a cascade of fire sales and so on.

In order to quantify the default contagion cascade we use Theorem \ref{3:thm:general:weights}. That is, if we denote by $\hat{\bm{z}}_1\in\R_{+,0}^V$ the smallest vector such that
\[ \E\left[W^{+,r,\alpha}\P\left(\sum_{s\in[R]}s\mathrm{Poi}\left(\sum_{\gamma\in[T]}W^{-,s,\gamma}\hat{z}_1^{s,\beta,\gamma}\right)> C-L\right)\1\{A=\beta\}\right]=\hat{z}_1^{r,\alpha,\beta} \]
for all $(r,\alpha,\beta)\in V$, then the systemic importance of finally defaulted banks is lower bounded by
\[ (1-\delta)n\sum_{\beta\in[T]}\E\left[S\P\left(\sum_{s\in[R]}s\mathrm{Poi}\left(\sum_{\gamma\in[T]}W^{-,s,\gamma}\hat{z}_1^{s,\beta,\gamma}\right)> C-L\right)\1\{A=\beta\}\right] \]
w.\,h.\,p.~for any fixed $\delta>0$. In fact, by finitariness of the system we can find $\theta>0$ small enough such that
\begin{align*}
&\E\left[W^{+,r,\alpha}\P\left(\sum_{s\in[R]}s\mathrm{Poi}\left(\sum_{\gamma\in[T]}W^{-,s,\gamma}z^{s,\beta,\gamma}\right)> C-L\right)\1\{A=\beta\}\right]\\
&\hspace{3cm} = 
\E\left[W^{+,r,\alpha}\P\left(\sum_{s\in[R]}s\mathrm{Poi}\left(\sum_{\gamma\in[T]}W^{-,s,\gamma}z^{s,\beta,\gamma}\right)\geq \left\lceil C-L+\theta\right\rceil\right)\1\{A=\beta\}\right]
\end{align*}
(note that $\left\lceil C-L + \theta\right\rceil$ is the weak limit of $\left\lceil C_n-L_n + \theta\right\rceil$ again by finitariness) and we are thus in the setting of Chapter \ref{chap:block:model}. 
However, while Theorem \ref{3:thm:general:weights} focuses on the systemic damage due to defaulted banks only, here it is also important to keep track of all losses due to defaults. In fact, the proof of Theorem \ref{3:thm:general:weights} for finitary systems shows that the number $q_{j,k}^\beta$ of institutions of class $j$ and type $\beta$ with a total edge weight from finally defaulted neighbors of at least $k\leq\tilde{c}_j$ is lower bounded by
\[ (1-\delta)np_j^\beta\P\left(\sum_{s\in[R]}s\mathrm{Poi}\left(\sum_{\gamma\in[T]}\tilde{w}_j^{-,s,\gamma}\hat{z}_1^{s,\beta,\gamma}\right)\geq k\right) \]
w.\,h.\,p.~for $\delta>0$. Since in this part of the proof we are interested in lower bounds, we assume in the following that $q_{j,\lceil\tilde{c}_j-\tilde{\ell}_j\rceil}^\beta=(1-\delta)np_j^\beta\P(\sum_{s\in[R]}s\mathrm{Poi}(\sum_{\gamma\in[T]}\tilde{w}_j^{-,s,\gamma}\hat{z}_1^{s,\beta,\gamma})\geq \lceil\tilde{c}_j-\tilde{\ell}_j\rceil)$, $q_{j,k}^\beta=(1-\delta)np_j^\beta\P(\sum_{s\in[R]}s\mathrm{Poi}(\sum_{\gamma\in[T]}\tilde{w}_j^{-,s,\gamma}\hat{z}_1^{s,\beta,\gamma}) = k)$ for $1\leq k\leq \lceil\tilde{c}_j-\tilde{\ell}_j\rceil-1$, and \mbox{$q_{j,0}^\beta=np_j^\beta(n)-\sum_{k=1}^{\lceil\tilde{c}_j-\tilde{\ell}_j\rceil}q_{j,k}^\beta$} w.\,h.\,p. That is, we increase the losses due to default contagion.

Next, we want to use Theorem \ref{4:thm:fire:sale:final:fraction} to quantify the impact of the round of fire sales. We need to consider losses (and defaults in particular) due to the previous default contagion cascade. That is, we need to add to the exogenous losses $\ell_i$ the edge weight from defaulted debtors of each bank $i\in[n]$. This leads to a new loss vector $\big(\ell_i'\big)_{i\in[n]}$. Note that we can set $\ell_i'=\tilde{\ell}_j+\lceil\tilde{c}_j-\tilde{\ell}_j\rceil$ if $i$ is of type $j$ and the total edge-weight $k$ from finally defaulted debtors of $i$ is larger or equal to $\lceil\tilde{c}_j-\tilde{\ell}_j\rceil$. Denoting by $L_n'$ a random vector distributed according to the empirical distribution function of $\big(\ell_i'\big)_{i\in[n]}$, we thus derive that w.\,h.\,p.
\begin{align*}
&\P\left(\bm{W}_n^-=\tilde{\bm{w}}_j^-,\bm{W}_n^+=\tilde{\bm{w}}_j^+,\bm{X}_n=\tilde{\bm{x}}_j,S_n=\tilde{s}_j,C_n=\tilde{c}_j,L_n'=\tilde{\ell}_j+k,A_n=\beta\right) \\
&\hspace{0.1cm} =\begin{cases}
(1-\delta)p_j^\beta\P\left(\sum_{s\in[R]}s\mathrm{Poi}\left(\sum_{\gamma\in[T]}\tilde{w}_j^{-,s,\gamma}\hat{z}_1^{s,\beta,\gamma}\right)\geq \left\lceil\tilde{c}_j-\tilde{\ell}_j\right\rceil\right),&\text{if }k=\left\lceil\tilde{c}_j-\tilde{\ell}_j\right\rceil,\\
(1-\delta)p_j^\beta\P\left(\sum_{s\in[R]}s\mathrm{Poi}\left(\sum_{\gamma\in[T]}\tilde{w}_j^{-,s,\gamma}\hat{z}_1^{s,\beta,\gamma}\right) = k\right),&\text{if }1\leq k\leq\left\lceil\tilde{c}_j-\tilde{\ell}_j\right\rceil-1, \\
p_j^\beta(n) - (1-\delta)p_j^\beta\P\left(\sum_{s\in[R]}s\mathrm{Poi}\left(\sum_{\gamma\in[T]}\tilde{w}_j^{-,s,\gamma}\hat{z}_1^{s,\beta,\gamma}\right)\geq 1\right),&\text{if }k=0.
\end{cases}
\end{align*}
For simplicity in the notation, we assume that from $(\tilde{\bm{x}}_j,\tilde{s}_j,\tilde{c}_j,\tilde{\ell}_j)=(\tilde{\bm{x}}_k,\tilde{s}_k,\tilde{c}_k,\tilde{\ell}_k)$ it follows \mbox{$j=k$} (i.\,e.~classes $j$ and $k$ are not distinguished by their in- and out-weights only) in the following. Otherwise consider sums over classes with the same asset holdings, systemic importance, capital and exogenous loss.

In particular, for the weak limit $(\bm{X},S,C,L',A)$ of $(\bm{X}_n,S_n,C_n,L_n',A_n)$ and $0\leq k\leq \lceil\tilde{c}_j-\tilde{\ell}_j\rceil$ it holds,
\[ \P\Big(\bm{X}=\tilde{\bm{x}}_j,S=\tilde{s}_j,C=\tilde{c}_j,L'=\tilde{\ell}_j+k,A=\beta\Big) \geq (1-\delta)p_j^\beta\P\Bigg(\sum_{s\in[R]}s\mathrm{Poi}\Bigg(\sum_{\gamma\in[T]}\tilde{w}_j^{-,s,\gamma}\hat{z}_1^{s,\beta,\gamma}\Bigg) = k\Bigg). \]
Let now
\[ \circFSuperSub{m}{\delta}(\bm{\chi}) := \E\left[X^m\circRho\left(\frac{L'+\bm{X}\cdot h(\bm{\chi})}{C}\right)\right] - \chi^m \]
the corresponding functions as in Chapter \ref{chap:fire:sales} and $\hat{\bm{\chi}}_\delta$ its smallest fixed point. Then
\begin{align*}
&\circFSuperSub{m}{\delta}(\bm{\chi})+\chi^m\\
&\qquad= \sum_{\beta\in[T]}\sum_{j\in [J]}\sum_{k\geq0}\tilde{x}_j^m\circRho\left(\frac{\tilde{\ell}_j+k+\tilde{\bm{x}}_j\cdot h(\bm{\chi})}{\tilde{c}_j}\right) \P(\bm{X}=\tilde{\bm{x}}_j,S=\tilde{s}_j,C=\tilde{c}_j,L'=\tilde{\ell}_j+k,A=\beta)\\
&\qquad\geq (1-\delta)\sum_{\beta\in[T]}\sum_{j\in [J]}\sum_{k\geq0}\tilde{x}_j^m\circRho\left(\frac{\tilde{\ell}_j+k+\tilde{\bm{x}}_j\cdot h(\bm{\chi})}{\tilde{c}_j}\right) p_j^\beta\P\Bigg(\sum_{s\in[R]}s\mathrm{Poi}\Bigg(\sum_{\gamma\in[T]}\tilde{w}_j^{-,s,\gamma}\hat{z}_1^{s,\beta,\gamma}\Bigg) = k\Bigg)\\
&\qquad= (1-\delta)(\circFSuper{m}(\hat{\bm{z}}_1,\bm{\chi})+\chi^m)
\end{align*}
and $\circFSuperSub{m}{\delta}(\bm{\chi})\geq(1-\delta)\circFSuper{m}(\hat{\bm{z}}_1,\bm{\chi})-\delta\chi^m\geq\circFSuper{m}(\hat{\bm{z}}_1,\bm{\chi})-\delta\E[X^m]$. In particular, $\liminf_{\delta\to0+}\circFSuperSub{m}{\delta}(\bm{\chi})\geq\circFSuper{m}(\hat{\bm{z}}_1,\bm{\chi})$ and by Lemma \ref{4:lem:convergence:chi:hat} we derive that $\liminf_{\delta\to0+}\hat{\bm{\chi}}_\delta\geq\hat{\bm{\chi}}_1$, where $\hat{\bm{\chi}}_1$ denotes the smallest joint root of the functions $\circFSuper{m}(\hat{\bm{z}}_1,\bm{\chi})$, $m\in[M]_0$, for fixed $\bm{z}=\hat{\bm{z}}_1$.

We can hence choose $\delta$ small enough such that the number of finally sold shares of asset $m$ is lower bounded by $n(1-\epsilon)\hat{\chi}_1^m$ w.\,h.\,p.~by Theorem \ref{4:thm:fire:sale:final:fraction}. Further, for\begin{align*}
&\circG_\delta(\bm{\chi}) := \E\left[S\1\left\{L'+\bm{X}\cdot h(\bm{\chi})> C\right\}\right]\\
&\hphantom{:}=\sum_{\beta\in[T]}\sum_{j\in[J]}\tilde{s}_j\sum_{k\geq0}\1\left\{ k > \tilde{c}_j-\tilde{\ell}_j-\tilde{\bm{x}}_j\cdot h(\bm{\chi}) \right\}\P\left(\bm{X}=\tilde{\bm{x}}_j,S=\tilde{s}_j,C=\tilde{c}_j,L'=\tilde{\ell}_j+k,A=\beta\right)\\
&\hphantom{:}\geq(1-\delta)\sum_{\beta\in[T]}\sum_{j\in[J]}\tilde{s}_j\sum_{k\geq0} \1\left\{ k > \tilde{c}_j-\tilde{\ell}_j-\tilde{\bm{x}}_j\cdot h(\bm{\chi}) \right\}p_j^\beta\P\Bigg(\sum_{s\in[R]}s\mathrm{Poi}\Bigg(\sum_{\gamma\in[T]}\tilde{w}_j^{-,s,\gamma}\hat{z}_1^{s,\beta,\gamma}\Bigg) = k\Bigg)\\
&\hphantom{:}= (1-\delta)\circG(\hat{\bm{z}}_1,\bm{\chi})
\end{align*}
and possibly further reducing $\delta$, we derive $n^{-1}\hat{\mathcal{S}}_n\geq\sqrt{1-\epsilon}\circG_\delta(\hat{\bm{\chi}}_\delta)\geq(1-\epsilon)\circG(\hat{\bm{z}}_1,\hat{\bm{\chi}}_1)$ w.\,h.\,p. So if $(\hat{\bm{z}}_1,\hat{\bm{\chi}}_1)=(\hat{\bm{z}},\hat{\bm{\chi}})$, then this finishes the proof of the lower bounds.

If $(\hat{\bm{z}}_1,\hat{\bm{\chi}}_1)\neq(\hat{\bm{z}},\hat{\bm{\chi}})$, then by construction of $(\hat{\bm{z}},\hat{\bm{\chi}})$, $\hat{\bm{z}}_1$ and $\hat{\bm{\chi}}_1$ it must hold that $\hat{\bm{\chi}}_1\leq\hat{\bm{\chi}}$ and $\hat{\bm{z}}_1\lneq\hat{\bm{z}}_2\leq\hat{\bm{z}}$, where $\hat{\bm{z}}_2\in\R_{+,0}^V$ denotes the smallest vector such that for all $(r,\alpha,\beta)\in V$,
\[ \E\left[W^{+,r,\alpha}\P\left(\sum_{s\in[R]}s\mathrm{Poi}\left(\sum_{\gamma\in[T]}W^{-,s,\gamma}\hat{z}_2^{s,\beta,\gamma}\right)> C-L-\bm{X}\cdot h(\hat{\bm{\chi}}_1)\right)\1\{A=\beta\}\right] = \hat{z}_2^{r,\alpha,\beta}. \]
The next step in the cascade process would now be the default contagion cascade from (i') starting from the state of the system after the fire sales cascade. Note, however, that (w.\,h.\,p.) we can equivalently restart the whole cascade process if for the default contagion cascade we choose capitals $c_i-\ell_i-\bm{x}_i\cdot h((1-\epsilon)\hat{\bm{\chi}}_1)$. If anything this reduces contagion effects which is alright because we are interested in lower bounds.

By the finitariness of the system, if we choose $\epsilon$ small enough, then $\hat{\bm{z}}_2$ is also the smallest solution of
\[ \E\left[W^{+,r,\alpha}\P\left(\sum_{s\in[R]}s\mathrm{Poi}\left(\sum_{\gamma\in[T]}W^{-,s,\gamma}\hat{z}_2^{s,\beta,\gamma}\right)> C-L-\bm{X}\cdot h((1-\epsilon)\hat{\bm{\chi}}_1)\right)\1\{A=\beta\}\right] = \hat{z}_2^{r,\alpha,\beta}. \]
and we can hence use Theorem \ref{3:thm:general:weights} to quantify the default contagion cascade. By exactly the same means as above we can then translate the losses due to default contagion into exogenous losses and investigate the fire sales process by Theorem \ref{4:thm:fire:sale:final:fraction} and we derive that the vector of finally sold shares is lower bounded by $n(1-\epsilon)\hat{\bm{\chi}}_2$ w.\,h.\,p., where $\hat{\bm{\chi}}_2$ denotes the smallest joint root of the functions $\circFSuper{m}(\hat{\bm{z}}_2,\bm{\chi})$ for fixed $\bm{z}=\hat{\bm{z}}_2$, and $n^{-1}\hat{\mathcal{S}}_n\geq(1-\epsilon)\circG(\hat{\bm{z}}_2,\hat{\bm{\chi}}_2)$ w.\,h.\,p.

Again, if $(\hat{\bm{z}}_2,\hat{\bm{\chi}}_2)=(\hat{\bm{z}},\hat{\bm{\chi}})$, then this finishes the proof of the lower bounds. Otherwise we can continue on for $t\geq3$. Note, however, that $\hat{\bm{z}}\gneq\hat{\bm{z}}_t$ is only possible if $\hat{\bm{\chi}}_{t-1}$ and $\hat{\bm{\chi}}_t$ are separated by a hyperplane of discontinuity of $f^{r,\alpha,\beta}$ for some $(r,\alpha,\beta)\in V$ (and hence the fire sales lead to further defaults in the system). However, as remarked earlier, there can only be finitely many such hyperplanes for finitary systems. Hence by the procedure outlined above, we will reach $\hat{\bm{z}}$ in finitely many steps and hence the end results still hold w.\,h.\,p.

We can now turn to the second part of the proof. We consider the contagion process in rounds (i) and (ii) to derive upper bounds on $\overline{\mathcal{S}}_n$ and $\overline{\bm{\chi}}_n$. Let $(\tilde{\bm{z}}(\delta),\tilde{\bm{\chi}}(\delta))_{\delta>0}$ be the constructing sequence of $(\bm{z}^*,\bm{\chi}^*)$ analogue to Remark \ref{4:rem:sequence:chi:*}. Then note that by upper semi-continuity and the discrete nature of $f^{r,\alpha,\beta}$ we can find $\Delta>0$ such that $f^{r,\alpha,\beta}(\bm{z},\tilde{\bm{\chi}}(\delta))=f^{r,\alpha,\beta}(\bm{z},\bm{\chi}^*)$ for all $0\leq\delta<\Delta$, $(r,\alpha,\beta)\in V$ and $\bm{z}\in\R_{+,0}^V$. 

Fix now some $\delta\in(0,\Delta)$ and consider the financial system with reduced capital values $c_i-\ell_i-\bm{x}_i\cdot h(\tilde{\bm{\chi}}(\delta))$ for each bank $i\in[n]$. We only want to consider the default contagion process in this new financial system and we are hence in the setting of Chapter \ref{chap:block:model} with limiting random variables $(\bm{W}^-,\bm{W}^+,S,\lceil C-L-\bm{X}\cdot h(\tilde{\bm{\chi}}(\delta))\rceil^+,A)$. Note that by finitariness the regularity transfers. By the choice of $\delta$ above, we derive that for $\bm{z}_\delta^*$ in this new financial system, it holds $\bm{z}_\delta^*=\bm{z}^*$ 
and by Theorem \ref{3:thm:general:weights} 
we know that the final systemic damage in the new system is upper bounded by
\begin{equation}\label{5:eqn:upper:bound:D:delta}
n\sum_{\beta\in[T]}\E\left[S\P\left(\sum_{s\in[R]}s\mathrm{Poi}\left(\sum_{\gamma\in[T]}W^{-,s,\gamma}(z^*)^{s,\beta,\gamma}\right)\geq C-L-\bm{X}\cdot h(\tilde{\bm{\chi}}(\delta))\right)\1\{A=\beta\}\right]+o_p(n).
\end{equation}
The proof of Theorem \ref{3:thm:general:weights} actually shows that the number of banks of type $\beta$ and class $j$ with at least an edge-weight of $k$ from defaulted neighbors at the end of the default contagion process is upper bounded by
\[ (1+\epsilon)np_j^\beta\P\left(\sum_{s\in[R]}s\mathrm{Poi}\left(\sum_{\gamma\in[T]}\tilde{w}_j^{-,s,\gamma}(z^*)^{s,\beta,\gamma}\right)\geq k\right) \]
w.\,h.\,p.~for any fixed $\epsilon>0$.

Similarly as in the first part of this proof (for the lower bounds) we can then construct a fire sales system (as in Chapter \ref{chap:fire:sales}) with limiting random vector $(\bm{X},S,C,L')$ such that
\begin{align*}
&\P(\bm{X}=\tilde{\bm{x}}_j,S=\tilde{s}_j,C=\tilde{c}_j,L'=\tilde{\ell}_j+k,A=\beta)\\
&\hspace{0.75cm}= (1+\epsilon) p_j^\beta\P\left(\sum_{s\in[R]}s\mathrm{Poi}\left(\sum_{\gamma\in[T]}\tilde{w}_j^{-,s,\gamma}(z^*)^{s,\beta,\gamma}\right) = k\right),\quad 0\leq k<\left\lceil\tilde{c}_j-\tilde{\ell}_j-\tilde{\bm{x}}_j\cdot h(\tilde{\bm{\chi}}(\delta))\right\rceil,\\
&\P\left(\bm{X}=\tilde{\bm{x}}_j,S=\tilde{s}_j,C=\tilde{c}_j,L'=\tilde{\ell}_j+\left\lceil\tilde{c}_j-\tilde{\ell}_j-\tilde{\bm{x}}_j\cdot h(\tilde{\bm{\chi}}(\delta))\right\rceil,A=\beta\right)\\
&\hspace{2.75cm}= p_j^\beta - \sum_{k=0}^{\left\lceil\tilde{c}_j-\tilde{\ell}_j-\tilde{\bm{x}}_j\cdot h(\tilde{\bm{\chi}}(\delta))\right\rceil-1}(1+\epsilon) p_j^\beta\P\left(\sum_{s\in[R]}s\mathrm{Poi}\left(\sum_{\gamma\in[T]}\tilde{w}_j^{-,s,\gamma}(z^*)^{s,\beta,\gamma}\right) = k\right)\\
&\hspace{2.75cm}\leq (1+\epsilon) p_j^\beta\P\left(\sum_{s\in[R]}s\mathrm{Poi}\left(\sum_{\gamma\in[T]}\tilde{w}_j^{-,s,\gamma}(z^*)^{s,\beta,\gamma}\right) \geq \left\lceil\tilde{c}_j-\tilde{\ell}_j-\tilde{\bm{x}}_j\cdot h(\tilde{\bm{\chi}}(\delta))\right\rceil \right)
\end{align*}
which dominates the stochastic final state after the default contagion cascade w.\,h.\,p.

Let now
\[ f_\epsilon^m(\bm{\chi}) := \E\left[X^m\rho\left(\frac{L'+\bm{X}\cdot h(\bm{\chi})}{C}\right)\right] - \chi^m \]
the corresponding functions for the 
fire sales system as in Chapter \ref{chap:fire:sales} and $\bm{\chi}_\epsilon^*$ the corresponding value for $\bm{\chi}^*$ in Chapter \ref{chap:fire:sales}. Then for $f^m(\bm{z},\bm{\chi})$ as in Section \ref{5:asymp:res:sys},
\begin{align*}
&f_\epsilon^m(\bm{\chi})+\chi^m\\
&\hspace{0.25cm}= \sum_{j\in J}\sum_{k\geq0}\tilde{x}_j^m\rho\left(\frac{\tilde{\ell}_j+k+\tilde{\bm{x}}_j\cdot h(\bm{\chi})}{\tilde{c}_j}\right) \P(\bm{X}=\tilde{\bm{x}}_j,C=\tilde{c}_j,L'=\tilde{\ell}_j+k)\\
&\hspace{0.25cm}\leq \sum_{\beta\in[T]}\sum_{j\in J}\sum_{k\geq0}\tilde{x}_j^m\rho\left(\frac{\tilde{\ell}_j+k+\tilde{\bm{x}}_j\cdot h(\bm{\chi})}{\tilde{c}_j}\right) (1+\epsilon) p_j^\beta\P\left(\sum_{s\in[R]}s\mathrm{Poi}\left(\sum_{\gamma\in[T]}\tilde{w}_j^{-,s,\gamma}(z^*)^{s,\beta,\gamma}\right) = k\right)\\
&\hspace{0.25cm}= (1+\epsilon)(f^m(\bm{z}^*,\bm{\chi})+\chi^m).
\end{align*}
In particular, we can choose $\epsilon>0$ small enough such that
\[ f_\epsilon^m(\tilde{\bm{\chi}}(\delta/2)) \leq f^m(\bm{z}^*,\tilde{\bm{\chi}}(\delta/2))+\epsilon\big(f^m(\bm{z}^*,\tilde{\bm{\chi}}(\delta/2))+\tilde{\chi}^m(\delta/2)\big) \leq -\delta/2 + \epsilon\E[X^m] < 0, \]
where in the last inequality it was used that $f^m(\bm{z}^*, \tilde{\bm{\chi}}(\delta/2))\leq f^m(\tilde{\bm{z}}(\delta/2), \tilde{\bm{\chi}}(\delta/2))=-\delta/2$. We can hence conclude that $\bm{\chi}_\epsilon^* \leq \tilde{\bm{\chi}}(\delta/2)$ componentwise. By Theorem \ref{4:thm:fire:sale:final:fraction} we thus derive that the number of finally sold shares of asset $m$ in the fire sales system $(\bm{X},S,C,L')$ is upper bounded by $n((\chi_\epsilon^*)^m+o(1))\leq n(\tilde{\chi}^m(\delta/2)+o(1)) \leq n\tilde{\chi}^m(\delta)$, where the last inequality holds for $n$ large enough since $\tilde{\chi}^m(\delta)>0$. 

The idea for the rest of this proof is now to apply this upper bound on the number of finally sold shares inductively in each step of the contagion process. Again we consider the contagion process with steps (i) and (ii). Then in iteration $1\leq k\leq n-1$ we derive the smallest set $\mathcal{D}_k\subseteq[n]$ such that
\[ \mathcal{D}_k = \left\{ i\in[n]\,:\,\sum_{j\in \mathcal{D}_k}e_{j,i}\geq c_i-\ell_i-\bm{x}_i\cdot h(\bm{\chi}_{k-1}) \right\} \]
and the smallest vector $\bm{\chi}_k$ such that
\[ \bm{\chi}_k = n^{-1}\sum_{i\in[n]}\bm{x}_i \rho\left(\frac{\sum_{j\in\mathcal{D}_k}e_{j,i}+\ell_i+\bm{x}_i\cdot h(\bm{\chi}_k)}{c_i}\right). \]
In particular, since $\bm{\chi}_0=\bm{0}\leq\tilde{\bm{\chi}}(\delta)$, we derive that $\mathcal{D}_1\subseteq\mathcal{D}^\delta$, where $\mathcal{D}^\delta\subseteq[n]$ is the smallest set such that
\begin{equation}\label{5:eqn:D:delta}
\mathcal{D}^\delta = \left\{ i\in[n]\,:\,\sum_{j\in \mathcal{D}^\delta}e_{j,i}\geq c_i-\ell_i-\bm{x}_i\cdot h(\tilde{\bm{\chi}}(\delta)) \right\}
\end{equation}
and hence $\bm{\chi}_1\leq\bm{\chi}^\delta$, where $\bm{\chi}^\delta$ denotes the smallest vector such that
\begin{equation}\label{5:eqn:chi:delta}
\bm{\chi}^\delta = n^{-1}\sum_{i\in[n]}\bm{x}_i \rho\left(\frac{\sum_{j\in\mathcal{D}^\delta}e_{j,i}+\ell_i+\bm{x}_i\cdot h(\bm{\chi}^\delta)}{c_i}\right).
\end{equation}
However, \eqref{5:eqn:D:delta} is exactly the cascade of default contagion with initial capitals given by \mbox{$c_i-\ell_i-\bm{x}_i\cdot h(\tilde{\bm{\chi}}(\delta))$}, $i\in[n]$, which we considered before and \eqref{5:eqn:chi:delta} the subsequent cascade of fire sales for which we showed that the vector of finally sold shares is upper bounded by $n\tilde{\bm{\chi}}(\delta)$ w.\,h.\,p. We can then consider the second iteration and derive that w.\,h.\,p.~$\mathcal{D}_2\subseteq\mathcal{D}^\delta$ and $\bm{\chi}_2\leq\bm{\chi}^\delta$. Inductively this shows that w.\,h.\,p.~$\mathcal{D}_k\subseteq\mathcal{D}^\delta$ and $\bm{\chi}_k\leq\bm{\chi}^\delta$ for each fixed $k\in\N$ \mbox{(independent of $n$).}

Now note that because of the finitariness of the system, the contagion process stops after a bounded (independent of $n$) number of iterations. We have thus shown that also for the final vector of sold shares $\bm{\chi}_n$ it holds $\bm{\chi}_n\leq\tilde{\bm{\chi}}(\delta)$ w.\,h.\,p. Letting $\delta\to0$ this proves the upper \mbox{bound on $\bm{\chi}_n$.}

For the final systemic damage note that w.\,h.\,p.~$\mathcal{D}_n\subseteq\mathcal{D}^\delta$ and hence $n^{-1}\mathcal{S}_n \leq n^{-1}\mathcal{S}^\delta$. But
\begin{align*}
n^{-1}\mathcal{S}^\delta &\leq \E\left[S\P\left(\sum_{s\in[R]}s\mathrm{Poi}\left(\sum_{\gamma\in[T]}W^{-,s,\gamma}(z^*)^{s,\beta,\gamma}\right)\geq C-L-\bm{X}\cdot h(\tilde{\bm{\chi}}(\delta))\right)\right] + o_p(1)\\
&= g(\bm{z}^*,\tilde{\bm{\chi}}(\delta)) + o_p(1)
\end{align*}
by \eqref{5:eqn:upper:bound:D:delta}. Using upper semi-continuity of $g$ and letting $\delta\to0$ this finishes the proof.
\end{proof}

\subsection{Proof of Theorem \ref{5:thm:final:fraction:combined:general}}\label{5:ssec:proofs:main}
In this section, we show how the validity of Theorem \ref{5:thm:final:fraction:combined:finitary} can be extended to the case of general (non-finitary) regular vertex sequences. The idea is the following: We will approximate the given regular vertex sequence from below and from above by finitary vertex sequences and couple the contagion processes in those system such that the final default fraction and the number of sold shares is under- or overestimated by the finitary systems.

We will describe the finitary systems by their distribution functions $\{F_k^A\}_{k\in\N}$ and $\{F_k^B\}_{k\in\N}$ respectively in the following and we need to ensure that the functions $g$ and $f^m$, $m\in[M]$, are approximated close enough. To this end, consider the integrands
\begin{align*}
&h_g(\bm{z},\bm{\chi};\bm{w}^-,\bm{w}^+,\bm{x},s,c,\ell,\tau)\\
&\hspace*{0.1cm}:= s \sum_{\beta\in[T]}\psi\left(\sum_{\gamma\in[T]}w^{-,1,\gamma}z^{1,\beta,\gamma},\ldots,\sum_{\gamma\in[T]}w^{-,R,\gamma}z^{R,\beta,\gamma};c-l-\bm{x}\cdot h(\bm{\chi})\right)\1\{\tau=\beta\},\\
&h_f^{r,\alpha,\beta}(\bm{z},\bm{\chi};\bm{w}^-,\bm{w}^+,\bm{x},s,c,\ell,\tau)\\
&\hspace*{0.1cm}:= w^{+,r,\alpha}\psi\left(\sum_{\gamma\in[T]}w^{-,1,\gamma}z^{1,\beta,\gamma},\ldots,\sum_{\gamma\in[T]}w^{-,R,\gamma}z^{R,\beta,\gamma};c-l-\bm{x}\cdot h(\bm{\chi})\right)\1\{\tau=\beta\},\hspace*{0.1cm}(r,\alpha,\beta)\in V,\\
&h_f^m(\bm{z},\bm{\chi};\bm{w}^-,\bm{w}^+,\bm{x},s,c,\ell,\tau)\\
&\hspace*{0.1cm}:= x^m \sum_{\beta\in[T]}\phi\left(\sum_{\gamma\in[T]}w^{-,1,\gamma}z^{1,\beta,\gamma},\ldots,\sum_{\gamma\in[T]}w^{-,R,\gamma}z^{R,\beta,\gamma};l+\bm{x}\cdot h(\bm{\chi}),c\right)\1\{\tau=\beta\},\hspace*{0.1cm} m\in[M],
\end{align*}
for $(\bm{z},\bm{\chi},\bm{w}^-,\bm{w}^+,\bm{x},s,c,\ell,\tau)\in\R_{+,0}^V\times\R_{+,0}^M\times D_\infty$, with $\bm{w}^{\pm}=(w^{\pm,r,\alpha})_{r\in[R],\alpha\in[T]}$ as well as \mbox{$D_\infty:=\big(\R_{+,0}^{[R]\times[T]}\big)^2\times\R_{+,0}^{M+3}\times[T]$} and where $\psi$ and $\phi$ are as defined in Section \ref{5:asymp:res:sys}. Then
\begin{align*}
g(\bm{z},\bm{\chi}) &= \int_{D_\infty} h_g(\bm{z},\bm{\chi};\bm{w}^-,\bm{w}^+,\bm{x},s,c,\ell,\tau) \dd F(\bm{w}^-,\bm{w}^+,\bm{x},s,c,\ell,\tau),\\
f^{r,\alpha,\beta}(\bm{z},\bm{\chi}) &= \int_{D_\infty} h_f^{r,\alpha,\beta}(\bm{z},\bm{\chi};\bm{w}^-,\bm{w}^+,\bm{x},s,c,\ell,\tau) \dd F(\bm{w}^-,\bm{w}^+,\bm{x},s,c,\ell,\tau) - z^{r,\alpha,\beta},\hspace*{0.1cm}(r,\alpha,\beta)\in V,\\
f^m(\bm{z},\bm{\chi}) &= \int_{D_\infty} h_f^m(\bm{z},\bm{\chi};\bm{w}^-,\bm{w}^+,\bm{x},s,c,\ell,\tau) \dd F(\bm{w}^-,\bm{w}^+,\bm{x},s,c,\ell,\tau) - \chi^m,\hspace*{0.1cm} m\in[M],
\end{align*}
where $F$ denotes the distribution function of $(\bm{W}^-,\bm{W}^+,\bm{X},S,C,L,A)$ (note that the integrands vanish for $c=\infty$ and it is thus sufficient to integrate over $D_\infty$). We denote in the following
\[ H := \{h_g\}\cup\bigcup_{(r,\alpha,\beta)\in V}\{h_f^{r,\alpha,\beta}\}\cup\bigcup_{m\in[M]}\{h_f^m\} \]
and $Z:=[\bm{0},\bm{\zeta}]\times[\bm{0},\bm{\eta}]\subset\R_{+,0}^V\times\R_{+,0}^M$, for $\zeta^{r,\alpha,\beta}=\E[W^{+,r,\alpha}\1\{A=\beta\}]$, $(r,\alpha,\beta)\in V$, and $\eta^m=\E[X^m]$, $m\in[M]$. If we further let
\begin{align*}
&\circHSub{g}(\bm{z},\bm{\chi};\bm{w}^-,\bm{w}^+,\bm{x},s,c,\ell,\tau)\\
&\hspace*{0.1cm}:= s \sum_{\beta\in[T]}\circPsi\left(\sum_{\gamma\in[T]}w^{-,1,\gamma}z^{1,\beta,\gamma},\ldots,\sum_{\gamma\in[T]}w^{-,R,\gamma}z^{R,\beta,\gamma};c-l-\bm{x}\cdot h(\bm{\chi})\right)\1\{\tau=\beta\},\\
&\circHSuperSub{r,\alpha,\beta}{f}(\bm{z},\bm{\chi};\bm{w}^-,\bm{w}^+,\bm{x},s,c,\ell,\tau)\\
&\hspace*{0.1cm}:= w^{+,r,\alpha}\circPsi\left(\sum_{\gamma\in[T]}w^{-,1,\gamma}z^{1,\beta,\gamma},\ldots,\sum_{\gamma\in[T]}w^{-,R,\gamma}z^{R,\beta,\gamma};c-l-\bm{x}\cdot h(\bm{\chi})\right)\1\{\tau=\beta\},\hspace*{0.1cm}(r,\alpha,\beta)\in V,\\
&\circHSuperSub{m}{f}(\bm{z},\bm{\chi};\bm{w}^-,\bm{w}^+,\bm{x},s,c,\ell,\tau)\\
&\hspace*{0.1cm}:= x^m \sum_{\beta\in[T]}\circPhi\left(\sum_{\gamma\in[T]}w^{-,1,\gamma}z^{1,\beta,\gamma},\ldots,\sum_{\gamma\in[T]}w^{-,R,\gamma}z^{R,\beta,\gamma};l+\bm{x}\cdot h(\bm{\chi}),c\right)\1\{\tau=\beta\},\hspace*{0.1cm} m\in[M],
\end{align*}
where $\circPsi$ and $\circPhi$ are defined as in Section \ref{5:asymp:res:sys}, then it holds
\begin{align*}
\circG(\bm{z},\bm{\chi}) &= \int_{D_\infty} \circHSub{g}(\bm{z},\bm{\chi};\bm{w}^-,\bm{w}^+,\bm{x},s,c,\ell,\tau) \dd F(\bm{w}^-,\bm{w}^+,\bm{x},s,c,\ell,\tau),\\
\circFSuper{r,\alpha,\beta}(\bm{z},\bm{\chi}) &= \int_{D_\infty} \circHSuperSub{r,\alpha,\beta}{f}(\bm{z},\bm{\chi};\bm{w}^-,\bm{w}^+,\bm{x},s,c,\ell,\tau) \dd F(\bm{w}^-,\bm{w}^+,\bm{x},s,c,\ell,\tau) - z^{r,\alpha,\beta},\hspace*{0.1cm}(r,\alpha,\beta)\in V,\\
\circFSuper{m}(\bm{z},\bm{\chi}) &= \int_{D_\infty} \circHSuperSub{m}{f}(\bm{z},\bm{\chi};\bm{w}^-,\bm{w}^+,\bm{x},s,c,\ell,\tau) \dd F(\bm{w}^-,\bm{w}^+,\bm{x},s,c,\ell,\tau) - \chi^m,\hspace*{0.1cm} m\in[M].
\end{align*}
Also denote $\accentset{\circ}{H}:=\{\circHSub{g}\}\cup\bigcup_{(r,\alpha,\beta)\in V}\{\circHSuperSub{r,\alpha,\beta}{f}\}\cup\bigcup_{m\in[M]_0}\{\circHSuperSub{m}{f}\}$. 

For $j\in\N$, consider now discretizations
\begin{gather*}
\tilde{F}_j^A(\bm{w}^-,\bm{w}^+,\bm{x},s,c,\ell,\tau) := F\left(\frac{\lceil j\bm{w}^-\rceil}{j}, \frac{\lceil j\bm{w}^+\rceil}{j},\frac{\lceil j\bm{x}\rceil}{j},\frac{\lfloor js\rfloor}{j},\frac{\lfloor jc\rfloor}{j},\frac{\lceil j\ell\rceil}{j},\tau\right),\\
\tilde{F}_j^B(\bm{w}^-,\bm{w}^+,\bm{x},s,c,\ell,\tau) := F\left(\frac{\lfloor j\bm{w}^-\rfloor}{j}, \frac{\lfloor j\bm{w}^+\rfloor}{j},\frac{\lfloor j\bm{x}\rfloor}{j},\frac{\lceil js\rceil}{j},\frac{\lceil jc\rceil}{j},\frac{\lfloor j\ell\rfloor}{j},\tau\right),
\end{gather*}
where $\lceil\cdot\rceil$ and $\lfloor\cdot\rfloor$ shall be applied componentwise on the vectors $j\bm{w}^-$, $j\bm{w}^+$ and $j\bm{x}$. 
In particular, the distributions $\big\{\tilde{F}_j^{A,B}\big\}_{j\in\N}$ converge to $F$. Choose now $k\in\N$, $h\in H$ and let for $0\leq s\leq 2k^2$ the sets
\[ I_s(z,\bm{\chi}) := \left\{(\bm{w}^-,\bm{w}^+,\bm{x},s,c,\ell,\tau)\,:\,h(\bm{z},\bm{\chi};\bm{w}^-,\bm{w}^+,\bm{x},s,c,\ell,\tau) \geq \frac{s}{2k}\right\} \]
which are closed by upper semi-continuity of $h$. Further, let
\[ \hat{h}(\bm{z},\bm{\chi};\bm{w}^-,\bm{w}^+,\bm{x},s,c,\ell,\tau) := \frac{1}{2k}\sum_{s=0}^{2k^2}\1\left\{(\bm{w}^-,\bm{w}^+,\bm{x},s,c,\ell,\tau)\in I_s(z,\bm{\chi})\right\} \]
such that in particular $h-(2k)^{-1}\leq \hat{h} \leq h \leq k$ on $\R_{+,0}^V\times\R_{+,0}^M\times D_k$ with
\[ D_k := \left\{(\bm{w}^-,\bm{w}^+,\bm{x},s,c,\ell,\tau)\,:\,\bm{w}^-\leq k\bm{\1}, \bm{w}^+\leq k\bm{\1}, \bm{x}\leq k\bm{\1}, s\leq k, c\leq k, \ell\leq k\right\}. \] By the Portmanteau theorem, we then know that for $j\geq j_k$ large enough it holds
\[ \int_{D_k}\hat{h}\dd\tilde{F}_j^{A,B} - \int_{D_k}\hat{h}\dd F \leq \frac{1}{2k} \]
and hence
\begin{equation}\label{5:eqn:portmanteau:upper}
\int_{D_k} h(\bm{z},\bm{\chi};\bm{w}^-,\bm{w}^+,\bm{x},s,c,\ell,\tau)\dd \tilde{F}_j^{A,B}
 - \int_{D_k} h(\bm{z},\bm{\chi};\bm{w}^-,\bm{w}^+,\bm{x},s,c,\ell,\tau)\dd F
  \leq k^{-1}
\end{equation}
Completely analogue, but choosing
\[ \accentset{\circ}{I}_s(z,\bm{\chi}) := \left\{(\bm{w}^-,\bm{w}^+,\bm{x},s,c,\ell,\tau)\,:\,\circH(\bm{z},\bm{\chi};\bm{w}^-,\bm{w}^+,\bm{x},s,c,\ell,\tau) > \frac{s}{2k}\right\}, \]
we derive for $\circH\in\accentset{\circ}{H}$ and $j\geq j_k$ (possibly increase $j_k$) that
\begin{equation}\label{5:eqn:portmanteau:lower}
\int_{\circDSub{k}} \circH(\bm{z},\bm{\chi};\bm{w}^-,\bm{w}^+,\bm{x},s,c,\ell,\tau)\dd \tilde{F}_j^{A,B}
 - \int_{\circDSub{k}} \circH(\bm{z},\bm{\chi};\bm{w}^-,\bm{w}^+,\bm{x},s,c,\ell,\tau)\dd F
  \geq -k^{-1},
\end{equation}
where $\circDSub{k}:=\{(\bm{w}^-,\bm{w}^+,\bm{x},s,c,\ell,\tau)\,:\,\bm{w}^-<k\bm{\1}, \bm{w}^+<k\bm{\1}, \bm{x}<k\bm{\1}, s<k, c<k, \ell<k\}$. We denote $\overline{F}_k^A:=\tilde{F}_{j_k}^A$ and $\overline{F}_k^B:=\tilde{F}_{j_k}^B$ in the following.

Note that $\overline{F}_k^{A,B}$ already describe discrete distributions approximating $F$ from below and from above. To qualify as a distribution function of a finitary vertex sequence, however, only finitely many atoms are allowed. For the lower bound, we thus choose
\[ F_k^A(\bm{w}^-,\bm{w}^+,\bm{x},s,c,\ell,\tau) := \begin{cases}\overline{F}_k^A(\bm{w}^-\wedge k,\bm{w}^+\wedge k,\bm{x}\wedge k,s\wedge k,c\wedge k,\ell\wedge k,\tau),&\text{if }c<\infty,\\1,&\text{else},\end{cases} \]
thus setting capital to $\infty$ for banks with in-weight, out-weight, asset holdings, capital or exogenous losses larger than $k$. We call such banks \emph{large} in the following. Banks with infinite capital keep their capital. As banks with infinite capital cannot ever default or sell any asset shares anyway we set their weights, asset holdings and losses all to zero. In particular, above choice reduces contagion in the system even further and thus for all $k\in\N$ the final systemic damage $n^{-1}(\mathcal{S}_k^A)_n$ is stochastically dominated by $n^{-1}\mathcal{S}_n$. The same holds for the number of finally sold shares of the assets.

We now want to construct the upper bound distribution function $F_k^B$. That is, we need to accumulate the contagious potential of all large banks to finitely many point masses. Thus denote the fraction of large $\beta$-banks in the system by
\[ \gamma_k^\beta := \int_{D_k^c} \1\{\tau=\beta\} \dd F(\bm{w}^-,\bm{w}^+,\bm{x},s,c,\ell,\tau), \]
where $D_k^c:=D_\infty\backslash D_k$, and
\begin{align*}
(\overline{w}_k^\beta)^{r,\alpha} &:= \begin{cases}2(\gamma_k^\beta)^{-1}\int_{D_k^c}w^{+,r,\alpha}\1\{\tau=\beta\}\dd F(\bm{w}^-,\bm{w}^+,\bm{x},s,c,\ell,\tau)\geq 2k,\hspace*{-0.2cm}&\text{if }\gamma_k^\beta>0,\\2k,&\text{if }\gamma_k^\beta=0,\end{cases}~(r,\alpha,\beta)\in V,\\
(\overline{x}_k^\beta)^m &:= \begin{cases}2(\gamma_k^\beta)^{-1}\int_{D_k^c}x^m\1\{\tau=\beta\}\dd F(\bm{w}^-,\bm{w}^+,\bm{x},s,c,\ell,\tau)\geq 2k,\hspace*{0.3cm}&\text{if }\gamma_k^\beta>0,\\2k,&\text{if }\gamma_k^\beta=0,\end{cases}~~ m\in[M],\\
\overline{s}_k^\beta &:= \begin{cases}2(\gamma_k^\beta)^{-1}\int_{D_k^c}s\1\{\tau=\beta\}\dd F(\bm{w}^-,\bm{w}^+,\bm{x},s,c,\ell,\tau)\geq 2k,\hspace*{0.6cm}&\text{if }\gamma_k^\beta>0,\\2k,&\text{if }\gamma_k^\beta=0.\end{cases}
\end{align*}
Similar as for the lower bound before, we now let $F_k^B$ be given by $\overline{F}_k^B$ on $D_k$. Moreover, let $F_k^B$ assign the remaining masses $\gamma_k^\beta$ to the points $(\bm{0},\overline{\bm{w}}_k^\beta,\overline{\bm{x}}_k^\beta,\overline{s}_k^\beta,0,0,\beta)$. As we have left out institutions with infinite capital thus far, finally let $F_k^B$ assign masses $\P(C=\infty,A=\beta)$ to the point $(\bm{0},\bm{0},\bm{0},0,\infty,0,\beta)$ for each $\beta\in[T]$.

The construction above ensures that small banks are more contagious than in the original system as their weights, asset holdings and losses are increased whereas their capitals are decreased. Moreover, all large banks are initially defaulted and their total number of shares held of each asset $m$ is given by
\[ n\sum_{\beta\in[T]}(\overline{x}_k^\beta)^m(\gamma_k^\beta+o(1)) = 2n\int_{D_k^c}x^m\dd F(\bm{w}^-,\bm{w}^+,\bm{x},s,c,\ell,\tau)(1+o(1)) \]
which is larger than in the original system,
\[ n\int_{D_k^c}x^m\dd F(\bm{w}^-,\bm{w}^+,\bm{x},s,c,\ell,\tau)(1+o(1)). \]
Finally, also the total $r$-out-weight of large $\beta$-banks with respect to each type $\alpha\in[T]$
\[ n(\overline{w}_k^\beta)^{r,\alpha}\left(\gamma_k^\beta+o(1)\right) = 2n \int_{D_k^c}w^{+,r,\alpha}\1\{\tau=\beta\}\dd F(\bm{w}^-,\bm{w}^+,\bm{x},s,c,\ell,\tau) (1+o(1)) \]
is increased by approximation $F_k^B$ compared to the original system with
\[ n \int_{D_k^c}w^{+,r,\alpha}\1\{\tau=\beta\}\dd F(\bm{w}^-,\bm{w}^+,\bm{x},s,c,\ell,\tau) (1+o(1)). \]
Similar as in \cite{Detering2015a} for each $r\in[R]$ the number of $r$-edges from large banks to a specific small bank in the approximating system thus stochastically dominates their analogue in the original system and in particular this property transfers to the total direct exposure from large banks by summing over all $r\in[R]$. Altogether we derive the following result.
\begin{lemma}\label{5:lem:stochastic:domination}
Consider a regular vertex sequence and let sequences $\{F_k^A\}_{k\in\N}$ and $\{F_k^B\}_{k\in\N}$ be constructed as above. Further let $\left(\mathcal{S}_k^A\right)_n$ and $\left(\mathcal{S}_k^B\right)_n$ be the total systemic importance of finally defaulted institutions in the finitary approximating systems. Then it holds that
\[ n^{-1}\left(\mathcal{S}_k^A\right)_n \preceq n^{-1}\mathcal{S}_n \preceq n^{-1}\left(\mathcal{S}_k^B\right)_n, \]
where $\preceq$ denotes stochastic domination. If further $(\chi_k^{A,B})_n^m$ denotes the number of finally sold shares of asset $m$ divided by $n$, then it holds
\[ \left(\chi_k^A\right)_n^m \preceq \chi_n^m \preceq \left(\chi_k^B\right)_n^m. \]
\end{lemma}

\vspace*{-0.2cm}
\noindent Denote

\vspace*{-0.7cm}
\begin{gather*}
g_k^{A,B}(\bm{z},\bm{\chi}) = \int_{D_\infty} h_g(\bm{z},\bm{\chi};\bm{w}^-,\bm{w}^+,\bm{x},s,c,\ell,\tau)\dd F_k^{A,B}(\bm{w}^-,\bm{w}^+,\bm{x},s,c,\ell,\tau),\\
\left(f_k^{A,B}\right)^{r,\alpha,\beta}(\bm{z},\bm{\chi}) = \int_{D_\infty} h_f^{r,\alpha,\beta}(\bm{z},\bm{\chi};\bm{w}^-,\bm{w}^+,\bm{x},s,c,\ell,\tau)\dd F_k^{A,B}(\bm{w}^-,\bm{w}^+,\bm{x},s,c,\ell,\tau) - z^{r,\alpha,\beta},\\
\left(f_k^{A,B}\right)^m(\bm{z},\bm{\chi}) = \int_{D_\infty} h_f^m(\bm{z},\bm{\chi};\bm{w}^-,\bm{w}^+,\bm{x},s,c,\ell,\tau)\dd F_k^{A,B}(\bm{w}^-,\bm{w}^+,\bm{x},s,c,\ell,\tau) - \chi^m,
\end{gather*}
analogue to $g$, $f^{r,\alpha,\beta}$, $(r,\alpha,\beta)\in V$, and $f^m$, $m\in[M]$. Moreover, let $\circGSuperSub{{A,B}}{k}$, $\big(\circFSuperSub{A,B}{k}\big)^{r,\alpha,\beta}$, $\big(\circFSuperSub{A,B}{k}\big)^m$, $(\hat{\bm{z}}_k^{A,B},\hat{\bm{\chi}}_k^{A,B})$ and $((\bm{z}^*)_k^{A,B},(\bm{\chi}^*)_k^{A,B})$ the analogues of $\circG$, $\circFSuper{r,\alpha,\beta}$, $(r,\alpha,\beta)\in V$, $\circFSuper{m}$, $m\in[M]$, $(\hat{\bm{z}},\hat{\bm{\chi}})$ and $(\bm{z}^*,\bm{\chi}^*)$. Then by Theorem \ref{5:thm:final:fraction:combined:finitary} we derive lower and upper bounds for the approximating systems in terms of those quantities. The following lemma compares them to the original quantities.
\begin{lemma}\label{5:lem:convergence:g}
It holds
\[ \liminf_{k\to\infty} \circGSuperSub{A}{k}\left(\hat{\bm{z}}_k^A,\hat{\bm{\chi}}_k^A\right) \geq \vphantom{g}\circG(\hat{\bm{z}},\hat{\bm{\chi}}) \]
and
\[ \limsup_{k\to\infty} g_k^B\left(\left(\bm{z}^*\right)_k^B,(\bm{\chi}^*)_k^B\right) \leq g(\bm{z}^*,\bm{\chi}^*), \]
as well as $\limsup_{k\to\infty}(\hat{\bm{\chi}}_k^A)^m\geq\hat{\chi}^m$ and $\liminf_{k\to\infty}((\bm{\chi}_k^B)^*)^m\leq(\chi^*)^m$ for all $m\in[M]$.
\end{lemma}

\begin{proof}
For $\circH\in\accentset{\circ}{H}$, using \eqref{5:eqn:portmanteau:lower} we derive
\begin{align*}
&\int_{\circDSub{k}} \circH(\bm{z},\bm{\chi};\bm{w}^-,\bm{w}^+,\bm{x},s,c,\ell,\tau)\dd F_k^A
 - \int_{\circDSub{k}} \circH(\bm{z},\bm{\chi};\bm{w}^-,\bm{w}^+,\bm{x},s,c,\ell,\tau)\dd F
 \\
&\hspace{3cm}= \int_{\circDSub{k}} \circH(\bm{z},\bm{\chi};\bm{w}^-,\bm{w}^+,\bm{x},s,c,\ell,\tau)\dd \tilde{F}_{j_k}^A
 - \int_{\circDSub{k}} \circH(\bm{z},\bm{\chi};\bm{w}^-,\bm{w}^+,\bm{x},s,c,\ell,\tau)\dd F
 \\
&\hspace{3cm}\geq -k^{-1} \to 0,\quad\text{as }k\to\infty.
\end{align*}
Moreover, $\int_{\circDSuperSub{c}{k}}s\,\dd F\to0$, $\int_{\circDSuperSub{c}{k}}w^{+,r,\alpha}\,\dd F\to0$, $(r,\alpha)\in[R]\times[T]$, and $\int_{\circDSuperSub{c}{k}}x^m\,\dd F\to0$, $m\in[M]$, as $k\to\infty$. In particular, it must then hold that $\int_{\circDSuperSub{c}{k}}\circH(\bm{z},\bm{\chi};\bm{w}^-,\bm{w}^+,\bm{x},s,c,\ell,\tau)\dd F\to0$.  Together with 
\[ \int_{\circD_k^c}\accentset{\circ}{h}(\bm{z},\bm{\chi};\bm{w}^-,\bm{w}^+,\bm{x},s,c,\ell,\tau)\dd F_k^A=\int_{D_k\cap\circD_k^c}\accentset{\circ}{h}(\bm{z},\bm{\chi};\bm{w}^-,\bm{w}^+,\bm{x},s,c,\ell,\tau)\dd F_k^A=o(1), \]
we can then conclude that
\begin{equation}\label{5:eqn:strong:G:convergence:A}
\int_{D_\infty}\circH(\bm{z},\bm{\chi};\bm{w}^-,\bm{w}^+,\bm{x},s,c,\ell,\tau)\dd F_k^A - \int_{D_\infty} \circH(\bm{z},\bm{\chi};\bm{w}^-,\bm{w}^+,\bm{x},s,c,\ell,\tau)\dd F \geq o(1).
\end{equation}
For $\{F_k^B\}_{k\in\N}$, we further obtain
\begin{gather*}
\int_{D_k^c}s\,\dd F_k^B(\bm{w}^-,\bm{w}^+,\bm{x},s,c,\ell,\tau) = \sum_{\beta\in[T]}\overline{s}_k^\beta\gamma_k^\beta,\\
\int_{D_k^c}w^{+,r,\alpha}\1\{\tau=\beta\}\dd F_k^B(\bm{w}^-,\bm{w}^+,\bm{x},s,c,\ell,\tau) = (\overline{w}_k^\beta)^{r,\alpha} \gamma_k^\beta,\\
\int_{D_k^c}x^m\dd F_k^B(\bm{w}^-,\bm{w}^+,\bm{x},s,c,\ell,\tau) = \sum_{\beta\in[T]}(\overline{x}_k^\beta)^m \gamma_k^\beta
\end{gather*}
and as $k\to\infty$, by definition of $\gamma_k^\beta$, $\overline{s}_k^\beta$, $(\overline{w}_k^\beta)^{r,\alpha}$ and $(\overline{x}_k^\beta)^m$ all those terms vanish. In particular,
\[ \int_{D_k^c}h(\bm{z},\bm{\chi};\bm{w}^-,\bm{w}^+,\bm{x},s,c,\ell,\tau)\dd F_k^B(\bm{w}^-,\bm{w}^+,\bm{x},s,c,\ell,\tau) \to 0,\quad\text{as }k\to\infty \]
and by \eqref{5:eqn:portmanteau:upper}
\begin{equation}\label{5:eqn:strong:G:convergence:B}
\int_{D_\infty}h(\bm{z},\bm{\chi};\bm{w}^-,\bm{w}^+,\bm{x},s,c,\ell,\tau)\dd F_k^B - \int_{D_\infty} h(\bm{z},\bm{\chi};\bm{w}^-,\bm{w}^+,\bm{x},s,c,\ell,\tau)\dd F \leq o(1).
\end{equation}
By \eqref{5:eqn:strong:G:convergence:A} we can apply Lemma \ref{4:lem:convergence:chi:hat} (extend it by the $\bm{z}$-dimensions) and thus derive that \mbox{$\liminf_{k\to\infty}\hat{\bm{z}}_k^A\geq\hat{\bm{z}}$} and $\liminf_{k\to\infty}\hat{\bm{\chi}}_k^A\geq\hat{\bm{\chi}}$, where $(\hat{\bm{z}}_k^A,\hat{\bm{\chi}}_k^A)$ denotes the smallest joint root of the functions
\[ (\circFSuperSub{A}{k})^{r,\alpha,\beta}(\bm{z},\bm{\chi}) = \int_{D_\infty}\circHSuperSub{m}{f}(\bm{z},\bm{\chi};\bm{w}^-,\bm{w}^+,\bm{x},s,c,\ell,\tau)\dd F_k^A(\bm{w}^-,\bm{w}^+,\bm{x},s,c,\ell,\tau) - z^{r,\alpha,\beta}, \]
\[ (\circFSuperSub{A}{k})^m(\bm{z},\bm{\chi}) = \int_{D_\infty}\circHSuperSub{m}{f}(\bm{z},\bm{\chi};\bm{w}^-,\bm{w}^+,\bm{x},s,c,\ell,\tau)\dd F_k^A(\bm{w}^-,\bm{w}^+,\bm{x},s,c,\ell,\tau) - \chi^m. \]
Now choose some $\delta>0$ and $k$ large enough such that $(\hat{\bm{z}}_k^A,\hat{\bm{\chi}}_k^A)\geq(1-\delta)(\hat{\bm{z}},\hat{\bm{\chi}})$. Then by \eqref{5:eqn:strong:G:convergence:A},
\[ \liminf_{k\to\infty}\circGSuperSub{A}{k}(\hat{\bm{z}}_k^A,\hat{\bm{\chi}}_k^A) \geq \liminf_{k\to\infty}\circGSuperSub{A}{k}((1-\delta)(\hat{\bm{z}},\hat{\bm{\chi}})) \geq \circG((1-\delta)(\hat{\bm{z}},\hat{\bm{\chi}})) \]
and using lower semi-continuity of $\circG$, as $\delta\to0$,
\[ \liminf_{k\to\infty} \circGSuperSub{A}{k}\left(\hat{\bm{z}}_k^A,\hat{\bm{\chi}}_k^A\right) \geq \vphantom{g}\circG(\hat{\bm{z}},\hat{\bm{\chi}}). \]
For the second statement apply a small additional shock to the system in the sense that each solvent institution defaults with probability $\epsilon>0$. Then the analogues of $f^{r,\alpha,\beta}$ and $f^m$ in the shocked system are given by
\begin{align*}
f_\epsilon^{r,\alpha,\beta}(\bm{z},\bm{\chi}) &= (1-\epsilon)f^{r,\alpha,\beta}(\bm{z},\bm{\chi}) + \epsilon(\E[W^+]-z^{r,\alpha,\beta}),\\
f_\epsilon^m(\bm{z},\bm{\chi}) &= (1-\epsilon)f^m(\bm{z},\bm{\chi}) + \epsilon(\E[X^m]-\chi^m).
\end{align*}
Denote the analogues of $(\bm{z}^*,\bm{\chi}^*)$ for these functions by $(\bm{z}^*(\epsilon),\bm{\chi}^*(\epsilon))$. Then using \eqref{5:eqn:strong:G:convergence:B} for $k$ large enough it holds $(f_k^B)^{r,\alpha,\beta}(\bm{z}^*(\epsilon),\bm{\chi}^*(\epsilon))\leq f^{r,\alpha,\beta}(\bm{z}^*(\epsilon),\bm{\chi}^*(\epsilon))/2<0$, \mbox{$(r,\alpha,\beta)\in V$}, and \mbox{$(f_k^B)^m(\bm{z}^*(\epsilon),\bm{\chi}^*(\epsilon))\leq f^m(\bm{z}^*(\epsilon),\bm{\chi}^*(\epsilon))/2<0$}, $m\in[M]$. Note that actually it is possible that $f^{r,\alpha,\beta}(\bm{z}^*(\epsilon),\bm{\chi}^*(\epsilon))=0$ resp.~$f^m(\bm{z}^*(\epsilon),\bm{\chi}^*(\epsilon))=0$ if $\E[W^{+,r,\alpha}\1\{A=\beta\}]=0$ resp.~$\E[X^m]=0$. In this case, however, the corresponding coordinates $z^{r,\alpha,\beta}$ resp.~$\chi^m$ are trivial and can be left out. Thus $(\bm{z}^*)_k^B\leq \bm{z}^*(\epsilon)$  and $(\bm{\chi}^*)_k^B\leq\bm{\chi}^*(\epsilon)$ componentwise and in particular $\limsup_{k\to\infty}((\chi_k^B)^*)^m\leq(\chi^*)^m$ for all $m\in[M]$. Using \eqref{5:eqn:strong:G:convergence:B} again, we now obtain
\[ \limsup_{k\to\infty}g_k^B\left((\bm{z}^*)_k^B,(\bm{\chi}^*)_k^B\right) \leq \limsup_{k\to\infty} g_k^B(\bm{z}^*(\epsilon),\bm{\chi}^*(\epsilon)) \leq g(\bm{z}^*(\epsilon),\bm{\chi}^*(\epsilon)) \]
and as $\epsilon\to0$ using upper semi-continuity of $g$ we can conclude that
\[ \limsup_{k\to\infty} g_k^B\left((\bm{z}^*)_k^B,(\bm{\chi}^*)_k^B\right) \leq g(\bm{z}^*,\bm{\chi}^*). \qedhere\]
\end{proof}
We can then state the proof of the main theorem for general regular vertex sequences:
\begin{proof}[Proof of Theorem \ref{5:thm:final:fraction:combined:general}]
For arbitrary $\epsilon>0$ we can apply Lemma \ref{5:lem:stochastic:domination} to derive
\[ \P\left(n^{-1}\mathcal{S}_n-\circG(\hat{\bm{z}},\hat{\bm{\chi}})<-\epsilon\right)\leq \P\left(n^{-1}\left(\mathcal{S}_k^A\right)_n-\circG(\hat{\bm{z}},\hat{\bm{\chi}})<-\epsilon\right). \]
Moreover, it holds $\circGSuperSub{A}{k}(\hat{\bm{z}}_k^A,\hat{\bm{\chi}}_k^A)>\circG(\hat{\bm{z}},\hat{\bm{\chi}})-\epsilon/2$ for $k$ large enough by Lemma \ref{5:lem:convergence:g} and then
\[ \P\left(n^{-1}\mathcal{S}_n-\circG(\hat{\bm{z}},\hat{\bm{\chi}})<-\epsilon\right) \leq \P\left(n^{-1}\left(\mathcal{S}_k^A\right)_n-\circGSuperSub{A}{k}(\hat{\bm{z}}_k^A,\hat{\bm{\chi}}_k^A)<-\epsilon/2\right). \]
Theorem \ref{5:thm:final:fraction:combined:finitary} now yields
\[ \P\left(n^{-1}\mathcal{S}_n-\circG(\hat{\bm{z}},\hat{\bm{\chi}})<-\epsilon\right) \to 0,\quad\text{as }n\to\infty, \]
and thus $n^{-1}\mathcal{S}_n\geq \circG(\hat{\bm{z}},\hat{\bm{\chi}})+o_p(1)$ as $\epsilon>0$ was arbitrary. Similarly,
\[ \P(\chi_n^m - \hat{\chi}^m<-\epsilon) \leq \P\left((\chi_k^A)_n^m - \hat{\chi}^m < -\epsilon\right) \leq \P\left((\chi_k^A)_n^m - (\hat{\chi}_k^A)^m < -\epsilon/2\right) \to0 \]
as $n\to\infty$ and hence $\chi_n^m \geq \hat{\chi}^m+o_p(1)$ for all $m\in[M]$.

In the same way, by Lemma \ref{5:lem:stochastic:domination}
\[ \P\left(n^{-1}\mathcal{S}_n-g(\bm{z}^*,\bm{\chi}^*)>\epsilon\right) \leq \P\left(n^{-1}\left(\mathcal{S}_k^B\right)_n-g(\bm{z}^*,\bm{\chi}^*)>\epsilon\right) \]
and hence for $k$ large enough it holds $g_k^B\left((\bm{z}^*)_k^B,(\bm{\chi}^*)_k^B\right)<g(\bm{z}^*,\bm{\chi}^*)+\epsilon/2$ by Lemma \ref{5:lem:convergence:g}. Using Theorem \ref{5:thm:final:fraction:combined:finitary} we thus conclude that
\begin{align*}
\P\left(n^{-1}\mathcal{S}_n-g(\bm{z}^*,\bm{\chi}^*)>\epsilon\right) &\leq \P\left(n^{-1}\left(\mathcal{S}_k^B\right)_n-g_k^B((\bm{z}^*)_k^B,(\bm{\chi}^*)_k^B)>\epsilon/2\right) \to 0,\quad\text{as }n\to\infty,
\end{align*}
as well as
\[ \P(\chi_n^m - (\chi^*)^m > \epsilon) \leq \P\left((\chi_k^B)_n^m - (\chi^*)^m > \epsilon\right) \leq \P\left((\chi_k^A)_n^m - ((\chi^*)_k^A)^m > \epsilon/2\right) \to0. \qedhere \]
\end{proof}

\subsection{Proofs for Section \ref{5:sec:resilience}}\label{5:ssec:proofs:resilience}
We keep the notation $g$, $f^{r,\alpha,\beta}$, $f^m$, $\bm{z}^*$, $\bm{\chi}^*$, $\circG$, $\circFSuper{r,\alpha,\beta}$, $\circFSuper{m}$, $\hat{\bm{z}}$ and $\hat{\bm{\chi}}$ for the quantities from Section \ref{5:asymp:res:sys} for the unshocked system and add the index $\cdot_L$ to indicate the corresponding quantities and functions in the system shocked by $L$.
\begin{proof}[Proof of Theorem \ref{5:thm:resilience}]
For arbitrary $\alpha>0$, we derive
\begin{align*}
&f_L^{r,\alpha,\beta}(\bm{z},\bm{\chi})\\
&\quad= \E\left[W^{+,r,\alpha}\P\left(\sum_{s\in[R]}s\mathrm{Poi}\left(\sum_{\gamma\in[T]}W^{-,s,\gamma}z^{s,\beta,\gamma}\right)\geq C-L-X\cdot h(\bm{\chi})\right)\1\{A=\beta\}\right] - z^{r,\alpha,\beta}\\
&\quad\leq \E\left[W^{+,r,\alpha}\1\{L\geq\alpha C\}\right] - z^{r,\alpha,\beta}\\
&\quad\hspace{1cm} + \E\left[W^{+,r,\alpha}\P\left(\sum_{s\in[R]}s\mathrm{Poi}\left(\sum_{\gamma\in[T]}W^{-,s,\gamma}z^{s,\beta,\gamma}\right)\geq C(1-\alpha)-X\cdot h(\bm{\chi})\right)\1\{A=\beta\}\right].
\end{align*}
Using that $\E[L/C]<\delta$, we derive with Markov's inequality that $\P(L\geq\alpha C)<\delta/\alpha$ and as $\E[W^{+,r,\alpha}]<\infty$ it thus holds that $\E[W^{+,r,\alpha}\1\{L\geq\alpha C\}]\leq\gamma/3$ for any arbitrary $\gamma>0$ if we choose $\delta>0$ small enough. Also the second summand in above inequality can be bounded by $f^{r,\alpha,\beta}(\bm{z},\bm{\chi})+\gamma/3$ if $\alpha$ is chosen small enough using the dominated convergence theorem.

Let now $((\tilde{\bm{z}}(\gamma),\tilde{\bm{\chi}}(\gamma)))_{\gamma>0}\subset\R_{+,0}^V\times\R_{+,0}^M$ be such that $f^{r,\alpha,\beta}(\tilde{\bm{z}}(\gamma),\tilde{\bm{\chi}}(\gamma))=-\gamma$ for all $(r,\alpha,\beta)\in V$ resp.~$f^m(\tilde{\bm{z}}(\gamma),\tilde{\bm{\chi}}(\gamma))=-\gamma$ for all $m\in[M]$, which exists analogue to Remark \ref{4:rem:sequence:chi:*} (extending it by the $\bm{z}$-coordinates). By the above result then \mbox{$f_L^{r,\alpha,\beta}(\tilde{\bm{z}}(\gamma),\tilde{\bm{\chi}}(\gamma)) \leq -\gamma/3<0$} for $\delta$ small enough. Similarly, one derives that $f_L^m(\tilde{\bm{z}}(\gamma),\tilde{\bm{\chi}}(\gamma)) \leq -\gamma/3<0$ for $\delta$ small enough. We can thus conclude that $(\bm{z}_L^*,\bm{\chi}_L^*)<(\tilde{\bm{z}}(\gamma),\tilde{\bm{\chi}}(\gamma))$. However, by Remark \ref{4:rem:sequence:chi:*} we further know that $(\tilde{\bm{z}}(\gamma),\tilde{\bm{\chi}}(\gamma))\to(\bm{z}^*,\bm{\chi}^*)$ and hence by upper semi-continuity of $g$ and possibly further decreasing $\delta$ it holds
\begin{equation}\label{5:eqn:g:z*}
g(\bm{z}_L^*,\bm{\chi}_L^*)\leq g(\bm{z}^*,\bm{\chi}^*)+\epsilon/3=\epsilon/3.
\end{equation}
By similar means as for $f^{r,\alpha,\beta}$ and $f^m$ above, we also derive that $g_L(\bm{z},\bm{\chi})\leq\epsilon/3+g(\bm{z},\bm{\chi})$ for $\delta$ small enough. By Theorem \ref{5:thm:final:fraction:combined:general} we can thus conclude that w.\,h.\,p.
\[ n^{-1}\mathcal{S}_{n,L} \leq g_L(\bm{z}_L^*,\bm{\chi}_L^*) + \epsilon/3 \leq g(\bm{z}_L^*,\bm{\chi}_L^*) + 2\epsilon/3 \leq \epsilon. \]
Now let $\delta$ small enough such that also $(\chi_L^*)^m\leq (\chi^*)^m+\epsilon/2$. Applying Theorem \ref{5:thm:final:fraction:combined:general} we thus derive that w.\,h.\,p.
\[ \chi_{n,L}^m \leq (\chi_L^*)^m + \epsilon/2 \leq (\chi^*)^m+\epsilon. \qedhere\]
\end{proof}

\begin{proof}[Proof of Theorem \ref{5:thm:non-resilience}]
Define for $\epsilon>0$ and $I_V\subset V$ resp.~$I_M\subset [M]$ the set
\begin{align*}
T(\epsilon,I) &:= \bigcap_{(r,\alpha,\beta)\in I_V}\left\{(\bm{z},\bm{\chi})\in\R_{+,0}^V\times\R_{+,0}^M\,:\,\circFSuper{r,\alpha,\beta}(\bm{z},\bm{\chi})\leq-\epsilon\right\}\\
&\hspace{3cm}\cap \bigcap_{(s,\theta,\lambda)\in I_V^c}\left\{(\bm{z},\bm{\chi})\in\R_{+,0}^V\times\R_{+,0}^M\,:\,z^{s,\theta,\lambda}\geq\E[W^{+,s,\theta}\1\{A=\lambda\}]\right\}\\
&\hspace{3cm}\cap\bigcap_{m\in I_M}\left\{(\bm{z},\bm{\chi})\in\R_{+,0}^V\times\R_{+,0}^M\,:\,\circFSuper{m}(\bm{z},\bm{\chi})\leq-\epsilon\right\}\\
&\hspace{3cm}\cap \bigcap_{k\in I_M^c}\left\{(\bm{z},\bm{\chi})\in\R_{+,0}^V\times\R_{+,0}^M\,:\,\chi^k\geq\E[X^k]\right\}
\end{align*}
where we denote $I_V^c:=V\backslash I_V$ and $I_M^c=[M]\backslash I_M$. Moreover, denote by $(\hat{\bm{z}}(\epsilon,I),\hat{\bm{\chi}}(\epsilon,I))\in\R_{+,0}^V\times\R_{+,0}^M$ the smallest vector such that $\circFSuper{r,\alpha,\beta}(\hat{\bm{z}}(\epsilon,I),\hat{\bm{\chi}}(\epsilon,I))=-\epsilon$ for $(r,\alpha,\beta)\in I_V$ and $\hat{z}^{s,\theta,\lambda}(\epsilon,I)=\E[W^{+,r,\alpha}\1\{A=\beta\}]$ for $(s,\theta,\lambda)\in I_V^c$ resp.~$\circFSuper{m}(\hat{\bm{z}}(\epsilon,I),\hat{\bm{\chi}}(\epsilon,I))=-\epsilon$ for $m\in I_M$ and $\hat{\chi}^k(\epsilon,I)=\E[X^k]$ for $k\in I_M^c$. The existence of such a vector is ensured analogue to Lemma \ref{5:lem:existence:z:chi:hat}. In particular, it then holds $(\hat{\bm{z}}(\epsilon,I),\hat{\bm{\chi}}(\epsilon,I))\in T(\epsilon,I)$ and by the construction of $(\hat{\bm{z}}(\epsilon,I),\hat{\bm{\chi}}(\epsilon,I))$ analogue to Lemma \ref{5:lem:existence:z:chi:hat} we obtain that $(\hat{\bm{z}}(\epsilon,I),\hat{\bm{\chi}}(\epsilon,I))\leq(z,\bm{\chi})$ for any other $(\bm{z},\bm{\chi})\in T(\epsilon,I)$

In particular, this implies that $(\hat{\bm{z}}(\epsilon,I),\hat{\bm{\chi}}(\epsilon,I))$ is non-decreasing in $\epsilon$ and thus continuous for almost every $\epsilon>0$. As moreover, the expressions $f^{r,\alpha,\beta}(\hat{\bm{z}}(\epsilon,I),\hat{\bm{\chi}}(\epsilon,I))+\hat{z}^{r,\alpha,\beta}(\epsilon,I)$ and $f^m(\hat{\bm{z}}(\epsilon,I),\hat{\bm{\chi}}(\epsilon,I))+\hat{\chi}^m(\epsilon,I)$ are bounded and increasing in $\epsilon$, we derive that for almost every $\epsilon>0$ and $\delta>0$, we can choose $\gamma>0$ small enough such that
\[ f^{r,\alpha,\beta}(\hat{\bm{z}}(\epsilon,I),\hat{\bm{\chi}}(\epsilon,I))+\hat{z}^{r,\alpha,\beta}(\epsilon,I) \leq f^{r,\alpha,\beta}(\hat{\bm{z}}(\epsilon-\gamma,I),\hat{\bm{\chi}}(\epsilon-\gamma,I))+\hat{z}^{r,\alpha,\beta}(\epsilon-\gamma,I) + \delta \]
and
\[ f^m(\hat{\bm{z}}(\epsilon,I),\hat{\bm{\chi}}(\epsilon,I))+\hat{\chi}^m(\epsilon,I) \leq f^m(\hat{\bm{z}}(\epsilon-\gamma,I),\hat{\bm{\chi}}(\epsilon-\gamma,I))+\hat{\chi}^m(\epsilon-\gamma,I) + \delta. \]
If $(r,\alpha,\beta)\in I_V$, as $\hat{\bm{z}}(\epsilon,I)$ is strictly increasing, we derive that
\begin{align*}
\circFSuper{r,\alpha,\beta}(\hat{\bm{z}}(\epsilon,I),\hat{\bm{\chi}}(\epsilon,I)) &\leq f^{r,\alpha,\beta}(\hat{\bm{z}}(\epsilon,I),\hat{\bm{\chi}}(\epsilon,I))\\
&\leq f^{r,\alpha,\beta}(\hat{\bm{z}}(\epsilon-\gamma,I),\hat{\bm{\chi}}(\epsilon-\gamma,I)) + \delta\\
&\leq f^{r,\alpha,\beta}(\hat{\bm{z}}(\epsilon-\gamma,I),\hat{\bm{\chi}}(\epsilon,I)) + \delta\\
&\leq \circFSuper{r,\alpha,\beta}(\hat{\bm{z}}(\epsilon,I),\hat{\bm{\chi}}(\epsilon,I)) + \delta
\end{align*}
and choosing $\delta$ arbitrarily small we conclude that for almost every $\epsilon>0$,
\[ f^{r,\alpha,\beta}(\hat{\bm{z}}(\epsilon,I),\hat{\bm{\chi}}(\epsilon,I))=\circFSuper{r,\alpha,\beta}(\hat{\bm{z}}(\epsilon,I),\hat{\bm{\chi}}(\epsilon,I))=-\epsilon. \]
Moreover, if $m\in I_M$, then $\hat{\chi}^m(\epsilon,I)$ is strictly increasing and by the assumption of $h^m(\bm{\chi})$ being strictly increasing in $\chi^m$, we derive for $x^m>0$ that
\[ \rho\left(\frac{\ell+\bm{x}\cdot h(\hat{\bm{\chi}}(\epsilon-\gamma,I))}{c}\right) \leq \circRho\left(\frac{\ell+\bm{x}\cdot h(\hat{\bm{\chi}}(\epsilon,I))}{c}\right). \]
Hence similarly as above, $f^m(\hat{\bm{z}}(\epsilon,I),\hat{\bm{\chi}}(\epsilon,I))=\circFSuper{m}(\hat{\bm{z}}(\epsilon,I),\hat{\bm{\chi}}(\epsilon,I))=-\epsilon$.

Let us now show that $(\bm{z}^*,\bm{\chi}^*)\leq(\hat{\bm{z}}(\epsilon,I),\hat{\bm{\chi}}(\epsilon,I))$. To this end, suppose that we could choose some $(\bm{z},\bm{\chi})\in P_0$ such that $z^{r,\alpha,\beta}>\hat{z}^{r,\alpha,\beta}(\epsilon,I)$ for some $(r,\alpha,\beta)\in V$ or $\chi^m>\hat{\chi}(\epsilon,I)$ for some $m\in[M]$. Then by $P_0\subset[0,\bm{\zeta}]\times[\bm{0},\bm{\eta}]$, where $\zeta^{r,\alpha,\beta}=\E[W^{+,r,\alpha}\1\{A=\beta\}]$ and $\eta^m=\E[X^m]$, and monotonicity of the functions $f^{r,\alpha,\beta}$ resp.~$f^m$ we would derive the existence of a point $P_0\ni(\tilde{\bm{z}},\tilde{\bm{\chi}})\leq(\hat{\bm{z}}(\epsilon,I),\hat{\bm{\chi}}(\epsilon,I))$ such that either $\tilde{z}^{r,\alpha,\beta}=\hat{z}^{r,\alpha,\beta}(\epsilon,I)$ for some $(r,\alpha,\beta)\in I_V$ or $\tilde{\chi}^m=\hat{\chi}^M(\epsilon,I)$ for some $m\in I_M$. But then $f^{r,\alpha,\beta}(\tilde{\bm{z}},\tilde{\bm{\chi}})\leq f^{r,\alpha,\beta}(\hat{\bm{z}}(\epsilon,I),\hat{\bm{\chi}}(\epsilon,I))=-\epsilon$ resp.~$f^m(\tilde{\bm{z}},\tilde{\bm{\chi}})\leq f^m(\hat{\bm{z}}(\epsilon,I),\hat{\bm{\chi}}(\epsilon,I))=-\epsilon$ and thus a contradiction to $(\tilde{\bm{z}},\tilde{\bm{\chi}})\in P_0$. It must therefore hold $(\bm{z}^*,\bm{\chi}^*)\leq(\hat{\bm{z}}(\epsilon,I),\hat{\bm{\chi}}(\epsilon,I))$.
Let now for given shock $L$,
\[ I_M := \{m\in[M]\,:\,\hat{\chi}_L^m<\E[X^m]\}\quad\text{and}\quad I_V := \{(r,\alpha,\beta)\in V\,:\,\hat{z}_L^{r,\alpha,\beta}<\E[W^{+,r,\alpha}\1\{A=\beta\}]. \]
If $I_V=\emptyset$ and $I_M=\emptyset$ the result is trivial, so assume that either $I_V\neq\emptyset$ or $I_M\neq\emptyset$. For $m\in I_M$, using $\circFSuperSub{m}{L}(\hat{\bm{z}}_L,\hat{\bm{\chi}}_L)=0$ we then derive
\[ \circFSuper{m}(\hat{\bm{z}}_L,\hat{\bm{\chi}}_L) = \frac{\circFSuperSub{m}{L}(\hat{\bm{z}}_L,\hat{\bm{\chi}}_L) - \P(L=2C) (\E[X^m]-\hat{\chi}_L^m)}{\P(L=0)} <0 \]
and analogously $\circFSuper{r,\alpha,\beta}(\hat{\bm{z}}_L,\hat{\bm{\chi}}_L)<0$ for $(r,\alpha,\beta)\in I_V$. For the choice
\[ \epsilon := -\max\left\{\max_{(r,\alpha,\beta)\in I_V} \circFSuper{r,\alpha,\beta}(\hat{\bm{z}}_L,\hat{\bm{\chi}}_L), \max_{m\in I_M} \circFSuper{m}(\hat{\bm{z}}_L,\hat{\bm{\chi}}_L)\right\} >0 \]
it then holds $(\hat{\bm{z}}_L,\hat{\bm{\chi}}_L)\in T(\epsilon,I)$ and further $(\hat{\bm{z}}_L,\hat{\bm{\chi}}_L)\geq(\hat{\bm{z}}(\epsilon,I),\hat{\bm{\chi}}(\epsilon,I))\geq(\bm{z}^*,\bm{\chi}^*)$. We can then apply Theorem \ref{5:thm:final:fraction:combined:general} and conclude
\[ n^{-1}\mathcal{S}_{n,L} \geq \circG_L(\hat{\bm{z}}_L,\hat{\bm{\chi}}_L) + o_p(1) \geq \circG(\bm{z}^*,\bm{\chi}^*) + o_p(1) \]
as well as
\[ \chi_{n,L}^m \geq \hat{\chi}_L^m + o_p(1) \geq (\chi^*)^m + o_p(1).\qedhere \]
\end{proof}

\cleardoublepage
\bibliography{finance}
\bibliographystyle{abbrv}
\addcontentsline{toc}{chapter}{Bibliography}
\end{document}